%% file: draft.tex
\newcommand{\shortVersion}{false} 

\newcommand{\techReportAppendix}[1]{%
    \IfEqCase{\shortVersion}{%
        {false}{Appendix \ref{#1}}%
        {true}{\cite{technicalReport}}%
    }
}%

\newcommand{\techReportAppendices}[2]{%
    \IfEqCase{\shortVersion}{%
        {false}{Appendices \ref{#1}--\ref{#2}}%
        {true}{\cite{technicalReport}}%
    }
}%

\newcommand{\onlyTechReport}[1]{%
    \IfEqCase{\shortVersion}{%
        {false}{#1}%
    }
}%

\newcommand{\onlyShortVersion}[1]{%
    \IfEqCase{\shortVersion}{%
        {true}{#1}%
    }
}%

		\documentclass{sig-alternate}		      
       	\newdef{definitionInt}{Definition}
		\numberwithin{definitionInt}{section}
		\newdef{attackInt}{Attack}
		\newdef{exampleInt}{Example}
		\numberwithin{exampleInt}{section}
		\newdef{decisionproblem}{Problem}
		\newtheorem{theorem}{Theorem}
		\numberwithin{theorem}{section}
		
		\numberwithin{corollary}{section}
		\newtheorem{lemma}{Lemma}
		\numberwithin{lemma}{section}
		\newtheorem{proposition}{Proposition}
		\numberwithin{proposition}{section}

\usepackage{footnote}
\usepackage{pifont}
\usepackage{setspace}
\usepackage{xstring}
\usepackage[nocompress]{cite}

\usepackage{etex}
\usepackage [table]{xcolor} 
\usepackage{array}
\usepackage{balance}
\usepackage{subcaption}
\usepackage{pgfplots}
\pgfplotsset{compat=newest}
\usepgfplotslibrary{units}
\usepackage{pgfplotstable}

\usepackage{proof}
\usepackage{bussproofs}
\usepackage{amssymb}
\usepackage{amsmath}

\makeatletter
\newcommand{\thickhline}{%
    \noalign {\ifnum 0=`}\fi \hrule height 1.5pt
    \futurelet \reserved@a \@xhline
}
\newcolumntype{"}{@{\hskip\tabcolsep\vrule width 1.5pt\hskip\tabcolsep}}
\makeatother

\usepackage{theoremref} 
\usepackage{multirow} 
\usepackage{listings}
\usepackage{comment}
\usepackage{todonotes}
\usepackage[labelfont=bf,font=bf]{caption}
\usepackage{balance}
\usepackage{paralist}
\usepackage{stmaryrd}
\usepackage{rotating}
\usepackage{cals}
\usepackage{float}

\newfloat{derivation}{thp}{lop}
\floatname{derivation}{Derivation}

\usepackage{pgf}
\usepackage{tikz}
\usetikzlibrary{arrows,automata}
\usetikzlibrary{positioning}
\usetikzlibrary{shapes.geometric}
\usetikzlibrary{plotmarks}
\usetikzlibrary{patterns}
\usetikzlibrary{calc}
\usepackage[normalem]{ulem}

\onlyTechReport{
\usepackage{placeins}
}

 \newenvironment{example}
 {\begin{exampleInt} } 
 { $\blacksquare$\end{exampleInt} }

  \newenvironment{definition}
 {\begin{definitionInt} } 
 { $\square$\end{definitionInt} }

\makeatletter
\newcommand{\xRightarrow}[2][]{\ext@arrow 0359\Rightarrowfill@{#1}{#2}}
\makeatother

 \newcommand{\concat}{\cdot}
 
 \newcommand{\admin}{\mathit{admin}}

 \newcommand{\correctness}{database integrity}
 \newcommand{\confidentiality}{data confidentiality}
 \newcommand{\Correctness}{Database Integrity}
 \newcommand{\Confidentiality}{Data Confidentiality}
 
 \newcommand{\auth}{\leadsto_{\mathit{auth}}}
 
 \newcommand{\attMod}{\vdash_{u}}
 \newcommand{\attackerModel}{{\cal ATK}_{u}}
 
 \newcommand{\accessControlFunction}{Policy Decision Point}
  \newcommand{\acf}{PDP}
  \newcommand{\accessControlConfiguration}{extended configuration}

\title{Strong and Provably Secure Database Access Control}

\onlyTechReport{
\numberofauthors{3} 
\author{
\alignauthor
Marco Guarnieri\\
       \affaddr{{\emph{Institute~of~Information~Security}}}\\
       \affaddr{{\emph{Department~of~Computer~Science}}}\\
       \affaddr{{\emph{ETH Zurich, Switzerland}}}\\
       \email{{\texttt{marco.guarnieri@inf.ethz.ch}}}
\alignauthor
Srdjan Marinovic\\
       \affaddr{{\emph{The Wireless Registry, Inc.}}}\\
       \affaddr{{\emph{Washington DC, US}}}\\
       \email{{\texttt{srdjan@wirelessregistry.com}}}
\alignauthor
David Basin\\
       \affaddr{{ \emph{Institute~of~Information~Security}}}\\
       \affaddr{{\emph{Department of Computer Science}}}\\
       \affaddr{{\emph{ETH Zurich, Switzerland}}}\\
       \email{{\texttt{basin@inf.ethz.ch}}}
       }
}

\onlyShortVersion{
\author{\IEEEauthorblockN{Marco Guarnieri}
\IEEEauthorblockA{
	 Institute of Information Security\\
       Department of Computer Science\\
       ETH Zurich, Switzerland\\
       {\small {\tt marco.guarnieri@inf.ethz.ch}}}
\and
\IEEEauthorblockN{Srdjan Marinovic}
\IEEEauthorblockA{The Wireless Registry, Inc.\\
		Washington DC, US\\
       	{\small {\tt srdjan@wirelessregistry.com}}}
\and
\IEEEauthorblockN{David Basin}
\IEEEauthorblockA{Institute of Information Security\\
       Department of Computer Science\\
       ETH Zurich, Switzerland\\
       {\small {\tt basin@inf.ethz.ch}}}
}
}

\begin{document}

\maketitle

\begin{abstract}
Existing SQL access control mechanisms are extremely limited.  
Attackers can leak information and escalate their privileges using advanced database features such as views, triggers, and integrity constraints.  
This is not merely a problem of vendors lagging behind the state-of-the-art.  
The theoretical foundations for database security lack adequate security definitions and a realistic attacker model, both of which are needed to evaluate the security of modern databases.
We address these issues and present a provably secure access control mechanism that prevents attacks that defeat popular SQL database systems.
\end{abstract}

\input{introduction.tex}

\input{motivating.tex}

\input{sql_syn.tex}

\input{databases.tex}

\input{system.tex}
\input{lts_short.tex}

\input{advModel_short.tex}

\input{securityDef.tex}
\input{enforcement.tex}

\input{related.tex}

\input{conclusion.tex}

\bibliographystyle{IEEEtranS}
\onlyShortVersion{
\balance
}
\bibliography{IEEEabrv,bib/main}

\newpage
\appendix
In this appendix we formalize the system's operational semantics, the attacker model, and the security properties.
Furthermore, we present complete proofs of all results.

For simplicity's sake, in the following we  assume, without loss of generality, that all the relational calculus formulae do not use constant symbols inside predicates.
For instance, instead of the formula $\exists x.\, R(x,5,10)$, we consider the equivalent formula $\exists x,y,z.\,  R(x,y,z) \wedge y = 5 \wedge z = 10$.
Note that this does not restrict the scope of our work as all formulae can be trivially expressed without using constant symbols inside predicates. 

\smallskip
\noindent
{\bf Structure.}
In Appendix~\ref{app:lts}, we provide a complete formalization of our system model.
In Appendix~\ref{app:adv:model}, we present all the rules defining the $\attMod$ relation and we prove the soundness of $\attMod$ with respect to the relational calculus semantics.
In Appendix~\ref{app:eop}, we provide the complete formalization of the $\auth$ relation.
In Appendix~\ref{app:indistinguishability}, we formalize $u$-projections and the indistinguishability relation $\cong_{P,u}$.
In Appendix~\ref{app:enforcement:eopsec} we formalize the access control function $f_{\mathit{int}}$, we prove that it provides \correctness{}, and we prove its data complexity.
In Appendix~\ref{app:enforcement:ibsec} we formalize the access control function $f_{\mathit{conf}}^{u}$, we prove that it provides \confidentiality{}, and we prove its data complexity.
In Appendix~\ref{app:composition} we prove that the function $f$, which is obtained by composing the \acf{}s  $f_{\mathit{int}}$ and $f_{\mathit{conf}}^{u}$, provides both \correctness{} and \confidentiality{}.
We also prove that its data complexity is $\mathit{AC}^0$.
Finally, in Appendix~\ref{app:data:conf:non:interference} we show that the concepts of secure judgment and \confidentiality{} have precise interpretations in terms of non-interference.

\input{lts_long_small_step.tex}

\input{attacker_long.tex}

\input{eop.tex}

\input{indistinguishability.tex}

\input{eopsecEnforcement.tex}
\input{ibsecEnforcement.tex}

\input{composition.tex}

\input{nonInterference.tex}

\input{sql_fragment.tex}


\end{document}

%% file: introduction.tex
\section{Introduction}

It is essential to control access to databases that store sensitive information.
To this end, the SQL standard  defines access control rules and all SQL database vendors have accordingly developed  access control mechanisms.
The standard however fails to define a precise access control semantics, the attacker model, and the security properties that the mechanisms ought to satisfy.
As a consequence, existing access control mechanisms are implemented in an ad hoc fashion, with neither precise security guarantees nor the means to verify them.

This deficit has dire and immediate consequences.
We show that popular   database systems are susceptible to two types of attacks.
{Integrity attacks} allow an attacker to perform non-authorized changes to the database.
{Confidentiality attacks} allow an attacker to learn sensitive data.
These attacks exploit advanced SQL features, such as triggers, views, and integrity constraints, and they are  easy to carry~out.

Current research efforts in data\-base security are neither adequate for evaluating the security of modern databases, nor do they account for their advanced features. 
In more detail, existing research \cite{ wang2007correctness,agrawal2005extending, rizvi2004extending,brodsky2000secure} implicitly considers  attackers who use \texttt{SELECT} commands. 
But the capabilities offered by databases go far beyond \texttt{SELECT}.
Users, in general, can modify the database's state and security policy, as well as use features such as triggers,~views, and integrity constraints.
Consequently, all proposed research solutions  fail to prevent attacks such as those we present~in~\S2.

In summary, the database vendors have been left to develop access control mechanisms without guidance from either the SQL standard or existing research in database security.
It is therefore not surprising that modern databases are open to abuse.

\smallskip
\noindent
{\bf Contributions.} 
We develop a comprehensive formal framework for the design and analysis of database access control. 
We use it to design and verify an access control mechanism that prevents confidentiality and integrity attacks that defeat existing mechanisms. 

First, we develop an operational semantics for databases that supports SQL's core features, as well as triggers, views, and integrity constraints.
Our semantics models both the security-critical aspects of these features and the database's dynamic behaviour at the level needed to capture realistic attacks. 
Our semantics is substantially more detailed than those used in previous works~\cite{rizvi2004extending, wang2007correctness}, which ignore the database's dynamics.

Second, we develop a novel attacker model that, in~addition to SQL's core features, incorporates advanced features such as triggers, views, and integrity constraints.
Furthermore, our attacker can infer information based on the semantics of these features. 
Note that our attacker model subsumes the \texttt{SELECT}-only attacker considered in previous works~\cite{rizvi2004extending},~\cite{wang2007correctness}.
We also develop an executable version of our operational semantics and attacker model using the Maude term-rewriting framework  \cite{clavel2003maude}.
The executable model acts as a reference implementation for our semantics. 
Given the complexity of databases and their features, having an executable version of our models provides a way to validate them against existing database systems and against the examples we use in this paper. 

Third, we present two security definitions---\correctness{} and \confidentiality{}---that reflect two principal security requirements for database access control.
There is a natural and intuitive relationship between these definitions and the types of attacks that we identify.
We thus argue that these definitions provide a strong measure of whether a given access control mechanism prevents our attacker from exploiting modern SQL databases.

Finally, using our framework, we build a database access control mechanism that is provably secure with respect to our attacker model and security definitions.
In contrast to existing mechanisms, our solution prevents all the attacks that we report on in \S\ref{sect:motivating}. 

\smallskip
\noindent
{\bf Related Work.} 
Surprisingly, and in contrast to other areas of information security\onlyTechReport{ \cite{dolev1983security}}, there does not exist a well-defined  attacker model for database access control. 
From the literature, we extracted the \texttt{SELECT}-only attacker model, where the attacker uses just \texttt{SELECT} commands. 
A number of access control mechanisms, such as~\onlyShortVersion{\cite{ wang2007correctness,agrawal2005extending, rizvi2004extending,Bender:2013:FDC:2463676.2467798,Bender:2014:ESR:2588555.2593663}}\onlyTechReport{\cite{ wang2007correctness,agrawal2005extending, rizvi2004extending,browder2002virtual, sybase2003security,shi2009soundness,lefevre2004limiting,halder2013fine,stonebraker1974access,Bender:2013:FDC:2463676.2467798,Bender:2014:ESR:2588555.2593663}}, 
 implicitly consider this attacker model.
The boundaries of this model are blurred and the attacker's capabilities are unclear.  
For instance, only a few works, such as~\cite{wang2007correctness}, explicitly state that update commands are not supported, whereas others~\cite{rizvi2004extending, agrawal2005extending, Bender:2013:FDC:2463676.2467798,Bender:2014:ESR:2588555.2593663} ignore what the attacker can learn from update commands.
Works on Inference Control~\cite{ brodsky2000secure, toland2010inference, farkas2002inference} and Controlled Query Evaluation~\cite{bonatti1995foundations} consider a variation of the \texttt{SELECT}-only attacker, in which the attacker additionally  has some initial knowledge about the data and can derive new information from the query's results through inference rules.
Note that while~\cite{toland2010inference} supports update commands, it treats them just as a way of increasing data availability, rather than considering them as a possible attack vector.

Database access control mechanisms can be classified into two distinct families~\cite{rizvi2004extending}. 
Mechanisms in the \emph{Truman model}\onlyShortVersion{~}\cite{wang2007correctness,agrawal2005extending} transparently modify query results to restrict the user's access to the data authorized by the  policy.
In contrast, mechanisms in the \emph{Non-Truman model}~\cite{rizvi2004extending,Bender:2013:FDC:2463676.2467798,Bender:2014:ESR:2588555.2593663} either accept or reject queries without modifying their  results.
Different notions of security have been proposed for these models~\cite{rizvi2004extending,wang2007correctness,guarnieri2014optimal}.
They are, however, based on \texttt{SELECT}-only attackers and provide no security guarantees against realistic attackers that can alter the database and the policy or use advanced SQL~features.
We refer the reader to \S\ref{sect:related:work} for further comparison with related~work.

\smallskip
\noindent
{\bf Organization.} 
In \S\ref{sect:motivating} we present attacks that illustrate serious weaknesses in existing Data\-base Management Systems (DBMSs).
In \S\ref{sect:preliminaries} we introduce background and notation about queries, views, triggers, and access control.
In \S\ref{sect:syst:attk:model} we formalize our system and attacker models, and in \S\ref{sect:sec:prop} we define the desired security properties.
In \S\ref{sect:enf:alg} we present our access control mechanism, and in \S\ref{sect:related:work} we discuss related work.
Finally,  we draw conclusions in \S\ref{sect:conclusion}.
\onlyShortVersion{An extended version of this paper, with the system's operational semantics, the attacker model, and complete proofs of all results, is available at~\cite{technicalReport}.
}
\onlyTechReport{The system's operational semantics, the attacker model, and complete proofs of all results are in Appendices~\ref{app:lts}--\ref{app:data:conf:non:interference}.
}
A prototype of our enforcement mechanism and its executable semantics are available at~\cite{prototype}.
\onlyTechReport{
This technical report is an extended version of~\cite{guarnieri2016strong}.
}

%% file: motivating.tex
\section{Illustrative Attacks}
\label{sect:motivating}

\begin{table*}[!hbtp]
\centering
\begin{tabular}{l| c c  c | c c }
 
   &  \multicolumn{3}{c | }{{\bf Integrity Attacks}} & \multicolumn{2}{c}{{\bf Confidentiality Attacks}}  \\
\multicolumn{1}{c|}{\multirow{2}{*}{ \bf  DBMS}} & {Triggers with} &  Granting  & Revoking  & Table updates and & {Triggers with} \\
  & activator's privileges &  views & views & integrity constraints & owner's privileges \\
 \hline
 IBM DB2 (v.~10.5) 		& $\dagger$ 		&  ${\cal X}$ 	&  $\checkmark$  	& $\checkmark$ 	& $\checkmark$ \\
 Oracle (v.~11g) 		& $\dagger$ 		&  ${\cal X}$ 	&  ${\cal X}$ 		& $\checkmark$	& $\checkmark$ \\
 PostgreSQL (v.~9.3.5) 	& $\checkmark$	&  $\checkmark$ 	&  $\checkmark$	 	& $\checkmark$ 	& $\checkmark$ \\
 
 MySQL (v.~14.14) 		& $\dagger$		&  ${\cal X}$ 	&  $\checkmark$ 	 	& $\checkmark$ 	& $\checkmark$ \\
 
 SQL Server (v.~12.0) 	& $\checkmark$	&  $\dagger$		&  $\dagger$ 		& $\checkmark$	& $\checkmark$  \\
 Firebird (v.~2.5.2) 	& $\checkmark$ 	&  ${\cal X}$ 	&  $\checkmark$ 	 	& $\checkmark$ 	& $\checkmark$ \\
 \hline
\end{tabular}
\captionof{figure}{The $\checkmark$ symbol indicates a successful attack, whereas ${\cal X}$ indicates a failed attack. The $\dagger$ symbol indicates that the DBMS does not support the features necessary to launch the attack. } 
\label{table:attacks}
\end{table*}

We demonstrate here  how attackers can exploit existing
DBMSs using standard SQL features.
We classify these attacks as either  \emph{Integrity Attacks} or
\emph{Confidentiality Attacks}.
In the former, an attacker makes unauthorized changes to the database, which stores the data, the policy, the triggers, and the views.
In the latter, an attacker learns sensitive~data by interacting with
the system and observing the outcome.
No existing access control mechanism prevents all the attacks we  present. 
Moreover, many related attacks can be constructed using variants of the ideas presented here.
We manually carried out the attacks against IBM~DB2, Oracle~Database, PostgreSQL,  MySQL, SQL~Server,  and  Firebird.
We summarize our findings at the end of this~section.

\subsection{Integrity Attacks}
Our three integrity attacks combine different database features: \texttt{INSERT}, \texttt{DELETE}, \texttt{GRANT}, and \texttt{REVOKE} 
commands together with views and triggers.
In the first attack, an attacker creates a trigger, i.e., a procedure automatically executed by the DBMS in response to user commands, that will be
activated by an unaware user with a higher security clearance and will perform unauthorized changes to the database.
The attack requires triggers to be executed under the privileges of
the users activating them.
Such triggers are supported  by PostgreSQL, SQL Server,  and Firebird.

\begin{attackInt}\label{example:trigger:attack:activator}
{\bf {Triggers with activator's privileges.}}
Consider a data\-base with two tables $P$ and $S$ and two users $u_1$ and $u_2$.
The attacker is the user $u_1$, whose goal is to delete the content of $S$.
The policy is that
$u_1$ is not authorized\footnote{As is common in SQL, a user is authorized to execute a command if and only if the policy assigns him the corresponding permission.
} to alter  $S$, $u_1$ can  create triggers on $P$, and  $u_2$  can read and modify $S$ and $P$.
The attack is as follows:
\begin{compactenum}
\item  $u_1$ creates the trigger:
\begin{lstlisting}[mathescape, basicstyle=\ttfamily\small]
CREATE TRIGGER $t$ ON $P$ AFTER INSERT
  DELETE FROM $S$;
\end{lstlisting}

\item $u_1$ waits until $u_2$  inserts a tuple into the table $P$. The trigger will then be invoked using $u_2$'s privileges and $S$'s content  will be deleted.  $\blacksquare$
\end{compactenum}
\end{attackInt} 

An attacker can use similar attacks to execute arbitrary commands with administrative privileges.
Despite the threat posed by such simple attacks, the existing
countermeasures~\cite{microsoft2014} are unsatisfactory; they are
either too restrictive, for instance completely disabling triggers in the~data\-base, or too time consuming and error prone, namely manually checking if
``dangerous'' triggers have been created.

In our second attack, an attacker escalates his privileges by delegating the read permission for a table without being authorized to delegate this permission.
The attacker first~creates a view over the table and, afterwards, delegates the access to the view to another user.
This attack exploits DBMSs, such as PostgreSQL, where a user can grant any read permission over his own views.
Note that  \texttt{GRANT} and \texttt{REVOKE} commands are \emph{write operations}, which target the database's internal configuration instead of the tables.

\begin{attackInt}\label{example:view:escalation1}
{\bf {Granting views.}}
Consider a database with a table $S$, two users $u_{1}$ and $u_{2}$, and the following policy:
$u_{1}$ can create views and read $S$ (without being able to
  delegate this permissions), and $u_{2}$ cannot read  $S$. 
The attack is as~follows:
\begin{compactenum}
\item $u_{1}$ creates the view: \texttt{CREATE VIEW} $v$ \texttt{AS SELECT} $*$ \texttt{FROM}\onlyTechReport{~}\onlyShortVersion{ }$S$. 
\item $u_{1}$ issues the command \texttt{GRANT SELECT ON} $v$ \texttt{TO} $u_{2}$. 
Now, $u_{2}$ can read $S$ through $v$.
However, $u_{1}$ is not authorized to delegate the read permission on $S$.
 $\blacksquare$ %
\end{compactenum}
\end{attackInt}

This attack exploits several subtleties in the commands' semantics:\begin{inparaenum}[(a)]
\item users can create views over all tables they can read,
\item the views are executed under the owner's privileges, and
\item view's owners can grant arbitrary permissions over their own views.
\end{inparaenum}
These features give $u_{1}$ the implicit ability to delegate the read access over $S$.
As a result, the overall system's behaviour does not conform with the given policy.
That is, $u_{1}$
should not be permitted to delegate the read access to
$S$ or to any view that depends on it.
Note that the commands' semantics may vary between different~DBMSs. 

In our third attack, an attacker exploits the failure of access control
mechanisms to propagate \texttt{REVOKE} commands.

\begin{attackInt}\label{example:view:escalation2}
{\bf {Revoking views.}}
Consider a database with a table $S$, three users $u_{1}$, $u_{2}$, and $u_{3}$, and the following policy:
$u_{1}$ can read $S$ and delegate this permission,
$u_{2}$ can create views, and
$u_{3}$ cannot read $S$.
The attack proceeds as follows:
\begin{compactenum}
\item $u_{1}$ issues the command \texttt{GRANT SELECT ON} $S$ \texttt{TO} $u_{2}$ \texttt{WITH GRANT OPTION}.
\item $u_{2}$ creates the view: \texttt{CREATE VIEW} $v$ \texttt{AS SELECT} $*$ \texttt{FROM}\onlyTechReport{~}\onlyShortVersion{ }$S$.
\item $u_{2}$ issues the command \texttt{GRANT SELECT ON} $v$ \texttt{TO} $u_{3}$.  
\item $u_{1}$ revokes the permission to read  $S$ (and to delegate the permission) from $u_{2}$: \texttt{REVOKE SELECT ON}  $S$ \texttt{FROM} $u_{2}$. 
Now, $u_{3}$ cannot read $v$ because $u_{2}$, which is $v$'s owner, cannot read $S$.

\item $u_{1}$ grants again the permission to read  $S$ to $u_{2}$: \texttt{GRANT SELECT ON} $S$ \texttt{TO} $u_{2}$.
Now, $u_{3}$ can again read $v$ but $u_{2}$ can no longer delegate the read permission on $v$. $\blacksquare$
\end{compactenum}
\end{attackInt} 

This attack succeeds because, in the fourth step,
the \texttt{REVOKE} statement does not remove the \texttt{GRANT} granted by $u_{2}$ to $u_{3}$ to read $v$.
This \texttt{GRANT} only becomes ineffective because $u_{2}$ is no longer authorized to read $S$.
However, after the fifth step, this \texttt{GRANT} becomes 
effective again, even though $u_{2}$ can no longer delegate the read permission on $v$. 
Thus, the policy is left in an inconsistent state.

\subsection{Confidentiality Attacks}
We now present two attacks that use \texttt{INSERT} and
\texttt{SELECT} commands together with triggers and integrity
constraints. 
In our fourth attack, an attacker exploits integrity constraint violations to
learn sensitive information.
An integrity constraint is an invariant that must be
satisfied for a database state to be considered \emph{valid}.
\emph{Integrity constraint violations} arise when the
execution of an SQL command leads the database from a valid state into
an invalid one.

\begin{attackInt}\label{example:integrity:delete:information}
{\bf {Table updates and integrity constraints.}}
\\
Consider a database with two tables $P$ and $S$.
Suppose the primary key of both tables is the user's identifier.
Furthermore, the set of user identifiers in $S$ is contained in  the set of user identifiers in $P$, i.e., there is a foreign key from $S$ to $P$.
The attacker is the user $u$ whose goal is to learn whether $\mathtt{Bob}$ is in $S$.
The access control policy is that $u$ can read ${P}$ and insert tuples in ${S}$.
The attacker $u$ can learn whether $\mathtt{Bob}$ is in $S$ as follows:
\begin{compactenum}
\item He reads ${P}$ and learns $\mathtt{Bob}$'s identifier.
\item He issues an \texttt{INSERT} statement in ${S}$ using  $\mathtt{Bob}$'s id. 
\item If $\mathtt{Bob}$ is already in ${S}$, then $u$ gets an error message about the primary key's violation.
Alternatively, there is no violation and $u$ learns that $\mathtt{Bob}$ is not in $S$. $\blacksquare$
\end{compactenum}
\end{attackInt}

Even though similar attacks have been identified before~\cite{jajodia1990polyinstantiation,schultz2013ifdb}, existing~DBMSs~are~still~vulnerable.

In our fifth attack, an attacker learns sensitive information by exploiting
the system's triggers.
The trigger in this attack is executed under the privileges of
the trigger's owner.
Such triggers are supported by IBM DB2, Oracle Database, 
PostgreSQL, MySQL, SQL Server,  and Firebird.

\begin{attackInt}\label{example:trigger:attack:owner}
{\bf {Triggers with owner's privileges.}}
Consider a data\-base with three tables $N$, $P$, and $T$. 
The attacker is the user $u$, who wishes to learn whether $v$ is in $T$.
The  policy is that
  $u$ is not authorized to read the table $T$, and 
  he can read and modify the tables $N$   and $P$. 
Moreover, the following trigger has been defined by the administrator.
\begin{lstlisting}[mathescape, basicstyle=\ttfamily\small]
CREATE$\enspace$TRIGGER$\enspace$t$\enspace$ON$\enspace P\enspace$AFTER$\enspace$INSERT$\enspace$FOR$\enspace$EACH$\enspace$ROW
  IF exists(SELECT * FROM $T$ WHERE id = NEW.id)
   INSERT INTO $N$ VALUES (NEW.id);
\end{lstlisting}

The attack is as follows:
\begin{compactenum}
\item $u$ deletes $v$ from  $\mathit{N}$.
\item $u$ issues the command \texttt{INSERT}\,\texttt{INTO}\,$P$\;\texttt{VALUES}\,\texttt{(}$v$\texttt{)}.
\item $u$ checks the table $N$. If it contains $v$'s id, then $v$ is in $T$. Otherwise, $v$ is not in $T$. $\blacksquare$
\end{compactenum}
\end{attackInt} 

This attack exploits that the trigger $t$ conditionally
modifies the database.
Furthermore, the attacker can activate $t$, by inserting
tuples in $P$, and then observe $t$'s effects,
by reading the table $N$.
He therefore can exploit $t$'s execution to learn whether $t$'s condition holds.
We assume here that the attacker knows the triggers in the system.
This is, in general, a weak assumption as triggers usually
describe the domain-specific rules regulating a system's behaviour and
 users are usually aware of them.

\subsection{Discussion}\label{sect:motivating:discussion}

We manually carried out all five attacks against IBM DB2, Oracle Database, PostgreSQL, MySQL, SQL Server, and Firebird.
Figure~\ref{table:attacks} summarizes our findings.
None of these systems prevent the confidentiality attacks.
They are however more successful in preventing the integrity
attacks.
The most successful is Oracle Database, which prevents two of the three attacks, while Attack~\ref{example:trigger:attack:activator} cannot be carried out due to missing features. 
IBM DB2, MySQL,  and Firebird prevent just one of the three attacks, namely~Attack~\ref{example:view:escalation1}.
However, they all fail to prevent Attack~\ref{example:view:escalation2}.
Note that Firebird also fails to prevent Attack~\ref{example:trigger:attack:activator}.
In contrast, Attack~\ref{example:trigger:attack:activator} cannot be carried out against MySQL and IBM DB2 due to missing features. 
SQL Server also fails to prevent Attack~\ref{example:trigger:attack:activator}; however the remaining two
attacks cannot be carried out due to missing features.
PostgreSQL fails to prevent all three attacks.

We argue that the dire state of database access control mechanisms,
as illustrated by these attacks, comes from the lack of clearly
defined security properties that such mechanisms ought to satisfy and
the lack of a well-defined attacker model.
We therefore develop a formal attacker model and precise
security properties  and we use them to design a provably secure access control
mechanism that prevents all the above~attacks.

%% file: sql_syn.tex
\onlyTechReport{\newpage}
\section{Database Model}\label{sect:preliminaries}

We now formalize   databases including features like views, access control policies, and triggers.
Our formalization of databases and queries follows~\cite{abiteboul1995foundations}, and our access control policies formalize SQL policies.

\subsection{Overview}\label{sect:fragment}
 
In this paper we consider the following SQL features: \texttt{SELECT}, \texttt{INSERT}, \texttt{DELETE}, \texttt{GRANT}, \texttt{REVOKE}, \texttt{CREATE TRIGGER}, \texttt{CREATE VIEW}, and \texttt{ADD USER} commands. 

For \texttt{SELECT} commands, rather than using SQL, we use the relational calculus ($\mathit{RC}$), i.e., function-free first-order logic, which has a simple and well-defined  semantics~\cite{abiteboul1995foundations}.
We support \texttt{GRANT} commands with  the \texttt{GRANT OPTION} and \texttt{REVOKE} commands with the \texttt{CASCADE OPTION}, i.e., when a user revokes a privilege, he also revokes all the privileges that depend on it.
We support \texttt{INSERT} and \texttt{DELETE} commands that explicitly identify the tuple to be inserted or deleted, i.e., commands of the form 
\texttt{INSERT INTO} $\mathit{table}(x_{1}, \ldots,  x_{n})$ \texttt{VALUES} $(v_{1}, \ldots, v_{n})$ and \texttt{DELETE FROM} $\mathit{table}$ \texttt{WHERE} $x_{1} = v_{1} \wedge \ldots \wedge x_{n} = v_{n}$, where $x_{1}, \ldots, x_{n}$ are $\mathit{table}$'s attributes and $v_{1}, \ldots, v_{n}$ are the tuple's values.
More complex \texttt{INSERT} and \texttt{DELETE} commands, as well as \texttt{UPDATE}s, can be simulated by combining \texttt{SELECT}, \texttt{INSERT}, and \texttt{DELETE} commands.

We support only  \texttt{AFTER} triggers on \texttt{INSERT} and \texttt{DELETE} events, i.e., triggers that are executed in response to \texttt{INSERT} and \texttt{DELETE} commands. 
The triggers' \texttt{WHEN} conditions are arbitrary boolean queries and their actions are \texttt{GRANT}, \texttt{REVOKE}, \texttt{INSERT}, or  \texttt{DELETE} commands. 
Note that DBMSs usually impose severe restrictions on the \texttt{WHEN} clause, such as it must not contain sub-queries. 
However, most DBMSs can express arbitrary conditions on triggers by combining control flow statements with \texttt{SELECT} commands inside the trigger's body.
Thus, we support the class of triggers whose body is of the form \lstinline[mathescape, basicstyle=\ttfamily]|BEGIN IF $\mathit{expr}$ THEN $\mathit{act}$ END|, where $\mathit{act}$ is either a \texttt{GRANT}, \texttt{REVOKE}, \texttt{INSERT}, or  \texttt{DELETE} command. 
Note that all triggers used in \S\ref{sect:motivating} belong to this class.

%
We support two kinds of integrity constraints: functional dependencies and inclusion dependencies \cite{abiteboul1995foundations}. They model the most widely used families of SQL  integrity constraints, namely the \texttt{UNIQUE}, \texttt{PRIMARY KEY}, and \texttt{FOREIGN KEY} constraints. 
We also support views with both the owner's privileges and the activator's privileges. 

The SQL fragment we support\onlyShortVersion{ }\onlyTechReport{, shown in Figure~\ref{table:sql:syntax}, } contains the most common SQL commands for data manipulation and access control as well as the core commands for creating triggers and views.
The ideas and the techniques presented in this paper are general and can be extended to the entire SQL~standard.

%% file: databases.tex
\subsection{Data\-bases and Queries}
Let ${\cal R}$, ${\cal U}$, ${\cal V}$, and ${\cal T}$ be mutually disjoint, countably infinite sets, respectively representing identifiers of relation schemas, users, views, and triggers.

A \emph{data\-base schema} $D$ is a pair $\langle \Sigma, \mathbf{dom} \rangle$, where $\Sigma$ is a first-order signature and $\mathbf{dom}$ is a fixed countably infinite domain.
The signature $\Sigma$ consists of a set of \emph{relation schemas} $R \in {\cal R}$, also called  \emph{tables}, with arity $|R|$ and sort $\mathit{sort}(R)$. 
A \emph{state} $s$ of $D$ is a finite $\Sigma$-structure over \textbf{dom}.
We denote  by $\Omega_{D}$ the set of all states.
Given a table $R \in D$, $s(R)$ denotes the set of tuples that belong to $R$ in $s$.

A \emph{query} $q$ over a schema $D$ is of the form $\{\overline{x}\,|\,\phi \}$, where $\overline{x}$ is a sequence of variables, $\phi$ is a relational calculus formula over $D$, and $\phi$'s free variables are those in $\overline{x}$.
A \emph{boolean query} is a query $\{\,|\, \phi\}$, also written  as $\phi$, where $\phi$ is a sentence. 
The result of executing a query $q$ on a state $s$, denoted by $[q]^{s}$, is a boolean value in $\{\top,\bot\}$, if $q$ is a boolean query, or a set of tuples otherwise.
We denote by $\mathit{RC}$ (respectively $\mathit{RC}_{\mathit{bool}}$) the set of all  relational calculus queries (respectively sentences).
We consider only \emph{domain-independent queries} as is standard, and we employ the standard relational calculus semantics~\cite{abiteboul1995foundations}.

Let $D = \langle \Sigma, \mathbf{dom} \rangle$ be a schema,  $s$ be a state in $\Omega_{D}$,  $R$ be a table in $D$, and  $\overline{t}$ be a tuple in $\mathbf{dom}^{|R|}$.
The result of inserting  (respectively deleting) $\overline{t}$ in $R$ in the state $s$ is the state $s'$, denoted by $s[R \oplus \overline{t}]$ (respectively $s[R \ominus \overline{t}]$), where $s'(T) = s(T)$ for all $T \in \Sigma$ such that $T \neq R$, and $s'(R) = s(R) \cup \{\overline{t}\}$ (respectively $s'(R) = s(R) \setminus \{\overline{t}\}$).

An \emph{integrity constraint over $D$} is a  relational calculus sentence $\gamma$ over $D$.
Given a state $s$, we say that \emph{$s$ satisfies the constraint $\gamma$} iff $[\gamma]^{s} = \top$.
Given a set of constraints $\Gamma$, $\Omega_{D}^{\Gamma}$ denotes the set of all states satisfying the constraints in $\Gamma$, i.e., $\Omega_{D}^{\Gamma} = \{s \in \Omega_{D}\,|\, \bigwedge_{\gamma \in  \Gamma} [\gamma]^{s} = \top\}$.
We consider two types of integrity constraints:
\emph{functional dependencies}, which are sentences of the form $\forall \overline{x}, \overline{y}, \overline{y}', \overline{z}, \overline{z}'.\,( (R(\overline{x}, \overline{y}, \overline{z}) \wedge R(\overline{x}, \overline{y}', \overline{z}') )\Rightarrow \overline{y} = \overline{y}')$,
and \emph{inclusion dependencies}, which are sentence of the form $\forall \overline{x}, \overline{y}.\,( R(\overline{x}, \overline{y}) \Rightarrow \exists \overline{z}.\, S(\overline{x}, \overline{z}) )$.

\subsection{Views}

Let $D$ be a schema.
A \emph{view $V$ over $D$} is a tuple $\langle \mathit{id},\mathit{o}, q,$ $\mathit{m}\rangle$, where $\mathit{id} \in {\cal V}$ is the view identifier, $\mathit{o} \in {\cal U}$ is the view's owner, $q$ is the non-boolean query over $D$ defining the view, and $\mathit{m} \in \{A,O\}$ is the security mode, where  $A$ stands for \emph{activator's privileges} and $O$ stands for \emph{owner's privileges}.
Note that the query $q$ may refer to other views. 
We assume, however, that views have no cyclic dependencies between them.
We denote by ${\cal VIEW}_{D}$ the set of all views over $D$. 
The \emph{materialization of a view $\langle V,\mathit{o}, q, \mathit{m}\rangle$ in a state $s$}, denoted by $s(V)$, is $[q]^{s}$.
We extend the relational calculus in the standard way to work with views~\cite{abiteboul1995foundations}.

\subsection{Access Control Policies}

We now formalize the SQL access control model. 
We first formalize five privileges.
Let $D$ be a database schema.
A \emph{\texttt{SELECT} privilege over $D$} is a tuple $\langle \texttt{SELECT}, R\rangle$, where $R$ is a relation schema in $D$ or a view over $D$.
A \emph{\texttt{CREATE VIEW} privilege over $D$} is a tuple $\langle \texttt{CREATE VIEW}\rangle$.
An \emph{\texttt{INSERT} privilege over $D$} is a tuple $\langle \texttt{INSERT}, R \rangle$,  a \emph{\texttt{DELETE} privilege over $D$} is a tuple $\langle \texttt{DELETE}, R \rangle$, and a \emph{\texttt{CREATE TRIGGER} privilege over $D$} is a tuple  $\langle \texttt{CREATE TRIGGER}, R\rangle$, where $R$ is a relation schema in $D$.
We denote by ${\cal PRIV}_{D}$ the set of privileges over  $D$.

Following SQL, we use \texttt{GRANT} commands to assign privileges to users.  
Let $U \subseteq {\cal U}$ be a set of users and $D$ be a data\-base schema.
We now define $(U,D)$-grants and $(U,D)$-revokes.
There are two types of \emph{$(U,D)$\emph{-grants}}.
A $(U,D)$\emph{-simple grant} is a tuple $\langle \oplus, u,p,u'\rangle$, where $u \in U$ is the user receiving the privilege $p \in {\cal PRIV}_{D}$  and $u' \in U$ is the user granting this privilege.
A $(U,D)$\emph{-grant with grant option} is a tuple $\langle \oplus^{*}, u,p,u'\rangle$, where $u$, $p$, and $u'$ are as before.
A \emph{$(U,D)$-revoke} is a tuple $\langle \ominus, u, p, u'\rangle$, where $u \in U$ is the user from which the privilege  $p\in {\cal PRIV}_{D}$ will be revoked and  $u' \in U$ is the user revoking this privilege.
We denote by $\Omega^{\mathit{sec}}_{U,D}$ the set of all $(U,D)$-grants and $(U,D)$-revokes.
A grant $\langle \oplus, u,p,u'\rangle$ models the command \texttt{GRANT} $p$ \texttt{TO} $u$ issued by $u'$, a grant with grant option $\langle \oplus^{*}, u,p,u'\rangle$ models the command \texttt{GRANT} $p$ \texttt{TO} $u$ \texttt{WITH GRANT OPTION} issued by $u'$, and a revoke $\langle \ominus, u, p, u'\rangle$ models the command \texttt{REVOKE} $p$ \texttt{FROM} $u$ \texttt{CASCADE} issued by $u'$.

Finally, we define a $(U,D)$\emph{-access control policy} $S$ as a finite set of $(U,D)$-grants. 
We denote by ${\cal S}_{U,D}$ the set of all $(U,D)$-policies.

\begin{example}\label{example:attack:five}
Consider the policy described in Attack~\ref{example:trigger:attack:owner}.
The database $D$ has three tables: $N$, $\mathit{P}$, and $\mathit{T}$.
The set $U$ is $\{u,\mathit{admin}\}$ and the policy $S$ contains the following grants: $\langle \oplus,$$u,$$\langle \texttt{SELECT},$$\mathit{P}\rangle,$$\mathit{admin}\rangle$, $\langle \oplus,$$u,$$\langle \texttt{INSERT},$$\mathit{P}\rangle,$$ \mathit{admin}\rangle$, $\langle \oplus,$$ u,$$ \langle \texttt{DELETE},$$ \mathit{P}\rangle,$$ \mathit{admin}\rangle$,  $ \langle \oplus,$$  u,$$   \langle \texttt{SELECT},$$ \mathit{N}\rangle,$$  \mathit{admin}\rangle$, $\langle \oplus,$$  u,$\onlyTechReport{\\}$ \langle \texttt{INSERT},$$ \mathit{N}\rangle,$$ \mathit{admin}\rangle$, and  $\langle \oplus,$$ u,$$  \langle \texttt{DELETE},$$ \mathit{N}\rangle,$$ \mathit{admin}\rangle$.
\end{example}

\subsection{Triggers}
Let $D$ be a database schema.
A \emph{trigger over $D$} is a tuple $\langle \mathit{id}, u,  \mathit{e}, R, \phi, \mathit{a}, \mathit{m}\rangle$, where $\mathit{id} \in {\cal T}$ is the trigger identifier, $u \in {\cal U}$ is the trigger's owner, $\mathit{e} \in \{ \mathit{INS}, \mathit{DEL}\}$ is the trigger event (where $\mathit{INS}$ stands for $\mathtt{INSERT}$ and $\mathit{DEL}$ stands for $\mathtt{DELETE}$), $R \in D$ is a relation schema, the trigger condition $\phi$ is a relational calculus formula such that $\mathit{free}(\phi) \subseteq \{x_{1}, \ldots, x_{|R|}\}$, and the trigger action $\mathit{a}$ is one of:
(1) $\langle \mathtt{INSERT}, R', \overline{t} \rangle$, where $R' \in D$ and $\overline{t}$ is a $|R'|$-tuple of values in $\mathbf{dom}$ and variables in $\{x_{1}, \ldots, x_{|R|}\}$,
(2) $\langle \mathtt{DELETE}, R', \overline{t}\rangle$, where $R'$ and $\overline{t}$ are as before, or
(3) $\langle \mathit{op}, u, p \rangle$, where $\mathit{op} \in \{\oplus, \oplus^{*}, \ominus\}$, $u \in {\cal U}$, and $p$ is a privilege over $D$.
Finally,  $\mathit{m} \in \{A,O\}$ is the security mode, where  $A$ stands for \emph{activator's privileges} and $O$ stands for \emph{owner's privileges}. 
We denote by ${\cal TRIGGER}_{D}$ the set of all triggers over $D$.

We assume that any command $a$ is executed atomically together with all the triggers activated by $a$.
We also assume that triggers do not recursively activate other triggers.
Hence all executions terminate.
We enforce this condition syntactically at the trigger's creation time; see~\techReportAppendix{app:lts}~for additional details.
The trigger $\langle \mathit{t},\mathit{admin},  \mathit{INS}, \mathit{P}, \mathit{T}(x_{1}),\onlyTechReport{\\} \langle \mathtt{INSERT},  \mathit{N}, x_{1} \rangle,  O\rangle$ models the trigger in Attack~\ref{example:trigger:attack:owner}. 
Here, $x_{1}$ is bound, at run-time, to the value inserted in $P$ by the trigger's invoker.

%% file: system.tex
\section{System and Attacker Model}\label{sect:syst:attk:model}
We next present our system and attacker models.
Executable versions of these models, built in the Maude framework~\cite{clavel2003maude}, are available at~\cite{prototype}.
The models can be used for simulating the execution of our operational semantics, as well as computing the information that an attacker can infer from the system's behaviour.
We have executed and validated all of our examples using these models.

\subsection{Overview}

In our system model, shown in Figure \ref{fig:systemModel}, users interact with two components: a data\-base system and an access control system.
The access control system contains both a policy enforcement point and a policy decision point.
We assume that all the communication between the users and the components is over secure channels.

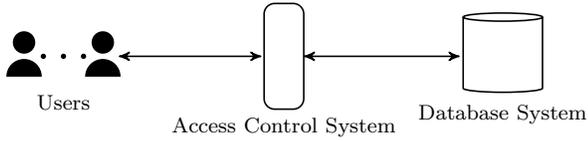
\begin{figure}
\centering

\begin{tikzpicture}[->,>=stealth',shorten >=1pt,auto, semithick]

\node[fill=none,draw=black, shape = rectangle, rounded corners, inner sep=0pt, outer sep=0pt, minimum height = 40pt, minimum width = 15pt]  (pdp) at (0,0) {};
   \node[below = 0pt of pdp] (pdpLabel) {\small{Access Control System}};

\node[shape=circle,fill=none,inner sep=0pt, minimum size=2pt, anchor=south, outer sep=0pt, left = 60pt of pdp](hiddenNeck1)  {}; 
\node[shape=circle,fill=none,inner sep=0pt, minimum size=2pt, anchor=south, outer sep=0pt, left = 55pt of pdp](hiddenNeck3)  {}; 
\node[shape=semicircle,fill=black,inner sep=3pt, anchor=south, outer sep=0pt, below =0pt of hiddenNeck1](body1)  {}; 
\node[shape=circle,fill=black,inner sep=3pt, anchor=south, outer sep=0pt, above=0pt of hiddenNeck1](head1)  {};
\node[fill = black, shape = circle,inner sep=0pt,minimum size=2pt,  left = 5pt of hiddenNeck1] (dot1) {};

\node[fill = black, shape = circle,inner sep=0pt,minimum size=2pt,  left = 5pt of dot1] (dot2) {};

\node[fill = black, shape = circle,inner sep=0pt,minimum size=2pt,  left = 5pt of dot2] (dot3) {};
\node[shape=circle,fill=none,inner sep=0pt, minimum size=2pt, anchor=south, outer sep=0pt, left = 5pt of dot3](hiddenNeck2)  {};
\node[shape=semicircle,fill=black,inner sep=3pt, anchor=south, outer sep=0pt, below = 0pt of hiddenNeck2](body2)  {}; 
\node[shape=circle,fill=black,inner sep=3pt, anchor=south, outer sep=0pt, above=0pt of hiddenNeck2](head2)  {};

   \node[below = 10pt of dot2] (userLabel) {{\small{Users}}};

\node[shape = cylinder, shape border rotate=90, aspect = 0.5, fill=none,draw=black,  minimum height = 30pt, minimum width = 30pt, right = 60pt of pdp ]  (db) {};
   \node[below = 1pt of db] (databaseLabel) {\small{Database System}};

\draw[<->, black] (hiddenNeck3) -- (pdp);
\draw[<->, black] (pdp) --  (db);
\end{tikzpicture}
\caption{System model.}
\label{fig:systemModel}
\end{figure}

\smallskip
\noindent {\bf Database System.}
The database system (or database for short) manages the data. 
The data\-base's state is represented by a mapping from relation schemas to sets of tuples. 
We assume that all data\-base~operations~are~atomic.

\smallskip
\noindent {\bf Users.} 
Users interact with the database  where each command is checked by the access control system.
Each user has a unique account through which he can  issue \texttt{SELECT}, \texttt{INSERT}, \texttt{DELETE}, \texttt{GRANT}, \texttt{REVOKE}, \texttt{CREATE TRIGGER}, and \texttt{CREATE VIEW} commands.

The \emph{system administrator} is a distinguished user responsible for defining the database schema and the access control policy.
In addition to issuing queries and commands, he can  create user accounts and assign them to users. 
The administrator interacts with the access control system through a special account $\mathit{admin}$. 

The \emph{attacker} is a user, other than the administrator, with an assigned user account who attempts to violate the access control policy. 
Namely, his goals are:\begin{inparaenum}[(1)]
\item to read or infer data from the database for which he lacks the necessary
\texttt{SELECT} privileges, and
\item  to alter the system state in unauthorized ways, e.g., changing data in relations for which he lacks the
necessary \texttt{INSERT} and \texttt{DELETE} privileges.
\end{inparaenum}
The attacker can issue any command available to users and he sees the results of his commands.
The attacker's inference capabilities are specified using deduction rules. 

\smallskip
\noindent {\bf Access Control System.} 
The access control system protects the confidentiality and integrity of the data in the database. 
It is configured with an access control policy $S$, it intercepts all commands issued by the users, and it prevents the execution of commands that are not authorized by $S$.
When a user $u$ issues a command $c$, the access control system decides whether $u$ is authorized to execute $c$.
If $c$ complies with the policy, then the access control system forwards the command to the DBMS, which executes $c$ and returns its result to $u$. 
Otherwise, it raises a \emph{security exception} and rejects $c$.
Note that this corresponds to the Non-Truman model \cite{rizvi2004extending}; see related work for more details.

The access control system also logs all issued commands. 
When evaluating a command, the access control system can access the data\-base's current state and the log.

%% file: lts_short.tex
\subsection{System Model}\label{sect:lts}
We formalize our system model as a labelled transition system (LTS).
First, we define a system configuration, which describes the data\-base schema and the integrity constraints, and the user actions.
Afterwards, we define the  system's state, which represents a snapshot of the system that contains the data\-base's state, the identifiers of the users interacting with the system, the access control policy, and the current triggers and views in the system.
Finally, we formalize the system's behaviour as a small step operational semantics, including all features necessary to reason about security, even in the presence of attacks like those illustrated in \S\ref{sect:motivating}.

A \emph{system configuration} is a tuple $\langle D,\Gamma\rangle$ such that $D$ is a schema and $\Gamma$ is a finite set of integrity constraints over $D$. 
Let $M = \langle D,\Gamma\rangle$ be a system configuration and $u \in {\cal U}$ be a user.
A \emph{$(D,u)$-action} is one of the following tuples:
\begin{compactitem}
\item $\langle u, \mathtt{ADD\_USER}, u'\rangle$, where $u \onlyTechReport{\hspace{-0.5mm} } = \onlyTechReport{\hspace{-0.5mm} } \admin$ and $u'\onlyTechReport{\hspace{-0.5mm} } \in {\cal U}\setminus  \{\mathit{admin}\}$,

\item $\langle u, \mathtt{SELECT}, q\rangle$, where $q$ is a boolean query\footnote{
Without loss of generality, we focus only on boolean queries \cite{abiteboul1995foundations}.
We can support non-boolean queries  as follows.
Given a database state $\mathit{s}$ and a query $q:=\{ \overline{x}\ |\ \phi\}$, if the access control mechanism authorizes the boolean query $\bigwedge_{\overline{t} \in [q]^{\mathit{s}}} \phi[\overline{x} \mapsto \overline{t}] \wedge (\forall \overline{x}.\, \phi \Rightarrow \bigvee_{\overline{t} \in [q]^{\mathit{s}}} \overline{x} = \overline{t})$,   then we return  $q$'s result, and otherwise we reject  $q$ as unauthorized.} over $D$,

\item $\langle u, \mathtt{INSERT}, R, \overline{t}\rangle$, where $R \in D$ and $\overline{t} \in \mathbf{dom}^{|R|}$,

\item $\langle u, \mathtt{DELETE}, R, \overline{t}\rangle$, where $R$ and $\overline{t}$ are as above, 

\item $\langle \mathit{op}, u',p,u\rangle$, where  $\langle \mathit{op}, u',p,u\rangle \in \Omega^{\mathit{sec}}_{D,{\cal U}}$, or

\item $\langle u,\mathtt{CREATE}, o\rangle$, where  $o \in {\cal TRIGGER}_{D} \cup {\cal VIEW}_{D}$.
\end{compactitem}
We denote by ${\cal A}_{D,u}$ the set of all $(D,u)$-actions and by ${\cal A}_{D,U}$, for some $U \subseteq {\cal U}$, the set $\bigcup_{u \in U} {\cal A}_{D,u}$.

An \emph{$M$-context} describes the system's history, 
the scheduled triggers that must be executed, 
and how to modify the system's state in case a roll-back occurs.
We denote by ${\cal C}_{M}$ the set of all $M$-contexts.
We assume that ${\cal C}_{M}$ contains a distinguished element $\epsilon$ representing the empty context, which is the context in which the system starts.
\onlyTechReport{Contexts~are formalized in \techReportAppendix{app:lts}.}

An $M$\emph{-state} is a tuple $\langle \mathit{db}, U, \mathit{sec}, T,V, c\rangle$ such that $\mathit{db} \in \Omega_{D}^{\Gamma}$ is a data\-base state, $U \subset {\cal U}$ is a finite set of users such that $\mathit{admin}\in U$, $\mathit{sec} \in  {\cal S}_{U,D}$ is a security policy, $T$ is a finite set of triggers over $D$ owned by users in $U$, $V$ is a finite set of views over $D$ owned by users in $U$, and $c \in {\cal C}_{M}$ is an $M$-context.
We denote by $\Omega_{M}$ the set of all $M$-states.
An $M$-state $\langle \mathit{db}, U, \mathit{sec}, T,V, c\rangle$ is \emph{initial} iff\begin{inparaenum}[(a)]
\item $\mathit{sec}$ contains only grants issued by $\mathit{admin}$,
\item $T$ (respectively  $V$) contains only triggers (respectively views) owned by $\mathit{admin}$, and 
\item $c = \epsilon$.
\end{inparaenum} 
We denote by ${\cal I}_{M}$ the set of all initial states.

An \emph{$M$-\accessControlFunction{}} ($M$-PDP) is a total function $f: \Omega_{M} \times  {\cal A}_{D,{\cal U}} \rightarrow \{\top,\bot\}$ that maps each state $s$ and action $a$ to an access control decision represented by a boolean value, where $\top$ stands for permit and $\bot$ stands for deny.
An \emph{\accessControlConfiguration{}}  is a tuple $\langle M, f\rangle$, where $M$ is a system configuration and $f$ is an $M$-\acf{}.

\begin{figure*}

\centering

\begin{minipage}{.49\textwidth}
    \begin{center}
    \scalebox{.85}[.85]{
        \begin{tabular}{c }
$\infer[\text{\begin{tabular}{c} \texttt{SELECT}\\ Success\end{tabular}}]
{s \xrightarrow{\langle u, \mathtt{SELECT},  q\rangle}_{f} s'}
{
\hfill s = \langle \mathit{db},\mathit{sec}, U, T,V,c \rangle \hfill  \enspace  
 \hfill  \mathit{f}(s, \langle u, \mathtt{SELECT},  q\rangle) = \top \hfill \enspace
 \hfill \mathit{trg}(s) = \epsilon \hfill \\
 \hfill s' = \langle \mathit{db},\mathit{sec}, U, T,V, c' \rangle \hfill \quad c' = \mathit{upd}(s,  \langle u, \mathtt{SELECT},  q\rangle)  \hfill}$
\\\\
$\infer[\text{\begin{tabular}{c} \texttt{INSERT}\\ Success\end{tabular}}]
{s \xrightarrow{\langle u, \mathtt{INSERT}, R, \overline{t}\rangle}_{f} s'}
{
\hfill s = \langle \mathit{db},\mathit{sec}, U, T,V,c \rangle \hfill  \quad 
\hfill \mathit{f}(s, \langle u, \mathtt{INSERT}, R, \overline{t}\rangle) = \top  \hfill \\
\hfill s' = \langle \mathit{db}[R \oplus \overline{t}],\mathit{sec}, U, T,V, c' \rangle \hfill \quad
\hfill \mathit{db} [R \oplus \overline{t}] \in \Omega_{D}^{\Gamma} \hfill \\
\hfill c' = \mathit{upd}(s,\langle u, \mathtt{INSERT}, R, \overline{t}\rangle)  \hfill \quad
\hfill \mathit{trg}(s) = \epsilon \hfill
}$
\end{tabular}
}
    \end{center}
\end{minipage}
\hspace{.12cm}
\begin{minipage}{.49\textwidth}
   \begin{center}
    \scalebox{.85}[.85]{
\begin{tabular}{c }
$\infer[\text{\begin{tabular}{c} Trigger \\\texttt{INSERT}\\ Success\end{tabular}}]
{\hfill s \xrightarrow{\mathit{trg}(s)}_{f} s' \hfill }
{
\hfill s = \langle \mathit{db},\mathit{sec}, U, T,V,c \rangle \hfill  \quad
\hfill \overline{v} = \mathit{tpl}(s) \hfill \\
\hfill u= \mathit{user}(m,\mathit{owner},\mathit{invoker}(s)) \hfill \\
\hfill \mathit{trg}(s) = \langle \mathit{id},\mathit{owner},  \mathit{ev}, R', \phi, \mathit{st}, m\rangle \hfill  \\
\hfill \mathit{f}(s, \langle u, \mathtt{SELECT}, \phi[\overline{x} \mapsto \overline{v}]\rangle) = \top \hfill \quad
\hfill [ \phi[\overline{x} \mapsto \overline{v}]]^{\mathit{db}} = \top \hfill \\
\hfill \langle u , \mathtt{INSERT}, R, \overline{v}' \rangle = \mathit{act}(\mathit{st}, u, \overline{v})  \hfill \\
\hfill \mathit{f}(s, \langle u , \mathtt{INSERT}, R, \overline{v}' \rangle) = \top \hfill \quad
\hfill c' = \mathit{upd}(s,\mathit{trg}(s))  \hfill  \\
\hfill s' = \langle \mathit{db}[R \oplus \overline{v}'],\mathit{sec}, U, T,V, c' \rangle \hfill \quad
\hfill \mathit{db} [R \oplus \overline{v}'] \in \Omega_{D}^{\Gamma} \hfill 
}$
\end{tabular}
}
    \end{center}
\end{minipage}

\caption{Examples of system model's rules.}\label{table:rules:example}
\end{figure*}

We now define the LTS representing the system model.

\begin{definitionInt}
Let $\mathit{P}=\langle M, f\rangle$ be an \accessControlConfiguration{}, where $M = \langle D,\Gamma\rangle$ and $f$ is an $M$-\acf{}.
The \emph{$\mathit{P}$-LTS} is the labelled transition system $\langle S, A, \rightarrow_{f}, I\rangle$ where
$S = \Omega_{M}$ is the set of states, $\mathit{A} = {\cal A}_{D,{\cal U}} \cup {\cal TRIGGER}_{D}$ is  the set of actions, $\rightarrow_{f} \; \subseteq S \times A \times S$ is the transition relation, and   $I = {\cal I}_{M}$  is the set of initial states. $\square$
\end{definitionInt}

Let $\mathit{P}=\langle M, f\rangle$ be an \accessControlConfiguration{}.
A \emph{run} $r$ of a $P$-LTS $L$ is a finite alternating sequence of states and actions, which starts with an initial state $s$, ends in some state $s'$, and respects the transition relation $\rightarrow_{f}$. 
We denote by $\mathit{traces}(L)$ the set of all  $L$'s runs.
Given a run $r$, $|r|$ denotes  the number of states in $r$,  $\mathit{last}(r)$ denotes $r$'s last state, 
\onlyShortVersion{and $r^{i}$, where $1 \leq i \leq |r|$, denotes the run obtained by truncating $r$ at the $i$-th state.}
\onlyTechReport{and $r^{i}$, where $1 \leq i \leq |r|$, denotes the run obtained by truncating $r$ at the $i$-th state.} 

The relation $\rightarrow_{f}$ formalizes the system's  small step operational semantics.
Figure \ref{table:rules:example} shows three rules describing the successful execution  of \texttt{SELECT} and \texttt{INSERT} commands, as well as triggers.
In the rules, we represent context changes using the update  function $\mathit{upd}$, which takes as input an $M$-state and an action $a \in {\cal A}_{D,{\cal U}} \cup {\cal TRIGGER}_D$, and returns the updated context. 
This function, for instance, updates the system's history stored in the context.
The  function $\mathit{trg}$ takes as input a system state $s$ and returns the first trigger in the list of scheduled triggers stored in $s$'s context. 
If there are no triggers to be executed, then $\mathit{trg}(s) = \epsilon$.
The rule \emph{\texttt{SELECT} Success} models the system's  behaviour when the user $u$ issues a \texttt{SELECT} query $q$ that is authorized by the \acf{} $f$.
The only component of the $M$-state $s$ that changes is the context $c$.
Namely, $c'$ is obtained from $c$ by updating the history and storing $q$'s result. 
Similarly, the rule \emph{\texttt{INSERT} Success} describes how the system behaves after a successful \texttt{INSERT} command, i.e., one that neither violates the integrity constraints nor causes security exceptions.
The database state $\mathit{db}$ is updated by adding the tuple $\overline{t}$ to $R$ and the context is updated from $c$ to $c'$ by\begin{inparaenum}[(a)]
\item storing the action's result,
\item storing the triggers that must be executed in response to the \texttt{INSERT} event, and 
\item keeping track of the previous state in case a roll-back is needed.
\end{inparaenum}

The \emph{Trigger \texttt{INSERT} Success} rule describes how the system executes a trigger whose action is an \texttt{INSERT}.
The system extracts from the context the trigger $t$ to be executed, i.e., $t = \mathit{trg}(s)$. 
It determines, using the function $\mathit{user}$, the user $u$ under whose privileges the trigger $t$ is executed, which is, depending on $t$'s security mode, either the invoker $\mathit{invoker}(s)$ or $t$'s owner.
It then checks that $u$ is authorized to execute the \texttt{SELECT} statement associated with $t$'s \texttt{WHEN} condition, and that this condition is satisfied.
Afterwards, it computes the actual action using the function $\mathit{act}$, which instantiates the free variables in $t$'s definition with the values in the tuple $\mathit{tpl}(s)$, i.e., the tuple associated with the action that fired $t$.
Finally, the system updates the database state $\mathit{db}$ by adding the tuple $\overline{v}'$ to $R$ and the context by
 storing the results of $t$'s execution and
 removing $t$ from the list of scheduled triggers.

In \techReportAppendix{app:lts}, we give the complete formalization of our labelled transition system.
This includes formalizing contexts and  all the rules defining the transition relation $\rightarrow_{f}$. 
Our operational semantics can be tailored to model the behaviour of specific DBMSs.
Thus, using our executable model, available at~\cite{prototype}, it is possible to validate our operational semantics against different existing DBMSs.

%% file: advModel_short.tex
\subsection{Attacker Model}\label{sect:adv:model:concrete}
\begin{figure*}[!hbtp]
\centering

\scalebox{.95}[1]{
\begin{tabular}{c}

$\infer[\begin{tabular}{c}\texttt{DELETE}\\ \text{Success}\end{tabular}]
{r, i \attMod  \neg R(\overline{t})}
{
  \hfill r^{i} = r^{i-1} \concat \langle u, \mathtt{DELETE}, R, \overline{t}\rangle \concat s  \hfill  \quad
  \hfill 1 < i \leq |r| \hfill  \\
	\hfill s \in \Omega_{M} \hfill \quad
  \hfill \mathit{secEx}(s) =\;\bot  \hfill \quad
   \hfill \mathit{Ex}(s) = \emptyset \hfill }$

 $\infer[\begin{tabular}{c}\texttt{SELECT}\\ \text{Success}\end{tabular}]
{r,i \attMod  \phi}
{
 \hfill  r^{i} = r^{i-1} \concat \langle u, \mathtt{SELECT},  \phi\rangle \concat s \hfill  \quad
 \hfill 1 < i \leq |r| \hfill \quad 
 \hfill s \in \Omega_{M} \hfill \\
 \hfill \mathit{secEx}(s) =\;\bot  \hfill \quad
   \hfill \mathit{Ex}(s) = \emptyset \hfill \quad 
 \hfill   \mathit{res}(s) =\top \hfill }$ \\\\

$\infer[\text{\begin{tabular}{c}Learn\\ \texttt{INSERT}\\ Backward\end{tabular}}]
{r, i \attMod   \phi[\overline{x} \mapsto \mathit{tpl}(\mathit{last}(r^{i}))] }
{ 
 \hfill r^{i+1} = r^{i}  \concat t \concat s \hfill \quad 
 \hfill \mathit{invoker}(\mathit{last}(r^{i})) = u  \hfill \quad 
 \hfill     s \in \Omega_{M} \hfill \quad
 \hfill 1 \leq i < |r| \hfill \\
\hfill \mathit{secEx}(s) = \bot \hfill  \quad
 \hfill \mathit{Ex}(s) = \emptyset \hfill   \quad
  \hfill r, i \attMod   \neg \psi   \hfill \quad
 \hfill  r, i+1 \attMod   \psi  \hfill \\
 \hfill t=  \langle \mathit{id},\mathit{ow}, \mathit{ev}, R', \phi(\overline{x}), \langle \mathtt{INSERT}, R, \overline{t}\rangle,m\rangle \hfill \\
}$

$\infer[\text{\begin{tabular}{c} Propagate \\ Backward \\ \texttt{SELECT}\end{tabular}}]
{r, i\attMod  \psi}
{
 \hfill  r^{i+1} = r^{i}  \concat \langle u, \mathtt{SELECT},  \phi\rangle \concat s  \hfill \\
 \hfill r,i+1 \attMod  \psi  \hfill  \quad
 \hfill     s \in \Omega_{M} \hfill \quad
 \hfill 1 \leq i < |r| \hfill 
 }$ \\\\

$\infer[\text{\begin{tabular}{c}Propagate Forward\\ Update Success\end{tabular}}]
{r , i\attMod   \phi}
{
 \hfill r,i-1 \attMod  \phi \hfill  \quad
 \hfill r^{i} = r^{i-1} \concat \langle u, \mathit{op}, R, \overline{t}\rangle \concat s  \hfill \quad
  \hfill     s \in \Omega_{M} \hfill \quad
 \hfill 1 < i \leq |r| \hfill \\
 \hfill \mathit{secEx}(s) = \bot  \hfill \quad
 \hfill  \mathit{Ex}(s) = \emptyset  \hfill   \quad
 \hfill  \mathit{revise}(r^{i-1}, \phi, r^{i}) = \top  \hfill \quad
 \hfill    \mathit{op} \in \{\mathtt{INSERT}, \mathtt{DELETE}\} \hfill } $ 
\end{tabular}
}
\caption{
Example of attacker inference rules, where $r,i \attMod \phi$ denotes that this judgment holds in $\attackerModel$.
}\label{figure:adv:model:rules}
\end{figure*}

We model attackers that interact with the system through SQL commands and infer information from the system's behaviour by exploiting triggers, views, and integrity constraints.
We argue that database access control mechanisms should be secure with respect to such strong attackers, as this reflects how (malicious) users may interact with modern databases.
Furthermore, any mechanism secure against such strong attackers is also secure against weaker attackers.

Any user other than the administrator can be an attacker, and we assume that users do not collude to subvert the system.
Note that our attacker model, the security properties in  \S\ref{sect:sec:prop}, and the mechanism we develop in  \S\ref{sect:enf:alg}, can  easily be extended to support colluding users.
We also assume that an attacker can issue any command available to the system's users, and he knows the system's operational semantics, the database schema, and the integrity constraints.

We assume that an attacker has access to the system's security policy, the set of users, and the definitions of the triggers and views in the system's state.
In more detail, given an $M$-state $\langle \mathit{db}, U, \mathit{sec}, T,V, c\rangle$, an attacker can access $U$, $\mathit{sec}$, $T$, and $V$.
Users interacting with existing DBMSs typically have access to some, although not all, of this information.
For instance, in PostgreSQL a user can read all the information about the triggers defined on the tables for which he has some non-\texttt{SELECT} privileges.
Note that the more information an attacker has, the more attacks he can launch.
Finally, we assume that an attacker knows whether any two of his commands $c$ and $c'$ have been executed consecutively by the system, i.e., if there are commands executed by other users occurring between $c$ and $c'$.
The attacker's knowledge about the sequential execution of his commands is needed to soundly propagate his knowledge about the system's state between his commands.
Since the mechanism we develop in \S\ref{sect:enf:alg} is secure with respect to this attacker, it is also secure with respect to weaker attackers who have less information or cannot detect whether their commands have been executed consecutively.

An attacker model describes what information an attacker knows, how he interacts with the system, and what he learns about the system's data by observing the system's behaviour.
Since every user is a potential attacker, for each user $u \in {\cal U}$ we define an attacker model specifying $u$'s inference capabilities.
To represent $u$'s knowledge, we introduce judgments. 
A judgment is a four-tuple $\langle r,i, u, \phi\rangle$, written $r,i \attMod \phi$, denoting that from the run $r$, which represents the system's behaviour, the user $u$ can infer that $\phi$ holds in the $i$-th state of $r$.
An attacker model for $u$ is thus  a set of judgments associating to each position of each run, the sentences that  $u$ can infer from the system's behaviour.
The idea of representing the attacker's knowledge using sentences $\phi$ is inspired by existing formalisms for Inference Control~\cite{ brodsky2000secure,  farkas2002inference} and Controlled Query Evaluation~\cite{bonatti1995foundations}.

\begin{definition}
Let $P$ be an \accessControlConfiguration{}, $L$ be the $P$-LTS, and $u \in {\cal U}$ be a user.
A \emph{$(P,u)$-judgment} is a tuple $\langle r,i,u,\phi \rangle$, written $r,i \attMod \phi$, where $r \in \mathit{traces}(L)$, $1 \leq i \leq |r|$, and $\phi \in \mathit{RC}_{\mathit{bool}}$. 
A \emph{$(P,u)$-attacker model} is a set of $(P,u)$-judgments. 
A $(P,u)$-judgment $r,i \attMod \phi$ \emph{holds in a $(P,u)$-attacker model $A$} iff $r,i \attMod \phi \in A$.
\end{definition}

For each user $u \in {\cal U}$, we now define the $(P,u)$-attacker model $\attackerModel$ that we use in the rest of the paper.
We formalize this model using a set of inference rules, where $\attackerModel$ is the smallest set of judgments satisfying the inference rules.
Figure \ref{figure:adv:model:rules} shows five representative rules. 
The complete formalization of all rules is given in \techReportAppendix{app:adv:model}.
In the following, when we say that a judgment $r,i \attMod \phi$ holds, we always mean with respect to the attacker model $\attackerModel$.

Note that $\attackerModel$ is sound with respect to the $\mathit{RC}$~semantics, i.e., if $r, i \attMod \phi$ holds, then the formula $\phi$ holds in the $i$-th state of $r$.
Intuitively, $\attackerModel$ models how $u$~infers information from the system's behaviour, namely\begin{inparaenum}[(a)]
\item how $u$  learns information from his commands and their results,
\item how $u$ learns information from triggers, their execution, their interleavings, and their side effects,
\item how $u$ propagates his knowledge along a run, and
\item how $u$ learns information from  exceptions caused  by either integrity constraint violations or security violations.
\end{inparaenum}
This model is substantially more powerful than the \texttt{SELECT}-only attacker~model.

The rules \emph{\texttt{DELETE} Success} and \emph{\texttt{SELECT} Success} describe how the user $u$ infers information from his successful actions, i.e., those actions that generate neither security exceptions nor integrity violations.
In the rules,  $\mathit{secEx}(s) = \bot$  denotes that there were no security exceptions caused by the action leading to $s$, and $\mathit{Ex}(s) = \emptyset$ denotes that the action leading to $s$ has not violated the integrity constraints.
After a successful \texttt{DELETE}, $u$ knows that the deleted tuple is no longer in the database, and after a successful \texttt{SELECT} he learns the query's result, denoted by $\mathit{res}(s)$.

The rules \emph{Propagate Backward \texttt{SELECT}} and \emph{Propagate Forward Update Success} describe how $u$ propagates information along the run.
\emph{Propagate Backward \texttt{SELECT}} states that if the user $u$ knows that $\phi$ holds after a \texttt{SELECT} command, then he knows that $\phi$ also holds just before the \texttt{SELECT} command because \texttt{SELECT} commands do not modify the database state.
\emph{Propagate Forward Update Success} states that if $u$ knows that $\phi$ holds before a successful \texttt{INSERT} or \texttt{DELETE} command and he can determine that the command's execution does not influence $\phi$'s truth value, denoted by $\mathit{revise}(r^{i-1},\phi,r^{i}) = \top$, then he also knows that $\phi$ holds after the command. 
\onlyTechReport{The function $\mathit{revise}$ is formalized in \techReportAppendix{app:adv:model}.} 
 
Finally, the rule \emph{Learn \texttt{INSERT} Backward} models $u$'s reasoning when he activates a trigger that successfully inserts a tuple in the database. 
If $u$ knows that immediately before the trigger the formula $\psi$ does not hold and immediately after the trigger the formula $\psi$ holds, then the trigger's execution is the cause of the database state's change.
Therefore, $u$ can infer that the trigger's condition $\phi$ holds just before the trigger's execution. 
Note that $\mathit{invoker}(s)$ denotes the user who fired the trigger that is executed in the state $s$, whereas  $\mathit{tpl}(s)$  denotes the tuple associated with the action that fired the trigger that is executed in the state $s$.

\begin{figure*}[!hbtp]
\centering
$
	\infer[\text{\begin{tabular}{c}Learn \texttt{INSERT} Backward\end{tabular}}]
			{r,3 \attMod  \mathit{T}(v)}
			{\vspace{5pt}
					\infer[\text{\begin{tabular}{c}Propagate Forward \\Update Success\end{tabular}}]
					{r,3 \attMod \neg \mathit{N}(v)}
					{\vspace{5pt}
						\infer[\text{\begin{tabular}{c}\texttt{DELETE} Success\end{tabular}}]
						{r,2 \attMod   \neg \mathit{N}(v)}
						{}
					}
				 \qquad
					\infer[\text{\begin{tabular}{c}Propagate Backward \texttt{SELECT}\end{tabular}}]
					{r,4 \attMod  \mathit{N}(v) }
					{\vspace{5pt}
						\infer[\text{\begin{tabular}{c}\texttt{SELECT} Success\end{tabular}}]
						{r,5 \attMod  \mathit{N}(v)}		
						{}
					}						
			}
$
\caption{Template Derivation of Attack \ref{example:trigger:attack:owner} (contains just selected subgoals)}\label{deriv:example:trigger:owner}
\end{figure*}

\begin{example}\label{example:derivation}
Let the schema, the set of users $U$,~and the policy $S$ be as in Example \ref{example:attack:five}.
The database state $\mathit{db}$  is $\mathit{db}(N)\,=\,\{v\}$, $\mathit{db}(\mathit{P}) =\emptyset$,~and $\mathit{db}(\mathit{T})$ $=\{v\}$. 
The only trigger in the system is $t = \langle \mathit{id},\mathit{admin},\mathit{INS},$ $ \mathit{P}, \mathit{T}(x_{1}), \langle \mathtt{INSERT}, \onlyTechReport{\\} \mathit{N}, x_{1} \rangle,  O\rangle$.
The run $r$ is as follows: 
\begin{compactenum}
\item $u$  deletes $v$ from $N$.
\item $u$ inserts $v$ in $\mathit{P}$. This activates the trigger $t$, which inserts $v$ in $N$. 
\item $u$ issues the \texttt{SELECT} query $N(v)$.
\end{compactenum}

We used Maude to generate the following run, which illustrates how the system's state changes. 
Note that there are no exceptions during the run.

\noindent
\begin{tikzpicture}[->,>=stealth',shorten >=1pt,auto, semithick]
  \tikzstyle{every state}=[rectangle, rounded corners,fill=none,draw=black,text=black,minimum height=1em,
           inner sep=1pt, ultra thin]

  \node[state] 		   (A) at (0,0)  	{ \onlyShortVersion{\footnotesize}\onlyTechReport{\small} {$\langle \mathit{db}, U,S,  \{t\}, \emptyset, c_{1}\rangle$}};
  \node[state]         (B) at (5,0) 	{ \onlyShortVersion{\footnotesize}\onlyTechReport{\small} {$\langle \mathit{db}[N \ominus {v}], U,S,  \{t\}, \emptyset, c_{2}\rangle$}};
  \node[state]         (C) at (5,-0.75) 	{ \onlyShortVersion{\footnotesize}\onlyTechReport{\small} {$\langle \mathit{db}[P \oplus {v}, N \ominus {v}], U, S,  \{t\}, \emptyset, c_{3}\rangle$}};
  \node[state]         (D) at (0.5,-0.75)  	{ \onlyShortVersion{\footnotesize}\onlyTechReport{\small} {$\langle \mathit{db}[P \oplus {v}], U, S,  \{t\}, \emptyset, c_{4} \rangle$}};
  \node[state]         (E) at (0.5,-1.5)  	{ \onlyShortVersion{\footnotesize}\onlyTechReport{\small} {$\langle \mathit{db}[P \oplus {v}], U, S,  \{t\}, \emptyset, c_{5} \rangle$}};

  \path (A) edge              node { \onlyShortVersion{\footnotesize}\onlyTechReport{\small} {$\langle u, \mathtt{DELETE}, N, {v} \rangle$}} (B)
        (B) edge		[left]	  node { \onlyShortVersion{\footnotesize}\onlyTechReport{\small} {$\langle u, \mathtt{INSERT}, \mathit{P}, {v} \rangle$}} (C)
        (C) edge        [above]      node { \onlyShortVersion{\footnotesize}\onlyTechReport{\small} {$t$}} (D)
        (D) edge  		      node { \onlyShortVersion{\footnotesize}\onlyTechReport{\small} {$\langle u, \mathtt{SELECT}, N(v) \rangle$}} (E);
\end{tikzpicture}

Figure \ref{deriv:example:trigger:owner} models $u$'s reasoning in Attack \ref{example:trigger:attack:owner}.
The user $u$ first applies the  \emph{\texttt{SELECT} Success} rule to derive $r,5 \attMod  \mathit{N}(v)$, i.e., he learns the query's result. By applying the rule \emph{Propagate Backward \texttt{SELECT}} to $r,5 \attMod  \mathit{N}(v)$,  he obtains $r,4 \attMod  \mathit{N}(v)$, i.e., he learns that $\mathit{N}(v)$ holds before the \texttt{SELECT} query.
Similarly, he applies the rule \emph{\texttt{DELETE} Success} to derive  $r,2 \attMod  \neg \mathit{N}(v)$, and he obtains $r,3 \attMod  \neg \mathit{N}(v)$ by applying the \emph{\text{Propagate Forward Update Success}} rule.
Finally, by applying the rule \emph{Learn \texttt{INSERT} Backward} to $r,3 \attMod  \neg \mathit{N}(v)$ and $r,4 \attMod  \mathit{N}(v)$, he learns the value of the trigger's \texttt{WHEN} condition $r,3 \attMod  \mathit{T}(v)$.
Since the user $u$ should not be able to learn information about $\mathit{T}$, the attack violates the intended confidentiality guarantees. 
We used our executable attacker model~\cite{prototype} to derive  the judgments.
\end{example}

%% file: securityDef.tex
\section{Security Properties}\label{sect:sec:prop}

Here we define two security properties: \correctness{} and  \confidentiality{}. 
These properties capture the two essential aspects of database security.
Database integrity states that all actions modifying the system's state are authorized by the system's policy.
In contrast, \confidentiality{} states that all information that an attacker can learn by observing the system's behaviour is authorized.

These two properties formalize security guarantees with respect to the two different classes of attacks previously identified.  
An access control mechanism providing \correctness{} prevents non-authorized changes to the system's state and, thereby, prevents integrity attacks.
Similarly, by preventing the leakage of sensitive data, a mechanism providing \confidentiality{} prevents confidentiality attacks.

\onlyTechReport{\newpage}
\subsection{\Correctness{}}\label{sect:sec:prop:correctness}

Database integrity requires a formalization of authorized actions.
We therefore define the relation $\auth$ between states and actions, modelling which actions are authorized~in a given  state. 
Let $P = \langle M, f\rangle$ be an \accessControlConfiguration{}, where $M = \langle D,\Gamma\rangle$ and $f$ is an $M$-\acf{}.
The~relation $\auth \subseteq \Omega_{M} \times ({\cal A}_{D,{\cal U}} \cup {\cal TRIGGER}_{D})$~is~defined~by~a~set~of~rules~given in \techReportAppendix{app:eop}. 
Figure~\ref{figure:eop:auth:example} shows three representative rules. 
The  \emph{\texttt{GRANT}} rule says that the owner $o$ of a view $v$ with owner's privileges  is authorized to delegate the \texttt{SELECT} privilege over $v$ to a user $u$ in the state $s$, if $o$ has the \texttt{SELECT} privilege with grant option over a set of tables and views that determine $v$'s materialization \cite{nash2010views}, denoted by $\mathit{hasAccess}(s, v, o, \oplus^{*})$.
The  \emph{\texttt{TRIGGER}} rule says that the execution of an enabled trigger, i.e., one whose \texttt{WHEN} condition is satisfied, with the activator's privileges is authorized if both the invoker and the trigger's owner are authorized to execute the  trigger's action according to $\auth$. 
Note that the $\mathit{act}$ function instantiates the action given in the trigger's definition to a concrete action by identifying the user performing the action and replacing the free variables with values from ${\bf dom}$. 
Finally, the \emph{\texttt{REVOKE}} rule says that a \texttt{REVOKE} statement is authorized if the resulting state,
obtained using the function $\mathit{apply}$, has a consistent policy, namely one in which all the \texttt{GRANT}s are 
authorized~by~$\auth$.

\begin{figure}
\centering
\scalebox{.90}[1]{
\begin{tabular}{c}

$
\infer[\texttt{GRANT}]
{ s \auth \langle \mathit{op}, u, \mathit{priv}, \mathit{o} \rangle
}
{
\hfill s = \langle\mathit{db}, U, \mathit{sec}, T,V,c\rangle \hfill \quad
\hfill u, \mathit{o} \in U \hfill \quad 
\hfill \mathit{op} \in \{\oplus,\oplus^{*}\} \hfill \\ 
\hfill \mathit{priv} = \langle \texttt{SELECT}, v\rangle \hfill \quad
\hfill v = \langle \mathit{id}, \mathit{o}, q, O \rangle \hfill \quad
\hfill v \in V \hfill \\
\hfill \mathit{hasAccess}(s ,v, o, \oplus^{*}) \hfill
}$
\\\\

$
\infer[\texttt{TRIGGER}]
{ s \auth t
}
{
\hfill s = \langle\mathit{db}, U, \mathit{sec}, T,V,c  \rangle \hfill  \quad
\hfill t = \langle \mathit{id},\mathit{ow},  \mathit{ev}, R, \phi, \mathit{st}, A\rangle \hfill \\  
\hfill [\phi[\overline{x} \mapsto \mathit{tpl}(s)]]^{\mathit{db}} = \top \hfill  \quad
\hfill s \auth \mathit{act}(\mathit{st},\mathit{ow},\mathit{tpl}(s)) \hfill \\
\hfill s \auth \mathit{act}(\mathit{st},\mathit{invoker}(s),\mathit{tpl}(s)) \hfill \quad
\hfill t \in T \hfill \\
}
$
\\\\
$
\infer[\texttt{REVOKE}]
{ s \auth \langle \ominus, u, p, u' \rangle
}
{
\hfill s = \langle\mathit{db}, U, \mathit{sec}, T,V,c  \rangle \hfill \quad 
\hfill s' = \langle\mathit{db}, U, \mathit{sec}', T,V,c  \rangle \hfill \\
\hfill s' = \mathit{apply}(\langle \ominus, u, p, u' \rangle, s) \quad
\hfill \forall g \in \mathit{sec}'.\, s' \auth g \hfill
}
$
\end{tabular}
}
\caption{Examples of $\auth$ rules.}\label{figure:eop:auth:example}
\end{figure}

We  now define \correctness{}. 
Intuitively, a \acf{} provides \correctness{}  iff all the actions it authorizes are explicitly authorized by the policy, i.e., they are authorized by $\auth$. 
This notion comes directly from the SQL standard, and it is reflected in existing enforcement mechanisms.
Recall that, given a state $s$,  $\mathit{secEx}(s) = \bot$  denotes that there were no security exceptions caused by the action or trigger leading to $s$.

\begin{definition}
Let $P = \langle M, f\rangle$ be an \accessControlConfiguration{}, where $M = \langle D,\Gamma\rangle$ and $f$ is an $M$-\acf{}, and let $L$ be the $P$-LTS.
We say that  \emph{$f$ provides \correctness{} with respect to $P$} iff for all reachable states $s,s' \in \Omega_{M}$, 
if $s'$ is reachable in one step from $s$ by an action $a \in {\cal A}_{D,{\cal U}} \cup {\cal TRIGGER}_{D}$ and $\mathit{secEx}(s') = \bot$, then $s \auth a$. 
\end{definition}

\begin{example}\label{example:eop:1}
We consider a run corresponding to Attack~\ref{example:trigger:attack:activator}, which illustrates a violation of \correctness{}. 
The data\-base $\mathit{db}$ is such that $\mathit{db}(P) =\emptyset$ and  $\mathit{db}(S) = \{z\}$,
the policy $\mathit{sec}$ is 
$\{ 
$$  \langle \oplus, $$ u_1, $$ \langle \texttt{CREATE TRIGGER},  $$ P \rangle,  $$ \admin \rangle, 
$ $  \langle \oplus,  $$ u_2,  $$ \langle \texttt{INSERT}, $$ P \rangle, $\onlyTechReport{\\}$  \admin \rangle, 
$$  \langle \oplus, $$ u_2,  $$ \langle \texttt{DELETE}, $$ S \rangle,  $$ \admin \rangle, 
$ $  \langle \oplus, $$ u_2,  $$ \langle \texttt{SELECT}, $$ P \rangle, $$  \admin \rangle, \onlyTechReport{\\}
$$  \langle \oplus, $$  u_2, $$  \langle \texttt{SELECT},$$  S \rangle, $$  \admin \rangle \}
$,
 and  the set  $U$ is $\{u_1, u_2,  \admin\}$.
 %
%
The run $r$ is as follows: 
\begin{compactenum}
\item The user $u_1$ creates the trigger $\hfill t \hfill  = \hfill \langle \mathit{id}, \hfill  u_1, \hfill  \mathit{INS},\hfill  P,\hfill  \top, \\ \langle \mathtt{DELETE}, S, z\rangle, A \rangle$. 

\item The user $u_2$ inserts the value $v$ in $P$. This activates the trigger $t$ and deletes the content of $S$, i.e., the value $z$.
\end{compactenum}
We used Maude to generate the following run, which illustrates how the system's state changes. 
Note that there are no exceptions during the run.
\onlyTechReport{\\}
\noindent
\begin{tikzpicture}[->,>=stealth',shorten >=1pt,auto, semithick]
  \tikzstyle{every state}=[rectangle, rounded corners,fill=none,draw=black,text=black,minimum height=1em,
           inner sep=1pt, ultra thin]

  \node[state] 		   (A) at (0,0)  	{ \onlyShortVersion{\footnotesize}\onlyTechReport{\small} {$\langle\mathit{db}, U,\mathit{sec},  \emptyset, \emptyset, c_{1}\rangle$}};
  \node[state]         (B) at (5,0) 	{ \onlyShortVersion{\footnotesize}\onlyTechReport{\small} {$\langle\mathit{db}, U,\mathit{sec},  \{t\}, \emptyset, c_{2} \rangle$}};
  \node[state]         (C) at (5,-0.75) 	{ \onlyShortVersion{\footnotesize}\onlyTechReport{\small} {$\langle\mathit{db}[P\oplus {v}], U, \mathit{sec},  \{t\}, \emptyset, c_{3}\rangle$}};
  \node[state]         (D) at (0.5,-0.75)  	{ \onlyShortVersion{\footnotesize}\onlyTechReport{\small} {$\langle\mathit{db}[P\oplus {v}, S \ominus  {z}], U, \mathit{sec},  \{t\}, \emptyset, c_{4} \rangle$}};

  \path (A) edge              node { \onlyShortVersion{\footnotesize}\onlyTechReport{\small} {$\langle u_1, \mathtt{CREATE}, t \rangle$}} (B)
        (B) edge		[left]	  node { \onlyShortVersion{\footnotesize}\onlyTechReport{\small} {$\langle u_2, \mathtt{INSERT}, P,  {v} \rangle$}} (C)
        (C) edge              node { \onlyShortVersion{\footnotesize}\onlyTechReport{\small} {$t$}} (D);
\end{tikzpicture}

Access control mechanisms that do not restrict the execution of triggers with activator's privileges violate \correctness{} because they do not throw security exceptions when $\langle\mathit{db}[P\oplus {v}], U, \mathit{sec},  \{t\}, \emptyset, c_{3}\rangle \not\auth t$.
\end{example}

\subsection{\Confidentiality{}}\label{sect:sec:prop:confidentiality}

To model data confidentiality, we first introduce the concept of indistinguishability of runs, which formalizes the desired confidentiality guarantees by specifying whether users can distinguish between different runs based on their observations.
Formally, a \emph{$P$-indistinguishability relation} is an equivalence relation over $\mathit{traces}(L)$, where $P$ is an extended configuration and $L$ is the $P$-LTS.
Indistinguishable runs, intuitively, should disclose the same information.

We now define the concept of a secure judgment, which is a judgment that does not leak sensitive information or, equivalently, one that cannot be used to differentiate between indistinguishable runs.

\begin{definition}\label{definition:secure:judgment}
Let $P$ be an \accessControlConfiguration{}, $L$ be the $P$-LTS, and $\cong$ be a $P$-indistinguishability relation.
  A judgment $r,i \attMod \phi$~is \emph{secure with respect to $P$ and $\cong$}, written $\mathit{secure}_{P,\cong}(r,i \attMod \phi)$, iff for all $r' \in \mathit{traces}(L)$ such that $r^{i} \cong r'$, it holds that $[\phi]^{\mathit{db}} = [\phi]^{\mathit{db}'}$,~ where $\mathit{last}(r^{i}) = \langle \mathit{db},  U,\mathit{S},T,  V,c\rangle$ and $\mathit{last}(r') = \langle \mathit{db}', U', \mathit{S}', T', V',c'\rangle$.\end{definition}

We are now ready to define \confidentiality{}.
Intuitively, an access control mechanism provides \confidentiality{} iff all judgments that an attacker can derive are~secure.

\begin{definition}\label{definition:data:confidentiality}
Let $P = \langle M,f\rangle$ be an \accessControlConfiguration{}, $L$ be the $P$-LTS,
    $u \in {\cal U}$ be a user, $A$ be  a $(P,u)$-attacker model, and $\cong$ be a $P$-indistinguishability relation.
	We say that \emph{$f$ provides \confidentiality{} with respect to $P$,   $u$, $A$, and $\cong$} iff 
$\mathit{secure}_{P,\cong}(r,i \attMod \phi)$ for all judgments $r,i \attMod \phi$ that hold in $A$.~\end{definition}

We now define the indistinguishability relation that we use in the rest of the paper, which captures what each user can observe (as stated in  \S\ref{sect:adv:model:concrete}) and the effects of the system's access control policy.
Let $P = \langle \langle D, \Gamma \rangle,f\rangle$ be an \accessControlConfiguration{}, $L$ be the $P$-LTS,  and $u$ be a user in ${\cal U}$. 
Given a run $r \in \mathit{traces}(L)$, the user $u$ is  aware only of his actions  and not of the actions of the other users in $r$.
This is represented by the $u$-projection of $r$, which is obtained by masking all sequences of actions that are not issued by $u$ using a distinguished symbol $*$.
Specifically, the \emph{$u$-projection of $r$} is a sequence of states in $\Omega_{M}$ and actions in ${\cal A}_{D,u} \cup {\cal TRIGGER}_{D} \cup \{*\}$ that is obtained from $r$ by 
(1) replacing each action not issued by $u$ with $*$, 
(2) replacing each trigger whose invoker is not $u$ with $*$, and 
(3) replacing all non-empty sequences of $*$-transitions with a single $*$-transition.
For each user $u \in {\cal U}$, we define the {$P$-in\-dis\-tin\-guish\-a\-bility} relation $\cong_{P,u}$, which is formally defined in \techReportAppendix{app:indistinguishability}. 
Intuitively, two runs $r$ and $r'$ are $\cong_{P,u}$-indistinguishable, denoted $r \cong_{P,u} r'$, iff\begin{inparaenum}[(1)]
\item the labels of the $u$-pro\-jec\-tions of $r$ and $r'$ are the same,
\item $u$ executes the same actions $a_{1}$, \ldots, $a_{n}$ in $r$ and $r'$, in the same order, and with the same results,~and
\item before each action $a_{i}$, where $1 \leq i \leq n$,  as well as in the last states of $r$ and $r'$, the views, 
the triggers,
the users,
and the data disclosed by the policy 
 are the same in $r$ and~$r'$.
\end{inparaenum}

We remark that there is a close relation between $\cong_{P,u}$ and state-based indistinguishability~\cite{guarnieri2014optimal, rizvi2004extending, wang2007correctness}. For any two $\cong_{P,u}$-indistinguishable runs $r$ and $r'$, the database states that precede all actions issued by $u$ as well as the last states in $r$ and $r'$ are  pairwise indistinguishable under existing state-based notions~\cite{guarnieri2014optimal, rizvi2004extending, wang2007correctness}.

\newcommand*{\horizontalSpace}{.5cm}
\newcommand*{\verticalSpace}{.3cm}
\newcommand*{\labelHorizontalSpace}{.05cm}

\begin{figure}

\centering
\scalebox{.70}[1]{
\begin{tikzpicture}[->,>=stealth',shorten >=1pt,auto, semithick]
  \tikzstyle{every state}=[rectangle, fill=none,draw=black,text=black,minimum size=0mm, minimum width=0mm, minimum height=0em,
           inner sep=0pt, ultra thin]

\node[state] (A) at (0,0) { \footnotesize {\begin{tabular}{| c " c |} \hline  $\mathbf{N}$&$\{v\}$ \\ \hline  $\mathbf{P}$&$\emptyset$ \\ \hline \rowcolor[gray]{.8} $\mathbf{T}$&  $\{v\}$ \\ \hline \end{tabular}}};
  \node[state, right = \horizontalSpace of A] (B) { \footnotesize {\begin{tabular}{| c " c |} \hline $\mathbf{N}$&$\emptyset$ \\ \hline  $\mathbf{P}$&$\emptyset$ \\ \hline \rowcolor[gray]{.8} $\mathbf{T}$&  $\{v\}$ \\ \hline \end{tabular}}};
  \node[state, right = \horizontalSpace of B] (C) { \footnotesize {\begin{tabular}{| c " c |} \hline $\mathbf{N}$&$\emptyset$ \\ \hline  $\mathbf{P}$&$\{v\}$ \\ \hline \rowcolor[gray]{.8} $\mathbf{T}$&  $\{v\}$ \\ \hline \end{tabular}}};
  \node[state, right = \horizontalSpace of C] (D) { \footnotesize {\begin{tabular}{| c " c |} \hline $\mathbf{N}$&$\{v\}$ \\ \hline $\mathbf{P}$&$\{v\}$ \\ \hline \rowcolor[gray]{.8} $\mathbf{T}$&  $\{v\}$ \\ \hline \end{tabular}}};
  \node[state, right = \horizontalSpace of D] (E) { \footnotesize {\begin{tabular}{| c " c |} \hline $\mathbf{N}$&$\{v\}$ \\ \hline $\mathbf{P}$&$\{v\}$ \\ \hline \rowcolor[gray]{.8} $\mathbf{T}$&  $\{v\}$ \\ \hline \end{tabular}}};
  \node[left = \labelHorizontalSpace of A] (label) {$r(\mathit{db}_{1})$};

  \node[state, below = \verticalSpace of A] (A1) { \footnotesize {\begin{tabular}{| c " c |} \hline $\mathbf{N}$&$\{v\}$ \\ \hline $\mathbf{P}$&$ \emptyset$ \\ \hline \rowcolor[gray]{.8} $\mathbf{T}$&  $\{j,v\}$ \\ \hline \end{tabular}}};
  \node[state, below = \verticalSpace of B] (B1) { \footnotesize {\begin{tabular}{| c " c |} \hline $\mathbf{N}$&$ \emptyset$ \\ \hline $\mathbf{P}$&$ \emptyset$ \\ \hline \rowcolor[gray]{.8} $\mathbf{T}$&  $\{j,v\}$ \\ \hline \end{tabular}}};
  \node[state, below = \verticalSpace of C] (C1) { \footnotesize {\begin{tabular}{| c " c |} \hline $\mathbf{N}$&$ \emptyset$ \\ \hline $\mathbf{P}$&$ \{v\}$ \\ \hline \rowcolor[gray]{.8} $\mathbf{T}$&  $\{j,v\}$ \\ \hline \end{tabular}}};
  \node[state, below = \verticalSpace of D] (D1) { \footnotesize {\begin{tabular}{| c " c |} \hline $\mathbf{N}$&$ \{v\}$ \\ \hline $\mathbf{P}$&$ \{v\}$ \\ \hline \rowcolor[gray]{.8} $\mathbf{T}$&  $\{j,v\}$ \\ \hline \end{tabular}}};
  \node[state, below = \verticalSpace of E] (E1) { \footnotesize {\begin{tabular}{| c " c |} \hline $\mathbf{N}$&$ \{v\}$ \\ \hline $\mathbf{P}$&$ \{v\}$ \\ \hline \rowcolor[gray]{.8} $\mathbf{T}$&  $\{j,v\}$ \\ \hline \end{tabular}}};
   \node[left = \labelHorizontalSpace of A1] (label1) {$r(\mathit{db}_{2})$};

  \node[state, below = \verticalSpace of A1] (A2) { \footnotesize {\begin{tabular}{| c " c |} \hline $\mathbf{N}$&$\{v\}$ \\ \hline $\mathbf{P}$&$ \emptyset$ \\ \hline \rowcolor[gray]{.8} $\mathbf{T}$&  $\emptyset$ \\ \hline \end{tabular}}};
  \node[state, below = \verticalSpace of B1] (B2) { \footnotesize {\begin{tabular}{| c " c |} \hline $\mathbf{N}$&$ \emptyset$ \\ \hline $\mathbf{P}$&$ \emptyset$ \\ \hline \rowcolor[gray]{.8} $\mathbf{T}$&  $\emptyset$ \\ \hline \end{tabular}}};
  \node[state, below = \verticalSpace of C1] (C2) { \footnotesize {\begin{tabular}{| c " c |} \hline $\mathbf{N}$&$ \emptyset$ \\ \hline $\mathbf{P}$&$ \{v\}$ \\ \hline \rowcolor[gray]{.8} $\mathbf{T}$&  $\emptyset$ \\ \hline \end{tabular}}};
  \node[state, below = \verticalSpace of D1] (D2) { \footnotesize {\begin{tabular}{| c " c |} \hline $\mathbf{N}$&$ \emptyset$ \\ \hline $\mathbf{P}$&$ \{v\}$ \\ \hline \rowcolor[gray]{.8} $\mathbf{T}$&  $\emptyset$ \\ \hline \end{tabular}}};
  \node[state, below = \verticalSpace of E1] (E2) { \footnotesize {\begin{tabular}{| c " c |} \hline $\mathbf{N}$&$ \emptyset$ \\ \hline $\mathbf{P}$&$ \{v\}$ \\ \hline \rowcolor[gray]{.8} $\mathbf{T}$& $\emptyset$ \\ \hline \end{tabular}}};  
  \node[left = \labelHorizontalSpace of A2] (label2) {$r(\mathit{db}_{3})$};

 
  \path (A) edge (B)
        (B) edge (C)
        (C) edge (D)
        (D) edge (E);

  \path (A1) edge (B1)
        (B1) edge (C1)
        (C1) edge (D1)
        (D1) edge (E1);

  \path (A2) edge (B2)
        (B2) edge (C2)
        (C2) edge (D2)
        (D2) edge (E2);
\end{tikzpicture}
}
\vspace{2pt}
\caption{The runs $r(\mathit{db}_{1})$ and $r(\mathit{db}_{2})$ are indistinguishable, whereas $r(\mathit{db}_{1})$ and $r(\mathit{db}_{3})$ are not.}\label{fig:indistinguishability}
\end{figure}

Example \ref{example:ibsec:1} illustrates our indistinguishability notion.

\begin{example}\label{example:ibsec:1}
Let the schema, the set of users, the policy, and the triggers be as in Example \ref{example:derivation}.
Consider the following run $r(\mathit{db})$, parametrized by the initial data\-base state~$\mathit{db}$:
\begin{compactenum}
\item $u$  deletes  $v$ from $N$.
\item $u$ inserts $v$ in $\mathit{P}$.  
If $v$ is in $T$, this activates the trigger $t$, which, in turn, inserts  $v$ in $N$. 
\item $u$ issues the \texttt{SELECT} query $N(v)$.
\end{compactenum}

\noindent
Let $\mathit{db}_{1}$, $\mathit{db}_{2}$, and $\mathit{db}_{3}$ be three database states such that $\mathit{db}_{1}(\mathit{T}) = \{v\}$, $\mathit{db}_{2}(\mathit{T}) = \{j,v\}$, and $\mathit{db}_{3}(\mathit{T}) = \emptyset$, whereas $\mathit{db}_{i}(N) = \{v\}$ and $\mathit{db}_{i}(\mathit{P}) =\emptyset$, for  $1 \leq i \leq 3$. 
Note that $r(\mathit{db}_{1})$ is the run used in Example \ref{example:derivation}.
Figure \ref{fig:indistinguishability} depicts how the database's state changes during the runs $r(\mathit{db}_{i})$, for $1 \leq i\leq 3$. 
Gray indicates those tables that the user $u$ cannot read.
The runs $r(\mathit{db}_{1})$ and $r(\mathit{db}_{2})$ are indistinguishable for the user  $u$. 
The only difference between them is the content of the table $\mathit{T}$, which $u$ cannot read. 
In contrast, $u$ can distinguish between $r(\mathit{db}_{1})$ and $r(\mathit{db}_{3})$ because the trigger has been executed in the former and not in the latter.

Indistinguishability may also depend  on the actions of the other users.
Consider the runs $r'$ and $r''$ obtained by extending $r(\mathit{db}_{1})$ respectively with one and two \texttt{SELECT} queries issued by the administrator just after $u$'s query. 
The user $u$ can distinguish between $r(\mathit{db}_{1})$ and $r'$ because  he knows that other users interacted with the system in $r'$ but not in $r(\mathit{db}_{1})$, i.e.,  the $u$-projections  have different labels.
In contrast, the runs $r'$ and $r''$ are indistinguishable for $u$ because he only knows that, after his own \texttt{SELECT}, other users interacted with the system, i.e., the $u$-projections  have the same labels.
However, he does not know the number of commands, the commands themselves, or their results. 
\end{example}

Example~\ref{example:ibsec:3} shows that existing \acf{}s leak sensitive information and therefore do not provide \confidentiality{}. 

\begin{example}\label{example:ibsec:3}
In Example \ref{example:derivation}, we showed how the user $u$ derives $r, 3 \attMod \mathit{T}(v)$. 
The judgment is not secure because there is a run indistinguishable from $r^{3}$, i.e., the run $r^{3}(\mathit{db}_{3})$  in Example \ref{example:ibsec:1}, in which $\mathit{T}(v)$ does not hold. 
\end{example}

Example \ref{example:ibsec:2} shows how views may leak information about the underlying tables.
Even though this leakage  might be considered legitimate, there is no way in our setting to distinguish between intended and unintended leakages.
If this is desired, data confidentiality can be extended with the concept of \emph{declassification}~\cite{askarov2009tight ,askarov2007gradual}. 

\begin{example}\label{example:ibsec:2}
Consider a database with two tables $T$ and $Z$ and 
 a view $V = \langle \mathit{v}, \mathit{admin}, \{x\,|\,T(x) \wedge Z(x)\}, O\rangle$. 
The set $U$ is $\{u,\mathit{admin}\}$ and the policy $S$ is \onlyTechReport{\hspace{-0.3mm}}$\{\langle \oplus,\onlyTechReport{\hspace{-0.3mm}} u\onlyTechReport{\hspace{-0.3mm}},\langle \texttt{SELECT},\onlyTechReport{\hspace{-0.3mm}} T\rangle,\onlyTechReport{\hspace{-0.3mm}} \mathit{admin}\rangle, \onlyTechReport{\\} \langle \oplus, u,\langle \texttt{SELECT}, \mathit{V}\rangle,  \mathit{admin}\rangle,\onlyShortVersion{\\} \langle \oplus, u,\langle \texttt{INSERT}, T\rangle,  \mathit{admin}\rangle \}$.
Consider the following run $r$, para\-metrized by the initial database state $\mathit{db}$, where $u$ first inserts $27$ into $T$ and afterwards issues the \texttt{SELECT} query $V(27)$.
We assume there are no exceptions in $r$.

\noindent
\begin{tikzpicture}[->,>=stealth',shorten >=1pt,auto, semithick]
  \tikzstyle{every state}=[rectangle, rounded corners,fill=none,draw=black,text=black,minimum height=1em,
           inner sep=1pt, ultra thin]

  \node[state] 		   (A) at (0,0)  	{ \onlyShortVersion{\footnotesize}\onlyTechReport{\small} {$\langle \mathit{db}, U,S, \emptyset, \{V\}, c_{1}\rangle$}};
  \node[state]         (B) at (5.3,0) 	{ \onlyShortVersion{\footnotesize}\onlyTechReport{\small} {$\langle \mathit{db}[T \oplus 27], U,S,  \emptyset, \{V\}, c_{2}\rangle$}};
  \node[state]         (C) at (5.3,-0.75) 	{ \onlyShortVersion{\footnotesize}\onlyTechReport{\small} {$\langle \mathit{db}[T \oplus 27], U,S,  \emptyset, \{V\}, c_{3}\rangle$}};

  \path (A) edge              node { \onlyShortVersion{\footnotesize}\onlyTechReport{\small} {$\langle u, \mathtt{INSERT}, T, 27 \rangle$}} (B)
        (B) edge		[left]	  node { \onlyShortVersion{\footnotesize}\onlyTechReport{\small} {$\langle u, \mathtt{SELECT}, V(27) \rangle$}} (C);
\end{tikzpicture}

We used Maude to generate the runs $r(d)$ and $r(d')$ with the initial database states $d$ and $d'$ such that $d(T) = d(Z) = d'(T) = \emptyset$ and $d'(Z) = \{27\}$.
The runs $r^{1}(d)$ and $r^{1}(d')$ are indistinguishable for $u$ because they differ only in the content of $Z$, which $u$ cannot read.
After the $\mathtt{INSERT}$, $u$ can distinguish between $r^{2}(d)$ and $r^{2}(d')$ by reading $V$.
Indeed, $d[T \oplus 27](V) = \emptyset$, because $d(Z) = \emptyset$, whereas $d'[T \oplus 27](V) = \{27\}$.
The user $u$  derives  $r(d'),1 \attMod Z(27)$, which is not secure 
 because $r^{1}(d)$~and $r^{1}(d')$ are indistinguishable for $u$, but $Z(27)$ holds just~in~the~latter.
\end{example}

In contrast to existing security notions \cite{guarnieri2014optimal,wang2007correctness,rizvi2004extending},  we have defined \confidentiality{}  over runs.
This is essential to model and detect attacks, such as those  in Examples \ref{example:ibsec:3} and \ref{example:ibsec:2}, where an attacker infers sensitive information from the  transitions between states.
For instance, the leakage in Example \ref{example:ibsec:2} is due to the execution of the \texttt{INSERT} command.
Although the \texttt{SELECT} command is authorized by the policy, $u$ can use it to infer sensitive information about the system's state before the \texttt{INSERT} execution.

%% file: enforcement.tex
\onlyTechReport{\newpage}
\section{A Provably Secure \acf{}}\label{sect:enf:alg}
\begin{figure*}[!hbtp]
\centering
\small{
\begin{tabular}{c c}
$\begin{aligned}
&\vartriangleright \emph{$s$  is a state  and  $a$ is an action}\\
&\enspace\textbf{function } \mathit{f}(s, a) \\
1.&\enspace \enspace \textbf{return } f_{\mathit{int}}(s,a) \wedge f_{\mathit{conf}}(s,a, \mathit{user}(s,a))\\
\\\\
&\vartriangleright \emph{$s$ is a state  and $a$ is an action}\\
&\enspace \textbf{function } \mathit{f}_{\mathit{int}}(s, a)\\
1.&\enspace \enspace \textbf{if } \mathit{trg}(s) = \epsilon  \textbf{ return } \mathit{auth}(s,a)\\
2.&\enspace \enspace\enspace \textbf{else } \textbf{if } a = \mathit{cond}(\mathit{trg}(s),s) \textbf{ return }\top\\
3.&\enspace  \enspace\enspace\enspace \textbf{else if } a = \mathit{act}(\mathit{trg}(s),s) \textbf{ return } \mathit{auth}(s,\mathit{trg}(s))\\
4.&\enspace  \enspace\enspace\enspace\enspace	\textbf{else return } \bot \\
\end{aligned}$
&
$\begin{aligned}
&\vartriangleright \emph{$s$ is a state, $a$  is an action, and $u$ is a user}\\
& \textbf{function } \mathit{f}_{\mathit{conf}}(s, \mathit{a}, u)\\
1.&  \enspace \textbf{switch } \mathit{a}\\
2.& \enspace \enspace\enspace \textbf{case } \langle u', \mathtt{SELECT}, q\rangle: \textbf{ return }\mathit{secure}(u,q, s)\\
3.& \enspace \enspace\enspace \textbf{case } \langle u', \mathtt{INSERT}, R, \overline{t}\rangle: \textbf{ case } \langle u', \mathtt{DELETE}, R, \overline{t}\rangle:  \\
4.& \enspace \enspace\enspace\enspace \textbf{if }  \mathit{leak}(\mathit{a}, s, u) \vee \neg \mathit{secure}(u,\mathit{getInfo}(\mathit{a}),s) \textbf{ return }\bot\\
5.& \enspace \enspace\enspace\enspace\textbf{for } \gamma \in {\mathit{Dep}}(\mathit{a}, \Gamma)\\
6.& \enspace \enspace\enspace\enspace\enspace \textbf{if }(\neg \mathit{secure}(u,\mathit{getInfoS}(\gamma, \mathit{a}),s) \vee  \neg \mathit{secure}(u,\mathit{getInfoV}(\gamma, \mathit{a}),s))\\
7.& \enspace \enspace\enspace\enspace\enspace\enspace \textbf{ return } \bot\\
8.& \enspace \enspace\enspace \textbf{case } \langle \oplus, u'', \mathit{pr}, u'\rangle, \langle \oplus^{*}, u'', \mathit{pr}, u'\rangle: \textbf{return } \neg \mathit{leak}(\mathit{a}, s, u)\\
9.& \enspace \textbf{return }\top \\
\end{aligned}$
\end{tabular}
}
\vspace{2pt}
\caption{
The \acf{} $f$ uses the two subroutines $\mathit{f}_{\mathit{int}}$ and $\mathit{f}_{\mathit{conf}}$.
The former provides \correctness{} and the latter provides \confidentiality{} with respect to the user $\mathit{user}(s,a)$, which denotes either the user issuing the action, when the system is not executing  a trigger, or the trigger's invoker. }\label{figure:algorithms}
\end{figure*}
We now present a \acf{} that provides both \correctness{} and \confidentiality{}.
We first explain the ideas behind it using examples.
Afterwards, we show that it 
satisfies the desired security properties and has  acceptable overhead.
Finally, we argue that it is more permissive  than existing access control solutions. 

Figure \ref{figure:algorithms} depicts our \acf{} $f$ together with the functions  $\mathit{f}_{\mathit{int}}$ and $f_{\mathit{conf}}$.
Additional details about the \acf{} are given in~\techReportAppendices{app:enforcement:ibsec}{app:composition}.
The \acf{}  takes as input a state $s$ and an action $a$ and  outputs $\top$ iff both  $\mathit{f}_{\mathit{int}}$ and $f_{\mathit{conf}}$ authorize $a$ in $s$, i.e., iff $a$'s execution neither violates  \correctness{} nor \confidentiality{}. 
Note that our algorithm is not \emph{complete} in that it may reject some secure commands.
However, from the results in \cite{Koutris:2012:QDP:2213556.2213582, guarnieri2014optimal,nash2010views}, it follows that no algorithm can be complete and provide  \correctness{} and \confidentiality{} for the relational calculus.

Our \acf{} is invoked by the database system each time a user $u$ issues an action $a$ to check whether $u$ is authorized to execute $a$.
The \acf{}  is also invoked whenever the database system executes a scheduled trigger $t$:
once to check if the \texttt{SELECT} statement associated with $t$'s \texttt{WHEN} condition is authorized and once, in case $t$ is enabled, to check if $t$'s action is authorized.

\subsection{Enforcing \Correctness{}}

The function $f_{\mathit{int}}$ takes as input a state $s$ and an action $a$. 
If the system is not executing a trigger, denoted by $\mathit{trg}(s) = \epsilon$, $f_{\mathit{int}}$ checks (line 1) whether $a$ is authorized with respect to $s$. 
In line 2, $f_{\mathit{int}}$ checks whether $a$ is the current trigger's condition. 
 If this is the case, it returns $\top$ because the triggers' conditions do not violate \correctness{}.
Finally, the algorithm  checks (line 3) whether $a$ is the current trigger's action, and if this is the case, it checks whether the current trigger $\mathit{trg}(s)$ is authorized with respect to $s$.
The function $\mathit{auth}$, which checks if $a$ is authorized with respect to $s$,  is a sound and computable under-approximation of  $\auth$.
Thus, any action authorized by $f_{\mathit{int}}$ is authorized according to $\auth$.
This ensures database integrity.
Note that $\auth$ relies on the concept of \emph{determinacy}~\cite{nash2010views} to decide whether a query is determined by a set of views.
Since determinacy is undecidable~\cite{nash2010views}, in $\mathit{auth}$ we implement a sound under-approximation of  it, given in \techReportAppendix{app:enforcement:eopsec}, that checks syntactically if a query is determined by a set of views.

\begin{example}\label{example:enforcement:integrity}
Consider a database with three tables: $R$, $T$, and $Z$. 
The set $U$ is $\{u,u',\mathit{admin}\}$ and the policy $S$ is $\{\langle \oplus,$  $u,\langle \texttt{SELECT}, R\rangle, \mathit{admin}\rangle,$ $  \langle \oplus^{*},  u,\langle \texttt{SELECT}, \mathit{T}\rangle, \mathit{admin}\rangle,$ $\langle \oplus^{*}, u,\langle \texttt{SELECT}, \mathit{Z}\rangle, \mathit{admin}\rangle\}$.
There are two views $V = \langle \mathit{v}, \onlyTechReport{\\} \mathit{admin}, \{x\,|\,T(x) \wedge Z(x)\}, O\rangle$ and $W = \langle w, u,$ $ \{x\,|\,R(x) \vee V(x)\}, O\rangle$. 
The user $u$ tries to grant to $u'$ read access to $W$, i.e., he issues $\langle \oplus, u',\langle \texttt{SELECT}, W\rangle, u\rangle$.
The \acf{} $f_{\mathit{int}}$ rejects the command and raises a security exception because $u$ is authorized to delegate the read access only for  $T$ and $Z$ but $W$'s result depends also on $R$, for which $u$ cannot delegate read access. 
Assume now that the policy is $\{\langle \oplus^{*},$$ u,$$ \langle \texttt{SELECT},$$ R\rangle,$$ \mathit{admin}\rangle,$ $\langle \oplus^{*},$$ u,$$  \langle \texttt{SELECT},$$ \mathit{T}\rangle,$$ \mathit{admin}\rangle,$$\langle \oplus^{*},$$  u,$\onlyTechReport{\\}$ \langle \texttt{SELECT},$$ \mathit{Z}\rangle,$$ \mathit{admin}\rangle\}$.
In this case, $f_{\mathit{int}}$ authorizes the \texttt{GRANT}. 
The reason is that $W$'s definition can be equivalently rewritten as $\{x\,|\,R(x) \vee (T(x) \wedge Z(x))\}$ and $u$ is authorized~to~delegate~the~read~access~for~$R$,~$T$,~and~$Z$.
\end{example}

\onlyTechReport{\newpage}
\subsection{Enforcing \Confidentiality{}}

The function $f_{\mathit{conf}}$, shown in Figure \ref{figure:algorithms}, takes as input an action $a$, a state $s$, and a user $u$. 
Note that any user other than the administrator is a potential attacker.
The requirement for $f_{\mathit{conf}}$ is that it authorizes only those commands that result in secure  judgments for $u$ as required by Definition \ref{definition:data:confidentiality}.
To achieve this,  $f_{\mathit{conf}}$ over-approximates the set of judgments that $u$ can derive from $a$'s execution.
For instance, the algorithm assumes that $u$ can always derive the trigger's condition from the run, even though this is not always the case.
Then,  $f_{\mathit{conf}}$ authorizes  $a$ iff it can determine that all $u$'s judgements are secure.
This can be done by analysing just a finite subset of the over-approximated set of $u$'s judgments.

In more detail, $f_{\mathit{conf}}$ performs  a case distinction on the action $a$ (line 1). 
If $a$ is a \texttt{SELECT} command (line 2), $f_{\mathit{conf}}$ checks whether the query is secure with respect to the current state $s$ and the user $u$ using the  $\mathit{secure}$ procedure. 
If $a$ is an \texttt{INSERT} or \texttt{DELETE} command (lines 3--7), $f_{\mathit{conf}}$  checks (line 4), using the $\mathit{leak}$ procedure,  whether $a$'s execution may leak sensitive information through the views that $u$ can read, as in Example \ref{example:ibsec:2}.
Afterwards, $f_{\mathit{conf}}$ also checks (line 4) whether the information $u$ can learn from $a$'s execution, modelled by the sentence computed by the procedure $\mathit{getInfo}(a)$, is secure.
In line 5--7, $f_{\mathit{conf}}$ computes the set of all integrity constraints that $a$'s execution may violate, denoted by $\mathit{Dep}(a,\Gamma)$, and  for all constraints $\gamma$, it checks whether the information that $u$ may learn from $\gamma$ is secure.
The procedure $\mathit{getInfoS}$ (respectively $\mathit{getInfoV}$) computes the sentence modelling the information learned by $u$ from $\gamma$ if $a$ is executed successfully (respectively violates $\gamma$).
If $a$ is a \texttt{GRANT} command (line 8), $f_{\mathit{conf}}$ checks whether $a$'s successful execution discloses sensitive information to $u$.
In the remaining cases (line 9), $f_{\mathit{conf}}$ authorizes $a$.

\smallskip
\noindent
{\bf Secure judgments.}
Determining if a given judgment is secure is undecidable for $\mathit{RC}$~\cite{guarnieri2014optimal,Koutris:2012:QDP:2213556.2213582}.
Hence, the $\mathit{secure}$ procedure implements a sound and computable under-ap\-prox\-i\-ma\-tion of this notion. 
We now present our solution.
Other sound under-ap\-prox\-i\-ma\-tions  can alternatively be used without affecting $f_{\mathit{conf}}$'s \confidentiality{} guarantees.

Let $M = \langle D,\Gamma \rangle$ be a system configuration, $r, i \attMod \phi$ be a judgment, and $s = \langle \mathit{db}, U, \mathit{sec}, T,V, c \rangle$ be the $i$-th state in $r$.
As a first under-approximation, instead of the set of all runs indistinguishable from $r^{i}$, we consider the larger set of all runs $r'$ whose last state $s' = \langle \mathit{db}', U, \mathit{sec}, T,V, c' \rangle$ is such that the disclosed data in $\mathit{db}$ and $\mathit{db}'$ are the same.
Note that if a judgment is secure with respect to this larger set, it is secure also with respect to the set of indistinguishable runs because the former set contains the latter.
This larger set depends just on the database state $\mathit{db}$ and the policy $\mathit{sec}$, not on the run or the attacker model $\attackerModel$.
Determining judgment's security is, however, still undecidable even on this larger set.
We therefore employ a second under-approximation that uses query rewriting. 
We rewrite the sentence $\phi$ to a sentence $\phi_{\mathit{rw}}$ such that if $r, i \attMod \phi$ is not secure for the user $u$, then $[\phi_{\mathit{rw}}]^{\mathit{db}} = \top$. 
The formula  $\phi_{\mathit{rw}}$ is $\neg \phi^{\top}_{s,u} \wedge \phi^{\bot}_{s,u}$, where $\phi^{\top}_{s,u}$ and $\phi^{\bot}_{s,u}$ are defined inductively  over $\phi$.
A formal definition of $\mathit{secure}$ is~given~in~\techReportAppendix{app:enforcement:ibsec}.

We now explain how we construct $\phi^{\top}_{s,u}$ and $\phi^{\bot}_{s,u}$.
We assume that both $\phi$ and  $V$ contain only views with the owner's privileges.
The extension to the general case is given in \techReportAppendix{app:enforcement:ibsec}. 
First, for each table or view $o \in D \cup V$, we create additional views representing any possible projection of $o$.
The \emph{extended vocabulary} contains the tables in $D$, the views in $V$, and their projections.
For instance, given a table $R(x,y)$, we create the views $R_{x}$ and $R_{y}$ representing respectively $\{ y\,|\,\exists x.\, R(x,y)\}$ and $\{x\,|\,\exists y.\, R(x,y)\}$.
Second, we compute the formula $\phi'$  by replacing each sub-formula $\exists \overline{x}.\, R(\overline{x}, \overline{y})$ in $\phi$ with the view $R_{\overline{x}}(\overline{y})$ associated with the corresponding projection.
Third, for each predicate $R$ in the formula $\phi'$, we compute the sets $R^{\top}_{s,u}$ and $R^{\bot}_{s,u}$.
The set $R^{\top}_{s,u}$ (respectively $R^{\bot}_{s,u}$) contains all the tables and views $K$ in the extended vocabulary such that\begin{inparaenum}[(1)]
\item $K$ is contained in (respectively contains) $R$, and
\item the user $u$ is authorized to read $K$ in $s$, i.e., there is a grant $\langle \mathit{op}, u, \langle \mathtt{SELECT}, K' \rangle, u' \rangle \in \mathit{sec}$ such that either $K' = K$ or $K$ is obtained from $K'$ through a projection. 
\end{inparaenum}
The formula $\phi^{v}_{s,u}$, where $v \in \{\top,\bot\}$, is: 
\[
\phi^{v}_{s,u} = \left\{ 
  \begin{array}{l l}
    \bigvee_{S \in R^{\top}_{s,u}} S(\overline{x}) & \text{if } \phi= R(\overline{x}) \text{ and } v = \top\\
    \bigwedge_{S \in R^{\bot}_{s,u}} S(\overline{x}) & \text{if } \phi= R(\overline{x}) \text{ and } v = \bot\\
    \neg \psi^{\neg v}_{s,u} & \text{if } \phi = \neg \psi\\
    \psi^{v}_{s,u}\,\mathit{*}\, \gamma^{v}_{s,u} & \text{if } \phi = \psi\,\mathit{*}\,\gamma \text{ and } \mathit{*} \in \{\vee, \wedge\}\\
    Q\, x.\, \psi^{\mathit{v}}_{s,u} &  \text{if } \phi =  Q\, x.\, \psi \text{ and } \mathit{Q} \in \{\exists, \forall\}\\
	\phi & \text{otherwise}\\
  \end{array} \right.
\]
The formulae are such that if $\phi^{\top}_{s,u}$ holds, then $\phi$ holds and if $\neg \phi^{\bot}_{s,u}$ holds, then $\neg \phi$ holds.
To compute the sets $R^{\top}_{s,u}$ and $R^{\bot}_{s,u}$, we check the containment between queries. 
Since query containment is undecidable \cite{abiteboul1995foundations}, we implement a sound under-approximation of it, described in \techReportAppendix{app:enforcement:ibsec}. 
Other sound under-approximations can be used as well.

Our $\phi^{\top}_{s,u}$ and $\phi^{\bot}_{s,u}$ rewritings share similarities with the \emph{low} and \emph{high evaluations} of Wang et al.~\cite{wang2007correctness}.
Both try to approximate the result of a query just by looking at the authorized data.
However, we use $\phi^{\top}_{s,u}$ and $\phi^{\bot}_{s,u}$ to determine a judgment's security, whereas  Wang et al.~use evaluations to restrict the query's results only to authorized data.

\begin{figure}

\centering
\scalebox{.87}[.87]{
\begin{tikzpicture}[->,>=stealth',shorten >=1pt,auto, semithick]
  \tikzstyle{every state}=[rectangle, fill=none,draw=none,text=black,minimum size=0mm, minimum width=0mm, minimum height=0em,
           inner sep=0pt, ultra thin, anchor=north west]

\node[anchor=north west] (S) at (0,0) { 
							\footnotesize {
								\begin{tabular}{| c |  c |}
								 \hline \multicolumn{2}{|c|}{$\mathbf{S}$} \\ 
								 \hline \rowcolor[gray]{.8}  $1$&$1$ \\ 
								 \hline  \rowcolor[gray]{.8} $2$&$3$ \\ 
								 \hline  \rowcolor[gray]{.8} $4$&$2$ \\ 
								 \hline 
								\end{tabular}
							}
						};

\node[anchor=north west] (dbLabel) at ($(S.north west) + (0, 0.4)$) {$\mathbf{Database\ State}$};				
\node[anchor=north west] (R) at ($(S.north east)$) { 
							\footnotesize {
								\begin{tabular}{ |c | }
								 \hline $\mathbf{R}$ \\ 
								 \hline  \rowcolor[gray]{.8} $3$\\ 
								 \hline 
								\end{tabular}
							}
						};
						
\node[anchor=north west] (Q) at ($(R.north east)$) { 
							\footnotesize {
								\begin{tabular}{| c | }
								 \hline $\mathbf{Q}$ \\ 
								 \hline  \rowcolor[gray]{.8} $4$\\ 
								 \hline 
								\end{tabular}
							}
						};

\node[anchor=north west] (V) at ($(Q.north east) + (0, 0)$) { 
							\footnotesize {
								\begin{tabular}{ | c |  c |}
								 \hline \multicolumn{2}{|c|}{$\mathbf{V}$} \\ 
								 \hline $1$&$1$ \\ 
								 \hline  $2$&$3$ \\ 
								 \hline 
								\end{tabular}
							}
						};

\node[anchor=north west] (viewLabel) at ($(V.north west) + (0, 0.4)$) {$\mathbf{Views}$};	
\node[anchor=north west] (W) at ($(V.north east) $) { 
							\footnotesize {
								\begin{tabular}{ | c | }
								 \hline $\mathbf{W}$ \\ 
								 \hline  $3$\\ 
								 \hline  $4$\\ 
								 \hline 
								\end{tabular}
							}
						};

\node[anchor=north west] (Vdef) at ($(V.south west)$) {\small {$V = \{x,y\,|\, S(x,y) \wedge (x=1 \vee y = 3)\}$}};

\node[anchor=north west] (Wdef) at ($(Vdef.south west)$) {\small {$W = \{x\,|\, R(x) \vee Q(x)\}$}};
						
\node[anchor=north west] (Sx) at ($(S.south west) - (0,1.1)$) { 
							\footnotesize {
								\begin{tabular}{| c |}
								 \hline {$\mathbf{S}_x$} \\ 
								 \hline \rowcolor[gray]{.8}  $1$ \\ 
								 \hline  \rowcolor[gray]{.8} $3$ \\ 
								 \hline  \rowcolor[gray]{.8} $2$ \\ 
								 \hline 
								\end{tabular}
							}
						};

\node[anchor=north west] (extVocLabel) at ($(Sx.north west) + (0, 0.4)$) {$\mathbf{Extended\ Vocabulary}$};	

\node[anchor=north west] (Sy) at ($(Sx.north east)$) { 
							\footnotesize {
								\begin{tabular}{| c |}
								 \hline {$\mathbf{S}_y$} \\ 
								 \hline \rowcolor[gray]{.8}  $1$ \\ 
								 \hline  \rowcolor[gray]{.8} $2$ \\ 
								 \hline  \rowcolor[gray]{.8} $4$ \\ 
								 \hline 
								\end{tabular}
							}
						};
						
\node[anchor=north west] (Vx) at ($(Sy.north east)$) { 
							\footnotesize {
								\begin{tabular}{| c |}
								 \hline {$\mathbf{V}_x$} \\ 
								 \hline   $1$ \\ 
								 \hline   $3$ \\ 
								 \hline 
								\end{tabular}
							}
						};

\node[anchor=north west] (Vy) at ($(Vx.north east)$) { 
							\footnotesize {
								\begin{tabular}{| c |}
								 \hline {$\mathbf{V}_y$} \\ 
								 \hline  $1$ \\ 
								 \hline  $2$ \\ 
								 \hline 
								\end{tabular}
							}
						};
							
\node[anchor=north west] (Sxdef) at ($(Sx.south west)$) {\small {$S_x = \{y\,|\, \exists x.\,S(x,y)\}$}};

\node[anchor=north west] (Sydef) at ($(Sxdef.south west)$) {\small {$S_y = \{x\,|\, \exists y.\,S(x,y)\}$}};

\node[anchor=north west] (Vxdef) at ($(Sxdef.north east)$) {\small {$V_x = \{y\,|\, \exists x.\,V(x,y)\}$}};

\node[anchor=north west] (Vydef) at ($(Vxdef.south west)$) {\small {$V_y = \{x\,|\, \exists y.\,V(x,y)\}$}};

\node[anchor=north west] (syTop) at ($(Vy.north east) + (0.2,0)$) { \small {${S_{y}}^{\top}_{s,u} = \{V_{y}\}$}};

\node[anchor=north west] (containmentLabel) at ($(syTop.north west) + (0, 0.4)$) {$\mathbf{Containment\ Sets}$};

\node[anchor=north west] (syBot) at ($(syTop.south west)$) { \small {${S_{y}}^{\bot}_{s,u} = \emptyset$}};
\node[anchor=north west] (RTop) at ($(syTop.north east)$){ \small {$R^{\top}_{s,u} = \emptyset$}};
\node[anchor=north west] (RBot) at ($(RTop.south west)$) { \small {$R^{\bot}_{s,u} = \{W\}$}};

\node[anchor=north west] (phi) at ($(Sx.south west) - (0, 1.5)$) {\small {$\phi := (\exists y.\, S(2,y)) \wedge ( \neg R(5)  \vee \exists y.\, S(4,y) ) \equiv S_{y}(2) \wedge ( \neg R(5)  \vee S_{y}(4) )$}};

\node[anchor=north west] (phiLabel) at ($(phi.north west) + (0, 0.4)$) {$\mathbf{Original\ Sentence}$};

\node[anchor=north west] (phiRw) at ($(phi.south west)- (0, .5)$) {\small {$\phi_{\mathit{rw}} := \neg \phi^{\top}_{s,u} \wedge \phi^{\bot}_{s,u}$}};

\node[anchor=north west] (phiRwLabel) at ($(phiRw.north west) + (0, 0.4)$) {$\mathbf{Rewriting}$};

\node[anchor=north west] (phiTop) at ($(phiRw.south west)$) {\small {$\phi^{\top}_{s,u} := S_{y}(2)^{\top}_{s,u} \wedge ( \neg R(5)^{\bot}_{s,u}  \vee S_{y}(4)^{\top}_{s,u} ) \equiv V_{y}(2) \wedge ( \neg W(5)  \vee V_{y}(4) )$}};

\node[anchor=north west] (phiBottom) at ($(phiTop.south west)$) {\small {$\phi^{\bot}_{s,u} := S_{y}(2)^{\bot}_{s,u} \wedge ( \neg R(5)^{\top}_{s,u}  \vee S_{y}(4)^{\bot}_{s,u} ) \equiv \top$}};
\end{tikzpicture}
}
\caption{
Checking the security of the judgment $r, 1 \attMod (\exists y.\, S(2,y)) \wedge ( \neg R(5)  \vee \exists y.\, S(4,y))$ from Example~\ref{example:enforcement:confidentiality}.
}\label{fig:pdp:example}
\end{figure}

\begin{example}\label{example:enforcement:confidentiality}
Consider a database with three tables $S$, $R$, and $Q$,  and two views $V = \langle v, \admin, \{x,y\,|\, S(x,y) \wedge (x=1 \vee y = 3)\}, O \rangle$ and $W = \langle w, \admin, \{x\,|\, R(x) \vee Q(x)\}, O \rangle$.
The database state $\mathit{db}$ is $\mathit{db}(S) = \{(1,1),(2,3),(4,2)\}$, $\mathit{db}(R) = \{3\}$, and $\mathit{db}(Q) = \{4\}$, the set $U$ is $\{u,\admin\}$, and the policy $\mathit{sec}$ is $\{\langle \oplus, u , \langle \mathtt{SELECT}, V\rangle, \admin\rangle,   \langle \oplus, u , \langle \mathtt{SELECT},  W\rangle, \onlyTechReport{\\} \admin\rangle\}$.
Let the state $s$ be $\langle \mathit{db}, U,  \mathit{sec}, \emptyset, \{V, W\}, \epsilon \rangle$ and the run $r$ be $s$.
We want to check the security of $r, 1 \attMod \phi$, where $\phi := (\exists y.\, S(2,y)) \wedge ( \neg R(5)  \vee \exists y.\, S(4,y) )$, for the user $u$.
Figure~\ref{fig:pdp:example} depicts the database state $\mathit{db}$, the materializations of the views $V$ and $W$, and the materializations of the views $S_x$, $S_y$, $V_x$, and $V_y$ in the extended vocabulary.
Gray indicates those tables and views that $u$ cannot read.
 
The rewriting process, depicted also in Figure~\ref{fig:pdp:example}, proceeds as follows.
We first rewrite the formula $\phi$ as $S_{y}(2) \wedge ( \neg R(5)  \vee S_{y}(4) )$.
The sets ${S_{y}}^{\top}_{s,u}$, ${S_{y}}^{\bot}_{s,u}$, $R^{\top}_{s,u}$, and $R^{\bot}_{s,u}$ are respectively $\{V_{y}\}$, $\emptyset$, $\emptyset$, and $\{W\}$.
The formulae $\phi^{\top}_{s,u}$ and $\phi^{\bot}_{s,u}$  are respectively $S_{y}(2)^{\top}_{s,u} \wedge ( \neg R(5)^{\bot}_{s,u}  \vee S_{y}(4)^{\top}_{s,u} )$, which is equivalent to $V_{y}(2) \wedge ( \neg W(5)   \vee V_{y}(4) )$, and  $S_{y}(2)^{\bot}_{s,u} \wedge ( \neg R(5)^{\top}_{s,u}  \vee S_{y}(4)^{\bot}_{s,u} )$, which is equivalent to $\top$.
They are both secure, as they depend only on $V$ and $W$.
Furthermore, since $\phi^{\top}_{s,u}$ holds in $s$, then $\phi$ holds as well.
Finally, $\phi_{\mathit{rw}}$ is $\neg \phi^{\top}_{s,u} \wedge \phi^{\bot}_{s,u}$.
Since $\phi_{\mathit{rw}}$ does not hold in $s$, it follows that $r, 1 \attMod \phi$ is secure.
\end{example}

\subsection{Theoretical Evaluation}
\newcommand{\complexity}{$\mathit{AC}^{0}$}
Our \acf{} provides the desired security guarantees and its data complexity, i.e., the complexity of executing $f$ when the action, the policy, the triggers, and the views are fixed, is  \complexity{}.
This means that $f$ can be evaluated in logarithmic space in the database's size, as $\mathit{AC}^{0} \subseteq \mathit{LOGSPACE} $, and evaluation is highly parallelizable.
Note that $\mathit{secure}$'s data complexity is \complexity{} because it relies on  query evaluation, whose data complexity is \complexity{}~\cite{abiteboul1995foundations}.
In contrast, all other operations in $f$ are executed in constant time in terms of data complexity.
Note also that on a single processor,  $f$'s data complexity is polynomial in the database's size. 
We believe that this is acceptable because the database is usually very large, whereas the query, which determines the degree of the polynomial, is small. 
\onlyShortVersion{The proof of  Theorem \ref{theorem:enforcement:security}~is~given~in~\cite{technicalReport}.}
\onlyTechReport{The proofs  are given in Appendices \ref{app:enforcement:ibsec}--\ref{app:composition}.}

\begin{theorem}\thlabel{theorem:enforcement:security}
Let $P = \langle M, f \rangle$ be an \accessControlConfiguration{}, where $M$ is a system configuration and $f$ is as above.
The \acf{} $f$\begin{inparaenum}[(1)]
\item provides \confidentiality{}  with respect to  $P$, $u$, $\attackerModel$, and $\cong_{P,u}$, for any user $u \in {\cal U}$, and
\item provides \correctness{} with respect to $P$. 
\end{inparaenum}
Moreover, the data complexity of $f$ is \complexity{}.
\end{theorem}

 \renewcommand{\complexity}{$\mathit{AC}^{0}$}

As the Examples \ref{example:eval:1} and \ref{example:eval:2} below show, $f$ is more permissive than existing \acf{}s for those actions that violate neither \correctness{} nor \confidentiality{}.

\begin{example}\label{example:eval:1}
Our \acf{} is more permissive than existing mechanisms for commands of the form \texttt{GRANT SELECT ON} $V$ \texttt{TO} $u$, where $V$ is a view with owner's privileges, $u$ is a user, and the statement is issued by the view's owner $o$.
Such mechanisms, in general, authorize the \texttt{GRANT} iff $o$ is authorized to delegate the read permission for all tables and views that occur in $v$'s definition.
Consider again Example \ref{example:enforcement:integrity}. Our \acf{} authorizes $ \langle  \oplus, u', \langle \texttt{SELECT}, W\rangle, u\rangle$ under the policy $S'$.
However, existing mechanisms reject it because $u$ is not directly authorized to read $V$, although $u$ can read the underlying tables.
Our \acf{} also authorizes all the {secure} \texttt{GRANT} statements authorized by existing mechanisms.
\end{example}
\begin{example}\label{example:eval:2}
Our \acf{} is more permissive than the mechanisms used in existing DBMSs for secure \texttt{SELECT} statements.
Such mechanisms, in general, authorize a \texttt{SELECT} statement issued by a user $u$ iff $u$ is authorized to read all tables and views used in the query.
They will reject the query in Example \ref{example:enforcement:confidentiality} even though the query is secure.
Furthermore, any secure \texttt{SELECT} statement authorized by them will be authorized by our solution as well.
Also the \acf{} proposed by Rizvi et al. \cite{rizvi2004extending} rejects the query in Example \ref{example:enforcement:confidentiality} as insecure.
However, our solution and the proposal of Rizvi et al. \cite{rizvi2004extending} are incomparable in terms of permissiveness, i.e., some secure \texttt{SELECT} queries are authorized by one mechanism and not by the other. 
\end{example}

\subsection{Implementation}
\begin{figure}
\begin{tabular}{c c}
    \begin{subfigure}{0.48\textwidth}
\centering
\scalebox{0.8}{
    \begin{tikzpicture}
	\tikzstyle{every node}=[font=\small]    
    
      \begin{axis}[
          width=\linewidth, 
          grid=major, 
          grid style={dashed,gray!30},
          xlabel=Number of tuples, 
          ylabel={Time [$ms$]},
          scaled y ticks = false,
          scaled x ticks = false,
		  legend style={at={(0,1)},anchor=north west},
          x tick label style={at={(axis description cs:0.5,-0.1)},anchor=north, 
       /pgf/number format/.cd,
            fixed zerofill = false,
            precision=2,
            fixed,
            1000 sep={},
        /tikz/.cd,
        ymax=3
    }
            ]
        
 		\addplot[mark = o, /tikz/solid, every mark/.append style={solid} ] table[x index={0},y  expr=((\thisrowno{1})) ,col sep=comma] {attack7summary.csv};
  		\addplot[mark = triangle, /tikz/densely dotted, every mark/.append style={solid} ] table[x index={0},y  expr=((\thisrowno{2})),col sep=comma] {attack7summary.csv};
        \legend{$f_{\mathit{int}}$, $f_{\mathit{conf}}$}
      \end{axis}
    \end{tikzpicture}
   }
		 \caption{Example \ref{example:enforcement:integrity}}\label{figure:result:integrity}
    \end{subfigure}
    \\
    \begin{subfigure}{.48\textwidth}
		\centering
		\scalebox{0.8}{
		    \begin{tikzpicture}
		    \tikzstyle{every node}=[font=\small]

		      \begin{axis}[
		          width=\linewidth, 
		          grid=major, 
		          grid style={dashed,gray!30},
		          xlabel=Number of tuples, 
		          ylabel={Time [$ms$]},
		          scaled y ticks = false,
		          scaled x ticks = false,
		          legend style={at={(0,1)},anchor=north west},
		          x tick label style={at={(axis description cs:0.5,-0.1)},anchor=north, 
		       /pgf/number format/.cd,
		            fixed zerofill = false,
		            precision=2,
		            fixed,
		            1000 sep={},
		        /tikz/.cd
		    }
		            ]

  				\addplot[mark = o, /tikz/solid, every mark/.append style={solid} ] table[x index={0}, y expr=(\thisrowno{1}) ,col sep=comma] {attack8summary.csv};
  				\addplot[mark = triangle, /tikz/densely dotted, every mark/.append style={solid} ] table[x index={0},y expr=(\thisrowno{2}),col sep=comma] {attack8summary.csv};
  				\addplot[mark=+, /tikz/densely dashdotted , every mark/.append style={solid} ] table[x index={0},y expr=(\thisrowno{6}),col sep=comma] {attack8summary.csv};
  				
		        \legend{$f_{\mathit{int}}$, $f_{\mathit{conf}}$,  Command Execution}
		        
		      \end{axis}
		    \end{tikzpicture}
		 }
		 \caption{Example \ref{example:enforcement:confidentiality}}\label{figure:result:confidentiality}

    \end{subfigure}
    \end{tabular}
    \caption{PDP Execution time.}\label{figure:results}
\end{figure}

To evaluate the feasibility and security of our approach in practice, we implemented our \acf{} in Java.
The prototype, available at \cite{prototype}, implements both our \acf{} and  the operational semantics of our system model.
It relies on the underlying PostgreSQL database for executing the \texttt{SELECT}, \texttt{INSERT}, and \texttt{DELETE} commands.
Note that our prototype also handles  all the differences between the  relational calculus and SQL.
For instance, it translates every relational calculus query into an equivalent \texttt{SELECT} SQL query over the underlying database.
We performed a preliminary experimental evaluation of our prototype.
Our experiments were run on a PC with an Intel i7 processor and 32GB of RAM. 
Note  that we materialized the content of all the views.

Our evaluation has two objectives:\begin{inparaenum}[(1)]
\item to empirically validate that the prototype provides the desired security guarantees, and
\item to evaluate its overhead.
\end{inparaenum}
For (1), we ran the attacks in \S\ref{sect:motivating} against our prototype.
As expected, our \acf{} prevents all the attacks. 
For (2), we simulated Examples \ref{example:enforcement:integrity} and \ref{example:enforcement:confidentiality} on database states where the number of tuples ranges from 1,000 to 100,000.
Figure \ref{figure:results} shows the PDP's execution time.
Our results show that our solution is feasible.
In more detail, $f_{\mathit{int}}$'s execution time  does not depend on the database size, whereas  $f_{\mathit{conf}}$'s execution time does. 
We believe that the overhead introduced by the PDP is acceptable for a proof of concept. 
Even with 100,000 tuples, the \acf{}'s running time is under a second.
In Example \ref{example:enforcement:confidentiality}, $f_{\mathit{conf}}$'s execution time is comparable to the execution time of the user's query.
As Figure \ref{figure:result:confidentiality:detail} shows,  $f_{\mathit{conf}}$'s query rewriting time does not depend on the database's size, whereas $f_{\mathit{conf}}$'s query execution time does.

To improve $f_{\mathit{conf}}$'s performance, one could strike a different balance between simple syntactic checks and our query rewriting solution.
This, however, would result in more restrictive \acf{}s.
We will investigate further optimizations as a future work.

\begin{figure}
    \centering
		\scalebox{0.8}{
		    \begin{tikzpicture}
		\begin{axis}[
		    ybar stacked,
			bar width=5pt,
		   enlarge y limits={0.15,upper},
			 width=\linewidth, 
			 xlabel=Number of tuples, 
			 ylabel={Time [$ms$]},
			 scaled y ticks = false,
			 scaled x ticks = false,
		 	 legend style={at={(0,1)},anchor= north west},
			 x tick label style={at={(axis description cs:0.5,-0.1)},anchor=north, 
				       /pgf/number format/.cd,
				            fixed zerofill = false,
				            precision=2,
				            fixed,
				            1000 sep={},
				        /tikz/.cd
				    },
			  legend columns={2},
		    ]

		\addplot[ybar, black, fill = white] table[x index={0},y expr=(\thisrowno{5}),col sep=comma] {attack8summary.csv};
		\addplot[ybar, black, fill = black] table[x index={0},y expr=(\thisrowno{4}),col sep=comma] {attack8summary.csv};
		  
		\legend{Query execution, Query rewriting}
		
		\end{axis}
		\end{tikzpicture}
		}
    \caption{Example 8 : $f_{\mathit{conf}}$'s execution time.}\label{figure:result:confidentiality:detail}
\end{figure}

%% file: related.tex
\section{Related Work and Discussion}\label{sect:related:work}

We compare our work against two lines of research: data\-base access
control and information flow control.
Both of these  have similar goals, namely preventing the leakage and corruption of sensitive information.

\subsection{Database Access Control}

\smallskip
\noindent
{\bf Discretionary Database Access Control.}
Our framework builds on prior research in database access control~
\cite{wang2007correctness,rizvi2004extending,guarnieri2014optimal} as
well as established notions from database theory, such as determinacy~\cite{nash2010views} and instance-based~determinacy~\cite{Koutris:2012:QDP:2213556.2213582}.

Specifically, our notion of secure judgments extends instance\onlyShortVersion{-}\onlyTechReport{ }based determinacy from database states to runs, while data confidentiality extends existing security notions \cite{wang2007correctness,rizvi2004extending,guarnieri2014optimal} to dynamic settings, where both the database and the policy may change.
Similarly, our indistinguishability notion extends those in \cite{wang2007correctness,guarnieri2014optimal} from database states to runs.
Finally, our formalization of $\auth$ relies on determinacy to decide whether the content of a view is fully determined by a set of other views.
Griffiths and Wade propose a \acf{}~\cite{griffiths1976authorization} that prevents Attacks~\ref{example:view:escalation1} and~\ref{example:view:escalation2} by using syntactic checks and by removing all  views whose owners lack the necessary permissions.
In contrast, we prevent the execution of \texttt{GRANT} and \texttt{REVOKE} operations leading to inconsistent policies.

\smallskip
\noindent
{\bf Mandatory Database Access Control.}
Research on mandatory database access control  has historically 
focused on Multi-Level Security (MLS)~\cite{denning1987multilevel}, where
both the data and the users are associated with security levels,
which are compared to control data access.
Our \acf{} extends the SQL discretionary access control model with additional \emph{mandatory} checks to provide \correctness{} and \confidentiality{}.
In the following, we compare our work with the access control policies and
semantics used by~MLS~systems.

With respect to policies, our work uses the SQL
access control model, where policies are sets of \texttt{GRANT}
statements.  In this model,  users can dynamically modify policies by delegating permissions.
In contrast, MLS policies are usually expressed by labelling each subject and object in the system with labels from a security lattice~\cite{samarati2001access}.
The policy is, in general, fixed (cf.~the \emph{tranquillity principle}~\cite{samarati2001access}).

With respect to semantics, existing MLS solutions are based on the so-called \emph{Truman model}~\cite{rizvi2004extending}, where they transparently modify the commands issued by the users to restrict the access  to only the authorized data.
In contrast, we use the same semantics as SQL,  that is, we execute only the secure commands.
This is called the \emph{Non-Truman model}~\cite{rizvi2004extending}.
For an in-depth comparison of these access control models, see~\cite{rizvi2004extending,guarnieri2014optimal}.
Operationally,  MLS mechanisms  use 
poly-instantiation~\cite{jajodia1990polyinstantiation}, which is neither
supported by the relational model nor by the SQL standard, and requires ad-hoc extensions~\cite{denning1987multilevel, sandhu1998multilevel}.
Furthermore, the operational semantics of MLS systems differs from the standard relational semantics.
In contrast, our operational semantics supports the relational model and
is  directly inspired by SQL.

The above differences influence how security properties are expressed.
Data confidentiality, which relies on a precise characterization of security based on a possible worlds semantics, is a key component of the Non-Truman model (and SQL) access control semantics.
Similarly, \correctness{} requires that any ``write'' operation is authorized according to the policy and is directly inspired by the SQL access control semantics.
The security properties in MLS systems, in contrast,  combine the  multilevel relational semantics~\cite{denning1987multilevel,sandhu1998multilevel} with MLS and BIBA properties~\cite{samarati2001access}.

MLS systems prevent attacks similar to Attacks~\ref{example:integrity:delete:information} and \ref{example:trigger:attack:owner} using poly-instantiated tuples and triggers~\cite{sandhu1998multilevel,smith1993multilevel}, whereas
attacks similar to Attack~\ref{example:trigger:attack:activator} cannot be carried out because triggers with activator's privileges are not supported~\cite{smith1993multilevel}.
The SeaView system~\cite{denning1987multilevel}, which combines discretionary access control and MLS, additionally prevents attacks similar to Attacks~\ref{example:view:escalation1} and~\ref{example:view:escalation2} by relying on Griffiths and Wade's \acf{}~\cite{griffiths1976authorization}.
However, these solutions cannot be applied to SQL databases for the aforementioned reasons.

\subsection{Information Flow Control}

Various authors have applied ideas from information flow control to databases.
Davis and Chen~\cite{davis2010dbtaint} study how~cross-application information flows can be tracked through~data\-bases.
Other researchers~\cite{schoepe2014selinq,corcoran2009cross,li2005practical} present languages for developing secure applications that use databases.
They employ secure type systems to track information flows through databases. 
However, they neither model nor prevent the attacks we identified because they do not account for the advanced database features and the strong attacker model we study in this paper.

Schultz and Liskov~\cite{schultz2013ifdb} extend decentralized information flow control \cite{Myers:1997:DMI:268998.266669} to databases, based on concepts from multi-level security~\cite{denning1987multilevel}.
They identify one attack on data confidentiality that exploits integrity constraints. 
Their solution relies on poly-instantiation~\cite{jajodia1990polyinstantiation} and cannot be applied to SQL databases that do not support multi-level security. 
Their mechanism neither prevents the other attacks we identify nor provides provable and precise security guarantees.

Several researchers have studied attacker models in information flow control~\cite{Giacobazzi:2004:ANP:964001.964017, askarov2012learning}.
Giacobazzi and Mastroeni~\cite{Giacobazzi:2004:ANP:964001.964017} model attackers as data-flow analysers that observe the program's behaviour, whereas Askarov and Chong~\cite{askarov2012learning} model attackers as automata that observe the program's events.
They both model passive attackers, who can observe, but do not influence, the program's execution.
In contrast, our attacker is active and interacts with the database.

\subsection{Discussion}

Historically, database access control and information flow control rely on different foundations, formalisms, security notions, and techniques.
We see our paper as a starting point for bridging these topics: we combine database access control theory  with an operational semantics and an attacker model, which are common in information flow control, but have been largely ignored in database access control.
We thereby give a precise logical characterization of the attacker's capabilities and of a judgment's security. 
Furthermore, our indistinguishability notion has similarities with the low-equivalence notions used in~\cite{askarov2007gradual, askarov2009tight, bohannon2009reactive}, whereas both data confidentiality and the notion of secure judgments have a precise characterization as instances of non-interference~\cite{sabelfeld2003language,GoguenM82}; see~\techReportAppendix{app:data:conf:non:interference}~for more details.

We believe our framework provides a basis for (1) further investigating the connections between these two topics, (2) applying  information flow mechanisms, such as type systems or multi-execution~\cite{devriese2010noninterference}, to database access control, and (3) investigating how integrity notions used in information flow control can best be applied to databases.

%% file: conclusion.tex
\section{Conclusion}\label{sect:conclusion}

Motivated by practical attacks against existing databases, we have initiated several new research directions.
First, we developed the idea that attacker models should be studied and formalized for databases. Rather than being implicit,
the relevant models must be made explicit, just like when analysing security in other domains.
In this respect, we presented a concrete attacker model that accounts for relevant features of modern databases, like triggers and views,
and attacker inference capabilities.

Second, access control mechanisms must be designed to be secure, and provably so, with respect to the formalized attacker capabilities.
This requires research on mechanism design, complemented by a formal operational semantics of databases as a basis for security proofs.
We presented such a mechanism, proved that it is secure, and built and evaluated a prototype of it in PostgreSQL.
As a future work, we will extend our framework and our \acf{} to directly support  the SQL language, and we will investigate efficiency improvements for our \acf{}.

\smallskip
{
\noindent
{\bf Acknowledgments.} 
We thank \`{U}lfar Erlingsson, Erwin Fang,  Andreas Lochbihler, Ognjen Maric, Mohammad Torabi Dashti, Dmitriy Traytel, Petar Tsankov,  Thilo Weghorn, Der-Yeuan Yu,   Eugen Zalinescu, as well as  the anonymous reviewers for their comments. 
}

%% file: lts_long_small_step.tex
\newpage
\section{Formalizing the System Model}\label{app:lts}
In this section, we precisely formalize the system model introduced in \S\ref{sect:lts}.
We first introduce some auxiliary definitions about queries, views, privileges, and triggers.
Afterwards, we introduce the concept of partial state.
Then, we formalize contexts and we refine the notion of $M$-state given in \S\ref{sect:lts}.
Finally, we formalize the transition relation $\to_f$ together with some auxiliary predicates and functions.

\subsection{Auxiliary definitions on queries, views, \\ privileges, and triggers}
Triggers can give rise to non-terminating executions, for example when the  action associated with trigger $t_{1}$ activates trigger $t_{2}$, which in turn activates $t_{1}$.
We say that a set $T$ of triggers  is \emph{safe} iff no trigger in $T$ can activate another trigger in $T$. 
Note that safety ensures termination.
Even though this termination condition is simple, it is sufficient for the purpose of this paper.
Note that more complex and permissive termination conditions do not influence our results.
We say that a set of triggers $T$ is \emph{safe}, denoted by  $\mathit{safe}(T)$, iff for all triggers $t_1, t_2 \in T$:
\begin{compactitem}
\item if the $t_1$'s activation event is an \texttt{INSERT} on the table $R$, then $t_2$'s action is not of the form $\langle \mathtt{INSERT}, R, \overline{t}\rangle$, or
\item if the $t_1$'s activation event is a \texttt{DELETE} on the table $R$,  then  $t_2$'s action is not of the form $\langle \mathtt{DELETE}, R, \overline{t}\rangle$.
\end{compactitem}
Let $D$ be a database schema, $U$ be a set of users, $t$ be a trigger over $D$, and $V$ be a set of views over $D$.
We say that $t$ is a \emph{$U$-trigger}, denoted by $\mathit{usersIn}(t,U)$, if and only if $\mathit{owner}(t) \in U$ and $t$'s statement mentions just users in $U$.
We say that \emph{a query $q$ is defined over $D$ and $V$}, denoted by $\mathit{defined}(q,D,V)$, iff all the predicates in $q$ are either tables in $D$ or views in $V$.
We say that \emph{a privilege $p$ is defined over $D$ and $V$}, denoted by $\mathit{defined}(p,D,V)$, iff the table or view referred in $p$ is in $D \cup V$.
We say that \emph{a view $v$ is defined over $D$ and $V$}, denoted by $\mathit{defined}(v,D,V)$, iff its definition is defined over $D$ and $V$.
Finally, we say that \emph{a trigger $t$ is defined over $D$ and $V$}, denoted by $\mathit{defined}(t,D,V)$, iff (1) the table on which $t$ is defined is in $D$,
(2) $t$'s \texttt{WHEN} condition is defined over $D$ and $V$, and
(3) $t$'s action refers only to tables and views in $D \cup V$.

\subsection{Revoke Semantics}

We now define the function $\mathit{revoke}$ that models the semantics of SQL's \texttt{REVOKE} statements with cascade.
In the following, let $S$ be a security policy, i.e., a set of \texttt{GRANT}s, $u_1, u_2, u_3, u_4$, $u$, and $u'$ be user identifiers, $\mathit{op}, \mathit{op}' \in \{\oplus, \oplus^*\}$, and $p$ be a privilege.
We say that a \emph{chain} is a sequence of grants $g_1\concat g_2 \concat \ldots \concat g_n$ such that 
(1) $g_1 = \langle \mathit{op}', u_4, p, \mathit{start}\rangle$,
(2) if $p \neq \langle \mathtt{SELECT}, V \rangle$, where $V$ is a view with owner's privileges, then $\mathit{start} = \admin$, whereas if $p = \langle \mathtt{SELECT}, V \rangle$, then $\mathit{start} \in \{\admin , \mathit{owner}(V)\}$, and
(3) for each $1 \leq i \leq n-1$, 
 $g_i = \langle \oplus^*, u_2, p, u_1\rangle$ and
 $g_{i+1} = \langle \mathit{op}, u_3, p, u_2\rangle$. 
We first define the $\mathit{chain}$ function that takes as input a policy $S$ and constructs all possible chains.
\begin{align*}
\mathit{chain}(S) := &
\{ \langle \mathit{op}, u, p, u'\rangle \in S \,|\, u' = \admin \} \cup \\
&\quad\{ \langle \mathit{op}, u, \langle \mathtt{SELECT}, V\rangle, u'\rangle \in S \,|\, V = \langle v, o, q, O \rangle \\
&\qquad\wedge u' = o  \} \cup \\
&\bigcup_{g_1\concat \ldots \concat g_n \in \mathit{chain}(S)} \{ g_1\concat \ldots \concat g_n \concat g \,|\, g \in S \wedge\\
&\qquad g = \langle \mathit{op} , u , p , u' \rangle \wedge g_n = \langle \oplus^*, u' , p , u''\rangle \wedge\\
&\qquad \forall i \in \{1,\ldots, n\}.\, g_i \neq g \}.
\end{align*}

The function $\mathit{filter}$ takes as input a set of chains $C$ and a grant $g$ and returns as output the set of all chains in $C$ that do not contain $g$:
\begin{align*}
\mathit{filter}(C,g) := & \{ g_1\concat \ldots \concat g_n \in C \,|\, \forall i \in \{1,\ldots, n\}.\, g_i \neq g   \}.
\end{align*}
The function $\mathit{policy}$ takes as input a set of chains and constructs a policy, i.e., a set of grants, out of it:
\begin{align*}
\mathit{policy}(C) := & \bigcup_{g_1\concat \ldots \concat g_n \in C} \bigcup_{1 \leq i \leq n} \{g_i\}.
\end{align*}
Finally, the function $\mathit{revoke}$, which models the semantics of the \texttt{REVOKE} command with cascade, is as follows:
\begin{align*}
\mathit{revoke}(S, u, p , u') := & \mathit{policy}(\mathit{filter}(\mathit{chain}(
\mathit{policy}(\mathit{filter}(\\
&\mathit{chain}(S), \langle \oplus, u, p, u'\rangle))), \langle \oplus^*, u , p , u' \rangle)).
\end{align*}
Given a policy $S$, $\mathit{revoke}(S, u, p , u')$ denotes the policy obtained by applying $\langle \ominus, u, p, u'\rangle$ to $S$.

\subsection{Partial States}
An $M$\emph{-partial state} is a tuple $\langle \mathit{db}, U, \mathit{sec}, T,V\rangle$ such that $\mathit{db} \in \Omega_{D}^{\Gamma}$ is a data\-base state, $U \subset {\cal U}$ is a finite set of users such that $\mathit{admin}\in U$, $\mathit{sec} \in  {\cal S}_{U,D}$ is a security policy, $T$ is a finite set of safe triggers over $D$,  and $V$ is a finite set of views over $D$. 
We denote by $\Pi_{M}$ the set of all $M$-partial states.
Given an $M$-state $s = \langle \mathit{db}, U, \mathit{sec}, T,V,c \rangle$, we denote by $\mathit{pState}(s)$ the $M$-partial state $\langle \mathit{db}, U, \mathit{sec}, T,V\rangle$ obtained from $s$ by dropping the context $c$.

\subsection{Contexts}
Let $M = \langle D, \Gamma\rangle$ be a system configuration and $u$ be a user.
An \emph{$(M,u)$-action effect} is a tuple $\langle \mathit{act}, \mathit{accDec}, \mathit{res}, E\rangle$, where $\mathit{act} \in {\cal A}_{D,u}$ is an action, $\mathit{accDec} \in \{\top,\bot\}$ is the access control decision for that action (where $\top$ stands for \emph{permit} and $\bot$ stands for \emph{deny}), $\mathit{res} \in \{\top,\bot\}$ is the action result, and $E \subseteq \Gamma$ is the set of integrity constraints violated by the action. 
We denote by $\Omega_{M,u}^{\mathit{actEff}}$ the set of all $(M,u)$-action effects and by $\Omega_{M,U}^{\mathit{actEff}}$, for some $U \subseteq {\cal U}$, the set $\bigcup_{u \in U} \Omega_{M,u}^{\mathit{actEff}}$.
An \emph{$(M,u)$-trigger effect} is a triple $\langle t, \mathit{when}, \mathit{stmt} \rangle$ where $t \in {\cal TRIGGER}_{D}$ is a trigger, $\mathit{when} \in \Omega_{M,u}^{\mathit{actEff}}$ is the action effect associated with the \texttt{WHEN} condition of the trigger, and  $\mathit{stmt}\in \Omega_{M,u}^{\mathit{actEff}} \cup \{\epsilon\}$ is the action effect associated with the statement in the trigger's body.
We denote by $\Omega_{M,u}^{\mathit{triEff}}$ the set of all $(M,u)$-trigger effects and by $\Omega_{M,U}^{\mathit{triEff}}$, for some $U \subseteq {\cal U}$, the set $\bigcup_{u \in U} \Omega_{M,u}^{\mathit{triEff}}$.

An \emph{$M$-pending trigger transaction} is a 4-tuple $\langle s, \overline{t}, u, \mathit{tr} \rangle$, where $s \in \Pi_{M} \cup \{\epsilon\}$ is an $M$-partial state representing the ``roll-back state'', i.e., the state that we must restore in case a roll-back happens, $\overline{t} \in \{\epsilon\} \cup \bigcup_{n \in \mathbb{N}^{+}} \mathbf{dom}^{n}$ is the tuple involved in the event that has fired the transaction, $u \in {\cal U} \cup \{\epsilon\}$ is the user that has activated the triggers in the transactions, and $\mathit{tr} \in {\cal TRIGGER}_{D}^{*}$ is a sequence of triggers.
Note that we denote by $\cdot$ the concatenation operation between strings over ${\cal TRIGGER}_{D}^{*}$, by $\epsilon$ the empty string in ${\cal TRIGGER}_{D}^{*}$, and by  $\langle \epsilon, \epsilon, \epsilon, \epsilon\rangle$ the empty $M$-pending transaction.

An \emph{$M$-history} $h$ is a  sequence of action effects and trigger effects, i.e., $ h \in (\Omega_{M,{\cal U}}^{\mathit{actEff}} \cup \Omega_{M,{\cal U}}^{\mathit{triEff}})^{*}$.
We denote by ${\cal H}_{M}$ the set of all $M$-histories, by $\cdot$ the concatenation operation over ${\cal H}_{M}$, and by $\epsilon$ the empty history.

We are now ready to formally define contexts.
Let $M=\langle D, \Gamma\rangle$ be a system configuration. 
An \emph{$M$-context} is a tuple $\langle h,  \mathit{actEff}, \mathit{tr} \rangle$, where $h \in {\cal H}_{M}$ models the system's history, $\mathit{actEff} \in \Omega_{M,{\cal U}}^{\mathit{actEff}}  \cup \Omega_{M,{\cal U}}^{\mathit{triEff}} \cup \{\epsilon\}$ describes the effect of the last action, i.e., whether the action has been accepted by the access control mechanism and the action's result, and $\mathit{tr}$ is an $M$-pending transaction.
Furthermore, the empty context $\epsilon$ is the element $\langle  \epsilon, \epsilon, \langle  \epsilon, \epsilon, \epsilon, \epsilon \rangle \rangle$.

\subsubsection{Auxiliary Functions over contexts}
Given an $M$-context $c=\langle h,  \mathit{actEff}, \mathit{tr} \rangle$, we denote by $\mathit{secEx}$ the following function, which returns $\top$ if the last action has caused a security exception.
\[
	\mathit{secEx}(\langle h,  \mathit{aE}, \mathit{tr} \rangle) = \left\{ 
  \begin{array}{l l}
  \top & \text{if }\mathit{aE} = \langle \mathit{act}, \bot, \mathit{res}, E\rangle\\
  \top & \text{if }\mathit{aE} = \langle t, \langle \mathit{act}, \bot, \mathit{res}, E\rangle, \epsilon \rangle\\
  \top & \text{if }\mathit{aE} = \langle t, \mathit{when} , \langle \mathit{act}, \bot, \mathit{res}, E\rangle \rangle\\
  \bot & \text{otherwise}\\
  \end{array}\right.
\]
Similarly, we denote by $\mathit{Ex}(c)$ the function extracting the integrity constraints violated by the last action.
\[
	\mathit{Ex}(\langle h,  \mathit{aE}, \mathit{tr} \rangle) = \left\{ 
  \begin{array}{l l}
  E & \text{if }\mathit{aE} = \langle \mathit{act}, \mathit{aC}, \mathit{res}, E\rangle\\
  E & \text{if }\mathit{aE} = \langle t, \mathit{when} , \langle \mathit{act}, \mathit{aC}, \mathit{res}, E\rangle \rangle\\
  \emptyset & \text{otherwise}
  \end{array}\right.
\]
We also denote by $\mathit{res}(c)$ the function extracting the last action's result:
\[
	\mathit{res}(\langle h,  \mathit{aE}, \mathit{tr} \rangle) = \left\{ 
  \begin{array}{l l}
  \mathit{res} & \text{if }\mathit{aE} = \langle \mathit{act}, \mathit{aC}, \mathit{res}, E\rangle\\
  \mathit{aC}  & \text{if }\mathit{aE} = \langle t, \langle \mathit{act}, \mathit{aC}, \mathit{res}, E\rangle, \epsilon \rangle\\
    \mathit{aC} \wedge \mathit{aC}' & \text{if }\mathit{aE} = \langle t, \langle \mathit{act}, \mathit{aC}, \mathit{res}, E\rangle,\\
 \quad  \wedge \mathit{res}' & \quad \langle \mathit{act}', \mathit{aC}', \mathit{res}', E'\rangle \rangle \wedge \\
  & \quad \langle \mathit{act}', \mathit{aC}', \mathit{res}', E'\rangle \neq \epsilon 
  \end{array}\right.
\]

Similarly, we denote by $\mathit{acA}(c)$ and $\mathit{acC}(c)$ the functions that extract the access control decision for the trigger's action and condition:
\[
	\mathit{acA}(\langle h,  \mathit{aE}, \mathit{tr} \rangle) = \left\{ 
  \begin{array}{l l}
     \mathit{aC}' & \text{if }\mathit{aE} = \langle t, \langle \mathit{act}, \mathit{aC}, \mathit{res}, E\rangle,\\
 & \quad \langle \mathit{act}', \mathit{aC}', \mathit{res}', E'\rangle \rangle \wedge \\
  & \quad \langle \mathit{act}', \mathit{aC}', \mathit{res}', E'\rangle \neq \epsilon \\
  \bot & \text{otherwise}
  \end{array}\right.
\]
\[
	\mathit{acC}(\langle h,  \mathit{aE}, \mathit{tr} \rangle) = \left\{ 
  \begin{array}{l l}
	\mathit{aC}  & \text{if }\mathit{aE} = \langle t, \langle \mathit{act}, \mathit{aC}, \mathit{res}, E\rangle, \epsilon \rangle\\
    \mathit{aC}  & \text{if }\mathit{aE} = \langle t, \langle \mathit{act}, \mathit{aC}, \mathit{res}, E\rangle,\\
  & \quad \langle \mathit{act}', \mathit{aC}', \mathit{res}', E'\rangle \rangle \wedge \\
  & \quad \langle \mathit{act}', \mathit{aC}', \mathit{res}', E'\rangle \neq \epsilon \\
  \bot & \text{otherwise}
  \end{array}\right.
\]

We denote by $\mathit{invoker}(c)$ the function extracting the user in the transaction, i.e., $\mathit{invoker}(\langle h,  \mathit{aE}, \langle  s, \overline{t}, u, \mathit{trL}\rangle \rangle) = u$. 
Similarly, we denote by $\mathit{tpl}(c)$ the function extracting the tuple that has fired the transaction, namely $\mathit{tpl}(\langle h,  \mathit{aE}, \langle  s, \overline{t}, u, \\ \mathit{trL}\rangle \rangle) = \overline{t}$, by $\mathit{triggers}(c)$ the function extracting the list of triggers, i.e., $\mathit{triggers}(\langle h,  \mathit{aE}, \langle  s, \overline{t}, u, \mathit{trL}\rangle \rangle) =  \mathit{trL}$, and by $\mathit{trigger}(c)$, or $\mathit{tr}(c)$ for short, the first trigger in the sequence $\mathit{triggers}(c)$.

\subsection{States}

We can now define $M$-states.
Let $M = \langle D, \Gamma\rangle$ be a system configuration.
An $M$\emph{-state} is a tuple $\langle \mathit{db}, U, \mathit{sec}, T,V, c\rangle$ such that $\mathit{db} \in \Omega_{D}^{\Gamma}$ is a data\-base state, $U \subset {\cal U}$ is a finite set of users such that $\mathit{admin}\in U$, $\mathit{sec} \in  {\cal S}_{U,D}$ is a security policy, $T$ is a finite set of safe triggers over $D$ such that for any trigger $t \in T$, both $\mathit{usersIn}(t,U)$ and $\mathit{defined}(t,D,V)$ hold, $V$ is a finite set of views over $D$ such that (1) there are no cyclic dependencies between the views in $V$, and (2) for any view $v \in V$, $\mathit{defined}(t,D,V')$, for some $V'\subseteq V$, and $v$'s owner is in $U$, and $c \in {\cal C}_{M}$ is an $M$-context.

In Section \ref{sect:lts}, we denoted an $M$-state as a tuple $\langle \mathit{db}, \mathit{sec}, U,\\T,V, c \rangle$, where $c = \langle h,  \mathit{actEff}, \mathit{tr} \rangle$ is an element of ${\cal C}_{M}$.
In the following, instead of representing states as $\langle  \mathit{db}, \mathit{sec}, U,T,V,\\ \langle h,  \mathit{aE}, \mathit{tr} \rangle\rangle$, we represent them as $\langle \mathit{db}, \mathit{sec}, U,T,V,  h,  \mathit{aE}, \mathit{tr} \rangle$.
Given an $M$-state $s:= \langle \mathit{db}, \mathit{sec}, U,T,V,  h,  \mathit{aE}, \mathit{tr} \rangle$, we denote by $\mathit{ctx}(s)$ the context $\langle h,  \mathit{aE}, \mathit{tr} \rangle$.
With a slight abuse of notation, we extend the functions $\mathit{Ex}$, $\mathit{secEx}$,  $\mathit{res}$, $\mathit{tpl}$, $\mathit{acA}$, $\mathit{acC}$, $\mathit{triggers}$, $\mathit{tr}$, and $\mathit{invoker}$ from contexts to $M$-states, e.g., $\mathit{Ex}(s)$ is just $\mathit{Ex}(\mathit{ctx}(s))$. 
Furthermore, given an $M$-state $s:= \langle \mathit{db}, \mathit{sec}, U,T,V,  h,  \mathit{aE}, \mathit{tr} \rangle$, we use a dot notation to refer to its components. For instance, we use $s.\mathit{db}$ to refer to the database's state in $s$ and $s.\mathit{sec}$ to refer to the policy in $s$.

\subsection{Transition Relation $\to_f$}
The transition rules describing the $\rightarrow_{f}$ transition relation are shown in Figures \ref{table:rules:lts:1}--\ref{table:rules:lts:8}.
The $\rightarrow_{f}$ relation is, thus, the smallest relation satisfying all the inference rules. 
Note that we ignore the $\mathit{upd}$ function introduced in Section \ref{sect:lts} since the rules explicitly encode the changes to the contexts.

We now define the functions we used in the rules in  Figures \ref{table:rules:lts:1}--\ref{table:rules:lts:8}.
The $\mathit{getActualUser}(m,\mathit{invk}, \mathit{ow})$ function, where $m\in \{A,O\}$ and $\mathit{invk}, \mathit{ow} \in {\cal U}$, is defined as follows:
\[
\mathit{getActualUser}(m,\mathit{invk}, \mathit{ow}) = \left\{ 
  \begin{array}{l l}
  \mathit{invk} & \text{if }m = A\\
  \mathit{ow} & \text{if }m = O\\
  \end{array}\right.
\]
  
The $\mathit{ID}$ function takes as input an action $\mathit{act} \in {\cal A}_{D,{\cal U}}$ and returns $\top$ if $\mathit{act}$ is either $\langle u, \mathtt{INSERT},R, \overline{t}\rangle$ or $\langle u, \mathtt{DELETE}, R,  \overline{t}\rangle$, for some $u \in {\cal U}$, $R \in D$, and $\overline{t} \in \mathbf{dom}^{|R|}$. 
The function   $\mathit{ID}$ returns $\bot$ otherwise.

The $\mathit{apply}$ function, which takes as input an action $\mathit{act} \in {\cal A}_{M,{\cal U}}$ that is either an \texttt{INSERT} or a \texttt{DELETE} action and a database state $\mathit{db} \in \Omega_{D}$, is defined as follows:
\[
	\mathit{apply}(\mathit{act},  \mathit{db}) = \left\{ 
  \begin{array}{l l}
  \mathit{db}[R \oplus \overline{t}] & \text{if } \mathit{act} = \langle u, \mathit{INSERT}, R, \overline{t}\rangle\\
  \mathit{db}[R \ominus \overline{t}] & \text{if }\mathit{act} = \langle u, \mathit{DELETE}, R, \overline{t}\rangle\\
  \end{array}\right.
\]

Let $t = \langle \mathit{id},\mathit{ow},  \mathit{ev}, R', \phi, \mathit{stmt}, m\rangle$ be a trigger and $R$ be a relation schema.
We denote $t$'s owner by $\mathit{owner}(t)$, i.e., $\mathit{owner}(t) = \mathit{ow}$.
Similarly, given a view $V$, we denote by $\mathit{owner}(V)$ the owner of $V$.
We also denote by $\overline{x}^{|R|}$ the tuple of variables $\langle x_{1}, \ldots, x_{|R|}\rangle$.
Furthermore, given a tuple $\overline{t} := \langle t_{1}, \ldots, t_{n} \rangle $, we denote by $\overline{t}(i)$ the $i$-th value $t_{i}$.
Finally, we denote by  $\overline{t}[\overline{x}^{|R|} \mapsto \overline{v}]$, where $\overline{t}$ is a tuple of values in $\mathbf{dom}$ and variables in $\{x_{1}, \ldots, x_{|R|}\}$ and $\overline{v}$ is a tuple in $\mathbf{dom}^{|R|}$, the tuple $\overline{z} \in \mathbf{dom}^{n}$ obtained as follows: for each $i \in \{1, \ldots, n\}$, if $\overline{t}(i) = x_{j}$, where $x_{j} \in \{x_{1}, \ldots, x_{|R|}\}$, then $\overline{z}(i) = \overline{v}(j)$, and otherwise $\overline{z}(i) = \overline{t}(i)$.
We are now ready to define the function $\mathit{getAction}$, which takes as input the trigger's statement $\mathit{stmt}$, a user $u$, and a tuple $\overline{t}' \in \mathbf{dom}^{|R'|}$, and returns the concrete action executed by the system.
Formally, $\mathit{getAction}$ is as follows:
\begin{compactitem}
\item $\mathit{getAction}(\langle \mathtt{INSERT}, R, \overline{t}\rangle, u, \overline{t}') = \langle u, \mathtt{INSERT}, R, \overline{t}[\overline{x}^{|R'|} \\\mapsto \overline{t'}]\rangle$, 
\item $\mathit{getAction}(\langle \mathtt{DELETE}, R, \overline{t}\rangle, u, \overline{t}') = \langle u, \mathtt{DELETE}, R, \overline{t}[\overline{x}^{|R'|} \\\mapsto \overline{t'}]\rangle$, and
\item $\mathit{getAction}(\langle \mathit{op}, u, p\rangle, u, \overline{t}') = \langle \mathit{op}, u, p,  u\rangle$, where $\mathit{op} \in \{\ominus,\oplus, \oplus^{*}\}$.
\end{compactitem}

We assume there is a total-order relation $\preceq_{\cal T}$ over ${\cal T}$. 
We use this ordering to determine the order in which triggers are executed.
Given a set of triggers $T$ and a database schema $D$, we denote by $\mathit{filter}(T,\mathit{ev},R)$, where $\mathit{ev} \in \{\mathit{INS}, \mathit{DEL}\}$ and $R \in D$, the sequence  of triggers in $T$ (ordered according to $\preceq_{\cal T}$)  whose event is $\mathit{ev}$ and whose relation schema is $R$.

\begin{figure*}

\centering
\begin{tabular}{c}
$\infer[\text{\begin{tabular}{c}Add\\ User\\ Success\end{tabular}}]
{
\hfill \langle \mathit{db}, U, \mathit{sec}, T, V, h, \mathit{aE}, \langle \mathit{rS}, \overline{t}', u', \epsilon\rangle \rangle \xrightarrow{\langle \mathit{admin},\mathtt{ADD\_USER},u\rangle}_{f} \langle \mathit{db}, U \cup \{u\}, \mathit{sec}, T, V, h \concat \mathit{aE}, \mathit{aE}', \langle \epsilon, \epsilon, \epsilon,  \epsilon\rangle\rangle \hfill}
{
\hfill \mathit{admin}\in U \hfill \quad
\hfill \mathit{aE}' = \langle \langle \mathit{admin}, \mathtt{ADD\_USER}, u\rangle, \top, \top, \emptyset \rangle \hfill \quad
\hfill \mathit{f}(\langle \mathit{db}, U, \mathit{sec}, T,V, h, \mathit{aE}, \langle \mathit{rS}, \overline{t}', u', \epsilon\rangle\rangle,\langle \mathit{admin}, \mathtt{ADD\_USER}, u\rangle) = \top \hfill
}$ \\
 \\
 
 $\infer[\text{\begin{tabular}{c}Add\\ User\\ Deny\end{tabular}}]
{
\hfill \langle \mathit{db}, U, \mathit{sec}, T, V, h, \mathit{aE}, \langle \mathit{rS}, \overline{t}', u', \epsilon\rangle \rangle \xrightarrow{\langle \mathit{admin},\mathtt{ADD\_USER},u\rangle}_{f} \langle \mathit{db}, U, \mathit{sec}, T, V, h \concat \mathit{aE}, \mathit{aE}', \langle \epsilon, \epsilon, \epsilon,  \epsilon\rangle\rangle \hfill}
{
\hfill \mathit{admin}\in U \hfill \quad
\hfill \mathit{aE}' = \langle \langle \mathit{admin}, \mathtt{ADD\_USER}, u\rangle, \bot, \bot, \emptyset \rangle \hfill \quad
\hfill \mathit{f}(\langle \mathit{db}, U, \mathit{sec}, T,V, h, \mathit{aE}, \langle \mathit{rS}, \overline{t}', u', \epsilon\rangle\rangle, \langle \mathit{admin}, \mathtt{ADD\_USER}, u\rangle) = \bot \hfill
}$ \\
 \\

$\infer[\text{\begin{tabular}{c} \texttt{SELECT} \\\text{ Success}\end{tabular}}]
{\langle \mathit{db}, U, \mathit{sec}, T, V, h, \mathit{aE}, \langle \mathit{rS}, \overline{t}', u', \epsilon\rangle \rangle \xrightarrow{\langle u, \mathtt{SELECT},  q\rangle}_{f} \langle \mathit{db}, U, \mathit{sec}, T, V, h \concat \mathit{aE}, \mathit{aE}', \langle \epsilon,\epsilon, \epsilon, \epsilon\rangle\rangle}
{
\hfill \mathit{u}\in U \hfill \quad
\hfill \mathit{f}(\langle \mathit{db}, U, \mathit{sec}, T,V, h, \mathit{aE}, \langle \mathit{rS}, \overline{t}', u', \epsilon\rangle\rangle, \langle u, \mathtt{SELECT},  q\rangle) = \top \hfill \quad
\hfill  [q]^{\mathit{db}} = v \hfill \\
\hfill  \mathit{aE}' = \langle \langle u, \mathtt{SELECT},  q\rangle, \top, v, \emptyset \rangle \hfill \quad
\hfill \mathit{defined}(q,D,V) \hfill }$ \\
\\

$\infer[\text{\begin{tabular}{c} \texttt{SELECT}\\ \text{ Deny}\end{tabular}}]
{\langle \mathit{db}, U, \mathit{sec}, T, V, h, \mathit{aE}, \langle \mathit{rS}, \overline{t}', u', \epsilon\rangle \rangle \xrightarrow{\langle u, \mathtt{SELECT},  q\rangle}_{f} \langle \mathit{db}, U, \mathit{sec}, T, V, h \concat \mathit{aE}, \mathit{aE}', \langle \epsilon,\epsilon, \epsilon, \epsilon\rangle\rangle}
{
\hfill \mathit{u}\in U \hfill  \quad 
\hfill \mathit{f}(\langle \mathit{db}, U, \mathit{sec}, T,V, h, \mathit{aE}, \langle \mathit{rS},\overline{t}', u', \epsilon\rangle\rangle, \langle u, \mathtt{SELECT},  q\rangle) =\bot \hfill \\ 
\hfill \mathit{aE}' = \langle \langle u, \mathtt{SELECT},  q\rangle, \bot, \bot, \emptyset \rangle \hfill  \quad
\hfill \mathit{defined}(q,D,V) \hfill }$ \\
\\
\end{tabular}
\captionof{figure}{Rules defining the $\rightarrow_{f}$ relation for \texttt{SELECT} and \texttt{ADD USER}}\label{table:rules:lts:1}
\end{figure*}

\begin{figure*}
\centering
\begin{tabular}{c}

$\infer[\text{\begin{tabular}{c} \texttt{INSERT}\\ Success 1\end{tabular}}]
{\langle \mathit{db}, U, \mathit{sec}, T, V, h, \mathit{aE}, \langle \mathit{rS}, \overline{t}', u', \epsilon\rangle \rangle \xrightarrow{\langle u, \mathtt{INSERT}, R, \overline{t}\rangle}_{f} \langle \mathit{db}[R \oplus \overline{t}], U, \mathit{sec}, T, V, h\concat \mathit{aE}, \mathit{aE}', \langle \epsilon, \epsilon, \epsilon,  \epsilon\rangle \rangle}
{
\hfill u\in U \hfill   \quad
\hfill R \in D \hfill \quad
\hfill \mathit{f}(\langle \mathit{db}, U, \mathit{sec}, T,V, h, \mathit{aE}, \langle \mathit{rS}, \overline{t}', u', \epsilon\rangle \rangle, \mathit{act}) = \top \hfill  \quad
\hfill  \mathit{act}= \langle u, \mathtt{INSERT}, R, \overline{t}\rangle \hfill  \\
\hfill \mathit{db} [R \oplus \overline{t}] \in \Omega_{D}^{\Gamma} \hfill \quad
\hfill \mathit{aE}' = \langle \mathit{act}, \top, \top, \emptyset \rangle \hfill \quad
\hfill \mathit{filter}(T,\mathit{INS},R) = \epsilon \vee \overline{t} \in \mathit{db}(R)\hfill 
}$
 \\\\
 
 $\infer[\text{\begin{tabular}{c} \texttt{INSERT}\\ Success 2\end{tabular}}]
{\hfill \langle \mathit{db}, U, \mathit{sec}, T, V, h, \mathit{aE}, \langle \mathit{rS}, \overline{t}', u', \epsilon\rangle \rangle \xrightarrow{\langle u, \mathtt{INSERT}, R, \overline{t}\rangle}_{f} \langle \mathit{db}[R \oplus \overline{t}], U, \mathit{sec}, T, V, h\concat \mathit{aE}, \mathit{aE}', \langle \mathit{rS}',\overline{t},u, \mathit{tr}\rangle \rangle\hfill }
{
\hfill u\in U \hfill   \quad
\hfill R \in D \hfill \quad
\hfill \mathit{f}(\langle \mathit{db}, U, \mathit{sec}, T,V, h, \mathit{aE}, \langle \mathit{rS}, \overline{t}', u', \epsilon\rangle\rangle, \mathit{act}) = \top \hfill  \quad
\hfill  \mathit{act}= \langle u, \mathtt{INSERT}, R, \overline{t}\rangle \hfill  \quad
\hfill \mathit{db} [R \oplus \overline{t}] \in \Omega_{D}^{\Gamma} \hfill  \\
\hfill \mathit{aE}' = \langle \mathit{act}, \top, \top, \emptyset \rangle \hfill  \quad
\hfill \mathit{tr} =\mathit{filter}(T,\mathit{INS},R) \hfill \quad
\hfill \mathit{tr} \neq \epsilon \hfill \quad 
\hfill \overline{t} \not\in \mathit{db}(R) \hfill \quad
\hfill \mathit{rS}' = \langle \mathit{db}, U, \mathit{sec}, T, V \rangle\hfill 
}$
 \\\\

 $\infer[\text{\begin{tabular}{c} \texttt{INSERT}\\ Exception\end{tabular}}]
{\hfill \langle \mathit{db}, U, \mathit{sec}, T, V, h, \mathit{aE}, \langle \mathit{rS}, \overline{t}', u', \epsilon\rangle \rangle \xrightarrow{ \langle u, \mathtt{INSERT}, R, \overline{t} \rangle }_{f} \langle \mathit{db}, U, \mathit{sec}, T, V, h \concat \mathit{aE}, \mathit{aE}',  \langle \epsilon, \epsilon, \epsilon,  \epsilon\rangle \rangle\hfill }
{
\hfill u\in U \hfill \quad
\hfill R \in D \hfill \quad
\hfill \mathit{f}(\langle \mathit{db}, U, \mathit{sec}, T,V, h, \mathit{aE}, \langle \mathit{rS}, \overline{t}', u', \epsilon\rangle \rangle, \langle u, \mathtt{INSERT}, R, \overline{t}\rangle) = \top \hfill  \\
\hfill E'= \{\phi \in \Gamma | [\phi]^{\mathit{db} [R \oplus \overline{t}]} =\bot\} \hfill  \quad 
\hfill E' \neq \emptyset \hfill  \quad
\hfill  \mathit{aE}' = \langle \langle u, \mathtt{INSERT}, R, \overline{t}\rangle, \top, \bot, E' \rangle \hfill 
}$
\\\\

 $\infer[\text{\begin{tabular}{c} \texttt{INSERT}\\ Deny\end{tabular}}]
{\hfill \langle \mathit{db}, U, \mathit{sec}, T, V, h, \mathit{aE}, \langle \mathit{rS}, \overline{t}', u', \epsilon\rangle \rangle \xrightarrow{ \langle u, \mathtt{INSERT}, R, \overline{t} \rangle }_{f}  \langle \mathit{db}, U, \mathit{sec}, T, V, h \concat \mathit{aE},   \mathit{aE}',  \langle \epsilon, \epsilon, \epsilon, \epsilon\rangle \rangle \hfill }
{
\hfill u\in U \hfill  \quad
\hfill R \in D \hfill \quad
\hfill \mathit{f}(\langle \mathit{db}, U, \mathit{sec}, T,V, h, \mathit{aE}, \langle \mathit{rS},\overline{t}',  u', \epsilon\rangle\rangle, \langle u, \mathtt{INSERT}, R, \overline{t}\rangle) =\bot \hfill  \quad
\hfill  \mathit{aE}' = \langle \langle u, \mathtt{INSERT}, R, \overline{t}\rangle, \bot, \bot, \emptyset \rangle \hfill 
}$

\end{tabular}
\captionof{figure}{Rules defining the $\rightarrow_{f}$ relation for \texttt{INSERT}}\label{table:rules:lts:2}
\end{figure*}

\begin{figure*}
\centering
\begin{tabular}{c}

$\infer[\text{\begin{tabular}{c} \texttt{DELETE}\\ Success 1\end{tabular}}]
{\hfill \langle \mathit{db}, U, \mathit{sec}, T, V, h, \mathit{aE}, \langle \mathit{rS}, \overline{t}', u', \epsilon\rangle \rangle \xrightarrow{\langle u, \mathtt{DELETE}, R, \overline{t}\rangle}_{f} \langle \mathit{db}[R \ominus \overline{t}], U, \mathit{sec}, T, V, h\concat \mathit{aE}, \mathit{aE}', \langle \epsilon, \epsilon, \epsilon, \epsilon\rangle \rangle \hfill }
{
\hfill u\in U \hfill   \quad
\hfill R \in D \hfill \quad
\hfill \mathit{f}(\langle \mathit{db}, U, \mathit{sec}, T,V, h, \mathit{aE}, \langle \mathit{rS},\overline{t}', u', \epsilon\rangle \rangle, \mathit{act}) = \top \hfill  \quad
\hfill  \mathit{act}= \langle u, \mathtt{DELETE}, R, \overline{t}\rangle \hfill  \\
\hfill \mathit{db} [R \ominus \overline{t}] \in \Omega_{D}^{\Gamma} \hfill  \quad
\hfill \mathit{aE}' = \langle \mathit{act}, \top, \top, \emptyset \rangle \hfill \quad 
\hfill \mathit{filter}(T,\mathit{DEL},R) = \epsilon \vee \overline{t} \not \in \mathit{db}(R) \hfill
}$
 \\\\
 
 $\infer[\text{\begin{tabular}{c} \texttt{DELETE}\\ Success 2\end{tabular}}]
{\hfill \langle \mathit{db}, U, \mathit{sec}, T, V, h, \mathit{aE}, \langle \mathit{rS},\overline{t}', u', \epsilon\rangle \rangle \xrightarrow{\langle u, \mathtt{DELETE}, R, \overline{t}\rangle}_{f} \langle \mathit{db}[R \ominus \overline{t}], U, \mathit{sec}, T, V, h\concat \mathit{aE}, \mathit{aE}', \langle \mathit{rS}', \overline{t}, u,\mathit{tr}\rangle \rangle \hfill }
{
\hfill u\in U \hfill   \quad
\hfill R \in D \hfill \quad
\hfill \mathit{f}(\langle \mathit{db}, U, \mathit{sec}, T,V, h, \mathit{aE}, \langle \mathit{rS}, \overline{t}', u', \epsilon\rangle\rangle, \mathit{act}) = \top \hfill  \quad
\hfill  \mathit{act}= \langle u, \mathtt{DELETE}, R, \overline{t}\rangle \hfill \quad
\hfill \mathit{db} [R \ominus \overline{t}] \in \Omega_{D}^{\Gamma} \hfill \\
\hfill \mathit{aE}' = \langle \mathit{act}, \top, \top, \emptyset \rangle \hfill \quad
\hfill \mathit{tr} =\mathit{filter}(T,\mathit{DEL},R)  \hfill \quad
\hfill  \mathit{tr} \neq \epsilon \hfill \quad 
\hfill \overline{t} \in \mathit{db}(R) \hfill \quad
\hfill \mathit{rS}' = \langle \mathit{db}, U, \mathit{sec}, T, V \rangle \hfill 
}$
 \\\\

 $\infer[\text{\begin{tabular}{c} \texttt{DELETE}\\ Exception\end{tabular}}]
{\langle \mathit{db}, U, \mathit{sec}, T, V, h, \mathit{aE}, \langle \mathit{rS}, \overline{t}',u', \epsilon\rangle \rangle \xrightarrow{ \langle u, \mathtt{DELETE}, R, \overline{t} \rangle }_{f} \langle \mathit{db}, U, \mathit{sec}, T, V, h \concat \mathit{aE}, \mathit{aE}',  \langle \epsilon, \epsilon, \epsilon, \epsilon\rangle \rangle}
{
\hfill u\in U \hfill \quad
\hfill R \in D \hfill \quad
\hfill \mathit{f}(\langle \mathit{db}, U, \mathit{sec}, T,V, h, \mathit{aE}, \langle \mathit{rS}, \overline{t}', u', \epsilon\rangle\rangle, \langle u, \mathtt{DELETE}, R, \overline{t}\rangle) = \top \hfill  \\
\hfill E'= \{\phi \in \Gamma | [\phi]^{\mathit{db} [R \ominus \overline{t}]} =\bot\} \hfill \quad 
\hfill E' \neq \emptyset \hfill \quad
\hfill  \mathit{aE}' = \langle \langle u, \mathtt{DELETE}, R, \overline{t}\rangle, \top, \bot, E' \rangle \hfill 
}$
\\\\

 $\infer[\text{\begin{tabular}{c} \texttt{DELETE}\\ Deny\end{tabular}}]
{\hfill \langle \mathit{db}, U, \mathit{sec}, T, V, h, \mathit{aE}, \langle \mathit{rS}, \overline{t}', u', \epsilon\rangle \rangle \xrightarrow{ \langle u, \mathtt{DELETE}, R, \overline{t} \rangle }_{f}  \langle \mathit{db}, U, \mathit{sec}, T, V, h \concat \mathit{aE},   \mathit{aE}',  \langle \epsilon, \epsilon, \epsilon, \epsilon\rangle \rangle\hfill }
{
\hfill u\in U \hfill \quad
\hfill R \in D \hfill \quad
\hfill \mathit{f}(\langle \mathit{db}, U, \mathit{sec}, T,V, h, \mathit{aE}, \langle \mathit{rS}, \overline{t}', u', \epsilon\rangle\rangle, \langle u, \mathtt{DELETE}, R, \overline{t}\rangle) =\bot  \hfill \quad
\hfill  \mathit{aE}' = \langle \langle u, \mathtt{DELETE}, R, \overline{t}\rangle, \bot, \bot, \emptyset \rangle \hfill 
}$

\end{tabular}
\captionof{figure}{Rules defining the $\rightarrow_{f}$ relation for \texttt{DELETE}}\label{table:rules:lts:3}
\end{figure*}

\begin{figure*}
\centering
\begin{tabular}{c}

$\infer[\text{\begin{tabular}{c} Trigger \\\texttt{DELETE}-\\\texttt{INSERT}\\ Success\end{tabular}}]
{\hfill \langle \mathit{db}, U, \mathit{sec}, T, V, h, \mathit{aE}, \langle \mathit{rS}, \overline{t}, \mathit{invk}, t \concat \mathit{tr}\rangle \rangle \xrightarrow{t}_{f} \langle \mathit{db}', U, \mathit{sec}, T, V, h\concat \mathit{aE}, \mathit{tE}', \langle \mathit{rS}, \overline{t}, \mathit{invk}, \mathit{tr}\rangle \rangle \hfill }
{
\hfill \mathit{invk}, \mathit{ow}\in U  \hfill  \quad
\hfill t = \langle \mathit{id},\mathit{ow},  \mathit{ev}, R', \phi, \mathit{stmt}, m\rangle \hfill  \quad
\hfill u= \mathit{getActualUser}(m,\mathit{ow},\mathit{invk}) \hfill \quad
\hfill \phi' =  \phi[\overline{x}^{|R'|}\mapsto\overline{t}]\hfill \\
\hfill \mathit{f}(\langle \mathit{db}, U, \mathit{sec}, T,V, h, \mathit{aE}, \langle \mathit{rS},\overline{t}, \mathit{invk}, t \concat \mathit{tr}\rangle\rangle, \langle u, \mathtt{SELECT}, \phi'\rangle) = \top \hfill \quad
\hfill [ \phi']^{\mathit{db}} = \top \hfill \quad
\hfill  \mathit{aE}' = \langle \langle u, \mathtt{SELECT},  \phi' \rangle, \top, \top, \emptyset\rangle \hfill \\
\hfill \mathit{act} = \mathit{getAction}(\mathit{stmt}, u, \overline{t}) \hfill \quad
\hfill  \mathit{db'} = \mathit{apply}(\mathit{act}, \mathit{db}) \hfill \quad
\hfill \mathit{f}(\langle \mathit{db}, U, \mathit{sec}, T,V, h \concat \mathit{aE}, \mathit{aE}', \langle \mathit{rS}, \overline{t}, \mathit{invk}, t \concat \mathit{tr}\rangle\rangle, \mathit{act}) = \top \hfill \\
\hfill \mathit{db}' \in \Omega_{D}^{\Gamma} \hfill \quad
\hfill \mathit{aE}'' = \langle \mathit{act}, \top, \top, \emptyset \rangle \hfill \quad
\hfill \mathit{tE}' = \langle t, \mathit{aE}', \mathit{aE}''\rangle \hfill \quad
\hfill \mathit{ID}(\mathit{act}) =\top\hfill
}$
 \\\\
 
 $\infer[\text{\begin{tabular}{c} Trigger \\\texttt{DELETE}-\\\texttt{INSERT}\\ Exception\end{tabular}}]
{\hfill \langle \mathit{db}, U, \mathit{sec}, T, V, h, \mathit{aE}, \langle \mathit{rS}, \overline{t}, \mathit{invk}, t \concat \mathit{tr}\rangle \rangle \xrightarrow{t}_{f} \langle \mathit{db}', U', \mathit{sec}', T', V', h\concat \mathit{aE}, \mathit{tE}', \langle \epsilon, \epsilon, \epsilon, \epsilon\rangle \rangle \hfill }
{
\hfill \mathit{invk}, \mathit{ow}\in U \hfill   \quad
\hfill  t = \langle \mathit{id},\mathit{ow},  \mathit{ev}, R', \phi, \mathit{stmt}, m\rangle \hfill  \quad
\hfill u= \mathit{getActualUser}(m,\mathit{ow},\mathit{invk}) \hfill \quad
\hfill  \mathit{rS} = \langle \mathit{db}', U', \mathit{sec}', T', V' \rangle \hfill \\
\hfill \mathit{f}(\langle \mathit{db}, U, \mathit{sec}, T,V, h, \mathit{aE},  \langle \mathit{rS}, \overline{t}, \mathit{invk}, t \concat \mathit{tr}\rangle\rangle, \langle u, \mathtt{SELECT}, \phi'\rangle) = \top \hfill \quad
\hfill  [\phi']^{\mathit{db}} = \top \hfill \quad 
\hfill \mathit{aE}' = \langle \langle u, \mathtt{SELECT}, \phi'\rangle, \top, \top, \emptyset\rangle \hfill \\
\hfill \mathit{act} = \mathit{getAction}(\mathit{stmt}, u, \overline{t}) \hfill \quad
\hfill \mathit{f}(\langle \mathit{db}, U, \mathit{sec}, T,V,  h \concat \mathit{aE}, \mathit{aE}', \langle \mathit{rS},  \overline{t}, \mathit{invk}, t \concat \mathit{tr}\rangle\rangle, \mathit{act}) = \top  \quad \mathit{ID}(\mathit{act}) =\top \hfill \\
\hfill \phi' =  \phi[\overline{x}^{|R'|}\mapsto\overline{t}]\hfill \quad
\hfill E'= \{\phi \in \Gamma | [\phi]^{ \mathit{apply}(\mathit{act}, \mathit{db})}\} \hfill \quad 
\hfill E' \neq \emptyset \hfill \quad
\hfill  \mathit{aE}'' = \langle \mathit{act}, \top, \bot, E' \rangle \hfill  \quad
\hfill \mathit{tE}' = \langle t, \mathit{aE}', \mathit{aE}''\rangle \hfill 
}$
\end{tabular}
\captionof{figure}{Rules defining the $\rightarrow_{f}$ relation for triggers with \texttt{INSERT}/\texttt{DELETE} action}\label{table:rules:lts:5a}
\end{figure*}
 
 \begin{figure*}
\centering
\begin{tabular}{c}
$\infer[\text{\begin{tabular}{c} Trigger\\ \texttt{GRANT}\\ Success\end{tabular}}]
{\hfill \langle \mathit{db}, U, \mathit{sec}, T, V, h, \mathit{aE}, \langle \mathit{rS}, \overline{t}, \mathit{invk}, t \concat \mathit{tr}\rangle \rangle \xrightarrow{t}_{f} \langle \mathit{db}, U, \mathit{sec} \cup \{\langle \mathit{op}, u', p,  u\rangle\}, T, V, h\concat \mathit{aE}, \mathit{tE}', \langle \mathit{rS}, \overline{t}, \mathit{invk}, \mathit{tr}\rangle \rangle \hfill }
{
\hfill \mathit{invk}, \mathit{ow}\in U  \hfill  \quad
\hfill t = \langle \mathit{id},\mathit{ow},  \mathit{ev}, R', \phi, \mathit{stmt}, m\rangle \hfill  \quad
\hfill u= \mathit{getActualUser}(m,\mathit{ow},\mathit{invk}) \hfill \quad
\hfill \phi' =  \phi[\overline{x}^{|R'|}\mapsto\overline{t}]\hfill \\
\hfill \mathit{f}(\langle \mathit{db}, U, \mathit{sec}, T,V, h, \mathit{aE},  \langle \mathit{rS}, \overline{t}, \mathit{invk}, t \concat \mathit{tr}\rangle\rangle, \langle u, \mathtt{SELECT}, \phi'\rangle) = \top \hfill \quad
\hfill  [\phi']^{\mathit{db}} = \top \hfill \quad
\hfill  \mathit{aE}' = \langle \langle u, \mathtt{SELECT}, \phi'\rangle, \top, \top, \emptyset\rangle \hfill \\
\hfill \langle \mathit{op}, u', p,  u\rangle = \mathit{getAction}(\mathit{stmt}, u, \overline{t}) \hfill \quad
\hfill \mathit{f}(\langle \mathit{db}, U, \mathit{sec}, T,V, h \concat \mathit{aE}, \mathit{aE}',  \langle \mathit{rS}, \overline{t}, \mathit{invk}, t \concat \mathit{tr}\rangle\rangle, \langle \mathit{op}, u', p,  u\rangle) = \top \hfill \\
\hfill \mathit{aE}'' = \langle \langle \mathit{op}, u', p,  u\rangle, \top, \top, \emptyset \rangle \hfill \quad 
\hfill \mathit{tE}' = \langle t, \mathit{aE}', \mathit{aE}''\rangle \hfill \quad
\hfill  \mathit{op} \in  \{\oplus, \oplus^{*}\}\hfill
}$
\\
\\
$\infer[\text{\begin{tabular}{c} Trigger\\ \texttt{REVOKE}\\ Success\end{tabular}}]
{\hfill \langle \mathit{db}, U, \mathit{sec}, T, V, h, \mathit{aE}, \langle \mathit{rS}, \overline{t}, \mathit{invk}, t \concat \mathit{tr}\rangle \rangle \xrightarrow{t}_{f} \langle \mathit{db}, U, \mathit{revoke}(\mathit{sec}, u,p,u'), T, V, h\concat \mathit{aE}, \mathit{tE}', \langle \mathit{rS}, \overline{t}, \mathit{invk}, \mathit{tr}\rangle \rangle\hfill }
{
\hfill \mathit{invk}, \mathit{ow}\in U  \hfill  \quad
\hfill t = \langle \mathit{id},\mathit{ow},  \mathit{ev}, R', \phi, \mathit{stmt}, m\rangle \hfill  \quad
\hfill u= \mathit{getActualUser}(m,\mathit{ow},\mathit{invk}) \hfill \quad
\hfill \phi' =  \phi[\overline{x}^{|R'|}\mapsto\overline{t}] \hfill \\
\hfill \mathit{f}(\langle \mathit{db}, U, \mathit{sec}, T,V, h, \mathit{aE}, \langle \mathit{rS}, \overline{t}, \mathit{invk}, t \concat \mathit{tr}\rangle\rangle, \langle u, \mathtt{SELECT}, \phi'\rangle) = \top \hfill \quad 
\hfill [\phi']^{\mathit{db}} = \top \hfill \quad 
\hfill \mathit{aE}' = \langle \langle u, \mathtt{SELECT}, \phi'\rangle, \top, \top, \emptyset\rangle \hfill \\
\hfill \langle \ominus, u', p,  u\rangle = \mathit{getAction}(\mathit{stmt}, u, \overline{t})\hfill  \quad
\hfill \mathit{f}(\langle \mathit{db}, U, \mathit{sec}, T,V, h \concat \mathit{aE}, \mathit{aE}',  \langle \mathit{rS}, \overline{t}, \mathit{invk}, t \concat \mathit{tr}\rangle\rangle, \langle \ominus, u', p,  u\rangle) = \top \hfill \\
\hfill \mathit{aE}'' = \langle \langle \ominus, u', p,  u\rangle, \top, \top, \emptyset \rangle \hfill  \quad
\hfill \mathit{tE}' = \langle t, \mathit{aE}', \mathit{aE}''\rangle \hfill
}$
\end{tabular}
\captionof{figure}{Rules defining the $\rightarrow_{f}$ relation for triggers with \texttt{GRANT}/\texttt{REVOKE} action}\label{table:rules:lts:5b}
\end{figure*}

 \begin{figure*}
\centering
\begin{tabular}{c}
  $\infer[\text{\begin{tabular}{c} Trigger \\ Disabled\end{tabular}}]
{\hfill \langle \mathit{db}, U, \mathit{sec}, T, V, h, \mathit{aE}, \langle \mathit{rS}, \overline{t}, \mathit{invk}, t \concat \mathit{tr}\rangle \rangle \xrightarrow{t}_{f} \langle \mathit{db}, U, \mathit{sec}, T, V, h\concat \mathit{aE}, \mathit{tE}', \langle \mathit{rS}, \overline{t}, \mathit{invk}, \mathit{tr}\rangle \rangle \hfill }
{\hfill \mathit{invk}, \mathit{ow}\in U  \hfill  \quad
\hfill t = \langle \mathit{id},\mathit{ow},  \mathit{ev}, R', \phi, \mathit{stmt}, m\rangle \hfill  \quad
\hfill u= \mathit{getActualUser}(m,\mathit{ow},\mathit{invk}) \hfill \\
\hfill \phi' =  \phi[\overline{x}^{|R'|}\mapsto\overline{t}]\hfill \quad
\hfill \mathit{f}(\langle \mathit{db}, U, \mathit{sec}, T,V,  h, \mathit{aE}, \langle \mathit{rS}, \overline{t}, \mathit{invk}, \mathit{tr}\rangle \rangle, \langle u, \mathtt{SELECT}, \phi'\rangle) = \top \hfill \\
\hfill [\phi']^{\mathit{db}} =\bot \hfill \quad
\hfill  \mathit{aE}' = \langle \langle u, \mathtt{SELECT}, \phi'\rangle, \top, \bot, \emptyset\rangle \hfill \quad
\hfill \mathit{tE}' = \langle t, \mathit{aE}', \epsilon\rangle\hfill 
}$
 \\\\
 
   $\infer[\text{\begin{tabular}{c} Trigger \\ Deny\\ Condition\end{tabular}}]
{\hfill \langle \mathit{db}, U, \mathit{sec}, T, V, h, \mathit{aE}, \langle \mathit{rS}, \overline{t}, \mathit{invk}, t \concat \mathit{tr}\rangle \rangle \xrightarrow{t}_{f} \langle \mathit{db}', U', \mathit{sec}', T', V', h\concat \mathit{aE}, \mathit{tE}', \langle \epsilon, \epsilon, \epsilon, \epsilon \rangle \rangle\hfill }
{\hfill \mathit{invk}, \mathit{ow}\in U \hfill   \quad
\hfill  t = \langle \mathit{id},\mathit{ow},  \mathit{ev}, R', \phi, \mathit{stmt}, m\rangle \hfill  \quad
\hfill u= \mathit{getActualUser}(m,\mathit{ow},\mathit{invk}) \hfill \\
\hfill  \mathit{rS} = \langle \mathit{db}', U', \mathit{sec}', T', V' \rangle \hfill \quad
\hfill \mathit{f}(\langle \mathit{db}, U, \mathit{sec}, T,V,  h, \mathit{aE},   \langle \mathit{rS}, \overline{t}, \mathit{invk}, \mathit{tr} \rangle\rangle, \langle u, \mathtt{SELECT}, \phi'\rangle) =\bot\hfill \\
\hfill \mathit{aE}' = \langle \langle u, \mathtt{SELECT}, \phi'\rangle, \bot, \bot, \emptyset\rangle \hfill \quad
\hfill \mathit{tE}' = \langle t, \mathit{aE}', \epsilon\rangle \hfill \quad
\hfill \phi' =  \phi[\overline{x}^{|R'|}\mapsto\overline{t}] \hfill 
}$
\\\\
  $\infer[\text{\begin{tabular}{c} Trigger \\ Deny\\ Action\end{tabular}}]
{\hfill \langle \mathit{db}, U, \mathit{sec}, T, V, h, \mathit{aE}, \langle \mathit{rS}, \overline{t}, \mathit{invk}, t \concat \mathit{tr}\rangle \rangle \xrightarrow{t}_{f} \langle \mathit{db}', U', \mathit{sec}', T', V', h\concat \mathit{aE}, \mathit{tE}', \langle \epsilon, \epsilon, \epsilon, \epsilon\rangle \rangle\hfill }
{
\hfill \mathit{invk}, \mathit{ow}\in U \hfill   \quad
\hfill  t = \langle \mathit{id},\mathit{ow},  \mathit{ev}, R', \phi, \mathit{stmt}, m\rangle  \hfill \quad
\hfill u= \mathit{getActualUser}(m,\mathit{ow},\mathit{invk}) \hfill  \\
\hfill \mathit{f}(\langle \mathit{db}, U, \mathit{sec}, T,V, h, \mathit{aE}, \langle \mathit{rS}, \overline{t},\mathit{invk}, t \concat \mathit{tr}\rangle\rangle, \langle u, \mathtt{SELECT}, \phi'\rangle) = \top  \hfill \quad
\hfill  [\phi']^{\mathit{db}} = \top \hfill \quad 
\hfill \mathit{aE}' = \langle \langle u, \mathtt{SELECT}, \phi'\rangle, \top, \top, \emptyset\rangle \hfill \\
\hfill \mathit{act} = \mathit{getAction}(\mathit{stmt}, u, \overline{t}) \hfill \quad
\hfill \mathit{f}(\langle \mathit{db}, U, \mathit{sec}, T,V, h \concat \mathit{aE}, \mathit{aE}', \langle \mathit{rS}, \overline{t}, \mathit{invk}, t \concat \mathit{tr}\rangle\rangle, \mathit{act}) =\bot \hfill \quad
\hfill \phi' =  \phi[\overline{x}^{|R'|}\mapsto\overline{t}]\hfill \\
\hfill \mathit{aE}'' = \langle \mathit{act}, \bot, \bot, \emptyset \rangle \hfill \quad
\hfill \mathit{tE}' = \langle t, \mathit{aE}', \mathit{aE}''\rangle \hfill \quad
\hfill  \mathit{rS} = \langle \mathit{db}', U', \mathit{sec}', T', V' \rangle \hfill  
}$
\end{tabular}
\captionof{figure}{Rules defining the $\rightarrow_{f}$ relation for triggers}\label{table:rules:lts:6}
\end{figure*}

\begin{figure*}
\centering
\begin{tabular}{c}

$\infer[\text{\begin{tabular}{c}\texttt{GRANT}\\ Success\end{tabular}}]
{\hfill \langle \mathit{db}, U, \mathit{sec}, T, V, h, \mathit{aE}, \langle \mathit{rS}, \overline{t}, u'', \epsilon\rangle \rangle \xrightarrow{ \langle \mathit{op}, u, p,  u'\rangle}_{f} 
\langle \mathit{db}, U,  \mathit{sec} \cup \{\langle \mathit{op}, u, p,  u'\rangle\}, T,V,h \concat \mathit{aE}, \mathit{aE}', \langle \epsilon, \epsilon, \epsilon,  \epsilon\rangle\rangle \hfill }
{\hfill u,u'\in U \hfill \quad 
\hfill  \mathit{f}(\langle \mathit{db}, U, \mathit{sec}, T,V, h, \mathit{aE}, \langle \mathit{rS}, \overline{t}, u'', \epsilon\rangle \rangle,  \langle \mathit{op}, u, p, u'\rangle) = \top \hfill \\
\hfill \mathit{aE}' = \langle \langle \mathit{op}, u, p,  u'\rangle, \top, \top, \emptyset \rangle \hfill \quad
\hfill  \mathit{op} \in \{\oplus,\oplus^{*}\} \hfill \quad
\hfill \mathit{defined}(p,D,V) \hfill }$ 
\\
\\
$\infer[\text{\begin{tabular}{c}\texttt{REVOKE}\\ Success\end{tabular}}]
{\hfill \langle \mathit{db}, U, \mathit{sec}, T, V, h, \mathit{aE}, \langle \mathit{rS}, \overline{t}, u'', \epsilon\rangle \rangle \xrightarrow{ \langle \ominus, u, p,  u'\rangle}_{f} 
\langle \mathit{db}, U,  \mathit{revoke}(\mathit{sec}, u,p,u'), T,V,h \concat \mathit{aE}, \mathit{aE}', \langle \epsilon, \epsilon, \epsilon,  \epsilon\rangle\rangle \hfill }
{\hfill u,u'\in U \hfill \quad
\hfill  \mathit{f}(\langle \mathit{db}, U, \mathit{sec}, T,V, h, \mathit{aE}, \langle \mathit{rS}, \overline{t}, u'', \epsilon\rangle\rangle,  \langle \ominus, u, p, u'\rangle) = \top \hfill \\
\hfill \mathit{aE}' = \langle \langle \ominus, u, p,  u'\rangle, \top, \top, \emptyset \rangle \hfill \quad
\hfill \mathit{defined}(p,D,V) \hfill}$ 
\\\\
$\infer[\text{\begin{tabular}{c}\texttt{GRANT}-\\\texttt{REVOKE}\\ Deny\end{tabular}}]
{\hfill \langle \mathit{db}, U, \mathit{sec}, T, V, h, \mathit{aE}, \langle \mathit{rS}, \overline{t}, u'', \epsilon\rangle \rangle \xrightarrow{ \langle \mathit{op}, u, p,  u'\rangle}_{f} 
\langle \mathit{db}, U,  \mathit{sec}, T,V,h \concat \mathit{aE}, \mathit{aE}', \langle \epsilon,\epsilon, \epsilon,  \epsilon\rangle\rangle\hfill }
{
\hfill u,u'\in U \hfill \quad
\hfill  \mathit{f}(\langle \mathit{db}, U, \mathit{sec}, T,V, h, \mathit{aE}, \langle \mathit{rS}, \overline{t}, u'', \epsilon\rangle \rangle,  \langle \mathit{op}, u, p, u'\rangle) =\bot \hfill \\
\hfill \mathit{aE}' = \langle \langle \mathit{op}, u, p,  u'\rangle, \bot, \bot, \emptyset \rangle \hfill  \quad
\hfill  \mathit{op} \in \{\oplus,\oplus^{*}, \ominus\} \quad
\hfill \mathit{defined}(p,D,V) \hfill}$ 
\end{tabular}
\captionof{figure}{Rules defining the $\rightarrow_{f}$ relation for \texttt{GRANT} and \texttt{REVOKE}}\label{table:rules:lts:7}
\end{figure*}

\begin{figure*}

\centering
\begin{tabular}{c}
$\infer[\text{\begin{tabular}{c}\texttt{CREATE} \\ \texttt{TRIGGER}\\ Success\end{tabular}}]
{\hfill \langle \mathit{db}, U, \mathit{sec}, T, V, h, \mathit{aE}, \langle \mathit{rS}, \overline{t}, u', \epsilon\rangle \rangle \xrightarrow{\langle u, \mathtt{CREATE}, t\rangle}_{f} 
\langle \mathit{db}, U, \mathit{sec}, T \cup \{t\},V, h \concat \mathit{aE}, \mathit{aE}', \langle \epsilon, \epsilon, \epsilon,  \epsilon\rangle\rangle\hfill
 }
{
\hfill u\in U \hfill \quad
\hfill \mathit{defined}(t,D,V) \hfill \quad
\hfill  \mathit{safe}(\{t\}\cup T) \hfill  \quad
\hfill  \mathit{usersIn}(t,U) \hfill  \quad
\hfill  \mathit{f}(\langle \mathit{db}, U, \mathit{sec}, T,V, h, \mathit{aE}, \langle \mathit{rS}, \overline{t}, u', \epsilon\rangle \rangle, \langle u, \mathtt{CREATE}, t\rangle) = \top \hfill \\
\hfill t = \langle \mathit{id},u,  \mathit{ev}, R, \phi, \mathit{stmt}, m\rangle \hfill \quad
\hfill \mathit{aE}' = \langle \langle u, \mathtt{CREATE}, t\rangle, \top, \top, \emptyset \rangle \hfill  \quad
\hfill \neg \exists t' \in T.\, t' = \langle \mathit{id},\mathit{ow}',  \mathit{ev}', R', \phi', \mathit{stmt}', m'\rangle \hfill}$ \\
 \\
 
$\infer[\text{\begin{tabular}{c}\texttt{CREATE} \\ \texttt{TRIGGER}\\ Deny \end{tabular}}]
{\hfill \langle \mathit{db}, U, \mathit{sec}, T, V, h, \mathit{aE}, \langle \mathit{rS}, \overline{t}, u', \epsilon\rangle \rangle \xrightarrow{\langle u, \mathtt{CREATE}, t\rangle}_{f} 
\langle \mathit{db}, U, \mathit{sec}, T, V, h \concat \mathit{aE}, \mathit{aE}', \langle \epsilon, \epsilon, \epsilon,  \epsilon\rangle\rangle\hfill }
{
\hfill u\in U \hfill \quad
\hfill \mathit{defined}(t,D,V) \hfill \quad
\hfill  \mathit{safe}(\{t\}\cup T) \hfill  \quad
\hfill  \mathit{usersIn}(t,U) \hfill  \quad
\hfill  \mathit{f}(\langle \mathit{db}, U, \mathit{sec}, T,V, h, \mathit{aE}, \langle \mathit{rS}, \overline{t}, u', \epsilon\rangle \rangle, \langle u, \mathtt{CREATE}, t\rangle) = \top \hfill \\
\hfill t = \langle \mathit{id},u,  \mathit{ev}, R, \phi, \mathit{stmt}, m\rangle \hfill \quad
\hfill \mathit{aE}' = \langle \langle u, \mathtt{CREATE}, t\rangle, \top, \bot, \emptyset \rangle \hfill \quad
\hfill t' = \langle \mathit{id},\mathit{ow}',  \mathit{ev}', R', \phi', \mathit{stmt}', m'\rangle \hfill \quad
\hfill t' \in T \hfill \quad 
\hfill t' \neq t \hfill }$ \\
 \\

$\infer[\text{\begin{tabular}{c}\texttt{CREATE}\\ \texttt{VIEW}\\ Success\end{tabular}}]
{\hfill \langle \mathit{db}, U, \mathit{sec}, T, V, h, \mathit{aE}, \langle \mathit{rS}, \overline{t}, u', \epsilon\rangle \rangle \xrightarrow{\langle u, \mathtt{CREATE}, v\rangle}_{f} 
\langle \mathit{db}, U, \mathit{sec}, T ,V\cup \{v\}, h \concat \mathit{aE}, \mathit{aE}', \langle \epsilon,\epsilon, \epsilon,  \epsilon\rangle\rangle \hfill }
{\hfill u\in U \hfill   \quad
\hfill \mathit{defined}(v,D,V) \hfill \quad
\hfill  \mathit{f}(\langle \mathit{db}, U, \mathit{sec}, T,V, h, \mathit{aE}, \langle \mathit{rS}, \overline{t}, u', \epsilon\rangle \rangle,  \langle u, \mathtt{CREATE}, v\rangle) = \top \hfill \\
\hfill v = \langle \mathit{id},u, q, m\rangle \hfill \quad
\hfill \mathit{aE}' = \langle \langle u, \mathtt{CREATE}, v\rangle, \top, \top, \emptyset \rangle \hfill \quad
\hfill \neg \exists v' \in V.\,  v' = \langle \mathit{id},\mathit{ow}', q', m'\rangle \hfill}$ \\
 \\
 
 $\infer[\text{\begin{tabular}{c}\texttt{CREATE}\\ \texttt{VIEW}\\ Deny\end{tabular}}]
{\hfill \langle \mathit{db}, U, \mathit{sec}, T, V, h, \mathit{aE}, \langle \mathit{rS}, \overline{t}, u', \epsilon\rangle \rangle \xrightarrow{\langle u, \mathtt{CREATE}, v\rangle}_{f} 
\langle \mathit{db}, U, \mathit{sec}, T ,V, h \concat \mathit{aE}, \mathit{aE}', \langle \epsilon,\epsilon, \epsilon , \epsilon\rangle\rangle \hfill }
{\hfill u\in U \hfill   \quad
\hfill \mathit{defined}(v,D,V) \hfill \quad
\hfill  \mathit{f}(\langle \mathit{db}, U, \mathit{sec}, T,V, h, \mathit{aE}, \langle \mathit{rS}, \overline{t}, u', \epsilon\rangle \rangle,  \langle u, \mathtt{CREATE}, v\rangle) = \top \hfill \\
\hfill v = \langle \mathit{id},u, q, m\rangle \hfill \quad
\hfill \mathit{aE}' = \langle \langle u, \mathtt{CREATE}, v\rangle, \top, \bot, \emptyset \rangle \hfill \quad
\hfill  v' = \langle \mathit{id},\mathit{ow}', q', m'\rangle \hfill \quad
\hfill v' \in V \hfill \quad
\hfill v \neq v' \hfill}$ \\
 \\

$\infer[\text{\begin{tabular}{c}\texttt{CREATE}\\ Deny\end{tabular}}]
{\hfill \langle \mathit{db}, U, \mathit{sec}, T, V, h, \mathit{aE}, \langle \mathit{rS}, \overline{t}, u', \epsilon\rangle \rangle \xrightarrow{\langle u, \mathtt{CREATE}, o\rangle}_{f} 
\langle \mathit{db}, U, \mathit{sec}, T,V, h \concat \mathit{aE}, \mathit{aE}', \langle \epsilon,\epsilon, \epsilon, \epsilon\rangle\rangle \hfill }
{ \hfill u\in U \hfill   \quad
\hfill \mathit{defined}(o,D,V) \hfill \quad
\hfill  \mathit{f}(\langle \mathit{db}, U, \mathit{sec}, T,V, h, \mathit{aE}, \langle \mathit{rS}, \overline{t}, u', \epsilon\rangle\rangle,  \langle u, \mathtt{CREATE}, o\rangle) = \bot \hfill  \quad
\hfill \mathit{aE}' = \langle \langle u, \mathtt{CREATE}, o\rangle, \bot, \bot, \emptyset \rangle \hfill }$
\end{tabular}

\captionof{figure}{Rules defining the $\rightarrow_{f}$ relation for \texttt{CREATE} triggers and views}\label{table:rules:lts:8}
\end{figure*}

%% file: attacker_long.tex
\clearpage

\section{Attacker Model}\label{app:adv:model}

\begin{figure*}[!hbtp]
\[
\mathit{reviseBelief}(p', \phi,p' \concat \mathit{act} \concat s)) = \left\{ 
  \begin{array}{l l}
    \top & \text{if } \mathit{act}= \langle u, \mathit{op}, R, \overline{t}\rangle  \wedge R \not\in \mathit{tables}(\phi) \wedge \mathit{op} \in \{\texttt{INSERT}, \texttt{DELETE}\}\\
    \top & \text{if } \mathit{act}= \langle \mathit{id},\mathit{ow},  \mathit{ev}, R', \phi, \langle \mathit{op}, R,\overline{t}  \rangle,m\rangle  \wedge R \not\in \mathit{tables}(\phi) \wedge \mathit{op} \in \{\texttt{INSERT}, \texttt{DELETE}\}\\
  \top & \text{if } \mathit{act}= \langle \mathit{id},\mathit{ow},  \mathit{ev}, R, \phi, \langle \mathit{op}, u, p  \rangle,m\rangle  \wedge \mathit{op} \in \{\oplus,\oplus^{*}, \ominus\}\\
    \bot & \text{otherwise}
  \end{array} \right.
  \]
  \caption{Belief Revision}\label{function:revise}
\end{figure*}

In this section, we formalize our attacker model $\attackerModel$. Let $P = \langle M, f \rangle$ be an \accessControlConfiguration{}, where $M = \langle D, \Gamma\rangle$ is a system configuration and $f$ is an $M$-\acf{}, $L$ be the $P$-LTS, and $u \in {\cal U}$ be a user.
The set $\attackerModel$ is the smallest set of judgments  satisfying the inference rules in Figures \ref{table:rules:adversary:1}--\ref{table:rules:adversary:11}. 
With a slight abuse of notation, in the rules we use $r,i \attMod \phi$ to denote that this judgment holds in $\attackerModel$, i.e., $ r,i \attMod \phi \in \attackerModel$.
Note that we redefine here also the rules we presented before in Figure \ref{figure:adv:model:rules}. 

In the rules, we use $\models_{\mathit{fin}}$ to denote the standard semantic entailment relation for first-order logic over finite models.
We also denote by $\mathit{replace}(\psi,o)$, where $\psi$ is a sentence and $o$ is a view $\langle V, \mathit{ow}, \{\overline{x} | \phi \}, m\rangle \in {\cal VIEW}_{D}$, the formula $\psi'$ obtained from $\psi$ by replacing all occurrences of $V(\overline{x})$ with $\phi(\overline{x})$. Note that $\psi$ and $\mathit{replace}(\psi,o)$ are semantically equivalent.
Finally, given a database schema $D$, a state $s = \langle \mathit{db}, U, \mathit{sec}, T, V, \mathit{ctx} \rangle$, and an action $a \in {\cal A}_{D,{\cal U}} \cup {\cal TRIGGER}_{D}$, we denote by $\mathit{user}(s,a)$ the following function:
\[
\mathit{user}(s,a) = \left\{ 
  \begin{array}{l l}
    \mathit{invoker}(s) & \text{if } \mathit{tr}(s) \neq \epsilon\\
    u & \text{if } \mathit{tr}(s) = \epsilon  \wedge  u \in {\cal U} \wedge a \in {\cal A}_{D,u}\\
  \end{array} \right.
  \]

In the rules, we omit some details when dealing with integrity constraints.
For instance, when we refer to functional dependencies of the form $\forall \overline{x}, \overline{y}, \overline{y}', \overline{z}, \overline{z}'.\, (R(\overline{x}, \overline{y},  \overline{z}) \wedge R(\overline{x}, \overline{y}',  \overline{z}') )\Rightarrow \overline{y} = \overline{y}'$, we implicitly assume that $|\overline{y}| = |\overline{y}'|$ and $|\overline{z}| = |\overline{z}'|$.
Furthermore, when we refer to tuples in $R$, we use the notation $(\overline{v}, \overline{w}, \overline{q})$ and we implicitly assume that $|\overline{v}| = |\overline{x}|$, $|\overline{w}| = |\overline{y}|$, and $|\overline{q}| = |\overline{z}|$.
We make similar simplifications for the inclusion dependencies.

In our attacker model, we consider a very simple syntactic criterion for revising believes.
Intuitively, the attacker is able to propagate the knowledge of a sentence $\phi$ after (or before) an \texttt{INSERT} or a \texttt{DELETE} action on a table $R$ iff the predicate $R$ does not occur in  $\phi$. 
We formalize this using the function $\mathit{reviseBelief} : \mathit{traces}(L) \times \mathit{RC}_{\mathit{bool}} \times \mathit{traces}(L) \rightarrow \{\top,\bot\}$.
In Figure \ref{function:revise}, we give the definition for the function only for the inputs $r'$, $\phi$, $r$ such that $\phi \in \mathit{RC}_{\mathit{bool}}$ is a sentence and  $r=r' \concat \mathit{act} \concat s$, where $\mathit{act} \in {\cal A}_{D, {\cal U }} \cup {\cal TRIGGER}_{D}$ and $s \in \Omega_{M}$. 
If this is not the case, then $\mathit{reviseBelif}(r',\phi,r) = \bot$.
Note that the function $\mathit{tables}$ takes as input a formula $\phi$ and returns as output the set of all tables mentioned in $\phi'$, where $\phi'$ is the formula obtained from $\phi$ by replacing all views with their definitions.
We remark that, given a formula $\phi$, if $R \not\in \mathit{tables}(\phi)$, then the value of $\phi$ is independent on $R$, i.e., $R$ does not determine $\phi$.

In \thref{theorem:attacker:model:sound}, we prove that our attacker model is \emph{sound} with respect to the relational calculus semantics, i.e., every judgment $r,i \attMod \phi$ that holds in $\attackerModel$ is such that $\phi$ is satisfied in the $i$-th state of $r$.
We first introduce the concept of \emph{derivation}.
Given a judgment $r, i \attMod \phi$, a \emph{derivation of  $r, i \attMod \phi$ with respect to $\attackerModel$}, or \emph{a derivation of $r,i \attMod \phi$} for short, is a proof tree, obtained by applying the rules defining $\attackerModel$, that ends in $r, i \attMod \phi$.
With a slight abuse of notation, we use $r,i \attMod \phi$ to denote both the judgment and its derivation.
The length of a derivation, denoted $|r, i \attMod \phi|$, is the number of rule applications in it.
Note that a judgments $r,i \attMod \phi$ holds in $\attackerModel$ iff there is a derivation for it.

\begin{lemma}\thlabel{theorem:attacker:model:sound}
Let $P$ be an \accessControlConfiguration{}, $L$ be the $P$-LTS, $u$ be a user, $r \in \mathit{traces}(L)$ be an $L$ run, $\phi \in RC_{\mathit{bool}}$ be a sentence, and $1 \leq i \leq |r|$.
If $r,i \attMod \phi$ holds in $\attackerModel$, then $[\phi]^{\mathit{db}} = \top$, where $\mathit{last}(r^{i}) = \langle \mathit{db}, U, \mathit{sec}, T, V, c \rangle$.
\end{lemma}
\begin{proof}
Let $P$ be an \accessControlConfiguration{}, $L$ be the $P$-LTS, $u$ be a user, $r \in \mathit{traces}(L)$ be an $L$ run, $\phi \in RC_{\mathit{bool}}$ be a sentence, and $1 \leq i \leq |r|$.
Furthermore, let $r,i \attMod \phi$ be a judgment that holds, i.e., there is a derivation $d$ that ends on this judgment.
We prove our claim by induction on the length of $d$.

\smallskip
\noindent
{\bf Base Case: }
The base case is a derivation of length $1$.
Thus, there are a number of cases depending on the rule used to obtain $r,i \attMod \phi$.
\begin{compactenum}
\item \emph{\texttt{SELECT} Success - 1}.
Let $i$ be such that $r^{i} = r^{i-1} \concat \langle u, \mathtt{SELECT}, \phi \rangle \concat s$, where $s = \langle \mathit{db}, U, \mathit{sec}, T, V, c \rangle \in \Omega_{M}$ and $\mathit{last}(r^{i-1}) =  \langle \mathit{db}, U, \mathit{sec}, T, V, c' \rangle$.
From the rules, it follows that $\mathit{res}(s) = \top$.
From this and the LTS rules, it follows that $[\phi]^{\mathit{db}} = \top$.

\item \emph{\texttt{SELECT} Success - 2}.
The proof for this case is similar to that of \emph{\texttt{SELECT} Success - 1}.

\item \emph{\texttt{INSERT} Success}.
Let $i$ be such that $r^{i} = r^{i-1} \concat \langle u, \mathtt{INSERT}, \\ R, \overline{t} \rangle \concat s$ , where $s = \langle \mathit{db}, U, \mathit{sec}, T, V, c \rangle \in \Omega_{M}$ and $\mathit{last}(r^{i-1}) =  \langle \mathit{db}', U, \mathit{sec}, T, V, c' \rangle$, and $\phi$ be $R(\overline{t})$.
From the LTS rules, it follows that $\mathit{db} = \mathit{db}'[R \oplus \overline{t}]$.
From $\oplus$'s definition, it follows that $\overline{t} \in \mathit{db}(R)$.
Therefore, from the RC's semantics, it follows that $[R(\overline{t})]^{\mathit{db}} = \top$.
Since $\phi := R(\overline{t})$, it follows that $[\phi]^{\mathit{db}} = \top$.

\item \emph{\texttt{INSERT} Success - FD}.
Let $i$ be such that $r^{i} = r^{i-1} \concat \langle u, \mathtt{INSERT}, R, (\overline{v}, \overline{w}, \overline{q}) \rangle \concat s$, where $s = \langle \mathit{db}, U, \mathit{sec}, T, V, c \rangle \\ \in \Omega_{M}$ and $\mathit{last}(r^{i-1}) =  \langle \mathit{db}', U, \mathit{sec}, T, V, c' \rangle$, and $\phi$ be $\neg \exists \overline{y},\overline{z}.\, R(\overline{v}, \overline{y}, \overline{z}) \wedge \overline{y} \neq \overline{w}$.
We claim that $[\phi]^{\mathit{db}'}$ holds.
From this claim and the LTS semantics, it follows that there is no tuple $(\overline{v}', \overline{w}', \overline{q}')$ in $\mathit{db}'(R)$ such that $\overline{v}' = \overline{v}$ and $\overline{w}' \neq \overline{w}$.
There are two cases:
\begin{compactenum}
\item The \texttt{INSERT} command causes an integrity exception, i.e., $\mathit{Ex}(s) \neq \emptyset$.
From this and the LTS semantics, it follows that $\mathit{db} = \mathit{db}'$.
From this and $[\phi]^{\mathit{db}'}$ holds, it follows that also $[\phi]^{\mathit{db}}$ holds.

\item The \texttt{INSERT} command does not cause any integrity exception, i.e., $\mathit{Ex}(s) = \emptyset$.
From this, $[\phi]^{\mathit{db}'} = \top$, and $\mathit{db}(R) = \mathit{db}'(R) \cup \{(\overline{v}, \overline{w}, \overline{q})\}$, it follows that there is no tuple $(\overline{v}', \overline{w}', \overline{q}')$ in $\mathit{db}(R)$ such that $\overline{v}' = \overline{v}$ and $\overline{w}' \neq \overline{w}$.
From this, it follows that also $[\phi]^{\mathit{db}}$ holds.
\end{compactenum}

We now prove our claim that $[\phi]^{\mathit{db}'}$ holds.
Assume, for contradiction's sake, that this is not the case.
This means that there is a tuple $(\overline{v}', \overline{w}', \overline{q}')$ in $\mathit{db}'(R)$ such that $\overline{v}' = \overline{v}$ and $\overline{w}' \neq \overline{w}$.
Let $\mathit{db}''$ be the state $\mathit{db}'[R\oplus (\overline{v}, \overline{w}, \overline{q})]$.
From $\mathit{db}''=\mathit{db}'[R\oplus (\overline{v}, \overline{w}, \overline{q})]$, and the fact that there is a tuple $(\overline{v}', \overline{w}', \overline{q}')$ in $\mathit{db}'(R)$ such that $\overline{v}' = \overline{v}$ and $\overline{w}' \neq \overline{w}$,
it follows that there are two tuples $(\overline{v}, \overline{w}, \overline{q})$ and $(\overline{v}, \overline{w}', \overline{q}')$ in $\mathit{db}''(R)$ such that $\overline{w}' \neq \overline{w}$.
From this and the relational calculus semantics, it follows that $[\forall \overline{x}, \overline{y}, \overline{y}', \overline{z}, \overline{z}'.\,( (R(\overline{x}, \overline{y},  \overline{z}) \wedge R(\overline{x}, \overline{y}',  \overline{z}') )\Rightarrow \overline{y} = \overline{y}'  ]^{\mathit{db}''} = \bot$.
This contradicts the fact that $\forall \overline{x}, \overline{y}, \overline{y}', \\ \overline{z}, \overline{z}'.\,( (R(\overline{x}, \overline{y},  \overline{z}) \wedge R(\overline{x}, \overline{y}',  \overline{z}') )\Rightarrow \overline{y} = \overline{y}' $ is not in $\mathit{Ex}(s)$.

\item \emph{\texttt{INSERT} Success - ID}.
Let $i$ be such that $r^{i} = r^{i-1} \concat \langle u, \mathtt{INSERT}, R, (\overline{v}, \overline{w}) \rangle \concat s$, where $s = \langle \mathit{db}, U, \mathit{sec}, T, V, c \rangle \\ \in \Omega_{M}$ and $\mathit{last}(r^{i-1}) =  \langle \mathit{db}', U, \mathit{sec}, T, V, c' \rangle$, and $\phi$ be $\exists \overline{y}.\, S(\overline{v}, \overline{y})$.
We claim that $[\phi]^{\mathit{db}'}$ holds.
From this claim and the LTS semantics, it follows that there is a tuple $(\overline{v}', \overline{w}')$ in $\mathit{db}'(S)$ such that $\overline{v}' = \overline{v}$.
There are two cases:
\begin{compactenum}
\item The \texttt{INSERT} command causes an integrity exception, i.e., $\mathit{Ex}(s) \neq \emptyset$.
From this and the LTS semantics, it follows that $\mathit{db} = \mathit{db}'$.
From this and $[\phi]^{\mathit{db}'}$ holds, it follows that also $[\phi]^{\mathit{db}}$ holds.

\item The \texttt{INSERT} command does not cause any integrity exception, i.e., $\mathit{Ex}(s) = \emptyset$.
From this, $[\phi]^{\mathit{db}'} = \top$, and $\mathit{db}(S) = \mathit{db}'(S)$, it follows that there a tuple $(\overline{v}', \overline{w}')$ in $\mathit{db}(S)$ such that $\overline{v}' = \overline{v}$.
From this, it follows that also $[\phi]^{\mathit{db}}$ holds.
\end{compactenum}

We now prove our claim that $[\phi]^{\mathit{db}'}$ holds.
Assume, for contradiction's sake, that this is not the case.
This means that there is no tuple $(\overline{v}', \overline{w}')$ in $\mathit{db}'(S)$ such that $\overline{v}' = \overline{v}$.
Let $\mathit{db}''$ be the state $\mathit{db}'[R\oplus (\overline{v}, \overline{w})]$.
From $\mathit{db}''=\mathit{db}'[R\oplus (\overline{v}, \overline{w})]$, and the fact that there is no tuple $(\overline{v}', \overline{w}')$ in $\mathit{db}'(S)$ such that $\overline{v}' = \overline{v}$,
it follows that there is a tuple $(\overline{v}, \overline{w})$ in $\mathit{db}''(R)$ and there is no tuple $(\overline{v}', \overline{w}')$ in $\mathit{db}''(S)$ such that $\overline{v}' = \overline{v}$.
From this and the relational calculus semantics, it follows that $[\forall \overline{x}, \overline{z}.\,( R(\overline{x}, \overline{z}) \Rightarrow \exists \overline{w}.\, S(\overline{x}, \overline{w})) ]^{\mathit{db}''} = \bot$.
This contradicts the fact that $\forall \overline{x}, \overline{z}.\,( R(\overline{x}, \overline{z}) \Rightarrow \exists \overline{w}.\, S(\overline{x}, \overline{w})) $ is not in $\mathit{Ex}(s)$.

\item \emph{\texttt{DELETE} Success}.
The proof for this case is similar to that of \emph{\texttt{INSERT} Success}.

\item \emph{\texttt{DELETE} Success - ID}.
Let $i$ be such that $r^{i} = r^{i-1} \concat \langle u, \mathtt{DELETE}, R, (\overline{v}, \overline{w}) \rangle \concat s$, where $s = \langle \mathit{db}, U, \mathit{sec}, T, V, c \rangle \\ \in \Omega_{M}$ and $\mathit{last}(r^{i-1}) =  \langle \mathit{db}', U, \mathit{sec}, T, V, c' \rangle$, and $\phi$ be $\forall \overline{x}, \overline{z}.\,( S(\overline{x}, \overline{z}) \Rightarrow \overline{x} \neq \overline{v}) \vee \exists \overline{y}.\, (R(\overline{v}, \overline{y}) \wedge \overline{y} \neq \overline{w})$.
We claim that $[\phi]^{\mathit{db}}$ holds.
From this claim and the LTS semantics, it follows that there are two cases:
\begin{compactenum}
\item  all tuples $(\overline{x},\overline{y}) \in \mathit{db}(S)$ are such that $\overline{v} \neq \overline{x}$. 
There are two cases:
\begin{compactenum}
\item The \texttt{DELETE} command causes an integrity exception, i.e., $\mathit{Ex}(s) \neq \emptyset$.
From this and the LTS semantics, it follows that $\mathit{db} = \mathit{db}'$.
From this and $[\phi]^{\mathit{db}'}$ holds, it follows that also $[\phi]^{\mathit{db}}$ holds.

\item The \texttt{DELETE} command does not cause any integrity exception, i.e., $\mathit{Ex}(s) = \emptyset$.
From this, $[\phi]^{\mathit{db}'} = \top$, and $\mathit{db}(S) = \mathit{db}'(S)$, it follows that all tuples $(\overline{x},\overline{y}) \in \mathit{db}(S)$ are such that $\overline{v} \neq \overline{x}$.
Therefore, also $[\phi]^{\mathit{db}}$ holds.
\end{compactenum}

\item there is a tuple $(\overline{v},\overline{w}') \in \mathit{db}(R)$ such that $\overline{w} \neq \overline{w}'$. 
There are two cases:
\begin{compactenum}
\item The \texttt{DELETE} command causes an integrity exception, i.e., $\mathit{Ex}(s) \neq \emptyset$.
From this and the LTS semantics, it follows that $\mathit{db} = \mathit{db}'$.
From this and $[\phi]^{\mathit{db}'}$ holds, it follows that also $[\phi]^{\mathit{db}}$ holds.

\item The \texttt{DELETE} command does not cause any integrity exception, i.e., $\mathit{Ex}(s) = \emptyset$.
From this, $[\phi]^{\mathit{db}'} = \top$, and $\mathit{db}(R) = \mathit{db}'(R) \setminus \{(\overline{v}, \overline{w})\}$, it follows that there is a tuple $(\overline{v},\overline{w}') \in \mathit{db}(R)$ such that $\overline{w} \neq \overline{w}'$.
Therefore, also $[\phi]^{\mathit{db}}$ holds.
\end{compactenum}

\end{compactenum}

We now prove our claim that $[\phi]^{\mathit{db}'}$ holds.
Assume, for contradiction's sake, that this is not the case.
This means that there is a tuple $(\overline{v}, \overline{z})$ in $\mathit{db}'(S)$ and there is no tuple $(\overline{v},\overline{y}) \in \mathit{db}'(R)$ such that $\overline{y} \neq \overline{w}$.
Let $\mathit{db}''$ be the state $\mathit{db}'[R\ominus (\overline{v}, \overline{w})]$.
From $\mathit{db}''=\mathit{db}'[R\ominus (\overline{v}, \overline{w})]$, and the fact that there is a tuple $(\overline{v}, \overline{z})$ in $\mathit{db}'(S)$ and there is no tuple $(\overline{v},\overline{y}) \in \mathit{db}'(R)$ such that $\overline{y} \neq \overline{w}$,
it follows that there is a tuple $(\overline{v}, \overline{z})$ in $\mathit{db}''(S)$ and there is no tuple $(\overline{v},\overline{y}) \in \mathit{db}''(R)$ such that $\overline{y} \neq \overline{w}$.
From this and the relational calculus semantics, it follows that $[\forall \overline{x}, \overline{z}.\,( S(\overline{x}, \overline{z}) \Rightarrow \exists \overline{w}.\, R(\overline{x}, \overline{w}) ]^{\mathit{db}''} = \bot$.
This contradicts the fact that $\forall \overline{x}, \overline{z}.\,( S(\overline{x}, \overline{z}) \Rightarrow \exists \overline{w}.\, R(\overline{x}, \overline{w})$ is not in $\mathit{Ex}(s)$.

\item \emph{\texttt{INSERT} Exception}.
Let $i$ be such that $r^{i} = r^{i-1} \concat \langle u, \mathtt{INSER}, R, \overline{t} \rangle \concat s$, where $s = \langle \mathit{db}, U, \mathit{sec}, T, V, c \rangle  \in  \Omega_{M}$ and $\mathit{last}(r^{i-1}) =  \langle \mathit{db}', U, \mathit{sec}, T, V, c' \rangle$, and $\phi$ be $\neg R(\overline{t})$.
We claim that $[\neg R(\overline{t})]^{\mathit{db}'} = \top$ holds.
From the LTS semantics, it follows that $\mathit{db} = \mathit{db}'$.
Therefore, also $[\neg R(\overline{t})]^{\mathit{db}} = \top$  holds.

We now prove our claim.
Assume, for contradiction's sake, that $[\neg R(\overline{t})]^{\mathit{db}'} = \bot$.
Therefore, $\overline{t} \in \mathit{db}'(R)$.
From this and the definition of $\oplus$, it follows that $\mathit{db}' = \mathit{db}' [R \oplus \overline{t}]$.
From the rules, it follows that $\mathit{Ex}(s) \neq \emptyset$.
Therefore, from the LTS semantics, it follows that $\mathit{db}' [R \oplus \overline{t}] \not\in \Omega_{D}^{\Gamma}$.
From $\mathit{last}(r^{i-1}) =  \langle \mathit{db}', U, \mathit{sec}, T, V, c' \rangle$, it follows that $\mathit{db}'  \in \Omega_{D}^{\Gamma}$.
However, from $\mathit{db}' = \mathit{db}' [R \oplus \overline{t}]$ and $\mathit{db}'  \in \Omega_{D}^{\Gamma}$, it follows that $\mathit{db}' [R \oplus \overline{t}] \in \Omega_{D}^{\Gamma}$ leading to a contradiction.

\item \emph{\texttt{DELETE} Exception}. 
The proof for this case is similar to that of \emph{\texttt{INSERT} Exception}.

\item \emph{\texttt{INSERT} FD Exception}.
Let $i$ be such that $r^{i} = r^{i-1} \concat \langle u, \mathtt{INSERT}, R, (\overline{v}, \overline{w}, \overline{q}) \rangle \concat s$, where $s = \langle \mathit{db}, U, \mathit{sec}, T, V, c \rangle \\ \in \Omega_{M}$ and $\mathit{last}(r^{i-1}) =  \langle \mathit{db}', U, \mathit{sec}, T, V, c' \rangle$, and $\phi$ be $\exists \overline{y},\overline{z}.\, R(\overline{v}, \overline{y}, \overline{z}) \wedge \overline{y} \neq \overline{w}$.
We claim that $[\phi]^{\mathit{db}'}$ holds.
From this claim and the LTS semantics, it follows that there is a tuple $(\overline{v}, \overline{w}', \overline{q}')$ in $\mathit{db}'(R)$ such that $\overline{w}' \neq \overline{w}$.
From this and $\mathit{db} = \mathit{db}'$, it follows that there is a tuple $(\overline{v}, \overline{w}', \overline{q}')$ in $\mathit{db}(R)$ such that $\overline{w}' \neq \overline{w}$.
From this, it follows that also $[\phi]^{\mathit{db}}$ holds.

We now prove our claim that $[\phi]^{\mathit{db}'}$ holds.
Assume, for contradiction's sake, that this is not the case.
This means that there is no tuple $(\overline{v}', \overline{w}', \overline{q}')$ in $\mathit{db}'(R)$ such that $\overline{v}' = \overline{v}$ and $\overline{w}' \neq \overline{w}$.
Therefore, for all tuples $(\overline{v}', \overline{w}', \overline{q}')$ in $\mathit{db}'(R)$, if $\overline{v} = \overline{v}'$, then  $\overline{w}' = \overline{w}$.
From this and $\mathit{db}'[R \oplus (\overline{v}, \overline{w}, \overline{q})](R) = \mathit{db}'(R) \cup \{(\overline{v}, \overline{w}, \overline{q})\}$, it follows that for all tuples $(\overline{v}', \overline{w}', \overline{q}')$ in $\mathit{db}'[R \oplus (\overline{v}, \overline{w}, \overline{q})](R)$, if $\overline{v} = \overline{v}'$, then  $\overline{w}' = \overline{w}$.
Furthermore, from $\mathit{db}' \in \Omega_{D}^{\Gamma}$ , it follows that for all tuples $(\overline{v}', \overline{w}', \overline{q}')$ and $(\overline{v}', \overline{w}'', \overline{q}'')$ in $\mathit{db}(R)$  such that $\overline{v}' \neq \overline{v}$, then $\overline{w}' = \overline{w}$.
From this and $\mathit{db}[R \oplus (\overline{v}, \overline{w}, \overline{q})](R) = \mathit{db}'(R) \cup \{(\overline{v}, \overline{w}, \overline{q})\}$,  it follows that for all tuples $(\overline{v}', \overline{w}', \overline{q}')$ and $(\overline{v}', \overline{w}'', \overline{q}'')$ in $\mathit{db}'[R \oplus (\overline{v}, \overline{w}, \overline{q})](R)$  such that $\overline{v}' \neq \overline{v}$, then $\overline{w}' = \overline{w}$.
From these facts and the relational calculus semantics, it follows that $[\forall \overline{x}, \overline{y}, \overline{y}', \overline{z}, \overline{z}'.\,( (R(\overline{x}, \overline{y},  \overline{z}) \wedge R(\overline{x}, \overline{y}',  \overline{z}') )\Rightarrow \overline{y} = \overline{y}'  ]^{\mathit{db}'[R \oplus (\overline{v}, \overline{w}, \overline{q})]} = \top$.
This is in contradiction with the fact that the constraint $\forall \overline{x}, \overline{y}, \overline{y}', \overline{z}, \overline{z}'.\,( (R(\overline{x}, \overline{y},  \overline{z}) \wedge R(\overline{x}, \overline{y}',  \overline{z}') )\Rightarrow \overline{y} = \overline{y}'$ is in $\mathit{Ex}(\mathit{last}(r^{i}))$.

\item \emph{\texttt{INSERT} ID Exception}.
Let $i$ be such that $r^{i} = r^{i-1} \concat \langle u, \mathtt{INSERT}, R, (\overline{v}, \overline{w}) \rangle \concat s$, where $s = \langle \mathit{db}, U, \mathit{sec}, T, V, c \rangle \\ \in \Omega_{M}$ and $\mathit{last}(r^{i-1}) =  \langle \mathit{db}', U, \mathit{sec}, T, V, c' \rangle$, and $\phi$ be $\forall \overline{x},\overline{y}.\, S(\overline{x}, \overline{y}) \Rightarrow \overline{x} \neq \overline{v}$.
We claim that $[\phi]^{\mathit{db}'}$ holds.
From this claim and the LTS semantics, it follows that there is no tuple $(\overline{v}, \overline{w}')$ in $\mathit{db}'(S)$.
From this and $\mathit{db}(S) = \mathit{db}'(S)$, it follows that there no tuple $(\overline{v}, \overline{w}')$ in $\mathit{db}(S)$.
From this, it follows that also $[\phi]^{\mathit{db}}$ holds.

We now prove our claim that $[\phi]^{\mathit{db}'}$ holds.
Assume, for contradiction's sake, that this is not the case.
This means that there is a tuple $(\overline{v}, \overline{w}')$ in $\mathit{db}'(S)$, for some $\overline{w}'$.
From  $\mathit{db}' \in \Omega_{D}^{\Gamma}$, it follows that for all tuples $(\overline{x},\overline{z}) \in \mathit{db}'(R)$ such that $\overline{x} \neq \overline{v}$, there is a tuple $(\overline{x}, \overline{y}) \in \mathit{db}'(S)$.
From this, $(\overline{v}, \overline{w}')$ in $\mathit{db}'(S)$, $\mathit{db}'[R\oplus(\overline{v},\overline{w})](S) = \mathit{db}'(S)$, and $\mathit{db}'[R\oplus(\overline{v},\overline{w})](R) = \mathit{db}'(R)\cup\{(\overline{v},\overline{w})\}$, it follows that for all tuples $(\overline{x},\overline{z}) \in \mathit{db}'[R\oplus(\overline{v},\overline{w})](R)$, there is a tuple $(\overline{x}, \overline{y}) \in \mathit{db}'[R\oplus(\overline{v},\overline{w})](S)$.
From these facts and the relational calculus semantics, it follows that $[\forall \overline{x}, \overline{z}.\,( R(\overline{x}, \overline{z}) \Rightarrow \exists \overline{w}.\, S(\overline{x}, \overline{w})) ]^{\mathit{db}'[R\oplus(\overline{v},\overline{w})]} = \top$.
This is in contradiction with the fact that the constraint $\forall \overline{x}, \overline{z}.\,( R(\overline{x}, \overline{z})  \Rightarrow \exists \overline{w}.\, S(\overline{x}, \overline{w}))$ is in $\mathit{Ex}(\mathit{last}(r^{i}))$.

\item \emph{\texttt{DELETE} FD Exception}.
Let $i$ be such that $r^{i} = r^{i-1} \concat \langle u, \mathtt{DELETE}, R, (\overline{v}, \overline{w}) \rangle \concat s$, where $s = \langle \mathit{db}, U, \mathit{sec}, T, V, c \rangle \\ \in \Omega_{M}$ and $\mathit{last}(r^{i-1}) =  \langle \mathit{db}', U, \mathit{sec}, T, V, c' \rangle$, and $\phi$ be $\exists \overline{z}.\, S(\overline{v},\overline{z}) \wedge \forall \overline{y}.\,(R(\overline{v},\overline{y}) \Rightarrow \overline{y} = \overline{w} )$.
We claim that $[\phi]^{\mathit{db}'}$ holds.
From this claim and the LTS semantics, it follows that there is a tuple $(\overline{v},\overline{z})$ in $\mathit{db}'(S)$ and all tuples $(\overline{v},\overline{y}) \in \mathit{db}'(R)$ are such that $\overline{y} = \overline{w}$. 
From $(\overline{v},\overline{z})$ in $\mathit{db}'(S)$ and $\mathit{db}(S) = \mathit{db}'(S)$, it follows that $(\overline{v},\overline{z})$ in $\mathit{db}'(S)$.
From the fact that all tuples $(\overline{v},\overline{y}) \in \mathit{db}'(R)$ are such that $\overline{y} = \overline{w}$ and  $\mathit{db}(R) = \mathit{db}'(R) \}$, it follows that all tuples $(\overline{v},\overline{y}) \in \mathit{db}(R)$ are such that $\overline{y} = \overline{w}$. 
From $(\overline{v},\overline{z})$ in $\mathit{db}(S)$ and the fact that all tuples $(\overline{v},\overline{y}) \in \mathit{db}(R)$ are such that $\overline{y} = \overline{w}$, it follows that $[\phi]^{\mathit{db}}$ holds.

We now prove our claim that $[\phi]^{\mathit{db}'}$ holds.
Assume, for contradiction's sake, that this is not the case.
There are two cases:
\begin{compactenum}
\item  all tuples $(\overline{x},\overline{y}) \in \mathit{db}'(S)$ are such that $\overline{v} \neq \overline{x}$. 
Furthermore, from $\mathit{db}' \in \Omega_{D}^{\Gamma}$, it follows that for all tuples $(\overline{x},\overline{y}) \in \mathit{db}(S)$ such that $\overline{v} \neq \overline{x}$, there is a tuple $(\overline{x},\overline{z}) \in \mathit{db}(R)$.
From these facts,  $\mathit{db}'[R\ominus(\overline{v}, \overline{w})] (S) = \mathit{db}'(S)$, and $\mathit{db}'[R\ominus(\overline{v}, \overline{w})] (R) = \mathit{db}'(R) \\ \setminus \{(\overline{v}, \overline{w})\}$, it follows that for all tuples $(\overline{x},\overline{y}) \in \mathit{db}'[R\ominus(\overline{v}, \overline{w})](S)$, there is a tuple $(\overline{x},\overline{z}) \in \mathit{db}'[R\ominus(\overline{v}, \overline{w})](R)$.
From this  and the relational calculus semantics, it follows that \[[\forall \overline{x}, \overline{z}.\,( S(\overline{x}, \overline{z}) \Rightarrow \exists \overline{w}.\, R(\overline{x}, \overline{w}) ]^{\mathit{db}'[R\ominus(\overline{v}, \overline{w}))]}  = \top.\]
This contradicts the fact that the constraint $\forall \overline{x}, \overline{z}. \\ ( S(\overline{x}, \overline{z}) \Rightarrow \exists \overline{w}.\, R(\overline{x}, \overline{w}))$ is in $\mathit{Ex}(\mathit{last}(r^{i}))$.

\item there is a tuple $(\overline{v},\overline{w}') \in \mathit{db}'(R)$ such that $\overline{w} \neq \overline{w}'$. 
Furthermore, from $\mathit{db}' \in \Omega_{D}^{\Gamma}$, it follows that for all tuples $(\overline{x},\overline{y}) \in \mathit{db}'(S)$ such that $\overline{v} \neq \overline{x}$, there is a tuple $(\overline{x},\overline{z}) \in \mathit{db}'(R)$.
From these facts, $\mathit{db}'[R\ominus(\overline{v}, \overline{w})] (S) = \mathit{db}'(S)$, and $\mathit{db}'[R\ominus(\overline{v}, \overline{w})] (R) = \mathit{db}'(R) \\ \setminus \{(\overline{v}, \overline{w})\}$, it follows that for all tuples $(\overline{x},\overline{y}) \in \mathit{db}'[R\ominus(\overline{v}, \overline{w})](S)$, there is a tuple $(\overline{x},\overline{z}) \in \mathit{db}'[R\ominus(\overline{v}, \overline{w})](R)$.
From this and the relational calculus semantics, it follows that \[[\forall \overline{x}, \overline{z}.\,( S (\overline{x}, \overline{z}) \Rightarrow \exists \overline{w}.\, R(\overline{x}, \overline{w}) ]^{\mathit{db}'[R\ominus(\overline{v}, \overline{w}))]}  = \top.\]
This contradicts the fact that the constraint $\forall \overline{x}, \overline{z}.\\( S(\overline{x}, \overline{z}) \Rightarrow \exists \overline{w}.\, R(\overline{x}, \overline{w}))$ is in $\mathit{Ex}(\mathit{last}(r^{i}))$.
\end{compactenum}

\item \emph{Integrity Constraint}.
The proof of this case follows trivially from the fact that for any state $s = \langle \mathit{db}, U, \mathit{sec},  T, \\ V,  c \rangle  \in \Omega_{M}$ and any $\gamma \in \Gamma$, $[\gamma]^{\mathit{db}} = \top$ holds by definition.

\item \emph{Learn \texttt{GRANT}/\texttt{REVOKE} Backward}. 
Let $i$ be such that $r^{i} = r^{i-1} \concat t \concat s$, where $s = \langle \mathit{db}, U, \mathit{sec}, T, V, c \rangle  \in \Omega_{M}$, $\mathit{last}(r^{i-1}) =  \langle \mathit{db}, U, \mathit{sec}', T, V, c' \rangle$, and $t$ be a trigger whose \texttt{WHEN} condition is $\phi$ and whose action is either a \texttt{GRANT} or a \texttt{REVOKE}.
From the rule's definition, it follows  $\mathit{sec} \neq \mathit{sec}'$.
We now prove that $[\phi]^{\mathit{db}} = \top$.
Assume, for contradiction's sake, that $[\phi]^{\mathit{db}} = \bot$.
From this and the LTS rules for the triggers, it follows that the trigger $t$ is disabled.
Therefore, according to the \emph{Trigger Disabled} rule, $\mathit{sec} =\mathit{sec}'$, which leads to a contradiction.

\item \emph{Trigger \texttt{GRANT} Disabled Backward}. 
Let $i$ be such that $r^{i} = r^{i-1} \concat t \concat s$, where $s = \langle \mathit{db}, U, \mathit{sec}, T, V, c \rangle  \in \Omega_{M}$, $\mathit{last}(r^{i-1}) =  \langle \mathit{db}, U, \mathit{sec}', T, V, c' \rangle$, and $t$ be a trigger whose \texttt{WHEN} condition is $\psi$, and $\phi$ be $\neg \psi$.
Furthermore, let $g \in \Omega^{\mathit{sec}}_{{\cal U}, D}$ be the \texttt{GRANT} added by the trigger.
From the rule's definition, it follows  $g \not\in \mathit{sec}'$.
We now prove that $[\phi]^{\mathit{db}} = \top$.
Assume, for contradiction's sake, that $[\phi]^{\mathit{db}} = \perp$.
This would imply that the trigger $t$ is enabled.
There are two cases:
\begin{compactenum}
\item $t$'s execution is authorized. Therefore, $g \in \mathit{sec}'$, which contradicts $g \not\in \mathit{sec}'$.

\item  $t$'s execution is not authorized. This contradicts $\mathit{secEx}(s) = \bot$.
\end{compactenum}

\item \emph{Trigger \texttt{REVOKE} Disabled Backward}. 
The proof for this case is similar to that of \emph{Trigger \texttt{GRANT} Disabled Backward}.

\item \emph{Trigger \texttt{INSERT} FD Exception}.
The proof for this case is similar to that of \emph{\texttt{INSERT} FD Exception}.

\item \emph{Trigger \texttt{INSERT} ID Exception}.
The proof for this case is similar to that of \emph{\texttt{INSERT} ID Exception}.

\item \emph{Trigger \texttt{DELETE} ID Exception}.
The proof for this case is similar to that of \emph{\texttt{DELETE} ID Exception}.

\item \emph{Trigger Exception}.
Let $i$ be such that $r^{i} = r^{i-1} \concat t \concat s$, where $s = \langle \mathit{db}, U, \mathit{sec}, T, V, c \rangle  \in \Omega_{M}$, $\mathit{last}(r^{i-1}) =  \langle \mathit{db}, U, \mathit{sec}', T, V, c' \rangle$, and $t$ be a trigger whose \texttt{WHEN} condition is $\phi$ and whose action is $\mathit{act}$.
From the rule's definition, it follows that $t$ is enabled and that the evaluation of the \texttt{WHEN} condition is authorized.
From this and the LTS's rules, it follows that $[\phi]^{\mathit{db}} = \top$. 

\item \emph{Trigger \texttt{INSERT} Exception}.
The proof for this case is similar to that of \emph{\texttt{INSERT} Exception}.

\item \emph{Trigger \texttt{DELETE} Exception}.
The proof for this case is similar to that of \emph{\texttt{DELETE} Exception}.

\item \emph{Trigger Rollback \texttt{INSERT}}.
Let $i$ be such that $r^{i} = r^{i-n-1}  \concat \langle u, \mathtt{INSERT}, R, \overline{t}\rangle \concat s_{1} \concat t_{1} \concat s_{2} \concat \ldots \concat t_{n} \concat s_{n}$, where $s_{1}, s_{2}, \ldots, s_{n} \\ \in \Omega_{M}$ and $t_{1}, \ldots, t_{n} \in {\cal TRIGGER}_{D}$, and $\phi$ be $\neg R(\overline{t})$.
Furthermore, let  $\mathit{last}(r^{i-n-1}) = \langle \mathit{db}', U', \mathit{sec}', T', V', c' \rangle$ and $s_n$ be $\langle \mathit{db}, U, \mathit{sec}, T, V, c \rangle$. 
Assume, for contradiction's sake, that $[\phi]^{\mathit{db}} = \bot$.
Therefore, $\overline{t} \in \mathit{db}(R)$.
From the LTS rules, it follows that $\mathit{db}' = \mathit{db}$.
From this and $\overline{t} \in \mathit{db}(R)$, it follows $\overline{t} \in \mathit{db}'(R)$.
From $r$'s definition and the LTS rule \emph{\texttt{INSERT} Success - 2}, it follows that $\overline{t} \not\in \mathit{db}'(R)$, which leads to a contradiction.

\item \emph{Trigger Rollback \texttt{DELETE}}.
The proof for this case is similar to that of \emph{Trigger Rollback \texttt{INSERT}}.
\end{compactenum}
This completes the proof of the base step.

\smallskip
\noindent
{\bf Induction Step: }
Assume that the claim hold for any derivation of $r, j \attMod \psi$ such that $|r, j \attMod \psi| < |r,i \attMod \phi|$.
We now prove that the claim also holds for $r,i \attMod \phi$.
There are a number of cases depending on the rule used to obtain $r,i \attMod \phi$.
\begin{compactenum}
\item \emph{View}.
The proof of this case follows trivially from the semantics of the relational calculus extended over views.

\item \emph{Propagate Forward \texttt{SELECT}}.
Let $i$ be such that $r^{i+1} = r^{i} \concat \langle u, \mathtt{SELECT}, \psi \rangle \concat s$, where $s = \langle \mathit{db}, U, \mathit{sec}, T, V, c \rangle  \in \Omega_{M}$ and $\mathit{last}(r^{i}) =  \langle \mathit{db}', U', \mathit{sec}', T', V', c' \rangle$.
From the rule's definition, $r, i \attMod \phi$ holds.
From this, the induction hypothesis, and $\mathit{last}(r^{i}) =  \langle \mathit{db}', U', \mathit{sec}', T', V', c' \rangle$, it follows that $[\phi]^{\mathit{db}'} = \top$.
From the LTS semantics, it follows that $\mathit{db} = \mathit{db}'$.
From this and $[\phi]^{\mathit{db}'} = \top$, it follows that $[\phi]^{\mathit{db}} = \top$.

\item \emph{Propagate Forward \texttt{GRANT/REVOKE}}.
The proof for this case is similar to that of \emph{Propagate Forward \texttt{SELECT}}.

\item \emph{Propagate Forward \texttt{CREATE}}.
The proof for this case is similar to that of \emph{Propagate Forward \texttt{SELECT}}.

\item \emph{Propagate Backward \texttt{SELECT}}.
Let $i$ be such that $r^{i+1} = r^{i} \concat \langle u, \mathtt{SELECT}, \psi \rangle \concat s$, where $s = \langle \mathit{db}', U', \mathit{sec}', T', V', c' \rangle  \\ \in \Omega_{M}$ and $\mathit{last}(r^{i}) =  \langle \mathit{db}, U, \mathit{sec}, T, V, c \rangle$.
From the rule's definition, $r, i+1 \attMod \phi$ holds.
From this, the induction hypothesis, $r^{i+1} = r^{i} \concat \langle u, \mathtt{SELECT}, \psi \rangle \concat s$, and $s = \langle \mathit{db}, U, \mathit{sec}, T, V, c \rangle$, it follows that $[\phi]^{\mathit{db}'} = \top$.
From the LTS semantics, it follows that $\mathit{db} = \mathit{db}'$.
From this and $[\phi]^{\mathit{db}'} = \top$, it follows that $[\phi]^{\mathit{db}} = \top$.

\item \emph{Propagate Backward \texttt{GRANT/REVOKE}}.
The proof for this case is similar to that of \emph{Propagate Backward \texttt{SELECT}}.

\item \emph{Propagate Backward \texttt{CREATE TRIGGER}}.
The proof for this case is similar to that of \emph{Propagate Backward \texttt{SELECT}}.

\item \emph{Propagate Backward \texttt{CREATE VIEW}}.
Let $i$ be such that $r^{i+1} = r^{i} \concat \langle u, \mathtt{CREATE}, o \rangle \concat s$, where $s = \langle \mathit{db}', U', \mathit{sec}', T', \\ V', c' \rangle  \in \Omega_{M}$ and $\mathit{last}(r^{i}) =  \langle \mathit{db}, U, \mathit{sec}, T, V, c \rangle$.
From the rule's definition, $r, i+1 \attMod \phi'$ holds.
From this, the induction hypothesis, $r^{i+1} = r^{i} \concat \langle u, \mathtt{SELECT}, \psi \rangle \concat s$, and $s = \langle \mathit{db}, U, \mathit{sec}, T, V, c \rangle$, it follows that $[\phi']^{\mathit{db}'} = \top$.
From the definition of $\mathit{replace}$, it follows that $\mathit{replace}(\phi', \\ o)$ and $\phi'$ are semantically equivalent.
From this and  $[\phi']^{\mathit{db}'} = \top$,  $[\mathit{replace}(\phi',o)]^{\mathit{db}'} = \top$.
From the LTS semantics, it follows that $\mathit{db} = \mathit{db}'$.
From this and $[\mathit{replace}(\phi',o)]^{\mathit{db}'} = \top$, it follows that $[\mathit{replace}(\phi',o)]^{\mathit{db}} \\ = \top$.

\item \emph{Rollback Backward - 1}.
Let $i$ be such that $r^{i} = r^{i-n-1}  \concat \langle u, \mathit{op}, R, \overline{t}\rangle \concat s_{1} \concat t_{1} \concat s_{2} \concat \ldots \concat t_{n} \concat s_{n}$, where  $s_{1}, s_{2}, \ldots, s_{n} \in \Omega_{M}$, $t_{1}, \ldots, t_{n} \in {\cal TRIGGER}_{D}$, and  $\mathit{op}$ is one of  $\{\mathtt{INSERT}, \\ \texttt{DELETE}\}$.
Furthermore, let $s_{n}$ be $\langle \mathit{db}', U', \mathit{sec}', T',  V', c' \rangle$ and $\mathit{last}(r^{i-n-1})$ be $\langle \mathit{db}, U, \mathit{sec}, T, V, c \rangle$.
From the rule's definition, $r, i \attMod \phi$ holds.
From this, the induction hypothesis, and $s_{n} = \langle \mathit{db}, U, \mathit{sec}, T, V, c \rangle  \in \Omega_{M}$, it follows that $[\phi]^{\mathit{db}'} = \top$.
From the LTS semantics, it follows that $\mathit{db} = \mathit{db}'$ (because a roll-back happened).
From this and $[\phi]^{\mathit{db}'} = \top$, it follows that $[\phi]^{\mathit{db}} = \top$.

\item \emph{Rollback Backward - 2}.
Let $i$ be such that $r^{i} = r^{i-1} \concat \langle u, op, R, \overline{t} \rangle \concat s$, where $s = \langle \mathit{db}', U', \mathit{sec}', T', V', c' \rangle  \in \Omega_{M}$, $\mathit{last}(r^{i-1}) =  \langle \mathit{db}, U,  \mathit{sec}, T, V, c \rangle$, and $\mathit{op}$ is one of $\{\mathtt{INSERT}, \mathtt{DELETE}\}$.
From the rule's definition, $r, i \attMod \phi$ holds.
From this, the induction hypothesis, $r^{i} = r^{i-1} \concat \langle u, op, R,  \overline{t} \rangle \concat s$, and $s = \langle \mathit{db}', U', \mathit{sec}', T', V', c' \rangle$, it follows that $[\phi]^{\mathit{db}'} = \top$.
From the LTS semantics, it follows that $\mathit{db} = \mathit{db}'$ (because a roll-back happened).
From this and $[\phi]^{\mathit{db}'} = \top$, it follows that $[\phi]^{\mathit{db}} = \top$.

\item \emph{Rollback Forward - 1}.
Let $i$ be such that $r^{i} = r^{i-n-1}  \concat \langle u, \mathit{op}, R, \overline{t}\rangle \concat s_{1} \concat t_{1} \concat s_{2} \concat \ldots \concat t_{n} \concat s_{n}$, where $s_{1}, s_{2}, \ldots, s_{n} \in \Omega_{M}$, $t_{1}, \ldots, t_{n} \in {\cal TRIGGER}_{D}$, and $\mathit{op}$ is one of $\{\mathtt{INSERT}, \\ \mathtt{DELETE}\}$.
Furthermore, let $s_{n}$ be $\langle \mathit{db}, U, \mathit{sec}, T, V, c \rangle $ and $\mathit{last}(r^{i-n-1})$ be $\langle \mathit{db}', U', \mathit{sec}', T', V', c' \rangle$.
From the rule's definition, $r, i -n -1 \attMod \phi$ holds.
From this, the induction hypothesis, and $\mathit{last}(r^{i-n-1}) = \langle \mathit{db}', U', \mathit{sec}',\\ T', V', c' \rangle$., it follows that $[\phi]^{\mathit{db}'} = \top$.
From the LTS semantics, it follows that $\mathit{db} = \mathit{db}'$ (because a roll-back happened).
From this and $[\phi]^{\mathit{db}'} = \top$, it follows that $[\phi]^{\mathit{db}} = \top$.

\item \emph{Rollback Forward - 2}.
Let $i$ be such that $r^{i} = r^{i-1} \concat \langle u, op, R, \overline{t} \rangle \concat s$, where $\mathit{op} \in \{\mathtt{INSERT}, \mathtt{DELETE}\}$, $s = \langle \mathit{db}, U, \\ \mathit{sec}, T, V, c \rangle  \in \Omega_{M}$ and $\mathit{last}(r^{i-1}) =  \langle \mathit{db}', U', \mathit{sec}',  T', V', \\ c' \rangle$.
From the rule's definition, $r, i-1 \attMod \phi$ holds.
From this, the induction hypothesis, and $\mathit{last}(r^{i-1})  =  \langle \mathit{db}', U', \mathit{sec}',  T', V', c' \rangle$, it follows that $[\phi]^{\mathit{db}'} = \top$.
From the LTS semantics, it follows that $\mathit{db} = \mathit{db}'$ (because a roll-back happened).
From this and $[\phi]^{\mathit{db}'} = \top$, it follows that $[\phi]^{\mathit{db}} = \top$.

\item \emph{Propagate Forward \texttt{INSERT/DELETE} Success}.
Let $i$ be such that $r^{i} = r^{i-1} \concat \langle u, op, R, \overline{t} \rangle \concat s$, where $\mathit{op} \in \{\mathtt{INSERT}, \\ \mathtt{DELETE}\}$, $s = \langle \mathit{db}, U, \mathit{sec}, T, V, c \rangle  \in \Omega_{M}$ and $\mathit{last}(r^{i-1}) =  \langle \mathit{db}', U', \mathit{sec}',  T', V', c' \rangle$.
From the rule's definition, $r, i-1 \attMod \phi$ holds.
From this, the induction hypothesis, and $\mathit{last}(r^{i-1})  =  \langle \mathit{db}', U', \mathit{sec}',  T', V', c' \rangle$, it follows that $[\phi]^{\mathit{db}'} = \top$.
From $\mathit{reviseBelief}(r^{i-1}, \phi, r^{i}) = \top$, it follows that $R$ does not occur in $\phi$.
From the LTS semantics, it follows that $\mathit{db}(R') = \mathit{db}'(R')$ for all $R' \neq R$.
From this and the fact that $R$ does not occur in $\phi$, it follows that $[\phi]^{\mathit{db}} = \top$.

\item \emph{Propagate Forward \texttt{INSERT} Success - 1}.
Let $i$ be such that $r^{i} = r^{i-1} \concat \langle u, op, R, \overline{t} \rangle \concat s$, where $\mathit{op}$ is one of $\{\mathtt{INSERT},  \mathtt{DELETE}\}$, $s = \langle \mathit{db}, U, \mathit{sec}, T, V, c \rangle  \in \Omega_{M}$ and $\mathit{last}(r^{i-1}) =  \langle \mathit{db}', U', \mathit{sec}',  T', V', c' \rangle$.
From the rule's definition, $r, i-1 \attMod \phi$ and $r, i-1 \attMod R(\overline{t})$ hold.
From this, the induction hypothesis, and $\mathit{last}(r^{i-1})  =  \langle \mathit{db}', U', \mathit{sec}',  T', V', c' \rangle$, it follows that $[\phi]^{\mathit{db}'} = \top$ and $[R(\overline{t})]^{\mathit{db}'} = \top$.
From $[R(\overline{t})]^{\mathit{db}'} = \top$ and the relational calculus' semantics, it follows that $\overline{t} \in \mathit{db}'(R)$.
From the LTS semantics, $\mathit{db} = \mathit{db}' [R \oplus \overline{t}]$.
From this, it follows that  $\mathit{db}(R') = \mathit{db}'(R')$ for all $R' \neq R$ and $\mathit{db}(R) = \mathit{db}'(R) \cup \{\overline{t}\}$.
From this and $\overline{t} \in \mathit{db}'(R)$, it follows that $\mathit{db}(R) = \mathit{db}'(R)$.
Therefore, $\mathit{db} = \mathit{db}'$.
From this and $[\phi]^{\mathit{db}'} = \top$, it follows that $[\phi]^{\mathit{db}} = \top$.

\item \emph{Propagate Forward \texttt{DELETE} Success - 1}.
The proof for this case is similar to that of \emph{Propagate Forward \texttt{INSERT} Success - 1}.

\item \emph{Propagate Backward \texttt{INSERT/DELETE} Success}.
The proof for this case is similar to that of \emph{Propagate Forward \texttt{INSERT/DELETE} Success}.

\item \emph{Propagate Backward \texttt{INSERT} Success - 1}.
The proof for this case is similar to that of \emph{Propagate Forward \texttt{INSERT} Success - 1}.

\item \emph{Propagate Backward \texttt{DELETE} Success - 1}.
The proof for this case is similar to that of \emph{Propagate Forward \texttt{DELETE} Success - 1}.

\item \emph{Reasoning}.
Let $\Phi$ be a subset of $\{ \phi \,| \, r,i \attMod \phi\}$ and $\mathit{last}(r^{i}) =  \langle \mathit{db}, U, \mathit{sec},  T, V, c \rangle$.
From the induction hypothesis, it follows that $[\phi]^{\mathit{db}} = \top$ for any $\phi \in \Phi$.
From the rule's definition, it follows that $\Phi \models_{\mathit{fin}} \gamma$.
From this and  $[\phi]^{\mathit{db}} = \top$ for any $\phi \in \Phi$, it follows that $[\gamma]^{\mathit{db}} = \top$.

\item \emph{Learn \texttt{INSERT} Backward - 3}.
Let $i$ be such that $r^{i} = r^{i-1} \concat \langle u, \mathtt{INSERT}, R, \overline{t} \rangle \concat s$, where $s = \langle \mathit{db}', U', \mathit{sec}', T', V', \\ c' \rangle  \in \Omega_{M}$ and $\mathit{last}(r^{i-1}) =  \langle \mathit{db}, U, \mathit{sec},  T, V, c \rangle$, and $\phi$ be $\neg R(\overline{t})$.
We prove that $[\neg R(\overline{t})]^{\mathit{db}} = \top$.
Assume, for contradiction's sake, that $[\neg R(\overline{t})]^{\mathit{db}} = \bot$.
From this and the relational calculus semantics, it follows that $\overline{t} \in \mathit{db}(R)$.
From this and the LTS semantics, it follows that $\mathit{db} = \mathit{db}'$ because $\mathit{db}' = \mathit{db}[R\oplus\overline{t}]$.
However, from the rule's definition, there is a $\psi$ such that $r, i-1 \attMod \psi$ and $r, i \attMod \neg \psi$ hold.
From this, the induction hypothesis, $s = \langle \mathit{db}', U', \mathit{sec}', T', V', c' \rangle$, and $\mathit{last}(r^{i-1})  =  \langle \mathit{db}, U, \mathit{sec},  T, V, c \rangle$, it follows that $[\psi]^{\mathit{db}} = \top$ and $[\neg \psi]^{\mathit{db}'} = \top$.
Therefore, $[\psi]^{\mathit{db}} = \top$ and $[\psi]^{\mathit{db}'} = \bot$.
Hence, $\mathit{db} \neq \mathit{db}'$ leading to a contradiction with $\mathit{db} = \mathit{db}'$.

\item \emph{Learn \texttt{DELETE} Backward - 3}.
The proof for this case is similar to that of \emph{Learn \texttt{INSERT} Backward - 3}.

\item \emph{Propagate Forward Disabled Trigger}.
Let $i$ be such that $r^{i} = r^{i-1} \concat t \concat s$, where $s = \langle \mathit{db}, U, \mathit{sec}, T, V, c \rangle  \in \Omega_{M}$, $\mathit{last}(r^{i-1}) =  \langle \mathit{db}, U, \mathit{sec},  T, V, c \rangle$, and $t$ be a trigger.
Furthermore, let $\psi$ be $t$'s condition where all free variables are replaced with the values in $\mathit{tpl}(\mathit{last}(r^{i-1}))$.
From the rule's definition, it follows that $r, i-1 \attMod \neg\psi$ holds.
From this and the induction hypothesis, it follows that $[\psi]^{\mathit{db}'} = \bot$.
From this, the fact that $\psi$ is $t$'s \texttt{WHEN} condition, and the rule \emph{Trigger Disabled}, it follows that $\mathit{db} = \mathit{db}'$.
From the rule's definition, it follows that $r, i-1 \attMod \phi$ holds.
From this, the induction hypothesis, and $\mathit{last}(r^{i-1}) =  \langle \mathit{db}, U, \mathit{sec},  T, V, c \rangle$, it follows that $[\phi]^{\mathit{db}'} = \top$.
From this and $\mathit{db} = \mathit{db}'$, it follows that $[\phi]^{\mathit{db}} = \top$.

\item \emph{Propagate Backward Disabled Trigger}.
The proof for this case is similar to that of  \emph{Propagate Forward Disabled Trigger}.

\item \emph{Learn \texttt{INSERT} Forward}.
Let $i$ be such that $r^{i} = r^{i-1} \concat t \concat s$, where $s = \langle \mathit{db}, U, \mathit{sec}, T, V, c \rangle  \in \Omega_{M}$, $\mathit{last}(r^{i-1}) =  \langle \mathit{db}, U, \mathit{sec},  T, V, c \rangle$, and $t$ be a trigger, and $\phi$ be $R(\overline{t})$.
Furthermore, let $\psi$ be $t$'s condition where all free variables are replaced with the values in $\mathit{tpl}(\mathit{last}(r^{i-1}))$.
From the rule's definition, it follows that $r, i-1 \attMod \psi$ holds.
From this and the induction hypothesis, it follows that $[\psi]^{\mathit{db}'} = \bot$.
Furthermore, from the rule's definition, it follows that $\mathit{secEx}(s) = \bot$ and $\mathit{Ex}(s) = \emptyset$.
From this, the fact that $\psi$ is $t$'s \texttt{WHEN} condition, $[\psi]^{\mathit{db}'} = \bot$, and the rule \emph{Trigger \texttt{DELETE-INSERT} Success}, it follows that $\mathit{db} = \mathit{db}'[R \oplus \overline{t}]$.
From the definition of $\oplus$, it follows that $\overline{t} \in \mathit{db}(R)$.
From this and the relational calculus semantics, it follows that $[\phi]^{\mathit{db}} = \top$.

\item \emph{Learn \texttt{INSERT} - FD}.
Let $i$ be such that $r^{i} = r^{i-1} \concat t \concat s$, where $s = \langle \mathit{db}, U, \mathit{sec}, T, V, c \rangle \in \Omega_{M}$, $\mathit{last}(r^{i-1}) =  \langle \mathit{db}', U', \mathit{sec}', T', V', c' \rangle$, and $t \in {\cal TRIGGER}_{D}$, and $\phi$ be $\neg \exists \overline{y},\overline{z}.\, R(\overline{v}, \overline{y}, \overline{z}) \wedge \overline{y} \neq \overline{w}$.
Furthermore,  let $\psi$ be $t$'s condition where all free variables are replaced with the values in $\mathit{tpl}(\mathit{last}(r^{i-1}))$ and $\langle u', \mathtt{INSERT}, R, (\overline{v}, \overline{w}, \overline{q}) \rangle$ be $t$'s actual action.
We claim that $\mathit{db}(R) = \mathit{db}'(R) \cup \{(\overline{v}, \overline{w}, \overline{q})\}$.
Furthermore, we claim that $[\phi]^{\mathit{db}}$ holds.
From this claim and the relational calculus semantics, it follows that there is no tuple $(\overline{v}', \overline{w}', \overline{q}')$ in $\mathit{db}(R)$ such that $\overline{v}' = \overline{v}$ and $\overline{w}' \neq \overline{w}$.
From this and $\mathit{db}(R) = \mathit{db}'(R) \cup \{(\overline{v}, \overline{w}, \overline{q})\}$, it follows that there is no tuple $(\overline{v}', \overline{w}', \overline{q}')$ in $\mathit{db}'(R)$ such that $\overline{v}' = \overline{v}$ and $\overline{w}' \neq \overline{w}$.
From this, it follows that also $[\phi]^{\mathit{db}'}$ holds.

We now prove our claim that $\mathit{db}(R) = \mathit{db}'(R) \cup \{(\overline{v}, \overline{w}, \overline{q})\}$.
Assume, for contradiction's sake, that this is not the case.
Since $\mathit{db}$ is obtained from $\mathit{db}'$, this would imply that the trigger $t$ is disabled.
Hence, this would imply that $[\psi]^{\mathit{db}'} = \bot$.
From the rule's definition, $r, i-1 \attMod \psi$.
From this, the induction's hypothesis, and $\mathit{last}(r^{i-1}) =  \langle \mathit{db}', U, \mathit{sec}, T, V, c' \rangle$, it follows that $[\psi]^{\mathit{db}'} = \top$, which contradicts $[\psi]^{\mathit{db}'} = \bot$.

We now prove our claim that $[\phi]^{\mathit{db}}$ holds.
Assume, for contradiction's sake, that this is not the case.
This means that there is a tuple $(\overline{v}', \overline{w}', \overline{q}')$ in $\mathit{db}(R)$ such that $\overline{v}' = \overline{v}$ and $\overline{w}' \neq \overline{w}$.
Note that, as we proved before, $(\overline{v}, \overline{w}, \overline{q}) \in \mathit{db}(R)$.
Therefore, there are two tuples $(\overline{v}, \overline{w}, \overline{q})$ and $(\overline{v}, \overline{w}', \overline{q}')$ in $\mathit{db}(R)$ such that $\overline{w}' \neq \overline{w}$.
From this and the relational calculus semantics, it follows that $[\forall \overline{x}, \overline{y}, \overline{y}', \overline{z}, \overline{z}'.\,( (R(\overline{x}, \overline{y},  \overline{z}) \wedge R(\overline{x}, \overline{y}',  \overline{z}') )\Rightarrow \overline{y} = \overline{y}'  ]^{\mathit{db}} = \bot$.
This is in contradiction with the fact that the constraint $\forall \overline{x}, \overline{y}, \overline{y}', \overline{z}, \overline{z}'.\,( (R(\overline{x}, \overline{y},  \overline{z}) \wedge R(\overline{x}, \overline{y}', \\ \overline{z}') )\Rightarrow \overline{y} = \overline{y}'$ is in $\Gamma$.
Indeed, since the constraint is in $\Gamma$, any state in $\Omega_{D}^{\Gamma}$ must satisfy it.

\item \emph{Learn \texttt{INSERT} - FD - 1}.
Let $i$ be such that $r^{i} = r^{i-1} \concat t \concat s$, where $s = \langle \mathit{db}, U, \mathit{sec}, T, V, c \rangle \in \Omega_{M}$, $\mathit{last}(r^{i-1}) =  \langle \mathit{db}', U', \mathit{sec}', T', V', c' \rangle$, and $t \in {\cal TRIGGER}_{D}$, and $\phi$ be $\neg \exists \overline{y},\overline{z}.\, R(\overline{v}, \overline{y}, \overline{z}) \wedge \overline{y} \neq \overline{w}$.
Furthermore,  let $\psi$ be $t$'s condition where all free variables are replaced with the values in $\mathit{tpl}(\mathit{last}(r^{i-1}))$ and $\langle u', \mathtt{INSERT}, R, (\overline{v}, \overline{w}, \overline{q}) \rangle$ be $t$'s actual action.
From the rule's definition, $r, i-1 \attMod \psi$.
From this, the induction's hypothesis, and $\mathit{last}(r^{i-1}) =  \langle \mathit{db}', U, \mathit{sec}, T, V, c' \rangle$, it follows that $[\psi]^{\mathit{db}'} = \top$.
From this and the LTS semantics, it follows that the trigger $t$ is enabled in $\mathit{last}(r^{i-1})$.
We now prove our claim that $[\phi]^{\mathit{db}'}$ holds.
Assume, for contradiction's sake, that this is not the case.
This means that there is a tuple $(\overline{v}', \overline{w}', \overline{q}')$ in $\mathit{db}'(R)$ such that $\overline{v}' = \overline{v}$ and $\overline{w}' \neq \overline{w}$.
Let $\mathit{db}''$ be the state $\mathit{db}'[R\oplus (\overline{v}, \overline{w}, \overline{q})]$.
From $\mathit{db}''=\mathit{db}'[R\oplus (\overline{v}, \overline{w}, \overline{q})]$, and the fact that there is a tuple $(\overline{v}', \overline{w}', \overline{q}')$ in $\mathit{db}'(R)$ such that $\overline{v}' = \overline{v}$ and $\overline{w}' \neq \overline{w}$,
it follows that there are two tuples $(\overline{v}, \overline{w}, \overline{q})$ and $(\overline{v}, \overline{w}', \overline{q}')$ in $\mathit{db}''(R)$ such that $\overline{w}' \neq \overline{w}$.
From this and the relational calculus semantics, it follows that $[\forall \overline{x}, \overline{y}, \overline{y}', \overline{z}, \overline{z}'.\,( (R(\overline{x}, \overline{y},  \overline{z}) \wedge R(\overline{x}, \overline{y}',  \overline{z}') )\Rightarrow \overline{y} = \overline{y}'  ]^{\mathit{db}''} = \bot$.
Since the trigger $t$ is enabled, this contradicts the fact that $\forall \overline{x}, \overline{y}, \overline{y}', \overline{z}, \overline{z}'.\,( (R(\overline{x}, \\ \overline{y},  \overline{z}) \wedge R(\overline{x}, \overline{y}',  \overline{z}') )\Rightarrow \overline{y} = \overline{y}' $ is not in $\mathit{Ex}(s)$.

\item \emph{Learn \texttt{INSERT} - ID}.
The proof of this case is similar to that of \emph{Learn \texttt{INSERT} - FD}.
See also the proof of \emph{\texttt{INSERT} Success - ID}.

\item \emph{Learn \texttt{INSERT} - ID - 1}.
The proof of this case is similar to that of \emph{Learn \texttt{INSERT} - FD - 1}.
See also the proof of \emph{\texttt{INSERT} Success - ID}.

\item \emph{Learn \texttt{INSERT} Backward - 1}.
Let $i$ be such that $r^{i} = r^{i-1} \concat t \concat s$, where $s = \langle \mathit{db}', U', \mathit{sec}', T', V', c' \rangle \in \Omega_{M}$, $\mathit{last}(r^{i-1}) =  \langle \mathit{db}, U, \mathit{sec}, T, V, c \rangle$, and $t \in {\cal TRIGGER}_{D}$, and $\phi$ be $t$'s actual \texttt{WHEN} condition, where all free variables are replaced with the values in $\mathit{tpl}(\mathit{last}(r^{i-1}))$.
From the rule's definition, it follows that there is a $\psi$ such that $r, i-1 \attMod \psi$ and $r, i \attMod \neg \psi$.
From this, the induction's hypothesis, $s = \langle \mathit{db}', U', \mathit{sec}', T', V', c' \rangle$, and $\mathit{last}(r^{i-1}) =  \langle \mathit{db}, U, \mathit{sec}, T, V, c \rangle$, it follows that $[\psi]^{\mathit{db}}  = \\ \top$ and $[\neg \psi]^{\mathit{db}'} = \top$.
Therefore, $[\psi]^{\mathit{db}} = \top$ and $[\psi]^{\mathit{db}'} = \bot$.
Hence, $\mathit{db} \neq \mathit{db}'$.
We now prove that $[\phi]^{\mathit{db}} = \top$.
Assume, for contradiction's sake, that $[\phi]^{\mathit{db}} = \bot$.
From the rule's definition, it follows that $\mathit{secEx}(s) = \bot$.
Therefore, $f(\mathit{last}(r^{i-1}), \langle u', \mathtt{SELECT}, \phi\rangle) = \top$.
From this, $[\phi]^{\mathit{db}} = \bot$, and the rule \texttt{Trigger Disabled}, it follows that $\mathit{db} = \mathit{db}'$, which contradicts  $\mathit{db} \neq \mathit{db}'$.

\item \emph{Learn \texttt{INSERT} Backward - 2}.
Let $i$ be such that $r^{i} = r^{i-1} \concat t \concat s$, where $s = \langle \mathit{db}', U', \mathit{sec}', T', V', c' \rangle \in \Omega_{M}$, $\mathit{last}(r^{i-1}) =  \langle \mathit{db}, U, \mathit{sec}, T, V, c \rangle$, and $t \in {\cal TRIGGER}_{D}$, and $\phi$ be $\neg R(\overline{t})$.
Furthermore, let $\mathit{act}=\langle u', \mathtt{INSERT}, R, \\ \overline{t} \rangle$ be $t$'s actual action.
From the rule's definition, it follows that there is a $\psi$ such that $r, i-1 \attMod \psi$ and $r, i \attMod \neg \psi$.
From this, the induction's hypothesis, $s = \langle \mathit{db}', U', \mathit{sec}', T', V', c' \rangle$, and $\mathit{last}(r^{i-1}) =  \langle \mathit{db}, U, \mathit{sec}, T, \\ V, c \rangle$, it follows that $[\psi]^{\mathit{db}} = \top$ and $[\neg \psi]^{\mathit{db}'} = \top$.
Therefore, $[\psi]^{\mathit{db}} = \top$ and $[\psi]^{\mathit{db}'} = \bot$.
Hence, $\mathit{db} \neq \mathit{db}'$.
We now prove that $[\phi]^{\mathit{db}} = \top$.
Assume, for contradiction's sake, that $[\phi]^{\mathit{db}} = \bot$.
Therefore, $\overline{t} \in \mathit{db}(R)$.
From this and $\mathit{act}=\langle u', \mathtt{INSERT}, R, \overline{t} \rangle$, it follows that $\mathit{db}'= \mathit{db}[R \oplus \overline{t}]$.
From this and $\oplus$'s definition, it follows that $\mathit{db}'(R') = \mathit{db}(R')$ for all $R' \neq R$ and $\mathit{db}'(R) =\mathit{db}(R) \cup \{\overline{t}\}$.
From $\mathit{db}'(R) =\mathit{db}(R) \cup \{\overline{t}\}$ and $\overline{t} \in \mathit{db}(R)$, it follows that $\mathit{db}'(R) = \mathit{db}(R)$.
From this and  $\mathit{db}'(R') = \mathit{db}(R')$ for all $R' \neq R$, it follows that $\mathit{db}' = \mathit{db}$, which contradicts $\mathit{db} \neq \mathit{db}'$.

\item \emph{Learn \texttt{DELETE} Forward}.
The proof of this case is similar to that of \emph{Learn \texttt{INSERT} Forward}.

\item \emph{Learn \texttt{DELETE} - ID}.
The proof of this case is similar to that of \emph{Learn \texttt{INSERT} - FD}.
See also the proof of \emph{\texttt{DELETE} Success - ID}.

\item \emph{Learn \texttt{DELETE} - ID - 1}.
The proof of this case is similar to that of \emph{Learn \texttt{INSERT} - FD - 1}.
See also the proof of \emph{\texttt{DELETE} Success - ID}.

\item \emph{Learn \texttt{DELETE} Backward - 1}.
The proof of this case is similar to that of \emph{Learn \texttt{INSERT} Backward - 1}.

\item \emph{Learn \texttt{DELETE} Backward - 2}.
The proof of this case is similar to that of \emph{Learn \texttt{INSERT} Backward - 2}.

\item \emph{Propagate Forward Trigger Action}.
The proof of this case is similar to  \emph{Propagate Forward \texttt{INSERT/DELETE} Success}.

\item \emph{Propagate Backward Trigger Action}.
The proof of this case is similar to \emph{Propagate Backward \texttt{INSERT/DELETE} Success}.

\item \emph{Propagate Forward \texttt{INSERT} Trigger Action}.
The proof of this case is similar to that of \emph{Propagate Forward \texttt{INSERT} Success - 1}.

\item \emph{Propagate Forward \texttt{DELETE} Trigger Action}.
The proof of this case is similar to that of \emph{Propagate Forward \texttt{DELETE} Success - 1}.

\item \emph{Propagate Backward \texttt{INSERT} Trigger Action}.
The proof of this case is similar to that of \emph{Propagate Backward \texttt{INSERT} Success - 1}.

\item \emph{Propagate Backward \texttt{DELETE} Trigger Action}.
The proof of this case is similar to that of \emph{Propagate Backward \texttt{DELETE} Success - 1}.

\item \emph{Trigger FD \texttt{INSERT} Disabled Backward}.
Let $i$ be such that $r^{i} = r^{i-1} \concat t \concat s$, where $s = \langle \mathit{db}', U', \mathit{sec}', T', V', c' \rangle \in \Omega_{M}$, $t \in {\cal TRIGGER}_{D}$, and $\mathit{last}(r^{i-1}) =  \langle \mathit{db}, U, \mathit{sec}, T, V, \\ c \rangle$, and $\psi$ be $\neg \phi [\overline{x}^{|R'|} \mapsto \mathit{tpl}(\mathit{last}(r^{i-1}))]$.
Furthermore, let $\mathit{act}=\langle u', \mathtt{INSERT}, R, (\overline{v}, \overline{w}, \overline{q}) \rangle$ be $t$'s actual action.
From the rule's definition, it follows that  $r, i-1 \attMod \exists \overline{y},\overline{z}. R(\overline{v},  \overline{y}, \overline{z}) \wedge \overline{y} \neq \overline{w}$ holds.
From this, the induction hypothesis, and $\mathit{last}(r^{i-1}) =  \langle \mathit{db}, U, \mathit{sec}, T, V, c \rangle$, it follows that $[\exists \overline{y},\overline{z}. R(\overline{v}, \overline{y}, \overline{z}) \wedge \overline{y} \neq \overline{w}]^{\mathit{db}} = \top$.
Therefore, there is a tuple $(\overline{v},\overline{w}', \overline{z}') \in \mathit{db}(R)$ such that $\overline{w}' \neq \overline{w}$.
We now prove that $[\psi]^{\mathit{db}} = \top$.
Assume, for contradiction's sake, that this is not the case, namely that $[\phi [\overline{x} \mapsto \mathit{tpl}(\mathit{last}(r^{i-1}))]]^{\mathit{db}} = \top$.
There are two cases:
\begin{compactenum}
\item the trigger $t$ is enabled and the action $\mathit{act}$ is authorized.
In this case, the database $\mathit{db}[R \oplus \{(\overline{v}, \overline{w}, \overline{q}) \}] \\ \not\in \Omega_{D}^{\mathit{\Gamma}}$ because $\forall \overline{x}, \overline{y}, \overline{y}', \overline{z}, \overline{z}'.\, (R(\overline{x}, \overline{y},  \overline{z}) \wedge R(\overline{x},  \overline{y}',  \overline{z}') ) \\ \Rightarrow \overline{y} = \overline{y}' \in \Gamma$ and there is a tuple $(\overline{v},\overline{w}', \overline{z}') \in \mathit{db}(R)$ such that $\overline{w}' \neq \overline{w}$.
Therefore, the resulting state would be such that $\mathit{Ex}(s) \neq \emptyset$.
This contradicts the fact that, according to the rule's definition, $\mathit{Ex}(s) = \emptyset$.

\item  the trigger $t$ is enabled and the action $\mathit{act}$ is not authorized.
Therefore, the resulting state would be such that $\mathit{secEx}(s) = \top$.
This contradicts the fact that, according to the rule's definition, $\mathit{secEx}(s) = \bot$.
\end{compactenum}

\item \emph{Trigger ID \texttt{INSERT} Disabled Backward}.
The proof of this case is similar to that of \emph{Trigger FD \texttt{INSERT} Disabled Backward}.

\item \emph{Trigger ID \texttt{DELETE} Disabled Backward}.
The proof of this case is similar to that of \emph{Trigger FD \texttt{INSERT} Disabled Backward}.
\end{compactenum}
This completes the proof of the induction step.
	
This completes the proof of the theorem.	
\end{proof}

\begin{figure*}

\centering
\begin{tabular}{c}

$\infer[\begin{tabular}{c}\text{Propagate Forward }\\ \texttt{SELECT}\end{tabular}]
{r, i+1 \attMod \psi}
{
\hfill r, i \attMod  \psi \hfill \quad
\hfill r^{i+1} = r^{i} \concat \langle u, \mathtt{SELECT},  \phi\rangle \concat s \hfill \quad
\hfill 1 \leq i < |r| \hfill \quad
\hfill s \in \Omega_{M}\hfill }$ \\
\\

$\infer[\begin{tabular}{c}\text{Propagate Forward  }\\ \texttt{GRANT}/\texttt{REVOKE}\end{tabular}]
{r, i+1 \attMod  \psi }
{
\hfill r, i \attMod   \psi  \hfill \quad
\hfill r^{i+1} = r^{i} \concat \langle \mathit{op}, u', \mathit{pr},  u\rangle \concat s \hfill \quad
\hfill 1 \leq i < |r| \hfill \quad
\hfill \mathit{op} \in \{\oplus,\oplus^{*}, \ominus\} \hfill \quad
\hfill s \in \Omega_{M} \hfill }$ \\
\\

$\infer[\text{\begin{tabular}{c} Propagate Forward\\  \texttt{CREATE}\end{tabular}}]
{r, i+1 \attMod \psi }
{
\hfill r, i \attMod \psi \hfill  \quad
\hfill r^{i+1} = r^{i} \concat \langle u, \mathtt{CREATE}, o\rangle \concat s \hfill \quad
\hfill 1 \leq i < |r| \hfill \quad
\hfill o \in {\cal TRIGGER}_{D} \cup {\cal VIEW}_{D} \hfill \quad
\hfill s \in \Omega_{M} \hfill }$
\end{tabular}
\caption{Rules defining how the attacker propagates (forward) the knowledge}\label{table:rules:adversary:1}
\end{figure*}

\begin{figure*}

\centering
\begin{tabular}{c}

$\infer[\begin{tabular}{c} \text{Propagate Backward }\\ \texttt{SELECT}\end{tabular}]
{r, i \attMod \psi}
{
\hfill r, i+1 \attMod \psi \hfill  \quad
\hfill r^{i+1} = r^{i}  \concat \langle u, \mathtt{SELECT},  \phi\rangle \concat s \hfill  \quad
\hfill 1 \leq i < |r| \hfill \quad
\hfill s \in \Omega_{M} \hfill }$ \\
\\

$\infer[\begin{tabular}{c} \text{Propagate Backward }\\ \texttt{GRANT}/\texttt{REVOKE}\end{tabular}]
{r,i  \attMod  \psi}
{
\hfill r, i+1 \attMod  \psi \hfill  \quad
\hfill r^{i+1} = r^{i}  \concat \langle \mathit{op}, u', \mathit{pr},  u\rangle \concat s \hfill  \quad
\hfill 1 \leq i < |r| \hfill \quad
\hfill \mathit{op} \in \{\oplus,\oplus^{*}, \ominus\} \hfill \quad
\hfill s \in \Omega_{M} \hfill }$ \\
\\

$\infer[\text{\begin{tabular}{c} Propagate Backward\\ \texttt{CREATE TRIGGER}\end{tabular}}]
{r, i \attMod \psi}
{
\hfill r, i+1 \attMod \psi \hfill  \quad
\hfill r^{i+1} = r^{i} \concat \langle u, \mathtt{CREATE}, o\rangle \concat s \hfill  \quad
\hfill 1 \leq i < |r| \hfill \quad
\hfill o \in {\cal TRIGGER}_{D} \hfill \quad
\hfill s \in \Omega_{M} \hfill }$\\\\

$\infer[\text{\begin{tabular}{c} Propagate Backward\\ \texttt{CREATE VIEW}\end{tabular}}]
{r, i \attMod \psi'}
{
\hfill r, i+1 \attMod \psi  \hfill \quad
\hfill r^{i+1} = r^{i} \concat \langle u, \mathtt{CREATE}, o\rangle \concat s \hfill  \quad
\hfill 1 \leq i < |r| \hfill \quad
\hfill o \in {\cal VIEW}_{D} \hfill \quad
\hfill s \in \Omega_{M} \hfill \quad
\hfill \psi' = \mathit{replace}(\psi,o) \hfill }$
\end{tabular}

\caption{Rules defining how the attacker propagates (backward) the knowledge}\label{table:rules:adversary:2}
\end{figure*}

\begin{figure*}

\centering
\begin{tabular}{c}
$\infer[\texttt{SELECT} \text{ Success - 1}]
{r, i \attMod  \phi }
{
\hfill 1 < i \leq |r| \hfill \quad
\hfill r^{i} = r^{i-1} \concat \langle u, \mathtt{SELECT},  \phi\rangle \concat s \hfill  \quad
\hfill s = \langle \mathit{db}, U, \mathit{sec}, T, V, h, \langle \langle u, \mathtt{SELECT},  \phi\rangle, \top, \top, \emptyset \rangle, \langle \epsilon, \epsilon, \epsilon, \epsilon \rangle \rangle\hfill }$ \\
\\

$\infer[\texttt{SELECT} \text{ Success - 2}]
{r, i \attMod  \neg \phi }
{
\hfill 1 < i \leq |r| \hfill \quad
\hfill\hfill r^{i} = r^{i-1} \concat \langle u, \mathtt{SELECT},  \phi\rangle \concat s  \hfill \quad
\hfill s = \langle \mathit{db}, U, \mathit{sec}, T, V, h, \langle \langle u, \mathtt{SELECT},  \phi\rangle, \top, \bot, \emptyset \rangle, \langle \epsilon, \epsilon, \epsilon, \epsilon \rangle \rangle\hfill }$ \\
\\

$\infer[\texttt{INSERT} \text{ Success}]
{r, i \attMod   R(\overline{t})}
{
\hfill 1 < i \leq |r| \hfill \quad
\hfill r^{i} = r^{i-1} \concat \langle u, \mathtt{INSERT}, R, \overline{t}\rangle \concat s \hfill  \quad
\hfill s = \langle \mathit{db}, U, \mathit{sec}, T, V, h, \langle \langle u, \mathtt{INSERT}, R, \overline{t}\rangle, \top, \top, \emptyset \rangle, \langle rS', \overline{t}, u, \mathit{tr} \rangle \rangle\hfill }$ \\
\\

$\infer[\texttt{INSERT} \text{ Success - FD}]
{r, l \attMod  \neg \exists \overline{y},\overline{z}.\, R(\overline{v}, \overline{y}, \overline{z}) \wedge \overline{y} \neq \overline{w} }
{
\hfill 1 < i \leq |r| \hfill \quad
\hfill r^{i} = r^{i-1} \concat \langle u, \mathtt{INSERT}, R, \overline{t}\rangle \concat s \hfill  \quad
\hfill l \in \{i,i-1\} \hfill\\ 
\hfill s = \langle \mathit{db}, U, \mathit{sec}, T, V, h, \langle \langle u, \mathtt{INSERT}, R, \overline{t}\rangle, \top, \top, E \rangle, \langle rS', \overline{t}, u, \mathit{tr} \rangle \rangle\hfill \\ 
\hfill \forall \overline{x}, \overline{y}, \overline{y}', \overline{z}, \overline{z}'.\,( (R(\overline{x}, \overline{y},  \overline{z}) \wedge R(\overline{x}, \overline{y}',  \overline{z}') )\Rightarrow \overline{y} = \overline{y}' \in \Gamma \setminus E \hfill \quad
\hfill \overline{t} = (\overline{v}, \overline{w}, \overline{q})\hfill }$ \\
\\

$\infer[\texttt{INSERT} \text{ Success - ID}]
{r , l \attMod   \exists \overline{y}.\, S(\overline{v}, \overline{y}) }
{
\hfill 1 < i \leq |r| \hfill \quad
\hfill r^{i} = r^{i-1} \concat \langle u, \mathtt{INSERT}, R, \overline{t}\rangle \concat s  \hfill \quad
\hfill l \in \{i,i-1\} \hfill \\ 
\hfill s = \langle \mathit{db}, U, \mathit{sec}, T, V, h, \langle \langle u, \mathtt{INSERT}, R, \overline{t}\rangle, \top, \top, E \rangle, \langle rS', \overline{t}, u, \mathit{tr} \rangle \rangle\hfill \\
\hfill \forall \overline{x}, \overline{z}.\,( R(\overline{x}, \overline{z}) \Rightarrow \exists \overline{y}.\, S(\overline{x}, \overline{y}) ) \in \Gamma \setminus E \hfill  \quad
\hfill \overline{t} = (\overline{v}, \overline{w})\hfill }$ \\
\\

$\infer[\texttt{DELETE} \text{ Success}]
{r, i \attMod \neg R(\overline{t})}
{
\hfill 1 < i \leq |r| \hfill \quad
\hfill r^{i} = r^{i-1} \concat \langle u, \mathtt{DELETE}, R, \overline{t}\rangle \concat s \hfill  \quad
\hfill s = \langle \mathit{db}, U, \mathit{sec}, T, V, h, \langle \langle u, \mathtt{DELETE}, R, \overline{t}\rangle, \top, \top, \emptyset \rangle, \langle rS', \overline{t}, u, \mathit{tr} \rangle \rangle\hfill }$ \\
\\

$\infer[\texttt{DELETE} \text{ Success - ID}]
{r, l \attMod  \forall \overline{x}, \overline{z}.\,( S(\overline{x}, \overline{z}) \Rightarrow \overline{x} \neq \overline{v}) \vee \exists \overline{y}.\, (R(\overline{v}, \overline{y}) \wedge \overline{y} \neq \overline{w})}
{
\hfill 1 < i \leq |r| \hfill \quad
\hfill r^{i} = r^{i-1} \concat \langle u, \mathtt{DELETE}, R, \overline{t}\rangle \concat s \hfill  \quad
\hfill l \in \{i,i-1\} \hfill\\ 
\hfill s = \langle \mathit{db}, U, \mathit{sec}, T, V, h, \langle \langle u, \mathtt{DELETE}, R, \overline{t}\rangle, \top, \top, E \rangle, \langle rS', \overline{t}, u, \mathit{tr} \rangle \rangle\hfill \\
\hfill \forall \overline{x}, \overline{z}.\,( S(\overline{x}, \overline{z}) \Rightarrow \exists \overline{y}.\, R(\overline{x}, \overline{y}) ) \in \Gamma \setminus E \hfill \quad
\hfill \overline{t} = (\overline{v}, \overline{w})\hfill }$ \\
\\

$\infer[\texttt{INSERT} \text{ Exception}]
{r, l \attMod   \neg R(\overline{t}) }
{
\hfill 1 < i \leq |r| \hfill \quad
\hfill r^{i} = r^{i-1} \concat \langle u, \mathtt{INSERT}, R, \overline{t}\rangle \concat s  \hfill \quad
\hfill l \in \{i,i-1\} \hfill\\
\hfill s = \langle \mathit{db}, U, \mathit{sec}, T, V, h, \langle \langle u, \mathtt{INSERT}, R, \overline{t}\rangle, \top, \bot, E \rangle, \langle \epsilon, \epsilon, \epsilon, \epsilon \rangle \rangle \hfill \quad 
\hfill E \neq \emptyset \hfill } $ \\
\\

$\infer[\texttt{DELETE} \text{ Exception}]
{r, l \attMod   R(\overline{t}) }
{
\hfill 1 < i \leq |r| \hfill \quad
\hfill r^{i} = r^{i-1} \concat \langle u, \mathtt{DELETE}, R, \overline{t}\rangle \concat s  \hfill \quad
\hfill l \in \{i,i-1\} \hfill\\ 
\hfill s = \langle \mathit{db}, U, \mathit{sec}, T, V, h, \langle \langle u, \mathtt{DELETE}, R, \overline{t}\rangle, \top, \bot, E \rangle, \langle \epsilon, \epsilon, \epsilon, \epsilon \rangle \rangle \hfill \quad 
\hfill E \neq \emptyset \hfill } $ \\
\\

$\infer[\texttt{INSERT} \text{ FD Exception}]
{r, l \attMod  \exists \overline{y},\overline{z}.\, R(\overline{v}, \overline{y}, \overline{z}) \wedge \overline{y} \neq \overline{w} }
{
\hfill 1 < i \leq |r| \hfill \quad
\hfill r^{i} = r^{i-1} \concat \langle u, \mathtt{INSERT}, R, \overline{t}\rangle \concat s \hfill  \quad
\hfill l \in \{i,i-1\} \hfill\\ 
\hfill s = \langle \mathit{db}, U, \mathit{sec}, T, V, h, \langle \langle u, \mathtt{INSERT}, R, \overline{t}\rangle, \top, \bot, E \rangle, \langle \epsilon, \epsilon, \epsilon, \epsilon \rangle \rangle \hfill \\ 
\hfill (\forall \overline{x}, \overline{y}, \overline{y}', \overline{z}, \overline{z}'.\,( (R(\overline{x}, \overline{y},  \overline{z}) \wedge R(\overline{x}, \overline{y}',  \overline{z}') )\Rightarrow \overline{y} = \overline{y}') \in E \hfill \quad
\hfill \overline{t} = (\overline{v}, \overline{w}, \overline{q})} $ \\
\\

$\infer[\texttt{INSERT} \text{ ID Exception}]
{r, l \attMod   \forall \overline{x}, \overline{y}.\, S(\overline{x},\overline{y}) \Rightarrow \overline{x} \neq \overline{v} }
{
\hfill 1 < i \leq |r| \hfill \quad
\hfill r^{i} = r^{i-1} \concat \langle u, \mathtt{INSERT}, R, \overline{t}\rangle \concat s  \hfill \quad
\hfill l \in \{i,i-1\} \hfill\\ 
\hfill s = \langle \mathit{db}, U, \mathit{sec}, T, V, h, \langle \langle u, \mathtt{INSERT}, R, \overline{t}\rangle, \top, \bot, E \rangle, \langle \epsilon, \epsilon, \epsilon, \epsilon \rangle \rangle \hfill \\ 
\hfill \forall \overline{x}, \overline{z}.\,( R(\overline{x}, \overline{z}) \Rightarrow \exists \overline{y}.\, S(\overline{x}, \overline{y}) ) \in E \hfill \quad
\hfill \overline{t} = (\overline{v}, \overline{w})} $ \\
\\

$\infer[\texttt{DELETE} \text{ ID Exception}]
{r, l \attMod   \exists \overline{z}.\, S(\overline{v},\overline{z}) \wedge \forall \overline{y}.\,(R(\overline{v},\overline{y}) \Rightarrow \overline{y} = \overline{w} )}
{
\hfill 1 < i \leq |r| \hfill \quad
\hfill r^{i} = r^{i-1} \concat \langle u, \mathtt{DELETE}, R, \overline{t}\rangle \concat s  \hfill \quad
\hfill l \in \{i,i-1\} \hfill\\ 
\hfill s = \langle \mathit{db}, U, \mathit{sec}, T, V, h, \langle \langle u, \mathtt{DELETE}, R, \overline{t}\rangle, \top, \bot, E \rangle, \langle \epsilon, \epsilon, \epsilon, \epsilon \rangle \rangle \hfill \\ 
\hfill \forall \overline{x}, \overline{z}.\,( S(\overline{x}, \overline{z}) \Rightarrow \exists \overline{y}.\, R(\overline{x}, \overline{y}) ) \in E \hfill \quad
\hfill \overline{t} = (\overline{v}, \overline{w})} $ \\
\\

$\infer[\text{Integrity Constraint}]
{r, i \attMod   \gamma }
{
\hfill 1 \leq i \leq |r|\quad
\hfill \gamma \in \Gamma \hfill} $ \\\\

$\infer[\text{View}]
{r, i \attMod   \psi' }
{\hfill 1 \leq i \leq |r| \hfill \quad
\hfill v \in \mathit{last}(r^{i}).V \hfill \quad
\hfill r, i \attMod   \psi \hfill \quad
\hfill \psi' = \mathit{replace}(\psi,v) \hfill } $ \\\\
\end{tabular}

\caption{Rules defining how the attacker extracts knowledge from the run}\label{table:rules:adversary:3}
\end{figure*}

\begin{figure*}

\centering
\begin{tabular}{c}

$\infer[\text{Rollback Backward - 1}]
{r, i-n-1 \attMod   \phi }
{
\hfill r, i \attMod   \phi  \hfill  \quad
\hfill n+1 < i \leq |r| \hfill \quad
\hfill s_{1}, s_{2}, \ldots, s_{n} \in \Omega_{M} \hfill  \quad
\hfill t_{1}, \ldots, t_{n} \in {\cal TRIGGER}_{D}\hfill \\ 
\hfill \mathit{secEx}(s_{n}) = \top \vee \mathit{Ex}(s_{n}) \neq \emptyset \hfill  \quad
\hfill r^{i} = r^{i-n-1}  \concat \langle u, \mathit{op}, R, \overline{t}\rangle \concat s_{1} \concat t_{1} \concat s_{2} \concat \ldots \concat t_{n} \concat s_{n}  \hfill \\ 
\hfill s_{n} = \langle \mathit{db}, U, \mathit{sec}, T, V, h, \langle t_{n}, \mathit{when}, \mathit{stmt} \rangle, \langle \epsilon, \epsilon, \epsilon, \epsilon \rangle \rangle  \hfill \quad
\hfill \mathit{op} \in \{\mathtt{INSERT}, \mathtt{DELETE}\}\hfill } $ \\
\\

$\infer[\text{Rollback Backward - 2}]
{r, i-1 \attMod   \phi  }
{
\hfill r, i \attMod   \phi \hfill   \quad
\hfill 1 < i \leq |r| \hfill \quad
\hfill \mathit{secEx}(s) = \top \vee \mathit{Ex}(s) \neq \emptyset \hfill \quad
\hfill \mathit{op} \in \{\mathtt{INSERT}, \mathtt{DELETE}\}\hfill \\ 
\hfill r^{i} = r^{i-1} \concat \langle u, \mathit{op}, R, \overline{t}\rangle \concat s \hfill \quad
\hfill s = \langle \mathit{db}, U, \mathit{sec}, T, V, h, \langle \langle u, \mathtt{op}, R, \overline{t}\rangle, v, v', E \rangle, \langle \epsilon, \epsilon, \epsilon, \epsilon \rangle \rangle  \hfill \quad
} $ \\ 
\\

$\infer[\text{Rollback Forward - 1}]
{r, i \attMod   \phi  }
{
\hfill r, i-n-1 \attMod  \phi  \hfill  \quad
\hfill n+1 < i \leq |r| \hfill \quad
\hfill s_{1}, s_{2}, \ldots, s_{n} \in \Omega_{M} \hfill  \quad
\hfill t_{1}, \ldots, t_{n} \in {\cal TRIGGER}_{D}\hfill \\ 
\hfill \mathit{secEx}(s_{n}) = \top \vee \mathit{Ex}(s_{n}) \neq \emptyset \hfill  \quad
\hfill r^{i} = r^{i-n-1}  \concat \langle u, \mathit{op}, R, \overline{t}\rangle \concat s_{1} \concat t_{1} \concat s_{2} \concat \ldots \concat t_{n} \concat s_{n}  \hfill \\ 
\hfill s_{n} = \langle \mathit{db}, U, \mathit{sec}, T, V, h, \langle t_{n}, \mathit{when}, \mathit{stmt} \rangle, \langle \epsilon, \epsilon, \epsilon, \epsilon \rangle \rangle  \hfill \quad
\hfill \mathit{op} \in \{\mathtt{INSERT}, \mathtt{DELETE}\} \hfill } $ \\
\\

$\infer[\text{Rollback Forward - 2}]
{r, i \attMod  \phi  }
{
\hfill r, i-1 \attMod \phi \hfill  \quad
\hfill 1 < i \leq |r| \hfill \quad
\hfill \mathit{secEx}(s) = \top \vee \mathit{Ex}(s) \neq \emptyset \hfill \quad
\hfill \mathit{op} \in \{\mathtt{INSERT}, \mathtt{DELETE}\} \hfill \\ 
\hfill r^{i} = r^{i-1}  \concat \langle u, \mathit{op}, R, \overline{t}\rangle \concat s \hfill  \quad
\hfill s = \langle \mathit{db}, U, \mathit{sec}, T, V, h, \langle \langle u, \mathtt{op}, R, \overline{t}\rangle, v, v', E \rangle, \langle \epsilon, \epsilon, \epsilon, \epsilon \rangle \rangle \hfill  \quad
} $ 
\end{tabular}
\caption{Rules regulating how information propagates in case of rollbacks}\label{table:rules:adversary:4}
\end{figure*}

\begin{figure*}

\centering
\begin{tabular}{c}

$\infer[\text{\begin{tabular}{c}Propagate Forward\\ \texttt{INSERT}/\texttt{DELETE} Success\end{tabular}}]
{r, i \attMod  \phi  }
{
\hfill 1 < i \leq |r| \hfill \quad
\hfill r, i-1 \attMod  \phi \hfill  \quad
\hfill  r^{i} = r^{i-1} \concat \langle u, \mathit{op}, R, \overline{t}\rangle \concat s \hfill \\ 
\hfill s \in \Omega_{M} \hfill  \quad
\hfill \mathit{secEx}(s_{n}) = \bot \hfill \quad
\hfill \mathit{Ex}(s_{n}) = \emptyset \hfill  \quad
\hfill s = \langle \mathit{db}, U, \mathit{sec}, T, V, h, \mathit{actEff}, \mathit{tr} \rangle  \hfill \\ 
\hfill \mathit{reviseBelief}(r^{i-1}, \phi, r^{i}) = \top \hfill \quad
\hfill \mathit{op} \in \{\mathtt{INSERT}, \mathtt{DELETE}\}\hfill } $ \\
\\

$\infer[\text{\begin{tabular}{c}Propagate Forward\\ \texttt{INSERT} Success - 1\end{tabular}}]
{r, i \attMod  \phi }
{
\hfill 1 < i \leq |r| \hfill \quad
\hfill r, i-1 \attMod  \phi \hfill  \quad
\hfill r, i-1 \attMod  R(\overline{t}) \hfill  \quad
\hfill  r^{i} = r^{i-1} \concat \langle u, \mathtt{INSERT}, R, \overline{t}\rangle \concat s \hfill \\ 
\hfill s \in \Omega_{M}  \hfill \quad
\hfill \mathit{secEx}(s_{n}) = \bot \hfill \quad
\hfill \mathit{Ex}(s_{n}) = \emptyset  \hfill \quad
\hfill s = \langle \mathit{db}, U, \mathit{sec}, T, V, h, \mathit{actEff}, \mathit{tr} \rangle \quad \hfill 
} $ \\
\\

$\infer[\text{\begin{tabular}{c}Propagate Forward\\ \texttt{DELETE} Success - 1\end{tabular}}]
{r, i \attMod   \phi }
{
\hfill 1 < i \leq |r| \hfill \quad
\hfill r, i-1 \attMod  \phi \hfill  \quad
\hfill r, i-1 \attMod \neg R(\overline{t}) \hfill  \quad
\hfill r^{i} = r^{i-1} \concat \langle u, \mathtt{DELETE}, R, \overline{t}\rangle \concat s \hfill \\ 
\hfill s \in \Omega_{M}  \hfill \quad
\hfill \mathit{secEx}(s_{n}) = \bot \hfill \quad
\hfill \mathit{Ex}(s_{n}) = \emptyset \hfill  \quad
\hfill s = \langle \mathit{db}, U, \mathit{sec}, T, V, h, \mathit{actEff}, \mathit{tr} \rangle \hfill \quad
} $ \\
\\

$\infer[\text{\begin{tabular}{c}Propagate Backward\\ \texttt{INSERT}/\texttt{DELETE} Success\end{tabular}}]
{r, i-1 \attMod  \phi }
{
\hfill 1 < i \leq |r| \hfill \quad
\hfill r, i \attMod  \phi \hfill \quad
\hfill r^{i} = r^{i-1} \concat \langle u, \mathit{op}, R, \overline{t}\rangle \concat s \hfill \\ 
\hfill s \in \Omega_{M} \hfill  \quad
\hfill \mathit{secEx}(s_{n}) = \bot \hfill \quad
\hfill \mathit{Ex}(s_{n}) = \emptyset  \hfill \quad
\hfill s = \langle \mathit{db}, U, \mathit{sec}, T, V, h, \mathit{actEff}, \mathit{tr} \rangle  \hfill \\ 
\hfill \mathit{reviseBelief}(r^{i-1}, \phi, r^{i}) = \top \hfill \quad
\hfill \mathit{op} \in \{\mathtt{INSERT}, \mathtt{DELETE}\}} $ \\\\
 
 $\infer[\text{\begin{tabular}{c}Propagate Backward\\ \texttt{INSERT} Success - 1\end{tabular}}]
{r,i-1 \attMod   \phi }
{
\hfill 1 < i \leq |r| \hfill \quad
\hfill r, i \attMod  \phi  \hfill \quad
\hfill r, i-1 \attMod R(\overline{t}) \hfill \quad
\hfill r^{i} = r^{i-1} \concat \langle u, \mathtt{INSERT}, R, \overline{t}\rangle \concat s \hfill \\ 
\hfill s \in \Omega_{M}  \hfill \quad
\hfill \mathit{secEx}(s_{n}) = \bot \hfill \quad
\hfill \mathit{Ex}(s_{n}) = \emptyset  \hfill \quad
\hfill s = \langle \mathit{db}, U, \mathit{sec}, T, V, h, \mathit{actEff}, \mathit{tr} \rangle  \hfill  \quad
} $  \\\\
 
 $\infer[\text{\begin{tabular}{c}Propagate Backward\\ \texttt{DELETE} Success - 1\end{tabular}}]
{r, i-1 \attMod  \phi }
{
\hfill 1 < i \leq |r| \hfill \quad
\hfill r, i \attMod  \phi \hfill  \quad
\hfill r, i-1 \attMod  \neg R(\overline{t})  \hfill  \quad
\hfill r^{i} = r^{i-1} \concat \langle u, \mathtt{DELETE}, R, \overline{t}\rangle \concat s \hfill \\ 
\hfill s \in \Omega_{M} \hfill  \quad
\hfill \mathit{secEx}(s_{n}) = \bot \hfill \quad
\hfill \mathit{Ex}(s_{n}) = \emptyset  \hfill \quad
\hfill s = \langle \mathit{db}, U, \mathit{sec}, T, V, h, \mathit{actEff}, \mathit{tr} \rangle  \hfill  \quad
} $ 
\end{tabular}
\caption{Rules regulating how information propagates in case of successful \texttt{INSERT} and \texttt{DELETE}}\label{table:rules:adversary:5}
\end{figure*}

\begin{figure*}

\centering
\begin{tabular}{c}
$\infer[\text{Reasoning}]
{r, i \attMod  \gamma }
{
\hfill 1 \leq i \leq |r| \hfill \quad
\hfill \Phi \subseteq \{ \phi \,|\, r, i \attMod  \phi  \} \hfill \quad
\hfill \Phi \models_{\mathit{fin}} \gamma \hfill} $ 
\end{tabular}

\caption{Rules regulating the reasoning}\label{table:rules:adversary:6}
\end{figure*}

\begin{figure*}

\centering
\begin{tabular}{c}
$\infer[\text{Learn \texttt{INSERT} Backward - 3}]
{r, i-1 \attMod  \neg R(\overline{t})}
{
\hfill  r^{i} = r^{i-1}  \concat \langle u, \mathtt{INSERT}, R, \overline{t}\rangle \concat s \hfill  \quad
\hfill 1 < i \leq |r| \hfill \\
\hfill s = \langle \mathit{db}, U, \mathit{sec}, T,V, h, \mathit{aE}, \mathit{tr} \rangle \hfill \quad
\hfill \mathit{secEx}(s) = \bot \hfill \\ 
\hfill \mathit{Ex}(s) = \emptyset \hfill \quad
\hfill r, i-1 \attMod \psi  \hfill \quad
\hfill r, i \attMod  \neg \psi  \hfill } $
\\
\\

$\infer[\text{Learn \texttt{DELETE} Backward - 3}]
{r, i-1 \attMod  R(\overline{t}) }
{
\hfill  r^{i} = r^{i-1}  \concat \langle u, \mathtt{DELETE}, R, \overline{t}\rangle \concat s \hfill  \quad
\hfill 1 < i \leq |r| \hfill \\
\hfill s = \langle \mathit{db}, U, \mathit{sec}, T,V, h, \mathit{aE}, \mathit{tr} \rangle \hfill \quad
\hfill \mathit{secEx}(s) = \bot \hfill \\ 
\hfill \mathit{Ex}(s) = \emptyset \hfill \quad
\hfill r, i-1 \attMod  \psi \hfill \quad
\hfill r, i \attMod     \neg \psi \hfill  } $
\\
\\

\end{tabular}

\caption{Rules describing how the attacker learns facts about \texttt{INSERT} and \texttt{DELETE} commands}\label{table:rules:adversary:7}
\end{figure*}

\begin{figure*}

\centering
\begin{tabular}{c}
$\infer[\text{\begin{tabular}{c}Propagate Forward \\Disabled Trigger\end{tabular}}]
{r, i \attMod \phi }
{
\hfill r, i-1 \attMod  \phi  \hfill \quad
\hfill r^{i} = r^{i-1}  \concat t \concat s \hfill  \quad
\hfill \mathit{invoker}(\mathit{last}(r^{i-1})) = u  \hfill \\ 
\hfill s = \langle \mathit{db}, U, \mathit{sec}, T, V, h, \langle t, \mathit{when}, \mathit{stmt} \rangle, \mathit{tr} \rangle \hfill \quad
\hfill \mathit{secEx}(s) = \bot \hfill \\ 
\hfill t=  \langle \mathit{id},\mathit{ow},  \mathit{ev}, R, \psi, \mathit{act},m\rangle \hfill \quad
\hfill r, i-1 \attMod  \neg \psi[\overline{x}^{|R|} \mapsto \mathit{tpl}(\mathit{last}(r^{i-1}))] \hfill } $ \\
\\

$\infer[\text{\begin{tabular}{c}Propagate Backward \\Disabled Trigger\end{tabular}}]
{r, i-1 \attMod  \phi }
{
\hfill r, i \attMod  \phi  \hfill \quad
\hfill r^{i} = r^{i-1}  \concat t \concat s  \hfill \quad
\hfill  \mathit{invoker}(\mathit{last}(r^{i-1})) = u  \hfill \\ 
\hfill s = \langle \mathit{db}, U, \mathit{sec}, T, V, h, \langle t, \mathit{when}, \mathit{stmt} \rangle, \mathit{tr} \rangle \hfill \quad
\hfill \mathit{secEx}(s) = \bot \hfill \\ 
\hfill t=  \langle \mathit{id},\mathit{ow},  \mathit{ev}, R, \psi, \mathit{act},m\rangle \hfill \quad
\hfill r, i-1 \attMod \neg \psi [\overline{x}^{|R|} \mapsto \mathit{tpl}(\mathit{last}(r^{i-1}))] \hfill  } $ \\
\\
\end{tabular}

\caption{Rules regulating the propagation of information through disabled triggers}\label{table:rules:adversary:8}
\end{figure*}

\begin{figure*}
\centering

\begin{tabular}{c}

$\infer[\text{Learn \texttt{INSERT} Forward}]
{r, i \attMod  R(\overline{t}) }
{
\hfill r, i-1 \attMod \phi[\overline{x}^{|R'|} \mapsto \mathit{tpl}(\mathit{last}(r^{i-1}))] \hfill \quad
\hfill 1 < i \leq |r| \hfill \quad
\hfill r^{i} = r^{i-1}  \concat t \concat s \hfill  \quad
\hfill \mathit{invoker}(\mathit{last}(r^{i-1})) = u  \hfill \\ 
\hfill s = \langle \mathit{db}, U, \mathit{sec}, T, V, h, \langle t, \mathit{when}, \langle \langle u', \texttt{INSERT}, R, \overline{t}\rangle, \top, \top, \emptyset\rangle \rangle, \mathit{tr} \rangle \hfill \\ 
\hfill \mathit{secEx}(s) = \bot \hfill \quad
\hfill \mathit{Ex}(s) = \emptyset \hfill \quad
\hfill  t=  \langle \mathit{id},\mathit{ow},  \mathit{ev}, R', \phi, \mathit{act},m\rangle\hfill }$
\\
\\
$\infer[\text{Learn \texttt{INSERT} - FD}]
{r, l \attMod  \neg \exists \overline{y},\overline{z}.\, R(\overline{v}, \overline{y}, \overline{z}) \wedge \overline{y} \neq \overline{w}  }
{
\hfill r, i-1 \attMod   \phi[\overline{x}^{|R'|} \mapsto \mathit{tpl}(\mathit{last}(r^{i-1}))]  \hfill  \quad
\hfill 1 < i \leq |r| \hfill \quad
\hfill r^{i} = r^{i-1}  \concat t \concat s  \hfill \quad
\hfill \mathit{invoker}(\mathit{last}(r^{i-1})) = u  \hfill \\ 
\hfill s = \langle \mathit{db}, U, \mathit{sec}, T, V, h, \langle t, \mathit{when}, \langle \langle u', \texttt{INSERT}, R, \overline{t}\rangle, \top, \top, \emptyset\rangle \rangle, \mathit{tr} \rangle \hfill \quad
\hfill l \in \{i,i-1\} \hfill\\ 
\hfill \mathit{secEx}(s) = \bot \hfill \quad
\hfill \mathit{Ex}(s) = \emptyset \hfill \quad
\hfill  t=  \langle \mathit{id},\mathit{ow},  \mathit{ev}, R', \phi, \mathit{act},m\rangle\hfill \\ 
\hfill \forall \overline{x}, \overline{y}, \overline{y}', \overline{z}, \overline{z}'.\,( (R(\overline{x}, \overline{y},  \overline{z}) \wedge R(\overline{x}, \overline{y}',  \overline{z}') )\Rightarrow \overline{y} = \overline{y}' \in \Gamma \hfill \quad
\hfill \overline{t} = (\overline{v}, \overline{w}, \overline{q})\hfill }$
\\
\\
$\infer[\text{Learn \texttt{INSERT} - FD - 1}]
{r, i-1 \attMod  \neg \exists \overline{y},\overline{z}.\, R(\overline{v}, \overline{y}, \overline{z}) \wedge \overline{y} \neq \overline{w}  }
{
\hfill r, i-1 \attMod   \phi[\overline{x}^{|R'|} \mapsto \mathit{tpl}(\mathit{last}(r^{i-1}))]  \hfill  \quad
\hfill 1 < i \leq |r| \hfill \quad
\hfill r^{i} = r^{i-1}  \concat t \concat s  \hfill \quad
\hfill \mathit{invoker}(\mathit{last}(r^{i-1})) = u  \hfill \\ 
\hfill s = \langle \mathit{db}, U, \mathit{sec}, T, V, h, \langle t, \mathit{when}, \langle \langle u', \texttt{INSERT}, R, \overline{t}\rangle, \top, \top, E \rangle \rangle, \mathit{tr} \rangle \hfill \quad
\hfill \overline{t} = (\overline{v}, \overline{w}, \overline{q})\hfill \quad 
\hfill \mathit{secEx}(s) = \bot \hfill \\
\hfill  t=  \langle \mathit{id},\mathit{ow},  \mathit{ev}, R', \phi, \mathit{act},m\rangle\hfill \quad 
\hfill \forall \overline{x}, \overline{y}, \overline{y}', \overline{z}, \overline{z}'.\,( (R(\overline{x}, \overline{y},  \overline{z}) \wedge R(\overline{x}, \overline{y}',  \overline{z}') )\Rightarrow \overline{y} = \overline{y}' \in \Gamma \setminus E \hfill }$
\\
\\
$\infer[\text{Learn \texttt{INSERT} - ID}]
{r, l \attMod  \exists \overline{y}.\, S(\overline{v}, \overline{y})  }
{
\hfill r, i-1 \attMod \phi[\overline{x}^{|R'|} \mapsto \mathit{tpl}(\mathit{last}(r^{i-1}))]  \hfill  \quad
\hfill 1 < i \leq |r| \hfill \quad
\hfill r^{i} = r^{i-1}  \concat t \concat s  \hfill \quad
\hfill \mathit{invoker}(\mathit{last}(r^{i-1})) = u  \hfill \\ 
\hfill s = \langle \mathit{db}, U, \mathit{sec}, T, V, h, \langle t, \mathit{when}, \langle \langle u', \texttt{INSERT}, R, \overline{t}\rangle, \top, \top, \emptyset\rangle \rangle, \mathit{tr} \rangle \hfill \quad
\hfill l \in \{i,i-1\} \hfill \\ 
\hfill \mathit{secEx}(s) = \bot \hfill  \quad
\hfill \mathit{Ex}(s) = \emptyset \hfill \quad
\hfill  t=  \langle \mathit{id},\mathit{ow},  \mathit{ev}, R', \phi, \mathit{act},m\rangle\hfill \\ 
\hfill (\forall \overline{x}, \overline{z}.\,( R(\overline{x}, \overline{z}) \Rightarrow \exists \overline{w}.\, S(\overline{x}, \overline{w}) ) \in \Gamma \hfill \quad
\hfill \overline{t} = (\overline{v}, \overline{w}) \hfill }$ \\
\\

$\infer[\text{Learn \texttt{INSERT} - ID - 1}]
{r, i-1 \attMod  \exists \overline{y}.\, S(\overline{v}, \overline{y})  }
{
\hfill r, i-1 \attMod \phi[\overline{x}^{|R'|} \mapsto \mathit{tpl}(\mathit{last}(r^{i-1}))]  \hfill  \quad
\hfill 1 < i \leq |r| \hfill \quad
\hfill r^{i} = r^{i-1}  \concat t \concat s  \hfill \quad
\hfill \mathit{invoker}(\mathit{last}(r^{i-1})) = u  \hfill \\ 
\hfill s = \langle \mathit{db}, U, \mathit{sec}, T, V, h, \langle t, \mathit{when}, \langle \langle u', \texttt{INSERT}, R, \overline{t}\rangle, \top, \top,  E \rangle \rangle, \mathit{tr} \rangle \hfill \quad
\hfill \overline{t} = (\overline{v}, \overline{w}) \hfill\\
\hfill \mathit{secEx}(s) = \bot \hfill  
\hfill  t=  \langle \mathit{id},\mathit{ow},  \mathit{ev}, R', \phi, \mathit{act},m\rangle\hfill \hfill
\hfill (\forall \overline{x}, \overline{z}.\,( R(\overline{x}, \overline{z}) \Rightarrow \exists \overline{w}.\, S(\overline{x}, \overline{w}) ) \in \Gamma \setminus E \hfill  }$ \\
\\

$\infer[\text{Learn \texttt{INSERT} Backward - 1}]
{r, i-1 \attMod  \phi[\overline{x}^{|R'|} \mapsto \mathit{tpl}(\mathit{last}(r^{i-1}))]  }
{
\hfill 1 < i \leq |r| \hfill \quad
\hfill  r^{i} = r^{i-1}  \concat t \concat s  \hfill \quad
\hfill \mathit{invoker}(\mathit{last}(r^{i-1})) = u \hfill \\ 
\hfill s = \langle \mathit{db}, U, \mathit{sec}, T, V, h, \langle t, \mathit{when}, \langle \langle u', \texttt{INSERT}, R, \overline{t}\rangle, \top, \top, \emptyset\rangle \rangle, \mathit{tr} \rangle \hfill \\ 
\hfill \mathit{secEx}(s) = \bot \hfill \quad
\hfill \mathit{Ex}(s) = \emptyset \hfill \quad
\hfill  t=  \langle \mathit{id},\mathit{ow},  \mathit{ev}, R', \phi, \mathit{act},m\rangle\hfill \\ 
\hfill r, i-1 \attMod  \psi \hfill  \quad
\hfill r, i \attMod \neg \psi \hfill }$
\\
\\

$\infer[\text{Learn \texttt{INSERT} Backward - 2}]
{r, i-1 \attMod   \neg R(\overline{t})  }
{
\hfill 1 < i \leq |r| \hfill \quad
\hfill r^{i} = r^{i-1}  \concat t \concat s  \hfill \quad
\hfill \mathit{invoker}(\mathit{last}(r^{i-1})) = u  \hfill \\ 
\hfill s = \langle \mathit{db}, U, \mathit{sec}, T, V, h, \langle t, \mathit{when}, \langle \langle u', \texttt{INSERT}, R, \overline{t}\rangle, \top, \top, \emptyset\rangle \rangle, \mathit{tr} \rangle \hfill \\ 
\hfill \mathit{secEx}(s) = \bot \hfill \quad
\hfill \mathit{Ex}(s) = \emptyset \hfill \quad
\hfill  t=  \langle \mathit{id},\mathit{ow},  \mathit{ev}, R', \phi, \mathit{act},m\rangle\hfill \\ 
\hfill r, i-1 \attMod   \psi  \hfill  \quad
\hfill r, i \attMod   \neg \psi  \hfill }$
\\
\\

\end{tabular}

\caption{Extracting knowledge from triggers}\label{table:rules:adversary:9}
\end{figure*}

\begin{figure*}
\centering

\begin{tabular}{c}

$\infer[\text{Learn \texttt{DELETE} Forward}]
{r, i \attMod    \neg R(\overline{t})  }
{
\hfill r, i-1 \attMod  \phi[\overline{x}^{|R'|} \mapsto \mathit{tpl}(\mathit{last}(r^{i-1}))]  \hfill  \quad
\hfill 1 < i \leq |r| \hfill \quad
\hfill r^{i} = r^{i-1}  \concat t \concat s \hfill  \quad
\hfill \mathit{invoker}(\mathit{last}(r^{i-1}))= u  \hfill \\ 
\hfill s = \langle \mathit{db}, U, \mathit{sec}, T, V, h, \langle t, \mathit{when}, \langle \langle u', \texttt{DELETE}, R, \overline{t}\rangle, \top, \top, \emptyset\rangle \rangle, \mathit{tr} \rangle \hfill \\ 
\hfill \mathit{secEx}(s) = \bot \hfill \quad
\hfill \mathit{Ex}(s) = \emptyset \hfill \quad
\hfill  t=  \langle \mathit{id},\mathit{ow},  \mathit{ev}, R', \phi, \mathit{act},m\rangle\hfill }$
\\
\\

$\infer[\text{Learn \texttt{DELETE} - ID}]
{r, l \attMod    \forall \overline{x}, \overline{z}.\,( S(\overline{x}, \overline{z}) \Rightarrow \overline{x} \neq \overline{v}) \vee \exists \overline{y}.\, (R(\overline{v}, \overline{y}) \wedge \overline{y} \neq \overline{w})  }
{
\hfill r, i-1 \attMod  \phi[\overline{x}^{|R'|} \mapsto \mathit{tpl}(\mathit{last}(r^{i-1}))] \hfill   \quad
\hfill 1 < i \leq |r| \hfill \quad
\hfill r^{i} = r^{i-1}  \concat t \concat s \hfill  \quad
\hfill \mathit{invoker}(\mathit{last}(r^{i-1}))= u  \hfill \\ 
\hfill s = \langle \mathit{db}, U, \mathit{sec}, T, V, h, \langle t, \mathit{when}, \langle \langle u', \texttt{DELETE}, R, \overline{t}\rangle, \top, \top, \emptyset\rangle \rangle, \mathit{tr} \rangle \hfill \quad 
\hfill l \in \{i,i-1\} \hfill \\ 
\hfill \mathit{secEx}(s) = \bot \hfill \quad
\hfill \mathit{Ex}(s) = \emptyset \hfill \quad
\hfill  t=  \langle \mathit{id},\mathit{ow},  \mathit{ev}, R', \phi, \mathit{act},m\rangle\hfill \\ 
\hfill (\forall \overline{x}, \overline{z}.\,( S(\overline{x}, \overline{z}) \Rightarrow \exists \overline{w}.\, R(\overline{x}, \overline{w}) ) \in \Gamma \hfill \quad
\hfill \overline{t} = (\overline{v}, \overline{w})\hfill }$ \\
\\

$\infer[\text{Learn \texttt{DELETE} - ID - 1}]
{r, i-1 \attMod    \forall \overline{x}, \overline{z}.\,( S(\overline{x}, \overline{z}) \Rightarrow \overline{x} \neq \overline{v}) \vee \exists \overline{y}.\, (R(\overline{v}, \overline{y}) \wedge \overline{y} \neq \overline{w})  }
{
\hfill r, i-1 \attMod  \phi[\overline{x}^{|R'|} \mapsto \mathit{tpl}(\mathit{last}(r^{i-1}))] \hfill   \quad
\hfill 1 < i \leq |r| \hfill \quad
\hfill r^{i} = r^{i-1}  \concat t \concat s \hfill  \quad
\hfill \mathit{invoker}(\mathit{last}(r^{i-1}))= u  \hfill \\ 
\hfill s = \langle \mathit{db}, U, \mathit{sec}, T, V, h, \langle t, \mathit{when}, \langle \langle u', \texttt{DELETE}, R, \overline{t}\rangle, \top, \top, E \rangle \rangle, \mathit{tr} \rangle \hfill \quad
\hfill \overline{t} = (\overline{v}, \overline{w})\hfill \\ 
\hfill \mathit{secEx}(s) = \bot \hfill \quad
\hfill  t=  \langle \mathit{id},\mathit{ow},  \mathit{ev}, R', \phi, \mathit{act},m\rangle\hfill \quad 
\hfill (\forall \overline{x}, \overline{z}.\,( S(\overline{x}, \overline{z}) \Rightarrow \exists \overline{w}.\, R(\overline{x}, \overline{w}) ) \in \Gamma \setminus E \hfill }$ \\
\\

$\infer[\text{Learn \texttt{DELETE} Backward - 1}]
{r, i-1 \attMod    \phi[\overline{x}^{|R'|} \mapsto \mathit{tpl}(\mathit{last}(r^{i-1}))]  }
{ 
\hfill 1 < i \leq |r| \hfill \quad
\hfill r^{i} = r^{i-1}  \concat t \concat s  \hfill \quad
\hfill \mathit{invoker}(\mathit{last}(r^{i-1}))= u  \hfill \\ 
\hfill s = \langle \mathit{db}, U, \mathit{sec}, T, V, h, \langle t, \mathit{when}, \langle \langle u', \texttt{DELETE}, R, \overline{t}\rangle, \top, \top, \emptyset\rangle \rangle, \mathit{tr} \rangle \hfill \\ 
\hfill \mathit{secEx}(s) = \bot \hfill \quad
\hfill \mathit{Ex}(s) = \emptyset \hfill \quad
\hfill  t=  \langle \mathit{id},\mathit{ow},  \mathit{ev}, R', \phi, \mathit{act},m\rangle\hfill \\ 
\hfill r, i-1 \attMod   \psi   \hfill \quad
\hfill r, i \attMod  \neg \psi \hfill }$
\\
\\

$\infer[\text{Learn \texttt{DELETE} Backward - 2}]
{r, i-1 \attMod  R(\overline{t})  }
{
\hfill 1 < i \leq |r| \hfill \quad
\hfill r^{i} = r^{i-1}  \concat t \concat s \hfill  \quad
\hfill \mathit{invoker}(\mathit{last}(r^{i-1}))= u  \hfill \\ 
\hfill s = \langle \mathit{db}, U, \mathit{sec}, T, V, h, \langle t, \mathit{when}, \langle \langle u', \texttt{DELETE}, R, \overline{t}\rangle, \top, \top, \emptyset\rangle \rangle, \mathit{tr} \rangle \hfill \\
\hfill \mathit{secEx}(s) = \bot \hfill \quad
\hfill \mathit{Ex}(s) = \emptyset \hfill \quad
\hfill  t=  \langle \mathit{id},\mathit{ow},  \mathit{ev}, R', \phi, \mathit{act},m\rangle\hfill \\ 
\hfill r, i-1 \attMod \psi  \hfill \quad
\hfill r, 1 \attMod  \neg \psi \hfill }$
\\
\\

$\infer[\text{Learn \texttt{GRANT}/\texttt{REVOKE} Backward}]
{r, i-1 \attMod  \phi[\overline{x}^{|R'|} \mapsto \mathit{tpl}(\mathit{last}(r^{i-1}))]  }
{
\hfill 1 < i \leq |r| \hfill \quad
\hfill  r^{i} = r^{i-1}  \concat t \concat s  \hfill \quad
\hfill \mathit{invoker}(\mathit{last}(r^{i-1})) = u  \hfill \\ 
\hfill s = \langle \mathit{db}, U, \mathit{sec}, T, V, h, \langle t, \mathit{when}, \langle \langle \mathit{op},u'', \mathit{pr},u'\rangle, \top, \top, \emptyset\rangle \rangle, \mathit{tr} \rangle \hfill \\ 
\hfill \mathit{secEx}(s) = \bot \hfill \quad
\hfill \mathit{Ex}(s) = \emptyset \hfill \quad
\hfill  t=  \langle \mathit{id},\mathit{ow},  \mathit{ev}, R', \phi, \mathit{act},m\rangle\hfill \\ 
\hfill u',u'' \in U \hfill \quad
\hfill \mathit{op} \in \{\oplus,\oplus^{*}, \ominus\} \hfill \quad 
\hfill \mathit{last}(r^{i-1}).\mathit{sec} \neq \mathit{last}(r^{i}).\mathit{sec} \hfill }$

\end{tabular}

\caption{Extracting knowledge from triggers}\label{table:rules:adversary:9}
\end{figure*}

\begin{figure*}

\centering
\begin{tabular}{c}
$\infer[\text{\begin{tabular}{c}Propagate Forward \\ Trigger Action\end{tabular}}]
{r, i \attMod \psi }
{
\hfill r, i-1 \attMod  \psi \hfill \quad
\hfill 1 < i \leq |r| \hfill \quad
\hfill r^{i} = r^{i-1}  \concat t \concat s \hfill  \quad
\hfill \mathit{invoker}(\mathit{last}(r^{i-1})) = u   \hfill \\ 
\hfill s = \langle \mathit{db}, U, \mathit{sec}, T, V, h, \langle t, \mathit{when}, \mathit{stmt} \rangle , \mathit{tr} \rangle \hfill  \quad
\hfill \mathit{Ex}(s) = \emptyset  \hfill \\ 
\hfill \mathit{secEx}(s) =\bot \hfill \quad
\hfill t=  \langle \mathit{id},\mathit{ow},  \mathit{ev}, R, \phi, \mathit{act},m\rangle\hfill \quad 
\hfill {\mathit{reviseBelief}(r^{i-1}, \psi,r^{i}) = \top}\hfill }$ \\
\\

$\infer[\text{\begin{tabular}{c}Propagate Backward \\ Trigger Action\end{tabular}}]
{r, i-1 \attMod   \psi }
{
\hfill r, i \attMod  \psi  \hfill \quad
\hfill 1 < i \leq |r| \hfill \quad
\hfill r^{i} = r^{i-1}  \concat t \concat s \hfill  \quad
\hfill \mathit{invoker}(\mathit{last}(r^{i-1})) = u   \hfill \\ 
\hfill s = \langle \mathit{db}, U, \mathit{sec}, T, V, h, \langle t, \mathit{when}, \mathit{stmt} \rangle , \mathit{tr} \rangle \hfill \quad
\hfill \mathit{Ex}(s) = \emptyset  \hfill \\ 
\hfill \mathit{secEx}(s) =\bot \hfill \quad
\hfill t=  \langle \mathit{id},\mathit{ow},  \mathit{ev}, R, \phi, \mathit{act},m\rangle\hfill \quad 
\hfill {\mathit{reviseBelief}(r^{i-1}, \psi,r^{i}) = \top}\hfill }$ \\
\\

$\infer[\text{\begin{tabular}{c}Propagate Forward \\ \texttt{INSERT} Trigger Action\end{tabular}}]
{r, i \attMod   \psi }
{
\hfill r, i-1 \attMod  \psi \hfill  \quad
\hfill r, i-1 \attMod   R(\overline{t})  \hfill \quad
\hfill 1 < i \leq |r| \hfill \quad
\hfill r^{i} = r^{i-1}  \concat t \concat s \hfill  \quad
\hfill \mathit{invoker}(\mathit{last}(r^{i-1})) = u  \hfill \\ 
\hfill s = \langle \mathit{db}, U, \mathit{sec}, T, V, h, \langle t, \mathit{when}, \langle \langle u', \mathtt{INSERT}, R, \overline{t}\rangle, \top, \top,\emptyset\rangle \rangle , \mathit{tr} \rangle \hfill \quad
\hfill \mathit{Ex}(s) = \emptyset  \hfill \\ 
\hfill \mathit{secEx}(s) =\bot \hfill \quad
\hfill t=  \langle \mathit{id},\mathit{ow},  \mathit{ev}, R', \phi, \mathit{act},m\rangle\hfill }$ \\
\\

$\infer[\text{\begin{tabular}{c}Propagate Forward \\ \texttt{DELETE} Trigger Action\end{tabular}}]
{r, i \attMod   \psi }
{
\hfill r, i-1 \attMod    \psi  \hfill \quad
\hfill r, i-1 \attMod   \neg R(\overline{t})  \hfill \quad
\hfill 1 < i \leq |r| \hfill \quad
\hfill r^{i} = r^{i-1}  \concat t \concat s  \hfill \quad
\hfill \mathit{invoker}(\mathit{last}(r^{i-1})) = u   \hfill \\ 
\hfill s = \langle \mathit{db}, U, \mathit{sec}, T, V, h, \langle t, \mathit{when}, \langle \langle u', \mathtt{DELETE}, R, \overline{t}\rangle, \top, \top,\emptyset\rangle \rangle , \mathit{tr} \rangle \hfill  \quad
\hfill \mathit{Ex}(s) = \emptyset  \hfill \\ 
\hfill \mathit{secEx}(s) =\bot \hfill \quad
\hfill t=  \langle \mathit{id},\mathit{ow},  \mathit{ev}, R', \phi, \mathit{act},m\rangle\hfill }$ \\
\\

$\infer[\text{\begin{tabular}{c}Propagate Backward \\ \texttt{INSERT} Trigger Action\end{tabular}}]
{r, i-1 \attMod     \psi }
{
\hfill r, i \attMod    \psi  \hfill \quad
\hfill r, i-1 \attMod    R(\overline{t})  \hfill \quad
\hfill 1 < i \leq |r| \hfill \quad
\hfill  r^{i} = r^{i-1}  \concat t \concat s  \hfill \quad
\hfill \mathit{invoker}(\mathit{last}(r^{i-1})) = u   \hfill \\ 
\hfill s = \langle \mathit{db}, U, \mathit{sec}, T, V, h, \langle t, \mathit{when}, \langle \langle u', \mathtt{INSERT}, R, \overline{t}\rangle, \top, \top,\emptyset\rangle \rangle , \mathit{tr} \rangle \hfill \quad
\hfill \mathit{Ex}(s) = \emptyset  \hfill \\ 
\hfill \mathit{secEx}(s) =\bot \hfill \quad
\hfill t=  \langle \mathit{id},\mathit{ow},  \mathit{ev}, R', \phi, \mathit{act},m\rangle\hfill }$ \\
\\

$\infer[\text{\begin{tabular}{c}Propagate Backward  \\  \texttt{DELETE} Trigger Action\end{tabular}}]
{r, i-1 \attMod  \psi }
{
\hfill r, i \attMod  \psi \hfill \quad
\hfill r, i-1 \attMod  \neg R(\overline{t})  \hfill \quad
\hfill 1 < i \leq |r| \hfill \quad
\hfill r^{i} = r^{i-1}  \concat t \concat s  \hfill \quad
\hfill \mathit{invoker}(\mathit{last}(r^{i-1})) = u  \hfill \\ 
\hfill s = \langle \mathit{db}, U, \mathit{sec}, T, V, h, \langle t, \mathit{when}, \langle \langle u', \mathtt{DELETE}, R, \overline{t}\rangle, \top, \top,\emptyset\rangle \rangle , \mathit{tr} \rangle \hfill \quad
\hfill \mathit{Ex}(s) = \emptyset  \hfill \\ 
\hfill \mathit{secEx}(s) =\bot \hfill \quad
\hfill t=  \langle \mathit{id},\mathit{ow},  \mathit{ev}, R', \phi, \mathit{act},m\rangle\hfill }$ 

\end{tabular}

\caption{Rules for propagating knowledge through triggers}\label{table:rules:adversary:10}
\end{figure*}

\begin{figure*}
\centering
\begin{tabular}{c}

$\infer[\text{\begin{tabular}{c} Trigger \\ FD \\\texttt{INSERT}\\ Disabled \\Backward\end{tabular}}]
{r, i-1 \attMod \neg \phi [\overline{x}^{|R'|} \mapsto \mathit{tpl}(\mathit{last}(r^{i-1}))] }
{
\hfill 1 < i \leq |r| \hfill \quad
\hfill r^{i} = r^{i-1}  \concat t \concat s  \hfill \quad
\hfill \mathit{invoker}(\mathit{last}(r^{i-1})) = u \hfill  \\
\hfill s = \langle \mathit{db}, U, \mathit{sec}, T, V, h, \langle t, \mathit{when}, \mathit{stmt} \rangle, \mathit{tr}  \rangle   \hfill \\ 
\hfill \mathit{secEx}(s) = \bot \hfill \quad
\hfill \mathit{Ex}(s) = \emptyset \hfill \quad
\hfill t=  \langle \mathit{id},\mathit{ow},  \mathit{ev}, R, \phi, \mathit{act},m\rangle\hfill \\
\hfill \mathit{getAction}(\mathit{act},\mathit{user}(\mathit{last}(r^{i-1}),t), \mathit{tpl}(\mathit{last}(r^{i-1})) = \langle u', \mathtt{INSERT}, R, (\overline{v}, \overline{w}, \overline{q}) \rangle \hfill \\
\hfill r, i-1 \attMod \exists \overline{y},\overline{z}. R(\overline{v}, \overline{y}, \overline{z}) \wedge \overline{y} \neq \overline{w} \hfill \\
\hfill \forall \overline{x}, \overline{y}, \overline{y}', \overline{z}, \overline{z}'.\, (R(\overline{x}, \overline{y},  \overline{z}) \wedge R(\overline{x}, \overline{y}',  \overline{z}') )\Rightarrow \overline{y} = \overline{y}' \in \Gamma \hfill}$ \\
\\

$\infer[\text{\begin{tabular}{c} Trigger \\ ID \\\texttt{INSERT}\\  Disabled\\ Backward\end{tabular}}]
{r, i-1 \attMod \neg \phi [\overline{x}^{|R'|} \mapsto \mathit{tpl}(\mathit{last}(r^{i-1}))] }
{
\hfill 1 < i \leq |r| \hfill \quad
\hfill r^{i} = r^{i-1}  \concat t \concat s  \hfill \quad
\hfill \mathit{invoker}(\mathit{last}(r^{i-1}))= u  \hfill \\
\hfill s = \langle \mathit{db}, U, \mathit{sec}, T, V, h, \langle t, \mathit{when}, \mathit{stmt} \rangle, \mathit{tr}  \rangle   \hfill \\ 
\hfill \mathit{secEx}(s) = \bot \hfill \quad
\hfill \mathit{Ex}(s) = \emptyset \hfill \quad
\hfill t=  \langle \mathit{id},\mathit{ow},  \mathit{ev}, R, \phi, \mathit{act},m\rangle\hfill \\
\hfill \mathit{getAction}(\mathit{act},\mathit{user}(\mathit{last}(r^{i-1}),t), \mathit{tpl}(\mathit{last}(r^{i-1})) = \langle u', \mathtt{INSERT}, R, (\overline{v}, \overline{w}) \rangle \hfill \\
\hfill r, i-1 \attMod \forall \overline{x}, \overline{y}.\, S(\overline{x},\overline{y}) \Rightarrow \overline{x} \neq \overline{v} \hfill \quad
\hfill \forall \overline{x}, \overline{z}.\, R(\overline{x}, \overline{z}) \Rightarrow \exists \overline{w}.\, S(\overline{x}, \overline{w}) \in \Gamma \hfill}$ \\
\\

$\infer[\text{\begin{tabular}{c} Trigger \\ ID \\\texttt{DELETE} \\ Disabled\\ Backward\end{tabular}}]
{r, i-1 \attMod \neg \phi [\overline{x}^{|R'|} \mapsto \mathit{tpl}(\mathit{last}(r^{i-1}))] }
{
\hfill 1 < i \leq |r| \hfill \quad
\hfill r^{i} = r^{i-1}  \concat t \concat s  \hfill \quad
\hfill \mathit{invoker}(\mathit{last}(r^{i-1}))= u \hfill \\
\hfill s = \langle \mathit{db}, U, \mathit{sec}, T, V, h, \langle t, \mathit{when}, \mathit{stmt} \rangle, \mathit{tr}  \rangle   \hfill \\ 
\hfill \mathit{secEx}(s) = \bot \hfill \quad
\hfill \mathit{Ex}(s) = \emptyset \hfill \quad
\hfill t=  \langle \mathit{id},\mathit{ow},  \mathit{ev}, R, \phi, \mathit{act},m\rangle\hfill \\
\hfill \mathit{getAction}(\mathit{act},\mathit{user}(\mathit{last}(r^{i-1}),t), \mathit{tpl}(\mathit{last}(r^{i-1})) = \langle u', \mathtt{DELETE}, R, (\overline{v}, \overline{w}) \rangle \hfill \\
\hfill r, i-1 \attMod \exists \overline{z}.\, S(\overline{v},\overline{z}) \wedge \forall \overline{y}.\,(R(\overline{x},\overline{y}) \Rightarrow \overline{y} = \overline{w} ) \hfill \quad
\hfill \forall \overline{x}, \overline{z}.\, S(\overline{x}, \overline{z}) \Rightarrow \exists \overline{w}.\, R(\overline{x}, \overline{w}) \in \Gamma \hfill}$ \\
\\

$\infer[\text{\begin{tabular}{c}Trigger \\ \texttt{GRANT} \\ Disabled \\ Backward\end{tabular}}]
{r, i-1 \attMod  \neg \phi[\overline{x}^{|R'|} \mapsto \mathit{tpl}(\mathit{last}(r^{i-1}))]  }
{ 
\hfill 1 < i \leq |r| \hfill \quad
\hfill r^{i} = r^{i-1}  \concat t \concat s  \hfill \quad
\hfill \mathit{invoker}(\mathit{last}(r^{i-1}))= u  \hfill \\ 
\hfill s = \langle \mathit{db}, U, \mathit{sec}, T, V, h, \langle t, \mathit{when}, \mathit{stmt} \rangle, \mathit{tr}  \rangle   \hfill \\ 
\hfill \mathit{secEx}(s) = \bot \hfill \quad
\hfill \mathit{Ex}(s) = \emptyset \hfill \quad
\hfill  t=  \langle \mathit{id},\mathit{ow},  \mathit{ev}, R', \phi, \mathit{act},m\rangle\hfill \\ 
\hfill \mathit{getAction}(\mathit{act},\mathit{user}(\mathit{last}(r^{i-1}),t), \mathit{tpl}(\mathit{last}(r^{i-1})) = \langle \mathit{op}, u'', p, u' \rangle  \hfill \\
\hfill u',u'' \in U \hfill \quad
\hfill \mathit{op} \in \{\oplus,\oplus^{*}\} \hfill \quad 
\hfill \langle \mathit{op}, u'', p, u' \rangle  \not\in \mathit{last}(r^{i-1}).\mathit{sec}  \hfill \quad
\hfill \mathit{last}(r^{i-1}).\mathit{sec} = \mathit{sec} \hfill }$\\\\

$\infer[\text{\begin{tabular}{c}Trigger \\ \texttt{REVOKE} \\ Disabled \\ Backward\end{tabular}}]
{r, i-1 \attMod  \neg \phi[\overline{x}^{|R'|} \mapsto \mathit{tpl}(\mathit{last}(r^{i-1}))]  }
{ 
\hfill 1 < i \leq |r| \hfill \quad
\hfill r^{i} = r^{i-1}  \concat t \concat s \hfill  \quad
\hfill \mathit{invoker}(\mathit{last}(r^{i-1}))= u  \hfill \\ 
\hfill s = \langle \mathit{db}, U, \mathit{sec}, T, V, h, \langle t, \mathit{when}, \mathit{stmt} \rangle, \mathit{tr}  \rangle   \hfill \\ 
\hfill \mathit{secEx}(s) = \bot \hfill \quad
\hfill \mathit{Ex}(s) = \emptyset \hfill \quad
\hfill  t=  \langle \mathit{id},\mathit{ow},  \mathit{ev}, R', \phi, \mathit{act},m\rangle\hfill \\ 
\hfill \mathit{getAction}(\mathit{act},\mathit{user}(\mathit{last}(r^{i-1}),t), \mathit{tpl}(\mathit{last}(r^{i-1})) = \langle \ominus, u'', p, u' \rangle  \hfill \\
\hfill u',u'' \in U \hfill \quad
\hfill \mathit{op} \in \{\oplus,\oplus^{*}\} \hfill \quad 
\hfill \langle \mathit{op}, u'', p, u' \rangle  \in \mathit{last}(r^{i-1}).\mathit{sec}  \hfill \quad
\hfill \mathit{last}(r^{i-1}).\mathit{sec} = \mathit{sec} \hfill }$\\\\
\end{tabular}

\caption{Extracting knowledge from triggers}\label{table:rules:adversary:12}
\end{figure*}

\begin{figure*}

\centering
\begin{tabular}{c}

$\infer[\text{\begin{tabular}{c}Trigger \\\texttt{INSERT} \\FD\\ Exception\end{tabular}}]
{r, i-1 \attMod  \exists \overline{y},\overline{z}.\, R(\overline{v}, \overline{y}, \overline{z}) \wedge \overline{y} \neq \overline{w} }
{
\hfill 1 < i \leq |r| \hfill \quad
\hfill r^{i} = r^{i-1}  \concat t \concat s \hfill  \quad
\hfill \mathit{invoker}(\mathit{last}(r^{i-1}))= u \hfill  \\ 
\hfill s = \langle \mathit{db}, U, \mathit{sec}, T, V, h, \langle t, \mathit{when}, \langle \langle u', \texttt{INSERT}, R, \overline{t}\rangle, \top, \top, E\rangle  , \mathit{tr} \rangle   \hfill \\ 
\hfill \mathit{secEx}(s) =\bot \hfill \quad
\hfill t=  \langle \mathit{id},\mathit{ow},  \mathit{ev}, R', \phi, \mathit{act},m\rangle\hfill \\ 
\hfill (\forall \overline{x}, \overline{y}, \overline{y}', \overline{z}, \overline{z}'.\,( (R(\overline{x}, \overline{y},  \overline{z}) \wedge R(\overline{x}, \overline{y}',  \overline{z}') )\Rightarrow \overline{y} = \overline{y}') \in \mathit{Ex}(s) \hfill \quad
\hfill \overline{t} = (\overline{v}, \overline{w}, \overline{q})\hfill }$ \\
\\

$\infer[\text{\begin{tabular}{c} Trigger \\\texttt{INSERT}\\ ID \\Exception\end{tabular}}]
{r, i-1 \attMod  \forall \overline{x}, \overline{y}.\, S(\overline{x},\overline{y}) \Rightarrow \overline{x} \neq \overline{v} }
{
\hfill 1 < i \leq |r| \hfill \quad
\hfill r^{i} = r^{i-1}  \concat t \concat s  \hfill \quad
\hfill \mathit{invoker}(\mathit{last}(r^{i-1}))= u \hfill  \\ 
\hfill s = \langle \mathit{db}, U, \mathit{sec}, T, V, h, \langle t, \mathit{when}, \langle \langle u', \texttt{INSERT}, R, \overline{t}\rangle, \top, \top, E\rangle  , \mathit{tr} \rangle   \hfill \\ 
\hfill \mathit{secEx}(s) =\bot \hfill \quad
\hfill t=  \langle \mathit{id},\mathit{ow},  \mathit{ev}, R', \phi, \mathit{act},m\rangle\hfill \\
 \hfill (\forall \overline{x}, \overline{z}.\,( R(\overline{x}, \overline{z}) \Rightarrow \exists \overline{w}.\, S(\overline{x}, \overline{w}) ) \in \mathit{Ex}(s) \hfill \quad
\hfill \overline{t} = (\overline{v}, \overline{w})\hfill }$ \\
\\

$\infer[\text{\begin{tabular}{c} Trigger \\\texttt{DELETE}\\ ID \\Exception\end{tabular}}]
{r, i-1 \attMod    \exists \overline{z}.\, S(\overline{v},\overline{z}) \wedge \forall \overline{y}.\,(R(\overline{v},\overline{y}) \Rightarrow \overline{y} = \overline{w} ) }
{
\hfill 1 < i \leq |r| \hfill \quad
\hfill r^{i} = r^{i-1}  \concat t \concat s  \hfill \quad
\hfill \mathit{invoker}(\mathit{last}(r^{i-1}))= u  \hfill \\ 
\hfill s = \langle \mathit{db}, U, \mathit{sec}, T, V, h, \langle t, \mathit{when}, \langle \langle u', \texttt{DELETE}, R, \overline{t}\rangle, \top, \top, E\rangle  , \mathit{tr} \rangle   \hfill \\
\hfill \mathit{secEx}(s) =\bot \hfill \quad
\hfill t=  \langle \mathit{id},\mathit{ow},  \mathit{ev}, R', \phi, \mathit{act},m\rangle\hfill \\ 
\hfill (\forall \overline{x}, \overline{z}.\,( S(\overline{x}, \overline{z}) \Rightarrow \exists \overline{w}.\, R(\overline{x}, \overline{w}) ) \in \mathit{Ex}(s) \hfill \quad
\hfill \overline{t} = (\overline{v}, \overline{w})\hfill }$ \\
\\

$\infer[\text{\begin{tabular}{c} Trigger \\ Exception\end{tabular}}]
{r, i-1 \attMod \phi [\overline{x}^{|R'|} \mapsto \mathit{tpl}(\mathit{last}(r^{i-1}))] }
{
\hfill 1 < i \leq |r| \hfill \quad
\hfill r^{i} = r^{i-1}  \concat t \concat s \hfill \quad
\hfill \mathit{invoker}(\mathit{last}(r^{i-1}))= u  \hfill \\ 
\hfill s = \langle \mathit{db}, U, \mathit{sec}, T, V, h, \langle t, \langle \langle u', \texttt{SELECT}, \phi[\overline{x} \mapsto \mathit{tpl}(\mathit{last}(r^{i-1}))]\rangle, \top, \top, \emptyset\rangle, \mathit{stmt} , \mathit{tr} \rangle   \hfill \\
\hfill \mathit{secEx}(s) = \top \vee \mathit{Ex}(s) \neq \emptyset \hfill \quad
\hfill t=  \langle \mathit{id},\mathit{ow},  \mathit{ev}, R, \phi, \mathit{act},m\rangle\hfill }$ \\
\\

$\infer[\text{\begin{tabular}{c} Trigger \\ \texttt{INSERT}\\ Exception\end{tabular}}]
{r, i-1 \attMod \neg R(\overline{t}) }
{
\hfill 1 < i \leq |r| \hfill \quad
\hfill r^{i} = r^{i-1}  \concat t \concat s  \hfill \quad
\hfill \mathit{invoker}(\mathit{last}(r^{i-1}))= u \hfill \\ 
\hfill s = \langle \mathit{db}, U, \mathit{sec}, T, V, h, \langle t, \langle \langle u', \texttt{SELECT}, \phi\rangle, \top, \top, \emptyset\rangle,  \langle \langle u', \texttt{INSERT}, R, \overline{t}\rangle, \mathit{res}, \mathit{aC}, E\rangle , \mathit{tr} \rangle   \hfill \\ 
\hfill\mathit{secEx}(s) = \bot \hfill \quad 
\hfill \mathit{Ex}(s) \neq \emptyset \hfill \quad
\hfill t=  \langle \mathit{id},\mathit{ow},  \mathit{ev}, R', \phi, \mathit{act},m\rangle\hfill }$ \\
\\

$\infer[\text{\begin{tabular}{c} Trigger \\ \texttt{DELETE} \\ Exception\end{tabular}}]
{r, i-1 \attMod R(\overline{t}) }
{
\hfill 1 < i \leq |r| \hfill \quad
\hfill r^{i} = r^{i-1}  \concat t \concat s  \hfill \quad
\hfill \mathit{invoker}(\mathit{last}(r^{i-1}))= u \hfill \\ 
\hfill s = \langle \mathit{db}, U, \mathit{sec}, T, V, h, \langle t, \langle \langle u', \texttt{SELECT}, \phi\rangle, \top, \top, \emptyset\rangle,  \langle \langle u', \texttt{DELETE}, R, \overline{t}\rangle, \mathit{res}, \mathit{aC}, E\rangle , \mathit{tr} \rangle   \hfill \\ 
\hfill \mathit{secEx}(s) = \bot \hfill \quad 
\hfill \mathit{Ex}(s) \neq \emptyset \hfill \quad
\hfill t=  \langle \mathit{id},\mathit{ow},  \mathit{ev}, R', \phi, \mathit{act},m\rangle\hfill }$ \\
\\

$\infer[\text{\begin{tabular}{c} Trigger \\ Rollback \\ \texttt{INSERT}\end{tabular}}]
{r, i \attMod   \neg R(\overline{t})  }
{
\hfill n+1 < i \leq |r| \hfill \quad
\hfill s_{1}, s_{2}, \ldots, s_{n} \in \Omega_{M}  \hfill \quad
\hfill t_{1}, \ldots, t_{n} \in {\cal TRIGGER}_{D}\hfill \\ 
\hfill \mathit{secEx}(s_{n}) = \top \vee \mathit{Ex}(s_{n}) \neq \emptyset  \hfill \quad
\hfill r^{i} = r^{i-n-1}  \concat  \langle u, \mathtt{INSERT}, R, \overline{t}\rangle \concat s_{1} \concat t_{1} \concat s_{2} \concat \ldots \concat t_{n} \concat s_{n}  \hfill \\ 
\hfill s_{n} = \langle \mathit{db}, U, \mathit{sec}, T, V, h, \langle t_{n}, \mathit{when}, \mathit{stmt} \rangle, \langle \epsilon, \epsilon, \epsilon, \epsilon \rangle \rangle \hfill  
} $ \\
\\

$\infer[\text{\begin{tabular}{c} Trigger \\ Rollback \\ \texttt{DELETE}\end{tabular}}]
{r, i \attMod   R(\overline{t})  }
{
\hfill n+1 < i \leq |r| \hfill \quad
\hfill s_{1}, s_{2}, \ldots, s_{n} \in \Omega_{M} \hfill  \quad
\hfill t_{1}, \ldots, t_{n} \in {\cal TRIGGER}_{D}\hfill \\ 
\hfill \mathit{secEx}(s_{n}) = \top \vee \mathit{Ex}(s_{n}) \neq \emptyset  \hfill \quad
\hfill r^{i} = r^{i-n-1}  \concat \langle u, \mathtt{DELETE}, R, \overline{t}\rangle \concat s_{1} \concat t_{1} \concat s_{2} \concat \ldots \concat t_{n} \concat s_{n}  \hfill \\ 
\hfill s_{n} = \langle \mathit{db}, U, \mathit{sec}, T, V, h, \langle t_{n}, \mathit{when}, \mathit{stmt} \rangle, \langle \epsilon, \epsilon, \epsilon, \epsilon \rangle \rangle  \hfill 
} $ \\
\\

\end{tabular}

\caption{Extracting knowledge from trigger's exceptions}\label{table:rules:adversary:11}
\end{figure*}

%% file: eop.tex
\clearpage
\section{Database Integrity}\label{app:eop}

In this section, we present the formal definition of the $\auth$ relation, which is used to define \correctness{}. 
Let $P = \langle M, f\rangle$ be an \accessControlConfiguration{}, where $M = \langle D, \Gamma\rangle$ is a system configuration and $f$ is an $M$-\acf{}.
We denote by   ${\cal VIEW}_{D}^{\mathit{owner}}$ the set of all $D$-views with the owner's privileges, i.e., ${\cal VIEW}_{D}^{\mathit{owner}} = \{\langle V, o, q, m\rangle \in {\cal VIEW}_D \;|\; m = O\}$, and by ${\cal PRIV}_{D}^{\mathtt{SELECT}, {\cal VIEW}^{\mathit{owner}}_{D}}$ the set of privileges $\{\mathit{pr} \in {\cal PRIV}_{D} \,|\, \mathit{pr} = \langle  \mathtt{SELECT}, V\rangle \wedge V \in {\cal VIEW}_{D}^{\mathit{owner}}\}$.
Given a state an $M$-state $s = \langle \mathit{db}, U,\mathit{sec},T,V , c \rangle $ and a revoke command $r = \langle \ominus, u, p, u' \rangle$, we denote by $\mathit{applyRev}(s, r)$ the state $\langle \mathit{db}, U,\mathit{revoke}(\mathit{sec}, u,p,u'),T,  V ,c\rangle$ obtained by executing the \texttt{REVOKE} command.
Given a system's configuration $M = \langle D, \Gamma\rangle$, a query $q$, a set of views $V$ with owner's privileges, and a set of tables $T$, we say that $V$ and $T$ determine $q$, denoted by $\mathit{determines}_{M}(T,  V,q)$, iff for all $\mathit{db} \in \Omega_{D}^{\Gamma}$,  for all $\mathit{db}_{1}$, $\mathit{db}_{2} \in \llbracket \mathit{db}\rrbracket_{V,T}$, $[q]^{\mathit{db}_{1}} = [q]^{\mathit{db}_{2}}$,  where $\llbracket \mathit{db}\rrbracket_{V,T}$ denotes the set $\{\mathit{db}' \in \Omega_{D}^{\Gamma}\,|\, \forall T_{1} \in T.\, T_{1}(\mathit{db}) = T_{1}(\mathit{db}') \wedge \forall V_{1} \in V.\, V_{1}(\mathit{db}) = V_{1}(\mathit{db}')\}$.
Further details on the concept of determinacy can be found in \cite{nash2010views}.
Finally, the relation $\auth \subseteq \Omega_{M} \times ({\cal A}_{D,{\cal U}} \cup {\cal TRIGGER}_{D})$ is the smallest relation satisfying the inference rules given in Figure~\ref{figure:eop:auth:full}.

\begin{figure*}
\centering
\begin{tabular}{c c}
$
\infer[\begin{tabular}{c}\texttt{INSERT}\\\texttt{DELETE}\end{tabular}]
{ \langle\mathit{db}, U, \mathit{sec}, T,V,c \rangle \auth \langle u, \mathit{op}', R, \overline{t} \rangle
}
{
\hfill u,u' \in U \hfill \quad 
\hfill R \in D \hfill \quad 
\hfill \overline{t} \in \mathbf{dom}^{|R|}  \hfill \quad
\hfill g = \langle \mathit{op}, u, \langle \mathit{op}', R\rangle, u' \rangle \hfill \quad
\hfill g \in \mathit{sec}  \hfill \\
\hfill \langle\mathit{db}, U, \mathit{sec}, T,V,c \rangle \auth g \hfill \quad 
\hfill \mathit{op}' \in \{\mathtt{INSERT}, \mathtt{DELETE}\} \hfill
}
$
\\
$
\infer[\begin{tabular}{c}\texttt{CREATE}\\\texttt{VIEW}\end{tabular}]
{ \langle\mathit{db}, U, \mathit{sec}, T,V,c \rangle \auth \langle u, \mathtt{CREATE}, v \rangle
}
{
\hfill u,u' \in U \hfill \quad 
\hfill v \in {\cal VIEW}_{D} \hfill \quad 
\hfill g = \langle \mathit{op}, u, \langle \mathtt{CREATE\ VIEW}\rangle, u' \rangle \hfill \\
\hfill g \in \mathit{sec}  \hfill \quad
\hfill \langle\mathit{db}, U, \mathit{sec}, T,V,c \rangle \auth g \hfill 
}
$ 
 
$
\infer[\begin{tabular}{c}\texttt{CREATE}\\\texttt{TRIGGER}\end{tabular}]
{ \langle\mathit{db}, U, \mathit{sec}, T,V,c \rangle \auth \langle u, \mathtt{CREATE}, t \rangle
}
{
\hfill u,u' \in U \hfill \quad 
\hfill t = \langle \mathit{id},\mathit{ow},  \mathit{ev}, R, \phi, \mathit{stmt}, \mathit{m}\rangle \hfill \\
\hfill t \in {\cal TRIGGER}_{D} \hfill \\ 
\hfill g = \langle \mathit{op}, u, \langle \texttt{CREATE TRIGGER}, R\rangle, u' \rangle \hfill \\
\hfill g \in \mathit{sec}  \hfill \quad
\hfill \langle\mathit{db}, U, \mathit{sec}, T,V,c \rangle \auth g \hfill 
}
$
\\
\\
$
\infer[\begin{tabular}{c}\texttt{INSERT}\\\texttt{DELETE}\\\text{admin}\end{tabular}]
{ \langle\mathit{db}, U, \mathit{sec}, T,V,c \rangle \auth \langle \mathit{admin}, \mathit{op}', R, \overline{t} \rangle
}
{\
\hfill R \in D \hfill \quad 
\hfill \overline{t} \in \mathbf{dom}^{|R|}  \hfill \quad
\hfill \mathit{op}' \in \{\mathtt{INSERT}, \mathtt{DELETE}\} \hfill
}
$
 
$
\infer[\begin{tabular}{c}\texttt{CREATE}\\\texttt{VIEW}\\\text{admin}\end{tabular}]
{ \langle\mathit{db}, U, \mathit{sec}, T,V,c \rangle \auth \langle \mathit{admin}, \mathtt{CREATE}, v \rangle
}
{
\hfill v \in {\cal VIEW}_{D} \hfill \quad
\hfill v = \langle \mathit{id}, \admin, q, m \rangle \hfill
}
$
\\
\\
$
\infer[\begin{tabular}{c}\texttt{CREATE}\\\texttt{TRIGGER}\\\text{admin}\end{tabular}]
{ \langle\mathit{db}, U, \mathit{sec}, T,V,c \rangle \auth \langle \mathit{admin}, \mathtt{CREATE}, t \rangle
}
{
\hfill t = \langle \mathit{id},\admin, \mathit{ev}, R, \phi, \mathit{stmt}, \mathit{m}\rangle \hfill \\
\hfill t \in {\cal TRIGGER}_{D} \hfill 
}
$
 
$
\infer[\mathtt{REVOKE}]
{ \langle\mathit{db}, U, \mathit{sec}, T,V,c \rangle \auth \langle \ominus, u, \mathit{priv}, u' \rangle
}
{
\hfill u,u' \in U \hfill \quad 
\hfill \mathit{priv} \in {\cal PRIV}_{D} \hfill \\
\hfill s = \langle\mathit{db}, U, \mathit{sec}, T,V,c  \rangle \hfill \quad 
\hfill s' = \langle\mathit{db}, U, \mathit{sec}', T,V,c  \rangle \hfill \\
\hfill s' = \mathit{applyRev}(s,\langle \ominus, u, p, u' \rangle) \hfill \\
\hfill \forall g \in \mathit{sec}'.\, s' \auth g \hfill
}
$
\\
\\
$
\infer[\mathtt{GRANT}\text{-1}]
{ \langle\mathit{db}, U, \mathit{sec}, T,V,c \rangle \auth \langle \mathit{op}, u, \mathit{priv}, u' \rangle
}
{
\hfill u,u', u'' \in U \hfill \quad 
\hfill \mathit{op} \in \{\oplus,\oplus^{*}\} \hfill \quad 
\hfill \mathit{priv} \in {\cal PRIV}_{D} \hfill \\
\hfill g = \langle \oplus^{*}, u', \mathit{priv}, u''\rangle  \hfill \quad
\hfill g \in \mathit{sec} \hfill \quad
\hfill  \langle\mathit{db}, U, \mathit{sec}, T,V,c \rangle \auth g \hfill
}
$
 
$
\infer[\mathtt{GRANT}\text{-2}]
{ \langle\mathit{db}, U, \mathit{sec}, T,V,c \rangle \auth \langle \mathit{op}, u, \mathit{priv}, \mathit{admin} \rangle
}
{
\hfill u \in U \hfill \quad 
\hfill \mathit{op} \in \{\oplus,\oplus^{*}\} \hfill \\ 
\hfill \mathit{priv} \in {\cal PRIV}_{D} \setminus {\cal PRIV}_{D}^{\mathtt{SELECT}, {\cal VIEW}^{\mathit{owner}_{D}}} \hfill 
}
$

\\
\\
$
\infer[\mathtt{GRANT}\text{-3}]
{ \langle\mathit{db}, U, \mathit{sec}, T,V,c \rangle \auth \langle \mathit{op}, u, \mathit{priv}, \mathit{owner} \rangle
}
{
\hfill u, \mathit{owner} \in U \hfill \quad 
\hfill \mathit{op} \in \{\oplus,\oplus^{*}\} \hfill \quad 
\hfill \mathit{priv} = \langle \mathtt{SELECT}, v\rangle \hfill \\
\hfill v = \langle \mathit{id}, \mathit{owner}, q, O \rangle \hfill \quad
\hfill v \in V \hfill \quad
\hfill V'\subseteq V \cap {\cal VIEW}_{D}^{\mathit{owner}}\hfill \\
\hfill T'\subseteq D \hfill\quad
\hfill \mathit{determines}_{M}(T',V', q) \hfill \\
\hfill \mathit{hasAccess}(\langle\mathit{db}, U, \mathit{sec}, T,V,c \rangle, V' \cup T', \mathit{owner}, \oplus^{*}) \hfill
}$
 
$
\infer[\mathtt{GRANT}\text{-4}]
{ \langle\mathit{db}, U, \mathit{sec}, T,V,c \rangle \auth \langle \mathit{op}, u, \mathit{priv}, \mathit{admin} \rangle
}
{
\hfill u, \mathit{owner} \in U \hfill \quad 
\hfill \mathit{op} \in \{\oplus,\oplus^{*}\} \hfill \quad 
\hfill \mathit{priv} = \langle \mathtt{SELECT}, v\rangle \hfill \\
\hfill v = \langle \mathit{id}, \mathit{owner}, q, O \rangle \hfill \quad
\hfill v \in V \hfill \quad
\hfill \mathit{owner} \neq \mathit{admin} \hfill \\
\hfill V'\subseteq V \cap {\cal VIEW}_{D}^{\mathit{owner}}\hfill \quad
\hfill T'\subseteq D \hfill\quad
\hfill \mathit{determines}_{M}(T',V', q) \hfill \\
\hfill \mathit{hasAccess}(\langle\mathit{db}, U, \mathit{sec}, T,V,c \rangle, V' \cup T', \mathit{owner}, \oplus) \hfill
}$
\\
\\
$
\infer[\mathtt{GRANT}\text{-5}]
{ \langle\mathit{db}, U, \mathit{sec}, T,V,c \rangle \auth \langle \mathit{op}, u, \mathit{priv}, \mathit{owner} \rangle
}
{
\hfill u, \mathit{owner} \in U \hfill \quad 
\hfill \mathit{op} \in \{\oplus,\oplus^{*}\} \hfill \quad
\hfill v \in V \hfill \\ 
\hfill \mathit{priv} = \langle \mathtt{SELECT}, v\rangle \hfill \quad
\hfill v = \langle \mathit{id}, \mathit{owner}, q, A \rangle \hfill \quad
}$
 
$
\infer[\texttt{ADD USER}]
{ \langle\mathit{db}, U, \mathit{sec}, T,V,c \rangle \auth \langle u', \mathtt{ADD\_USER}, u \rangle
}
{
\hfill u \in {\cal U} \hfill \quad \hfill u' = \admin \hfill
}
$
\\
\\
$
\infer[\begin{tabular}{c}\texttt{EXECUTE}\\\texttt{TRIGGER}-1\end{tabular}]
{ \langle\mathit{db}, U, \mathit{sec}, T,V,c \rangle \auth t
}
{
\hfill t = \langle \mathit{id},\mathit{ow},  \mathit{ev}, R, \phi, \mathit{stmt}, O\rangle \hfill \quad
\hfill t \in T \hfill \\
\hfill \langle\mathit{db}, U, \mathit{sec}, T,V,c \rangle \auth \mathit{getAction}(\mathit{stmt},\mathit{ow},\mathit{tpl}(c)) \hfill \\
\hfill [\phi[\overline{x}^{|R|} \mapsto \mathit{tpl}(c)]]^{\mathit{db}} = \top \hfill  
}
$
\\
\\

$
\infer[\begin{tabular}{c}\texttt{EXECUTE}\\\texttt{TRIGGER}-2\end{tabular}]
{ \langle\mathit{db}, U, \mathit{sec}, T,V,c \rangle \auth t
}
{
\hfill t = \langle \mathit{id},\mathit{ow},  \mathit{ev}, R, \phi, \mathit{stmt}, A\rangle \hfill \quad
\hfill t \in T \hfill \\ 
\hfill \langle\mathit{db}, U, \mathit{sec}, T,V,c \rangle \auth \mathit{getAction}(\mathit{stmt},\mathit{invoker}(c),\mathit{tpl}(c)) \hfill \\
\hfill \langle\mathit{db}, U, \mathit{sec}, T,V,c \rangle \auth \mathit{getAction}(\mathit{stmt},\mathit{ow},\mathit{tpl}(c)) \hfill \\
\hfill [\phi[\overline{x}^{|R|} \mapsto \mathit{tpl}(c)]]^{\mathit{db}} = \top \hfill  
}
$

\\
\\
$
\infer[\texttt{SELECT}]
{ \langle\mathit{db}, U, \mathit{sec}, T,V,c \rangle \auth \langle u, \mathtt{SELECT}, q \rangle
}
{
\hfill u \in U \quad \hfill q \in \mathit{RC} \hfill
}
$
 
$
\infer[\begin{tabular}{c}\texttt{EXECUTE}\\\texttt{TRIGGER}-3\end{tabular}]
{ \langle\mathit{db}, U, \mathit{sec}, T,V,c \rangle \auth t
}
{
\hfill t = \langle \mathit{id},\mathit{ow},  \mathit{ev}, R, \phi, \mathit{stmt}, m\rangle \hfill \\
\hfill t \in T \hfill \quad
\hfill [\phi[\overline{x}^{|R|} \mapsto \mathit{tpl}(c)]]^{\mathit{db}} =\bot \hfill  
}
$

\\
\\

$
\mathit{hasAccess}(\langle\mathit{db}, U, \mathit{sec}, T,V,c \rangle, S, u, \mathit{op}) =  \left\{ 
  \begin{array}{l l}
  \top & \text{if }  \mathit{u} \neq \mathit{admin} \wedge \forall v \in S.\,\exists u'' \in U, g \in \mathit{sec}, \mathit{op}' \in \{\mathit{op},\oplus^{*}\}.\\
  & \qquad  g =  \langle \mathit{op}', u, \langle \mathtt{SELECT}, v \rangle, u'' \rangle \wedge \langle\mathit{db}, U, \mathit{sec}, T,V,c \rangle \auth g \\
    \top & \text{if }  \mathit{u} = \mathit{admin} \wedge \forall v \in S.\,\exists u'' \in U, \mathit{op}' \in \{\mathit{op},\oplus^{*}\}.\\
  & \qquad   \langle\mathit{db}, U, \mathit{sec}, T,V,c \rangle \auth \langle \mathit{op}', u, \langle \mathtt{SELECT}, v \rangle, u'' \rangle  \\
  \bot & \mathit{otherwise}
  \end{array}\right.
$
\\\\
\end{tabular}

\caption{Definition of the $\auth$ relation}\label{figure:eop:auth:full}
\end{figure*}

%% file: indistinguishability.tex
\clearpage
\section{Data Confidentiality}\label{app:indistinguishability}

In this section, we  define indistinguishability of runs.
We first formalize the notion of $u$-projection.
Afterwards, we define the notion of consistency between $u$-projections.
Finally, we formalize the indistinguishability relation $\cong_{P,u}$.

We recall that, given a run $r$, we denote by $r^i$, where $1 \leq i \leq |r|$, the prefix of $r$ obtained by truncating $r$ at the $i$-th state.
In the rest of the paper, we use $r^0$ to denote the empty run.

\subsection{Projections}
Let $P = \langle M,f\rangle$ be an \accessControlConfiguration{},  where $M = \langle D,\Gamma\rangle$ and $f$ is an $M$-\acf{},  $L$ be the $P$-LTS, and $u$ be a user in ${\cal U}$.
Given a run $r \in \mathit{traces}(L)$, its \emph{$u$-projection}, denoted by $r|_u$, is obtained by 
(1) replacing each action not issued by $u$ with $*$, 
(2) replacing each trigger whose invoker is not $u$ with $*$, and 
(3) replacing all non-empty sequences of $*$-transitions with a single $*$-transition.
Note that the $*$-transitions in the $u$-projections represent whether $u$'s actions are executed consecutively or not.
With a slight abuse of notation, we extend all the notation we use for runs also to $u$-projections.
For instance, $r|_{u}^{i}$ denotes the prefix  obtained by truncating $r|_{u}$ at its $i$-th state.
Formally, the $u$-projection $r|_{u}$ is defined as $\mathit{c}(\mathit{v}(r,u))$.
The function $\mathit{v}$ takes as input a run $r$ and a user $u$ and returns another run in which all non-$u$ actions are replaced with $*$.
\[
\mathit{v}(r,u) = \left\{ 
  \begin{array}{l l}
	\mathit{v}(r^{|r|-1},u) \concat a \concat s & \text{if } r = r^{|r|-1}\concat a \concat s \text{ and } s \in \Omega_{M}\\
												&	 \text{and } a \in {\cal A}_{D,u} \text{ and } |r| > 1\\ 
 	\mathit{v}(r^{|r|-1},u) \concat * \concat s & \text{if } r = r^{|r|-1}\concat a \concat s \text{ and } s \in \Omega_{M}\\
												&	\text{and } a \in {\cal A}_{D,u'}  \text{ and } u'\neq u \text{ and}\\
												& |r| > 1\\
	\mathit{v}(r^{|r|-1},u) \concat t \concat s & \text{if } r = r^{|r|-1}\concat t \concat s \text{ and } s \in \Omega_{M}\\
												& 	\text{and } t \in {\cal TRIGGER}_{D} \text{ and }\\
												&	\mathit{invoker}(\mathit{last}(r^{|r|-1})) = u \text{ and}\\
												& |r| > 1\\ 
 	\mathit{v}(r^{|r|-1},u) \concat * \concat s & \text{if } r = r^{|r|-1}\concat t \concat s \text{ and } s \in \Omega_{M} \\
 												&	\text{and } t \in {\cal TRIGGER}_{D} \text{ and }\\
												&	\mathit{invoker}(\mathit{last}(r^{|r|-1})) \neq u \text{ and}\\
												& |r| > 1\\ 										    
	s											& \text{if } r = s \text{ and } s \in \Omega_{M}\\
  \end{array} \right.
  \]
The function $c$ takes as input a run $r$ containing $*$-transitions and replaces each sequence of $*$-transitions with a single $*$-transition.
Note that the function $c$ is obtained by repeatedly applying the function $c'$ until the computation reaches a fixed point.
The function $c'$ is as follows:
\[
\mathit{c'}(r) = \left\{ 
  \begin{array}{l l}
	\mathit{c'}(r^{|r|-1}) \concat a \concat s 	& \text{if } r = r^{|r|-1}\concat a \concat s \text{ and } a \neq *\\
												& \text{and } s \in \Omega_{M} \text{ and } |r| >1\\
 	\mathit{c'}(r^{|r|-2}) \concat * \concat s 	& \text{if } r = r^{|r|-2}\concat * \concat s' \concat * \concat s \text{ and }\\
 												& s,s' \in \Omega_{M} \text{ and } |r| > 2\\						    
	s											& \text{if } r = s \text{ and } s \in \Omega_{M}\\				    
	r											& \text{if } r = s \concat * \concat s' \text{ and } s, s' \in \Omega_{M}\\
  \end{array} \right.
  \]

\newcommand*{\hSp}{25pt}
\newcommand*{\vSp}{25pt}
\newcommand*{\labelSp}{3pt}
\begin{figure}
    \begin{tikzpicture}[scale=0.95, every node/.style={draw=black,fill=black,shape=circle,inner sep=1.2pt}]

    		\node[label=left:{$r_{1}$}]  (r10) at (0,0){};
    		\node[right = \hSp of r10] (r11) {};
    		\node[right = \hSp of r11] (r12) {};
    		\node[right = \hSp of r12] (r13) {};
    		\node[right = \hSp of r13] (r14) {};
    		\node[right = \hSp of r14] (r15) {};
    		\node[right = \hSp of r15] (r16) {};
    		\node[right = \hSp of r16] (r17) {};
    		\node[right = \hSp of r17] (r18) {};
    		\path 	(r10) edge   																					(r11) 
    			 	(r11) edge   																					(r12) 
    			 	(r12) edge node [above=\labelSp, fill=white, draw=white,inner sep=0pt, shape=rectangle] {$a_1$}		(r13) 
    			 	(r13) edge node [above=\labelSp, fill=white, draw=white,inner sep=0pt, shape=rectangle] {$a_2$} 	(r14) 
    			 	(r14) edge   																					(r15) 
    			 	(r15) edge   																					(r16) 
    			 	(r16) edge node [above=\labelSp, fill=white, draw=white,inner sep=0pt, shape=rectangle] {$a_3$} 	(r17) 
    			 	(r17) edge   																					(r18);

    		\node[below = \hSp of r10, label=left:{$r_{2}$}]  (r20) {};
    		\node[right = \hSp of r20] (r21) {};
    		\node[right = \hSp of r21] (r22) {};
    		\node[right = \hSp of r22] (r23) {};
    		\node[right = \hSp of r23] (r24) {};
    		\node[right = \hSp of r24] (r25) {};
    		\node[right = \hSp of r25] (r26) {};
    		\node[right = \hSp of r26] (r27) {};
    		\path 	(r20) edge   																					(r21) 
    			 	(r21) edge node [above=\labelSp, fill=white, draw=white,inner sep=0pt, shape=rectangle] {$a_1$} 	(r22) 
    			 	(r22) edge node [above=\labelSp, fill=white, draw=white,inner sep=0pt, shape=rectangle] {$a_2$} 	(r23) 
    			 	(r23) edge  																						(r24) 
    			 	(r24) edge node [above=\labelSp, fill=white, draw=white,inner sep=0pt, shape=rectangle] {$a_3$}		(r25) 
    			 	(r25) edge   																					(r26) 
    			 	(r26) edge  																						(r27); 
    		
    		\node[below = \hSp of r20, label=left:{$r_{3}$}]  (r30) {};
    		\node[right = \hSp of r30] (r31) {};
    		\node[right = \hSp of r31] (r32) {};
    		\node[right = \hSp of r32] (r33) {};
    		\node[right = \hSp of r33] (r34) {};
    		\node[right = \hSp of r34] (r35) {};
    		\path 	(r30) edge   																					(r31) 
    			 	(r31) edge node [above=\labelSp, fill=white, draw=white,inner sep=0pt, shape=rectangle] {$a_1$} 	(r32) 
    			 	(r32) edge node [above=\labelSp, fill=white, draw=white,inner sep=0pt, shape=rectangle] {$a_2$} 	(r33) 
    			 	(r33) edge  																						(r34) 
    			 	(r34) edge node [above=\labelSp, fill=white, draw=white,inner sep=0pt, shape=rectangle] {$a_3$}		(r35);
    		
    		\node[below = \hSp of r30, label=left:{$r_{4}$}]  (r40) {};
    		\node[right = \hSp of r40] (r41) {};
    		\node[right = \hSp of r41] (r42) {};
    		\node[right = \hSp of r42] (r43) {};
    		\node[right = \hSp of r43] (r44) {};
    		\node[right = \hSp of r44] (r45) {};
    		\node[right = \hSp of r45] (r46) {};
    		\node[right = \hSp of r46] (r47) {};
    		\path 	(r40) edge node [above=\labelSp, fill=white, draw=white,inner sep=0pt, shape=rectangle] {$a_1$} 	(r41) 
    			 	(r41) edge 																						(r42) 
    			 	(r42) edge node [above=\labelSp, fill=white, draw=white,inner sep=0pt, shape=rectangle] {$a_2$} 	(r43) 
    			 	(r43) edge  																						(r44) 
    			 	(r44) edge node [above=\labelSp, fill=white, draw=white,inner sep=0pt, shape=rectangle] {$a_3$}		(r45) 
    			 	(r45) edge   																					(r46) 
    			 	(r46) edge  																						(r47);

    		\node[below = \hSp of r40, label=left:{$r_{1}|_u$}]  (r10u) {};
    		\node[right = \hSp of r10u] (r11u) {};
    		\node[right = \hSp of r11u] (r12u) {};
    		\node[right = \hSp of r12u] (r13u) {};
    		\node[right = \hSp of r13u] (r14u) {};
    		\node[right = \hSp of r14u] (r15u) {};
    		\node[right = \hSp of r15u] (r16u) {};
    		\path 	(r10u) edge node [above=\labelSp, fill=white, draw=white,inner sep=0pt, shape=rectangle] {$*$}		(r11u) 
    			 	(r11u) edge node [above=\labelSp, fill=white, draw=white,inner sep=0pt, shape=rectangle] {$a_1$}	(r12u) 
    			 	(r12u) edge node [above=\labelSp, fill=white, draw=white,inner sep=0pt, shape=rectangle] {$a_2$} 	(r13u) 
    			 	(r13u) edge node [above=\labelSp, fill=white, draw=white,inner sep=0pt, shape=rectangle] {$*$} 		(r14u) 
    			 	(r14u) edge node [above=\labelSp, fill=white, draw=white,inner sep=0pt, shape=rectangle] {$a_3$} 	(r15u) 
    			 	(r15u) edge node [above=\labelSp, fill=white, draw=white,inner sep=0pt, shape=rectangle] {$*$} 		(r16u);

    		\node[below = \hSp of r10u, label=left:{$r_{2}|_u$}]  (r20u) {};
    		\node[right = \hSp of r20u] (r21u) {};
    		\node[right = \hSp of r21u] (r22u) {};
    		\node[right = \hSp of r22u] (r23u) {};
    		\node[right = \hSp of r23u] (r24u) {};
    		\node[right = \hSp of r24u] (r25u) {};
    		\node[right = \hSp of r25u] (r26u) {};
    		\path 	(r20u) edge node [above=\labelSp, fill=white, draw=white,inner sep=0pt, shape=rectangle] {$*$}		(r21u) 
    			 	(r21u) edge node [above=\labelSp, fill=white, draw=white,inner sep=0pt, shape=rectangle] {$a_1$}	(r22u) 
    			 	(r22u) edge node [above=\labelSp, fill=white, draw=white,inner sep=0pt, shape=rectangle] {$a_2$} 	(r23u) 
    			 	(r23u) edge node [above=\labelSp, fill=white, draw=white,inner sep=0pt, shape=rectangle] {$*$} 		(r24u) 
    			 	(r24u) edge node [above=\labelSp, fill=white, draw=white,inner sep=0pt, shape=rectangle] {$a_3$}	(r25u) 
    			 	(r25u) edge node [above=\labelSp, fill=white, draw=white,inner sep=0pt, shape=rectangle] {$*$} 		(r26u); 
    		
    		\node[below = \hSp of r20u, label=left:{$r_{3}|_u$}]  (r30u) {};
    		\node[right = \hSp of r30u] (r31u) {};
    		\node[right = \hSp of r31u] (r32u) {};
    		\node[right = \hSp of r32u] (r33u) {};
    		\node[right = \hSp of r33u] (r34u) {};
    		\node[right = \hSp of r34u] (r35u) {};
    		\path 	(r30u) edge node [above=\labelSp, fill=white, draw=white,inner sep=0pt, shape=rectangle] {$*$}		(r31u) 
    			 	(r31u) edge node [above=\labelSp, fill=white, draw=white,inner sep=0pt, shape=rectangle] {$a_1$}	(r32u) 
    			 	(r32u) edge node [above=\labelSp, fill=white, draw=white,inner sep=0pt, shape=rectangle] {$a_2$} 	(r33u) 
    			 	(r33u) edge node [above=\labelSp, fill=white, draw=white,inner sep=0pt, shape=rectangle] {$*$} 		(r34u) 
    			 	(r34u) edge node [above=\labelSp, fill=white, draw=white,inner sep=0pt, shape=rectangle] {$a_3$}	(r35u); 
    		
    		\node[below = \hSp of r30u, label=left:{$r_{4}|_u$}]  (r40u) {};
    		\node[right = \hSp of r40u] (r41u) {};
    		\node[right = \hSp of r41u] (r42u) {};
    		\node[right = \hSp of r42u] (r43u) {};
    		\node[right = \hSp of r43u] (r44u) {};
    		\node[right = \hSp of r44u] (r45u) {};
    		\node[right = \hSp of r45u] (r46u) {};
    		\path 	(r40u) edge node [above=\labelSp, fill=white, draw=white,inner sep=0pt, shape=rectangle] {$a_1$} 	(r41u) 
    			 	(r41u) edge node [above=\labelSp, fill=white, draw=white,inner sep=0pt, shape=rectangle] {$*$}		(r42u) 
    			 	(r42u) edge node [above=\labelSp, fill=white, draw=white,inner sep=0pt, shape=rectangle] {$a_2$} 	(r43u) 
    			 	(r43u) edge node [above=\labelSp, fill=white, draw=white,inner sep=0pt, shape=rectangle] {$*$}		(r44u) 
    			 	(r44u) edge node [above=\labelSp, fill=white, draw=white,inner sep=0pt, shape=rectangle] {$a_3$}	(r45u) 
    			 	(r45u) edge node [above=\labelSp, fill=white, draw=white,inner sep=0pt, shape=rectangle] {$*$}		(r46u);

    \end{tikzpicture}
  \caption{The runs $r_1$, $r_2$, $r_3$, and $r_4$, where the states are represented using black dots, the actions $a_1$, $a_2$, and $a_3$ issued by the user $u$ are written above the edges connecting the states, and the actions of the other users are omitted.
 The $u$-projections of these runs are, respectively, $r_1|_u$, $r_2|_u$, $r_3|_u$, and $r_4|_u$.
  The runs $r_{1}$ and $r_{2}$ have  $u$-projections with the same labels, whereas the runs $r_{3}$ and $r_{4}$ have $u$-projections with different labels.
}\label{figure:projections}
  \label{schedule}
\end{figure}
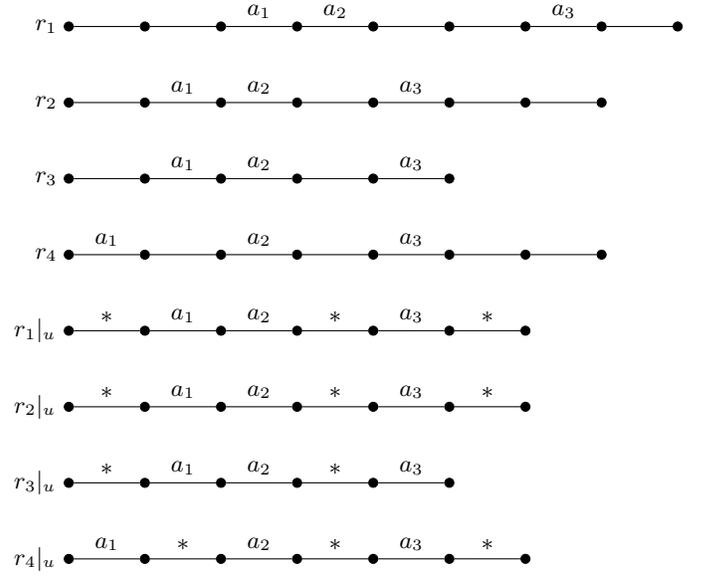

\subsection{Consistency}
Before defining the notion of consistency, we define the function $\mathit{labels}$ which takes as input a run $r$ and returns as output the sequence of labels in the run.
In more detail, $\mathit{labels}(r)$ is obtained from $r$  by dropping all the states.
We now define the notion of consistency between two $u$-projections.
\begin{definitionInt}
Let $P = \langle M,f\rangle$ be an \accessControlConfiguration{},  where $M = \langle D,\Gamma\rangle$ and $f$ is an $M$-\acf{},  $L$ be the $P$-LTS, and $u$ be a user in ${\cal U}$.
Furthermore, let $r|_{u}$ and $r'|_{u}$ be two $u$-projections for the runs $r$ and $r'$ in $\mathit{traces}(L)$.
We say that \emph{$r|_{u}$ and $r'|_{u}$ are consistent} iff the following conditions hold:
\begin{compactenum}
\item $|r|_{u}| =  |r'|_{u}|$.

\item $\mathit{labels}(r|_{u}) = \mathit{labels}(r'|_{u})$.

\item $\mathit{triggers}(\mathit{last}(r|_{u})) = \epsilon$ iff $\mathit{triggers}(\mathit{last}(r'|_{u})) = \epsilon$.

\item for all $i$ such that $1 \leq i \leq |r|_{u}|$, if  $r|_{u}^{i} = r|_{u}^{i-1} \concat a \concat s$ and $a \neq *$, then
\begin{compactitem}
\item $\mathit{res}(\mathit{last}(r|_{u}^{i})) = \mathit{res}(\mathit{last}(r'|_{u}^{i}))$, 
\item $\mathit{secEx}(\mathit{last}(r|_{u}^{i})) = \mathit{secEx}(\mathit{last}(r'|_{u}^{i}))$,
\item if $a$ is a trigger, then $\mathit{acC}(\mathit{last}(r|_{u}^{i})) = \mathit{acC}(\mathit{last}(r'|_{u}^{i}))$,
\item $\mathit{invoker}(\mathit{last}(r|_{u}^{i})) = \mathit{invoker}(\mathit{last}(r'|_{u}^{i}))$,
\item $\mathit{triggers}(\mathit{last}(r|_{u}^{i})) = \mathit{triggers}(\mathit{last}(r'|_{u}^{i}))$,
\item $\mathit{tpl}(\mathit{last}(r|_{u}^{i})) = \mathit{tpl}(\mathit{last}(r'|_{u}^{i}))$,
\item and $\mathit{Ex}(\mathit{last}(r|_{u}^{i})) = \mathit{Ex}(\mathit{last}(r'|_{u}^{i}))$. $\square$
\end{compactitem}
\end{compactenum}
\end{definitionInt}

Figure \ref{figure:projections} depicts four runs. 
The states are represented just as black dots and the action between two states is written above the edge connecting them.
Note that we  represent just the actions $a_{1}$, $a_{2}$, and $a_{3}$ issued by the user $u$.
Assume that\begin{inparaenum}[(a)]
\item the action's effects are the same in all the runs and
\item the $\mathit{invoker}$, $\mathit{res}$, $\mathit{secEx}$, $\mathit{triggers}$, $\mathit{tpl}$, and $\mathit{Ex}$ functions return the same results in all runs.
\end{inparaenum}
It is easy to see that  $r_{1}|_{u}$ and $r_{2}|_{u}$ are consistent projections, whereas $r_{3}|_{u}$ and $r_{4}|_{u}$ are not.
Furthermore, there is no other pair of consistent $u$-projections between the runs in the figure.

\subsection{Indistinguishability}

Let $M = \langle D, \Gamma \rangle$ be a system configuration, $s = \langle \mathit{db}, U, \mathit{sec}, \\ T, V\rangle$ be an $M$-partial state, and $u \in U$ be a   user.
The set $\mathit{permissions}(s, u)$ is  $ \mathit{permissions}(s,u) := \{\langle \oplus, \texttt{SELECT}, O \rangle\, |\, \exists u' \\ \in U, \mathit{op} \in \{\oplus,\oplus^{*}\}.\,\langle \mathit{op}, u, \langle \texttt{SELECT}, O\rangle,  u' \rangle \in \mathit{sec}  \}$.
Note that $\mathit{permissions}(s,\admin) = D \cup V$ since the administrator has read access to the whole database.
We extend permissions to $M$-states as follows.
Given an $M$-state $s' = \langle \mathit{db}, U, \mathit{sec}, T, V, c\rangle$, $\mathit{permissions}(s',u) =  \mathit{permissions}(\langle \mathit{db}, U, \\ \mathit{sec}, T, V\rangle,u)$.

We are now ready to introduce the notion of indistinguishability between two runs.
Intuitively, two runs $r$ and $r'$ are indistinguishable for a user $u$ iff 
(1) their $u$-projections are consistent, and
(2) for each action of the user $u$ as well as for the last states in the two runs, the policy, the triggers, the views, the users, and the data disclosed by the policy are the same in $r$ and $r'$.

\begin{definitionInt}
Let $P = \langle M,f \rangle$ be an \accessControlConfiguration{}, $L$ be the $P$-LTS, and $u$ be a user.

We say that two runs $r$ and $r'$ in $\mathit{traces}(L)$ are \emph{$(P,u)$-indistinguishable}, written $r \cong_{u,P} r'$, iff 
\begin{compactenum}
\item $r|_u$ and $r'|_u$ are consistent, 
\item $\mathit{pState}(\mathit{last}(r))$ and $\mathit{pState}(\mathit{last}(r'))$ are $(M,u)$-data indistinguishable, and
\item for all $i$ such that $1 \leq i \leq |r|_{u}|-1$, if  $r|_{u}^{i+1} = r|_{u}^{i} \concat a \concat s$, $a \neq *$, and $s \in \Omega_{M}$, then $\mathit{pState}(\mathit{last}(r|_{u}^{i}))$ and $\mathit{pState}(\mathit{last}(r'|_{u}^{i}))$ are $(M,u)$-data indistinguishable.
\end{compactenum}

We say that two $M$-partial states $s = \langle \mathit{db}, U, \mathit{sec}, T,V\rangle$ and $s' = \langle \mathit{db}', U', \mathit{sec}', T',V'\rangle$ are \emph{$(M,u)$-data indistinguishable}, written $s \cong_{u,M}^{data} s'$, iff
\begin{compactenum}
\item $U = U'$,
\item $\mathit{sec} = \mathit{sec}'$,
\item $T = T'$, 
\item $V = V'$, 
\item for all relation schema $R \in D$ for which $\langle\oplus, \mathtt{SELECT},  R\rangle \in \mathit{permissions}(s,u)$, $\mathit{db}(R) = \mathit{db}'(R)$, and
\item for all views $v \in {\cal VIEW}_{D}^{\mathit{owner}}$ for which $\langle\oplus, \mathtt{SELECT},  v\rangle \in \mathit{permissions}(s,u)$,  $\mathit{db}(v) = \mathit{db}'(v)$. $\square$
\end{compactenum}
\end{definitionInt}

\begin{proposition}
Let $P = \langle M,f \rangle$ be an \accessControlConfiguration{}, $L$ be the $P$-LTS, and $u \in {\cal U}$ be a user.
The indistinguishability relation $\cong_{P,u}$ is an equivalence relation over $\mathit{traces}(L)$. 
\end{proposition}

\begin{proof}
We now prove that $\cong_{P,u}$ is reflexive, symmetric, and transitive. 
This implies the fact that $\cong_{P,u}$ is an equivalence relation over $\mathit{traces}(L)$.
In the following, let $P = \langle M,f \rangle$ be an \accessControlConfiguration{}, $L$ be the $P$-LTS, and $u \in {\cal U}$ be a user. 
From the definition of data indistinguishability and the results in \cite{guarnieri2014optimal}, it follows that the data-indistinguishability relation  $\cong_{u,M}^{data}$ is an equivalence relation over the set of all partial states.

\smallskip
{\bf Reflexivity} 
Let $r \in \mathit{traces}(L)$ be a run.
It follows trivially that $r|_{u} = r|_{u}$.
From this, it follows that $r|_{u}$ and $r|_{u}$ are consistent.
It is easy to see that $r$ is indistinguishable from $r$.
Indeed, the database states are the same in $r$ and $r$ and the data-indistinguishability relation is reflexive \cite{guarnieri2014optimal}.

\smallskip
{\bf Symmetry}
Let $r, r' \in \mathit{traces}(L)$ be two runs such that $r \cong_{P,u} r'$.
From this, it follows that $r|_{u}$ and $r'|_{u}$ are consistent.
Note that the consistency definition is symmetric.
Therefore, also $r'|_{u}$ and $r|_{u}$ are consistent.
From this and the symmetry of data indistinguishability \cite{guarnieri2014optimal}, it follows the symmetry of  $\cong_{P,u}$.

\smallskip
{\bf Transitivity}
Let $r, r', r'' \in \mathit{traces}(L)$ be three runs such that $r \cong_{P,u} r'$ and $r' \cong_{P,u} r''$.
From this it follows that $r|_{u}$ and $r'|_{u}$ are consistent and $r'|_{u}$ and $r''|_{u}$ are consistent.
It is easy to see that also $r|_{u}$ and $r''|_{u}$ are consistent.
From this and the transitivity of data indistinguishability \cite{guarnieri2014optimal}, it follows the transitivity of  $\cong_{P,u}$. 
\end{proof}

Given a run $r$, we denote by $\llbracket r \rrbracket_{P, u}$ the equivalence class of $r$ defined by $\cong_{P,u}$ over $\mathit{traces}(L)$.
Similarly, we denote by $\llbracket s \rrbracket_{u,M}^{\mathit{data}}$ the equivalence class of $s$ defined by $\cong_{u,M}^{\mathit{data}}$ over $\Pi_M$.

%% file: eopsecEnforcement.tex
\clearpage
\section{Enforcing Database Integrity}\label{app:enforcement:eopsec}
In this section, we first define the access control function $f_{\mathit{int}}$, which models the $f_{\mathit{int}}$ procedure described in \S\ref{sect:enf:alg}. 
Afterwards, we prove that the function $f_{\mathit{int}}$ satisfies the \correctness{} property.
Finally, we prove that the data complexity of $f_{\mathit{int}}$ is \complexity{}.

The function $f_{\mathit{int}}$ is as follows: 
\[\mathit{f}_{\mathit{int}}(s, \mathit{a}) = \left\{ 
  \begin{array}{l l}
  \top & \text{if } \mathit{trigger}(s) = \epsilon \wedge 
   s \auth^{\mathit{appr}} \mathit{a} \\ 
  \top & \text{if } \mathit{trigger}(s) = t \wedge t \neq \epsilon \wedge  
   \mathit{a} = \mathit{trigCond}(s)  \\ 

  \top & \text{if } \mathit{trigger}(s) = t \wedge t \neq \epsilon \wedge  \mathit{a} = \mathit{trigAct}(s) \wedge
   \\
  &\quad  s \auth^{\mathit{appr}} t \\
  
  \bot & \text{otherwise}
  \\
  \end{array}\right.\]
  
The function $\mathit{trigCond}(s)$ (respectively $\mathit{trigAct}(s)$) returns the condition (respectively the action) associated with the trigger $\mathit{trigger}(s)$.
If $\mathit{trigger}(s) = \langle \mathit{id}, \mathit{ow}, e, R, \phi, \mathit{st}, O\rangle$, then $\mathit{trigAct}(s) = \mathit{getAction}(  \mathit{st}, \mathit{ow}, \mathit{tpl}(s))$ and $\mathit{trigCond}(s) = \langle \mathit{ow}, \mathtt{SELECT}, \phi[\overline{x}^{|R|} \mapsto \mathit{tpl}(s)]\rangle$.
If $\mathit{trigger}(s) = \langle \mathit{id}, \mathit{ow}, e, R, \phi, \\ \mathit{st}, A\rangle$, then $\mathit{trigAct}(s) = \mathit{getAction}(  \mathit{st}, \mathit{invoker}(s), \mathit{tpl}(s))$ and $\mathit{trigCond}(s) = \langle \mathit{invoker}(s), \mathtt{SELECT}, \phi[\overline{x}^{|R|} \mapsto \mathit{tpl}(s)]\rangle$.

Recall that, given  an $M$-state $s = \langle \mathit{db}, U,\mathit{sec},T,V ,c \rangle $ and a revoke statement $r = \langle \ominus, u, p, u' \rangle$,  $\mathit{applyRev}(s, r)$ denotes the state $\langle \mathit{db}, U,\mathit{revoke}(\mathit{sec},u,p,u'),T,V ,c\rangle$.  

The relation $\auth^{\mathit{appr}} \subseteq \Omega_{M} \times ({\cal A}_{D,{\cal U}} \cup {\cal TRIGGER}_{D})$ is the smallest relation satisfying the inference rules given in Figure~\ref{figure:eop:auth:approx}.
We remark that $\auth^{\mathit{appr}}$ is a sound and computable under-approximation of the relation $\auth$. 
In the rules, we use a number of auxiliary functions. The most important ones are:
\begin{compactenum}[(a)]
\item the $\mathit{aT}$ (respectively $\mathit{aV}$) function that takes as input a database state, an operator $\mathit{op}$ in $\{\oplus,\oplus^{*}\}$, and a user, and returns the set of tables (respectively views) that the user is authorized to read (if $\mathit{op} = \oplus$) or to delegate the read access (if $\mathit{op} = \oplus^{*}$) according to our approximation of $\auth$, and
\item the $\mathit{apprDet}$ function  is used to determine whether a set of tables and a set of views completely determine the result of a formula $\phi$ in all possible database states. Note that the function  $\mathit{apprDet}$ is a sound under-approximation of the concept of \emph{query determinacy} \cite{nash2010views}.
\end{compactenum}
In the following, we define the functions $\mathit{extend}$ and $\mathit{apprDet}$. The functions $\mathit{aT}$ and $\mathit{aV}$ are defined in Figure \ref{figure:eop:auth:approx}.
We assume that both the formula $\phi$ and the set of views $V$ in the state $s$ contain just views with owner's privileges.
This is without loss of generality.
Indeed, views with activator's privileges are just syntactic sugar, they do not disclose additional information to a user $u$ other than what he is already authorized to read because they are executed under $u$'s privileges.
If $\phi$ and $s$ contain views with activator's privileges, we can compute another formula $\phi'$ and a state $s'$ without views with activator's privileges as follows.
We replace, in the formula $\phi$, the predicates of the form $V(\overline{x})$, where $V$ is a view with activator's privileges,  with $V$'s definition, and we repeat this process until the resulting formula $\phi'$ no longer contains views with activator's privileges.
Similarly, the set $V'$ is obtained from $V$ by (1) removing all views with activator's privileges, and (2) for each view $v \in V$ with owner's privileges, replacing the predicates of the form $V(\overline{x})$ in $v$'s definition, where $V$ is a view with activator's privileges,  with $V$'s definition until $v$'s definition no longer contains views with activator's privileges.
The security policy $\mathit{sec}'$ is also obtained from $\mathit{sec}$ by removing all references to views with activator's privileges.
Therefore, in \S\ref{app:enforcement:eopsec:extend}-\ref{app:enforcement:eopsec:determinacy} we ignore views with activator's privileges as the extension to the general case is trivial.

\begin{figure*}
\centering
\begin{tabular}{c}
$\emph{apprDet}(T,V,\phi,s,M) = \left\{ 
  \begin{array}{l l}
	\top & \text{if } \exists \langle v,o,q,O\rangle \in \mathit{extend}(M,s,V).\, q =\{\overline{x} | \phi(\overline{x})\}\\
	\top & \text{if } \phi = (x = v) \vee \phi = \top \vee \phi = \bot\\
	\top & \text{if } \phi = R(\overline{x}) \wedge R \in T\\
	\top & \text{if } \phi = V(\overline{x}) \wedge \exists u \in {\cal U}, q \in \mathit{RC}.\, \langle V, u,q, O\rangle \in V\\ 
	\top & \text{if } \phi = (\psi \wedge \gamma) \wedge \emph{apprDet}(T,V,\psi,s,M) = \top \wedge \emph{apprDet}(T,V,\gamma,s,M)  = \top\\
	\top & \text{if } \phi = (\psi \vee \gamma) \wedge \emph{apprDet}(T,V,\psi,s,M)  = \top \wedge \emph{apprDet}(T,V,\gamma,s,M)  = \top\\
	\top & \text{if } \phi = (\neg \psi) \wedge \emph{apprDet}(T,V,\psi,s,M)  = \top\\
	\top & \text{if } \phi = (\exists x.\, \psi) \wedge \emph{apprDet}(T,V,\psi,s,M) = \top\\ 
	\top & \text{if } \phi = (\forall x.\, \psi) \wedge \emph{apprDet}(T,V,\psi,s,M) = \top\\ 
  	\bot & \text{otherwise}
  \end{array}\right.$
\end{tabular}
\caption{\emph{apprDet} function}\label{figure:access:control:function:EOP:approxDetermines}
\end{figure*}

\subsection{Extend function}\label{app:enforcement:eopsec:extend} 

We now define the \emph{extend} function, which takes as input a system configuration $M$, an $M$-state $s$, and a set of views with owner's privileges, and returns a set of views $V'$ such that $V \subseteq V'$.
Given a system configuration $M$, an $M$-partial state $s = \langle \mathit{db},U,\mathit{sec}, T,V\rangle$, and a normalized view $\langle v, o, q, O\rangle \in V$, we denote by $\mathit{inline}_{M}(\langle v, o, q, O\rangle,s)$ the view $\langle v, o, q', O\rangle$ where $q'$ is obtained from $q$ by replacing all occurrences of views in $V$ with owner's privileges with their definitions.
Note that $\mathit{inline}_{M}$ does not compute a fixpoint, i.e., if a view's definition refers to another view, the latter is not replaced with its definition.
The function $\mathit{extend}(M,s,V)$ returns the set $V \cup \{\mathit{inline}(v,s) | v \in \mathit{extend}(M,s,V)\}$.

\begin{lemma}\thlabel{theorem:eopsec:enforcement:extend}
Let $M = \langle D, \Gamma \rangle$ be a system configuration, $s = \langle \mathit{db}, U, \mathit{sec},  T,V \rangle$ be an $M$-partial state, $V' \subseteq V$ be a set of views with owner's privileges.
For each view $v \in \mathit{extend}(M,s,V')$, there is a view $v' \in V'$ such that $v$ and $v'$ disclose the same data.
\end{lemma}

\begin{proof}
\emph{Sketch: }Assume, for contradiction's sake, that there is a view $v \in \mathit{extend}(M,s,V')$ such that all the views in $V'$ disclose different data from $v$. 
This is impossible because $v$ has been obtained by a view $v' \in V'$ just by replacing the views with their definitions and the definitions of $v$ and $v'$ are semantically equivalent.
\end{proof}

\subsection{A sound under-approximation of query determinacy}\label{app:enforcement:eopsec:determinacy}
The definition of the function $\emph{apprDet}(T,V,q)$ is shown in Figure \ref{figure:access:control:function:EOP:approxDetermines}.
Before proving that $\emph{apprDet}$ is a sound approximation of $\emph{determines}$, we extend \emph{determines} from sentences to formulae.

We first introduce assignments.
Let ${\bf dom}$ be the universe and ${\bf var}$ be an infinite countably set of variable identifiers.
An \emph{assignment} $\nu$ is a partial function from ${\bf var}$ to ${\bf dom}$ that maps variables to values in the universe.
Given a formula $\phi$ and an assignment $\nu$, we say that \emph{$\nu$ is well-formed for $\phi$} iff $\nu$ is defined for all variables in $\mathit{free}(\phi)$.
Given an assignment $\nu$ and a sequence of variables $\overline{x}$ such that $\nu$ is defined for each $x \in \overline{x}$, we denote by $\nu(\overline{x})$ the tuple obtained by replacing each occurrence of $x \in \overline{x}$ with $\nu(x)$.
Given an assignment $\nu$, a variable $v \in {\bf var}$, and a value $u \in {\bf dom}$, we denote by $\nu \oplus [v \mapsto u]$ the assignment $\nu'$ obtained as follows: $\nu'(v') = \nu(v')$ for any $v' \neq v$, and $\nu'(v) = u$. 
Finally, given a formula $\phi$ with free variables $\mathit{free}(\phi)$ and an assignment $\nu$, we denote by $\phi \circ \nu$ the formula $\phi'$ obtained by replacing, for each free variable $x \in \mathit{free}(\phi)$ such that $\nu(x)$ is defined, all the free occurrences of $x$ with $\nu(x)$.

%
%
Given a system's configuration $M = \langle D, \Gamma\rangle$, a formula $\phi$, a set of views $V$ with owner's privileges, a set of tables $T$, and a well-formed assignment $\nu$ for $\phi$, we say that $V$ and $T$ determine $(\phi, \nu)$, denoted by $\mathit{determines}_{M}(T,V,\phi,\nu)$, iff for all $\mathit{db} \in \Omega_{D}^{\Gamma}$,  for all $\mathit{db}_{1}$, $\mathit{db}_{2} \in \llbracket \mathit{db}\rrbracket_{V,T}$, $[\phi \circ \nu]^{\mathit{db}_{1}} = [\phi \circ \nu]^{\mathit{db}_{2}}$.
In the following, given a view $\langle u, o,q,m\rangle$, we denote by $\mathit{def}(\langle u, o,q,m\rangle)$ its definition $q$.

In \thref{theorem:eopsec:enforcement:approxDetermines}, we show that $\mathit{apprDet}$ is, indeed, a sound under-approximation of query determinacy.

\begin{lemma}\thlabel{theorem:eopsec:enforcement:approxDetermines}
Let $M = \langle D, \Gamma \rangle$ be a system configuration, $s = \langle \mathit{db}, U, \mathit{sec},  T,V \rangle$ be an $M$-partial state, $T' \subseteq D$ be a set of tables, $V' \subseteq V$ be a set of views with owner's privileges, and $\phi$ be a formula.
If $\mathit{apprDet}(T',V',\phi,s,M) = \top$, then for all well-formed assignments $\nu$  for $\phi$, $\mathit{determines}_{M}(T',V', \\ \phi,\nu)$ holds.
\end{lemma}

\begin{proof}
Let $M = \langle D, \Gamma \rangle$ be a system configuration, $s = \langle \mathit{db}, U, \mathit{sec},  T,V \rangle$ be an $M$-partial state, $T' \subseteq D$ be a set of tables, $V' \subseteq V$ be a set of views with owner's privileges, and $\phi$ be a formula.
We prove the lemma by structural induction over the formula $\phi$.

\smallskip
\noindent
{\bf Base Case: } There are a number of alternatives.
\smallskip
\begin{compactitem}
\item[\boldmath${\phi := R(\overline{x})}$] Assume that $\mathit{apprDet}(T,V,R(\overline{x}),s,M) = \top$. There are two cases:
\begin{compactenum}
\item $R \in T'$.
In this case, the set $T'$ trivially determines the formula $R(\overline{x})$ for any well-formed assignment $\nu$.
Therefore, $\mathit{determines}_{M}(T',V',R(\overline{x}), \nu)$ holds.
Indeed, assume that this is not the case.
Thus, there are three database states $\mathit{db}$, $\mathit{db}_{1}$, and $\mathit{db}_{2}$ such that $\mathit{db}_1, \mathit{db}_2 \in \llbracket \mathit{db} \rrbracket_{V',T'}$ and $[R(\overline{x}) \circ \nu]^{\mathit{db}_1} \neq [R(\overline{x}) \circ \nu]^{\mathit{db}_2}$.
From this and the $\mathit{RC}$ semantics, it follows that $\mathit{db}_1(R) \neq \mathit{db}_2(R)$.
From this, $R \in T'$, and  $\mathit{db}_1, \mathit{db}_2 \in \llbracket \mathit{db} \rrbracket_{V',T'}$, it follows that  $\mathit{db}_1(R) = \mathit{db}_2(R)$ leading to a contradiction.

\item there is a view $v'$ in $\mathit{extend}(M,s,V')$ such that $\mathit{def}(v') = \{\overline{x} | R(\overline{x})\}$. This means that there is a sequences of views $V_{1}$, \ldots, $V_{n}$ in $s$ such that $\mathit{def}(V_{1}) = \{\overline{x} | R(\overline{x})\}$, $\mathit{def}(V_{2}) = \{\overline{x} | V_{1}(\overline{x})\}$, \ldots,  $\mathit{def}(V_{n}) = \{\overline{x} | V_{n-1}(\overline{x})\}$, and $V_{n} \in V'$. Therefore, the set $V'$ trivially determines the formula $R(\overline{x})$ for any well-formed assignment $\nu$, and $V_{n}$ and $R$ are equivalent. Therefore, $\mathit{determines}_{M}(T',  V', R(\overline{x}),\nu)$ holds.
\end{compactenum}

\item[\boldmath$\phi := V(\overline{x})$] Assume that $\mathit{apprDet}(T,V,V(\overline{x}),s,M) = \top$. There are two cases:
\begin{compactenum}
\item There is a view $\langle V, o, q, O\rangle \in V'$. In this case, the set $V'$ trivially determines the formula $V(\overline{x})$ for any  assignment $\nu$ that is well-formed for $\phi$. 
Therefore, $\mathit{determines}_{M}(T',V',V(\overline{x}),\nu)$ holds.

\item there is a view $v'$ in $\mathit{extend}(M,s,V')$ such that $\mathit{def}(v') = \{\overline{x} | V(\overline{x})\}$. This means that there is a sequences of views $V_{1}$, \ldots, $V_{n}$ in $s$ such that $\mathit{def}(V_{1}) = \{\overline{x} | V(\overline{x})\}$, $\mathit{def}(V_{2}) = \{\overline{x} | V_{1}(\overline{x})\}$, \ldots,  $\mathit{def}(V_{n}) = \{\overline{x} | V_{n-1}(\overline{x})\}$, and $V_{n} \in V'$. Therefore, the set $V'$ trivially determines the formula $V(\overline{x})$ for any well-formed assignment $\nu$, and  $V_{n}$ and $V$ are equivalent. Therefore, $\mathit{determines}_{M}(T',V', V(\overline{x}),\nu)$ holds.
\end{compactenum}

\item[\boldmath $\phi := x = v$] For any well-formed assignment $\nu$, the empty set trivially determines the formula $ x = v$ and $\mathit{apprDet}\\(T',V',x = v,s,M) = \top$.

\item[\boldmath $\phi := \top$] The proof of this case is similar to that of $\phi := x = v$.

\item[\boldmath $\phi := \bot$] The proof of this case is similar to that of $\phi := x = v$.

\end{compactitem}
This concludes the proof of the base case.

\smallskip
\noindent
{\bf Induction Step: } Assume that the claim holds for all sub-formulae of $\phi$.
There are a number of cases:
\begin{compactitem}

\item[\boldmath $\phi := \psi \wedge \gamma$] Assume that $\mathit{apprDet}(T',V',\psi \wedge \gamma ,s,M) = \top$. There are two cases:
\begin{compactenum}
\item $\mathit{apprDet}(T',V',\psi,s,M) = \top$ and $\mathit{apprDet}(T',V',\\ \gamma ,s,M) = \top$. From the induction hypothesis, it follows that both $\mathit{determines}_{M}(T',V',\psi,\nu)$  and \\$\mathit{determines}_{M}(T',V',\gamma,\nu)$ hold for all well-formed assignments $\nu$. Therefore, also $\mathit{determines}_{M}(T',V',\\\psi \wedge \gamma,\nu)$ holds for all well-formed assignments $\nu$.
Indeed, assume that this is not the case.
Then, there are three database states $\mathit{db}$, $\mathit{db}_{1}$, and $\mathit{db}_{2}$ such that $\mathit{db}_1, \mathit{db}_2 \in \llbracket \mathit{db} \rrbracket_{V',T'}$ and $[(\psi \wedge \gamma) \circ \nu]^{\mathit{db}_1} \neq [(\psi \wedge \gamma) \circ \nu]^{\mathit{db}_2}$.
From this and the $\mathit{RC}$ semantics, there are two cases:
\begin{compactenum}
\item $[\psi \circ \nu]^{\mathit{db}_1} \neq [\psi \circ \nu]^{\mathit{db}_2}$. 
From this, it follows that $\mathit{determines}_{M}(T',V',\psi,\nu)$ does not hold.
This contradicts the fact that $\mathit{determines}_{M}(T', V',\psi, \\ \nu)$ holds.

\item $[\gamma \circ \nu]^{\mathit{db}_1} \neq [\gamma \circ \nu]^{\mathit{db}_2}$. 
The proof of this case is similar to the previous one.
\end{compactenum}

\item there is a view $v'$ in $\mathit{extend}(M,s,V')$ such that $\mathit{def}(v') = \{\overline{x} | \psi \wedge \gamma\}$. From \thref{theorem:eopsec:enforcement:extend}, it follows that there is a view $v'' \in V'$ that is equivalent to $v'$, and, therefore, to $\{\overline{x} | \psi \wedge \gamma\}$. Thus, $\mathit{determines}_{M}(T', V',\psi \wedge \gamma,\nu)$ holds for all assignments $\nu$ that are well-formed for $\phi$.
\end{compactenum}

\item[\boldmath $\phi := \psi \vee \gamma$] This case is similar to $\psi \wedge \gamma$.

\item[\boldmath $\phi := \neg \psi$] Assume that $\mathit{apprDet}(T',V', \neg \psi,s,M) = \top$. There are two cases:
\begin{compactenum}
\item $\mathit{apprDet}(T',V',\psi,s,M) = \top$. 
From the induction hypothesis, it follows that $\mathit{determines}_{M}(T', V', \psi,  \nu)$ holds.
Therefore, also $\mathit{determines}_{M}(T', V',  \neg \psi,  \nu)$ holds.
Indeed, assume that this is not the case.
This means that there are three database states $\mathit{db}$, $\mathit{db}_{1}$, and $\mathit{db}_{2}$ such that $\mathit{db}_1$, $\mathit{db}_2$ are in $\llbracket \mathit{db} \rrbracket_{V',T'}$ and $[\neg \psi \circ \nu]^{\mathit{db}_1} \neq [\neg \psi \circ \nu]^{\mathit{db}_2}$.
From this and the $\mathit{RC}$ semantics, it follows that $[\psi \circ \nu]^{\mathit{db}_1} \neq [\psi \circ \nu]^{\mathit{db}_2}$. 
From this, it follows that $\mathit{determines}_{M}(T', \\ V',\psi,\nu)$ does not hold.
This contradicts the fact that $\mathit{determines}_{M}(T',  V',\psi,\nu)$ holds.

\item there is a view $v'$ in $\mathit{extend}(M,s,V')$ such that $\mathit{def}(v') = \{\overline{x} | \neg \psi\}$.  From \thref{theorem:eopsec:enforcement:extend}, it follows that there is a view $v'' \in V'$ that is equivalent to $v'$, and, therefore, to $\{\overline{x} | \neg \psi\}$. Thus, $\mathit{determines}_{M}\\ (T',V',\neg \psi,\nu)$ holds for all well-formed assignments $\nu$.
\end{compactenum}

\item[\boldmath $\phi := \exists x.\, \psi$] Assume that $\mathit{apprDet}(T',V', \exists x.\, \psi,s,M) = \top$. There are two cases:
\begin{compactenum}
\item $\mathit{apprDet}(T',V',\psi,s,M) = \top$.
From the induction hypothesis, it follows that $\mathit{determines}_{M}(T', V',\psi,\nu)$ holds for all well-formed assignments $\nu$. 
Therefore, also $\mathit{determines}_{M}(T',V',\exists x.\, \psi,\nu)$ holds for all well-formed assignments $\nu$ (note that any well-formed assignment for $\psi$ is also a well-formed assignment for $\exists x.\,\psi$).
Indeed, assume that this is not the case.
This means that there are three database states $\mathit{db}$, $\mathit{db}_{1}$, and $\mathit{db}_{2}$ such that $\mathit{db}_1$, $\mathit{db}_2$ are in $\llbracket \mathit{db} \rrbracket_{V',T'}$ and $[(\exists x.\,\psi) \circ \nu]^{\mathit{db}_1} \neq [(\exists x.\,\psi) \circ \nu]^{\mathit{db}_2}$.
From this and the $\mathit{RC}$ semantics, it follows that there is a value $v \in \mathbf{dom}$ such that $[\psi \circ \nu[x \mapsto v]]^{\mathit{db}_1} \neq [\psi \circ \nu[x \mapsto v]]^{\mathit{db}_2}$. 
Note that $\nu[x \mapsto v]$ is a well-formed assignment for $\psi$.
Let's call the assignment $\nu'$.
From this, it follows that $[\psi \circ \nu']^{\mathit{db}_1} \neq [\psi \circ \nu']^{\mathit{db}_2}$. 
From this, it follows that $\mathit{determines}_{M}(T',V',\psi,\nu')$ does not hold.
This contradicts the fact that $\mathit{determines}_{M}(T',  V',\psi,\nu)$ holds for any well-formed assignment $\nu$.

\item there is a view $v'$ in $\mathit{extend}(M,s,V')$ such that $\mathit{def}(v') = \{\overline{x} | \exists x.\, \psi\}$.  From \thref{theorem:eopsec:enforcement:extend}, it follows that there is a view $v'' \in V'$ that is equivalent to $v'$, and, therefore, to $\{\overline{x} | \exists x.\, \psi\}$. Thus, $\mathit{determines}_{M}(T',V',\exists x.\, \psi,\nu)$ holds for all well-formed assignments for $\exists x.\, \psi$.
\end{compactenum}

\item[\boldmath $\phi := \forall x.\, \psi$] This case is similar to $\exists x.\,\psi$.
\end{compactitem}
This concludes the proof of the induction step.

This completes the proof. 
\end{proof}

We now show that $\auth^{\mathit{appr}}$ is a sound approximation of $\auth$, i.e., if $s \auth^{\mathit{appr}}\mathit{act}$, then $s \auth \mathit{act}$.
A \emph{derivation of $s \auth^{\mathit{appr}}\mathit{act}$} is a proof tree, obtained using the rules defining $\auth^{\mathit{appr}}$, which ends in $s \auth^{\mathit{appr}}\mathit{act}$.
The size of a derivation is  the number of $\auth^{\mathit{appr}}$ rules that are used to show that $s \auth^{\mathit{appr}}\mathit{act}$.
In the following, we switch freely between statements of the form $s \auth^{\mathit{appr}}\mathit{act}$ and their derivations.
We  denote the size of the derivation of $s \auth^{\mathit{appr}}\mathit{act}$ as $|s \auth^{\mathit{appr}}\mathit{act}|$.

\begin{lemma}\thlabel{theorem:eopsec:enforcement:auth:approx}
Let $M = \langle D, \Gamma\rangle$ be a system configuration, $s$ be an $M$-state, $c$ be an $M$-context, and $\mathit{act} \in {\cal A}_{D,{\cal U}} \cup {\cal TRIGGER}_{D}$.
If $s \auth^{\mathit{appr}}\mathit{act}$, then $s \auth \mathit{act}$.
\end{lemma}

\begin{proof}
Let $M = \langle D, \Gamma\rangle$ be a system configuration, $s$ be an $M$-state, $c$ be an $M$-context, and $\mathit{act} \in {\cal A}_{D,{\cal U}} \cup {\cal TRIGGER}_{D}$.
Furthermore, we assume that there is a derivation of $s \auth^{\mathit{appr}}\mathit{act}$.
We prove our claim by structural induction on the size of $s \auth^{\mathit{appr}}\mathit{act}$'s derivation.

{\bf Base Case: } We now show that, for all $s$ and $\mathit{act}$ such that  $|s \auth^{\mathit{appr}}\mathit{act}| = 1$, if $s \auth^{\mathit{appr}}\mathit{act}$, then $s \auth \mathit{act}$. 
There are several cases:
\begin{compactenum}
\item Rule \emph{\texttt{INSERT DELETE} admin}: If $s \auth^{\mathit{appr}}\mathit{act}$, then $s \auth \mathit{act}$ follows trivially from the rule's definition.
\item Rule \emph{\texttt{CREATE VIEW} admin}: If $s \auth^{\mathit{appr}}\mathit{act}$, then $s \auth \mathit{act}$ follows trivially from the rule's definition.
\item Rule \emph{\texttt{CREATE TRIGGER} admin}: If $s \auth^{\mathit{appr}}\mathit{act}$, then $s \auth \mathit{act}$ follows trivially from the rule's definition.
\item Rule \emph{\texttt{SELECT}}: If $s \auth^{\mathit{appr}}\mathit{act}$, then $s \auth \mathit{act}$ follows trivially from the rule's definition.
\item Rule \emph{\texttt{EXECUTE TRIGGER}-3}: If $s \auth^{\mathit{appr}}\mathit{act}$, then $s \auth \mathit{act}$ follows trivially from the rule's definition.
\item Rule \emph{\texttt{GRANT}-2}: If $s \auth^{\mathit{appr}}\mathit{act}$, then $s \auth \mathit{act}$ follows trivially from the rule's definition.
\item Rule \emph{\texttt{GRANT}-5}: If $s \auth^{\mathit{appr}}\mathit{act}$, then $s \auth \mathit{act}$ follows trivially from the rule's definition.
\item Rule \emph{\texttt{ADD USER}}: If $s \auth^{\mathit{appr}}\mathit{act}$, then $s \auth \mathit{act}$ follows trivially from the rule's definition.
\end{compactenum}

{\bf Induction Step: } We now assume that, for all derivations of size less than  $|s \auth^{\mathit{appr}}\mathit{act}|$, it holds that  if $s' \auth^{\mathit{appr}}\mathit{act}'$, then $s' \auth \mathit{act}'$. 
There are several cases:
\begin{compactenum}
\item Rule \emph{\texttt{INSERT DELETE}}: Assume that $s \auth^{\mathit{appr}}\mathit{act}$ holds and that $\mathit{act} = \langle u, \mathtt{op}', R,\overline{t}\rangle$, where $\mathit{op}'$ is one of $\{\mathtt{INSERT}, \\ \mathtt{DELETE}\}$. 
From the rule's definition, it follows that there is a grant $g = \langle \mathit{op}, u, \langle \mathit{op}',R\rangle, u'\rangle$ in $s.\mathit{sec}$ such that $s \auth^{\mathit{appr}} g$.
From this and the induction hypothesis, it follows that  $s \auth g$.
Therefore, $s \auth \mathit{act}$ holds because we can apply the \emph{\texttt{INSERT DELETE}} rule in $\auth$.

\item Rule \emph{\texttt{CREATE VIEW}}: The proof is similar to the one for the \emph{\texttt{INSERT DELETE}} rule.

\item Rule \emph{\texttt{CREATE TRIGGER}}: The proof is similar to the one for the \emph{\texttt{INSERT DELETE}} rule.

\item Rule \emph{\texttt{EXECUTE TRIGGER}-2}: Assume that $s \auth^{\mathit{appr}}\mathit{act}$ holds and that $\mathit{act} = \langle \mathit{i},\mathit{o},  \mathit{e}, R, \phi, \mathit{st}, A\rangle$ such that $[\phi[\overline{x}^{|R|} \\ \mapsto \mathit{tpl}(s)]]^{s.\mathit{db}} = \top$. 
From the rule's definition, it follows that both $s \auth^{\mathit{appr}} \mathit{getAction}(\mathit{st}, \mathit{ow}, \mathit{tpl}(s))$ and $s \auth^{\mathit{appr}} \mathit{getAction}(\mathit{st},\mathit{invoker}(s),  \mathit{tpl}(s))$ hold.
From this and the induction hypothesis, both $s \auth \mathit{getAction}(\mathit{st},\mathit{ow}, \\ \mathit{tpl}(s))$ and $s \auth \mathit{getAction}(\mathit{st}, \mathit{invoker}(s), \mathit{tpl}(s))$ hold.
From this and the \emph{\texttt{EXECUTE TRIGGER}-2} rule in $\auth$, it follows that also $s \auth \mathit{act}$ holds.

\item Rule \emph{\texttt{EXECUTE TRIGGER}-1}: The proof is similar to the one for the \emph{\texttt{EXECUTE TRIGGER}-2} rule.

\item Rule \emph{\texttt{GRANT}-1}: Assume that $s \auth^{\mathit{appr}}\mathit{act}$ holds and that $\mathit{act} = \langle \mathit{op},u,p,u'\rangle$, where $\mathit{op} \in \{\oplus,\oplus^{*}\}$.
From the rule's definition, it follows that there is a grant $g = \langle \oplus^{*}, u', p,  u''\rangle$ in $s.\mathit{sec}$ such that $s \auth^{\mathit{appr}} g$. 
From this and the induction hypothesis, ti follows that  $s \auth g$.
From this and the  \emph{\texttt{GRANT}-1} rule in $\auth$, it follows that $s \auth \mathit{act}$ holds.

\item  Rule \emph{\texttt{GRANT}-3}: Assume that $s \auth^{\mathit{appr}}\mathit{act}$ holds and that $\mathit{act} = \langle \mathit{op},u,p,o\rangle$, where $p = \langle \mathtt{SELECT}, v\rangle$, $v \in {\cal VIEW}_{D}^{\mathit{owner}}$, $\mathit{op} \in \{\oplus,\oplus^{*}\}$, and $o = \mathit{owner}(v)$ such that $o\neq\mathit{admin}$.
Let $T'$ be the set obtained through the $\mathit{aT}$ function and $V'$ be the set obtained through the $\mathit{aV}$ function. 
From the rule's definition, it follows that $\mathit{apprDet}(T',V', \mathit{def}(v)) = \top$. 
From this and  \thref{theorem:eopsec:enforcement:approxDetermines}, it follows that $\mathit{determines}_{M}(T',V', \mathit{def}(v))$ holds.
We now show that for any $\mathit{obj} \in T' \cup V'$, $\mathit{hasAccess}(s', \{\mathit{obj}\}, o,\oplus^{*})$ holds.
There are four cases:
\begin{compactenum}
\item \emph{$o = \admin$ and $\mathit{obj} \in D$.} Since $\mathit{obj} \in T'$, it follows that there is a $g =  \langle \oplus^{*}, o, \langle \texttt{SELECT},\mathit{obj}\rangle, u'\rangle$ such that $s \auth^{\mathit{appr}} g$. From this and the induction hypothesis, it follows that $s \auth g$. Therefore, $\mathit{hasAccess}(s', \{\mathit{obj}\}, o,\oplus^{*})$ holds.
\item \emph{$o \neq \admin$ and $\mathit{obj} \in D$.} Since $\mathit{obj} \in T'$, it follows that there is a $g =  \langle \oplus^{*}, o, \langle \texttt{SELECT},\mathit{obj}\rangle, u'\rangle$ in $\mathit{sec}$ such that $s \auth^{\mathit{appr}} g$. From this and the induction hypothesis,  it follows that $s \auth g$. Thus, $\mathit{hasAccess}(s', \{\mathit{obj}\}, o,  \oplus^{*})$ holds.
\item \emph{$o = \admin$ and $\mathit{obj} \in V$.} 
The proof of this case is similar to that of $o = \admin$ and $\mathit{obj} \in D$.
\item \emph{$o \neq \admin$ and $\mathit{obj} \in V$.} 
The proof of this case is similar to that of $o \neq \admin$ and $\mathit{obj} \in D$.
\end{compactenum}
Note that from $\mathit{hasAccess}(s', A, o,\mathit{op})$ and $\mathit{hasAccess}(s',\\ B, o,\mathit{op})$, it follows that $\mathit{hasAccess}(s', A \cup B, o,\mathit{op})$.  
Thus, $\mathit{hasAccess}(s, T' \cup V', o,\oplus^{*})$ holds. 
From this, it follows that $s \auth \mathit{act}$ holds because we can apply the corresponding rule in $\auth$.

\item Rule \emph{\texttt{GRANT}-4}: The proof is similar to the one for the \emph{\texttt{GRANT}-3} rule.

\item Rule \emph{\texttt{REVOKE}}:  Assume that $s \auth^{\mathit{appr}}\mathit{act}$ holds and that $\mathit{act} = \langle \ominus,u,p,u'\rangle$.
From the rule's definition, it follows that $s' \auth^{\mathit{appr}} g$ for any $g \in s'.\mathit{sec}$, where $s' = \mathit{applyRev}(s, \langle \ominus,u,p,u'\rangle)$.
From the induction's hypothesis, it follows that $s' \auth g$ for any $g \in s'.\mathit{sec}$.
Therefore, we can apply the rule \emph{\texttt{REVOKE}} of $\auth$ to derive $s \auth \mathit{act}$.

\end{compactenum}

This completes our proof.
\end{proof}

\subsection{Database Integrity Proofs}

We are now ready to prove that $f_{\mathit{int}}$ satisfies the \correctness{} property.

\begin{lemma}\thlabel{theorem:f:int:actions:independent:from:db}
For any two states $s = \langle \mathit{db}, U,\mathit{sec}, T,V,c \rangle$, $s' = \langle \mathit{db}', U,\mathit{sec}, T,V,c' \rangle$  in $\Omega_{M}$ and any action $a \in {\cal A}_{D,{\cal U}}$:
\begin{compactenum} 
\item $s \auth a$ iff $s' \auth a$, and

\item $s \auth^{\mathit{appr}} a$ iff $s' \auth^{\mathit{appr}} a$.
\end{compactenum}
\end{lemma}

\begin{proof}
It is easy to see that the only rules that depends on $\mathit{db}$, $\mathit{db}'$, $c$, and $c'$ are \texttt{EXECUTE TRIGGER - 1}, \texttt{EXECUTE TRIGGER - 2}, and \texttt{EXECUTE TRIGGER - 3}.
Since they are not used to evaluate whether $s \auth a$ and $s \auth^{\mathit{appr}} a$ hold for actions in ${\cal A}_{D,{\cal U}}$, the lemma follows trivially.
\end{proof}

\begin{lemma}\thlabel{theorem:eopsec:enforcement:second:variant}
 Let $P = \langle M, f_{\mathit{int}}\rangle$ be an \accessControlConfiguration{}, where $M$ is a system configuration,  and $L$ be the $P$-LTS.
	Then, for all $M$-states $s = \langle \mathit{db}, U,\mathit{sec}, T,V,c \rangle \in \Omega_{M}$ such that $\mathit{trigger}(s) = \epsilon$ and all actions $\mathit{act} \in {\cal A}_{D, {\cal U}}$, if $f_{\mathit{int}}(s,\mathit{act}) = \top$, then $s \auth \mathit{act}$.
\end{lemma}

\begin{proof}
We prove the theorem by contradiction. 
Assume, for contradiction's sake, that the claim does not hold.
Therefore, there is a state $s$ and an action $\mathit{act}$ such that $f_{\mathit{int}}(s,\mathit{act}) = \top$, $\mathit{trigger}(s) = \epsilon$, and $s \not\auth \mathit{act}$. 
Thus, from $f_{\mathit{int}}(s,\mathit{act}) = \top$, $\mathit{trigger}(s) = \epsilon$, and $f_{\mathit{int}}$'s definition, it follows $s \auth^{\mathit{appr}} \mathit{act}$. From this and  \thref{theorem:eopsec:enforcement:auth:approx}, it follows that $s \auth \mathit{act}$ leading to a contradiction.
\end{proof}

\begin{lemma}\thlabel{theorem:eopsec:enforcement:second}
 Let $P = \langle M, f_{\mathit{int}}\rangle$ be an \accessControlConfiguration{}, where $M$ is a system configuration,  and $L$ be the $P$-LTS.
	Then, for all $M$-states $s = \langle \mathit{db}, U,\mathit{sec}, T,V,c \rangle \in \Omega_{M}$ such that $\mathit{trigger}(s) = \epsilon$ and all actions $\mathit{act} \in {\cal A}_{D, {\cal U}}$, and all $M$-states $s'$ reachable from $s$ in one step through $t$, if $\mathit{secEx}(s') =\bot$, then $s \auth \mathit{act}$.
\end{lemma}

\begin{proof}
We prove the theorem by contradiction. 
Assume, for contradiction's sake, that the claim does not hold.
Therefore, there are two states $s$ and $s'$ and an action $\mathit{act}$ such that $s'$ is reachable in one step from $s$ through $\mathit{act}$, $\mathit{secEx}(s') = \bot$, $\mathit{trigger}(s) = \epsilon$, and  $s \not\auth \mathit{act}$. 
From  $\mathit{secEx}(s') = \bot$ and the LTS's rules, it follows that $f_{\mathit{int}}(s,\mathit{act}) = \top$.
From this, $\mathit{trigger}(s) = \epsilon$, and $f_{\mathit{int}}$'s definition, it follows $s \auth^{\mathit{appr}} \mathit{act}$. 
From this and  \thref{theorem:eopsec:enforcement:auth:approx}, it follows that $s \auth \mathit{act}$ leading to a contradiction.
\end{proof}

\begin{lemma}\thlabel{theorem:eopsec:enforcement:third:variant}
 Let $P = \langle M, f_{\mathit{int}}\rangle$ be an \accessControlConfiguration{}, where $M$ is a system configuration,  and $L$ be the $P$-LTS.
	Then, for all $M$-states $s = \langle \mathit{db}, U,\mathit{sec}, T,V,c \rangle \in \Omega_{M}$ and all triggers $t \in {\cal TRIGGER}_{D}$ 
	such that $\mathit{trigger}(s) = t$, the following hold:
	\begin{compactenum}
	\item If $f_{\mathit{int}}(s,c) = \top$ and $[\psi]^{\mathit{db}} = \bot$, then $s \auth t$, where $c = \mathit{trigCond}(s) = \langle u, \mathtt{SELECT}, \psi \rangle$. 
	\item If $f_{\mathit{int}}(s,c) = \top$, $[\psi]^{\mathit{db}} = \top$, and $f_{\mathit{int}}(s,a) = \top$, then $s \auth t$, where $c = \mathit{trigCond}(s) = \langle u, \mathtt{SELECT}, \psi \rangle$ and $a = \mathit{trigAct}(s)$. 
	\end{compactenum}	
\end{lemma}

\begin{proof}
We prove both claims by contradiction.
Assume, for contradiction's sake, that the first claim does not hold. Therefore, there is a state $s$ and a trigger $t$ such that $f_{\mathit{int}}(s,c) \\ = \top$ and $[\psi]^{\mathit{db}} = \bot$ and $s \not\auth t$. 
From  $[\psi]^{\mathit{db}} = \bot$, $\mathit{trigger}(s) = t$, and the  rule \texttt{EXECUTE TRIGGER - 3}, it follows that $s \auth t$ holds, which  leads to a contradiction.
Assume, for contradiction's sake, that the second claim does not hold.
Therefore, there is a state $s$ and a trigger $t$ such that $f_{\mathit{int}}(s,c) = \top$, $[\psi]^{\mathit{db}} = \top$, $f_{\mathit{int}}(s,a) = \top$, and $s \not\auth t$.
 From $f_{\mathit{int}}(s,a) = \top$, it follows $s \auth^{\mathit{appr}} t$. 
From this and \thref{theorem:eopsec:enforcement:auth:approx}, it follows that $s \auth t$ holds leading to a contradiction.
This completes the proof.
\end{proof}

\begin{lemma}\thlabel{theorem:eopsec:enforcement:third}
 Let $P = \langle M, f_{\mathit{int}}\rangle$ be an \accessControlConfiguration{}, where $M$ is a system configuration,  and $L$ be the $P$-LTS.
	Then, for all $M$-states $s = \langle \mathit{db}, U,\mathit{sec}, T,V,c \rangle \in \Omega_{M}$, all triggers $t \in {\cal TRIGGER}_{D}$, and all $M$-states $s'$ reachable from $s$ in one step through $t$, if $\mathit{secEx}(s') =\bot$, then $s \auth t$.
\end{lemma}

\begin{proof}
We prove the theorem by contradiction. 
Assume, for contradiction's sake, that the claim does not hold.
Therefore, there are two states $s$ and $s'$ and a trigger $t$ such that $s'$ is reachable in one step through $t$ from $s$, $\mathit{secEx}(s') = \bot$, and $s \not\auth t$. 
In the following, let $c = \langle u, \mathtt{SELECT}, \psi\rangle$ be $\mathit{trigCond}(s)$ and $a$ be $\mathit{trigAct}(s)$.
Since $t$ is a trigger, $s'$ is reachable in one-step from $s$ through $t$, and $\mathit{secEx}(s') = \bot$, there are two cases, according to the LTS rules:
\begin{compactenum}
\item $f_{\mathit{int}}(s, c)  =\top$ and $[\psi]^{\mathit{db}} = \bot$. In this case, we can always apply the rule \texttt{EXECUTE TRIGGER - 3} in the state $s$ to derive $s \auth t$ leading to a contradiction.

\item $f_{\mathit{int}}(s, c)  =\top$, $[\psi]^{\mathit{db}} = \top$, and $f_{\mathit{int}}(s'',a) = \top$, where $s''$ is the state obtained from $s$ by updating the context according to the LTS rules.
From $f_{\mathit{int}}$'s definition and $f_{\mathit{int}}(s'',a) = \top$, it follows $s'' \auth^{\mathit{appr}} t$. From this and \thref{theorem:eopsec:enforcement:auth:approx}, it follows $s'' \auth t$. Since $s$ and $s''$ are equivalent modulo the context's history and the context's history is not used in the rules defining $\auth$, it also follows that $s \auth t$ holds. This lead to a contradiction.
\end{compactenum}
Both cases lead to a contradiction. This completes the proof.
\end{proof}

We are now ready to prove our main result, namely that $f_{\mathit{int}}$ provides \correctness{}.

\begin{theorem}
Let $P = \langle M, f_{\mathit{int}} \rangle$ be an \accessControlConfiguration{}, where $M = \langle D, \Gamma \rangle$ is a system configuration.
The \acf{} $f_{\mathit{int}}$ provides \correctness{} with respect to $P$.
\end{theorem}

\begin{proof}
To show that $f_{\mathit{int}}$ satisfies the \correctness{} property, we have to prove that for all reachable states $s = \langle \mathit{db},U,\mathit{sec},T,V,c\rangle$:
\begin{compactenum}
\item for all states $s'$ reachable from $s$ in one step through an action $a \in {\cal A}_{D, {\cal U}}$, if $\mathit{secEx}(s') = \bot$, then $s \auth a$,

\item for all states $s'$ reachable from $s$ in one step through a trigger $t \in {\cal TRIGGER}_{D}$, if $\mathit{secEx}(s') = \bot$, then $s \auth t$.

\end{compactenum}

The first condition has been proved in \thref{theorem:eopsec:enforcement:second} and the second one has been proved in \thref{theorem:eopsec:enforcement:third}. Therefore, $f_{\mathit{int}}$ satisfies the \correctness{} property.
\end{proof}

We also prove that, by using $f_{\mathit{int}}$, any reachable state has a consistent policy.
This is the underlying reason why $f_{\mathit{int}}$ prevents Attacks \ref{example:view:escalation1} and \ref{example:view:escalation2}.

\begin{lemma}\thlabel{theorem:integrity:policy:consistency}
Let $P = \langle M, f_{\mathit{int}} \rangle$ be an \accessControlConfiguration{}, where $M = \langle D, \Gamma \rangle$ is a system configuration.
For each reachable state $s = \langle \mathit{db}, U, \mathit{sec}, T, V, c\rangle$, $s \auth g$ for all $g \in \mathit{sec}$.
\end{lemma}

\begin{proof}
We claim that, for any run $r$, the state $\mathit{last}(r)$ is such that for all $p \in \mathit{last}(r).\mathit{sec}$, $\mathit{last}(r) \auth p$. From this, the lemma follows trivially.

We now prove that for any run $r$, the state $\mathit{last}(r)$ is such that for all $p \in \mathit{last}(r).\mathit{sec}$, $\mathit{last}(r) \auth p$.
We do this by structural induction on the length of the run $r$.

\smallskip
\noindent
{\bf Base Case:} The base case consists of the runs containing only one initial state. Note that an initial state contains only grants issued by $\admin$, together with views and triggers owned by $\admin$.
It is easy to see that for any permission $p = \langle \mathit{op}, u, \mathit{pr}, \admin\rangle$ in a policy $\mathit{sec}$ in an initial state $s$, it holds that $s \auth p$. There are two cases:
\begin{compactenum}
\item The privilege $\mathit{pr}$ in $p$ is such that $\mathit{pr} \in {\cal PRIV}_{D} \setminus \\{\cal PRIV}_{D}^{\texttt{SELECT}, {\cal VIEW}_{D}^{\mathit{owner}}}$.  Then, $s \auth p$  by the rule \emph{\texttt{GRANT}-2}.
\item The privilege $\mathit{pr}$ in $p$ is such that $\mathit{pr}$ is in the set ${\cal PRIV}_{D}^{\texttt{SELECT}, {\cal VIEW}_{D}^{\mathit{owner}}}$.  
Recall that $\admin$ is the owner of all views in the state.
Then, $s \auth p$  by the rule \emph{\texttt{GRANT}-3}. Indeed, $\admin$ can read (and delegate the \texttt{SELECT} permission over) all tables in the database. Therefore,  $\mathit{hasAccess}(s,D,\admin,\oplus^{*})$ and $\mathit{determines}_{M} \\ (D,\emptyset, q)$ hold for any query $q$.
\end{compactenum}
This complete the proof for the base case.

\smallskip
\noindent
{\bf Induction Step:} We now assume that for all runs $r'$ of length less than the length of $r$, the state $\mathit{last}(r')$ is such that for all $p \in \mathit{last}(r').\mathit{sec}$, $\mathit{last}(r') \auth p$.
Let $r'$ be the run $r^{|r| -1 }$.
There are two cases, depending on whether $\mathit{act}$ raises an exception or not.
\begin{compactenum}
\item $\mathit{secEx}(\mathit{last}(r)) =\bot$ and $\mathit{Ex}(\mathit{last}(r)) = \emptyset$. 
There are a number of cases depending on $\mathit{act}$:
\begin{compactenum}
\item $\mathit{act}$ is $\langle u, \mathtt{INSERT}, R, \overline{t} \rangle$, $\langle u, \mathtt{DELETE}, R, \overline{t} \rangle$, $\langle u, \mathtt{SELECT}, \\ q \rangle$, $\langle u, \mathtt{ADD\_USER}, u' \rangle$, or  $\langle u, \mathtt{CREATE}, o \rangle$. In these cases,  $\mathit{last}(r').\mathit{sec} = \mathit{last}(r).\mathit{sec}$.
Furthermore,  $\mathit{last}(r').U \subseteq \mathit{last}(r).U$, $\mathit{last}(r').T \subseteq \mathit{last}(r).T$, and $\mathit{last}(r').V \subseteq \mathit{last}(r).V$.
From this and the fact that $\mathit{last}(r') \auth g$ for all $g \in \mathit{last}(r').\mathit{sec}$, it follows that $\mathit{last}(r) \auth g$ for all $g \in \mathit{last}(r).\mathit{sec}$.

\item $\mathit{act}$ is $\langle \mathit{op}, u, p,u' \rangle$, where $\mathit{op} \in \{\oplus,\oplus^{*}\}$. 
From  $\mathit{secEx}(\mathit{last}(r)) =\bot$, it follows that $\mathit{last}(r') \auth \mathit{act}$.
From the induction hypothesis, it follows that $\mathit{last}(r') \auth g$ for all $g \in \mathit{last}(r').\mathit{sec}$.
We claim that, for any grant statement $g$, if $\langle \mathit{db}, U, \mathit{sec},  T, V, c\rangle \\ \auth g$, then $\langle \mathit{db}', U, \mathit{sec}', T, V, c'\rangle \auth g$ for any policy such that $\mathit{sec} \subseteq \mathit{sec}'$.
From the claim, it follows that $\mathit{last}(r)  \auth \mathit{act}$ and $\mathit{last}(r) \auth g$ for all $g \in \mathit{last}(r').\mathit{sec}$.
From this and $\mathit{last}(r).\mathit{sec} = \{\mathit{act}\} \cup \mathit{last}(r').\mathit{sec}$, it follows that $\mathit{last}(r) \auth g$ for all $g \in \mathit{last}(r).\mathit{sec}$.

Our claim that, for any grant statement $g$, if $\langle \mathit{db}, U, \\ \mathit{sec}, T, V, c\rangle  \auth g$, then $\langle \mathit{db}', U, \mathit{sec}', T, V, c'\rangle \auth g$, where $\mathit{sec} \subseteq \mathit{sec}'$, follows trivially from the definition of the rules for \texttt{GRANT} statements.

\item $\mathit{act}$ is $\langle \ominus, u, p,u' \rangle$.
From  $\mathit{secEx}(\mathit{last}(r)) =\bot$, it follows that $\mathit{last}(r') \auth \mathit{act}$.
From this, it follows that $s' \auth^{\mathit{appr}} g$ for all $g \in s'.\mathit{sec}$, where $\mathit{s'} = \mathit{applyRev}(\mathit{last}(r'), \mathit{act})$.
From this and \thref{theorem:eopsec:enforcement:auth:approx}, it follows that $s' \auth g$ for all $g \in s'.\mathit{sec}$.
Recall that $\mathit{last}(r)$ and $\mathit{s'}$ are equivalent modulo the database and the context.
From this, \thref{theorem:f:int:actions:independent:from:db}, and $s' \auth g$ for all $g \in s'.\mathit{sec}$, it follows that $\mathit{last}(r) \auth g$ for all $g \in \mathit{last}(r).\mathit{sec}$.

\item $\mathit{act}$ is a trigger and the \texttt{WHEN} condition is not satisfied. 
In this case, $\mathit{last}(r')$ and $\mathit{last}(r)$ are equivalent modulo the context.
From this, the induction hypothesis,  and \thref{theorem:f:int:actions:independent:from:db}, it follows that $\mathit{last}(r) \auth g$ for all $g \in \mathit{last}(r).\mathit{sec}$.

\item $\mathit{act}$ is a trigger and the \texttt{WHEN} condition is satisfied. 
In this case, the proof is the same as the previous cases depending on the trigger's action.
\end{compactenum}


\item $\mathit{secEx}(\mathit{last}(r)) =\top$ or $\mathit{Ex}(\mathit{last}(r)) \neq \emptyset$. From this and the LTS's rules, it follows that  there is a state $s' \in \{\mathit{last}(r^i) | 1 \leq i \leq |r|-1\}$ such that $\mathit{pState}(\mathit{last}(r)) = \mathit{pState}(s')$ (because there has been a roll-back). Let $\mathit{sec}$ be the policy in $s'$. From the induction hypothesis, it follows that for all $p \in \mathit{sec}$, $s' \auth p$. From this fact, the $\auth$'s definition, $\mathit{pState}(\mathit{last}(r)) = \mathit{pState}(s')$, and \thref{theorem:f:int:actions:independent:from:db}, it follows that for all $p \in \mathit{last}(r).\mathit{sec}$, also $\mathit{last}(r) \auth p$.

\end{compactenum}
This complete the proof for the induction step.

This completes the proof.
\end{proof}

\subsection{Complexity Proofs}

\begin{theorem}\thlabel{theorem:integrity:complexity}
The data complexity of $f_{\mathit{int}}$ is $O(1)$.
\end{theorem}

\begin{proof}
Let $M = \langle D, \Gamma\rangle$ be some fixed system configuration, $a \in {\cal A}_{D,U}$ be some fixed action, $u \in {\cal U}$ be some fixed  user, $U \subseteq {\cal U}$ be some fixed  set of users, $\mathit{sec} \in \Omega^{\mathit{sec}}_{U,D}$ be some fixed  policy, $T$ be some fixed  set of triggers over $D$ whose owners are in $U$, $V$ be some fixed  set of views over $D$ whose owners are in $U$, and $c$ be some fixed  context.
Furthermore, let $\mathit{db} \in \Omega_{D}^\Gamma$ be a database state such that $\langle \mathit{db},U,\mathit{sec},T,V,c\rangle \in \Omega_{M}$.
We denote by $s$ the state $\langle \mathit{db},U,\mathit{sec},T,V,c\rangle$.
We can check whether $f_{int}(s, a) = \top$ as follows:
\begin{compactenum}
\item If $\mathit{trigger}(s) = \epsilon$, then return $\top$ iff $s \auth^{\mathit{appr}} a$.

\item If $\mathit{trigger}(s) \neq \epsilon$ and $a = \mathit{trigCond}(s)$, return $\top$.

\item If $\mathit{trigger}(s) \neq \epsilon$ and $a = \mathit{trigAct}(s)$, return $\top$ iff both $s \auth^{\mathit{appr}} \mathit{getAction}(\mathit{stmt}, \mathit{ow}, \mathit{tpl}(s)) = \top$ and $s \auth^{\mathit{appr}} \mathit{getAction}(\mathit{stmt}, \mathit{invoker}(s), \mathit{tpl}(s)) = \top$, where $\mathit{trigger} \\ (s) = \langle \mathit{id}, \mathit{ow}, \mathit{ev}, R, \phi, \mathit{stmt}, m\rangle$.

\item Otherwise return $\perp$.
\end{compactenum}
Note that in the above algorithm we use all the rules in $\auth^{\mathit{appr}}$ other than \emph{EXECUTE TRIGGER - 1}, \emph{EXECUTE TRIGGER - 2}, and \emph{EXECUTE TRIGGER - 3}.
These are the only rules that depend on the database state.
Therefore, evaluating any statement of the form $s \auth^{\mathit{appr}} a$ can be done in constant time in terms of data complexity.
Therefore, all steps 1--4 can be executed in constant time in terms of data complexity.
\end{proof}

\begin{lemma}\thlabel{theorem:apprDet:complexity}
The complexity of $\mathit{apprDet}$ is $O(|\phi|^3 + |\phi|\concat \mathit{max}^{2}\concat |V|^{3})$.
\end{lemma}

\begin{proof}
Let $M = \langle D, \Gamma\rangle$ be a system configuration, $T \subseteq D$ be a set of tables, $V \subseteq {\cal VIEW}^{\mathit{owner}}_{D}$ be a set of views over $D$, $\phi$ be a formula over $D$, and $s$ be an $M$-state.
An algorithm that computes $\mathit{apprDet}(T,V,\phi,s,M)$ is as follows:
\begin{compactenum}
\item Compute the set $\mathit{extend}(M,s,V)$.

\item Compute the set $S$ of all sub-formulae of $\phi$, i.e., $S = \mathit{subF}(\phi)$.
Note that $\phi \in  \mathit{subF}(\phi)$.

\item Sort by length the set of sub-formulae in such a way that the shortest formula is the first one.

\item Let $S' := \emptyset$. 

\item For each sub-formula $\psi$ in the sequence:
\begin{compactenum}
\item Check whether there is a view $v \in \mathit{extend}(M,s,V)$ such that $\psi$ is $v$'s definition.
If this is the case, let $ S' = S' \cup \mathit{subF}(\psi)$.

\item Make a case distinction on $\psi$:
\begin{compactenum}
\item If $\psi := R(\overline{x})$ and $R \in T$, $S' = S' \cup \mathit{subF}(\psi)$.

\item If $\psi := V(\overline{x})$ and $\langle V, u, q, O\rangle \in V$, $S' = S' \cup \mathit{subF}(\psi)$.

\item If $\psi := \alpha \wedge \beta$ and $\alpha$, $\beta \in S'$, then $S' = S' \cup \mathit{subF}(\psi)$.

\item If $\psi := \alpha \vee \beta$ and $\alpha$, $\beta \in S'$, then $S' = S' \cup \mathit{subF}(\psi)$.

\item If $\psi := \neg \alpha$ and $\alpha \in S'$, then $S' = S' \cup \mathit{subF}(\psi)$.

\item If $\psi := \exists x. \alpha$ and $\alpha \in S'$, then $S' = S' \cup \mathit{subF}(\psi)$.

\item If $\psi := \forall x. \alpha$ and $\alpha \in S'$, then $S' = S' \cup \mathit{subF}(\psi)$.
\end{compactenum}
\end{compactenum}

\item$\mathit{apprDet}(T,V,\phi,s,M) = \top$ iff $S = S'$.
\end{compactenum} 
We claim that the size of $\mathit{subF}(\phi)$ is $O(|\phi|)$ and that computing  $\mathit{subF}(\psi)$ can be done in $O(|\phi|^{2})$.
Let $\mathit{max}$ be the maximum length of the definitions of the views in $V$.
We also claim that the size of the set $\mathit{extend}(M,s,V)$ is $O(\mathit{max}\concat |V|^{3})$ and that  $\mathit{extend}(M,s,V)$ can be computed in $O(|V|^{3} \concat \mathit{max}^{2})$.
From these claims, it follows that the fifth step can be executed in $O(|S|\concat ((|\mathit{extend}(M,s,V)| + |S|) +  |S|^2)))$.
After some simplification, it follows that the fifth step can be executed in  $O(|S|^3 + |S|\concat |\mathit{extend}(M,s,V)|)$.
From this, $|S| = O(|\phi|)$, and $|\mathit{extend}(M,s,V)| =  O(\mathit{max}\concat |V|^{3})$, it follows that the fifth step can be executed in $O(|\phi|^3 + |\phi|\concat \mathit{max}\concat |V|^{3})$.
Therefore, the overall complexity is $O(|\phi|^3 + |\phi|\concat \mathit{max}^{2}\concat |V|^{3})$.

We now prove our claims about $\mathit{subF}(\phi)$.
It is easy to see that the size of $\mathit{subF}(\phi)$ is $O(|\phi|)$.
Indeed, we can view the formula $\phi$ as a tree, where the operators are the internal nodes and the predicates and equalities are the leaves.
Then, there is a sub-formula for each sub-tree.
From this and from the fact that the number of sub-tree of a tree is linear in the number of nodes, it follows that $|\mathit{subF}(\phi)|$ is $O(|\phi|)$.
Note that computing  $\mathit{subF}(\psi)$ can be done in $O(|\phi|^{2})$.

We now prove our claims about $\mathit{extend}(M,s,V)$.
Let $\mathit{max}$ be the maximum length of the definitions of the views in $V$.
For each $v \in V$, computing $\mathit{inline}(v,s)$ can be done in $O(|v|\concat |V| \concat \mathit{max})$.
Furthermore, since the views' definitions are acyclic, after $|V|$ applications of $\mathit{inline}$ there are no views in the view's definition.
Therefore, for each $v \in V$, we can compute all views derivable from $v$ in $O(|v|\concat |V|^{2} \concat \mathit{max})$.
Therefore, we can compute the set  $\mathit{extend}(M,s,V)$ in $O(|V|^{3} \concat \mathit{max}^{2})$.
Therefore, also the size of $\mathit{extend}(M,s,V)$ is less than $O(\mathit{max}\concat |V|^{3})$.
\end{proof}

\begin{figure*}
\centering

\scalebox{.95}{
\begin{tabular}{c c}
\multicolumn{2}{c}{
$
\infer[\begin{tabular}{c}\texttt{INSERT}\\\texttt{DELETE}\end{tabular}]
{ \langle\mathit{db}, U, \mathit{sec}, T,V,c \rangle \auth^{\mathit{appr}} \langle u, \mathit{op}', R, \overline{t} \rangle
}
{
\hfill u,u' \in U \hfill \quad 
\hfill R \in D \hfill \quad 
\hfill \overline{t} \in \mathbf{dom}^{|R|}  \hfill \quad
\hfill g = \langle \mathit{op}, u, \langle \mathit{op}', R\rangle, u' \rangle \hfill \quad
\hfill g \in \mathit{sec}  \hfill \\
\hfill \langle\mathit{db}, U, \mathit{sec}, T,V,c \rangle \auth^{\mathit{appr}} g \hfill \quad 
\hfill \mathit{op} \in \{\oplus,\oplus^{*}\}\hfill \quad
\hfill \mathit{op}' \in \{\mathtt{INSERT}, \mathtt{DELETE}\} \hfill
}
$
}
\\
\\
$
\infer[\begin{tabular}{c}\texttt{CREATE}\\\texttt{VIEW}\end{tabular}]
{ \langle\mathit{db}, U, \mathit{sec}, T,V,c \rangle \auth^{\mathit{appr}} \langle u, \mathtt{CREATE}, v \rangle
}
{
\hfill u,u' \in U \hfill \quad 
\hfill v \in {\cal VIEW}_{D} \hfill \quad 
\hfill g = \langle \mathit{op}, u, \langle \mathtt{CREATE\ VIEW}\rangle, u' \rangle \hfill \\
\hfill g \in \mathit{sec}  \hfill \quad
\hfill \langle\mathit{db}, U, \mathit{sec}, T,V,c \rangle \auth^{\mathit{appr}} g \hfill \quad 
\hfill \mathit{op} \in \{\oplus,\oplus^{*}\}\hfill 
}
$
 &
$
\infer[\begin{tabular}{c}\texttt{CREATE}\\\texttt{TRIGGER}\end{tabular}]
{ \langle\mathit{db}, U, \mathit{sec}, T,V,c \rangle \auth^{\mathit{appr}} \langle u, \mathtt{CREATE}, t \rangle
}
{
\hfill u,u' \in U \hfill \quad 
\hfill t = \langle \mathit{id},\mathit{ow},  \mathit{ev}, R, \phi, \mathit{stmt}, \mathit{m}\rangle \hfill \\
\hfill t \in {\cal TRIGGER}_{D} \hfill \\ 
\hfill g = \langle \mathit{op}, u, \langle \mathtt{CREATE TRIGGER}, R\rangle, u' \rangle \hfill \\
\hfill g \in \mathit{sec}  \hfill \quad
\hfill \langle\mathit{db}, U, \mathit{sec}, T,V,c \rangle \auth^{\mathit{appr}} g \hfill \quad 
\hfill \mathit{op} \in \{\oplus,\oplus^{*}\}\hfill 
}
$
\\
\\
$
\infer[\begin{tabular}{c}\texttt{INSERT}\\\texttt{DELETE}\\\text{admin}\end{tabular}]
{ \langle\mathit{db}, U, \mathit{sec}, T,V,c \rangle \auth^{\mathit{appr}} \langle \mathit{admin}, \mathit{op}', R, \overline{t} \rangle
}
{\
\hfill R \in D \hfill \quad 
\hfill \overline{t} \in \mathbf{dom}^{|R|}  \hfill \quad
\hfill \mathit{op}' \in \{\mathtt{INSERT}, \mathtt{DELETE}\} \hfill
}
$
 &
$
\infer[\begin{tabular}{c}\texttt{CREATE}\\\texttt{VIEW}\\\text{admin}\end{tabular}]
{ \langle\mathit{db}, U, \mathit{sec}, T,V,c \rangle \auth^{\mathit{appr}} \langle \mathit{admin}, \mathtt{CREATE}, v \rangle
}
{
\hfill v \in {\cal VIEW}_{D} \hfill
}
$
\\
\\
$
\infer[\begin{tabular}{c}\texttt{CREATE}\\\texttt{TRIGGER}\\\text{admin}\end{tabular}]
{ \langle\mathit{db}, U, \mathit{sec}, T,V,c \rangle \auth^{\mathit{appr}} \langle \mathit{admin}, \mathtt{CREATE}, t \rangle
}
{
\hfill t \in {\cal TRIGGER}_{D} \hfill 
}
$
 &
$
\infer[\mathtt{REVOKE}]
{ \langle\mathit{db}, U, \mathit{sec}, T,V,c \rangle \auth^{\mathit{appr}} \langle \ominus, u, \mathit{priv}, u' \rangle
}
{
\hfill u,u' \in U \hfill \quad 
\hfill \mathit{priv} \in {\cal PRIV}_{D} \hfill \\
\hfill s = \langle\mathit{db}, U, \mathit{sec}, T,V,c  \rangle \hfill \quad 
\hfill s' = \langle\mathit{db}, U, \mathit{sec}', T,V,c  \rangle \hfill \\
\hfill s' = \mathit{applyRev}(s,\langle \ominus, u, p, u' \rangle) \hfill \\
\hfill \forall g \in \mathit{sec}'.\, s' \auth^{\mathit{appr}} g \hfill
}
$
\\
\\
$
\infer[\mathtt{GRANT}\text{-1}]
{ \langle\mathit{db}, U, \mathit{sec}, T,V,c \rangle \auth^{\mathit{appr}} \langle \mathit{op}, u, \mathit{priv}, u' \rangle
}
{
\hfill u,u', u'' \in U \hfill \quad 
\hfill \mathit{op} \in \{\oplus,\oplus^{*}\} \hfill \quad 
\hfill \mathit{priv} \in {\cal PRIV}_{D} \hfill \\
\hfill g = \langle \oplus^{*}, u', \mathit{priv}, u''\rangle  \hfill \quad
\hfill g \in \mathit{sec} \hfill \quad
\hfill  \langle\mathit{db}, U, \mathit{sec}, T,V,c \rangle \auth^{\mathit{appr}} g \hfill
}
$
 &
$
\infer[\mathtt{GRANT}\text{-2}]
{ \langle\mathit{db}, U, \mathit{sec}, T,V,c \rangle \auth^{\mathit{appr}} \langle \mathit{op}, u, \mathit{priv}, \mathit{admin} \rangle
}
{
\hfill u \in U \hfill \quad 
\hfill \mathit{op} \in \{\oplus,\oplus^{*}\} \hfill \\ 
\hfill \mathit{priv} \in {\cal PRIV}_{D} \setminus {\cal PRIV}_{D}^{\mathtt{SELECT}, {\cal VIEW}^{\mathit{owner}_{D}}} \hfill 
}
$

\\
\\
$
\infer[\mathtt{GRANT}\text{-3}]
{ \langle\mathit{db}, U, \mathit{sec}, T,V,c \rangle \auth^{\mathit{appr}} \langle \mathit{op}, u, \mathit{priv}, \mathit{owner} \rangle
}
{
\hfill u, \mathit{owner} \in U \hfill \quad 
\hfill \mathit{op} \in \{\oplus,\oplus^{*}\} \hfill \quad 
\hfill \mathit{priv} = \langle \mathtt{SELECT}, v\rangle \hfill \\
\hfill v = \langle \mathit{id}, \mathit{owner}, q, O \rangle \hfill \quad
\hfill v \in V \hfill \\
\hfill T' = \mathit{aT}( \langle\mathit{db}, U, \mathit{sec}, T,V,c \rangle, \oplus^{*}, \mathit{owner}) \hfill \\
\hfill V' = \mathit{aV}( \langle\mathit{db}, U, \mathit{sec}, T,V,c \rangle, \oplus^{*}, \mathit{owner}) \hfill \\
\hfill \mathit{apprDet}(T',V', q) = \top \hfill \\
}$
&
$
\infer[\mathtt{GRANT}\text{-4}]
{ \langle\mathit{db}, U, \mathit{sec}, T,V,c \rangle \auth^{\mathit{appr}} \langle \mathit{op}, u, \mathit{priv}, \mathit{admin} \rangle
}
{
\hfill u, \mathit{owner} \in U \hfill \quad 
\hfill \mathit{op} \in \{\oplus,\oplus^{*}\} \hfill \quad 
\hfill \mathit{priv} = \langle \mathtt{SELECT}, v\rangle \hfill \\
\hfill v = \langle \mathit{id}, \mathit{owner}, q, O \rangle \hfill \quad
\hfill v \in V \hfill \quad
\hfill \mathit{owner} \neq \mathit{admin} \hfill \\
\hfill T' = \mathit{aT}( \langle\mathit{db}, U, \mathit{sec}, T,V,c \rangle, \oplus,\mathit{owner} ) \hfill \\
\hfill V' = \mathit{aV}( \langle\mathit{db}, U, \mathit{sec}, T,V,c \rangle, \oplus, \mathit{owner}) \hfill \\
\hfill \mathit{apprDet}(T',V', q) = \top \hfill \\
}$
\\
\\
$
\infer[\mathtt{GRANT}\text{-5}]
{ \langle\mathit{db}, U, \mathit{sec}, T,V,c \rangle \auth^{\mathit{appr}} \langle \mathit{op}, u, \mathit{priv}, \mathit{owner} \rangle
}
{
\hfill u, \mathit{owner} \in U \hfill \quad 
\hfill \mathit{op} \in \{\oplus,\oplus^{*}\} \hfill \quad
\hfill v \in V \hfill \\ 
\hfill \mathit{priv} = \langle \mathtt{SELECT}, v\rangle \hfill \quad
\hfill v = \langle \mathit{id}, \mathit{owner}, q, A \rangle \hfill \quad
}$
&
$
\infer[\texttt{ADD USER}]
{ \langle\mathit{db}, U, \mathit{sec}, T,V,c \rangle \auth^{\mathit{appr}} \langle u', \mathtt{ADD\_USER}, u \rangle
}
{
\hfill u \in {\cal U} \hfill \quad \hfill u' = \admin \hfill
}
$
\\
\\
\multicolumn{2}{c}{
$
\infer[\begin{tabular}{c}\texttt{EXECUTE}\\\texttt{TRIGGER}-1\end{tabular}]
{ \langle\mathit{db}, U, \mathit{sec}, T,V,c \rangle \auth^{\mathit{appr}} t
}
{
\hfill t = \langle \mathit{id},\mathit{ow},  \mathit{ev}, R, \phi, \mathit{stmt}, O\rangle \hfill \quad
\hfill t \in T \hfill \\
\hfill \langle\mathit{db}, U, \mathit{sec}, T,V,c \rangle \auth^{\mathit{appr}} \mathit{getAction}(\mathit{stmt},\mathit{ow},\mathit{tpl}(c)) \hfill \\
\hfill [\phi[\overline{x}^{|R|} \mapsto \mathit{tpl}(c)]]^{\mathit{db}} = \top \hfill  
}
$}
\\
\\

\multicolumn{2}{c}{
$
\infer[\begin{tabular}{c}\texttt{EXECUTE}\\\texttt{TRIGGER}-2\end{tabular}]
{ \langle\mathit{db}, U, \mathit{sec}, T,V,c \rangle \auth^{\mathit{appr}} t
}
{
\hfill t = \langle \mathit{id},\mathit{ow},  \mathit{ev}, R, \phi, \mathit{stmt}, A\rangle \hfill \quad
\hfill t \in T \hfill \\ 
\hfill \langle\mathit{db}, U, \mathit{sec}, T,V,c \rangle \auth^{\mathit{appr}} \mathit{getAction}(\mathit{stmt},\mathit{invoker}(c),\mathit{tpl}(c)) \hfill \\
\hfill \langle\mathit{db}, U, \mathit{sec}, T,V,c \rangle \auth^{\mathit{appr}} \mathit{getAction}(\mathit{stmt},\mathit{ow},\mathit{tpl}(c)) \hfill \\
\hfill [\phi[\overline{x}^{|R|} \mapsto \mathit{tpl}(c)]]^{\mathit{db}} = \top \hfill  
}
$
}
\\
\\
$
\infer[\texttt{SELECT}]
{ \langle\mathit{db}, U, \mathit{sec}, T,V,c \rangle \auth^{\mathit{appr}} \langle u, \mathtt{SELECT}, q \rangle
}
{
\hfill u \in U \quad \hfill q \in \mathit{RC} \hfill
}
$
&
$
\infer[\begin{tabular}{c}\texttt{EXECUTE}\\\texttt{TRIGGER}-3\end{tabular}]
{ \langle\mathit{db}, U, \mathit{sec}, T,V,c \rangle \auth^{\mathit{appr}} t
}
{
\hfill t = \langle \mathit{id},\mathit{ow},  \mathit{ev}, R, \phi, \mathit{stmt}, m\rangle \hfill \\
\hfill t \in T \hfill \quad
\hfill [\phi[\overline{x}^{|R|} \mapsto \mathit{tpl}(c)]]^{\mathit{db}} =\bot \hfill  
}
$
\end{tabular}
}
  \begin{align*}
\mathit{aT}(\langle \mathit{db},U,\mathit{sec}, T, V, c\rangle, \mathit{op}, u) = \{ R \in D\,|\,
	&(u = \admin \wedge  \exists u' \in U, \mathit{op}' \in \{\oplus^{*}, \mathit{op}\}.\\
  	& \quad\quad \langle \mathit{db},U,\mathit{sec}, T, V, c\rangle \auth^{\mathit{appr}} \langle \mathit{op}', u, \langle \mathtt{SELECT}, R\rangle, u' \rangle ) \vee \\
  	&\exists u' \in U, g \in \mathit{sec}, \mathit{op}' \in \{\oplus^{*}, \mathit{op}\}.\, g = \langle \mathit{op}', u, \langle \mathtt{SELECT}, R\rangle, u' \rangle \\
  	& \quad\quad \wedge \langle \mathit{db},U,\mathit{sec}, T, V, c\rangle \auth^{\mathit{appr}} g  \}
  \end{align*}
  
  \begin{align*}
  	\mathit{aV}(\langle \mathit{db},U,\mathit{sec}, T, V, c\rangle, \mathit{op}, u) = \{ V \in V \cap {\cal VIEW}_{D}^{\mathit{owner}}\,|\,
	&(u = \admin \wedge  \exists u' \in U, \mathit{op}' \in \{\oplus^{*}, \mathit{op}\}.\\
  	& \quad\quad \langle \mathit{db},U,\mathit{sec}, T, V, c\rangle \auth^{\mathit{appr}} \langle \mathit{op}', u, \langle \mathtt{SELECT}, V\rangle, u' \rangle ) \vee \\
  	&\exists u' \in U, g \in \mathit{sec}, \mathit{op}' \in \{\oplus^{*}, \mathit{op}\}.\, g = \langle \mathit{op}', u, \langle \mathtt{SELECT}, V\rangle, u' \rangle \\
  	& \quad\quad \wedge \langle \mathit{db},U,\mathit{sec}, T, V, c\rangle \auth^{\mathit{appr}} g  \}
  \end{align*}

\caption{Definition of the $\auth^{\mathit{appr}}$ relation}\label{figure:eop:auth:approx}
\end{figure*}

%% file: ibsecEnforcement.tex
\clearpage
\section{Enforcing Data Confidentiality}\label{app:enforcement:ibsec}
Here, we first formalize the PDP $f_{\mathit{conf}}$.
Afterwards, we prove that it provides the \confidentiality{} property.
Finally, we show that its data complexity is \complexity{}.

Let $M = \langle D,\Gamma \rangle$ be a system configuration.
The \acf{} $f^{u}_{\mathit{conf}}$ is shown in Figure \ref{figure:access:control:function:IB}.
The function is parametrized by the user $u$ against which the \acf{} provides data confidentiality.
The \acf{} $f^{u}_{\mathit{conf}}(s,a)$ models the function  $f_{\mathit{conf}}(s,a,u)$  shown in Figure \ref{figure:algorithms}.
The mapping between the \acf{} $f^{u}_{\mathit{conf}}$ and the pseudo-code shown in Figure \ref{figure:algorithms} is immediate.

The \acf{} $f^{u}_{\mathit{conf}}$ uses a number of auxiliary functions. 
Recall that the function $\mathit{tr}$, defined in Appendix \ref{app:lts}, takes as input an $M$-state $s \in \Omega_{M}$ and returns the definition of the trigger that the system is executing. 
If the system is not executing any trigger, then $\mathit{tr}(s) = \epsilon$. 
Equivalently, $\mathit{tr}(s)$ is the first trigger in the sequence of triggers returned by $\mathit{triggers}(s)$.

The function $\mathit{tDet}$ takes as input a view $v = \langle \mathit{i}, \mathit{o}, \{\overline{x} | \phi\}, m\rangle \\\in {\cal VIEW}_{D}$, a state $s \in \Omega_{M}$, and a system configuration $M = \langle D, \Gamma \rangle$ and returns  as output the smallest set of tables in $D$ that determines $v$, namely the smallest set $T \in \mathbb{P}(D)$ such that $\mathit{apprDet}(T,\emptyset, \phi, s, M)$ holds, where $\mathit{apprDet}$ is defined in Appendix \ref{app:enforcement:eopsec}.
Note that such a set is always unique. 

The function $\mathit{noLeak}$, defined in Figure~\ref{figure:access:control:function:IB}, takes as input a state $s$, an \texttt{INSERT} or \texttt{DELETE} action $a$, and a user $u$ and checks whether the execution of the action $a$ may leak sensitive information through the views that the user $u$ can read, as shown in Example \ref{example:ibsec:2}.
Note that the function $\mathit{noLeak}$ returns $\top$ if there is no leakage of sensitive information and returns $\bot$ if the action $a$ may leak sensitive information through the views the user $u$ can read in the state $s$.
We remark that the function $\mathit{leak}(a,s,u)$ used in the algorithm in Section \ref{sect:enf:alg} returns is defined as $\mathit{leak}(a,s,u) = \mathit{noLeak}(s,a,u)$.

\begin{figure*}
\begin{tabular}{c}
$
\mathit{f}^{u}_{\mathit{conf}}(s, \mathit{act}) = \left\{ 
  \begin{array}{l l}
  \mathit{f}^{u}_{\mathit{conf},\mathtt{S}}(s, \mathit{act}) & \text{if } \mathit{act} = \langle u', \mathtt{SELECT}, q\rangle\\
  \mathit{f}^{u}_{\mathit{conf},\mathtt{I,D}}(s, \mathit{act}) & \text{if } \mathit{act} = \langle u', \mathtt{INSERT}, R, \overline{t}\rangle\\
  \mathit{f}^{u}_{\mathit{conf},\mathtt{I,D}}(s, \mathit{act}) & \text{if } \mathit{act} = \langle u', \mathtt{DELETE}, R, \overline{t}\rangle\\
  \mathit{f}^{u}_{\mathit{conf},\mathtt{G,R}}(s, \mathit{act}) & \text{if } \mathit{act} = \langle \mathit{op}, u'', p, u'\rangle \wedge \mathit{op} \in \{\oplus,  \oplus^{*}\}\\
 \top & \text{if } u = \admin\\
 \top & \text{otherwise}
  \end{array}\right.$
 \\ \\\\
$\mathit{f}^{u}_{\mathit{conf},\mathtt{I,D}}(s, \mathit{act}) = \left\{ 
  \begin{array}{l l}

{\mathit{secure}}(u,{\mathit{getInfo}}(\mathit{act}),s)	\wedge & \text{if } \mathit{act} = \langle u, \mathit{op}, R, \overline{t}\rangle \wedge {\mathit{trigger}}(s) = \epsilon \wedge \mathit{noLeak}(s, \mathit{act}, u) = \top  \\
 \bigwedge_{\gamma \in {\mathit{Dep}}(\mathit{act}, \Gamma)} {\mathit{secure}}(u,{\mathit{getInfoS}}(\gamma, \mathit{act}),s)  & \\
  \quad \wedge	{\mathit{secure}}(u,{\mathit{getInfoV}}(\gamma, \mathit{act}),s) & \\ 
  \\
  \bot & \text{if } \mathit{act} = \langle u, \mathit{op}, R, \overline{t}\rangle \wedge {\mathit{trigger}}(s) = \epsilon \wedge \mathit{noLeak}(s, \mathit{act}, u) =\bot \\
 \\
{\mathit{secure}}(u,{\mathit{getInfo}}(\mathit{act}),s)	\wedge  & \text{if }{\mathit{invoker}}(s) = u \wedge \mathit{trigger}(s) \neq \epsilon \wedge \mathit{noLeak}(s, \mathit{act}, u) = \top\\ 
 \bigwedge_{\gamma \in {\mathit{Dep}}(\mathit{act}, \Gamma)} {\mathit{secure}}(u,{\mathit{getInfoS}}(\gamma, \mathit{act}),s)  &  \\
  \quad \wedge	{\mathit{secure}}(u,{\mathit{getInfoV}}(\gamma, \mathit{act}),s)  &  \\ 
  \\
 \bot & \text{if }{\mathit{invoker}}(s) = u \wedge \mathit{trigger}(s) \neq \epsilon \wedge \mathit{noLeak}(s, \mathit{act}, u) = \bot\\ 
   &  \\
 \\ 
	\top & \text{otherwise}
  \end{array}\right.$
\\\\\\
$\mathit{f}^{u}_{\mathit{conf},\mathtt{S}}(s, \langle u', \mathtt{SELECT},  q\rangle) = \left\{ 
  \begin{array}{l l}
  	{\mathit{secure}}(u,q, s) & \text{if }  u' = u \wedge {\mathit{trigger}}(s) = \epsilon\\
  	
  	{\mathit{secure}}(u,q, s)  & \text{if }  {\mathit{invoker}}(s) = u \wedge {\mathit{trigger}}(s) \neq \epsilon\\ 
  	
  	\top & \text{otherwise}
  	\end{array}\right.$
\\\\\\

$\mathit{f}^{u}_{\mathit{conf},\mathtt{G}}(s, \langle \mathit{op}, u'', p, u'\rangle) = \left\{ 
  \begin{array}{l l}
  	\bot & \text{if } u'' = u \wedge  u' = u \wedge \mathit{trigger}(s) = \epsilon \wedge \mathit{op} \in \{\oplus,\oplus^{*}\} \wedge p = \langle \mathtt{SELECT}, O\rangle \wedge \\
  	& \quad \langle \oplus, \texttt{SELECT}, O \rangle \not\in \mathit{permissions}(s,u)\\
  	
  	\bot & \text{if } u'' = u \wedge \mathit{invoker}(s) = u \wedge \mathit{trigger}(s) \neq \epsilon \wedge \mathit{op} \in \{\oplus,\oplus^{*}\} \wedge p = \langle \mathtt{SELECT}, O\rangle \wedge \\
  	& \quad  \langle \oplus, \texttt{SELECT}, O \rangle \not\in \mathit{permissions}(s,u)\\
	  	
  	\top & \text{otherwise}
  	\end{array}\right.$
\\\\\\

$\mathit{noLeak}(s, \langle u', \mathit{op}, R, \overline{t}\rangle, u) = \left\{ 
  \begin{array}{l l}
  \top &  \text{if } u' = u \wedge \mathit{trigger}(s) = \epsilon \wedge \forall v \in {\cal VIEW}_{D}.\, ((\langle \oplus, \mathtt{SELECT}, v \rangle \in  \mathit{permissions}(s,u) \wedge \\ 
  & \quad R \in \mathit{tDet}(v,s,M)) \Rightarrow  (\forall o \in \mathit{tDet}(v,s,M).\, \langle \oplus, \mathtt{SELECT}, o \rangle \in \mathit{permissions}(s,u)))\\
  \\
  
  \top &  \text{if } \mathit{invoker}(s) = u \wedge \mathit{trigger}(s) \neq \epsilon \wedge \forall v \in {\cal VIEW}_{D}.\, ((\langle \oplus, \mathtt{SELECT}, v \rangle \in  \mathit{permissions}(s,u) \wedge \\
  & \quad R \in \mathit{tDet}(v,s,M)) \Rightarrow  (\forall o \in \mathit{tDet}(v,s,M).\, \langle \oplus, \mathtt{SELECT}, o \rangle \in  \mathit{permissions}(s,u)))\\  
  \\

  \bot & \text{otherwise}
  \end{array}\right.$
\end{tabular}

\caption{Access control function $f^u_{\mathit{conf}}$}\label{figure:access:control:function:IB}
\end{figure*}

We now define the $\mathit{Dep}$,  $\mathit{getInfoS}$,  $\mathit{getInfoV}$, and $\mathit{getInfo}$ functions.
The function $\mathit{Dep}$ is  as follows. $\mathit{Dep}(\langle u, \mathtt{INSERT}, R,\\ \overline{t}\rangle, \Gamma)$ returns the set containing all the formulae in $\Gamma$ of the form $\forall \overline{x}, \overline{y}, \overline{y}', \overline{z}, \overline{z}'.\, (R(\overline{x}, \overline{y},  \overline{z}) \wedge R(\overline{x}, \overline{y}',  \overline{z}') )\Rightarrow \overline{y} = \overline{y}'$ or $\forall \overline{x}, \overline{z}.\, R(\overline{x}, \overline{z}) \Rightarrow \exists \overline{w}.\, S(\overline{x}, \overline{w})$, whereas   $\mathit{Dep}(\langle u, \mathtt{DELETE}, R, \overline{t}\rangle, \\\Gamma)$ returns the set containing all the formulae in $\Gamma$ of the form $\forall \overline{x}, \overline{z}.\, S(\overline{x}, \overline{z}) \Rightarrow \exists \overline{w}.\, R(\overline{x}, \overline{w})$.

The function $\mathit{getInfoS}$ is defined as follows:
\begin{compactitem}
\item $\mathit{getInfoS}(\langle u, \mathtt{INSERT}, R, (\overline{v}, \overline{w}, \overline{q})\rangle, \phi^{R}_{\mathit{funct}})$ is the formula $\neg \exists \overline{y},\overline{z}. R(\overline{v}, \overline{y}, \overline{z}) \wedge \overline{y} \neq \overline{w}$, where $\phi^{R}_{\mathit{funct}}$ is a formula of the form $\forall \overline{x}, \overline{y}, \overline{y}', \overline{z}, \overline{z}'.\, (R(\overline{x}, \overline{y},  \overline{z}) \wedge R(\overline{x}, \overline{y}',  \overline{z}') )\Rightarrow \overline{y} = \overline{y}'$.
\item $\mathit{getInfoS}(\langle u, \mathtt{INSERT}, R, (\overline{v}, \overline{w})\rangle, \phi^{R,S}_{\mathit{incl}})$ is the formula $\exists \overline{y}.\, \\ S(\overline{v}, \overline{y})$, where $\phi^{R,S}_{\mathit{incl}}$ is a formula of the form $\forall \overline{x}, \overline{z}.\, R(\overline{x}, \overline{z}) \\ \Rightarrow \exists \overline{w}.\, S(\overline{x}, \overline{w})$.
\item $\mathit{getInfoS}(\langle u, \mathtt{DELETE}, R, (\overline{v}, \overline{w})\rangle, \phi^{S,R}_{\mathit{incl}})$ is the formula $\forall \overline{x}, \\ \overline{z}.\,( S(\overline{x}, \overline{z}) \Rightarrow \overline{x} \neq \overline{v}) \vee \exists \overline{y}.\, (R(\overline{v}, \overline{y}) \wedge \overline{y} \neq \overline{w})$, where $\phi^{S,R}_{\mathit{incl}}$ is a formula of the form $\forall \overline{x}, \overline{z}.\, S(\overline{x}, \overline{z}) \Rightarrow \exists \overline{w}.\, R(\overline{x}, \overline{w})$.
\item $\mathit{getInfoS}(\mathit{act}, \phi) = \top$ otherwise.
\end{compactitem}

The function $\mathit{getInfoV}$ is defined as follows:
\begin{compactitem}
\item $\mathit{getInfoV}(\langle u, \mathtt{INSERT}, R, (\overline{v}, \overline{w}, \overline{q})\rangle, \phi^{R}_{\mathit{funct}})$ is the formula $\exists \overline{y},\overline{z}. R(\overline{v}, \overline{y}, \overline{z}) \wedge \overline{y} \neq \overline{w}$, where $\phi^{R}_{\mathit{funct}}$ is a formula of the form $\forall \overline{x}, \overline{y}, \overline{y}', \overline{z}, \overline{z}'.\,  (R(\overline{x}, \overline{y},  \overline{z}) \wedge R(\overline{x}, \overline{y}',  \overline{z}') )\Rightarrow \overline{y} = \overline{y}'$.
\item $\mathit{getInfoV}(\langle u, \mathtt{INSERT}, R, (\overline{v}, \overline{w})\rangle, \phi^{R,S}_{\mathit{incl}})$ is the formula $\forall \overline{x}, \\ \overline{y}.\, S(\overline{x}, \overline{y}) \Rightarrow \overline{x} \neq \overline{v}$, where $\phi^{R,S}_{\mathit{incl}}$ is a formula of the form $\forall \overline{x}, \overline{z}.\, R(\overline{x}, \overline{z}) \Rightarrow \exists \overline{w}.\, S(\overline{x}, \overline{w})$.
\item $\mathit{getInfoV}(\langle u, \mathtt{DELETE}, R, (\overline{v}, \overline{w})\rangle, \phi^{S,R}_{\mathit{incl}})$ is the formula $\exists \overline{z}.\, \\ S(\overline{v},\overline{z}) \wedge \forall \overline{y}.\,(R(\overline{v},\overline{y}) \Rightarrow \overline{y} = \overline{w} )$, where $\phi^{S,R}_{\mathit{incl}}$ is a formula of the form $\forall \overline{x}, \overline{z}.\, S(\overline{x}, \overline{z}) \Rightarrow \exists \overline{w}.\, R(\overline{x}, \overline{w})$.
\item $\mathit{getInfoV}(\mathit{act}, \phi) = \top$ otherwise.
\end{compactitem}

The function $\mathit{getInfo}$ is as follows:
\[
\mathit{getInfo}(\langle u, \mathit{op}, R,\overline{t}\rangle) = \left\{
\begin{array}{l l}
\neg R(\overline{t}) & \text{if } \mathit{op} = \texttt{INSERT}\\
 R(\overline{t}) & \text{if } \mathit{op} = \texttt{DELETE}\\
\end{array}\right.
\]

In \S\ref{sect:secure} we describe the $\mathit{secure}$ function and we show that it is a sound, under-approximation of the concept of secure judgments.
Afterwards, in \S\ref{sect:data:conf:proofs} we prove that $f_{\mathit{conf}}^{u}$ provides \confidentiality{} with respect to the user $u$.
Finally, in \S\ref{sect:data:conf:complexity:proofs} we prove that the data complexity of $f_{\mathit{conf}}^{u}$ is \complexity.
In the rest of the paper, instead of writing $\mathit{secure}_{P,\cong_{P,u}}$ we simply write $\mathit{secure}_{P,u}$ and we omit the reference to the indistinguishability relation $\cong_{P,u}$ defined in Appendix~\ref{app:indistinguishability}.

\subsection{Checking a judgment's security}\label{sect:secure}

We still have to define the $\mathit{secure} : {\cal U} \times \mathit{RC}_{\mathit{bool}} \times \Omega_{M} \rightarrow \{\top,\perp\}$ function that determines a given judgment's security.
In more detail, the $\mathit{secure}$ function is as follows: 
\[
\mathit{secure}(u,\phi, s) = \left\{
\begin{array}{l l}
\top & \text{if } [\phi^{\mathit{rw}}_{s,u}]^{s.\mathit{db}} = \bot \\
\bot & \text{otherwise}
\end{array}\right.
\]

In the following, we assume that both the formula $\phi$ and the set of views $V$ in the state $s$ contain just views with owner's privileges.
This is without loss of generality.
Indeed, views with activator's privileges are just syntactic sugar, they do not disclose additional information to a user other than what he is already authorized to read because they are executed under the activator's privileges.
If $\phi$ and $s$ contain views with activator's privileges, we can compute another formula $\phi'$ and a state $s'$ without views with activator's privileges as follows.
We replace, in the formula $\phi$, the predicates of the form $V(\overline{x})$, where $V$ is a view with activator's privileges,  with $V$'s definition, and we repeat this process until the resulting formula $\phi'$ no longer contains views with activator's privileges.
Similarly, the set $V'$ is obtained from $V$ by (1) removing all views with activator's privileges, and (2) for each view $v \in V$ with owner's privileges, replacing the predicates of the form $V(\overline{x})$ in $v$'s definition, where $V$ is a view with activator's privileges,  with $V$'s definition until $v$'s definition no longer contains views with activator's privileges.
The security policy $\mathit{sec}'$ is also obtained from $\mathit{sec}$ by removing all references to views with activator's privileges.
Finally, $\mathit{secure}(u, \phi, \langle \mathit{db}, U,\mathit{sec}, T,V,c\rangle)$ is just $\mathit{secure}(u, \phi', \langle \mathit{db}, U,\mathit{sec}', T,V',c\rangle)$.

Before defining the $\phi^{\top}_{s,u}$ and $\phi^{\bot}_{s,u}$ rewritings, we define query containment.
Let $M = \langle D,\Gamma \rangle$ be a system configuration.
Given two formulae $\phi(\overline{x})$ and $\psi(\overline{y})$, we  write $\phi \subseteq_{M} \psi$ to denote that $\phi$ is contained in $\psi$, i.e., $\forall \mathit{d} \in \Omega_{D}^{\Gamma}.\, [\{\overline{x} | \phi\}]^{\mathit{d}} \subseteq [\{\overline{y} | \psi\}]^{\mathit{d}}$.
Determining whether $\phi \subseteq_{M} \psi$ holds is undecidable for $\mathit{RC}$~\cite{abiteboul1995foundations}.
Hence, we develop a sound, under-approximation of query containment.
Figure \ref{figure:containment} describes the rules defining our under-approximation.
For simplicity's sake, the rules are defined only for relational calculus formulae that do not use views.
To check whether  $\phi \subseteq_{M} \psi$ holds for two formulae $\phi$ and $\psi$  that use views, we first compute the formulae $\phi'$ and $\psi'$, obtained by replacing views' identifiers with their definitions, and then we check whether $\phi' \subseteq_{M} \psi'$ using the rules in Figure \ref{figure:containment}.
This preserves containment since $\phi$ and $\psi$ are semantically equivalent to $\phi'$ and $\psi'$.
Both in the rules and in the proof of \thref{theorem:containment:sound}, we assume that there is a total ordering $\preceq_{\mathit{var}}$ over the set of all possible variable identifiers.
This ensures that, given a formula $\phi$, there is a unique non-boolean query $\{\overline{x} \,|\, \phi\}$ associated to it, where the variables in $\overline{x}$ are those in $\mathit{free}(\phi)$ ordered according to  $\preceq_{\mathit{var}}$.
\thref{theorem:containment:sound} proves that the rules in Figure \ref{figure:containment} are a sound, under-approximation of query containment.

\begin{lemma}\thlabel{theorem:containment:sound}
Let $M = \langle D,\Gamma \rangle$ be a system configuration, and  $\phi(\overline{x})$ and $\psi(\overline{y})$ be two formulae.
If $\phi \subseteq_{M} \psi$, according to the rules in Figure \ref{figure:containment}, then  $\forall \mathit{d} \in \Omega_{D}^{\Gamma}.\, [\{\overline{x} | \phi\}]^{\mathit{d}} \subseteq [\{\overline{y} | \psi\}]^{\mathit{d}}$, where $\overline{x}$ (respectively $\overline{y}$) is the tuple defined by the variables in $\mathit{free}(\phi)$ (respectively $\mathit{free}(\psi)$) ordered according to $\preceq_{\mathit{var}}$.
\end{lemma}
\begin{proof}
$\phi \subseteq_{M} \psi$  iff there is a finite derivation that ends in $\phi \subseteq_{M} \psi$ created using the rules in Figure~\ref{figure:containment}.
We prove our claim by structural induction on the derivation's length.

\smallskip
\noindent
{\bf Base Case} 
Assume that the derivation has length 1.
There are four cases depending on the rule used to derive $\phi \subseteq_{M} \psi$:
\begin{compactenum}
\item Rule \emph{And}.
From the rule's definition, it follows that $\mathit{free}(\phi) = \mathit{free}(\phi\wedge \psi) = \overline{x}$.
Let $d \in \Omega_{D}^{\Gamma}$ and $\overline{t} \in [\{ \overline{x}\,|\,\phi \wedge \psi \}]^{d}$.
From $\overline{t} \in [\{ \overline{x}\,|\,\phi \wedge \psi \}]^{d}$ and the definition of non-boolean query, it follows that $[(\phi \wedge \psi )[\overline{x} \mapsto \overline{t}]]^{d} = \top$.
From this and the relational calculus semantics, it follows that $[\phi [\overline{x} \mapsto \overline{t}]]^{d} = \top$.
From this and the definition of non-boolean query,  $\overline{t} \in [\{ \overline{x}\,|\,\phi \}]^{d}$.
Therefore, $[\{\overline{x} | \phi \wedge \psi \}]^{\mathit{d}} \subseteq [\{\overline{x} | \phi\}]^d$.

\item Rule \emph{Or}.
From the rule's definition, it follows that $\mathit{free}(\phi) = \mathit{free}(\phi \vee \psi) = \overline{x}$.
Let $d \in \Omega_{D}^{\Gamma}$ and $\overline{t} \in [\{ \overline{x}\,|\,\phi  \}]^{d}$.
From $\overline{t} \in [\{ \overline{x}\,|\,\phi  \}]^{d}$ and the definition of non-boolean query, it follows that $[\phi [\overline{x} \mapsto \overline{t}]]^{d} = \top$.
From this and the relational calculus semantics, it follows that $[(\phi \vee \psi) [\overline{x} \mapsto \overline{t}]]^{d} = \top$.
From this and the definition of non-boolean query,  $\overline{t} \in [\{ \overline{x}\,|\,\phi \vee \psi \}]^{d}$.
Therefore, $[\{\overline{x} | \phi \}]^{\mathit{d}} \subseteq [\{\overline{x} | \phi \vee \psi \}]^d$.

\item Rule \emph{Identity}.
From the rule's definition, it follows that $\mathit{free}(\phi) = \overline{x}$, $\mathit{free}(\psi) = \overline{y}$, and $\phi[\overline{x} \mapsto \overline{y}] = \psi$.
Let $d \in \Omega_{D}^{\Gamma}$ and $\overline{t} \in [\{ \overline{x}\,|\,\phi  \}]^{d}$.
From $\overline{t} \in [\{ \overline{x}\,|\,\phi  \}]^{d}$ and the definition of non-boolean query, it follows that $[\phi [\overline{x} \mapsto \overline{t}]]^{d} = \top$.
From this and $\phi[\overline{x} \mapsto \overline{y}] = \psi$, it follows that $[\psi [\overline{y} \mapsto \overline{t}]]^{d} = \top$.
From this and the definition of non-boolean query,  $\overline{t} \in [\{ \overline{y}\,|\,\psi  \}]^{d}$.
Therefore, $[\{\overline{x} | \phi \}]^{\mathit{d}} \subseteq [\{\overline{y} | \psi \}]^d$.

\item Rule \emph{Inclusion Dependency}.
From the rule's definition, it follows that $\gamma:=\forall \overline{x}, \overline{z}.\,( R(\overline{x}, \overline{z}) \Rightarrow \exists \overline{w}.\, S(\overline{x}, \overline{w}) )$ is in $\Gamma$.
Let $d \in \Omega_{D}^{\Gamma}$ and $\overline{t} \in [\{ \overline{x}\,|\,\exists \overline{z}.\, R(\overline{x}, \overline{z})  \}]^{d}$.
From $\overline{t} \in [\{ \overline{x}\,|\,\exists \overline{z}.\, R(\overline{x}, \overline{z})  \}]^{d}$ and the definition of non-boolean query, it follows that $[\exists \overline{z}.\, R(\overline{t}, \overline{z})]^{d} = \top$.
Therefore, there is a tuple $(\overline{t}, \overline{w}) \in d(R)$.
From this and $\gamma \in \Gamma$, it follows that there is a tuple $(\overline{t}, \overline{w}') \in d(S)$.
From this, it follows that $[\exists \overline{w}.\, S(\overline{t}, \overline{w})]^{d} = \top$.
From this and the definition of non-boolean query, it follows that $\overline{t} \in [\{ \overline{x}\,|\,\exists \overline{w}.\, S(\overline{x}, \overline{w})  \}]^{d}$.
Therefore, it follows that $[\{ \overline{x}\,|\,\exists \overline{z}.\, \\ R(\overline{x}, \overline{z})  \}]^{d}  \subseteq [\{ \overline{x}\,|\,\exists \overline{w}.\, S(\overline{x}, \overline{w})  \}]^{d}$ holds.

\end{compactenum}
This completes the proof for the base case.

\smallskip
\noindent
{\bf Induction Step} 
Assume now that the claim holds for all derivations of length less than that of $\phi \subseteq_{M} \psi$.
We now prove that it holds also for $\phi \subseteq_{M} \psi$.
There is just one case, namely $\phi \subseteq_{M} \psi$ is of the form $\exists x_i.\, \alpha \subseteq_{M} \exists y_i.\, \beta$ and it is obtained by applying the rule \emph{Projection} to $\alpha \subseteq_{M} \beta$.
From the rule, it follows that $\alpha \subseteq_{M} \beta$ holds.
Let $1 \leq u \leq n$ and $\overline{t}'$ (respectively $\overline{x}'$ and $\overline{y}'$) be the tuple obtained from $\overline{t}$ (respectively $\overline{x}$ and $\overline{y}$) by dropping the $i$-th value (respectively variable).
We now prove that $[\{\overline{x}' | \exists x_{i}.\, \alpha \}]^{\mathit{d}} \subseteq [\{\overline{y}' | \exists y_{i}.\, \beta \}]^{\mathit{d}}$.
Assume, for contradiction's sake, that this is not the case, namely there is a tuple $\overline{v}$ such that $\overline{v} \in [\{\overline{x}' | \exists x_{i}.\, \alpha \}]^{\mathit{d}}$ but $\overline{v} \not\in  [\{\overline{y}' | \exists y_{i}.\, \beta \}]^{\mathit{d}}$.
From $\overline{v} \in [\{\overline{x}' | \exists x_{i}.\, \alpha \}]^{\mathit{d}}$ and the relational calculus semantics, it follows that there is a tuple $\overline{v}_1$, obtained by adding a value to $\overline{v}$ in the $i$-th position, such that  $\overline{v}_1 \in [\{\overline{x} |  \alpha \}]^{\mathit{d}}$.
From this, $\alpha \subseteq_{M} \beta$, and the induction hypothesis, it follows that  $\overline{v}_1 \in [\{\overline{y} |  \beta \}]^{\mathit{d}}$.
From this and the relational calculus semantics, it follows that $\overline{v} \in [\{\overline{y}' |  \exists y_{i}.\, \beta \}]^{\mathit{d}}$.
This contradicts the fact that $\overline{v} \not\in  [\{\overline{y}' | \exists y_{i}.\, \beta \}]^{\mathit{d}}$.

This completes the proof.
\end{proof}

\begin{figure}[!hbtp]
\centering
\scalebox{0.80}{
\begin{tabular}{c c}

$
\infer[\text{And}]
{ \phi  \wedge \psi \subseteq_{M}  \phi
}
{
\hfill \mathit{free}(\phi \wedge \psi) = \mathit{free}(\phi) \hfill \\
\hfill M = \langle D, \Gamma\rangle \hfill \quad
\hfill \mathit{free}(\phi) \neq \emptyset \hfill
}$
&
$
\infer[\text{Or}]
{ \phi \subseteq_{M}  \phi \vee \psi
}
{
\hfill \mathit{free}(\phi) = \mathit{free}(\phi \vee \psi) \hfill \\
\hfill M = \langle D, \Gamma\rangle \hfill \quad
\hfill \mathit{free}(\phi) \neq \emptyset \hfill
}$
\\\\

$
\infer[\text{Projection}]
{   \exists x_{i}. \phi \subseteq_{M} \exists y_{i}. \phi
}
{
\hfill M = \langle D, \Gamma\rangle \hfill \enspace 
\hfill n > 1 \hfill \\ 
\hfill \mathit{free}(\phi) = \{ x_{1}, \ldots, x_{n} \} \hfill \\
\hfill \mathit{free}(\psi) = \{ y_{1}, \ldots, y_{n} \} \hfill \\
\hfill 1 \leq i \leq n \hfill \enspace
\hfill  \phi  \subseteq_{M}  \psi \hfill \\
}
$
&

$
\infer[\text{Identity}]
{ \phi \subseteq_{M}  \psi
}
{
\hfill M = \langle D, \Gamma\rangle \hfill \enspace
\hfill n > 0 \hfill  \hfill \\
\hfill \mathit{free}(\phi) = \{x_{1}, \ldots, x_{n}\} \hfill \\
\hfill \mathit{free}(\psi) = \{y_{1}, \ldots, y_{n}\} \hfill \\
\hfill \phi[x_{1} \mapsto y_{1}, \ldots, x_{n} \mapsto y_{n}] = \psi \hfill
}$ 
\\\\
\multicolumn{2}{c}{
$
\infer[\text{\begin{tabular}{c}Inclusion\\ Dependency\end{tabular}}]
{  \exists \overline{z}.\, R(\overline{x}, \overline{z})\subseteq_{M}  \exists \overline{w}.\, S(\overline{x}, \overline{w})  
}
{
\hfill M = \langle D, \Gamma\rangle \hfill \enspace
\hfill |\overline{x}| > 0 \hfill \\
\hfill \forall \overline{x}, \overline{z}.\,( R(\overline{x}, \overline{z}) \Rightarrow \exists \overline{w}.\, S(\overline{x}, \overline{w}) ) \in \Gamma \hfill
}$
}
\end{tabular}
}
\caption{Containment rules}\label{figure:containment}
\end{figure}

Given a table or a view $O$ and a sequence of distinct integers $\overline{i} := (i_{1}, \ldots, i_{n})$ such that $1 \leq i_{j} \leq |O|$ for all $1 \leq j \leq n$, where $0 \leq n < |O|$, the \emph{$\overline{i}$-projection of $O$}, denoted by $O_{\overline{i}}$, is  the formula $\exists x_{i_{1}}, \ldots,  x_{i_{n}}.\, O(x_{1}, \ldots, x_{|O|})$.
Given a database schema $D$ and a set of views $V$ defined over $D$, we denote by $\mathit{extVocabulary}(D,V)$ the extended vocabulary obtained by defining all possible projections of tables in $D$ and views in $V$, i.e., for each $O \in D \cup V$, we define a predicate $O_{\overline{i}}$ for each projection $\exists x_{i_{1}}, \ldots,  x_{i_{n}}.\, O(x_{1}, \ldots, x_{|O|})$ of $O$.
Furthermore, given a relational calculus formula $\phi$ over $D$, we denote by $\mathit{extVoc}_{V,D}(\phi)$ the formula obtained by replacing all sub-formulae of the form $\exists \overline{x}. R(\overline{x}, \overline{y})$ with the predicates in $\mathit{extVocabulary}(D,V)$ representing the corresponding projections $R_{\overline{i}}$.
Finally, we denote by $\mathit{inline}_{D,V}(\phi)$, where $\phi$ is a relational calculus formula over  $\mathit{extVocabulary}(D,V)$, the formula $\phi'$ obtained by replacing all predicates associated with projections with the corresponding formulae.
Let $S$ be predicate in $\mathit{extVocabulary}(D,V)$ and $s$ be an $M$-state.
We denote by $S^{\top}_{s}$  the set of all projections of tables in $D$ and views in $V$ that are contained  in   $S$, i.e., $S^{\top}_{s} := \{ R \in \mathit{extVocabulary}(D,V) \,|\, R(\overline{x}) \subseteq_{M}  S(\overline{y})\}$\footnote{With a slight abuse of notation, we write $R(\overline{x}) \subseteq_{M}  S(\overline{y})$ instead of  $\mathit{inline}_{D,V}(R(\overline{x})) \subseteq_{M}  \mathit{inline}_{D,V}(S(\overline{y}))$.}.
Similarly, we denote by  $S^{\bot}_{s}$ the set of all projections of tables in $D$ and views in $V$ that contains  $S$, i.e., $S^{\bot}_{s} := \{ R \in \mathit{extVocabulary}(D,V) \,|\,  S(\overline{x}) \subseteq_{M}   R(\overline{y})\}$.
Furthermore, we denote by $\mathit{AUTH}_{s,u}$ the set of all tables and views that $u$ is authorized to read in $s$, i.e., $\mathit{AUTH}_{s,u} := \{ v \,|\, \langle \oplus, \mathtt{SELECT}, v \rangle \in \mathit{permissions}(s,u)\}$, and by $\mathit{AUTH}_{s,u}^{*}$ the set of all the projections obtained from tables and views in $\mathit{AUTH}_{s,u}$.

We are now ready to formally define the $\phi^{\top}_{s,u}$ and $\phi^{\bot}_{s,u}$ rewritings.
\begin{definitionInt}
Let $M = \langle D, \Gamma\rangle$ be a system configuration, $s = \langle \mathit{db}, U, \mathit{sec},  T,V,c  \rangle$ be an $M$-state, $u$ be a user, and $\phi$ be a relational calculus sentence over $\mathit{extVocabulary}(D,V)$.

The function $\mathit{bound}$ takes as input a formula $\phi$, a state $s$, a user $u$, a variable identifier $x$, and a value $v$ in $\{\top,\bot\}$.
It is inductively defined as follows:
\begin{compactitem}
\item $\mathit{bound}(R(\overline{y}), s, u, x, v)$, where $R$ is a predicate symbol in $\mathit{extVocabulary}(D,V)$, is $\top$ iff\begin{inparaenum}[(a)]
\item $x$ occurs in $\overline{y}$, and
\item the set $R^{v}_{s,u}$, which is $R^{v}_{s} \cap \mathit{AUTH}_{s,u}^{*}$, is not empty.
\end{inparaenum}
\item $\mathit{bound}(y = z, s, u, x, v)$ is $\top$ iff $x = y$ and $z$ is a constant symbol or $x = z$ and $y$ is a constant symbol.
\item $\mathit{bound}(\top, s, u, x, v) := \bot$.
\item $\mathit{bound}(\bot, s, u, x, v) := \bot$.
\item $\mathit{bound}(\neg \psi, s, u, x, v) := \mathit{bound}(\psi, s, u, x, \neg v)$, where $\psi$ is a relational calculus formula.
\item $\mathit{bound}(\psi \wedge \gamma, s, u, x, v) := \mathit{bound}(\psi, s, u, x, v) \vee \mathit{bound}(\gamma, \\ s, u, x, v)$, where $\psi$ and $\gamma$ are relational calculus formulae.
\item $\mathit{bound}(\psi \vee \gamma, s, u, x, v) := \mathit{bound}(\psi, s, u, x, v) \wedge \mathit{bound}(\gamma, \\ s, u, x, v)$, where $\psi$ and $\gamma$ are relational calculus formulae.
\item $\mathit{bound}(\exists y. \psi, s, u, x, v)$, where $\psi$ is a relational calculus formula, is $\mathit{bound}(\psi, s, u, x, v) \wedge \mathit{bound}(\psi, s, u, y, v)$ if $x \neq y$, and $\mathit{bound}(\exists y. \psi,  s, u, x, v):=\bot$ otherwise.
\item $\mathit{bound}(\forall y. \psi, s, u, x, v)$, where $\psi$ is a relational calculus formula, is $\mathit{bound}(\psi, s, u, x, v) \wedge \mathit{bound}(\psi, s, u, y, v)$ if $x \neq y$, and $\mathit{bound}(\forall y. \psi,  s, u, x, v):=\bot$ otherwise.
\end{compactitem}

The formula $\phi^{\mathit{\top}}_{s,u}$ is inductively defined as follows:
\begin{compactitem}
\item $R(\overline{x})^{\mathit{\top}}_{s,u} := \bigvee_{S \in R^{\top}_{s,u}} S(\overline{x})$, where $R$ is a predicate symbol in $\mathit{extVocabulary}(D,V)$ and $R^{\top}_{s,u} := R^{\top}_{s} \cap \mathit{AUTH}_{s,u}^{*}$.
\item $(x = v)^{\mathit{\top}}_{s,u} :=(x = v)$, where $x$ and $v$ are either variable identifiers or constant symbols. 
\item $(\top)^{\mathit{\top}}_{s,u} := \top$.
\item $(\bot)^{\mathit{\top}}_{s,u} := \bot$.
\item $(\neg \psi)^{\mathit{\top}}_{s,u} := \neg \psi^{\bot}_{s,u}$, where $\psi$ is a relational calculus formula.
\item $(\psi \wedge \gamma)^{\mathit{\top}}_{s,u} := \psi^{\top}_{s,u} \wedge \gamma^{\top}_{s,u}$, where $\psi$ and $\gamma$ are relational calculus formulae.
\item $(\psi \vee \gamma)^{\mathit{\top}}_{s,u} := \psi^{\top}_{s,u} \vee \gamma^{\top}_{s,u}$, where $\psi$ and $\gamma$ are relational calculus formulae.
\item $(\exists x.\, \psi)^{\mathit{\top}}_{s,u}$, where $\psi$ is a relational calculus formula and $x$ is a variable identifier, is $\exists x.\, \psi^{\mathit{\top}}_{s,u}$ if $\mathit{bound}(\psi,s, u, x, \top) \\ = \top$ and $(\exists x.\, \psi)^{\mathit{\top}}_{s,u}:= \bot$ otherwise.
\item $(\forall x.\, \psi)^{\mathit{\top}}_{s,u}$, where $\psi$ is a relational calculus formula and $x$ is a variable identifier, is $\forall x.\, \psi^{\mathit{\top}}_{s,u}$ if $\mathit{bound}(\psi,s, u, x, \top) \\ = \top$ and $(\forall x.\, \psi)^{\mathit{\top}}_{s,u}:= \bot$ otherwise.
\end{compactitem}

The formula $\phi^{\mathit{\bot}}_{s,u}$ is inductively defined as follows:
\begin{compactitem}
\item $R(\overline{x})^{\mathit{\bot}}_{s,u} := \bigwedge_{S \in R^{\bot}_{s,u}} S(\overline{x})$, where $R$ is a predicate symbol in $\mathit{extVocabulary}(D,V)$ and $R^{\bot}_{s,u} := R^{\bot}_{s} \cap \mathit{AUTH}_{s,u}^{*}$.
\item $(x = v)^{\mathit{\bot}}_{s,u} :=(x = v)$, where $x$ and $v$ are either variable identifiers or constant symbols. 
\item $(\top)^{\mathit{\bot}}_{s,u} := \top$.
\item $(\bot)^{\mathit{\bot}}_{s,u} := \bot$.
\item $(\neg \psi)^{\mathit{\bot}}_{s,u} := \neg \psi^{\top}_{s,u}$, where $\psi$ is a relational calculus formula.
\item $(\psi \wedge \gamma)^{\mathit{\bot}}_{s,u} := \psi^{\bot}_{s,u} \wedge \gamma^{\bot}_{s,u}$, where $\psi$ and $\gamma$ are relational calculus formulae.
\item $(\psi \vee \gamma)^{\mathit{\bot}}_{s,u} := \psi^{\bot}_{s,u} \vee \gamma^{\bot}_{s,u}$, where $\psi$ and $\gamma$ are relational calculus formulae.
\item $(\exists x.\, \psi)^{\mathit{\bot}}_{s,u}$, where $\psi$ is a relational calculus formula and $x$ is a variable identifier, is $\exists x.\, \psi^{\mathit{\bot}}_{s,u}$ if $\mathit{bound}(\psi,s, u, x, \bot) \\ = \top$ and  $(\exists x.\, \psi)^{\mathit{\bot}}_{s,u}:= \bot$ otherwise.
\item $(\forall x.\, \psi)^{\mathit{\bot}}_{s,u}$, where $\psi$ is a relational calculus formula and $x$ is a variable identifier, is $\forall x.\, \psi^{\mathit{\bot}}_{s,u}$ if $\mathit{bound}(\psi,s, u,x, \bot) \\ = \top$ and $(\forall x.\, \psi)^{\mathit{\bot}}_{s,u}:= \bot$ otherwise. $\square$
\end{compactitem}
\end{definitionInt}

Finally, we define the formula $\phi^{\mathit{rw}}_{s,u}$ which represents the overall rewritten formula.
\begin{definition}
Let $M = \langle D, \Gamma\rangle$ be a system configuration, $s = \langle \mathit{db}, U, \mathit{sec},  T,V,c  \rangle$ be an $M$-state, $u$ be a user, and $\phi$ be a relational calculus sentence over $D$.
The formula $\phi^{\mathit{rw}}_{s,u}$ is defined as $\mathit{inline}_{V,D}(\neg \psi^{\top}_{s,u} \wedge \psi^{\bot}_{s,u})$, where $\psi := \mathit{extVoc}_{V,D}(\phi)$.
\end{definition}

Let $P = \langle M, f \rangle$ be an \accessControlConfiguration{}, $L$ be the $P$-LTS, $u \in {\cal U}$ be a user, $r \in \mathit{traces}(L)$ be an $L$-run, $\phi \in RC_{\mathit{bool}}$ is a sentence, and $1 \leq i \leq |r|$.
Furthermore, let $s$ be the $i$-th state of $r$.
The judgment $r, i \attMod \phi$ is \emph{data-secure for $M$, $u$, and $s$}, denoted by $\mathit{secure}^{\mathit{data}}_{P,u}(r, i \attMod \phi)$, iff for all $s', s'' \in \llbracket \mathit{pState}(s)\rrbracket_{u,M}^\mathit{data}$, $[\phi]^{s'.\mathit{db}} = [\phi]^{s''.\mathit{db}}$, where $\cong_{u,M}^{data}$ is the data-indistinguishability relation  defined in Appendix \ref{app:indistinguishability} and $\llbracket s\rrbracket_{u,M}^{data} := \{ s' \in \Pi_{M} | s \cong_{u,M}^{\mathit{data}} s'\}$.

We first recall our definitions and notations for assignments.
Let ${\bf dom}$ be the universe and ${\bf var}$ be an infinite countably set of variable identifiers.
An \emph{assignment} $\nu$ is a partial function from ${\bf var}$ to ${\bf dom}$ that maps variables to values in the universe.
Given a formula $\phi$ and an assignment $\nu$, we say that \emph{$\nu$ is well-formed for $\phi$} iff $\nu$ is defined for all variables in $\mathit{free}(\phi)$.
Given an assignment $\nu$ and a sequence of variables $\overline{x}$ such that $\nu$ is defined for each $x \in \overline{x}$, we denote by $\nu(\overline{x})$ the tuple obtained by replacing each occurrence of $x \in \overline{x}$ with $\nu(x)$.
Given an assignment $\nu$, a variable $v \in {\bf var}$, and a value $u \in {\bf dom}$, we denote by $\nu \oplus [v \mapsto u]$ the assignment $\nu'$ obtained as follows: $\nu'(v') = \nu(v')$ for any $v' \neq v$, and $\nu'(v) = u$. 
Finally, given a formula $\phi$ with free variables $\mathit{free}(\phi)$ and an assignment $\nu$, we denote by $\phi \circ \nu$ the formula $\phi'$ obtained by replacing, for each free variable $x \in \mathit{free}(\phi)$ such that $\nu(x)$ is defined, all the free occurrences of $x$ with $\nu(x)$.

\thref{theorem:ibsec:correctness:secure:3} shows that $\mathit{secure}^{\mathit{data}}_{P,u}$ is a sound, under approximation of $\mathit{secure}_{P,u}$.
However, as shown in \cite{guarnieri2014optimal}, deciding whether $\mathit{secure}^{\mathit{data}}_{P,u}(r, i \attMod \phi)$ holds for a given judgment is still undecidable for the relational calculus.

\begin{lemma}\thlabel{theorem:ibsec:correctness:secure:3}
Let $P = \langle M, f \rangle$ be an \accessControlConfiguration{}, $L$ be the $P$-LTS, $u \in {\cal U}$ be a user, $r \in \mathit{traces}(L)$ be an $L$-run, $\phi \in RC_{\mathit{bool}}$ is a sentence, and $1 \leq i \leq |r|$.
Given a judgment $r, i \attMod \phi$, if $\mathit{secure}^{\mathit{data}}_{P,u}(r, i \attMod \phi)$, then $\mathit{secure}_{P,u}(r, i \attMod \phi)$. 
\end{lemma}

\begin{proof}
We prove the claim by contradiction.
Let $P = \langle M, f \rangle$ be an \accessControlConfiguration{}, $L$ be the $P$-LTS, $u \in {\cal U}$ be a user, $r \in \mathit{traces}(L)$ be an $L$-run, $\phi \in RC_{\mathit{bool}}$ is a sentence, and $1 \leq i \leq |r|$.
Furthermore, let $s = \langle \mathit{db}, U, \mathit{sec}, T, V, \\ c \rangle$ be the $i$-th state of $r$.
Assume, for contradiction's sake, that $\mathit{secure}^{\mathit{data}}_{P,u}(r, i \attMod \phi)$ holds and $\mathit{secure}_{P,u}(r, i \attMod \phi)$ does not hold.
We denote, for brevity's sake, the fact that $\mathit{secure}_{P,u}(r, i \attMod \phi)$ does not hold as $\neg \mathit{secure}_{P,u}(r, i  \attMod \phi)$.
From $\neg \mathit{secure}_{P,u}(r, i \attMod \phi)$, it follows that there is a run $r' \in \mathit{traces}(L)$, whose last state is $s' = \langle \mathit{db}', U', \mathit{sec}', T', V', c' \rangle$,  such that $r^{i} \cong_{P,u} r'$ and  $[\phi]^{\mathit{db}} \neq [\phi]^{\mathit{db}'}$.
From the $(P,u)$-indistinguishability definition, it follows that $\mathit{pState}(\mathit{last}(r^{i}))$ and $\mathit{pState}(\mathit{last}(r'))$ are data indistinguishable according   to $M$ and $u$, i.e., $\mathit{pState}(\mathit{last}(r^{i})) \cong_{u,M}^{\mathit{data}} \mathit{pState}(\mathit{last}(r'))$.
From $\mathit{secure}^{\mathit{data}}_{P,u}(r, i \attMod \phi)$, it also follows that for all $s', s'' \in \llbracket \mathit{pState}(s)\rrbracket_{u,M}^\mathit{data}$, $[\phi]^{s'.\mathit{db}} = [\phi]^{s''.\mathit{db}}$.
From this and the fact that $\mathit{pState}(\mathit{last}(r^{i})) \cong_{u,M}^{\mathit{data}} \mathit{pState}(\mathit{last}(r'))$, it follows that $[\phi]^{\mathit{db}} = [\phi]^{\mathit{db}'}$, which contradicts  $[\phi]^{\mathit{db}} \neq [\phi]^{\mathit{db}'}$.
This completes the proof.
\end{proof}

We now show that the rewritings $\phi_{s,u}^{\top}$ and $\phi_{s,u}^{\bot}$ provide the desired properties.
First, \thref{theorem:rewriting:invariants} proves that the two rewriting satisfy the following invariants: ``if $\phi_{s,u}^{\top}$ holds in $s$, then also $\phi$ holds in $s$'' and ``if $\phi_{s,u}^{\bot}$ does not hold in $s$, then also $\phi$ does not hold in $s$''.
Afterwards, \thref{theorem:rewriting:secure} shows that both  $\phi_{s,u}^{\top}$ and $\phi_{s,u}^{\bot}$ are secure.
Then, \thref{theorem:rewriting:equivalent:modulo:indistinguishable:state} shows that $\phi_{s,u}^{\top}$ and $\phi_{s,u}^{\bot}$ are equivalent to $\phi_{s',u}^{\top}$ and $\phi_{s',u}^{\bot}$ for any two data indistinguishable $M$-state $s$ and $s'$.
Finally, \thref{theorem:rewriting:domain:independence} shows that both $\phi_{s,u}^{\top}$ and $\phi_{s,u}^{\bot}$ are domain-independent.

\begin{lemma}\thlabel{theorem:rewriting:invariants}
Let $M = \langle D, \Gamma\rangle$ be a system configuration, $s =\langle \mathit{db}, U, \mathit{sec}, T, V \rangle$ be a partial $M$-state, $u \in U$ be a user, and $\phi$ be a $D$-formula.
For all assignments $\nu$ over $\textbf{dom}$ that are well-formed for $\phi$, the following conditions hold:
\begin{compactitem}
\item if $[\phi_{s,u}^{\top} \circ \nu]^{\mathit{db}} = \top$, then $[\phi \circ \nu]^{\mathit{db}} = \top$, and
\item if $[\phi_{s,u}^{\bot} \circ \nu]^{\mathit{db}} = \bot$, then $[\phi \circ \nu]^{\mathit{db}} = \bot$.
\end{compactitem}
\end{lemma}

\begin{proof}
Let $M = \langle D, \Gamma\rangle$ be a system configuration, $s =\langle \mathit{db}, U, \mathit{sec}, T, V \rangle$ be a partial $M$-state, $u \in U$ be a user, and $\phi$ be a $D$-formula.
Furthermore, let $\nu$ be an assignment that is well-formed for $\phi$.
We prove our claim by induction on the structure of the formula $\phi$.

\smallskip
\noindent
{\bf Base Case} 
There are four cases:
\begin{compactenum}
\item $\phi := x = y$. 
In this case, $\phi_{s,u}^{\top} = \phi_{s,u}^{\bot} = \phi$.
From this, it follows that $[(x = v)_{s,u}^{\top} \circ \nu]^{\mathit{db}} = [(x = v) \circ \nu]^{\mathit{db}}$ and $[(x = v)_{s,u}^{\bot} \circ \nu]^{\mathit{db}} = [(x = v) \circ \nu]^{\mathit{db}}$.
Therefore, our claim follows trivially.

\item $\phi := \top$. The proof of this case is similar to that of $\phi := x = y$.

\item $\phi := \bot$. The proof of this case is similar to that of $\phi := x = y$.

\item $\phi := R(\overline{x})$.
Let $\overline{t}$ be the tuple $\nu (\overline{x})$.
Note that since $\nu$ is well-formed for $\phi$, $\overline{t}$ is well-defined.

Assume that $[\phi_{s,u}^{\top} \circ \nu]^{\mathit{db}} = \top$.
From this and $\phi_{s,u}^{\top} := \bigvee_{S \in R^{\top}_{s,u}} S(\overline{x})$, it follows that there is an $S \in R^{\top}_{s,u}$ such that $\overline{t} \in \mathit{db}(S)$.
Since $S \in R^{\top}_{s,u}$, it follows that $S \subseteq_{M} R$.
From $S \subseteq_{M} R$, $\overline{t} \in \mathit{db}(S)$, and \thref{theorem:containment:sound}, it follows that $\overline{t} \in \mathit{db}(R)$.
From this and the relational calculus semantics, it follows that $[\phi \circ \nu]^{\mathit{db}} = \top$.

Assume that $[\phi_{s,u}^{\bot} \circ \nu]^{\mathit{db}} = \bot$.
From this and $\phi_{s,u}^{\bot} := \bigwedge_{S \in R^{\bot}_{s,u}} S(\overline{x})$, it follows that there is an $S \in R^{\bot}_{s,u}$ such that $\overline{t} \not\in \mathit{db}(S)$.
Since $S \in R^{\bot}_{s,u}$, it follows that $R \subseteq_{M} S$.
From $R \subseteq_{M} S$, $\overline{t} \not\in \mathit{db}(S)$, and \thref{theorem:containment:sound}, it follows that $\overline{t} \not\in \mathit{db}(R)$.
From this and the relational calculus semantics, it follows that $[\phi \circ \nu]^{\mathit{db}} = \bot$.

\end{compactenum}
This completes the proof of the base case.

\smallskip
\noindent
{\bf Induction Step} 
Assume that our claim holds for all formulae whose length is less than $\phi$'s length.
We now show that our claim holds also for $\phi$.
There are a number of cases depending on $\phi$'s structure.
\begin{compactenum}
\item $\phi := \psi \wedge \gamma$.
Assume that  $[\phi_{s,u}^{\top} \circ \nu]^{\mathit{db}} = \top$.
From this and $\phi_{s,u}^{\top} := \psi_{s,u}^{\top} \wedge \gamma_{s,u}^{\top}$, it follows that $[\psi_{s,u}^{\top} \circ \nu]^{\mathit{db}} = \top$ and $[\gamma_{s,u}^{\top} \circ \nu]^{\mathit{db}} = \top$.
Since $\nu$ is well-formed for $\phi$, it is also well-formed for $\psi$ and $\gamma$ because $\mathit{free}(\psi) \subseteq  \mathit{free}(\phi)$ and $\mathit{free}(\gamma) \subseteq  \mathit{free}(\phi)$.
From $[\psi_{s,u}^{\top} \circ \nu]^{\mathit{db}} = \top$ and the induction hypothesis, it follows that $[\psi \circ \nu]^{\mathit{db}} = \top$.
From $[\gamma_{s,u}^{\top} \circ \nu]^{\mathit{db}} = \top$ and the induction hypothesis, it follows that $[\gamma \circ \nu]^{\mathit{db}} = \top$.
From $[\psi \circ \nu]^{\mathit{db}} = \top$, $[\gamma \circ \nu]^{\mathit{db}} = \top$,  $\phi := \psi \wedge \gamma$, and the relational calculus semantics, it follows that $[\phi \circ \nu]^{\mathit{db}} = \top$.

Assume that  $[\phi_{s,u}^{\bot} \circ \nu]^{\mathit{db}} = \bot$.
From this and $\phi_{s,u}^{\bot} := \psi_{s,u}^{\bot} \wedge \gamma_{s,u}^{\bot}$, there are two cases:
\begin{compactenum}
\item $[\psi_{s,u}^{\bot} \circ \nu]^{\mathit{db}} = \bot$.
From $[\psi_{s,u}^{\bot} \circ \nu]^{\mathit{db}} = \bot$ and the induction hypothesis, it follows that $[\psi \circ \nu]^{\mathit{db}} = \bot$.
From this,  $\phi := \psi \wedge \gamma$, and the relational calculus semantics, it follows that $[\phi \circ \nu]^{\mathit{db}} = \bot$

\item $[\gamma_{s,u}^{\bot} \circ \nu]^{\mathit{db}} = \bot$.
From $[\gamma_{s,u}^{\bot} \circ \nu]^{\mathit{db}} = \bot$ and the induction hypothesis, it follows that $[\gamma \circ \nu]^{\mathit{db}} = \bot$.
From this,  $\phi := \psi \wedge \gamma$, and the relational calculus semantics, it follows that $[\phi \circ \nu]^{\mathit{db}} = \bot$
\end{compactenum}

\item $\phi := \psi \vee \gamma$.
The proof of this case is similar to that of $\phi := \psi \wedge \gamma$.

\item $\phi := \neg \psi$.
Assume that  $[\phi_{s,u}^{\top} \circ \nu]^{\mathit{db}} = \top$.
From this and $\phi_{s,u}^{\top} := \neg \psi_{s,u}^{\bot}$, it follows that $[\psi_{s,u}^{\bot} \circ \nu]^{\mathit{db}} = \bot$.
From this and the induction hypothesis, it follows that $[\psi \circ \nu]^{\mathit{db}} = \bot$.
From this, $\phi := \neg \psi$, and the relational calculus semantics, it follows that $[\phi \circ \nu]^{\mathit{db}} = \top$.

Assume that  $[\phi_{s,u}^{\bot} \circ \nu]^{\mathit{db}} = \bot$.
From this and $\phi_{s,u}^{\bot} := \neg \psi_{s,u}^{\top}$, it follows that $[\psi_{s,u}^{\top} \circ \nu]^{\mathit{db}} = \top$.
From this and the induction hypothesis, it follows that $[\psi \circ \nu]^{\mathit{db}} = \top$.
From this, $\phi := \neg \psi$, and the relational calculus semantics, it follows that $[\phi \circ \nu]^{\mathit{db}} = \bot$.

\item $\phi := \exists x.\, \psi$.
Assume that $[\phi_{s,u}^{\top} \circ \nu]^{\mathit{db}} = \top$.
From this and $\phi_{s,u}^{\top} := \exists x.\, \psi_{s,u}^{\top}$, it follows that there is a $v \in \mathbf{dom}$ such that $[\psi_{s,u}^{\top} \circ \nu[x \mapsto v]]^{\mathit{db}} = \top$.
Note that since $v$ is well-formed for $\phi$, $\nu[x \mapsto v]$ is well-formed for $\psi$ because $\phi := \exists x.\, \psi$.
From this,  $[\psi_{s,u}^{\top} \circ \nu[x \mapsto v]]^{\mathit{db}} = \top$, and the induction hypothesis, it follows that  $[\psi \circ \nu[x \mapsto v]]^{\mathit{db}} = \top$.
From this, $\phi := \exists x.\, \psi$, and the relational calculus semantics, it follows that $[\phi \circ \nu]^{\mathit{db}} = \top$.

Assume that $[\phi_{s,u}^{\bot} \circ \nu]^{\mathit{db}} = \bot$.
From this and $\phi_{s,u}^{\bot} := \exists x.\, \psi_{s,u}^{\bot}$, it follows that for all $v \in \mathbf{dom}$, $[\psi_{s,u}^{\bot} \circ \nu[x \mapsto v]]^{\mathit{db}} = \bot$.
Note that since $v$ is well-formed for $\phi$, $\nu[x \mapsto v]$ is well-formed for $\psi$ because $\phi := \exists x.\, \psi$.
From this,  $[\psi_{s,u}^{\bot} \circ \nu[x \mapsto v]]^{\mathit{db}} = \bot$, and the induction hypothesis, it follows that for all $v \in \mathbf{dom}$, $[\psi \circ \nu[x \mapsto v]]^{\mathit{db}} = \bot$.
From this, $\phi := \exists x.\, \psi$, and the relational calculus semantics, it follows that $[\phi \circ \nu]^{\mathit{db}} = \bot$.

\item $\phi := \forall x.\, \psi$.
The proof of this case is similar to that of $\phi:= \exists x.\, \psi$.
\end{compactenum}
This completes the proof of the induction step.

This completes the proof of our claim.
\end{proof}

In \thref{theorem:rewriting:secure}, we prove that our rewritings are secure.

\begin{lemma}\thlabel{theorem:rewriting:secure}
Let $P = \langle M, f \rangle$ be an extended configuration, where $M = \langle D, \Gamma\rangle$ is a system configuration and $f$ is an $M$-\acf{}, $r \in \mathit{traces}(L)$ be a run, $\phi$ be a $\mathit{RC}$-formula, and $1 \leq i \leq r$.
Furthermore, let $s$ be the $i$-th state of $r$.
For all assignments $\nu$ over $\textbf{dom}$ that are well-formed for $\phi$, $\mathit{secure}^{\mathit{data}}_{P,u}(r, i \attMod \phi_{s,u}^{\top}\circ \nu)$, $\mathit{secure}^{\mathit{data}}_{P,u}(r, i \attMod \phi_{s,u}^{\bot} \circ \nu)$, and $\mathit{secure}^{\mathit{data}}_{P,u}(r, i \attMod \phi_{s,u}^{\mathit{rw}} \circ \nu)$ hold.
\end{lemma}

\begin{proof}
The security of $r, i \attMod \phi_{s,u}^{\mathit{rw}}$ follows trivially from that of $r, i \attMod \phi_{s,u}^{\top}$ and $r, i \attMod \phi_{s,u}^{\bot}$.
Therefore, in the following we prove just that  $\mathit{secure}^{\mathit{data}}_{P,u}(r, i \attMod \phi_{s,u}^{\top}\circ \nu)$ and $\mathit{secure}^{\mathit{data}}_{P,u}(r, i \attMod \phi_{s,u}^{\bot} \circ \nu)$ hold.
Let $M = \langle D, \Gamma\rangle$ be a system configuration, $s =\langle \mathit{db}, U, \mathit{sec}, T, V \rangle$ be a partial $M$-state, $u \in U$ be a user, and $\phi$ be a $D$-formula.
Furthermore, let $\nu$ be an assignment that is well-formed for $\phi$.
We prove our claim by induction on the structure of the formula $\phi$.

\smallskip
\noindent
{\bf Base Case} 
There are four cases:
\begin{compactenum}
\item $\phi := x = y$. 
The claim holds trivially.
Indeed, $\phi_{s,u}^{\top}\circ \nu$ and $\phi_{s,u}^{\bot}\circ \nu$ are always equivalent either to $\top$ or to $\bot$.
Since for all $s', s'' \in \llbracket \mathit{pState}(\mathit{last}(r^{i}))\rrbracket_{u,M}^\mathit{data}$, $[\top]^{s'.\mathit{db}} = [\top]^{s''.\mathit{db}}$ and $[\bot]^{s'.\mathit{db}} = [\bot]^{s''.\mathit{db}}$, it follows that both $\mathit{secure}^{\mathit{data}}_{P,u}(r, i \attMod \phi_{s,u}^{\top}\circ \nu)$ and $\mathit{secure}^{\mathit{data}}_{P,u}(r, i \attMod \phi_{s,u}^{\bot} \circ \nu)$ hold.

\item $\phi := \top$. The proof of this case is similar to that of $\phi := x = y$.

\item $\phi := \bot$. The proof of this case is similar to that of $\phi := x = y$.

\item $\phi := R(\overline{x})$.
Assume, for contradiction's sake, that $\mathit{secure}^{\mathit{data}}_{P,u}(r, i \attMod \phi_{s,u}^{\top}\circ \nu)$ does not hold.
From this and $\mathit{secure}^{\mathit{data}}_{P,u}$'s definition, it follows that there are two $M$-partial states $s' = \langle \mathit{db}', U, \mathit{sec}, T, V \rangle$ and $s'' = \langle \mathit{db}'', U, \\ \mathit{sec}, T, V \rangle$ in  $\llbracket \mathit{pState}(\mathit{last}(r^{i}))\rrbracket_{u,M}^\mathit{data}$ such that $[\phi_{s,u}^{\top}\circ \nu]^{\mathit{db}'} \\ \neq [\phi_{s,u}^{\top}\circ \nu]^{\mathit{db}''}$.
Note that this rule out the cases in which $R^{v}_{s,u} = \emptyset$ for any $v \in \{\top,\bot\}$.
We assume without loss of generality that $[\phi_{s,u}^{\top}\circ \nu]^{\mathit{db}'} = \top$ and $[\phi_{s,u}^{\top}\circ \nu]^{\mathit{db}''} = \bot$.
From this and $\phi_{s,u}^{\top} := \bigvee_{S \in R^{\top}_{s,u}} S(\overline{x})$, it follows that there is an predicate symbol $S$ in the extended vocabulary such that $\nu(\overline{x}) \in \mathit{db}'(S)$ and $\nu(\overline{x}) \not\in \mathit{db}''(S)$.
There are two cases:
\begin{compactitem}
\item $S$ is a table in $D$ or a view in $V$.
Since $S \in R^{\top}_{s,u}$, it follows that $\langle \oplus, \mathtt{SELECT}, S \rangle \in \mathit{permissions}(\mathit{last}(r^{i}), \\ u)$.
Note that $\mathit{permissions}(s',u) = \mathit{permissions}(s'', \\ u) = \mathit{permissions}(\mathit{last}(r^{i}),u)$ because all the states are in the same equivalence class.
From $s' \cong_{u,M}^{\mathit{data}} s''$, $\langle \oplus, \mathtt{SELECT}, S \rangle \in \mathit{permissions}(s', u)$, and the definition of data indistinguishability, it follows that $\mathit{db}'(S) = \mathit{db}''(S)$.
From this, it follows that $\nu(\overline{x}) \in \mathit{db}'(S)$ iff $\nu(\overline{x}) \in \mathit{db}''(S)$, which contradicts $\nu(\overline{x}) \in \mathit{db}'(S)$ and $\nu(\overline{x}) \not\in \mathit{db}''(S)$.

\item $S$ is a projection of $O$, which is either a table  in $D$ or a view in $V$.
From $S \in R^{\top}_{s,u}$ and $R^{\top}_{s,u}$'s definition, it follows that $\langle \oplus, \mathtt{SELECT}, O \rangle \in \mathit{permissions}\\ (\mathit{last}(r^{i}), u)$.
From $s' \cong_{u,M}^{\mathit{data}} s''$, $\langle \oplus, \mathtt{SELECT}, O \rangle \in \mathit{permissions} (s', u)$, and the definition of data indistinguishability, it follows that $\mathit{db}'(O) = \mathit{db}''(O)$.
From this and the definition of $S$, it also follows that  $\mathit{db}'(S) = \mathit{db}''(S)$\footnote{With a slight abuse of notation, we consider $S$ as a view.}.
From this, it follows that $\nu(\overline{x}) \in \mathit{db}'(S)$ iff $\nu(\overline{x}) \in \mathit{db}''(S)$, which contradicts $\nu(\overline{x}) \in \mathit{db}'(S)$ and $\nu(\overline{x}) \not\in \mathit{db}''(S)$.
\end{compactitem}
The proof of $\mathit{secure}^{\mathit{data}}_{P,u}(r, i \attMod \phi_{s,u}^{\bot}\circ \nu)$ is analogous. 
 
\end{compactenum}
This completes the proof of the base case.

\smallskip
\noindent
{\bf Induction Step} 
Assume that our claim holds for all formulae whose length is less than $\phi$'s length.
We now show that our claim holds also for $\phi$.
There are a number of cases depending on $\phi$'s structure.
\begin{compactenum}
\item $\phi := \psi \wedge \gamma$.
Assume, for contradiction's sake, that $\mathit{secure}^{\mathit{data}}_{P,u}(r, i \attMod \phi_{s,u}^{\top}\circ \nu)$ does not hold.
From this and $\mathit{secure}^{\mathit{data}}_{P,u}$'s definition, it follows that there are two $M$-partial states $s' = \langle \mathit{db}', U, \mathit{sec}, T, V \rangle$ and $s'' = \langle \mathit{db}'', U, \\ \mathit{sec}, T, V \rangle$ in  $\llbracket \mathit{pState}(\mathit{last}(r^{i}))\rrbracket_{u,M}^\mathit{data}$ such that $[\phi_{s,u}^{\top}\circ \nu]^{\mathit{db}'} \\ \neq [\phi_{s,u}^{\top}\circ \nu]^{\mathit{db}''}$.
We assume, without loss of generality, that $[\phi_{s,u}^{\top}\circ \nu]^{\mathit{db}'} = \top$ and $[\phi_{s,u}^{\top}\circ \nu]^{\mathit{db}''} = \bot$.
From this and $\phi_{s,u}^{\top} = \psi_{s,u}^{\top} \wedge \gamma_{s,u}^{\top}$, it follows that either $[\psi_{s,u}^{\top}\circ \nu]^{\mathit{db}'} = \top$ and $[\psi_{s,u}^{\top}\circ \nu]^{\mathit{db}''} = \bot$ or  $[\gamma_{s,u}^{\top}\circ \nu]^{\mathit{db}'} = \top$ and $[\gamma_{s,u}^{\top}\circ \nu]^{\mathit{db}''} = \bot$.
We assume, without loss of generality, that $[\psi_{s,u}^{\top}\circ \nu]^{\mathit{db}'} = \top$ and $[\psi_{s,u}^{\top}\circ \nu]^{\mathit{db}''} = \bot$.
From this, it follows that $\mathit{secure}^{\mathit{data}}_{P,u}(r, i \attMod \psi_{s,u}^{\top}\circ \nu)$ does not hold. 
From the induction hypothesis, it follows that $\mathit{secure}^{\mathit{data}}_{P,u}(r, i \attMod \psi_{s,u}^{\top}\circ \nu)$ holds leading to a contradiction.

The proof of $\mathit{secure}^{\mathit{data}}_{P,u}(r, i \attMod \phi_{s,u}^{\bot}\circ \nu)$ is analogous. 

\item $\phi := \psi \vee \gamma$.
The proof of this case is similar to that of $\phi := \psi \wedge \gamma$.

\item $\phi := \neg \psi$.
Assume, for contradiction's sake, that \\ $\mathit{secure}^{\mathit{data}}_{P,u}(r, i \attMod \phi_{s,u}^{\top}\circ \nu)$ does not hold.
From this and $\mathit{secure}^{\mathit{data}}_{P,u}$'s definition, it follows that there are two $M$-partial states $s' = \langle \mathit{db}', U, \mathit{sec}, T, V \rangle$ and $s'' = \langle \mathit{db}'', U, \\ \mathit{sec}, T, V \rangle$ in  $\llbracket \mathit{pState}(\mathit{last}(r^{i}))\rrbracket_{u,M}^\mathit{data}$ such that $[\phi_{s,u}^{\top}\circ \nu]^{\mathit{db}'} \\ \neq [\phi_{s,u}^{\top}\circ \nu]^{\mathit{db}''}$.
We assume, without loss of generality, that $[\phi_{s,u}^{\top}\circ \nu]^{\mathit{db}'} = \top$ and $[\phi_{s,u}^{\top}\circ \nu]^{\mathit{db}''} = \bot$.
From this and $\phi_{s,u}^{\top} = \neg \psi_{s,u}^{\bot}$, it follows that $[\psi_{s,u}^{\bot}\circ \nu]^{\mathit{db}'} = \bot$ and $[\psi_{s,u}^{\bot}\circ \nu]^{\mathit{db}''} = \top$.
From this, it follows that $\mathit{secure}^{\mathit{data}}_{P,u}(r, i \attMod \psi_{s,u}^{\bot}\circ \nu)$ does not hold. 
From the induction hypothesis and $\phi := \neg \psi$, it follows that $\mathit{secure}^{\mathit{data}}_{P,u}(r, i \attMod \psi_{s,u}^{\bot}\circ \nu)$ holds leading to a contradiction.

The proof of $\mathit{secure}^{\mathit{data}}_{P,u}(r, i \attMod \phi_{s,u}^{\bot}\circ \nu)$ is analogous.

\item $\phi := \exists x.\, \psi$.
Assume, for contradiction's sake, that \\ $\mathit{secure}^{\mathit{data}}_{P,u}(r, i \attMod \phi_{s,u}^{\top}\circ \nu)$ does not hold.
From this and $\mathit{secure}^{\mathit{data}}_{P,u}$'s definition, it follows that there are two $M$-partial states $s' = \langle \mathit{db}', U, \mathit{sec}, T, V \rangle$ and $s'' = \langle \mathit{db}'', U, \\ \mathit{sec}, T, V \rangle$ in  $\llbracket \mathit{pState}(\mathit{last}(r^{i}))\rrbracket_{u,M}^\mathit{data}$ such that $[\phi_{s,u}^{\top}\circ \nu]^{\mathit{db}'} \\ \neq [\phi_{s,u}^{\top}\circ \nu]^{\mathit{db}''}$.
We assume, without loss of generality, that $[\phi_{s,u}^{\top}\circ \nu]^{\mathit{db}'} = \top$ and $[\phi_{s,u}^{\top}\circ \nu]^{\mathit{db}''} = \bot$.
From this and $\phi_{s,u}^{\top} = \exists x.\, \psi_{s,u}^{\top}$, it follows that there is a $v' \in \mathbf{dom}$ such that $[\psi_{s,u}^{\top}\circ \nu[x \mapsto v']]^{\mathit{db}'} = \top$ and there is no  $v'' \in \mathbf{dom}$ such that $[\psi_{s,u}^{\top}\circ \nu[x \mapsto v'']]^{\mathit{db}''} = \top$.
Therefore, $[\psi_{s,u}^{\top}\circ \nu[x \mapsto v']]^{\mathit{db}'} = \top$ and $[\psi_{s,u}^{\top}\circ \nu[x \mapsto v']]^{\mathit{db}''} = \bot$.
Note that $\nu[x \mapsto v']$ is well-formed for $\psi_{s,u}^{\top}$.
From this, it follows that $\mathit{secure}^{\mathit{data}}_{P,u}(r, i \attMod \psi_{s,u}^{\top}\circ \nu[x \mapsto v'])$ does not hold. 
However, from the fact that $\nu[x \mapsto v']$ is well-formed for $\psi_{s,u}^{\top}$ and the induction hypothesis, it follows that $\mathit{secure}^{\mathit{data}}_{P,u}(r, i \attMod \psi_{s,u}^{\top}\circ \nu [x \mapsto v'])$ holds leading to a contradiction.

The proof of $\mathit{secure}^{\mathit{data}}_{P,u}(r, i \attMod \phi_{s,u}^{\bot}\circ \nu)$ is analogous. 

\item $\phi := \forall x.\, \psi$.
The proof of this case is similar to that of $\phi:= \exists x.\, \psi$.
\end{compactenum}
This completes the proof of the induction step.

This completes the proof of our claim.
\end{proof}

\begin{proposition} \thlabel{theorem:rewriting:bound:equivalent:modulo:indistinguishable:state}
Let $M = \langle D, \Gamma\rangle$ be a system configuration, $s =\langle \mathit{db}, U, \mathit{sec}, T, V \rangle$ and $s'  =\langle \mathit{db}', U', \mathit{sec}', T',V' \rangle $ be two partial $M$-states, $u \in U$ be a user, $v \in \{\top,\bot\}$, and $\phi$ be a $D$-formula.
If $s \cong_{u,M}^{\mathit{data}} s'$, then $\mathit{bound}(\phi,s,u,x,v) = \mathit{bound}(\phi,s',u,x,v)$.
\end{proposition}

\begin{proof}
Let $M = \langle D, \Gamma\rangle$ be a system configuration, $s =\langle \mathit{db}, U, \mathit{sec}, T, V \rangle$ and $s'  =\langle \mathit{db}', U', \mathit{sec}', T', V' \rangle $ be two partial $M$-states, $u \in U$ be a user,  $v \in \{\top,\bot\}$, and $\phi$ be a $D$-formula.
We prove our claim by induction on the structure of the formula $\phi$.

\smallskip
\noindent
{\bf Base Case} 
There are four cases:
\begin{compactenum}
\item $\phi := y = z$. 
The result of $\mathit{bound}(\phi,s,u,x,v)$ and $\mathit{bound}(\phi,s',u,x,v)$ does not depend on $s$.
Therefore, $\mathit{bound}(\phi,s,u,x,v) = \mathit{bound}(\phi,s',u,x,v)$.  

\item $\phi := \top$. $\mathit{bound}(\phi,s,u,x,v) = \mathit{bound}(\phi,s',u,x,v) = \bot$.

\item $\phi := \bot$. $\mathit{bound}(\phi,s,u,x,v) = \mathit{bound}(\phi,s',u,x,v) = \bot$.

\item $\phi := R(\overline{x})$.
The result of $\mathit{bound}(\phi,s,u,x,v)$ and $\mathit{bound} \\ (\phi,s',u,x,v)$ depend only on the sets $R^{v}_{s,u}$ and $R^{v}_{s',u}$, which in turn depend on the content of the sets $R^{v}_{s}$, $R^{v}_{s'}$, $\mathit{AUTH}_{s,u}^{*}$, and $\mathit{AUTH}_{s',u}^{*}$.
Assume that $s \cong_{u,M}^{\mathit{data}} s'$.
From this, it follows that  $R^{v}_{s} = R^{v}_{s'}$ (because $D$ is the same and $V = V'$) and $\mathit{AUTH}_{s,u}^{*} = \mathit{AUTH}_{s',u}^{*}$ (because $\mathit{sec} =\mathit{sec}'$).
From this, it follows that $\mathit{bound}(\phi,s, \\ u,x,v) = \mathit{bound}(\phi,s',u,x,v)$.

\end{compactenum}
This completes the proof of the base case.

\smallskip
\noindent
{\bf Induction Step} 
Assume that our claim holds for all formulae whose length is less than $\phi$.
We now show that our claim holds also for $\phi$.
There are a number of cases depending on $\phi$'s structure.
\begin{compactenum}
\item $\phi := \psi \wedge \gamma$.
Assume that  $s \cong_{u,M}^{\mathit{data}} s'$.
From this and the induction hypothesis, it follows that  $\mathit{bound}(\psi,s,u,x,v) \\ = \mathit{bound}(\psi,s',u,x,v)$ and $\mathit{bound}(\gamma,s,u,x,v) = \mathit{bound} \\ (\gamma,s',u,x,v)$.
From this and $\mathit{bound}(\phi, s, u, x, v) := \\ \mathit{bound}(\psi, s, u, x, v) \vee \mathit{bound}(\gamma,  s, u, x, v)$, it follows that $\mathit{bound}(\phi, s, u, x, v) = \mathit{bound}(\phi, s', u, x, v)$.

\item $\phi := \psi \vee \gamma$.
The proof of this case is similar to that of $\phi := \psi \wedge \gamma$.

\item $\phi := \neg \psi$.
Assume that  $s \cong_{u,M}^{\mathit{data}} s'$.
From this and the induction hypothesis, it follows that  $\mathit{bound}(\psi,s,u,x,v) \\ = \mathit{bound}(\psi,s',u,x,v)$.
From this,  $\mathit{bound}(\neg \psi, s,u,x,v) \\ = \mathit{bound}(\psi,s,u,x, \neg v)$, and   $\mathit{bound}(\neg \psi, s',u,x,v) = \\ \mathit{bound}(\psi,s',u,x, \neg v)$, it follows that  $\mathit{bound}(\phi, s, u, x, v) \\ = \mathit{bound}(\phi, s', u, x, v)$.

\item $\phi := \exists y.\, \psi$.
Assume that  $s \cong_{u,M}^{\mathit{data}} s'$.
There are two cases:
\begin{compactenum}
\item $x = y$. In this case, the proof is trivial as $\mathit{bound}(\phi,s, \\ u,x,v) = \mathit{bound}(\phi,s',u,x,v) = \bot$. 

\item $x \neq y$. In this case, $\mathit{bound}(\phi,s,u,x,v) = \mathit{bound}(\psi, \\ s,u,x,v) \wedge \mathit{bound}(\psi,s,u,y,v)$ and $\mathit{bound}(\phi,s',u,x, \\ v) = \mathit{bound}(\psi,s',u,x,v) \wedge \mathit{bound}(\psi,s',u,y,v)$.
From $s \cong_{u,M}^{\mathit{data}} s'$ and the induction hypothesis, it follows that $\mathit{bound}(\psi,s,u,x,v) = \mathit{bound}(\psi,s',u,x,v)$ and $\mathit{bound}(\psi,s,u,y,v) = \mathit{bound}(\psi,s',u,y,v)$.
From this,  $\mathit{bound}(\phi,s,u,x,v) = \mathit{bound}(\psi,s,u,x,v) \wedge \mathit{bound}(\psi, \\ s,u,y,v)$, and $\mathit{bound}(\phi,s',u,x,v) = \mathit{bound}(\psi,s',u, \\ x,v) \wedge \mathit{bound}(\psi,s',u,y,v)$, it follows that  $\mathit{bound}(\phi,s, \\ u,x,v) = \mathit{bound}(\phi,s',u,x,v)$.
\end{compactenum}

\item $\phi := \forall x.\, \psi$.
The proof of this case is similar to that of $\phi:= \exists x.\, \psi$.
\end{compactenum}
This completes the proof of the induction step.

This completes the proof of our claim.
\end{proof}

\begin{lemma} \thlabel{theorem:rewriting:equivalent:modulo:indistinguishable:state}
Let $M = \langle D, \Gamma\rangle$ be a system configuration, $s =\langle \mathit{db}, U, \mathit{sec}, T, V \rangle$ and $s'  =\langle \mathit{db}', U', \mathit{sec}', T',V' \rangle $ be two partial $M$-states, $u \in U$ be a user, and $\phi$ be a $D$-formula.
If $s \cong_{u,M}^{\mathit{data}} s'$, then $\phi_{s,u}^{\top} = \phi_{s',u}^{\top}$, $\phi_{s,u}^{\bot} = \phi_{s',u}^{\bot}$, and $\phi^{\mathit{rw}}_{s,u} = \phi^{\mathit{rw}}_{s',u}$.
\end{lemma}

\begin{proof}
Let $M = \langle D, \Gamma\rangle$ be a system configuration, $s =\langle \mathit{db}, U, \mathit{sec}, T, V \rangle$ and $s'  =\langle \mathit{db}', U', \mathit{sec}', T', V' \rangle $ be two partial $M$-states, $u \in U$ be a user, and $\phi$ be a $D$-formula.
We prove our claim by induction on the structure of the formula $\phi$.

\smallskip
\noindent
{\bf Base Case} 
There are four cases:
\begin{compactenum}
\item $\phi := x = y$. 
The claim holds trivially.
Indeed, $\phi_{s,u}^{\top} = \phi_{s,u}^{\bot} = \phi$.

\item $\phi := \top$. The proof of this case is similar to that of $\phi := x = y$.

\item $\phi := \bot$. The proof of this case is similar to that of $\phi := x = y$.

\item $\phi := R(\overline{x})$.
The formulae $\phi_{s,u}^{\top}$ and $\phi_{s',u}^{\top}$ depend only on the sets $R^{\top}_{s,u}$ and $R^{\top}_{s',u}$, which in turn depends on $R^{\top}_{s}$, $R^{\top}_{s'}$, $\mathit{AUTH}_{s,u}^{*}$, and $\mathit{AUTH}_{s',u}^{*}$.
If $s \cong_{u,M}^{\mathit{data}} s'$, then $R^{\top}_{s} = R^{\top}_{s'}$ (because $D$ is the same and $V = V'$) and $\mathit{AUTH}_{s,u}^{*} = \mathit{AUTH}_{s',u}^{*}$ (because $\mathit{sec} =\mathit{sec}'$).
Therefore, $\phi_{s,u}^{\top} = \phi_{s',u}^{\top}$.
The proof for $\phi_{s,u}^{\bot}$ is analogous.

\end{compactenum}
This completes the proof of the base case.

\smallskip
\noindent
{\bf Induction Step} 
Assume that our claim holds for all formulae whose length is less than $\phi$.
We now show that our claim holds also for $\phi$.
There are a number of cases depending on $\phi$'s structure.
\begin{compactenum}
\item $\phi := \psi \wedge \gamma$.
Assume that  $s \cong_{u,M}^{\mathit{data}} s'$.
From this and the induction hypothesis, it follows that  $\psi_{s,u}^{\top} = \psi_{s',u}^{\top}$ and $\gamma_{s,u}^{\top} = \gamma_{s',u}^{\top}$.
From this, $\phi := \psi \wedge \gamma$, $\phi_{s,u}^{\top} := \psi_{s,u}^{\top} \wedge \gamma_{s,u}^{\top}$,  and $\phi_{s',u}^{\top} := \psi_{s',u}^{\top} \wedge \gamma_{s',u}^{\top}$, it follows that $\phi_{s,u}^{\top} = \phi_{s',u}^{\top}$.

The proof of $\phi_{s,u}^{\bot} = \phi_{s',u}^{\bot}$ is analogous. 

\item $\phi := \psi \vee \gamma$.
The proof of this case is similar to that of $\phi := \psi \wedge \gamma$.

\item $\phi := \neg \psi$.
Assume that  $s \cong_{u,M}^{\mathit{data}} s'$.
From this and the induction hypothesis, it follows that  $\psi_{s,u}^{\top} = \psi_{s',u}^{\top}$ and $\psi_{s,u}^{\bot} = \psi_{s',u}^{\bot}$.
From this, $\phi := \neg \psi$, $\phi_{s,u}^{\top} := \neg \psi_{s,u}^{\bot}$,  and $\phi_{s',u}^{\top} := \neg \psi_{s',u}^{\bot}$, it follows that $\phi_{s,u}^{\top} = \phi_{s',u}^{\top}$.

The proof of $\phi_{s,u}^{\bot} = \phi_{s',u}^{\bot}$ is analogous.

\item $\phi := \exists x.\, \psi$.
Assume that  $s \cong_{u,M}^{\mathit{data}} s'$.
From this and the induction hypothesis, it follows that  $\psi_{s,u}^{\top} = \psi_{s',u}^{\top}$.
We remark that $\mathit{bound}(\psi,s,u,x,\top) = \mathit{bound}(\psi,s',u,x,\top)$, as proved in \thref{theorem:rewriting:bound:equivalent:modulo:indistinguishable:state}.
There are two cases:
\begin{compactenum}
\item $\mathit{bound}(\psi,s,u,x,\top) = \top$. From this, $\mathit{bound}(\psi,s,u, \\ x,\top) = \mathit{bound}(\psi,s',u,x,\top)$, $\psi_{s,u}^{\top} = \psi_{s',u}^{\top}$, $\phi := \exists x.\,  \psi$, $\phi_{s,u}^{\top} := \exists x.\, \psi_{s,u}^{\top}$,  and $\phi_{s',u}^{\top} := \exists x.\, \psi_{s',u}^{\top}$, it follows that $\phi_{s,u}^{\top} = \phi_{s',u}^{\top}$.

\item $\mathit{bound}(\psi,s,u,x,\top) = \bot$. From this, $\mathit{bound}(\psi,s,u, \\ x,\top) = \mathit{bound}(\psi,s',u,x,\top)$, and $\phi_{s,u}^{\top}$ definition, it follows that $\phi_{s',u}^{\top} = \phi_{s',u}^{\top} = \bot$.
\end{compactenum}
The proof of $\phi_{s,u}^{\bot} = \phi_{s',u}^{\bot}$ is analogous.  

\item $\phi := \forall x.\, \psi$.
The proof of this case is similar to that of $\phi:= \exists x.\, \psi$.
\end{compactenum}
This completes the proof of the induction step.

The equivalence  $\phi^{\mathit{rw}}_{s,u} = \phi^{\mathit{rw}}_{s',u}$ follows trivially from $\phi^{\mathit{rw}}_{s,u}$'s definition and the fact that $\phi_{s,u}^{\top} = \phi_{s',u}^{\top}$ and $\phi_{s,u}^{\bot} = \phi_{s',u}^{\bot}$.
This completes the proof of our claim.
\end{proof}

Before proving the domain independence of $\phi^{\top}_{s,u}$ and $\phi^{\bot}_{s,u}$, we introduce some notation.
The relation $\mathit{gen}$, introduced in \cite{VanGelder:1991:STR:114325.103712}, is the smallest relation defined by the rules in Figure \ref{figure:domain:independence:gen}.
Note that we extended $\mathit{gen}$ by adding the rules \emph{Equiv}, \emph{Const 1}, and \emph{Const 2}.
A relational calculus formula $\phi$ is \emph{allowed} iff it satisfies the following conditions:
\begin{compactitem}
\item for all $x \in \mathit{free}(\phi)$,  $\mathit{gen}(x,\phi)$ holds,
\item for every sub-formula $\exists x. \psi$ in $\phi$,  $\mathit{gen}(x,\psi)$ holds, and
\item for every sub-formula $\forall x. \psi$ in $\phi$,  $\mathit{gen}(x, \neg \psi)$ holds.
\end{compactitem}
As shown in \cite{VanGelder:1991:STR:114325.103712}, every allowed formula is domain independent.
Note that the addition of the  \emph{Equiv}, \emph{Const 1}, and \emph{Const 2} rules does not modify this result.

\begin{figure}[!hbtp]
\centering
\scalebox{0.93}{
\begin{tabular}{c c}

$
\infer[\text{Pred}]
{ \mathit{gen}(x,R(\overline{x}))
}
{
\hfill x \in \overline{x} \hfill 
}$
&
$
\infer[\text{Neg}]
{ \mathit{gen}(x,\neg \phi)
}
{
\hfill \mathit{gen}(x,\mathit{push}(\neg \phi)) \hfill 
}$

\\\\
$
\infer[\text{Exists}]
{ \mathit{gen}(x,\exists y. \phi)
}
{
\hfill x \neq y \hfill \quad
\hfill \mathit{gen}(x,\phi) \hfill
}$
&

$
\infer[\text{For all}]
{ \mathit{gen}(x,\forall y. \phi)
}
{
\hfill x \neq y \hfill \quad
\hfill \mathit{gen}(x,\phi) \hfill
}$

\\\\

$
\infer[\text{Or}]
{ \mathit{gen}(x,\phi \vee \psi)
}
{
\hfill \mathit{gen}(x,\phi) \hfill \quad
\hfill \mathit{gen}(x,\psi) \hfill 
}$
&
$
\infer[\text{Equiv}]
{ \mathit{gen}(x,\phi)
}
{
\hfill \mathit{gen}(x,\psi) \hfill \quad
\hfill \phi \equiv \psi \hfill 
}$
\\\\
$
\infer[\text{Const 1}]
{ \mathit{gen}(x,x = v)
}
{
\hfill v \in \mathbf{dom} \hfill 
}$
&
$
\infer[\text{Const 2}]
{ \mathit{gen}(x,v = x)
}
{
\hfill v \in \mathbf{dom} \hfill 
}$
 
\\\\
$
\infer[\text{And 1}]
{ \mathit{gen}(x,\phi \wedge \psi)
}
{
\hfill \mathit{gen}(x,\phi) \hfill 
}$
&
$
\infer[\text{And 2}]
{ \mathit{gen}(x,\phi \wedge \psi)
}
{
\hfill \mathit{gen}(x,\psi) \hfill 
}$
\\\\
\multicolumn{2}{c}{
$\mathit{push}(\phi) = \left\{ 
  \begin{array}{l l}
	\neg \psi \vee \neg \gamma & \text{if } \phi := \neg (\psi \wedge \gamma) \\
	\neg \psi \wedge \neg \gamma & \text{if } \phi := \neg (\psi \vee \gamma) \\  	
  	\forall x.\, \neg \psi & \text{if } \phi := \neg \exists x. \psi \\  	
  	\exists x.\, \neg \psi & \text{if } \phi := \neg \forall x. \psi \\  	
  	\psi & \text{if } \phi := \neg\neg \psi 	\\
  	x \neq y & \text{if } \phi := \neg(x = y) \\ 
  	x = y & \text{if } \phi := \neg(x \neq y) \\
  	\end{array}\right.$
  	}
\end{tabular}
}
\caption{$\mathit{gen}$ rules}\label{figure:domain:independence:gen}
\end{figure}

\begin{proposition}\thlabel{theorem:rewriting:equivalences}
Let $M = \langle D, \Gamma\rangle$ be a system configuration, $s = \langle \mathit{db}, U, \mathit{sec},  T,V \rangle$ be an $M$-partial state, $u \in U$ be a user, and $v \in \{\top,\bot\}$.
For any formulae $\phi$ and $\psi$, the following equivalences hold:
\begin{compactitem}
\item $(\neg \phi)^{v}_{s,u} \equiv \neg \phi^{\neg v}_{s,u}$,

\item $(\phi)^{v}_{s,u} \wedge (\psi)^{v}_{s,u} \equiv (\phi \wedge \psi)^{v}_{s,u}$,

\item $(\phi)^{v}_{s,u} \vee (\psi)^{v}_{s,u} \equiv (\phi \vee \psi)^{v}_{s,u}$,

\item $(\exists x.\, \phi)^{v}_{s,u} \equiv (\neg \forall x.\, \neg \phi)^{v}_{s,u}$, and 

\item $(\forall x.\, \phi)^{v}_{s,u} \equiv (\neg \exists x.\, \neg \phi)^{v}_{s,u}$.
\end{compactitem}
\end{proposition}

\begin{proof}
Let $M = \langle D, \Gamma\rangle$ be a system configuration, $s = \langle \mathit{db}, U, \mathit{sec},  T,V \rangle$ be an $M$-partial state, $u \in U$ be a user,  $v \in \{\top,\bot\}$, and $\phi$, $\psi$ be two formulae.

\begin{compactitem}
\item $(\neg \phi)^{v}_{s,u} \equiv \neg \phi^{\neg v}_{s,u}$.
This case follows trivially from the definition of the rewriting.

\item $(\phi)^{v}_{s,u} \wedge (\psi)^{v}_{s,u} \equiv (\phi \wedge \psi)^{v}_{s,u}$.
This case follows trivially from the definition of the rewriting.

\item $(\phi)^{v}_{s,u} \vee (\psi)^{v}_{s,u} \equiv (\phi \vee \psi)^{v}_{s,u}$
This case follows trivially from the definition of the rewriting.

\item $(\exists x.\, \phi)^{v}_{s,u} \equiv (\neg \forall x.\, \neg \phi)^{v}_{s,u}$.
There are two cases:
\begin{compactenum}
\item $\mathit{bound}(\phi,s,u,x,v) = \top$.
From this, it follows that $(\exists x.\, \phi)^{v}_{s,u} = \exists x.\, \phi^{v}_{s,u}$.
From $(\neg \phi)^{v}_{s,u} \equiv \neg \phi^{\neg v}_{s,u}$ and $(\neg \forall x.\, \neg \phi)^{v}_{s,u}$, it follows that  $(\neg \forall x.\, \neg \phi)^{v}_{s,u} \equiv \neg (\forall x.\, \\ \neg \phi)^{\neg v}_{s,u}$.
From the definition of $\mathit{bound}$, it follows that $\mathit{bound}(\neg \phi,s,u,x,\neg v) = \mathit{bound}( \phi,s,u,x, \neg \neg v)$.
From this and $v = \neg \neg v$, it follows that $\mathit{bound}(\neg \phi,s, \\ u,x,\neg v) = \mathit{bound}( \phi,s,u,x,  v)$.
From this and $\mathit{bound} \\ (\phi,s,u,x,v) = \bot$, it follows that $\mathit{bound}(\neg \phi,s,u,x, \\ \neg v) = \top$.
From this, it follows that $(\forall x.\, \neg \phi)^{\neg v}_{s,u} = \forall x.\, (\neg \phi)^{\neg v}_{s,u}$.
From this and $(\neg \phi)^{v}_{s,u} \equiv \neg \phi^{\neg v}_{s,u}$, it follows that $(\forall x.\, \neg \phi)^{\neg v}_{s,u} \equiv \forall x.\, \neg \phi^{v}_{s,u}$.
From this and $(\neg \forall x.\, \neg \phi)^{v}_{s,u} \equiv \neg (\forall x.\, \neg \phi)^{\neg v}_{s,u}$, it follows that  $(\neg \forall x.\, \neg \phi)^{v}_{s,u} \equiv \neg  \forall x.\, \neg \phi^{v}_{s,u}$.
From this and standard RC equivalences, it follows that $(\neg \forall x.\, \neg \phi)^{v}_{s,u} \\ \equiv  \exists x.\, \phi^{v}_{s,u}$.

\item $\mathit{bound}(\phi,s,u,x,v) = \bot$.
From this, it follows that $(\exists x.\, \phi)^{v}_{s,u} = \neg v$.
From  $(\neg \phi)^{v}_{s,u} \equiv \neg \phi^{\neg v}_{s,u}$ and $(\neg \forall x.\, \\ \neg \phi)^{v}_{s,u}$, it follows that  $(\neg \forall x.\, \neg \phi)^{v}_{s,u} \equiv \neg (\forall x.\, \neg \phi)^{\neg v}_{s,u}$.
From the definition of $\mathit{bound}$, it follows that $\mathit{bound} \\ (\neg \phi,s,u,x,\neg v) = \mathit{bound}  ( \phi,s,u,x, \neg \neg v)$.
From this and $v = \neg \neg v$, it follows that $\mathit{bound}(\neg \phi,s,u,x, \neg v) \\  = \mathit{bound}( \phi,s,u,x,  v)$. 
From this and $\mathit{bound}(\phi, s,u, \\ x,v) = \bot$, it follows that $\mathit{bound}(\neg \phi,s,u,x,\neg v) = \bot$.
From this, it follows that $(\forall x.\, \neg \phi)^{\neg v}_{s,u} = v$.
From this and $(\neg \forall x.\, \neg \phi)^{v}_{s,u}   \equiv \neg (\forall x.\, \neg \phi)^{\neg v}_{s,u}$, it follows that  $(\neg \forall x.\, \neg \phi)^{v}_{s,u} \equiv \neg v$.
From this and $(\exists x.\, \phi)^{v}_{s,u} = \neg v$, it follows that $(\exists x.\, \phi)^{v}_{s,u}  \equiv (\neg \forall x.\, \neg \phi)^{v}_{s,u}$.
\end{compactenum}

\item $(\forall x.\, \phi)^{v}_{s,u} \equiv (\neg \exists x.\, \neg \phi)^{v}_{s,u}$.
The proof of this case is similar to that of $(\exists x.\, \phi)^{v}_{s,u} \equiv (\neg \forall x.\, \neg \phi)^{v}_{s,u}$.
\end{compactitem}
This completes the proof.
\end{proof}

\begin{proposition}\thlabel{theorem:rewriting:domain:independent:1} 
Let $M = \langle D, \Gamma\rangle$ be a system configuration, $s = \langle \mathit{db}, U, \mathit{sec},  T,V \rangle$ be an $M$-partial state, $u \in U$ be a user,  $v \in \{\top,\bot\}$, and $\phi$ be a formula.
Furthermore, let $x \in \mathit{free}(\phi) \cap \mathit{free}(\phi^{v}_{s,u})$.
If $\mathit{gen}(x,\phi)$ holds, then  $\mathit{gen}(x, \phi^{v}_{s,u})$ holds.
\end{proposition}

\begin{proof}
Let $M = \langle D, \Gamma\rangle$ be a system configuration, $s = \langle \mathit{db}, U, \mathit{sec},  T,V \rangle$ be an $M$-partial state, $u \in U$ be a user,  $v \in \{\top,\bot\}$, and $\phi$ be a formula.
Furthermore, let $x \in \mathit{free}(\phi) \cap \mathit{free}(\phi^{v}_{s,u})$.
We prove our claim by structural induction on the length of $\phi$.
In the following, the length of $\phi$ is the number of predicates, quantifiers, negations, conjunctions, and disjunctions in $\phi$.

\smallskip
\noindent
{\bf Base Case} 
There are four cases:
\begin{compactenum}
\item $\phi := x = y$. 
In this case, the claim holds trivially.

\item $\phi := \top$. 
In this case, the claim holds trivially.

\item $\phi := \bot$. 
In this case, the claim holds trivially.

\item $\phi := R(\overline{x})$.
Assume $\mathit{gen}(x,\phi)$ holds.
From this, it follows that $x$ is one of the free variables in $\overline{x}$.
Furthermore, from $x \in \mathit{free}(\phi^{v}_{s,u})$, it follows that $R^{v}_{s,u} \neq \emptyset$.
There are two cases:
\begin{compactenum}
\item $\phi^{v}_{s,u}$ is a conjunction of predicates $S(\overline{x})$ such that $\mathit{gen}(x,S(\overline{x}))$ holds.
From the rule \emph{And 1}, it follows that $\mathit{gen}(x,\phi^{v}_{s,u})$ holds.

\item $\phi^{v}_{s,u}$ is a disjunction of predicates $S(\overline{x})$ such that $\mathit{gen}(x,S(\overline{x}))$ holds.
From the rule \emph{Or}, it follows that $\mathit{gen}(x,\phi^{v}_{s,u})$ holds.
\end{compactenum}

\end{compactenum}
This completes the proof of the base case.

\smallskip
\noindent
{\bf Induction Step} 
Assume that our claim holds for all formulae whose length is less than $\phi$'s length.
We now show that our claim holds also for $\phi$.
There are a number of cases depending on $\phi$'s structure.
\begin{compactenum}
\item $\phi := \psi \wedge \gamma$.
Assume that $\mathit{gen}(x,\phi)$ holds.
From this and the rules \emph{And 1} and \emph{And 2}, it follows that either $\mathit{gen}(x,\psi)$ or $\mathit{gen}(x,\gamma)$ hold. 
Assume, without loss of generality, that  $\mathit{gen}(x,\psi)$ holds.
From this and the induction hypothesis, it follows that $\mathit{gen}(x,\psi^{v}_{s,u})$ holds.
From this, $\phi^{v}_{s,u} := \psi^{v}_{s,u} \wedge \gamma^{v}_{s,u}$, and the rule \emph{And 1}, it follows that $\mathit{gen}(x,\phi^{v}_{s,u})$ holds.

\item $\phi := \psi \vee \gamma$.
Assume that $\mathit{gen}(x,\phi)$ holds.
From this and the rule \emph{Or}, it follows that both $\mathit{gen}(x,\psi)$ and $\mathit{gen}(x,\gamma)$ hold. 
From this and the induction hypothesis, it follows that both $\mathit{gen}(x,\psi^{v}_{s,u})$ and $\mathit{gen}(x,\gamma^{v}_{s,u})$ hold.
From this, $\phi^{v}_{s,u} := \psi^{v}_{s,u} \vee \gamma^{v}_{s,u}$, and the rule \emph{Or}, it follows that $\mathit{gen}(x,\phi^{v}_{s,u})$ holds.

\item $\phi := \neg \psi$.
Assume that $\mathit{gen}(x,\phi)$ holds.
From this and the rule \emph{Not}, it follows that $\mathit{gen}(x, \mathit{push}(\neg \psi))$.
There are a number of cases depending on $\psi$.
In the following, we exploit standard relational calculus equivalences, see, for instance, \cite{abiteboul1995foundations}, and the equivalences we proved in \thref{theorem:rewriting:equivalences}.
\begin{compactenum}
\item $\psi := (\alpha \wedge \beta)$.
In this case, $\mathit{push}(\neg \psi)$ is $(\neg\alpha \vee \neg\beta)$.
From this and $\mathit{gen}(x, \mathit{push}(\neg \psi))$, it follows $\mathit{gen}(x, (\neg\alpha \vee \neg\beta))$.
From this and the \emph{Or} rule, it follows that $\mathit{gen}(x, \neg\alpha)$ and $\mathit{gen}(x, \neg\beta)$ hold.
From this and the induction hypothesis, it follows that $\mathit{gen}(x, (\neg\alpha)^{v}_{s,u})$ and $\mathit{gen}(x, (\neg\beta)^{v}_{s,u})$.
From this and the \emph{Or} rule, it follows that $\mathit{gen}(x, (\neg\alpha)^{v}_{s,u} \vee (\neg\beta)^{v}_{s,u})$.
From this, $(\neg\alpha)^{v}_{s,u} \vee (\neg\beta)^{v}_{s,u} \equiv (\neg\alpha \vee \neg\beta)^{v}_{s,u}$, and the \emph{Equiv} rule, it follows that $\mathit{gen}(x, (\neg\alpha \vee \neg\beta)^{v}_{s,u})$.
From this, $(\neg\alpha \vee \neg\beta)^{v}_{s,u} \equiv (\neg(\alpha \wedge \beta))^{v}_{s,u}$, and the \emph{Equiv} rule, it follows that $\mathit{gen}(x,  (\neg(\alpha \wedge \beta))^{v}_{s,u})$.
From this, $ (\neg(\alpha \wedge \beta))^{v}_{s,u} \equiv \neg(\alpha \wedge \beta)^{\neg v}_{s,u}$, and the \emph{Equiv} rule, it follows that $\mathit{gen}(x, \neg (\alpha \wedge \beta)^{\neg v}_{s,u})$.
From this and $\psi := \alpha \wedge \beta$, it follows that $\mathit{gen}(x, \neg \psi^{\neg v}_{s,u})$.
From this and $\phi^{v}_{s,u} := \neg \psi^{\neg v}_{s,u}$, it follows that $\mathit{gen}(x, \\ \phi^{v}_{s,u})$ holds.

\item $\psi := (\alpha \vee \beta)$. The proof is similar to the $\psi := \neg(\alpha \wedge \beta)$ case.

\item $\psi :=  \exists y.\,\alpha$.
In this case, $\mathit{push}(\neg \psi)$ is $\forall y.\, \neg\alpha$.
From this and $\mathit{gen}(x, \mathit{push}(\neg \psi))$, it follows $\mathit{gen}(x, \forall y.\, \neg\alpha)$.
From this and the induction hypothesis, it follows that  $\mathit{gen}(x, (\forall y.\, \neg\alpha)^{v}_{s,u})$.
From this, $\neg \neg (\forall y.\, \neg\alpha)^{v}_{s,u} \\ \equiv (\forall y.\, \neg\alpha)^{v}_{s,u}$, and the \emph{Equiv} rule, it follows that $\mathit{gen}(x, \neg \neg (\forall y.\, \neg\alpha)^{v}_{s,u})$.
From this, $\neg \neg (\forall y.\, \neg\alpha)^{v}_{s,u} \equiv \neg (\neg \forall y.\, \neg\alpha)^{\neg v}_{s,u}$, and the \emph{Equiv} rule, it follows that $\mathit{gen}(x, \neg (\neg \forall y.\, \neg\alpha)^{\neg v}_{s,u})$.
From this, $\neg (\neg \forall y.\, \neg\alpha)^{\neg v}_{s,u} \equiv \neg (\exists y.\, \neg\neg\alpha)^{\neg v}_{s,u}$, and the \emph{Equiv} rule, it follows that $\mathit{gen}(x, \neg (\exists y.\, \neg\neg\alpha)^{\neg v}_{s,u})$.
From this, $\neg (\exists y.\, \neg\neg\alpha)^{\neg v}_{s,u} \equiv \neg (\exists y.\, \alpha)^{\neg v}_{s,u}$, and the \emph{Equiv} rule, it follows that $\mathit{gen}(x, \\ \neg (\exists y.\, \alpha)^{\neg v}_{s,u})$.
From this and $\psi :=  \exists y.\,\alpha$, it follows that $\mathit{gen}(x, \neg \psi^{\neg v}_{s,u})$.
From this and $\phi^{v}_{s,u} := \neg \psi^{\neg v}_{s,u}$, it follows that $\mathit{gen}(x,\phi^{v}_{s,u})$ holds.

\item $\psi :=  \forall y.\,\alpha$. The proof is similar to the $\psi := \neg \exists y.\,\alpha$ case.

\item $\psi :=  \neg \alpha$.
In this case, $\mathit{push}(\neg \psi)$ is $\alpha$.
From this and $\mathit{gen}(x, \mathit{push}(\neg \psi))$, it follows $\mathit{gen}(x, \alpha)$.
From this and the induction hypothesis, it follows that  $\mathit{gen}(x, \alpha^{v}_{s,u})$.
From this, $\neg \neg \alpha^{v}_{s,u} \equiv \alpha^{v}_{s,u}$, and the \emph{Equiv} rule, it follows that $\mathit{gen}(x, \neg \neg \alpha^{v}_{s,u})$.
From this, $\neg \neg \alpha^{v}_{s,u} \equiv \neg (\neg \alpha)^{\neg v}_{s,u}$, and the \emph{Equiv} rule, it follows that $\mathit{gen}(x, \neg (\neg \alpha)^{\neg v}_{s,u})$.
From this and $\psi :=  \neg \alpha$, it follows that $\mathit{gen}(x, \neg \psi^{\neg v}_{s,u})$.
From this and $\phi^{v}_{s,u} := \neg \psi^{\neg v}_{s,u}$, it follows that $\mathit{gen}(x,\phi^{v}_{s,u})$ holds.

\item $\psi := x = y$. The proof for this case is trivial.

\item $\psi := x \neq y$. The proof for this case is trivial.

\end{compactenum}

\item $\phi := \exists x.\, \psi$.
Assume that $\mathit{gen}(x,\phi)$ holds.
From this and the rule \emph{Exists}, it follows that $\mathit{gen}(x,\psi)$ holds. 
From this and the induction hypothesis, it follows that $\mathit{gen}(x,\psi^{v}_{s,u})$ holds.
From this, $\phi^{v}_{s,u} := \exists x.\, \psi^{v}_{s,u}$, and the rule \emph{Exists}, it follows that $\mathit{gen}(x,\phi^{v}_{s,u})$ holds.

\item $\phi := \forall x.\, \psi$.
The proof of this case is similar to that of $\phi:= \exists x.\, \psi$.
\end{compactenum}
This completes the proof of the induction step.

This completes the proof of our claim.
\end{proof}

\begin{proposition}\thlabel{theorem:rewriting:domain:independent:2} 
Let $M = \langle D, \Gamma\rangle$ be a system configuration, $s = \langle \mathit{db}, U, \mathit{sec},  T,V \rangle$ be an $M$-partial state, $u \in U$ be a user,  $v \in \{\top,\bot\}$, and $\phi$ be a formula.
For every sub-formula $\exists x.\, \psi$ of $\phi$, if $\mathit{gen}(x,\psi)$ holds and $x \in \mathit{free}(\psi) \cap \mathit{free}(\psi^{v}_{s,u})$, then $\mathit{gen}(x,\psi^{v}_{s,u})$ holds.
\end{proposition}

\begin{proof}
Let $M = \langle D, \Gamma\rangle$ be a system configuration, $s = \langle \mathit{db}, U, \mathit{sec},  T,V \rangle$ be an $M$-partial state, $u \in U$ be a user,  $v \in \{\top,\bot\}$, and $\phi$ be a formula.
We prove our claim by structural induction on the length of $\phi$.
In the following, the size of $\phi$ is the number of predicates, quantifiers, negations, conjunctions, and disjunctions in $\phi$.

\smallskip
\noindent
{\bf Base Case} 
The claim is vacuously satisfied for the base cases as there is no sub-formula of the form $\exists x.\, \psi$.

\smallskip
\noindent
{\bf Induction Step} 
Assume that our claim holds for all formulae whose length is less than $\phi$.
We now show that our claim holds also for $\phi$.
There are a number of cases depending on $\phi$'s structure.
\begin{compactenum}
\item $\phi := \psi \wedge \gamma$.
Let $\alpha$ be a sub-formula of $\phi$ of the form $\exists x.\, \beta$ such that $\mathit{gen}(x,\beta)$ holds and $x \in \mathit{free}(\beta) \cap \mathit{free}(\beta^{v}_{s,u})$.
The formula $\alpha$ is either a sub-formula of $\psi$ or a sub-formula of $\gamma$.
From this and the induction hypothesis, it follows that $\mathit{gen}(x,\beta^{s}_{v,u})$ holds.

\item $\phi := \psi \vee \gamma$.
The proof of this case is similar to that of $\phi := \psi \wedge \gamma$.

\item $\phi := \neg \psi$.
Let $\alpha$ be a sub-formula of $\phi$ of the form $\exists x.\, \beta$ such that $\mathit{gen}(x,\beta)$ holds and $x \in \mathit{free}(\beta) \cap \mathit{free}(\beta^{v}_{s,u})$.
Since $\phi := \neg \psi$, the formula $\alpha$ is also a sub-formula of $\psi$.
From this and the induction hypothesis, it follows that $\mathit{gen}(x,\beta^{s}_{v,u})$ holds.

\item $\phi := \exists x.\, \psi$.
Let $\alpha$ be a sub-formula of $\phi$ of the form $\exists x.\, \beta$ such that $\mathit{gen}(x,\beta)$ holds and $x \in \mathit{free}(\beta) \cap \mathit{free}(\beta^{v}_{s,u})$.
There are two cases:
\begin{compactenum}
\item $\alpha$ is a sub-formula of $\psi$. 
From this and the induction hypothesis, it follows that $\mathit{gen}(x,\beta^{s}_{v,u})$ holds.

\item $\alpha = \phi$. 
From $\mathit{gen}(x,\beta)$, $x \in \mathit{free}(\beta) \cap \mathit{free}(\beta^{v}_{s,u})$, and  \thref{theorem:rewriting:domain:independent:1}, it follows that $\mathit{gen}(x,\beta^{v}_{s,u})$ holds.
\end{compactenum}

\item $\phi := \forall x.\, \psi$.
The proof of this case is similar to that of $\phi:= \exists x.\, \psi$.
\end{compactenum}
This completes the proof of the induction step.

This completes the proof of our claim.
\end{proof}

\begin{proposition}\thlabel{theorem:rewriting:domain:independent:3} 
Let $M = \langle D, \Gamma\rangle$ be a system configuration, $s = \langle \mathit{db}, U, \mathit{sec},  T,V \rangle$ be an $M$-partial state, $u \in U$ be a user,  $v \in \{\top,\bot\}$, and $\phi$ be a formula.
For every sub-formula $\forall x.\, \psi$ of $\phi$, if $\mathit{gen}(x,\psi)$ holds and $x \in \mathit{free}(\psi) \cap \mathit{free}(\psi^{v}_{s,u})$, then $\mathit{gen}(x,(\neg \psi)^{v}_{s,u})$ holds.
\end{proposition}

\begin{proof}
The proof is similar to that of \thref{theorem:rewriting:domain:independent:2}.
\end{proof}

\begin{proposition}\thlabel{theorem:rewriting:domain:independent:4} 
Let $M = \langle D, \Gamma\rangle$ be a system configuration, $s = \langle \mathit{db}, U, \mathit{sec},  T,V \rangle$ be an $M$-partial state, $u \in U$ be a user,  $v \in \{\top,\bot\}$, and $\phi$ be a formula.
Let $Q \in \{\exists, \forall\}$ be a quantifier and $\mathit{subs}_{Q}(\phi)$ be the set of sub-formulae of $\phi$ of the form $Q\, x.\, \psi$.
There is a surjective function $f$ from $\mathit{subs}_{Q}(\phi)$  to $\mathit{subs}_{Q}(\phi^{v}_{s,u})$ such that for any  $Q\, x.\, \psi$ in $\mathit{subs}_{Q}(\phi)$, if $f(Q\, x.\, \psi)$ is defined, then $f(Q\, x.\, \psi)^{v}_{s,u} = Q\, x.\, \psi^{v}_{s,u}$. 
\end{proposition}

\begin{proof}
The claim follows trivially from the definition of $\phi^{v}_{s,u}$.
\end{proof}

\begin{lemma}\thlabel{theorem:rewriting:domain:independence}
Let $M = \langle D, \Gamma\rangle$ be a system configuration, $s = \langle \mathit{db}, U, \mathit{sec},  T,V \rangle$ be an $M$-partial state, $u \in U$ be a user, and $\phi$ be a formula.
If $\phi$ is allowed and all views in $V$ are allowed, then $\phi^{\top}_{s,u}$, $\phi^{\bot}_{s,u}$, and  $\phi^{\mathit{rw}}_{s,u}$ are domain independent.
\end{lemma}

\begin{proof}
From \thref{theorem:rewriting:domain:independent:1}, \thref{theorem:rewriting:domain:independent:2}, \thref{theorem:rewriting:domain:independent:3}, and \thref{theorem:rewriting:domain:independent:4}, it follows that if $\phi$ is allowed, then  both $\phi^{\top}_{s,u}$ and $\phi^{\bot}_{s,u}$ are allowed.
Since every allowed formula is domain independent \cite{VanGelder:1991:STR:114325.103712}, it follows that both $\phi^{\top}_{s,u}$ and $\phi^{\bot}_{s,u}$ are domain independent.
Finally, the domain independence of $\phi^{\mathit{rw}}_{s,u}$  follows easily from its definition and the domain independence of $\phi^{\top}_{s,u}$ and $\phi^{\bot}_{s,u}$.
\end{proof}

We now prove the main result of this section, namely that the $\mathit{secure}$ function is, indeed, a sound, under-approximation of the notion of judgment's security.

\begin{lemma}\thlabel{theorem:secure:sound:under:approximation}
Let $P = \langle M, f \rangle$ be an \accessControlConfiguration{}, $L$ be the $P$-LTS, $u \in {\cal U}$ be a user, $r \in \mathit{traces}(L)$ be an $L$-run, $\phi \in RC_{\mathit{bool}}$ is a sentence, and $1 \leq i \leq |r|$.
Furthermore, let $s$ be the  $i$-th state in $r$.
The following statements hold:
\begin{compactenum}
\item Given a judgment $ r, i \attMod \phi $, if $\mathit{secure}(u,\phi,s) = \top$, then $\mathit{secure}^{\mathit{data}}_{P,u}(r, i \attMod \phi)$ holds. 
\item Given a judgment $ r, i \attMod \phi$, if $\mathit{secure}(u,\phi,s) = \top$, then $\mathit{secure}_{P,u}(r, i \attMod \phi)$ holds.
\end{compactenum}
\end{lemma}

\begin{proof}
Note that the second statement follows trivially from \thref{theorem:ibsec:correctness:secure:3} and the first statement.
Therefore, in the following we prove just that given a judgment $r, i \attMod \phi$, if $\mathit{secure}(u,\phi,s) = \top$, then $\mathit{secure}^{\mathit{data}}_{P,u}(r, i \attMod \phi)$ holds. 

Let $P = \langle M, f \rangle$ be an \accessControlConfiguration{}, $L$ be the $P$-LTS, $u \in {\cal U}$ be a user, $r \in \mathit{traces}(L)$ be an $L$-run, $\phi \in RC_{\mathit{bool}}$ is a sentence, and $1 \leq i \leq |r|$.
Furthermore, let $s = \langle \mathit{db}, U, \mathit{sec}, T, V, c \rangle$ be the  $i$-th state in $r$.
Assume that $\mathit{secure}(u,\phi,s) = \top$.
From this and $\emph{secure}$'s definition, $[\phi^{\mathit{rw}}_{s,u}]^{\mathit{db}} \\ = \bot$.
In the following, with a slight abuse of notation we ignore the $\mathit{inline}$ and $\mathit{ext}$ functions in $\phi^{\mathit{rw}}_{s,u}$'s definition.
This is without loss of generality since $\mathit{inline}$ and $\mathit{ext}$ do not modify $\phi$'s result.
From this and $\phi^{\mathit{rw}}_{s,u}$'s definition, it follows that either $[\phi^{\mathit{\top}}_{s,u}]^{\mathit{db}} = \top$ or $[\phi^{\mathit{\bot}}_{s,u}]^{\mathit{db}} = \bot$.
Note that from \thref{theorem:rewriting:secure}, it follows that $\mathit{secure}^{\mathit{data}}_{P,u}(r, i \attMod \phi^{\top}_{s,u})$ and  $\mathit{secure}^{\mathit{data}}_{P,u}(r, i \attMod \phi^{\bot}_{s,u})$.
Furthermore, let $\Delta$ be the equivalence class $\llbracket \mathit{pState}(s)\rrbracket_{u,M}^\mathit{data}$.
There are two cases:
\begin{compactenum}
\item $[\phi^{\mathit{\top}}_{s,u}]^{\mathit{db}} = \top$. 
From $\mathit{secure}^{\mathit{data}}_{P,u}(r, i \attMod \phi^{\top}_{s,u})$, it follows that for all $s', s'' \in \Delta$, $[\phi^{\top}_{s,u}]^{s'.\mathit{db}} = [\phi^{\top}_{s,u}]^{s''.\mathit{db}}$.
From this, $s \in \Delta$, and $[\phi^{\mathit{\top}}_{s,u}]^{\mathit{db}} = \top$, it follows that  $[\phi^{\top}_{s,u}]^{s'.\mathit{db}} = \top$ for all $s' \in \Delta$.
From \thref{theorem:rewriting:equivalent:modulo:indistinguishable:state}, it follows that for all $s', s'' \in \Delta$, $\phi^{\top}_{s,u} = \phi^{\top}_{s',u} = \phi^{\top}_{s'',u}$.
From this and the fact that for all $s' \in \Delta$, $[\phi^{\top}_{s,u}]^{s'.\mathit{db}} = \top$, it follows that for all $s' \in \Delta$, $[\phi^{\top}_{s',u}]^{s'.\mathit{db}} = \top$.
From this and \thref{theorem:rewriting:invariants}, it follows that for all $s' \in \Delta$, $[\phi]^{s'.\mathit{db}} = \top$.
From this, it follows that for all $s', s'' \in \Delta$, $[\phi]^{s'.\mathit{db}} = [\phi]^{s''.\mathit{db}}$.
From this, $r$'s definition, and $\mathit{secure}^{\mathit{data}}_{P,u}$, it follows that $\mathit{secure}^{\mathit{data}}_{P,u}(r,i \attMod \phi)$.

\item $[\phi^{\mathit{\bot}}_{s,u}]^{\mathit{db}} = \bot$.
From $\mathit{secure}^{\mathit{data}}_{P,u}(r, i \attMod \phi^{\bot}_{s,u})$, it follows that for all $s', s'' \in \Delta$, $[\phi^{\bot}_{s,u}]^{s'.\mathit{db}} = [\phi^{\bot}_{s,u}]^{s''.\mathit{db}}$.
From this, $s \in \Delta$, and $[\phi^{\mathit{\bot}}_{s,u}]^{\mathit{db}} = \bot$, it follows that  $[\phi^{\bot}_{s,u}]^{s'.\mathit{db}} = \bot$ for all $s' \in \Delta$.
From \thref{theorem:rewriting:equivalent:modulo:indistinguishable:state}, it follows that for all $s', s'' \in \Delta$, $\phi^{\bot}_{s,u} = \phi^{\bot}_{s',u} = \phi^{\bot}_{s'',u}$.
From this and the fact that for all $s' \in \Delta$, $[\phi^{\bot}_{s,u}]^{s'.\mathit{db}} = \bot$, it follows that for all $s' \in \Delta$, $[\phi^{\bot}_{s',u}]^{s'.\mathit{db}} = \bot$.
From this and \thref{theorem:rewriting:invariants}, it follows that for all $s' \in \Delta$, $[\phi]^{s'.\mathit{db}} = \bot$.
From this, it follows that for all $s', s'' \in \Delta$, $[\phi]^{s'.\mathit{db}} = [\phi]^{s''.\mathit{db}}$.
From this, $r$'s definition, and $\mathit{secure}^{\mathit{data}}_{P,u}$, it follows that $\mathit{secure}^{\mathit{data}}_{P,u}(r,i \attMod \phi)$.
\end{compactenum}
This completes the proof of our claim.
\end{proof}

\thref{theorem:secure:equivalent:modulo:indistinguishable:state} proves that the $\emph{secure}$ function produces the same result for any two indistinguishable states.

\begin{lemma}\thlabel{theorem:secure:equivalent:modulo:indistinguishable:state}
Let $M$ be a system configuration, $u \in U$ be a user, $s,s' \in \Omega_{M}$ be two $M$-states such that $\mathit{pState}(s) \cong_{u, M}^{\mathit{data}} \mathit{pState}(s')$,  and $\phi$ be a sentence.
Then, $\mathit{secure}(u,\phi, s) = \top$ iff $\mathit{secure}(u,\phi, s') = \top$.
\end{lemma}

\begin{proof}
Let $M$ be a system configuration, $u \in U$ be a user, $s = \langle \mathit{db}, U, \mathit{sec}, T, V, c\rangle$ and $s' = \langle \mathit{db}', U', \mathit{sec}', T', V', c'\rangle$ be two $M$-states such that $\mathit{pState}(s) \cong_{u,M}^{\mathit{data}} \mathit{pState}(s')$,  and $\phi$ be a sentence.
We now prove that  $\mathit{secure}(u,\phi, s) = \mathit{secure}(u, \\ \phi, s')$.
Assume, for contradiction's sake, that  $\mathit{secure}  (u,\phi, s) \neq \mathit{secure}  (u,\phi, s')$.
From this, it follows that $[\phi^{\mathit{rw}}_{s,u}]^{db}  \neq [\phi^{\mathit{rw}}_{s',u}]^{\mathit{db}'}$.
From $\mathit{pState}(s) \cong_{u,M}^{\mathit{data}} \mathit{pState}(s')$ and \thref{theorem:rewriting:equivalent:modulo:indistinguishable:state}, it follows that $\phi^{\mathit{rw}}_{s,u} = \phi^{\mathit{rw}}_{s',u}$.
From this and $[\phi^{\mathit{rw}}_{s,u}]^{db} \neq [\phi^{\mathit{rw}}_{s',u}]^{\mathit{db}'}$, it follows that $[\phi^{\mathit{rw}}_{s,u}]^{db}  \neq [\phi^{\mathit{rw}}_{s,u}]^{\mathit{db}'}$.
This contradicts $\mathit{secure}^{\mathit{data}}_{P,u}(r,i \\ \attMod \phi^{\mathit{rw}}_{s,u})$, which has been proved in \thref{theorem:rewriting:secure}.
This completes the proof of our claim.
\end{proof}

\subsection{Data Confidentiality Proofs}\label{sect:data:conf:proofs}
In this section, we first  prove some simple results about $f_{\mathit{conf}}^{u}$.
Afterwards, we prove our main result, namely that $f_{\mathit{conf}}^{u}$ provides \confidentiality{} with respect to the user $u$.

\begin{lemma}\thlabel{theorem:f:conf:soundness}
Let $M = \langle D, \Gamma\rangle$ be a system configuration, $u \in {\cal U}$ be a user, $s,s' \in \Omega_{M}$ be two $M$-states such that $\mathit{pState}(s) \cong_{u,M}^{\mathit{data}} \mathit{pState}(s')$, $\mathit{invoker}(s) = \mathit{invoker}(s')$, and $\mathit{tr}(s) = \mathit{tr}(s')$,  and $a$ be an action in ${\cal A}_{D,{\cal U}}$.
Then, $f^{u}_{\mathit{conf}}(s,a) 
\\ = f^{u}_{\mathit{conf}}(s',a)$.
\end{lemma}
\begin{proof}

Let $M = \langle D, \Gamma\rangle$ be a system configuration, $u \in {\cal U}$ be a user, $s =\langle \mathit{db}, U, \mathit{sec}, T, V, c \rangle$ and $s' =\langle \mathit{db}', U', \mathit{sec}', T', V', \\ c' \rangle$  be two $M$-states such that $\mathit{pState}(s) \cong_{u,M}^{\mathit{data}} \mathit{pState}(s')$, $\mathit{invoker}(s) = \mathit{invoker}(s')$, and $\mathit{tr}(s) = \mathit{tr}(s')$,  and $a$ be an action in ${\cal A}_{D,u}$.
There are a number of cases depending on the action $a$.
\begin{compactenum}
\item $a = \langle u', \mathtt{SELECT}, \phi \rangle$.
Assume, for contradiction's sake, that $f^{u}_{\mathit{conf},}(s,a) \neq f^{u}_{\mathit{conf}}(s',a)$.
This happens iff $\mathit{secure}(u, \\ \phi,s) \neq \mathit{secure}(u,\phi,s')$.
This contradicts \thref{theorem:secure:equivalent:modulo:indistinguishable:state} because $\mathit{pState}(s) \cong_{u,M}^{\mathit{data}} \mathit{pState}(s')$.

\item $a = \langle u', \mathtt{INSERT}, R, \overline{t} \rangle$.
We claim that $\mathit{noLeak}(s,a,u) =\mathit{noLeak}(s',a,u)$.
Assume, for contradiction's sake, that $f^{u}_{\mathit{conf},}(s,a) \neq f^{u}_{\mathit{conf}}(s',a)$.
This happens iff there is a formula $\phi$, which has been derived using the $\mathit{getInfo}$, $\mathit{getInfoV}$, or $\mathit{getInfoD}$ functions, such that $\mathit{secure}(u,  \phi, \\ s) \neq \mathit{secure}(u,\phi,s')$.
This contradicts \thref{theorem:secure:equivalent:modulo:indistinguishable:state} because $\mathit{pState}(s) \cong_{u,M}^{\mathit{data}} \mathit{pState}(s')$.

We prove our claim that $\mathit{noLeak}(s,a,u) =\mathit{noLeak}(s',a, \\ u)$ for any two states $s$ and $s'$ such that $\mathit{pState}(s) \cong_{u,M}^{\mathit{data}} \mathit{pState}(s')$.
Assume, for contradiction's sake, that this is not the case.
Without loss of generality we assume that $\mathit{noLeak}(s,a,u) = \top$ and $\mathit{noLeak}(s',a,u) = \bot$.
From $\mathit{noLeak}(s,a,u) = \top$, it follows that for all views $V$ such that  $\langle \oplus, \mathtt{SELECT}, V \rangle \in \mathit{permissions}(s,u)$ and $R \in \mathit{tDet}(V, s,M)$, for all $o \in \mathit{tDet} (V, s,M)$,  $\langle \oplus, \mathtt{SELECT}, o \rangle$ is in $\mathit{permissions}(s,u)$.
From $\mathit{pState}(s) \cong_{u,M}^{\mathit{data}} \mathit{pState}(s')$, it follows that $\mathit{sec} = \mathit{sec}'$.
From this, $\mathit{permissions}(s,u) \\ = \mathit{permissions}(s',  u)$.
From $\mathit{noLeak}(s',a,u) = \bot$, there  are two views $V'$ and $o$ such that $\langle \oplus, \mathtt{SELECT}, V' \rangle \in \mathit{permissions}(s',  u)$, $\langle \oplus, \mathtt{SELECT}, o \rangle \not\in \mathit{permissions}(s',u)$, $R \in \mathit{tDet}(V',  s',M)$, and $o \in \mathit{tDet} (V',s',M)$.
Note that $\mathit{tDet} (V', s',M) = \mathit{tDet} (V',s,M)$ because query determinacy does not consider the database state.
From this and $\mathit{permissions}(s,u) = \mathit{permissions}(s', u)$, it follows that there is a view $V'$ such that $\langle \oplus, \mathtt{SELECT}, V' \rangle \in \mathit{permissions}(s,u)$ and $R \in \mathit{tDet}(V',s,M)$, such that there is a table $o \in \mathit{tDet} (V',s,M)$ for which $\langle \oplus, \mathtt{SELECT}, \\ o \rangle \not\in \mathit{permissions}(s,u)$.
This contradicts $\mathit{noLeak}(s,a,u)\\ = \top$.

\item $a = \langle u', \mathtt{DELETE}, R, \overline{t} \rangle$.
The proof of this case is similar to the  $a = \langle u', \mathtt{INSERT}, R, \overline{t} \rangle$ case.

\item $a = \langle \mathit{op}, u'', p, u' \rangle$, where $\mathit{op} \in \{\oplus, \oplus^*\}$.
Assume, for contradiction's sake, that $f^{u}_{\mathit{conf},}(s,a) \neq f^{u}_{\mathit{conf}}(s',a)$.
Note that this happens iff $p = \langle \mathtt{SELECT}, o \rangle$ for some $o$.
Without loss of generality, we further assume that $f^{u}_{\mathit{conf},}(s,a) \\ = \top$ and $f^{u}_{\mathit{conf}}(s',a) = \bot$.
From $f^{u}_{\mathit{conf},}(s,a) = \top$, it follows that $\langle \oplus, \mathtt{SELECT}, o \rangle \in \mathit{permissions}(s,u)$.
From $\mathit{pState}(s) \cong_{u,M}^{\mathit{data}} \mathit{pState}(s')$,  it follows  $\mathit{permissions} (s,u) \\ = \mathit{permissions}(s',  u)$.
From this and $\langle \oplus, \mathtt{SELECT}, o \rangle \in \mathit{permissions}(s,u)$, it follows that $\langle \oplus, \mathtt{SELECT}, o \rangle$ is in $\mathit{permissions}(s',u)$.
From $f^{u}_{\mathit{conf},}(s',a) = \bot$, it follows that $\langle \oplus, \mathtt{SELECT}, o \rangle \not\in \mathit{permissions}(s,u)$.
This contradicts $\langle \oplus, \mathtt{SELECT}, o \rangle \in \mathit{permissions}(s',u)$.

\item For any other action $a$, the proof is trivial.
\end{compactenum}
This completes the proof of our claim.
\end{proof}

\begin{lemma}\thlabel{theorem:reasoning:preserves:security}
Let $P$ be an \accessControlConfiguration{}, $L$ be the $P$-LTS, $r \in \mathit{traces}(L)$ be a run, $u$ be a user, $\gamma$ be a sentence, and $\Phi$ be a set of sentences such that $\Phi \models_{\mathit{fin}} \gamma$.
If, for all $\phi \in \Phi$,  $\mathit{secure}_{P,u}(r,i \attMod \phi)$  holds and $[\phi]^{\mathit{last}(r).\mathit{db}} = \top$, then $\mathit{secure}_{P,u}(r,i \attMod \gamma)$ holds and $[\gamma]^{\mathit{last}(r^{i}).\mathit{db}} = \top$.
\end{lemma}

\begin{proof}
Let $P$ be an \accessControlConfiguration{}, $L$ be the $P$-LTS, $r \in \mathit{traces}(L)$ be a run, $u$ be a user, $\gamma$ be a sentence, and $\Phi$ be a set of sentences such that $\Phi \models_{\mathit{fin}} \gamma$ such that for all $\phi \in \Phi$,  $\mathit{secure}_{P,u}(r,i \attMod \phi)$ holds and $[\phi]^{\mathit{last}(r^{i}).\mathit{db}} = \top$.
We now show that $\mathit{secure}_{P,u}(r,i \attMod \gamma)$ holds and $[\gamma]^{\mathit{last}(r^{i}).\mathit{db}} = \top$.
From $\Phi \models_{\mathit{fin}} \gamma$, the fact that for all $\phi \in \Phi$, $[\phi]^{\mathit{last}(r^{i}).\mathit{db}} = \top$, and $ \models_{\mathit{fin}}$'s definition, it follows that $[\gamma]^{\mathit{last}(r^{i}).\mathit{db}} = \top$.
Assume, for contradiction's sake, that $\mathit{secure}_{P,u}(r,i \attMod \gamma)$ does not hold.
From this and $[\gamma]^{\mathit{last}(r^{i}).\mathit{db}} = \top$, it follows that there is a run $r' \in \mathit{traces}(L)$ such that $r^{i} \cong_{P,u} r'$ such that $[\gamma]^{\mathit{last}(r').\mathit{db}} = \bot$.
We claim that for all $\phi \in \Phi$, $[\phi]^{\mathit{last}(r').\mathit{db}} = \top$.
From this and $\Phi \models_{\mathit{fin}} \gamma$, it follows that $[\gamma]^{\mathit{last}(r').\mathit{db}} = \top$, which contradicts $[\gamma]^{\mathit{last}(r').\mathit{db}} = \bot$.

We now prove our claim that for all $\phi \in \Phi$, $[\phi]^{\mathit{last}(r').\mathit{db}} = \top$ for any trace $r'$ such that $r^{i} \cong_{P,u} r'$.
From $\mathit{secure}_{P,u}(r, i \attMod \phi)$, it follows that $[\phi]^{\mathit{last}(r^{i}).\mathit{db}} = [\phi]^{\mathit{last}(r').\mathit{db}}$.
From this and $[\phi]^{\mathit{last}(r^{i}).\mathit{db}} = \top$, it follows that $[\phi]^{\mathit{last}(r').\mathit{db}} = \top$.
\end{proof}

Before proving our main result, namely that $f_{\mathit{conf}}^{u}$ provides \confidentiality{} for the user $u$, we introduce the concept of an action that preserves the equivalence class induced by the indistinguishability relation $\cong_{P,u}$.

\begin{definition}\label{def:preserve:equiv:class}
Let $P = \langle M, f \rangle$ be an \accessControlConfiguration{}, where $M = \langle D,\Gamma\rangle$ is a system configuration and $f$ is an $M$-\acf{}, $L$ be the $P$-LTS, $r \in \mathit{traces}(L)$ be a run, $a$ be an action in ${\cal A}_{D,{\cal U}} \cup {\cal TRIGGER}_{D}$, and $u$ be a user in ${\cal U}$.
We denote by $\mathit{extend}(r, a)$, where $r$ is a run and $a$ is an action, the run $r' \in \mathit{traces}(L)$, where $s \in \Omega_{M}$ and  $r' = r \concat a \concat s$, obtained by executing the action $a$ at the end of the run $r'$. If there is no such run, then $\mathit{extend}(r,a)$ is undefined.
We say that $a$  \emph{preserves the equivalence class for $r$, $P$, and $u$} iff\begin{inparaenum}[(1)]
\item $\mathit{extend}(r,a)$ is defined, and
\item there is a bijection $b$ between $\llbracket r \rrbracket_{P,u}$ and $\llbracket \mathit{extend}(r,a) \rrbracket_{P,u}$ such that for all $r'  \in \llbracket r \rrbracket_{P,u}$, $\mathit{extend}(r',a)$ is defined and $b(r') = \mathit{extend}(r',a)$.
\end{inparaenum}
\end{definition}

\begin{lemma}\thlabel{theorem:secure:extend:on:runs:insert:delete}
Let $P = \langle M, f \rangle$ be an \accessControlConfiguration{}, where $M = \langle D,\Gamma\rangle$ is a system configuration and $f$ is an $M$-\acf{}, $L$ be the $P$-LTS, $u$ be a user in ${\cal U}$, $r$ be a run in $\mathit{traces}(L)$,  $a \in {\cal A}_{D,u}$ be an \texttt{INSERT} or \texttt{DELETE} action $\langle u, \mathit{op}, R, \overline{t} \rangle$,  $\phi$ be a sentence, and $i$ be such that $1 \leq i \leq |r|$, $\mathit{triggers}(\mathit{last}(r^{i})) = \epsilon$, and $r^{i+1} = \mathit{extend}(r^{i},a)$.
If\begin{inparaenum}[(1)]
\item $a$ preserves the equivalence class for $r^{i}$, $P$, and $u$, and
\item the execution of $a$ does not change any table in $\mathit{tables}(\phi)$ for any run $v \in \llbracket r^{i}\rrbracket_{P,u}$,
\end{inparaenum}
then $\mathit{secure}_{P,u}(r,i \attMod \phi)$ holds iff $\mathit{secure}_{P,u}(r,i+1 \attMod \phi)$ holds.
\end{lemma}

\begin{proof}
Let $P = \langle M, f \rangle$ be an \accessControlConfiguration{}, where $M = \langle D,\Gamma\rangle$ is a system configuration and $f$ is an $M$-\acf{}, $L$ be the $P$-LTS, $u$ be a user in ${\cal U}$, $r$ be a run in $\mathit{traces}(L)$,  $a \in {\cal A}_{D,u}$ be an \texttt{INSERT} or \texttt{DELETE} action $\langle u, \mathit{op}, R, \overline{t} \rangle$,  $\phi$ be a sentence, and $i$ be such that $1 \leq i \leq |r|$, $\mathit{triggers}(\mathit{last}(r^{i})) = \epsilon$, and $r^{i+1} = \mathit{extend}(r^{i},a)$.
Assume that\begin{inparaenum}[(1)]
\item  $a$ preserves the equivalence class for $r^{i}$, $P$, and $u$, and
\item the execution of $a$ does not change any table in $\mathit{tables}(\phi)$ for any run $v \in \llbracket r^{i}\rrbracket_{P,u}$. 
\end{inparaenum}
Without loss of generality, assume that $a$ is an \texttt{INSERT} action. 
In the following, we denote the $\mathit{extend}$ function by $e$.
Furthermore, we also denote the fact that $\mathit{secure}_{P,u}(r,i,u, \phi)$ does not hold as $\neg \mathit{secure}_{P,u}(r,i,u, \phi)$.
From Definition \ref{def:preserve:equiv:class} and $a$ preserves the equivalence class for $r^{i}$, $P$, and $u$, it follows that $e(r',a)$ is defined for any $r' \in \llbracket r^{i}\rrbracket_{P,u}$. 
Assume, for contradiction's sake, that our claim does not hold.
There are two cases:
\begin{compactitem}
\item $\mathit{secure}_{P,u}(r,i \attMod \phi)$ holds and $\mathit{secure}_{P,u}(r,i+1 \attMod \phi)$ does not hold.
From $\mathit{secure}_{P,u}(r,i \attMod \phi)$, it follows that for all $r' \in \llbracket r^{i}\rrbracket_{P,u}$, $[\phi]^{\mathit{last}(r').\mathit{db}} = [\phi]^{\mathit{last}(r^{i}).\mathit{db}}$.
We claim that $[\phi]^{\mathit{last}(r').\mathit{db}} = [\phi]^{\mathit{last}(e(r',a)).\mathit{db}}$ holds for any $r' \in \llbracket r^{i}\rrbracket_{P,u}$.
From this and  $[\phi]^{\mathit{last}(r').\mathit{db}} = [\phi]^{\mathit{last}(r^{i}).\mathit{db}}$ for all $r' \in \llbracket r^{i}\rrbracket_{P,u}$, it follows that $[\phi]^{\mathit{last}(r^{i}).\mathit{db}}  = [\phi]^{\mathit{last}(e(r',a)).\mathit{db}}$ holds for any $r' \in \llbracket r^{i}\rrbracket_{P,u}$.
From $\neg \mathit{secure}_{P,u}(r,i+1 \attMod \phi)$, it follows that there is a run $r' \in \llbracket r^{i+1}\rrbracket_{P,u}$ such that $[\phi]^{\mathit{last}(r^{i+1}).\mathit{db}} \neq [\phi]^{\mathit{last}(r').\mathit{db}}$.
From this, $[\phi]^{\mathit{last}(r').\mathit{db}} = [\phi]^{\mathit{last}(e(r',a)).\mathit{db}}$ for any $r' \in \llbracket r^{i}\rrbracket_{P,u}$, and $e(r^{i},a) = r^{i+1}$, it follows that $[\phi]^{\mathit{last}(r^{i}).\mathit{db}} \neq [\phi]^{\mathit{last}(r').\mathit{db}}$.
Let $b$ be the bijection showing that $a$ preserves the equivalence class with respect to $r^{i}$, $P$, and $u$.
From $e(r^{i},a) = r^{i+1}$ and $r' \in \llbracket r^{i+1}\rrbracket_{P,u}$, it follows that $r' \in \llbracket e(r^{i},a)\rrbracket_{P,u}$.
From this, it follows that there is a $r'' = b^{-1}(r')$ such that $r'' \in \llbracket r^{i}\rrbracket_{P,u}$ and $e(r'',a) = r'$.
From this and $[\phi]^{\mathit{last}(v).\mathit{db}} = [\phi]^{\mathit{last}(e(v,a)).\mathit{db}}$ for any $v \in \llbracket r^{i}\rrbracket_{P,u}$, it follows that $[\phi]^{\mathit{last}(r'').\mathit{db}} = [\phi]^{\mathit{last}(r').\mathit{db}}$.
From this and $[\phi]^{\mathit{last}(r^{i}).\mathit{db}} \neq [\phi]^{\mathit{last}(r').\mathit{db}}$, it follows that $[\phi]^{\mathit{last}(r^{i}).\mathit{db}} \neq [\phi]^{\mathit{last}(r'').\mathit{db}}$.
This contradicts the fact that for all $r' \in \llbracket r^{i}\rrbracket_{P,u}$, $[\phi]^{\mathit{last}(r').\mathit{db}} = [\phi]^{\mathit{last}(r^{i}).\mathit{db}}$.
Indeed, $r'' \in \llbracket r^{i}\rrbracket_{P,u}$ and $[\phi]^{\mathit{last}(r^{i}).\mathit{db}} \neq [\phi]^{\mathit{last}(r'').\mathit{db}}$.

We prove our claim that $[\phi]^{\mathit{last}(r').\mathit{db}} = [\phi]^{\mathit{last}(e(r',a)).\mathit{db}}$ holds for any $r' \in \llbracket r^{i}\rrbracket_{P,u}$.
Assume that this is not the case.
This implies that the content of one of the relations that determines $\phi$ is different in $\mathit{last}(r').\mathit{db}$ and $\mathit{last}(e(r',a)).\mathit{db}$.
This is impossible.
Indeed, if $a$'s execution has been successful, i.e., $\mathit{secEx}(\mathit{last}(e(r',a))) =\bot$ and $\mathit{Ex}(\mathit{last}(e(r',a))) = \emptyset$, then $a$'s execution does not change any table in  $\mathit{tables}(\phi)$, and the set of relations that determines $\phi$ is always a subset of $\mathit{tables}(\phi)$.
This leads to a contradiction, and, therefore,  $[\phi]^{\mathit{last}(r').\mathit{db}} \\ = [\phi]^{\mathit{last}(\mathit{e}(r',a)).\mathit{db}}$ holds. 
Similarly, if $a$'s execution has not been successful, i.e., $\mathit{secEx}(\mathit{last}(e(r',a))) = \top$ or $\mathit{Ex}(\mathit{last}(e(r',a))) \neq \emptyset$, then $\mathit{last}(r').\mathit{db}$ is the same as $\mathit{last}(\mathit{e}(r',a)).\mathit{db}$, and the claim holds trivially.

\item $\mathit{secure}_{P,u}(r,i+1 \attMod \phi)$ holds and  $\mathit{secure}_{P,u}(r,i \attMod \phi)$ does not hold.
We have already shown that $[\phi]^{\mathit{last}(r').\mathit{db}} \\ = [\phi]^{\mathit{last}(\mathit{e}(r',a)).\mathit{db}}$ holds for any $r' \in \llbracket r \rrbracket_{P,u}$.
From $\neg \mathit{secure}_{P,u}(r,i \attMod \phi)$, it follows that there is $r' \in \llbracket r^{i} \rrbracket_{P,u}$ such that $[\phi]^{\mathit{last}(r^{i}).\mathit{db}} \neq [\phi]^{\mathit{last}(r').\mathit{db}}$.
Let $b$ the bijection showing that $a$ preserves the equivalence class with respect to $r$, $P$, and $u$.
Since $r' \in \llbracket r^{i} \rrbracket_{P,u}$, then let $r'' = b(r') = e(r',a)$.
From $[\phi]^{\mathit{last}(r').\mathit{db}} = [\phi]^{\mathit{last}(\mathit{e}(r',a)).\mathit{db}}$ holds for any $r' \in \llbracket r^{i} \rrbracket_{P,u}$, it follows that  $[\phi]^{\mathit{last}(r^{i}).\mathit{db}} \neq [\phi]^{\mathit{last}(e(r',a)).\mathit{db}}$.
From this, $e(r^{i},a) = r^{i+1}$, and the fact that $[\phi]^{\mathit{last}(r').\mathit{db}} = [\phi]^{\mathit{last}(\mathit{e}(r',a)).\mathit{db}}$ holds for any $r' \in \llbracket r \rrbracket_{P,u}$, it follows that $[\phi]^{\mathit{last}(r^{i+1}).\mathit{db}} \neq [\phi]^{\mathit{last}(e(r',a)).\mathit{db}}$.
From this and $e(r',a) \in \llbracket r^{i+1} \rrbracket_{P,u}$, it follows $\neg \mathit{secure}_{P,u}\\(r,i+1 \attMod \phi)$.
This contradicts the fact that $\mathit{secure}_{P,u}(r, \\ i+1 \attMod \phi)$ holds.

\end{compactitem}
This completes the proof.
\end{proof}

\begin{lemma}\thlabel{theorem:secure:extend:on:runs:select:create}
Let $P = \langle M, f \rangle$ be an \accessControlConfiguration{}, where $M = \langle D,\Gamma\rangle$ is a system configuration and $f$ is an $M$-\acf{}, $L$ be the $P$-LTS, $u$ be a user in ${\cal U}$, $r$ be a run in $\mathit{traces}(L)$,  $a \in {\cal A}_{D,u}$ be a \texttt{SELECT} or \texttt{CREATE} action,  $\phi$ be a sentence, and $i$ be such that $1 \leq i \leq |r|$, $\mathit{triggers}(\mathit{last}(r^{i})) = \epsilon$, and $r^{i+1} = \mathit{extend}(r^{i},a)$.
If $a$ preserves the equivalence class for $r^{i}$, $P$, and $u$, then $\mathit{secure}_{P,u}(r,i \attMod \phi)$ holds iff $\mathit{secure}_{P,u}(r,i+1 \attMod \phi)$ holds.
\end{lemma}

\begin{proof}
Proof similar to that of \thref{theorem:secure:extend:on:runs:insert:delete}.
\end{proof}

\begin{lemma}\thlabel{theorem:secure:extend:on:runs:grant:revoke}
Let $P = \langle M, f \rangle$ be an \accessControlConfiguration{}, where $M = \langle D,\Gamma\rangle$ is a system configuration and $f$ is an $M$-\acf{}, $L$ be the $P$-LTS, $u$ be a user in ${\cal U}$, $r$ be a run in $\mathit{traces}(L)$,  $a \in {\cal A}_{D,u}$ be a \texttt{GRANT} or \texttt{REVOKE} action,  $\phi$ be a sentence, and $i$ be such that $1 \leq i \leq |r|$, $\mathit{triggers}(\mathit{last}(r^{i})) = \epsilon$, and $r^{i+1} = \mathit{extend}(r^{i},a)$.
If $a$ preserves the equivalence class for $r^{i}$, $P$, and $u$, then $\mathit{secure}_{P,u}(r,i \attMod \phi)$ holds iff $\mathit{secure}_{P,u}(r,i+1 \attMod \phi)$ holds.
\end{lemma}

\begin{proof}
Proof similar to that of \thref{theorem:secure:extend:on:runs:insert:delete}.
\end{proof}

\begin{lemma}\thlabel{theorem:secure:extend:on:runs:triggers}
Let $P = \langle M, f \rangle$ be an \accessControlConfiguration{}, where $M = \langle D,\Gamma\rangle$ is a system configuration and $f$ is an $M$-\acf{}, $L$ be the $P$-LTS,  $u$ be a user in ${\cal U}$, $r$ be a run in $\mathit{traces}(L)$,  $a$ be a trigger in ${\cal TRIGGER}_{D}$,  $\phi$ be a sentence, and $i$ be such that $1 \leq i \leq |r|$, $\mathit{invoker}(\mathit{last}(r^{i})) = u$, and $r^{i+1} = \mathit{extend}(r^{i},a)$.
If\begin{inparaenum}[(1)]
\item $a$ preserves the equivalence class for $r^{i}$, $P$, and $u$, 
\item if $a$'s action is either an \texttt{INSERT} or \texttt{DELETE}, then $t$'s execution does not change any table in $\mathit{tables}(\phi)$ for any run $v \in \llbracket r^{i}\rrbracket_{P,u}$, and
\item $\mathit{secEx}(\mathit{last}(\mathit{extend} \\ (r^{i},  a)) = \bot$ and  $\mathit{Ex}(\mathit{last}(\mathit{extend}(r^{i},a)) = \emptyset$,
\end{inparaenum}
then  $\mathit{secure}_{P,u}(r, \\ i \attMod \phi)$ holds iff $\mathit{secure}_{P,u}(r,i+1 \attMod \phi)$ holds.
\end{lemma}

\begin{proof}
Let $P = \langle M, f \rangle$ be an \accessControlConfiguration{}, where $M = \langle D,\Gamma\rangle$ is a system configuration and $f$ is an $M$-\acf{}, $L$ be the $P$-LTS,  $u$ be a user in ${\cal U}$, $r$ be a run in $\mathit{traces}(L)$,  $a$ be a trigger in ${\cal TRIGGER}_{D}$,  $\phi$ be a sentence, and $i$ be such that $1 \leq i \leq |r|$, $\mathit{invoker}(\mathit{last}(r^{i})) = u$, and $r^{i+1} = \mathit{extend}(r^{i},a)$.
Assume also\begin{inparaenum}[(1)]
\item that $a$ preserves the equivalence class for $r^{i}$, $P$, and $u$,  and
\item $\mathit{secEx}(\mathit{last}(\mathit{extend} \\ (r^{i},  a)) = \bot$ and  $\mathit{Ex}(\mathit{last}(\mathit{extend}(r^{i},a)) = \emptyset$.
\end{inparaenum}
In the following, we denote the $\mathit{extend}$ function by $e$.
Furthermore, we also denote the fact that $\mathit{secure}_{P,u}(r,i \attMod \phi)$ does not hold as $\neg \mathit{secure}_{P,u}(r,i \attMod \phi)$.
From Definition \ref{def:preserve:equiv:class} and the fact that $a$ preserves the equivalence class for $r^{i}$, $P$, and $u$, it follows that $e(r',a)$ is defined for any $r' \in \llbracket r^{i}\rrbracket_{P,u}$. 
Assume, for contradiction's sake, that our claim does not hold.
There are two cases:
\begin{compactitem}
\item $\mathit{secure}_{P,u}(r,i \attMod \phi)$ holds and $\mathit{secure}_{P,u}(r,i+1 \attMod \phi)$ does not hold.
From $\mathit{secure}_{P,u}(r,i \attMod \phi)$, it follows that $[\phi]^{\mathit{last}(r^{i}).\mathit{db}} = [\phi]^{\mathit{last}(r').\mathit{db}}$ for any $r' \in \llbracket r^{i} \rrbracket_{P,C}$.
We claim that $[\phi]^{\mathit{last}(r').\mathit{db}} = [\phi]^{\mathit{last}(e(r',a)).\mathit{db}}$ holds for any $r' \in \llbracket r^{i}\rrbracket_{P,u}$.
From $\neg \mathit{secure}_{P,u}(r,i+1 \attMod \phi)$, it follows that there is a $r'' \in \llbracket r^{i+1}\rrbracket_{P,u}$ such that $[\phi]^{\mathit{last}(r'').\mathit{db}} \neq [\phi]^{\mathit{last}(r^{i+1}).\mathit{db}}$.
Let $b$ the bijection showing that $a$ preserves the equivalence class with respect to $r^{i}$, $P$, and $u$.
Since $r^{i+1} = e(r^{i},a)$ and $r' \in \llbracket \mathit{e}(r,  a) \rrbracket_{P,u}$, then there is a run $v \in \llbracket r^{i}\rrbracket_{P,u}$ such that $v=b^{-1}(r'')$.
From this, $[\phi]^{\mathit{last}(r').\mathit{db}} = [\phi]^{\mathit{last}(e(r',a)).\mathit{db}}$ holds for any $r' \in \llbracket r^{i}\rrbracket_{P,u}$, and the fact that $[\phi]^{\mathit{last}(r'').\mathit{db}} \neq [\phi]^{\mathit{last}(r^{i+1}).\mathit{db}}$, it follows that $[\phi]^{\mathit{last}(v).\mathit{db}} \neq [\phi]^{\mathit{last}(r^{i+1}).\mathit{db}}$.
From this, $[\phi]^{\mathit{last}(r').\mathit{db}} = [\phi]^{\mathit{last}(e(r',a)).\mathit{db}}$ holds for any $r' \in \llbracket r^{i}\rrbracket_{P,u}$, and $r^{i+1} = e(r^{i},a)$, it follows $[\phi]^{\mathit{last}(v).\mathit{db}} \neq [\phi]^{\mathit{last}(r^{i}).\mathit{db}}$.
This contradicts the fact that $[\phi]^{\mathit{last}(r^{i}).\mathit{db}} = [\phi]^{\mathit{last}(r').\mathit{db}}$ for any $r' \in \llbracket r^{i} \rrbracket_{P,C}$.

We now prove that  $[\phi]^{\mathit{last}(r').\mathit{db}} = [\phi]^{\mathit{last}(e(r',a)).\mathit{db}}$ holds for any $r' \in \llbracket r^{i}\rrbracket_{P,u}$.
Assume, for contradiction's sake, that there is a run $r' \in \llbracket r^{i}\rrbracket_{P,u}$ such that $[\phi]^{\mathit{last}(r').\mathit{db}} \neq [\phi]^{\mathit{last}(e(r',a)).\mathit{db}}$.
There are three cases:
\begin{compactitem}
\item the trigger $a$ is not enabled in $e(r',a)$.
From this and the LTS semantics, it follows that $\mathit{last}(r').\mathit{db} = \mathit{last}(e(r',a)).\mathit{db}$.
From this, it therefore follows that $[\phi]^{\mathit{last}(r').\mathit{db}} = [\phi]^{\mathit{last}(e(r',a)).\mathit{db}}$.
This contradicts our assumption.

\item the trigger $a$ is enabled in $e(r',a)$ and its action is a \texttt{GRANT} or a \texttt{REVOKE}.
From this and the LTS semantics, it therefore follows that $\mathit{last}(r').\mathit{db} = \mathit{last}(e(r',a)).\mathit{db}$.
From this, it thus follows that $[\phi]^{\mathit{last}(r').\mathit{db}} = [\phi]^{\mathit{last}(e(r',a)).\mathit{db}}$.
This contradicts our assumption.

\item the trigger $a$ is enabled in $e(r',a)$ and its action is a \texttt{INSERT} or a \texttt{GRANT}.
Thus, from $[\phi]^{\mathit{last}(r').\mathit{db}} \neq [\phi]^{\mathit{last}(e(r',a)).\mathit{db}}$, it follows that the content of one of the relations that determines $\phi$ is different in $\mathit{last}(r').\mathit{db}$ and $\mathit{last}(e(r',a)).\mathit{db}$.
This contradicts the fact that the $a$'s execution does not change the tables in $\mathit{tables}(\phi)$ for any run  $r' \in \llbracket r^{i}\rrbracket_{P,u}$.

\end{compactitem}

\item $\mathit{secure}_{P,u}(r,i+1 \attMod \phi)$ holds and  $\mathit{secure}_{P,u}(r,  i \attMod \phi)$ does not hold.
We have already shown that $[\phi]^{\mathit{last}(r').\mathit{db}} \\ = [\phi]^{\mathit{last}(\mathit{e}(r',a)).\mathit{db}}$ holds for any $r' \in \llbracket r \rrbracket_{P,u}$.
From $\neg \mathit{secure}_{P,u}(r,i \attMod \phi)$, it follows that there is $r' \in \llbracket r^{i} \rrbracket_{P,u}$ such that $[\phi]^{\mathit{last}(r^{i}).\mathit{db}} \neq [\phi]^{\mathit{last}(r').\mathit{db}}$.
Let $b$ the bijection showing that $a$ preserves the equivalence class with respect to $r$, $P$, and $u$.
Since $r' \in \llbracket r^{i} \rrbracket_{P,u}$, then let $r'' = b(r') = e(r',a)$.
From $[\phi]^{\mathit{last}(r').\mathit{db}} = [\phi]^{\mathit{last}(\mathit{e}(r',a)).\mathit{db}}$ holds for any $r' \in \llbracket r^{i} \rrbracket_{P,u}$, it follows that  $[\phi]^{\mathit{last}(r^{i}).\mathit{db}} \neq [\phi]^{\mathit{last}(e(r',a)).\mathit{db}}$.
From this, $e(r^{i},a) = r^{i+1}$, and the fact that $[\phi]^{\mathit{last}(r').\mathit{db}} = [\phi]^{\mathit{last}(\mathit{e}(r',a)).\mathit{db}}$ holds for any $r' \in \llbracket r \rrbracket_{P,u}$, it follows that $[\phi]^{\mathit{last}(r^{i+1}).\mathit{db}} \neq [\phi]^{\mathit{last}(e(r',a)).\mathit{db}}$.
From this and $e(r',a) \in \llbracket r^{i+1} \rrbracket_{P,u}$, it follows $\neg \mathit{secure}_{P,u} \\ (r,i+1 \attMod \phi)$.
This contradicts $\mathit{secure}_{P,u}(r, i+1 \attMod \phi)$.

\end{compactitem}
This completes the proof.
\end{proof}

\begin{proposition}\thlabel{theorem:getInfo:sound:and:complete}
Let $P = \langle M, f \rangle$ be an \accessControlConfiguration{}, where $M = \langle D,\Gamma\rangle$ is a system configuration and $f$ is an $M$-\acf{}, $L$ be the $P$-LTS, $a \in {\cal A}_{D,u}$ be an \texttt{INSERT} or \texttt{DELETE} action,  and $r$ be a run such that $\mathit{tr}(\mathit{last}(r)) = \epsilon$.
For any constraint $\gamma$ in $\mathit{Dep}(\Gamma,a)$, the following statements hold:
\begin{compactitem}
\item $[\mathit{getInfoS}(\gamma,a)]^{\mathit{last}(r).\mathit{db}} = \top$ iff $\gamma \not\in \mathit{Ex}(\mathit{last}(\mathit{extend}(r,a)))$, and
\item $[\mathit{getInfoV}(\gamma,a)]^{\mathit{last}(r).\mathit{db}} = \top$ iff $\gamma \in \mathit{Ex}(\mathit{last}(\mathit{extend}(r,a)))$.
\end{compactitem}
\end{proposition}

\begin{proof}
Let $P = \langle M, f \rangle$ be an \accessControlConfiguration{}, where $M = \langle D,\Gamma\rangle$ is a system configuration and $f$ is an $M$-\acf{}, $L$ be the $P$-LTS, $a \in {\cal A}_{D,u}$ be an \texttt{INSERT} or \texttt{DELETE} action,  and $r$ be a run such that $\mathit{tr}(\mathit{last}(r)) = \epsilon$.
Furthermore, let $\gamma$ be a constraint in $\mathit{Dep}(\Gamma,a)$.
We first note that $\mathit{getInfoS}(\gamma,a) = \neg \mathit{getInfoV}(\gamma,a)$.
From this, it follows trivially that we can prove just one of the two claims.
We thus prove that  $[\mathit{getInfoS}(\gamma,a)]^{\mathit{last}(r).\mathit{db}} = \top$ iff $\gamma \not\in \mathit{Ex}(\mathit{last}(\mathit{extend}(r,a)))$.
There are two cases:
\begin{compactenum}
\item $a = \langle u,\mathtt{INSERT},R,\overline{t}\rangle$.
There are two cases depending on $\gamma$:
\begin{compactenum}
\item $\gamma$ is of the form $\forall \overline{x}, \overline{y}, \overline{y}', \overline{z}, \overline{z}'.\,  (R(\overline{x}, \overline{y},  \overline{z}) \wedge R(\overline{x}, \overline{y}', \\  \overline{z}') )\Rightarrow \overline{y} = \overline{y}'$.
Let $\overline{t}$ be $(\overline{v},\overline{w},\overline{q})$, $\mathit{db}$ be the state $\mathit{last}(r).\mathit{db}$, and $\mathit{db}'$ be the state $\mathit{db}[R\oplus\overline{t}]$.

\noindent
$(\Rightarrow)$ 
Assume that $[\mathit{getInfoS}(\gamma,a)]^{\mathit{last}(r).\mathit{db}} = \top$.
From this and $\mathit{getInfoS}(\gamma,a)$'s definition, it follows that for all tuples $(\overline{v},\overline{w}',\overline{q}') \in \mathit{db}(R)$, then $\overline{w}' = \overline{w}$.
From $a$'s definition and the LTS semantics, it follows that $\mathit{db}'(R) = \mathit{db}(R) \cup \{(\overline{v},\overline{w},\overline{q})\}$.
From this and the fact  that for all tuples $(\overline{v},\overline{w}',\overline{q}') \in \mathit{db}(R)$, then $\overline{w}' = \overline{w}$, it follows that for all tuples $(\overline{v},\overline{w}',\overline{q}') \\ \in \mathit{db}'(R)$, then $\overline{w}' = \overline{w}$.
Furthermore, since $\mathit{db} \in \Omega_{D}^{\Gamma}$, it follows that for all tuples $(\overline{v}',\overline{w}',\overline{q}'), (\overline{v}'',\overline{w}'', \\ \overline{q}'') \in \mathit{db}'(R)$, if $\overline{v}' = \overline{v}''$ and $\overline{v}' \neq \overline{v}$, then $\overline{w}' = \overline{w}$.
Therefore,  it follows that for all tuples $(\overline{v}',\overline{w}',\overline{q}')$, $(\overline{v}'',\overline{w}'',\overline{q}'') \in \mathit{db}'(R)$, if $\overline{v}' = \overline{v}''$, then $\overline{w}' = \overline{w}$.
Therefore, $[\gamma]^{\mathit{db}'} = \top$.
From this and the LTS semantics, it follows that $\gamma \not\in \mathit{Ex}(\mathit{last}(\mathit{extend}(r,a)))$.

\noindent
$(\Leftarrow)$
Assume that $\gamma \not\in \mathit{Ex}(\mathit{last}(\mathit{extend}(r,a))))$.
From this and the LTS semantics, it follows that  $[\gamma]^{\mathit{db}'} = \top$.
Therefore, for any two tuples $(\overline{v}',\overline{w}',\overline{q}')$ and $(\overline{v}'',\overline{w}'',\overline{q}'') \in \mathit{db}'(R)$, if $\overline{v}' = \overline{v}''$, then $\overline{w}' = \overline{w}$.
Assume, for contradiction's sake, that $[\mathit{getInfoS}(\gamma, a)]^{\mathit{db}} \\ = \bot$.
This means that there is a tuple $(\overline{v}, \overline{w}', \overline{q}')$ in $\mathit{db}(R)$ such that $\overline{w}' \neq \overline{w}$.
From $\mathit{db}'=\mathit{db}[R(\overline{v}, \overline{w}, \overline{q})]$ and the LTS semantics, it follows that both $(\overline{v}, \overline{w}', \overline{q}') $ and
$(\overline{v}, \overline{w}, \overline{q})$ are in $\mathit{db}'(R)$.
From this and $\overline{w}' \neq \overline{w}$, it follows that there are two tuples $(\overline{v}, \overline{w}, \overline{q})$ and $(\overline{v}, \overline{w}', \overline{q}')$ in $\mathit{db}(R)$ such that $\overline{w}' \neq \overline{w}$.
From this and the relational calculus semantics, it follows that $[\gamma]^{\mathit{db}} = \bot$.
This is in contradiction with $[\gamma]^{\mathit{db}'} = \top$.

\item $\gamma$ is of the form $\forall \overline{x}, \overline{z}.\, R(\overline{x}, \overline{z}) \Rightarrow \exists \overline{w}.\, S(\overline{x}, \overline{w})$.
Let $\overline{t}$ be $(\overline{v},\overline{w})$, $\mathit{db}$ be the state $\mathit{last}(r).\mathit{db}$, and $\mathit{db}'$ be the state $\mathit{db}[R\oplus\overline{t}]$.

\noindent
$(\Rightarrow)$ 
Assume that $[\mathit{getInfoS}(\gamma,a)]^{\mathit{db}} = \top$.
From this and $\mathit{getInfoS}(\gamma,a)$'s definition, it follows that there is a tuple $(\overline{v},\overline{y})$ in $\mathit{db}(S)$.
From $a$'s definition and the LTS semantics, it follows that $\mathit{db}'(S) = \mathit{db}(S)$.
From this, it follows that there is a tuple $(\overline{v},\overline{y})$ in $\mathit{db}'(S)$.
Furthermore, since $\mathit{db} \in \Omega_{D}^{\Gamma}$, it follows that for all tuples $(\overline{v}',\overline{w}') \in \mathit{db}(R)$, if $\overline{v}' \neq \overline{v}$, there is a tuple $(\overline{v}',\overline{y}') \in \mathit{db}(S)$.
From this and $\overline{db}' = \mathit{db} [R\oplus (\overline{v},\overline{w})]$, it follows that for all tuples $(\overline{v}',\overline{w}') \in \mathit{db}'(R)$,  there is a tuple $(\overline{v}',\overline{y}') \in \mathit{db}'(S)$.
Therefore, $[\gamma]^{\mathit{db}'} = \top$.
From this and the LTS semantics, it follows that $\gamma \not\in \mathit{Ex}(\mathit{last}(\mathit{extend}(r,a)))$.

\noindent
$(\Leftarrow)$
Assume that $\gamma \not\in \mathit{Ex}(\mathit{last}(\mathit{extend}(r,a)))$.
From this and the LTS semantics, it follows that  $[\gamma]^{\mathit{db}'} = \top$.
Therefore, for any tuple $(\overline{v}',\overline{w}') \in \mathit{db}'(R)$, there is a tuple $(\overline{v}',\overline{y}') \in \mathit{db}'(S)$.
Assume, for contradiction's sake, that $[\mathit{getInfoS}(\gamma,  a)]^{\mathit{db}} = \bot$.
This means that for any tuple $(\overline{v}',\overline{y}')$ in $\mathit{db}(S)$, $\overline{v}' \neq \overline{v}$.
From $\mathit{db}'(S)=\mathit{db}(S)$, it follows that for any tuple $(\overline{v}',\overline{y}')$ in $\mathit{db}'(S)$, $\overline{v}' \neq \overline{v}$.
From $\mathit{db}' = \mathit{db}[R\oplus (\overline{v},\overline{w}) ]$, it follows that there is a tuple $(\overline{v},\overline{w})$ in $\mathit{db}'(R)$ such that there is no tuple $(\overline{v},\overline{y}')$ in $\mathit{db}'(S)$.
From this and the relational calculus semantics, it follows that $[\gamma]^{\mathit{db}} = \bot$.
This is in contradiction with $[\gamma]^{\mathit{db}'} = \top$.

\end{compactenum}

\item $a = \langle u,\mathtt{DELETE},R,\overline{t}\rangle$.
In this case, $\gamma$ is of the form $\forall \overline{x}, \overline{z}.\, S(\overline{x},  \overline{z}) \Rightarrow \exists \overline{w}.\, R(\overline{x}, \overline{w})$.
Let $\overline{t}$ be $(\overline{v},\overline{w})$, $\mathit{db}$ be the state $\mathit{last}(r).\mathit{db}$, and $\mathit{db}'$ be the state $\mathit{db}[R \ominus \overline{t}]$.

\noindent
$(\Rightarrow)$ 
Assume that $[\mathit{getInfoS}(\gamma,a)]^{\mathit{db}} = \top$.
From this and $\mathit{getInfoS}(\gamma,a)$'s definition, it follows that either there is no tuple $(\overline{v}, \overline{y})$ in $\mathit{db}(S)$ or there is a tuple  $(\overline{v}, \overline{w}')$ in $\mathit{db}(R)$ such that $\overline{w}' \neq \overline{w}$.
There are two cases:
\begin{compactenum}
\item there is no tuple $(\overline{v}, \overline{y})$ in $\mathit{db}(S)$.
From this, $a$'s definition, and the LTS semantics, it follows that there is no tuple $(\overline{v}, \overline{y})$ in $\mathit{db}'(S)$.
From $\mathit{db} \in \Omega_{D}^{\Gamma}$, it follows that for all tuples $(\overline{v}', \overline{y}')$ in $\mathit{db}(S)$ such that $\overline{v}' \neq \overline{v}$, there is a tuple $(\overline{v}', \overline{w}')$ in $\mathit{db}(R)$.
From this, $\mathit{db}'(R) = \mathit{db}(R) \setminus \{(\overline{v},\overline{w})\}$, $\mathit{db}'(S) = \mathit{db}(S) $, and there is no tuple $(\overline{v}, \overline{y})$ in $\mathit{db}'(S)$, it follows that for all tuples $(\overline{v}', \overline{y}')$ in $\mathit{db}(S)$, there is a tuple $(\overline{v}', \overline{w}')$ in $\mathit{db}(R)$.
Therefore, $[\gamma]^{\mathit{db}'} = \top$.
From this and the LTS semantics, it follows that $\gamma \not\in \mathit{Ex}(\mathit{last}(\mathit{extend}(r,a)))$.

\item there is a tuple  $(\overline{v}, \overline{w}')$ in $\mathit{db}(R)$ such that $\overline{w}' \neq \overline{w}$.
From this, $a$'s definition, and the LTS semantics, it follows that there is  a tuple  $(\overline{v}, \overline{w}')$ in $\mathit{db}'(R)$ such that $\overline{w}' \neq \overline{w}$.
From $\mathit{db} \in \Omega_{D}^{\Gamma}$, it follows that for all tuples $(\overline{v}', \overline{y}')$ in $\mathit{db}(S)$ such that $\overline{v}' \neq \overline{v}$, there is a tuple $(\overline{v}', \overline{w}'')$ in $\mathit{db}(R)$.
From this, $\mathit{db}'(R) = \mathit{db}(R) \setminus \{(\overline{v},\overline{w})\}$, $\mathit{db}'(S) = \mathit{db}(S) $, and there is  a tuple  $(\overline{v}, \overline{w}')$ in $\mathit{db}'(R)$ such that $\overline{w}' \neq \overline{w}$, it follows that for all tuples $(\overline{v}', \overline{y}')$ in $\mathit{db}(S)$, there is a tuple $(\overline{v}', \overline{w}')$ in $\mathit{db}(R)$.
Therefore, $[\gamma]^{\mathit{db}'} = \top$.
From this and the LTS semantics, it follows that $\gamma \not\in \mathit{Ex}(\mathit{last}(\mathit{extend}(r,a)))$.

\end{compactenum}

\noindent
$(\Leftarrow)$
Assume that $\gamma \not\in \mathit{Ex}(\mathit{last}(\mathit{extend}(r,a)))$.
From this and the LTS semantics, it follows that  $[\gamma]^{\mathit{db}'} = \top$.
Therefore, for any tuple $(\overline{v}',\overline{y}') \in \mathit{db}'(S)$, there is a tuple $(\overline{v}',\overline{w}') \in \mathit{db}'(R)$.
Assume, for contradiction's sake, that $[\mathit{getInfoS}(\gamma,  a)]^{\mathit{db}} = \bot$.
Therefore, there is a tuple $(\overline{v},\overline{y})$ in $\mathit{db}(S)$ and for all tuples $(\overline{v},\overline{w}'')$ in $\mathit{db}(R)$, $\overline{w}'' = \overline{w}$.
From this,  $\mathit{db}'(S)=\mathit{db}(S)$, and $\mathit{db}' = \mathit{db}[R \ominus (\overline{v},\overline{w}) ]$, it follows that there is a tuple $(\overline{v},\overline{y})$ in $\mathit{db}'(S)$ and for all tuples $(\overline{v}'',\overline{w}'')$ in $\mathit{db}'(R)$, $\overline{v}'' \neq \overline{v}$.
From this and the relational calculus semantics, it follows that $[\gamma]^{\mathit{db}} = \bot$.
This is in contradiction with $[\gamma]^{\mathit{db}'} = \top$.

\end{compactenum}
This completes the proof.
\end{proof}

\begin{lemma}\thlabel{theorem:f:conf:pec:1}
Let $u$ be a user in ${\cal U}$, $P = \langle M, f_{\mathit{conf}}^{u} \rangle$ be an \accessControlConfiguration{}, where $M = \langle D,\Gamma\rangle$ is a system configuration and $f_{\mathit{conf}}^{u}$ is as above, and $L$ be the $P$-LTS.
For any run $r \in \mathit{traces}(L)$ and any action $a \in {\cal A}_{D,u}$, if $\mathit{extend}(r,a)$ is defined, then $a$ preserves the equivalence class for $r$, $P$, and $u$.
\end{lemma}

\begin{proof}
Let $u$ be a user in ${\cal U}$, $P = \langle M, f_{\mathit{conf}}^{u} \rangle$ be an \accessControlConfiguration{}, where $M = \langle D,\Gamma\rangle$ is a system configuration and $f_{\mathit{conf}}^{u}$ is as above, and $L$ be the $P$-LTS.
In the following, we use $e$ to refer to the $\mathit{extend}$ function and $f$ to refer to $f_{\mathit{conf}}^{u}$.
We prove our claim by contradiction.
Assume, for contradiction's sake, that there is a run $r  \in \mathit{traces}(L)$ and an action $a \in {\cal A}_{D,u}$ such that $\mathit{e}(r,a)$ is defined and $a$ does not preserve the equivalence class for $r$, $P$, and $u$.
According to the LTS semantics, the fact that  $\mathit{e}(r,a)$ is defined implies that $\mathit{triggers}(\mathit{last}(r)) = \epsilon$.
Therefore,  $\mathit{triggers}(\mathit{last}(r')) = \epsilon$ holds as well for any for any $r' \in \llbracket r \rrbracket_{P,u}$ (because $r$ and $r'$ are indistinguishable and, therefore, their projections are consistent), and, thus, $\mathit{e}(r',a)$ is defined as well for any $r' \in \llbracket r \rrbracket_{P,u}$.
There are a number of cases depending on $a$:
\begin{compactenum}
\item $a = \langle u, \mathtt{SELECT}, q \rangle$.
There are two cases:
\begin{compactenum}
\item $\mathit{secEx}(\mathit{last}(e(r,a))) = \bot$.
From the LTS rules and  $\mathit{secEx}(\mathit{last}(e(r,a))) = \bot$, it follows that $f(\mathit{last}(r),a) \\ = \top$.
From this and \thref{theorem:f:conf:soundness}, it follows that $f(\mathit{last}(r'), a) = \top$ for any $r' \in \llbracket r \rrbracket_{P,u}$.
From this and the LTS rules, it follows $\mathit{secEx}(\mathit{last}(e(r',a))) = \bot$ for any $r' \in \llbracket r \rrbracket_{P,u}$.
From $f(\mathit{last}(r'),a) = \top$ for any $r' \in \llbracket r \rrbracket_{P,u}$, it follows that $\mathit{secure}(u,q,\mathit{last}(r')) \\ = \top$ for any $r' \in \llbracket r \rrbracket_{P,u}$.
From this and \thref{theorem:secure:sound:under:approximation}, it follows that $[q]^{\mathit{last}(r').\mathit{db}} = [q]^{\mathit{last}(r).\mathit{db}}$ for all $r' \in \llbracket r \rrbracket_{P,u}$.
Furthermore, it follows trivially from the LTS rule \emph{\texttt{SELECT} Success}, that the state after $a$'s execution is data indistinguishable from $\mathit{last}(r)$.
It is also easy to see that $e(r',a)$ is well-defined for any $r' \in \llbracket r \rrbracket_{P,u}$.
From the considerations above and $r' \in \llbracket r \rrbracket_{P,u}$, it follows trivially that $e(r',a) \in  \llbracket e(r,a) \rrbracket_{P,u}$. The bijection $b$ is trivially $b(r') = e(r',a)$. 
This leads to a contradiction.

\item $\mathit{secEx}(\mathit{last}(e(r,a))) = \top$.
From the LTS rules and  $\mathit{secEx}(\mathit{last}(e(r,a))) = \top$, it follows that $f(\mathit{last}(r),a) \\ = \bot$.
From this and \thref{theorem:f:conf:soundness}, it follows that $f(\mathit{last}(r'), a) = \bot$ for any $r' \in \llbracket r \rrbracket_{P,u}$.
From this and the LTS rules, it follows $\mathit{secEx}(\mathit{last}(e(r',a))) = \top$ for any $r' \in \llbracket r \rrbracket_{P,u}$.
The data indistinguishability between $\mathit{last}(e(r',a))$ and $\mathit{last}(e(r,a))$ follows trivially from the data indistinguishability between $\mathit{last}(r')$ and $\mathit{last}(r)$.
Therefore, for any run $r' \in \llbracket r \rrbracket_{P,C}$, there is exactly one run $e(r',a)$. 
From the considerations above, it follows trivially that $e(r',a) \\ \in  \llbracket e(r,a) \rrbracket_{P,u}$.
The bijection $b$ is trivially $b(r') = e(r',a)$. This leads to a contradiction.

\end{compactenum}
Both cases leads to a contradiction. 
This completes the proof for $a = \langle u, \mathtt{SELECT}, q \rangle$.

\item $a = \langle u, \mathtt{INSERT}, R, \overline{t} \rangle$. 
In the following, we denote by $\mathit{gI}$ the function $\mathit{getInfo}$, by $\mathit{gS}$ the function $\mathit{getInfoS}$, and by $\mathit{gV}$ the function $\mathit{getInfoV}$.
There are three cases:

\begin{compactenum}
\item $\mathit{secEx}(\mathit{last}(e(r,a))) = \bot$ and $\mathit{Ex}(\mathit{last}(e(r,a))) = \emptyset$. 
From the LTS rules and  $\mathit{secEx}(\mathit{last}(e(r,a))) = \bot$, it follows that $f(\mathit{last}(r),a) = \top$.
From this and \thref{theorem:f:conf:soundness}, it follows that $f(\mathit{last}(r'),a) = \top$ for any $r' \in \llbracket r \rrbracket_{P,u}$.
From this and the LTS rules, it follows that $\mathit{secEx}(\mathit{last}(e(r',a))) = \bot$ for any $r' \in \llbracket r \rrbracket_{P,u}$.
From $f_{\mathit{conf}}^{u}$'s definition and  $f(\mathit{last}(r),a) = \top$, it follows that $\mathit{secure}(u,\mathit{gS}(\gamma, \mathit{act}), \mathit{last}(r))$ holds for any integrity constraint $\gamma$ in $\mathit{Dep}(\Gamma,a)$.
From $\mathit{Ex}(\mathit{last}(e(r,a))) = \emptyset$ and \thref{theorem:getInfo:sound:and:complete}, it follows $[\mathit{gS}(\gamma, \mathit{act})]^{\mathit{last}(r).\mathit{db}} = \top$.
From this,  $\mathit{secure}(u, \\ \mathit{gS}(\gamma, \mathit{act}), \mathit{last}(r))$, and \thref{theorem:secure:sound:under:approximation}, it follows that $[\mathit{gS}(\gamma, \mathit{act})]^{\mathit{last}(r').\mathit{db}} = \top$ for any $r'  \in \llbracket r\rrbracket_{P,u}$. 
From this and \thref{theorem:getInfo:sound:and:complete}, it follows that $\mathit{Ex}(\mathit{last}(e(r', \\ a))) = \emptyset$ for any $r'  \in \llbracket r\rrbracket_{P,u}$. 
We  claim that, for any $r'  \in \llbracket r\rrbracket_{P,u}$, $\mathit{last}(e(r,a))$ and $\mathit{last}(e(r', a))$ are data indistinguishable.
From this and the above considerations, it follows trivially that $e(r',a) \in  \llbracket e(r,a) \rrbracket_{P,u}$.
The bijection $b$ is trivially $b(r') = e(r',a)$.
This leads to a contradiction.

We now prove our claim that for any $r'  \in \llbracket r\rrbracket_{P,u}$, $\mathit{last}(e(r,a))$ and $\mathit{last}(e(r',a))$ are data indistinguishable.
We prove the claim by contradiction.
Let $s_{2} = \langle \mathit{db}_{2}, U_{2},\mathit{sec}_{2}, T_{2},V_{2} \rangle$ be $\mathit{pState}(\mathit{last}(e(r,  a)))$, $s_{2}' = \langle \mathit{db}_{2}', U_{2}',\mathit{sec}_{2}', T_{2}',V_{2}' \rangle$ be $\mathit{pState}(\mathit{last}(e (r',a)))$, $s_{1} = \langle \mathit{db}_{1}, U_{1},\mathit{sec}_{1}, T_{1},V_{1} \rangle$ be $\mathit{pState}(\mathit{last}  (r))$, and $s_{1}' = \langle \mathit{db}_{1}', U_{1}',\mathit{sec}_{1}', T_{1}',V_{1}'\rangle$ be $\mathit{pState}(\mathit{last}  (r'))$.
In the following, we denote  the $\mathit{permissions}$ function by $p$.
Furthermore, note that $s_{1}$ and $s_{1}'$ are data-indistinguishable because $r'  \in \llbracket r\rrbracket_{P,u}$.
There are a number of cases:
\begin{compactenum}

\item $U_{2} \neq U_{2}'$. 
Since $a$ is an \texttt{INSERT} operation, it follows that $U_{1} = U_{2}$ and $U_{1}' = U_{2}'$.
Furthermore, from $s_{1} \cong_{u,M}^{\mathit{data}} s_{1}'$, it follows that $U_{1} = U_{1}'$.
Therefore, $U_{2} = U_{2}'$ leading to a contradiction.

\item $\mathit{sec}_{2} \neq \mathit{sec}_{2}'$. 
The proof is similar to the case $U_{2} \neq U_{2}'$. 

\item $\mathit{T}_{2} \neq \mathit{T}_{2}'$. 
The proof is similar to the case $U_{2} \neq U_{2}'$. 

\item $\mathit{V}_{2} \neq \mathit{V}_{2}'$. 
The proof is similar to the case $U_{2} \neq U_{2}'$.

\item there is a table $R'$ for which $\langle \oplus, \mathtt{SELECT},  R\rangle \in p(s_{2},u)$ and $\mathit{db}_{2}(R') \neq \mathit{db}_{2}'(R')$.
Note that $p(s_{2},u) = p(s_{1},u)$.
There are two cases:
\begin{compactitem}
\item $R = R'$.
From $s_{1} \cong_{u,M}^{\mathit{data}} s_{1}'$ and $\langle \oplus, \mathtt{SELECT}, R\rangle \\ \in p(s_{2},u) $, it follows that $\mathit{db}_{1}(R') = \mathit{db}_{1}'(R')$.
From this and the fact that $a$ has been executed successfully both in $e(r,a)$ and $e(r',a)$, it follows that $\mathit{db}_{2}(R') = \mathit{db}_{1} (R') \cup \{\overline{t}\}$ and  $\mathit{db}_{2}'(R') = \mathit{db}_{1}'(R') \cup \{\overline{t}\}$.
From this and $\mathit{db}_{1}(R') = \mathit{db}_{1}'(R')$, it follows that $\mathit{db}_{2}(R') = \mathit{db}_{2}'(R')$ leading to a contradiction.

\item $R \neq R'$.
From $s_{1} \cong_{P,u}^{\mathit{data}} s_{1}'$ and $\langle \oplus, \mathtt{SELECT},  R\rangle \\ \in p(s_{2},u)$, it follows that $\mathit{db}_{1}(R') = \mathit{db}_{1}'(R')$.
From this and the fact that $a$ does not modify $R'$, it follows that $\mathit{db}_{1}(R')  = \mathit{db}_{2}(R')$ and  $\mathit{db}_{1}'(R') = \mathit{db}_{2}'(R')$.
From this and $\mathit{db}_{1}(R') = \mathit{db}_{1}'(R')$, it follows that $\mathit{db}_{2}(R') = \mathit{db}_{2}'(R')$ leading to a contradiction.
\end{compactitem}

\item there is a view $v$ for which $\langle \oplus, \mathtt{SELECT},  v\rangle \in p(s_{2}, u)$ and $\mathit{db}_{2}(v) \neq \mathit{db}_{2}'(v)$.
Note that $p(s_{2},\\ u) = p(s_{1},u)$.
Since $a$ has been successfully executed in both states, we know that $\mathit{noLeak}(s_{1},\\ a, u)$ hold.
There are two cases:
\begin{compactitem}
\item $R \not\in \mathit{tDet}(v,s,M)$.
Then, $v(s_{1}) = v(s_{2})$ and  $v(s_{1}') = v(s_{2}')$ (because $R$'s content does not determine $v$'s materialization).
From $s_{1} \cong_{u,M}^{\mathit{data}} s_{1}'$ and the fact that $a$ modifies only $R$, it follows that  $v(\mathit{db}_{2}) = v(\mathit{db}_{2}')$ leading to  a contradiction.

\item $R \in \mathit{tDet}(v,s,M)$ and for all $o \in \mathit{tDet}(v,s, \\ M)$, $\langle \oplus, \mathtt{SELECT}, o \rangle \in  p(s_{1},u)$. 
From this and $s_{1} \cong_{u,M}^{\mathit{data}} s_{1}'$, it follows that, for all $o \in \mathit{tDet}(v,\\ s,M)$, $o(s_{1}) = o(s_{1}')$.
If $o \neq R$, $o(s_{1}) = o(s_{1}') = o(s_{2}) = o(s_{2}')$.
From $s_{1} \cong_{u,M}^{\mathit{data}} s_{1}'$ and $\langle \oplus, \mathtt{SELECT},  R\rangle \in p(s_{1},u) $, it follows that $\mathit{db}_{1}(R) = \mathit{db}_{1}'(R)$.
From this and the fact that $a$ has been executed successfully both in $e(r,a)$ and $e(r',a)$, it follows that $\mathit{db}_{2}(R) = \mathit{db}_{1}  (R) \cup \{\overline{t}\}$ and  $\mathit{db}_{2}'(R) = \mathit{db}_{1}'(R) \cup \{\overline{t}\}$.
From this and $\mathit{db}_{1}(R) = \mathit{db}_{1}'(R)$, it follows that $\mathit{db}_{2}(R) = \mathit{db}_{2}'(R)$.
From this and for all $o \in \mathit{tDet}(v,s, M)$ such that $o \neq R$, $o(s_{2}) = o(s_{2}')$, it follows that for all $o \in \mathit{tDet}(v,s,M)$, $o(s_{2}) = o(s_{2}')$.
Since the content of all tables determining $v$ is the same in $s_{2}$ and $s_{2}'$, it follows that  $\mathit{db}_{2}(v) = \mathit{db}_{2}'(v)$ leading to  a contradiction.
\end{compactitem}

\end{compactenum}
All the cases lead to a contradiction.

\item $\mathit{secEx}(\mathit{last}(e(r,a))) = \bot$ and $\mathit{Ex}(\mathit{last}(e(r,a))) \neq \emptyset$.
From the LTS rules and  $\mathit{secEx}(e(r,a)) = \bot$, it follows that $f(\mathit{last}(r),a) = \top$.
From this and \thref{theorem:f:conf:soundness}, it follows that $f(\mathit{last}(r'),a) = \top$ for any $r' \in \llbracket r \rrbracket_{P,u}$.
From this and the LTS rules, it follows that $\mathit{secEx}(\mathit{last}(e(r',a))) = \bot$ for any $r' \in \llbracket r \rrbracket_{P,u}$.
Assume that the exception has been caused by the constraint $\gamma$, i.e., $\gamma \in \mathit{Ex}(\mathit{last}(e(r,a)))$.
From this and \thref{theorem:getInfo:sound:and:complete}, it follows that $\mathit{gV}(\gamma, \mathit{a})$ holds in $\mathit{last}(r).\mathit{db}$.
From $f_{\mathit{conf}}^{u}$'s definition, it thus follows that $\mathit{secure}(u,  \mathit{gV}(\gamma, \mathit{a}), \mathit{last}(r))$ holds.
From this,   $[\mathit{gV}(\gamma, \mathit{a})]^{\mathit{last}(r).\mathit{db}} = \top$, and \thref{theorem:secure:sound:under:approximation}, it  follows that $[\mathit{gV}(\gamma, \mathit{act})]^{\mathit{last}(r').\mathit{db}} = \top$ for any $r'  \in \llbracket r\rrbracket_{P,u}$. 
From this and \thref{theorem:getInfo:sound:and:complete}, it follows that $\gamma \in \mathit{Ex}(\mathit{last}(e(r',a)))$ for any $r'  \in \llbracket r\rrbracket_{P,u}$.
The data indistinguishability between $\mathit{last}(e(r,a))$ and $\mathit{last}(e(r',a))$ follows trivially from the data indistinguishability between $\mathit{last}(r)$ and $\mathit{last}(r')$ for any $r' \in \llbracket r \rrbracket_{P,u}$. 
Therefore, for any run $r' \in \llbracket r \rrbracket_{P,u}$, there is exactly one run $e(r',a)$.
From the considerations above, it follows trivially that $e(r',a) \in  \llbracket e(r,a) \rrbracket_{P,u}$.
The bijection $b$ is trivially $b(r') = e(r',a)$.
This leads to a contradiction.

\item $\mathit{secEx}(\mathit{last}(e(r,a))) = \top$.
From the LTS rules and  $\mathit{secEx}(\mathit{last}(e(r,a))) = \top$, it follows that $f(\mathit{last}(r),a) \\ = \bot$.
From this and \thref{theorem:f:conf:soundness}, it follows that $f(\mathit{last}(r'), a) = \bot$ for any $r' \in \llbracket r \rrbracket_{P,u}$.
From this and the LTS rules, it follows $\mathit{secEx}(\mathit{last}(e(r',a))) = \top$ for any $r' \in \llbracket r \rrbracket_{P,u}$.
The data indistinguishability between $\mathit{last}(e(r,a))$ and $\mathit{last}(e(r',a))$ follows trivially from the data indistinguishability between $\mathit{last}(r)$ and $\mathit{last}(r')$ for any $r' \in \llbracket r \rrbracket_{P,u}$.
Therefore, for any run $r' \in \llbracket r \rrbracket_{P,u}$, there is exactly one run $e(r',a)$.
From the considerations above, it follows trivially that $e(r',a) \in  \llbracket e(r,a) \rrbracket_{P,u}$.
The bijection $b$ is trivially $b(r') = e(r',a)$. This leads to a contradiction.
\end{compactenum}
All cases lead to a contradiction.
This completes the proof for $a = \langle u, \mathtt{INSERT}, R, \overline{t} \rangle$.

\item $a = \langle u, \mathtt{DELETE}, R, \overline{t} \rangle$. 
The proof is similar to that for $a = \langle u, \mathtt{INSERT}, R, \overline{t} \rangle$.

\item $a = \langle \oplus, u', p, u\rangle$. 
There are two cases:
\begin{compactenum}
\item $\mathit{secEx}(\mathit{last}(e(r,a))) = \bot$.
We assume that $p = \langle \mathtt{SELECT}, O\rangle$ for some $O \in D \cup V$.
If this is not the case, the proof is trivial.
Furthermore,  we also assume that $u' = u$, otherwise the proof is, again, trivial since the new permission does not influence $u$'s permissions.
From the LTS rules and  $\mathit{secEx}(\mathit{last}(e(r,a))) = \bot$, it follows that $f(\mathit{last}(r),a) \\ = \top$.
From this and \thref{theorem:f:conf:soundness}, it follows that $f(\mathit{last}(r'),a) = \top$ for any $r' \in \llbracket r \rrbracket_{P,u}$.
From this and the LTS rules, it follows $\mathit{secEx}(\mathit{last}(e(r',a))) = \bot$ for any $r' \in \llbracket r \rrbracket_{P,u}$. 
From $\mathit{secEx}(\mathit{last}(e(r,a))) = \bot$ and $f_{\mathit{conf}}^{u}$'s definition, it follows that $\mathit{last}(r').\mathit{sec} = \mathit{last}(e(r',  a)).\mathit{sec}$.
Therefore, since $\mathit{last}(r)$ and $\mathit{last}(r')$ are data indistinguishable, for any $r' \in \llbracket r \rrbracket_{P,u}$, then also   $\mathit{last}(e(r,a))$ and $\mathit{last}(e(r',a))$ are data indistinguishable.
Therefore, for any run $r' \in \llbracket r \rrbracket_{P,u}$, there is exactly one run $e(r',a)$.
From the considerations above, it follows trivially that $e(r',a) \in  \llbracket e(r,a) \rrbracket_{P,u}$.
The bijection $b$ is trivially $b(r') = e(r',a)$.
This leads to a contradiction.

\item $\mathit{secEx}(\mathit{last}(e(r,a))) = \top$.  
From the LTS rules and  $\mathit{secEx}(\mathit{last}(e(r,a))) = \top$, it follows that $f(\mathit{last}(r),a) \\ = \bot$.
From this and \thref{theorem:f:conf:soundness}, it follows that $f(\mathit{last}(r'), a) = \bot$ for any $r' \in \llbracket r \rrbracket_{P,u}$.
From this and the LTS rules, it follows $\mathit{secEx}(\mathit{last}(e(r',a))) = \top$ for any $r' \in \llbracket r \rrbracket_{P,u}$.
The data indistinguishability between $\mathit{last}(e(r',a))$ and $\mathit{last}(e(r,a))$ follows trivially from the data indistinguishability between $\mathit{last}(r')$ and $\mathit{last}(r)$.
Therefore, for any run $r' \in \llbracket r \rrbracket_{P,u}$, there is exactly one run $e(r',a)$.
From the considerations above, it follows trivially that $e(r',a) \\ \in  \llbracket e(r,a) \rrbracket_{P,u}$.
The bijection $b$ is trivially $b(r') = e(r',a)$.
This leads to a contradiction.

\end{compactenum}
Both cases lead to a contradiction.
This completes the proof for $a = \langle \oplus, u', p, u\rangle$.

\item $a = \langle \oplus^{*}, u', p, u\rangle$. 
The proof is similar to that for $a = \langle \oplus, u', p, u\rangle$.

\item $a = \langle \ominus, u', p, u\rangle$.
The proof is similar to that for $a = \langle u, \mathtt{SELECT}, q \rangle$.
The only difference is in proving that for any $r'  \in \llbracket r\rrbracket_{P,u}$, $\mathit{last}(e(r,a))$ and $\mathit{last}(e(r',a))$ are data indistinguishable.
Assume, for contradiction's sake, that this is not the case.
Let $s_{2} = \langle \mathit{db}_{2}, U_{2},\mathit{sec}_{2}, T_{2}, V_{2} \rangle$ be $\mathit{pState}(\mathit{last}(e(r, a)))$, $s_{2}' = \langle \mathit{db}_{2}', U_{2}',\mathit{sec}_{2}', T_{2}',V_{2}' \rangle$ be $\mathit{pState}(\mathit{last}(e (r',a)))$, $s_{1} = \langle \mathit{db}_{1}, U_{1},\mathit{sec}_{1}, T_{1},V_{1} \rangle$ be \\$\mathit{pState}(\mathit{last}  (r))$, and, finally, $s_{1}' = \langle \mathit{db}_{1}', U_{1}',\mathit{sec}_{1}', T_{1}',V_{1}'\rangle$ be $\mathit{pState}(\mathit{last}  (r'))$.
In the following, we denote  the $\mathit{permissions}$ function by $p$.
Furthermore, note that $s_{1}$ and $s_{1}'$ are data-indistinguishable because $r'  \in \llbracket r\rrbracket_{P,u}$.
There are a number of cases:
\begin{compactenum}

\item $U_{2} \neq U_{2}'$. 
Since $a$ is an \texttt{REVOKE} operation, it follows that $U_{1} = U_{2}$ and $U_{1}' = U_{2}'$.
Furthermore, from $s_{1} \cong_{u,M}^{\mathit{data}} s_{1}'$, it follows that $U_{1} = U_{1}'$.
Therefore, $U_{2} = U_{2}'$ leading to a contradiction.

\item $\mathit{sec}_{2} \neq \mathit{sec}_{2}'$. 
From $s_{1} \cong_{u,M}^{\mathit{data}} s_{1}'$, it follows that $\mathit{sec}_{1} = \mathit{sec}_{1}'$.
From $a$'s definition and the LTS rules, it follows that $\mathit{sec}_{2} = 
\mathit{revoke}(\mathit{sec}_{1}, u',  p,u)$ and $\mathit{sec}_{2}' = \mathit{revoke}(\mathit{sec}_{1}', u',p,u)$.
From this and $\mathit{sec}_{1} = \mathit{sec}_{1}'$, it follows that $\mathit{sec}_{2} = \mathit{sec}_{2}'$ leading to a contradiction.

\item $\mathit{T}_{2} \neq \mathit{T}_{2}'$. 
The proof is similar to the case $U_{2} \neq U_{2}'$. 

\item $\mathit{V}_{2} \neq \mathit{V}_{2}'$. 
The proof is similar to the case $U_{2} \neq U_{2}'$.

\item there is a table $R$ for which $\langle \oplus, \mathtt{SELECT},  R\rangle \in p(s_{2},u)$ and $\mathit{db}_{2}(R) \neq \mathit{db}_{2}'(R)$.
Since $a$ is an \texttt{REVOKE} operation, it follows that $\mathit{db}_{1} = \mathit{db}_{2}$ and $\mathit{db}_{1}' = \mathit{db}_{2}'$.
Furthermore, from $s_{1} \cong_{u,M}^{\mathit{data}} s_{1}'$, it follows that $\mathit{db}_{1}(R) = \mathit{db}_{1}'(R)$.
From this, $\mathit{db}_{1} = \mathit{db}_{2}$, and $\mathit{db}_{1}' = \mathit{db}_{2}'$, it follows that $\mathit{db}_{2}(R) = \mathit{db}_{2}'(R)$ leading to a contradiction.

\item there a view $v$ for which $\langle \oplus, \mathtt{SELECT},  v\rangle \in p(s_{2}, u)$ and $\mathit{db}_{2}(v) \neq \mathit{db}_{2}'(v)$.
Since $a$ is an \texttt{REVOKE} operation, it follows that $\mathit{db}_{1} = \mathit{db}_{2}$ and $\mathit{db}_{1}' = \mathit{db}_{2}'$.
Furthermore, from $s_{1} \cong_{u,M}^{\mathit{data}} s_{1}'$, it follows that $\mathit{db}_{1}(v) = \mathit{db}_{1}'(v)$.
From this, $\mathit{db}_{1} = \mathit{db}_{2}$, and $\mathit{db}_{1}' = \mathit{db}_{2}'$, it follows that $\mathit{db}_{2}(v) = \mathit{db}_{2}'(v)$ leading to a contradiction.

\end{compactenum}
All the cases lead to a contradiction.

\item $a = \langle u, \mathtt{CREATE}, o \rangle$. 
The proof is similar to that for $a = \langle \ominus, u', p, u\rangle$.

\item $a = \langle u, \mathtt{ADD\_USER}, u' \rangle$. 
The proof is similar to that for $a = \langle \ominus, u', p, u\rangle$.
\end{compactenum}

This completes the proof.
\end{proof}

\begin{lemma}\thlabel{theorem:f:conf:pec:2}
Let $u$ be a user in ${\cal U}$, $P = \langle M, f_{\mathit{conf}}^{u} \rangle$ be an \accessControlConfiguration{}, where $M = \langle D,\Gamma\rangle$ is a system configuration and $f_{\mathit{conf}}^{u}$ is as above, and $L$ be the $P$-LTS.
For any run $r \in \mathit{traces}(L)$ such that $\mathit{invoker}(\mathit{last}(r)) = u$ and any trigger $t \in {\cal TRIGGER}_{D}$,
if $\mathit{extend}(r,t)$ is defined, then $t$ preserves the equivalence class for $r$, $M$, and $u$.
\end{lemma}

\begin{proof}
Let $u$ be a user in ${\cal U}$, $P = \langle M, f_{\mathit{conf}}^{u} \rangle$ be an \accessControlConfiguration{}, where $M = \langle D,\Gamma\rangle$ is a system configuration and $f_{\mathit{conf}}^{u}$ is as above, and $L$ be the $P$-LTS.
In the following, we use $e$ to refer to the $\mathit{extend}$ function.
The proof in the cases where the trigger $t$ is not enabled, i.e., its \texttt{WHEN} condition is not satisfied, or $t$'s \texttt{WHEN} condition is not secure are similar to the proof of the \texttt{SELECT} case of \thref{theorem:f:conf:pec:1}.
In the following, we therefore assume that the trigger $t$ is enabled and that its \texttt{WHEN} condition is secure.
We prove our claim by contradiction.
Assume, for contradiction's sake, that there is a run $r  \in \mathit{traces}(L)$ such that $\mathit{invoker}(\mathit{last}(r)) = u$ and a trigger $t$ such that $\mathit{e}(r,t)$ is defined and $t$ does not preserve the equivalence class for $r$, $P$, and $u$.
Since $\mathit{invoker}(\mathit{last}(r))=u$ and $\mathit{e}(r,t)$ is defined, then $\mathit{e}(r',t)$ is defined as well for any $r' \in \llbracket r \rrbracket_{P,u}$  (indeed, from $\mathit{invoker}(\mathit{last}(r))=u$, it follows that the last action in $r$ is either an action issued by $u$ or a trigger invoker by $u$. 
From this, the fact that $e(r,t)$ is defined, and the fact that $r$ and $r'$ are indistinguishable, it follows that $\mathit{tr}(\mathit{last}(r)) = \mathit{tr}(\mathit{last}(r')) = t$).
Let $a$ be $t$'s action and $w = \langle u', \texttt{SELECT},q\rangle$ be the \texttt{SELECT} command associated with $t$'s \texttt{WHEN} condition.
Let $s$ be the state $\mathit{last}(r)$, $s'$ be the state obtained just after the execution of the \texttt{WHEN} condition, and $s''$ be the state $\mathit{last}(e(r,t))$.
There are a number of cases depending on $t$'s action $a$:
\begin{compactenum}

\item $a = \langle u', \mathtt{INSERT}, R, \overline{t} \rangle$. 
There are three cases:

\begin{compactenum}
\item $\mathit{secEx}(s'') = \bot$ and $\mathit{Ex}(s'') = \emptyset$. 
The proof of this case is similar to that of the corresponding case in \thref{theorem:f:conf:pec:1}.

\item $\mathit{secEx}(s'') = \bot$ and $\mathit{Ex}(s'') \neq \emptyset$.
The only difference between the proof of this case in this Lemma and in that of \thref{theorem:f:conf:pec:1} is that we have to establish again the data indistinguishability between $\mathit{last}(e(r,t))$ and $\mathit{last}(e(r',t))$.
Indeed, for triggers the roll-back state is, in general, different from the one immediately before the trigger's execution, i.e., it may be that $\mathit{pState}(\mathit{last}(e(r,t)))  \neq \mathit{pState}(\mathit{last}(r))$. 
We now prove that $\mathit{last}(e(r,t))$ and $\mathit{last}(e(r',t))$ are data indistinguishable. 
From the LTS semantics, it follows that $r = p \concat s_{0} \concat \langle \mathit{invoker}(\mathit{last}(r)), \mathit{op}, R',\overline{v} \rangle \concat  s_{1} \concat t_{1} \concat \ldots \concat s_{n-1} \concat t_{n} \concat s_{n}$, where $p \in \mathit{traces}(L)$ and $t_{1}, \ldots, t_{n} \in {\cal TRIGGER}_{D}$.
Similarly, $r' = p' \concat s_{0}' \concat \langle \mathit{invoker}(\mathit{last}(r)), \mathit{op}, R',\overline{v} \rangle \concat  s_{1}' \concat t_{1} \concat \ldots \concat s_{n-1}' \concat t_{n} \concat s_{n}'$, where $p' \in \mathit{traces}(L)$, $p \cong_{u,M} p'$, and all states $s_{i}$ and $s_{i}'$ are data indistinguishable.
Then, the roll-back states are, respectively, $s_{0}$ and $s_{0}'$, which are data indistinguishable.
From the LTS rules, $\mathit{last}(e(r,a)) = s_{0}$ and $\mathit{last}(e(r',a)) = s_{0}'$.
Therefore, the data indistinguishability between $\mathit{last}(e(r,\\ a))$ and $\mathit{last}(e(r', a))$ follows trivially for any $r' \in \llbracket r \rrbracket_{P,u}$.

\item $\mathit{secEx}(s'') = \top$. 
The proof is similar to the previous case.
\end{compactenum}
All cases lead to a contradiction.
This completes the proof for $a = \langle u', \mathtt{INSERT}, R, \overline{t} \rangle$.

\item $a = \langle u', \mathtt{DELETE}, R, \overline{t} \rangle$.
The proof is similar to that for $a = \langle u', \mathtt{INSERT}, R, \overline{t} \rangle$.

\item $a = \langle \oplus, u'', p, u'\rangle$.
There are two cases:
\begin{compactenum}
\item $\mathit{secEx}(s'') = \bot$. 
In this case, the proof is similar to the  corresponding case in \thref{theorem:f:conf:pec:1}.

\item $\mathit{secEx}(s'') = \top$. 
The proof is similar to the $\mathit{secEx}(s'') \\ = \top$ case of the $a = \langle u', \mathtt{INSERT}, R, \overline{t} \rangle$ case.
\end{compactenum}
Both cases lead to a contradiction.
This completes the proof for $a = \langle \oplus, u'', p, u'\rangle$.

\item $a = \langle \oplus^{*}, u'', p, u'\rangle$. 
The proof is similar to that for $a = \langle \oplus, u'', p, u'\rangle$.

\item $a = \langle \ominus, u'', p, u'\rangle$. 
The proof is similar to that for $a = \langle u', \mathtt{INSERT}, R, \overline{t} \rangle$.

\end{compactenum}
This completes the proof.
\end{proof}

We now prove our main result, namely that $f_{\mathit{conf}}^{u}$ provides \confidentiality{} with respect to the user $u$.
We first recall the concept of \emph{derivation}.
Given a judgment $r, i \attMod \phi$, a \emph{derivation of  $r, i \attMod \phi$ with respect to $\attackerModel$}, or \emph{a derivation of $r,i \attMod \phi$} for short, is a proof tree, obtained by applying the rules defining $\attackerModel$, that ends in $r, i \attMod \phi$.
With a slight abuse of notation, we use $r,i \attMod \phi$ to denote both the judgment and its derivation.
The length of a derivation, denoted $|r, i \attMod \phi|$, is the number of rule applications in it.

\begin{theorem}\thlabel{theorem:f:conf:confidentiality}
Let $u$ be a user in ${\cal U}$, $P = \langle M, f_{\mathit{conf}}^{u} \rangle$ be an \accessControlConfiguration{}, where $M$ is a system configuration and $f_{\mathit{conf}}^{u}$ is as above.
The \acf{} $f_{\mathit{conf}}^{u}$ provides \confidentiality{}  with respect to $P$, $u$, $\attackerModel$, and $\cong_{P,u}$.
\end{theorem}

\begin{proof}

Let $u$ be a user in ${\cal U}$, $P = \langle M, f_{\mathit{conf}}^{u} \rangle$ be an \accessControlConfiguration{}, where $M$ is a system configuration and $f_{\mathit{conf}}^{u}$ is as above, and $L$ be the $P$-LTS.
Furthermore, let $r$ be a run in $\mathit{traces}(L)$, $i$ be an integer such that $1 \leq i \leq |r|$, and $\phi$ be a sentence such that $r, i \attMod \phi$ holds.
We claim that also $\mathit{secure}_{P,u}(r,i \attMod \phi)$ holds.
The theorem follows trivially from the claim.

We now prove our claim that  $\mathit{secure}_{P,u}(r,i \attMod \phi)$ holds.
Let $r$ be a run in $\mathit{traces}(L)$, $i$ be an integer such that $1 \leq i \leq |r|$, and $\phi$ be a sentence such that $r, i \attMod \phi$ holds.
Furthermore, in the following we use $e$ to denote the $\mathit{extend}$ function.
We prove our claim by induction on the length of the derivation $r, i \attMod \phi$.

\smallskip
\noindent
{\bf Base Case: } 
Assume that $|r,i \attMod \phi| = 1$.
There are a number of cases depending on the rule used to obtain $r, i \attMod \phi$.
\begin{compactenum}
\item \emph{\texttt{SELECT} Success - 1}.
Let $i$ be such that $r^{i} = r^{i-1} \concat \langle u, \mathtt{SELECT}, \phi \rangle \concat s$, where $s = \langle \mathit{db}, U, \mathit{sec}, T,V,c\rangle \in \Omega_{M}$ and $\mathit{last}(r^{i-1}) = s'$, where $s' =  \langle \mathit{db}, U, \mathit{sec}, T,V,c'\rangle$. 
From the rules, it follows that $f_{\mathit{conf}}^{u}(s', \langle u, \mathtt{SELECT}, \phi\rangle) = \top$.
From this and $f_{\mathit{conf}}^{u}$'s definition, it follows that  $\mathit{secure}(u, \phi, s') = \top$ holds.
From this, \thref{theorem:secure:equivalent:modulo:indistinguishable:state}, and $\mathit{pState}(s) = \mathit{pState}(s')$, it follows that $\mathit{secure}(u, \phi, s) = \top$ holds.
From this,  \thref{theorem:secure:sound:under:approximation}, and $\mathit{last}(r^{i}) = s$, it follows that $\mathit{secure}_{P,u}(r,i \attMod \phi)$ holds.

\item \emph{\texttt{SELECT} Success - 2}.
The proof for this case is similar to that of \emph{\texttt{SELECT} Success - 1}.

\item \emph{\texttt{INSERT} Success}.
Let $i$ be such that $r^{i} = r^{i-1} \concat \langle u, \mathtt{INSERT}, \\ R, \overline{t} \rangle \concat s$ , where $s = \langle \mathit{db}, U, \mathit{sec}, T, V, c \rangle \in \Omega_{M}$ and $\mathit{last}(r^{i-1}) =  \langle \mathit{db}', U, \mathit{sec}, T, V, c' \rangle$, and $\phi$ be $R(\overline{t})$.
Then, $\mathit{secure}_{P,u}(r,i \attMod R(\overline{t}))$  holds.
Indeed, in all runs $r'$ indistinguishable from $r^{i}$ the last action is $\langle u, \mathtt{INSERT}, R, \\ \overline{t} \rangle$.
Furthermore, the action has been executed successfully.
Therefore, according to the LTS rules, $\overline{t} \in \mathit{db}''(R)$, where $\mathit{db}'' = \mathit{last}(r').\mathit{db}$.
From this and the relational calculus semantics, it follows that $[R(\overline{t})]^{\mathit{last}(r').\mathit{db}} \\ = \top$.
Therefore, $[R(\overline{t})]^{\mathit{last}(r').\mathit{db}} = \top$ for any run $r' \in \llbracket r^{i}\rrbracket_{P,u}$.
Hence, $\mathit{secure}_{P,u}(r,i \attMod R(\overline{t}))$ holds.

\item \emph{\texttt{INSERT} Success - FD}.
Let $i$ be such that $r^{i} = r^{i-1} \concat \langle u, \mathtt{INSERT}, R, (\overline{v}, \overline{w}, \overline{q}) \rangle \concat s$, where $s = \langle \mathit{db}, U, \mathit{sec}, T, V, c \rangle \\ \in \Omega_{M}$ and $\mathit{last}(r^{i-1}) =  \langle \mathit{db}', U, \mathit{sec}, T, V, c' \rangle$, and $\phi$ be $\neg \exists \overline{y},\overline{z}.\, R(\overline{v}, \overline{y}, \overline{z}) \wedge \overline{y} \neq \overline{w}$.
From the rule's definition, it follows that $\mathit{secEx}(s) = \bot$.
From this and the LTS rules, it follows that $f_{\mathit{conf}}^{u}(s', \langle u, \mathtt{INSERT}, R, (\overline{v}, \overline{w}, \overline{q}) \rangle) = \top$.
From this and $f_{\mathit{conf}}^{u}$'s definition, it follows that  $\mathit{secure}(u, \phi, \mathit{last}(r^{i-1})) = \top$ holds because $\phi$ is equivalent to $\mathit{getInfoS}(\gamma,a)$ for some $\gamma \in \mathit{Dep}(\Gamma, a)$, where $a = \langle u, \mathtt{INSERT}, R, (\overline{v}, \overline{w}, \overline{q}) \rangle$.
From this and \thref{theorem:secure:sound:under:approximation}, it follows that $\mathit{secure}_{P,u}(r,i-1 \attMod \phi)$ holds.
We claim that $\mathit{secure}^{\mathit{data}}_{P,u}(r,i \attMod \phi)$ holds.
From \thref{theorem:ibsec:correctness:secure:3} and $\mathit{secure}^{\mathit{data}}_{P,u}(r,i \attMod \phi)$, it follows $\mathit{secure}_{P,u}(r,i  \attMod \phi)$.

We now prove our claim that  $\mathit{secure}^{\mathit{data}}_{P,u}(r,i \attMod \phi)$ holds.
Let $s'$ be the state $\mathit{last}(r^{i-1})$.
Furthermore, for brevity's sake, in the following we omit the $\mathit{pState}$ function where needed.
For instance, with a slight abuse of notation, we write $\llbracket s' \rrbracket^{\mathit{data}}_{u,M}$ instead of $\llbracket \mathit{pState}(s') \rrbracket^{\mathit{data}}_{u,M}$.
There are two cases:
\begin{compactenum}
\item the \texttt{INSERT} command has caused an integrity constraint violation, i.e., $\mathit{Ex}(s) \neq \emptyset$.
From $\mathit{secure}(u, \phi, \\ s') = \top$ and \thref{theorem:secure:sound:under:approximation}, it follows that $\mathit{secure}^{\mathit{data}}_{P,u}(r,\\i-1 \attMod \phi)$ holds.
From this, it follows that $[\phi]^{v} = [\phi]^{s'}$ for any $v \in \llbracket s' \rrbracket^{\mathit{data}}_{u,M}$.
From this and the fact that the \texttt{INSERT} command caused an exception (i.e., $s' = s$), it follows that $[\phi]^{v} = [\phi]^{s}$ for any $v \in \llbracket s \rrbracket^{\mathit{data}}_{u,M}$.
From this, it follows that $\mathit{secure}^{\mathit{data}}_{P,u}(r,i \attMod \phi)$ holds.

\item the \texttt{INSERT} command has not caused exceptions, i.e., $\mathit{Ex}(s) = \emptyset$.
From $\mathit{secure}(u, \phi,s') = \top$ and \thref{theorem:secure:sound:under:approximation}, it follows that $\mathit{secure}^{\mathit{data}}_{P,u}(r,i-1 \attMod \phi)$ holds.
From this, it follows that $[\phi]^{v} = [\phi]^{s'}$ for any $v \in \llbracket s' \rrbracket^{\mathit{data}}_{u,M}$.
Furthermore, from \thref{theorem:getInfo:sound:and:complete} and $\mathit{Ex}(s) = \emptyset$, it follows that $\phi$ holds in $s'$.
Let $A_{s',R,\overline{t}}$ be the set $\{\langle \mathit{db}[R \oplus \overline{t}], U, \mathit{sec}, T,V\rangle \in \Pi_{M} \,|\, \exists \mathit{db}' \in \Omega_{D}.\, \langle \mathit{db}', U,  \mathit{sec},  T, V\rangle \in \llbracket s'\rrbracket^{\mathit{data}}_{u,M}\}$.
It is easy to see that $\llbracket s\rrbracket^{\mathit{data}}_{u,M} \subseteq A_{s',R,\overline{t}}$.
We now show that $\phi$ holds for any $z \in A_{s',R,\overline{t}}$.
Let $z_{1} \in \llbracket s' \rrbracket^{\mathit{data}}_{u,M}$.
From $[\phi]^{v} = [\phi]^{s'}$ for any $v \in \llbracket s' \rrbracket^{\mathit{data}}_{u,M}$ and the fact that $\phi$ holds in $s'$, it follows that $[\phi]^{z_{1}} = \top$.
Therefore, for any $(\overline{k}_{1},\overline{k}_{2}, \overline{k}_{3}) \in R(z_{1})$ such that $|\overline{k}_1| = |\overline{v}|$, $|\overline{k}_2| = |\overline{w}|$, and $|\overline{k}_3| = |\overline{z}|$, if $k_{1} = \overline{v}$, then $k_{2} = \overline{w}$. 
Then, for any $(\overline{k}_{1},\overline{k}_{2}, \overline{k}_{3}) \in R(z_{1}) \cup \{(\overline{v}, \overline{w}, \overline{q})\}$, if $k_{1} = \overline{v}$, then $k_{2} = \overline{w}$.
Therefore, $\phi$ holds also in $z_{1}[R \oplus \overline{t}] \in A_{\mathit{pState}(s'), R, \overline{t}}$.	
Hence, $[\phi]^{z} = \top$ for any $z \in A_{s',R,\overline{t}}$.
From this and $\llbracket s\rrbracket^{\mathit{data}}_{u,M} \subseteq A_{s',R,\overline{t}}$, it follows that  $[\phi]^{z} = \top$ for any $z \in \llbracket s\rrbracket^{\mathit{data}}_{u,M}$.
From this, it follows that $\mathit{secure}^{\mathit{data}}_{P,u}(r,i \attMod \phi)$ holds.

\end{compactenum}

\item \emph{\texttt{INSERT} Success - ID}.
The proof of this case is similar to that for the \emph{\texttt{INSERT} Success - FD}.

\item \emph{\texttt{DELETE} Success}.
The proof for this case is similar to that of \emph{\texttt{INSERT} Success}.

\item \emph{\texttt{DELETE} Success - ID}.
The proof of this case is similar to that for the \emph{\texttt{INSERT} Success - FD}.

\item \emph{\texttt{INSERT} Exception}.
Let $i$ be such that $r^{i} = r^{i-1} \concat \langle u, \mathtt{INSER}, R, \overline{t} \rangle \concat s$, where $s = \langle \mathit{db}, U, \mathit{sec}, T, V, c \rangle  \in  \Omega_{M}$ and $\mathit{last}(r^{i-1}) =  \langle \mathit{db}', U, \mathit{sec}, T, V, c' \rangle$, and $\phi$ be $\neg R(\overline{t})$.
From the rule's definition, it follows that $\mathit{secEx}(s) = \bot$.
From this and the LTS rules, it follows that $f_{\mathit{conf}}^{u}(s', \\ \langle u, \mathtt{INSERT}, R, \overline{t} \rangle) = \top$.
From this and $f_{\mathit{conf}}^{u}$'s definition, it follows that  $\mathit{secure}(u, \phi, \mathit{last}(r^{i-1})) = \top$ holds because $\phi = \mathit{getInfo}(\langle u, \mathtt{INSERT}, R, \overline{t} \rangle)$.
From this and \thref{theorem:secure:sound:under:approximation}, it follows that $\mathit{secure}_{P,u}(r,i-1 \attMod \phi)$ holds.
From the LTS semantics, it follows that $\mathit{pState}(s) \\ \cong^{data}_{u,M} \mathit{pState}(\mathit{last}(r^{i-1}))$.
From this, \thref{theorem:secure:equivalent:modulo:indistinguishable:state}, and $\mathit{secure}(u,  \phi, \mathit{last}(r^{i-1})) = \top$, it follows that  $\mathit{secure}(u, \phi, \\ \mathit{last}(r^{i})) = \top$.
From this and  \thref{theorem:secure:sound:under:approximation}, it follows that $\mathit{secure}_{P,u}(r,i \attMod \phi)$ holds.

\item \emph{\texttt{DELETE} Exception}. 
The proof for this case is similar to that of \emph{\texttt{INSERT} Exception}.

\item \emph{\texttt{INSERT} FD Exception}.
Let $i$ be such that $r^{i} = r^{i-1} \concat \langle u, \mathtt{INSERT}, R, (\overline{v}, \overline{w}, \overline{q}) \rangle \concat s$, where $s = \langle \mathit{db}, U, \mathit{sec}, T, V, c \rangle \\ \in \Omega_{M}$ and $\mathit{last}(r^{i-1}) =  \langle \mathit{db}', U, \mathit{sec}, T, V, c' \rangle$, and $\phi$ be $\exists \overline{y},\overline{z}.\, R(\overline{v}, \overline{y}, \overline{z}) \wedge \overline{y} \neq \overline{w}$.
From the rule's definition, it follows that $\mathit{secEx}(s) = \bot$.
From this and the LTS rules, it follows that $f_{\mathit{conf}}^{u}(s', \langle u, \mathtt{INSERT}, R, (\overline{v}, \overline{w}, \overline{q}) \rangle) = \top$.
From this and $f_{\mathit{conf}}^{u}$'s definition, it follows that  $\mathit{secure}(u, \phi, \mathit{last}(r^{i-1})) = \top$ because $\phi = \mathit{getInfoV}(\gamma, \\ \langle u, \mathtt{INSERT}, R, (\overline{v}, \overline{w}, \overline{q}) \rangle)$ for some constraint $\gamma \in \mathit{Dep}(\Gamma, \\ \langle u, \mathtt{INSERT}, R, (\overline{v}, \overline{w}, \overline{q}) \rangle)$.
From this and  \thref{theorem:secure:sound:under:approximation}, it follows that $\mathit{secure}_{P,u}(r,i-1 \attMod \phi)$ holds.
From the LTS semantics, it follows that $\mathit{pState}(s) \cong^{data}_{u,M} \mathit{pState}(\\ \mathit{last}(r^{i-1}))$.
From this, $\mathit{secure}(u,  \phi, \mathit{last}(r^{i-1})) = \top$, and \thref{theorem:secure:equivalent:modulo:indistinguishable:state}, it follows that  $\mathit{secure}(u, \phi, \mathit{last}(r^{i})) = \top$.
From this and  \thref{theorem:secure:sound:under:approximation}, it follows that also $\mathit{secure}_{P,u}(r,i \attMod \phi)$ holds.

\item \emph{\texttt{INSERT} ID Exception}.
The proof for this case is similar to that of \emph{\texttt{INSERT} FD Exception}.

\item \emph{\texttt{DELETE} FD Exception}.
The proof for this case is similar to that of \emph{\texttt{INSERT} FD Exception}.

\item \emph{Integrity Constraint}.
The proof of this case follows trivially from the fact that for any state $s = \langle \mathit{db}, U, \mathit{sec}, T, \\ V,  c \rangle  \in \Omega_{M}$ and any $\gamma \in \Gamma$, $[\gamma]^{\mathit{db}} = \top$ holds by definition.

\item \emph{Learn \texttt{GRANT}/\texttt{REVOKE} Backward}. 
Let $i$ be such that $r^{i} = r^{i-1} \concat t \concat s$, where $s = \langle \mathit{db}, U, \mathit{sec}, T, V, c \rangle  \in \Omega_{M}$, $\mathit{last}(r^{i-1}) =  \langle \mathit{db}, U, \mathit{sec}', T, V, c' \rangle$, and $t$ be a trigger whose \texttt{WHEN} condition is $\phi$ and whose action is either a \texttt{GRANT} or a \texttt{REVOKE}.
From the rule's definition, it follows that $\mathit{secEx}(s) = \bot$. 
From this and the LTS rules, it follows that $f_{\mathit{conf}}^{u}(\mathit{last}(r^{i-1}), \langle u', \mathtt{SELECT}, \phi \rangle) = \top$, where $u'$ is either the trigger's owner or the trigger's invoker depending on the security mode.
From this and $f_{\mathit{conf}}^{u}$'s definition, it follows that $\mathit{secure}(u,\phi,\mathit{last}  (r^{i-1})) = \top$.
From this and \ref{theorem:secure:sound:under:approximation}, it follows that $\mathit{secure}_{P,u}(r,i-1 \attMod \phi)$ holds.

\item \emph{Trigger \texttt{GRANT} Disabled Backward}. 
Let $i$ be such that $r^{i} = r^{i-1} \concat t \concat s$, where $s = \langle \mathit{db}, U, \mathit{sec}, T, V, c \rangle  \in \Omega_{M}$, $\mathit{last}(r^{i-1}) =  \langle \mathit{db}, U, \mathit{sec}', T, V, c' \rangle$, and $t$ be a trigger whose \texttt{WHEN} condition is $\psi$, and $\phi$ be $\neg \psi$.
From the rule's definition, it follows that $\mathit{secEx}(s) = \bot$. 
From this and the LTS rules, it follows that $f_{\mathit{conf}}^{u}(\mathit{last}(r^{i-1}), \\ \langle u', \mathtt{SELECT}, \phi \rangle) = \top$, where $u'$ is either the trigger's owner or the trigger's invoker depending on the security mode.
From this and $f_{\mathit{conf}}^{u}$'s definition, it follows that $\mathit{secure}(u,\phi,\mathit{last}(r^{i-1})) = \top$.
From this and \thref{theorem:secure:sound:under:approximation}, it follows that $\mathit{secure}_{P,u}(r,i-1 \attMod \phi)$ holds.

\item \emph{Trigger \texttt{REVOKE} Disabled Backward}. 
The proof for this case is similar to that of \emph{Trigger \texttt{GRANT} Disabled Backward}.

\item \emph{Trigger \texttt{INSERT} FD Exception}.
Let $i$ be such that $r^{i} = r^{i-1} \concat t \concat s$, where $s = \langle \mathit{db}, U, \mathit{sec}, T, V, c \rangle  \in \Omega_{M}$, $\mathit{last}(r^{i-1}) \\ =  \langle \mathit{db}, U, \mathit{sec}', T, V, c' \rangle$, and $t$ be a trigger whose \texttt{WHEN} condition is $\phi$ and whose action $\mathit{act}$ is a \texttt{INSERT} statement $\langle u', \mathtt{INSERT}, R, (\overline{v}, \overline{w}, \overline{q}) \rangle$.
Furthermore, let $\phi$ be $\exists \overline{y},\overline{z}.\, R(\overline{v}, \overline{y}, \overline{z}) \wedge \overline{y} \neq \overline{w}$.
From the rule's definition, it follows that $\mathit{secEx}(s) = \bot$.
From this and the LTS rules, it follows that $f_{\mathit{conf}}^{u}( \mathit{last}(r^{i-1}), \mathit{act}) = \top$.
From this and $f_{\mathit{conf}}^{u}$'s definition, it follows that  $\mathit{secure}(u, \phi, \mathit{last} \\ (r^{i-1})) = \top$ because $\phi = \mathit{getInfoV}(\gamma, \mathit{act})$ for some constraint $\gamma \in \mathit{Dep}(\Gamma,  \mathit{act})$.
From this and  \thref{theorem:secure:sound:under:approximation}, it follows that $\mathit{secure}_{P,u}(r,i-1 \attMod \phi)$ holds.

\item \emph{Trigger \texttt{INSERT} ID Exception}.
The proof for this case is similar to that of \emph{Trigger \texttt{INSERT} ID Exception}.

\item \emph{Trigger \texttt{DELETE} ID Exception}.
The proof for this case is similar to that of \emph{Trigger \texttt{DELETE} ID Exception}.

\item \emph{Trigger Exception}.
Let $i$ be such that $r^{i} = r^{i-1} \concat t \concat s$, where $s = \langle \mathit{db}, U, \mathit{sec}, T, V, c \rangle  \in \Omega_{M}$, $\mathit{last}(r^{i-1}) =  \langle \mathit{db}, U, \mathit{sec}', T, V, c' \rangle$, and $t$ be a trigger whose \texttt{WHEN} condition is $\phi$ and whose action is $\mathit{act}$.
From the rule's definition, it follows $f_{\mathit{conf}}^{u}(\mathit{last}(r^{i-1}),  \langle u', \mathtt{SELECT}, \phi \rangle) = \top$, where $u'$ is either the trigger's owner or the trigger's invoker depending on the security mode.
From this and $f_{\mathit{conf}}^{u}$'s definition, it follows $\mathit{secure}(u,\phi,\mathit{last}(r^{i-1})) = \top$.
From this and \ref{theorem:secure:sound:under:approximation}, it follows that $\mathit{secure}_{P,u}(r,i-1 \attMod \phi)$ holds.

\item \emph{Trigger \texttt{INSERT} Exception}.
The proof for this case is similar to that of \emph{\texttt{INSERT} Exception}.

\item \emph{Trigger \texttt{DELETE} Exception}.
The proof for this case is similar to that of \emph{\texttt{DELETE} Exception}.

\item \emph{Trigger Rollback \texttt{INSERT}}.
Let $i$ be such that $r^{i} = r^{i-n-1}  \concat \langle u, \mathtt{INSERT}, R, \overline{t}\rangle \concat s_{1} \concat t_{1} \concat s_{2} \concat \ldots \concat t_{n} \concat s_{n}$, where $s_{1}, s_{2}, \ldots, s_{n} \\ \in \Omega_{M}$ and $t_{1}, \ldots, t_{n} \in {\cal TRIGGER}_{D}$, and $\phi$ be $\neg R(\overline{t})$.
Furthermore, let  $\mathit{last}(r^{i-n-1}) = \langle \mathit{db}', U', \mathit{sec}', T', V', c' \rangle$ and $s_n$ be $\langle \mathit{db}, U, \mathit{sec}, T, V, c \rangle$. 
From the rule's definition, it follows that $\mathit{secEx}(s_{1}) = \bot$.
From this, it follows that $f_{\mathit{conf}}^{u}(\mathit{last}(r^{i-n-1}),  \langle u, \mathtt{INSERT}, R, \overline{t}\rangle) = \top$.
From this and $f_{\mathit{conf}}^{u}$'s definition, it follows that $\mathit{secure}(u, \\ \phi,\mathit{last}(r^{i-n-1})) = \top$ because $\phi = \mathit{getInfo}(\langle u,  \mathtt{INSERT}, R, \\ \overline{t}\rangle)$.
From the LTS semantics, it follows that $\mathit{last}(r^{i-n-1}) \\ \cong_{u,M}^{\mathit{data}} s_{n}$.
From this,   $\mathit{secure}(u,\phi,\mathit{last}(r^{i-n-1})) = \top$, and \thref{theorem:secure:equivalent:modulo:indistinguishable:state}, it follows that $\mathit{secure}(u,  \phi,s_{n}) = \top$.
From this and \thref{theorem:secure:sound:under:approximation}, it follows that $\mathit{secure}_{P,u}(r,i \\ \attMod \phi)$ holds.

\item \emph{Trigger Rollback \texttt{DELETE}}.
The proof for this case is similar to that of \emph{Trigger Rollback \texttt{INSERT}}.
\end{compactenum}
This completes the proof of the base step.

\smallskip
\noindent
{\bf Induction Step: }
Assume that the claim hold for any derivation of $r, j \attMod \psi$ such that $|r, j \attMod \psi| < |r,i \attMod \phi|$.
We now prove that the claim also holds for $r,i \attMod \phi$.
There are a number of cases depending on the rule used to obtain $r,i \attMod \phi$.
\begin{compactenum}
\item \emph{View}.
The proof of this case follows trivially from the semantics of the relational calculus extended over views.

\item \emph{Propagate Forward \texttt{SELECT}}.
Let $i$ be such that $r^{i+1} = r^{i} \concat \langle u, \mathtt{SELECT}, \psi \rangle \concat s$, where $s = \langle \mathit{db}, U, \mathit{sec}, T, V, c \rangle  \in \Omega_{M}$ and $\mathit{last}(r^{i}) =  \langle \mathit{db}', U', \mathit{sec}', T', V', c' \rangle$.
From the rule, it follows that $r,i \attMod \phi$ holds.
From this and the induction hypothesis, it follows that $\mathit{secure}_{P,u}(r,i  \attMod \phi)$ holds.
From \thref{theorem:f:conf:pec:1}, the action $\langle u, \mathtt{SELECT}, \psi \rangle$ preserves the equivalence class with respect to $r^{i}$, $P$, and $u$.
From this,  \thref{theorem:secure:extend:on:runs:select:create}, and $\mathit{secure}_{P,u}(r,i \attMod \phi)$, it follows that also $\mathit{secure}_{P,u}(r,i+1 \attMod \phi)$ holds.

\item \emph{Propagate Forward \texttt{GRANT/REVOKE}}.
Let $i$ be such that $r^{i+1} = r^{i} \concat \langle \mathit{op}, u', p, u \rangle \concat s$, where $s = \langle \mathit{db}, U, \mathit{sec}, T, V, c \rangle  \in \Omega_{M}$ and $\mathit{last}(r^{i}) =  \langle \mathit{db}', U', \mathit{sec}', T', V', c' \rangle$.
From the rule, it follows that $r,i \attMod \phi$ holds.
From this and the induction hypothesis, it follows that $\mathit{secure}_{P,u}(r,i  \attMod \phi)$ holds.
From \thref{theorem:f:conf:pec:1}, the action $\langle \mathit{op}, u', p, u \rangle$ preserves the equivalence class with respect to $r^{i}$, $P$, and $u$.
From this,  \thref{theorem:secure:extend:on:runs:grant:revoke}, and $\mathit{secure}_{P,u}(r,i  \attMod \phi)$, it follows that also $\mathit{secure}_{P,u}(r,i+1 \attMod \phi)$ holds.

\item \emph{Propagate Forward \texttt{CREATE}}.
The proof for this case is similar to that of \emph{Propagate Forward \texttt{SELECT}}.

\item \emph{Propagate Backward \texttt{SELECT}}.
Let $i$ be such that $r^{i+1} = r^{i} \concat \langle u, \mathtt{SELECT}, \psi \rangle \concat s$, where $s = \langle \mathit{db}', U', \mathit{sec}', T', V', c' \rangle  \\ \in \Omega_{M}$ and $\mathit{last}(r^{i}) =  \langle \mathit{db}, U, \mathit{sec}, T, V, c \rangle$.
From the rule, it follows that $r,i+1 \attMod \phi$ holds.
From this and the induction hypothesis, it follows that $\mathit{secure}_{P,u}  (r,i+1  \attMod \phi)$ holds.
From \thref{theorem:f:conf:pec:1}, the action $\langle u, \mathtt{SELECT}, \psi \rangle$ preserves the equivalence class with respect to $r^{i}$, $P$, and $u$.
From this,  \thref{theorem:secure:extend:on:runs:select:create}, and $\mathit{secure}_{P,u}(r,i+1 \\ \attMod \phi)$, it follows that also $\mathit{secure}_{P,u}(r,i \attMod \phi)$ holds.

\item \emph{Propagate Backward \texttt{GRANT/REVOKE}}.
Let $i$ be such that $r^{i+1} = r^{i} \concat \langle \mathit{op},u',p,u \rangle \concat s$, where $s = \langle \mathit{db}', U', \mathit{sec}', T', V',\\ c' \rangle   \in \Omega_{M}$ and $\mathit{last}(r^{i}) =  \langle \mathit{db}, U, \mathit{sec}, T, V, c \rangle$.
From the rule, it follows that $r,i+1 \attMod \phi$ holds.
From this and the induction hypothesis, it follows that $\mathit{secure}_{P,u}  (r,i+1 \attMod \phi)$ holds.
From \thref{theorem:f:conf:pec:1}, the action $\langle \mathit{op},u',p, \\ u \rangle$ preserves the equivalence class with respect to $r^{i}$, $P$, and $u$.
From this,  \thref{theorem:secure:extend:on:runs:grant:revoke}, and $\mathit{secure}_{P,u}(r,i+1 \attMod \phi)$, it follows that also $\mathit{secure}_{P,u}(r,i \attMod \phi)$ holds.

\item \emph{Propagate Backward \texttt{CREATE TRIGGER}}.
The proof for this case is similar to that of \emph{Propagate Backward \texttt{SELECT}}.

\item \emph{Propagate Backward \texttt{CREATE VIEW}}.
Note that the formulae $\psi$ and $\mathit{replace}(\psi,o)$ are semantically equivalent.
This is the only difference between the proof for this case and the one for the  \emph{Propagate Backward \texttt{SELECT}} case.

\item \emph{Rollback Backward - 1}.
Let $i$ be such that $r^{i} = r^{i-n-1}  \concat \langle u, \mathit{op}, R, \overline{t}\rangle \concat s_{1} \concat t_{1} \concat s_{2} \concat \ldots \concat t_{n} \concat s_{n}$, where  $s_{1}, s_{2}, \ldots, s_{n} \in \Omega_{M}$, $t_{1}, \ldots, t_{n} \in {\cal TRIGGER}_{D}$, and  $\mathit{op}$ is one of $\{\mathtt{INSERT}, \\ \texttt{DELETE}\}$.
Furthermore, let $s_{n}$ be $\langle \mathit{db}', U', \mathit{sec}', T',  V', c' \rangle$ and $\mathit{last}(r^{i-n-1})$ be $\langle \mathit{db}, U, \mathit{sec}, T, V, c \rangle$.
From the rule's definition, $r, i \attMod \phi$ holds.
From this and the induction hypothesis, it follows that $\mathit{secure}_{P,u}(r,i \attMod \phi)$ holds.
From \thref{theorem:f:conf:pec:2}, the triggers $t_{j}$ preserve the equivalence class with respect to $r^{i-n-1+j}$, $P$, and $u$ for any $1 \leq j \leq n$.
Therefore, for any $v \in \llbracket r^{i-1} \rrbracket_{{P,u}}$, the run $e(v,t_{n})$ contains the roll-back.
Therefore, for any $v \in \llbracket r^{i-1} \rrbracket_{P,u}$, the state $\mathit{last}(e(v,t_{n}))$ is the state just before the action $\langle u,\mathit{op}, R, \overline{t}\rangle$.
Let $A$ be the set of partial states associated with the roll-back states.
It is easy to see that $A$ is the same as $\{\mathit{pState}(\mathit{last}(t')) | t' \in \llbracket r^{i-n-1} \rrbracket_{P,u}\}$.
From $\mathit{secure}_{P,u}(r,i \attMod \phi)$, it follows that $\phi$ has the same result over all states in $A$.
From this and $A = \{\mathit{pState}(\mathit{last}(t')) | t' \in \llbracket r^{i-n-1} \rrbracket_{P,u}\}$, it follows that $\phi$ has the same result over all states in $ \{\mathit{pState}(\mathit{last}(t')) | \\ t' \in \llbracket r^{i-n-1} \rrbracket_{P,u}\}$.
From this, it follows that $\mathit{secure}_{P,u}\\(r,i-n-1 \attMod \phi)$ holds.

\item \emph{Rollback Backward - 2}.
Let $i$ be such that $r^{i} = r^{i-1} \concat \langle u, op, R, \overline{t} \rangle \concat s$, where $s = \langle \mathit{db}', U', \mathit{sec}', T', V', c' \rangle  \in \Omega_{M}$, $\mathit{last}(r^{i-1}) =  \langle \mathit{db}, U,  \mathit{sec}, T, V, c \rangle$, and $\mathit{op}$ is one of $\{\mathtt{INSERT}, \mathtt{DELETE}\}$.
From the rule's definition, $r, i \attMod \phi$ holds.
From this and  the induction hypothesis, it follows that $\mathit{secure}_{P,u}(r,i \attMod \phi)$ holds.
From \thref{theorem:f:conf:pec:1}, the action $\langle u,\mathit{op}, R, \overline{t}\rangle$ preserves the equivalence class with respect to $r^{i-1}$, $P$, and $u$. 	
From this, \thref{theorem:secure:extend:on:runs:insert:delete}, the fact that the action does not modify the database state, and $\mathit{secure}_{P,u}(r,i \attMod \phi)$, it follows $\mathit{secure}_{P,u}(r, i-1 \attMod \phi)$.

\item \emph{Rollback Forward - 1}.
Let $i$ be such that $r^{i} = r^{i-n-1}  \concat \langle u, \mathit{op}, R, \overline{t}\rangle \concat s_{1} \concat t_{1} \concat s_{2} \concat \ldots \concat t_{n} \concat s_{n}$, where $s_{1}, s_{2}, \ldots, s_{n} \in \Omega_{M}$, $t_{1}, \ldots, t_{n} \in {\cal TRIGGER}_{D}$, and $\mathit{op}$ is one of $\{\mathtt{INSERT}, \\ \mathtt{DELETE}\}$.
Furthermore, let $s_{n}$ be $\langle \mathit{db}, U, \mathit{sec}, T, V, c \rangle $ and $\mathit{last}(r^{i-n-1})$ be $\langle \mathit{db}', U', \mathit{sec}', T', V', c' \rangle$.
From the rule's definition, $r, i-n-1 \attMod \phi$ holds.
From this and the induction hypothesis, it follows that $\mathit{secure}_{P,u} (r,i-n-1 \attMod \phi)$ holds.
From \thref{theorem:f:conf:pec:2}, the triggers $t_{j}$ preserve the equivalence class with respect to $r^{i-n-1+j}$, $P$, and $u$ for any $1 \leq j \leq n$.
Independently on the cause of the roll-back (either a security exception or an integrity constraint violation), we claim that the set $A$ of roll-back partial states is $\{\mathit{pState}(\mathit{last}(t')) | t' \in \llbracket r^{i-n-1} \rrbracket_{P,u}\}$.
From $\mathit{secure}_{P,u}  (r,  i-n-1 \attMod \phi)$, the result of $\phi$ is the same for all states in $A$.
From this and $A = \{\mathit{pState}(\mathit{last}(t')) | t' \in \llbracket r^{i-n-1} \rrbracket_{P,u}\}$, it follows that also  $\mathit{secure}_{P,u}(r,i \attMod \phi)$ holds.

We now prove our claim.
It is trivial to see (from the LTS's semantics) that the set of rollback's states is a subset of $\{\mathit{pState}(\mathit{last}(v)) | v \in \llbracket r^{i-n-1} \rrbracket_{P,u}\}$.
Assume, for contradiction's sake, that there is a state in $\{\mathit{pState}(\mathit{last}(v)) | v \in \llbracket r^{i-n-1} \rrbracket_{P,u}\}$ that is not a rollback state for the runs in $\llbracket r^{i} \rrbracket_{P,u}$.
This is impossible since all triggers $t_{1}, \ldots, t_{n}$ preserve the equivalence class. 

\item \emph{Rollback Forward - 2}.
Let $i$ be such that $r^{i} = r^{i-1} \concat \langle u, op, R, \overline{t} \rangle \concat s$, where $\mathit{op} \in \{\mathtt{INSERT}, \mathtt{DELETE}\}$, $s = \langle \mathit{db}, U, \\ \mathit{sec}, T, V, c \rangle  \in \Omega_{M}$ and $\mathit{last}(r^{i-1}) =  \langle \mathit{db}', U', \mathit{sec}',    T', V', \\ c' \rangle$.
From the rule's definition, $r, i-1 \attMod \phi$ holds.
From this and  the induction hypothesis, it follows that $\mathit{secure}_{P,u}(r,i-1 \attMod \phi)$ holds.
From \thref{theorem:f:conf:pec:1}, the action $\langle u,\mathit{op}, R, \overline{t}\rangle$ preserves the equivalence class with respect to $r^{i-1}$, $P$, and $u$. 	
From this, \thref{theorem:secure:extend:on:runs:insert:delete}, the fact that the action does not modify the database state, and $\mathit{secure}_{P,u}(r,i-1 \attMod \phi)$, it follows that also $\mathit{secure}_{P,u}(r,i \attMod \phi)$ holds.

\item \emph{Propagate Forward \texttt{INSERT/DELETE} Success}.
Let $i$ be such that $r^{i} = r^{i-1} \concat \langle u, op, R, \overline{t} \rangle \concat s$, where $\mathit{op} \in \{\mathtt{INSERT}, \\ \mathtt{DELETE}\}$, $s = \langle \mathit{db}, U, \mathit{sec}, T, V, c \rangle  \in \Omega_{M}$ and $\mathit{last}(r^{i-1}) =  \langle \mathit{db}', U', \mathit{sec}',  T', V', c' \rangle$.
From the rule's definition, $r, i-1 \attMod \phi$ holds.
From this and the induction hypothesis, it follows that $\mathit{secure}_{P,u}(r,i-1 \attMod \phi)$ holds.
From \thref{theorem:f:conf:pec:1}, the action $\langle u,\mathit{op}, R, \overline{t}\rangle$ preserves the equivalence class with respect to $r^{i-1}$, $P$, and $u$. 	
From $\mathit{reviseBelif}(r^{i-1}, \phi, r^{i})$, it follows that the execution of $\langle u,\mathit{op}, R, \overline{t}\rangle$ does not alter the content of the tables in $\mathit{tables}(\phi)$  for any $v \in \llbracket r^{i-1}\rrbracket_{P,u}$.
From this, \thref{theorem:secure:extend:on:runs:insert:delete}, and $\mathit{secure}_{P,u}(r,i-1 \attMod \phi)$, it follows that $\mathit{secure}_{P,u} \\ (r,i \attMod \phi)$ holds.

\item \emph{Propagate Forward \texttt{INSERT} Success - 1}.
Let $i$ be such that $r^{i} = r^{i-1} \concat \langle u, op, R, \overline{t} \rangle \concat s$, where $\mathit{op}$ is one if $\{\mathtt{INSERT}, \\ \mathtt{DELETE}\}$, $s = \langle \mathit{db}, U, \mathit{sec}, T, V, c \rangle  \in \Omega_{M}$ and $\mathit{last}(r^{i-1}) =  \langle \mathit{db}', U', \mathit{sec}',  T', V', c' \rangle$.
From the rule's definition, $r, i-1 \attMod \phi$ holds.
From this and the induction hypothesis, it follows that $\mathit{secure}_{P,u}(r,i-1 \attMod \phi)$ holds.
From \thref{theorem:f:conf:pec:1}, the action $\langle u,\mathit{op}, R, \overline{t}\rangle$ preserves the equivalence class with respect to $r^{i-1}$, $P$, and $u$. 	
We claim that the execution of $\langle u,\mathit{\mathtt{INSERT}}, R, \overline{t}\rangle$ does not alter the content of the tables in $\mathit{tables}(\phi)$.
From this, $\mathit{secure}_{P,u}(r, \\ i-1 \attMod \phi)$, and  \thref{theorem:secure:extend:on:runs:insert:delete}, it follows that $\mathit{secure}_{P,u}  (r,i \\ \attMod \phi)$ holds.

We now prove our claim that the execution of $\langle u,\mathit{\mathtt{INSERT}}, \\ R, \overline{t}\rangle$ does not alter the content of the tables in $\mathit{tables}(\phi)$.
From the rule's definition, it follows that $r, i-1 \attMod R(\overline{t})$ holds.
From this and \thref{theorem:attacker:model:sound}, it follows that $[R(\overline{t})]^{\mathit{last}(r^{i-1}).\mathit{db}} = \top$.
From $r, i-1 \attMod R(\overline{t})$ and the induction hypothesis, it follows that $\mathit{secure}_{P,u}(r,i-1,u, R(\overline{t}))$ holds.
From this and $[R(\overline{t})]^{\mathit{last}(r^{i-1}).\mathit{db}} = \top$, it follows that $[R(\overline{t})]^{\mathit{last}(v).\mathit{db}} = \top$ for any $v \in \llbracket r^{i-1}\rrbracket_{P,u}$.
From this and the relational calculus semantics, it follows that the execution of $\langle u,\mathit{op}, R, \overline{t}\rangle$ does not alter the content of the tables in $\mathit{tables}(\phi)$ for any $v \in \llbracket r^{i-1}\rrbracket_{P,u}$.

\item \emph{Propagate Forward \texttt{DELETE} Success - 1}.
The proof for this case is similar to that of \emph{Propagate Forward \texttt{INSERT} Success - 1}.

\item \emph{Propagate Backward \texttt{INSERT/DELETE} Success}.
Let $i$ be such that $r^{i} = r^{i-1} \concat \langle u, op, R, \overline{t} \rangle \concat s$, where $\mathit{op} \in \{\mathtt{INSERT}, \\ \mathtt{DELETE}\}$, $s = \langle \mathit{db}, U, \mathit{sec}, T, V, c \rangle  \in \Omega_{M}$ and $\mathit{last}(r^{i-1}) =  \langle \mathit{db}', U', \mathit{sec}',  T', V', c' \rangle$.
From the rule's definition, $r, i  \\ \attMod \phi$ holds.
From this and the induction hypothesis, it follows that $\mathit{secure}_{P,u}(r,i \attMod \phi)$ holds.
From \thref{theorem:f:conf:pec:1}, the action $\langle u,\mathit{op}, R, \overline{t}\rangle$ preserves the equivalence class with respect to $r^{i-1}$, $P$, and $u$. 	
From $\mathit{reviseBelif}(r^{i-1}, \phi,  r^{i})$, it follows that the execution of $\langle u,\mathit{op}, R, \overline{t}\rangle$ does not alter the content of the tables in $\mathit{tables}(\phi)$  for any $v \in \llbracket r^{i-1}\rrbracket_{P,u}$.
From this, \thref{theorem:secure:extend:on:runs:insert:delete}, and $\mathit{secure}_{P,u}(r,i \attMod \phi)$, it follows that $\mathit{secure}_{P,u}  (r,i-1 \attMod \phi)$ holds.

\item \emph{Propagate Backward \texttt{INSERT} Success - 1}.
Let $i$ be such that $r^{i} = r^{i-1} \concat \langle u, op, R, \overline{t} \rangle \concat s$, where $\mathit{op}$ is one of $\{\mathtt{INSERT},  \mathtt{DELETE}\}$, $s = \langle \mathit{db}, U, \mathit{sec}, T, V, c \rangle  \in \Omega_{M}$, and $\mathit{last}(r^{i-1}) =  \langle \mathit{db}', U', \mathit{sec}',  T', V', c' \rangle$.
From the rule's definition, $r, i \attMod \phi$ holds.
From this and the induction hypothesis, it follows that $\mathit{secure}_{P,u}(r,i \attMod \phi)$ holds.
From \thref{theorem:f:conf:pec:1}, the action $\langle u,\mathit{op}, R, \overline{t}\rangle$ preserves the equivalence class with respect to $r^{i-1}$, $P$, and $u$. 	
We claim that the execution of $\langle u,\mathit{\mathtt{INSERT}}, R, \overline{t}\rangle$ does not alter the content of the tables in $\mathit{tables}(\phi)$ for any $v \in \llbracket r^{i-1}\rrbracket_{P,u}$ (the proof of this claim is in the proof of the \emph{Propagate Forward \texttt{INSERT} Success - 1} case).
From this, \thref{theorem:secure:extend:on:runs:insert:delete}, and $\mathit{secure}_{P,u}(r,i \attMod \phi)$, it follows that $\mathit{secure}_{P,u}  (r,i-1 \attMod \phi)$ holds.

\item \emph{Propagate Backward \texttt{DELETE} Success - 1}.
The proof for this case is similar to that of \emph{Propagate Forward \texttt{DELETE} Success - 1}.

\item \emph{Reasoning}.
Let $\Delta$ be a subset of $\{ \delta \,| \, r,i \attMod \delta\}$ and $\mathit{last}(r^{i}) =  \langle \mathit{db}, U, \mathit{sec},  T, V, c \rangle$.
From the induction hypothesis, it follows that $\mathit{secure}_{P,u}(r,i \attMod \delta)$ holds for any $\delta \in \Delta$.
Note that, given any $\delta \in \Delta$, from $r,i \attMod \delta$ and \thref{theorem:attacker:model:sound}, it follows that $\delta$ holds in $\mathit{last}(r^{i})$.
From this, $\mathit{secure}_{P,u}(r,i \attMod \delta)$ holds for any $\delta \in \Delta$, $\Delta \models_{\mathit{fin}} \phi$, and \thref{theorem:reasoning:preserves:security}, it follows that  $\mathit{secure}_{P,u}(r, \\ i \attMod \phi)$ holds.

\item \emph{Learn \texttt{INSERT} Backward - 3}.
Let $i$ be such that $r^{i} = r^{i-1} \concat \langle u, \mathtt{INSERT}, R, \overline{t} \rangle \concat s$, where $s = \langle \mathit{db}', U', \mathit{sec}', T', V', \\ c' \rangle  \in \Omega_{M}$ and $\mathit{last}(r^{i-1}) =  \langle \mathit{db}, U, \mathit{sec},  T, V, c \rangle$, and $\phi$ be $\neg R(\overline{t})$.
From the rule's definition, $\mathit{secEx}(s) = \bot$.
From this and the LTS rules, it follows that $f_{\mathit{conf}}^{u}(\mathit{last}(r^{i-1}), \\ \langle u, \mathtt{INSERT}, R, \overline{t} \rangle) = \top$.
From this and $f_{\mathit{conf}}^{u}$'s definition, it follows that $\mathit{secure}(u,\phi,\mathit{last}(r^{i-1})) = \top$ because $\phi = \mathit{getInfo}(\langle u,  \mathtt{INSERT}, R, \overline{t}\rangle)$.
From this and \thref{theorem:secure:sound:under:approximation}, it follows that $\mathit{secure}_{P,u}(r,i-1 \attMod \phi)$ holds.

\item \emph{Learn \texttt{DELETE} Backward - 3}.
The proof for this case is similar to that of \emph{Learn \texttt{INSERT} Backward - 3}.

\item \emph{Propagate Forward Disabled Trigger}.
Let $i$ be such that $r^{i} = r^{i-1} \concat t \concat s$, where $s = \langle \mathit{db}, U, \mathit{sec}, T, V, c \rangle  \in \Omega_{M}$, $\mathit{last}(r^{i-1}) =  \langle \mathit{db}, U, \mathit{sec},  T, V, c \rangle$, and $t$ be a trigger.
Furthermore, let $\psi$ be $t$'s condition where all free variables are replaced with $\mathit{tpl}(\mathit{last}(r^{i-1}))$.
From the rule, it follows that $r,i-1 \attMod \phi$.
From this and the induction hypothesis, it follows that $\mathit{secure}_{P,u} (r, i-1 \attMod \phi)$ holds.
Furthermore, from \thref{theorem:f:composition:pec:2}, it follows that $t$ preserves the equivalence class with respect to $r^{i-1}$, $P$, and $u$.	
If the trigger's action is an \texttt{INSERT} or a \texttt{DELETE} operation, we claim that the operation does not change the content of any table in $\mathit{tables}(\phi)$ for any run $v \in \llbracket r^{i-1} \rrbracket_{P,u}$.
From this, the fact that $t$ preserves the equivalence class with respect to $r^{i-1}$, $P$, and $u$, \thref{theorem:secure:extend:on:runs:triggers}, and $\mathit{secure}_{P,u} (r, i-1 \attMod \phi)$, it follows that also $\mathit{secure}_{P,u} (r, i \attMod \phi)$ holds.

We now prove our claim.
Assume that $t$'s action in either an  \texttt{INSERT} or a \texttt{DELETE} operation.
From the rule, it follows that $r,i-1 \attMod \neg \psi$.
From this and \thref{theorem:attacker:model:sound}, $[\psi]^{\mathit{last}(r^{i-1})} = \bot$.
From $r,i-1 \attMod \neg \psi$ and the induction hypothesis, it follows that $\mathit{secure}_{P,u}  (r, i-1 \attMod \psi)$ holds.
From this and $[\psi]^{\mathit{last}(r^{i-1}).\mathit{db}} = \bot$, it follows that $[\psi]^{v.\mathit{db}} = \bot$ for any run $v \in \llbracket r^{i-1} \rrbracket_{P,u}$.
Therefore, the trigger $t$ is disabled in any run $v \in \llbracket r^{i-1} \rrbracket_{P,u}$.
From this and the LTS semantics, it follows that $t$'s execution does not change the content of any  table in $\mathit{tables}(\phi)$ for any run $v \in \llbracket r^{i-1} \rrbracket_{P,u}$.

\item \emph{Propagate Backward Disabled Trigger}.
The proof for this case is similar to that of  \emph{Propagate Forward Disabled Trigger}.

\item \emph{Learn \texttt{INSERT} Forward}.
Let $i$ be such that $r^{i} = r^{i-1} \concat t \concat s$, where $s = \langle \mathit{db}, U, \mathit{sec}, T, V, c \rangle  \in \Omega_{M}$, $\mathit{last}(r^{i-1}) =  \langle \mathit{db}, U, \mathit{sec},  T, V, c \rangle$, and $t$ be a trigger, and $\phi$ be $R(\overline{t})$.
Furthermore, let $\psi$ be $t$'s condition where all free variables are replaced with $\mathit{tpl}(\mathit{last}(r^{i-1}))$.
From the rule's definition, it follows  that $t$'s action is  $\langle u',\mathtt{INSERT}, R, \overline{t}\rangle$ and that $r, i-1 \attMod \psi$ holds.
From \thref{theorem:attacker:model:sound} and $r, i-1 \attMod \psi$, it follows that $[\psi]^{\mathit{last}(r^{i-1}).\mathit{db}} = \top$.
From this, $\mathit{secEx}(s) = \bot$, and $\mathit{Ex}(s) = \emptyset$, it follows that $t$'s action has been executed successfully.
From this, it follows that $\overline{t} \in s.\mathit{db}(R)$.
From  $r, i-1 \attMod \psi$ and the induction hypothesis, it follows $\mathit{secure}_{P,u} (r, i-1 \attMod \psi)$.
From this and $[\psi]^{\mathit{last}(r^{i-1}).\mathit{db}} = \top$, it follows that $[\psi]^{\mathit{last}(v).\mathit{db}} = \top$ for any $v \in \llbracket r^{i-1}\rrbracket_{P,u}$.
From this, it follows that the trigger $t$ is enabled in any run $v \in \llbracket r^{i-1}\rrbracket_{P,u}$.
From \thref{theorem:f:conf:pec:2}, it follows that $t$ preserves the equivalence class with respect to $r^{i-1}$, $P$, and $u$.
From this, $\mathit{secEx}(s) = \bot$, $\mathit{Ex}(s) = \emptyset$, and the fact that the trigger $t$ is enabled in any run $v \in \llbracket r^{i-1}\rrbracket_{P,u}$, it follows that $t$'s action is executed successfully in any run $e(v,t)$, where $v \in \llbracket r^{i-1}\rrbracket_{P,u}$.
From this,  it follows that $\overline{t} \in \mathit{db}''(R)$ for any $v \in \llbracket r^{i-1}\rrbracket_{P,u}$, where $\mathit{db}'' = \mathit{last}(e(v,t)).\mathit{db}$.
Therefore, $\mathit{secure}_{P,u} (r, i \attMod \phi)$ holds.

\item \emph{Learn \texttt{INSERT} - FD}.
Let $i$ be such that $r^{i} = r^{i-1} \concat t \concat s$, where $s = \langle \mathit{db}, U, \mathit{sec}, T, V, c \rangle \in \Omega_{M}$, $\mathit{last}(r^{i-1}) =  \langle \mathit{db}', U', \mathit{sec}', T', V', c' \rangle$, and $t \in {\cal TRIGGER}_{D}$, and $\phi$ be $\neg \exists \overline{y},\overline{z}.\, R(\overline{v}, \overline{y}, \overline{z}) \wedge \overline{y} \neq \overline{w}$.
Furthermore,  let $\psi$ be $t$'s condition where all free variables are replaced with the values in $\mathit{tpl}(\mathit{last}(r^{i-1}))$ and $\langle u', \mathtt{INSERT}, R, (\overline{v}, \overline{w}, \overline{q}) \rangle$ be $t$'s actual action.
From the rule, it follows that $r, i-1 \attMod \psi$.
From this and \thref{theorem:attacker:model:sound}, it follows that $[\psi]^{\mathit{last}(r^{i-1}).\mathit{db}} = \top$.
From this, $\mathit{Ex}(s) = \emptyset$, and  $\mathit{secEx}(s) \\ = \bot$, it follows that $f_{\mathit{conf}}^{u} (s',\langle u',\texttt{INSERT},R, \overline{t}\rangle)  = \top$, where $s'$ is the state just after the execution of the \texttt{SELECT} statement associated with $t$'s \texttt{WHEN} clause.
From this and $f_{\mathit{conf}}^{u}$'s definition, it follows that $\mathit{secure}(u, \phi, s') \\ = \top$.
From this, $\mathit{pState}(s')  = \mathit{pState}(\mathit{last}(r^{i-1}))$, and \thref{theorem:secure:equivalent:modulo:indistinguishable:state}, it follows that $\mathit{secure}(u, \phi, \mathit{last}(r^{i-1})) = \top$.
From this and \thref{theorem:secure:sound:under:approximation}, it follows also that $\mathit{secure}_{P,u} (r, i-1 \attMod \phi)$ holds.
We claim that $\mathit{secure}^{\mathit{data}}_{P,u} \\ (r,i \attMod \phi)$ holds.
From this and \thref{theorem:ibsec:correctness:secure:3}, it follows that also  $\mathit{secure}_{P,u}(r,i  \attMod \phi)$ holds.

We now prove our claim that  $\mathit{secure}^{\mathit{data}}_{P,u}(r,i \attMod \phi)$ holds.
Let $s'$ be the state just after the execution of the \texttt{SELECT} statement associated with $t$'s \texttt{WHEN} clause and $s''$ be the state $\mathit{last}(r^{i-1})$.
Furthermore, for brevity's sake, in the following we omit the $\mathit{pState}$ function where needed.
For instance, with a slight abuse of notation, we write $\llbracket s' \rrbracket^{\mathit{data}}_{u,M}$ instead of $\llbracket \mathit{pState}(s') \rrbracket^{\mathit{data}}_{u,M}$.
From $\mathit{secure}(u, \phi,s') = \top$, $s' \cong_{u,M}^{\mathit{data}} s''$, \thref{theorem:secure:equivalent:modulo:indistinguishable:state}, and \thref{theorem:secure:sound:under:approximation}, it follows that $\mathit{secure}^{\mathit{data}}_{P,u}(r,i-1 \attMod \phi)$ holds.
From this, it follows that $[\phi]^{v} = [\phi]^{s''}$ for any $v \in \llbracket s'' \rrbracket^{\mathit{data}}_{u,M}$.
Furthermore, from \thref{theorem:getInfo:sound:and:complete} and $\mathit{Ex}(s) = \emptyset$, it follows that $\phi$ holds in $s''$.
Let $A_{s'',R,\overline{t}}$ be the set $\{\langle \mathit{db}[R \oplus \overline{t}], U, \mathit{sec}, T,V\rangle \in \Pi_{M} \,|\, \exists \mathit{db}' \in \Omega_{D}.\, \langle \mathit{db}', \\  U,  \mathit{sec},  T, V\rangle \in \llbracket s''\rrbracket^{\mathit{data}}_{u,M}\}$.
It is easy to see that $\llbracket s\rrbracket^{\mathit{data}}_{u,M} \subseteq A_{s'',R,\overline{t}}$.
We now show that $\phi$ holds for any $z \in A_{s'',R,\overline{t}}$.
Let $z_{1} \in \llbracket s'' \rrbracket^{\mathit{data}}_{u,M}$.
From $[\phi]^{v} = [\phi]^{s''}$ for any $v \in \llbracket s'' \rrbracket^{\mathit{data}}_{u,M}$ and the fact that $\phi$ holds in $s''$, it follows that $[\phi]^{z_{1}} = \top$.
Therefore, for any $(\overline{k}_{1},\overline{k}_{2}, \overline{k}_{3}) \in R(z_{1})$, if $k_{1} = \overline{v}$, then $k_{2} = \overline{w}$. 
Then, for any $(\overline{k}_{1},\overline{k}_{2}, \overline{k}_{3}) \in R(z_{1}) \cup \{(\overline{v}, \overline{w}, \overline{q})\}$, if $k_{1} = \overline{v}$, then $k_{2} = \overline{w}$.
Therefore, $\phi$ holds also in $z_{1}[R \oplus \overline{t}] \in A_{\mathit{pState}(s''), R, \overline{t}}$.	
Hence, $[\phi]^{z} = \top$ for any $z \in A_{s'',R,\overline{t}}$.
From this and $\llbracket s\rrbracket^{\mathit{data}}_{u,M} \subseteq A_{s'',R,\overline{t}}$, it follows that  $[\phi]^{z} = \top$ for any $z \in \llbracket s\rrbracket^{\mathit{data}}_{u,M}$.
From this, it follows that $\mathit{secure}^{\mathit{data}}_{P,u}(r,i \attMod \phi)$ holds.

\item \emph{Learn \texttt{INSERT} - FD - 1}.
The proof of this case is similar to that of \emph{Learn \texttt{INSERT} - FD}.

\item \emph{Learn \texttt{INSERT} - ID}.
The proof of this case is similar to that of \emph{Learn \texttt{INSERT} - FD}.
See also the proof of \emph{\texttt{INSERT} Success - ID}.

\item \emph{Learn \texttt{INSERT} - ID - 1}.
The proof of this case is similar to that of \emph{Learn \texttt{INSERT} - ID}.

\item \emph{Learn \texttt{INSERT} Backward - 1}.
Let $i$ be such that $r^{i} = r^{i-1} \concat t \concat s$, where $s = \langle \mathit{db}', U', \mathit{sec}', T', V', c' \rangle \in \Omega_{M}$, $\mathit{last}(r^{i-1}) =  \langle \mathit{db}, U, \mathit{sec}, T, V, c \rangle$, and $t \in {\cal TRIGGER}_{D}$, and $\phi$ be $t$'s actual \texttt{WHEN} condition, where all free variables are replaced with the values in $\mathit{tpl}(\mathit{last}(r^{i-1}))$.
From the rule's definition, it follows that $\mathit{secEx}(s) = \top$.
From this, the LTS semantics, and $\mathit{secEx}(s) = \top$, it follows that $f^{u}_{\mathit{conf}}(\mathit{last}(r^{i-1}), \langle u', \mathtt{SELECT}, \phi \rangle) = \top$.
From this and $f^{u}_{\mathit{conf}}$'s definition, it follows that $\mathit{secure}(u,\phi,\mathit{last}(r^{i-1})) = \top$.
From this and \thref{theorem:secure:sound:under:approximation}, it follows that also $\mathit{secure}_{P,u}(r,i -1 \attMod \phi)$ holds.

\item \emph{Learn \texttt{INSERT} Backward - 2}.
Let $i$ be such that $r^{i} = r^{i-1} \concat t \concat s$, where $s = \langle \mathit{db}', U', \mathit{sec}', T', V', c' \rangle \in \Omega_{M}$, $\mathit{last}(r^{i-1}) =  \langle \mathit{db}, U, \mathit{sec}, T, V, c \rangle$, and $t \in {\cal TRIGGER}_{D}$, and $\phi$ be $\neg R(\overline{t})$.
Furthermore, let $\mathit{act}=\langle u', \mathtt{INSERT}, R, \\ \overline{t} \rangle$ be $t$'s actual action and $\gamma$ be $t$'s actual \texttt{WHEN} condition obtained by replacing all free variables with the values in $\mathit{tpl}(\mathit{last}(r^{i-1}))$.
From the rule's definition, it follows that $\mathit{secEx}(s) = \top$ and there is a $\psi$ such that $r, i-1 \attMod \psi$ and $r, i \attMod \neg \psi$.
We claim that $[\gamma]^{\mathit{db}} = \top$.
From this and  $\mathit{secEx}(s) = \top$, it follows $f_{\mathit{conf}}^{u}(s', \langle u', \mathtt{INSERT}, R, \overline{t}\rangle) = \top$, where $s'$ is the state obtained after the evaluation of $t$'s \texttt{WHEN} condition.
From this and $f^{u}_{\mathit{conf}}$'s definition, it follows $\mathit{secure}(u,\phi,s') = \top$ since $\phi$ is equivalent to $\mathit{getInfo}(\langle u', \mathtt{INSERT}, R, \overline{t}\rangle)$.
From this, $\mathit{pState}(\mathit{last}(r^{i-1})) = \mathit{pState}(s')$, and \thref{theorem:secure:equivalent:modulo:indistinguishable:state},  it follows that  $\mathit{secure}(u, \phi,\mathit{last}(r^{i-1})) = \top$.
From this and \thref{theorem:secure:sound:under:approximation}, it follows $\mathit{secure}_{P,u}(r,i - 1 \attMod \phi)$.

We now prove our claim that $[\gamma]^{\mathit{db}} = \top$.
Assume, for contradiction's sake, that this is not the case.
From this and the LTS rules, it follows that $\mathit{db} = \mathit{db}'$.
From the rule's definition, it follows that there is a $\psi$ such that $r, i-1 \attMod \psi$ and $r, i \attMod \neg \psi$.
From this, \thref{theorem:attacker:model:sound}, $s = \langle \mathit{db}', U', \mathit{sec}', T', V', c' \rangle$, and $\mathit{last}(r^{i-1}) =  \langle \mathit{db}, U, \mathit{sec}, T, \\ V, c \rangle$, it follows that $[\psi]^{\mathit{db}} = \top$ and $[\neg \psi]^{\mathit{db}'} = \top$.
Therefore, $[\psi]^{\mathit{db}} = \top$ and $[\psi]^{\mathit{db}'} = \bot$.
Hence, $\mathit{db} \neq \mathit{db}'$, which contradicts $\mathit{db} = \mathit{db}'$.

\item \emph{Learn \texttt{DELETE} Forward}.
The proof of this case is similar to that of \emph{Learn \texttt{INSERT} Forward}.

\item \emph{Learn \texttt{DELETE} - ID}.
The proof of this case is similar to that of \emph{Learn \texttt{INSERT} - FD}.
See also the proof of \emph{\texttt{DELETE} Success - ID}.

\item \emph{Learn \texttt{DELETE} - ID - 1}.
The proof of this case is similar to that of \emph{Learn \texttt{DELETE} - ID}.

\item \emph{Learn \texttt{DELETE} Backward - 1}.
The proof of this case is similar to that of \emph{Learn \texttt{INSERT} Backward - 1}.

\item \emph{Learn \texttt{DELETE} Backward - 2}.
The proof of this case is similar to that of \emph{Learn \texttt{INSERT} Backward - 2}.

\item \emph{Propagate Forward Trigger Action}.
Let $i$ be such that $r^{i} = r^{i-1} \concat t \concat s$, where $t$ is a trigger, $s = \langle \mathit{db}, U, \mathit{sec}, T, V, c \rangle \\ \in \Omega_{M}$ and $\mathit{last}(r^{i-1}) =  \langle \mathit{db}', U', \mathit{sec}',  T', V', c' \rangle$.
From the rule's definition, $r, i-1 \attMod \phi$ holds.
From this and the induction hypothesis, it follows that $\mathit{secure}_{P,u}  (r,i-1 \attMod \phi)$ holds.
From \thref{theorem:f:conf:pec:2}, the trigger $t$ preserves the equivalence class with respect to $r^{i-1}$, $P$, and $u$. 	
We claim that the execution of $t$ does not alter the content of the tables in $\mathit{tables}(\phi)$.
From this, \thref{theorem:secure:extend:on:runs:insert:delete}, and $\mathit{secure}_{P,u}(r,i-1 \attMod \phi)$, it follows that also the judgment $r,i \attMod \phi$ is secure, i.e., $\mathit{secure}_{P,u} \\ (r,i \attMod \phi)$ holds.

We now prove our claim that the execution of $t$ does not alter the content of the tables in $\mathit{tables}(\phi)$.
If the trigger is not enabled, proving the claim is trivial.
In the following, we assume the trigger is enabled.
There are four cases:
\begin{compactitem}
\item $t$'s action is an \texttt{INSERT} statement.
This case amount to claiming that the \texttt{INSERT} statement $\langle u',\mathtt{INSERT}, \\ R,  \overline{t} \rangle$ does not alter the content of the tables in $\mathit{tables}(\phi)$ in case $\mathit{reviseBelif}(r^{i-1}, \phi, r^{i}) = \top$.
We proved the claim above in the \emph{Propagate Forward \texttt{INSERT/DELETE} Success} case.

\item $t$'s action is an \texttt{DELETE} statement.
The proof is similar to that of the \texttt{INSERT} case.

\item $t$'s action is an \texttt{GRANT} statement.
In this case, the action does not alter the database state and the claim follows trivially.

\item $t$'s action is an \texttt{REVOKE} statement.
The proof is similar to that of the \texttt{GRANT} case.

\end{compactitem}

\item \emph{Propagate Backward Trigger Action}.
The proof of this case is similar to \emph{Propagate Backward Trigger Action}.

\item \emph{Propagate Forward \texttt{INSERT} Trigger Action}.
Let $i$ be such that $r^{i} = r^{i-1} \concat t \concat s$, where $t$ is a trigger, $s = \langle \mathit{db}, U, \mathit{sec}, T, V, c \rangle  \in \Omega_{M}$ and $\mathit{last}(r^{i-1}) =  \langle \mathit{db}', U', \mathit{sec}', \\  T', V', c' \rangle$.
From the rule's definition, $r, i-1 \attMod \phi$ holds.
From this and the induction hypothesis, it follows that $\mathit{secure}_{P,u}  (r,i-1 \attMod \phi)$ holds.
From \thref{theorem:f:conf:pec:2}, the trigger $t$ preserves the equivalence class with respect to $r^{i-1}$, $P$, and $u$. 	
We claim that the execution of $t$ does not alter the content of the tables in $\mathit{tables}(\phi)$.
From this, \thref{theorem:secure:extend:on:runs:insert:delete}, and $\mathit{secure}_{P,u}  (r,i-1 \attMod \phi)$, it follows that also the judgment $r,i \attMod \phi$ is secure, i.e., $\mathit{secure}_{P,u}  (r,i \attMod \phi)$ holds.

We now prove our claim that the execution of $t$ does not alter the content of the tables in $\mathit{tables}(\phi)$.
If the trigger is not enabled, proving the claim is trivial.
In the following, we assume the trigger is enabled.
Then, $t$'s action is an \texttt{INSERT} statement.
This case amount to claiming that the \texttt{INSERT} statement $\langle u',\mathtt{INSERT}, R,  \overline{t} \rangle$ does not alter the content of the tables in $\mathit{tables}(\phi)$ in case $r,i-1 \attMod R(\overline{t})$ holds.
We proved the claim above in the \emph{Propagate Forward \texttt{INSERT} Success - 1} case.

\item \emph{Propagate Forward \texttt{DELETE} Trigger Action}.
The proof of this case is similar to that of  \emph{Propagate Forward \texttt{INSERT} Trigger Action}.

\item \emph{Propagate Backward \texttt{INSERT} Trigger Action}.
The proof of this case is similar to that of  \emph{Propagate Forward \texttt{INSERT} Trigger Action}.

\item \emph{Propagate Backward \texttt{DELETE} Trigger Action}.
The proof of this case is similar to that of  \emph{Propagate Forward \texttt{INSERT} Trigger Action}.

\item \emph{Trigger FD \texttt{INSERT} Disabled Backward}.
Let $i$ be such that $r^{i} = r^{i-1} \concat t \concat s$, where $s = \langle \mathit{db}', U', \mathit{sec}', T', V', c' \rangle \in \Omega_{M}$, $t \in {\cal TRIGGER}_{D}$,  $\mathit{last}(r^{i-1}) =  \langle \mathit{db}, U, \mathit{sec}, T, V, c \rangle$, and $\psi$ be $t$'s actual \texttt{WHEN} condition obtained by replacing all free variables with the values in $\mathit{tpl}(\mathit{last}(r^{i-1}))$.
Furthermore, let $\mathit{act}=\langle u', \mathtt{INSERT}, R, (\overline{v}, \overline{w}, \overline{q}) \rangle$ be $t$'s actual action and $\alpha$ be $\exists \overline{y},\overline{z}. R(\overline{v},  \overline{y}, \overline{z}) \wedge \overline{y} \neq \overline{w}$.
From the rule's definition, it follows that  $\mathit{secEx}(s) = \bot$.
From this, it follows that $f^{u}_{\mathit{conf}}(\mathit{last}(r^{i-1}), \langle u', \mathtt{SELECT}, \psi \rangle)  = \top$.
From this and $f^{u}_{\mathit{conf}}$'s definition, it follows that $\mathit{secure}(u, \neg \psi, \mathit{last}(r^{i-1})) = \top$.
From this, it follows  $\mathit{secure}(u,  \psi, \mathit{last}(r^{i-1})) = \top$.
From this and \thref{theorem:secure:sound:under:approximation}, it follows $\mathit{secure}_{P,u}  (r,i -1 \attMod \psi)$.

\item \emph{Trigger ID \texttt{INSERT} Disabled Backward}.
The proof of this case is similar to that of \emph{Trigger FD \texttt{INSERT} Disabled Backward}.

\item \emph{Trigger ID \texttt{DELETE} Disabled Backward}.
The proof of this case is similar to that of \emph{Trigger FD \texttt{INSERT} Disabled Backward}.
\end{compactenum}
This completes the proof of the induction step.

This completes the proof.
\end{proof}

\subsection{Complexity proofs}\label{sect:data:conf:complexity:proofs}
In this section, we prove that data complexity of $f_{\mathit{conf}}^{u}$ is \complexity{}.
Note that the complexity class \complexity{} identifies those problems that can be solved using constant-depth, polynomial-size boolean circuits with AND, OR, and NOT gates with unbounded fan-in.
Note also that, in the following, with \complexity{} we usually refer to \emph{uniform}-\complexity{}.
Given a database schema $D$ and a database state $\mathit{db} \in \Omega_{D}^{\Gamma}$, the \emph{size of $\mathit{db}$}, denoted also as $|\mathit{db}|$, is $|\mathit{db}| = \Sigma_{R \in D} \Sigma_{\overline{t} \in \mathit{db}(R)} |\overline{t}|$, where the size $|\overline{t}|$ of a tuple $\overline{t}$ is just its cardinality.
Similarly, the \emph{the size of the schema $D$}, denoted $|D|$, is $\Sigma_{R \in D} |R|$.
Finally, given a set of views $V$ over $D$, the \emph{size of the extended vocabulary $\mathit{extVocabulary}(D,V)$}, denoted $|\mathit{extVoc}(D,V)|$, is $\Sigma_{o \in R \cup V} \Sigma_{0 \leq i < |o|} \dfrac{|o|!}{(|o|-i)! \concat i!}$.
Note that, given a view $V$, we denote by $|V|$ its cardinality.
Furthermore, given a $\mathit{RC}$-formula $\phi$,  the \emph{size of  $\phi$}, denoted as $|\phi|$,  is defined as follows:
\[
|\phi| = \left\{
\begin{array}{l l}
1+|\overline{x}| & \text{if}\; \phi := R(\overline{x})\\
1 & \text{if}\; \phi := \top\\
1 & \text{if}\; \phi := \bot\\
3 & \text{if}\; \phi := x = y\\
1+|\psi|+|\gamma| & \text{if}\; \phi := \psi\;O\;\gamma \;\text{and}\; O \in \{\vee,\wedge\}\\
1+|\psi| & \text{if}\; \phi := \neg \psi\\
2 + |\psi| & \text{if}\; \phi := Q\,x.\,\psi \;\text{and}\; Q \in \{\exists,\forall\}
\end{array}\right.
\]

\thref{theorem:rewritten:formula:linear:size} shows that the rewritten formula $\phi^{v}_{s,u}$, for some $v \in \{\top,\bot\}$, is linear in the size of the original formula $\phi$.

\begin{lemma}\thlabel{theorem:rewritten:formula:linear:size:inductive}
Let $M = \langle D, \Gamma\rangle$ be a system configuration, $s =\langle \mathit{db}, U, \mathit{sec}, T, V \rangle$ be a partial $M$-state, $u \in U$ be a user,  and $\phi$ be a $D$-formula.
For all formulae $\phi$ and all $v \in \{\top,\bot\}$, $|\phi^{v}_{s,u}| \leq (|\mathit{extVoc}(D,V)|+1) \concat |\phi|$.
\end{lemma}

\begin{proof}
Let $M = \langle D, \Gamma\rangle$ be a system configuration, $s =\langle \mathit{db}, U, \mathit{sec}, T, V \rangle$ be a partial $M$-state, and $u \in U$ be a user.
Let $\phi$ be an arbitrary formula over $D \cup V$ and $v$ be an arbitrary value in $\{\top,\bot\}$.
We now prove that $|\phi^{v}_{s,u}| \leq m \concat |\phi|$ by induction over the structure of the formula $\phi$.

\smallskip
\noindent
{\bf Base Case} 
There are four cases:
\begin{compactenum}
\item $\phi := x = y$. 
In this case, $\phi_{s,u}^{v}  = \phi$.
From this, $|\phi_{s,u}^{v}|  = |\phi|$.
From this, it follows trivially that $|\phi_{s,u}^{v}|  \leq (|\mathit{extVoc}(D,V)|+1) \concat |\phi|$.

\item $\phi := \top$. The proof of this case is similar to that of $\phi := x = y$.

\item $\phi := \bot$. The proof of this case is similar to that of $\phi := x = y$.

\item $\phi := R(\overline{x})$.
Without loss of generality, we assume that $v = \top$.
From this, it follows that $\phi_{s,u}^{\top} := \bigvee_{S \in R^{\top}_{s,u}} S(\overline{x})$.
From this, it follows that $|\phi_{s,u}^{\top}| = (|R^{\top}_{s,u}| -1)+\Sigma_{ S \in R^{\top}_{s,u}} \\ |S(\overline{x})|$.
From this and $|S(\overline{x})| = 1 + |\overline{x}|$, it follows that  $|\phi_{s,u}^{\top}| = (|R^{\top}_{s,u}| -1)+\Sigma_{ S \in R^{\top}_{s,u}} (1 + |\overline{x}|)$.
From this, it follows that $|\phi_{s,u}^{\top}| = (|R^{\top}_{s,u}| -1) + |R^{\top}_{s,u}| \concat( 1 + |\overline{x}|)$.
From   $\phi := R(\overline{x})$, it follows that $|\phi| = 1 + |\overline{x}|$.
From this and $|\phi_{s,u}^{\top}| = (|R^{\top}_{s,u}| -1) + |R^{\top}_{s,u}| \concat( 1 + |\overline{x}|)$, it follows that $|\phi_{s,u}^{\top}| = |R^{\top}_{s,u}| \concat |\phi| + (|R^{\top}_{s,u}| -1)$.
We claim that $|R^{\top}_{s,u}| \leq |\mathit{extVoc}(D,V)|$.
From this and  $|\phi_{s,u}^{\top}| = |R^{\top}_{s,u}| \concat |\phi| + (|R^{\top}_{s,u}| -1)$, it follows that $|\phi_{s,u}^{\top}| \leq |\mathit{extVoc}(D,V)| \concat |\phi| + |\mathit{extVoc}(D,V)|$.
From this, it follows that $|\phi_{s,u}^{\top}| \leq (|\mathit{extVoc}(D,V)|+1) \concat |\phi|$.

We now prove our claim that $|R^{\top}_{s,u}| \leq |\mathit{extVoc}(D,V)|$.
The set $R^{\top}_{s,u}$ is a subset of $\mathit{extVocabulary}(D,V)$ by construction.
The set $\mathit{extVocabulary}(D,V)$ contains any possible projection of tables in $D$ and views in $V$.
It is easy to check that the cardinality of $\mathit{extVocabulary}(D,V)$ is, indeed, $|\mathit{extVoc}(D,V)|$.
\end{compactenum}
This completes the proof of the base case.

\smallskip
\noindent
{\bf Induction Step} 
Assume that our claim holds for all sub-formulae of $\phi$.
We now show that our claim holds also for $\phi$.
There are a number of cases depending on $\phi$'s structure.
\begin{compactenum}
\item $\phi := \psi \wedge \gamma$.
From this, it follows that $\phi_{s,u}^{v} := \psi_{s,u}^{v} \wedge \gamma_{s,u}^{v}$.
From this, it follows that $|\phi_{s,u}^{v}| = 1+|\psi_{s,u}^{v}|+|\gamma_{s,u}^{v}|$.
From the induction hypothesis, it follows that $|\psi_{s,u}^{v}| \leq (|\mathit{extVoc}(D,V)|+1) \concat |\psi|$ and $|\gamma_{s,u}^{v}| \leq (|\mathit{extVoc}(D,V)|+1) \concat |\gamma|$.
From this and $|\phi_{s,u}^{v}| = 1+|\psi_{s,u}^{v}|+|\gamma_{s,u}^{v}|$, it follows that $|\phi_{s,u}^{v}| \leq 1+ (|\mathit{extVoc}(D,V)|+1) \concat |\psi| + (|\mathit{extVoc}(D,V)|+1) \concat |\gamma|$.
From this and $|\mathit{extVoc}(D,V)| \\ \geq 0$, it follows that $|\phi_{s,u}^{v}| \leq |\mathit{extVoc}(D,V)|+1+(|\mathit{extVoc} \\ (D,V)|+1) \concat |\psi| + (|\mathit{extVoc}(D,V)|+1) \concat |\gamma|$.
From this, it follows that $|\phi_{s,u}^{v}| \leq (|\mathit{extVoc}(D,V)|+1)\concat (1 + |\psi| +  |\gamma|)$.
From this and $|\phi| = 1 + |\psi| + |\gamma|$, it follows that $|\phi_{s,u}^{v}| \leq (|\mathit{extVoc}(D,V)|+1)\concat |\phi|$.

\item $\phi := \psi \vee \gamma$.
The proof of this case is similar to that of $\phi := \psi \wedge \gamma$.

\item $\phi := \neg \psi$.
From this, it follows that $\phi_{s,u}^{v} := \neg \psi_{s,u}^{\neg v}$.
From this, it follows that $|\phi_{s,u}^{v}| = 1+|\psi_{s,u}^{\neg v}|$.
From the induction hypothesis, it follows that $|\psi_{s,u}^{\neg v}| \leq (|\mathit{extVoc} \\ (D,V)|+1) \concat |\psi|$.
From this and $|\phi_{s,u}^{v}| = 1+|\psi_{s,u}^{v}|$, it follows that $|\phi_{s,u}^{v}| \leq 1+ (|\mathit{extVoc}(D,V)|+1) \concat |\psi|$.
From this and $|\mathit{extVoc}(D,V)| \geq 0$, it follows that $|\phi_{s,u}^{v}| \leq |\mathit{extVoc}(D,V)|+1+(|\mathit{extVoc}(D,V)|+1) \concat |\psi|$.
From this, it follows that $|\phi_{s,u}^{v}| \leq (|\mathit{extVoc}(D,V)|+1)\concat (1 + |\psi|)$.
From this and $|\phi| = 1 + |\psi|$, it follows that $|\phi_{s,u}^{v}| \leq (|\mathit{extVoc}(D,V)|+1)\concat |\phi|$.

\item $\phi := \exists x.\, \psi$.
If $\phi^{v}_{s,u}$ is $\neg v$, then the claim holds trivially since $|\phi^{v}_{s,u}| = 1$.
In the following, we assume that $\phi_{s,u}^{v} := \exists x.\, \psi_{s,u}^{v}$.
From this, it follows that $|\phi_{s,u}^{v}| = 2+|\psi_{s,u}^{v}|$.
From the induction hypothesis, it follows that $|\psi_{s,u}^{v}| \leq (|\mathit{extVoc}(D,V)|+1) \concat |\psi|$.
From this and $|\phi_{s,u}^{v}| = 2+|\psi_{s,u}^{v}|$, it follows that $|\phi_{s,u}^{v}| \leq 2+(|\mathit{extVoc}(D,V)|+1) \concat |\psi|$.
From this and $|\mathit{extVoc}(D,V)| \\ \geq 0$, it follows that $|\phi_{s,u}^{v}| \leq 2\concat |\mathit{extVoc}(D,V)|+2+(|\mathit{extVoc}(D,V)|+1) \concat |\psi|$.
From this, it follows that $|\phi_{s,u}^{v}| \leq (|\mathit{extVoc}(D,V)|+1)\concat (2 + |\psi|)$.
From this and $|\phi| = 2 + |\psi|$, it follows that $|\phi_{s,u}^{v}| \leq (|\mathit{extVoc}(D,V)|+1)\concat |\phi|$.

\item $\phi := \forall x.\, \psi$.
The proof of this case is similar to that of $\phi:= \exists x.\, \psi$.
\end{compactenum}
This completes the proof of the induction step.

This completes the proof of our claim.
\end{proof}

\begin{lemma}\thlabel{theorem:rewritten:formula:linear:size}
Let $M = \langle D, \Gamma\rangle$ be a system configuration, $s =\langle \mathit{db}, U, \mathit{sec}, T, V \rangle$ be a partial $M$-state, $u \in U$ be a user,  and $\phi$ be a $D$-formula.
For all sentences $\phi$ and all $v \in \{\top,\bot\}$, $|\phi^{v}_{s,u}| \leq (|\mathit{extVoc}(D,V)|+1) \concat |\phi|$ and $|\neg \phi^{\top}_{s,u} \wedge \phi^{\bot}_{s,u}| \leq 2(|\mathit{extVoc}(D,V)|+1) \concat |\phi|$.
\end{lemma}

\begin{proof}
Let $M = \langle D, \Gamma\rangle$ be a system configuration, $s =\langle \mathit{db}, U, \mathit{sec}, T, V \rangle$ be a partial $M$-state, $u \in U$ be a user,  and $\phi$ be a $D$-formula.
Furthermore, let $\phi$ be a sentence and $v$ be a value in $\{\top,\bot\}$.
The fact that $|\phi^{v}_{s,u}| \leq (|\mathit{extVoc}(D,V)|+1) \concat |\phi|$ follows trivially from \thref{theorem:rewritten:formula:linear:size:inductive}.
Let $\psi$ be the formula $\neg \phi^{\top}_{s,u} \wedge \phi^{\bot}_{s,u}$.
The size of $\psi$ is $2+|\phi^{\top}_{s,u}| + |\phi^{\bot}_{s,u}|$.
From this and \thref{theorem:rewritten:formula:linear:size:inductive}, it follows that $|\psi| \leq 2+(|\mathit{extVoc}(D,V)|+1) \concat |\phi|+(|\mathit{extVoc}(D,V)|+1) \concat |\phi|$.
From this, it follows that $|\psi| \leq 2(|\mathit{extVoc}(D,V)|+1) \concat |\phi|$.
This completes the proof.
\end{proof}

In the following, we study the data complexity of our \acf{}.
Note that, given a \acf{} $f$, the \emph{data complexity of $f$} is the data complexity of the following decision problem:
\begin{definition}
Let $M = \langle D, \Gamma\rangle$ be some fixed system configuration, $a \in {\cal A}_{D,U}$ be some fixed  action, $u \in {\cal U}$ be some fixed  user, $U \subseteq {\cal U}$ be some fixed  set of users, $\mathit{sec} \in \Omega^{\mathit{sec}}_{U,D}$ be some fixed  policy, $T$ be some fixed  set of triggers over $D$ whose owners are in $U$, $V$ be some fixed  set of views over $D$ whose owners are in $U$, and $c$ be some fixed  context.

\noindent
\textbf{INPUT: } A database state $\mathit{db}$ such that $\langle \mathit{db},U,\mathit{sec},T,V,c\rangle \in \Omega_{M}$.

\noindent
\textbf{Question: } Is $f(\langle \mathit{db},U,\mathit{sec},T,V,c\rangle,a) = \top$?
\end{definition}

We define in a similar way the data complexity of the \emph{secure} procedure.

\begin{theorem}\thlabel{theorem:confidentiality:complexity}
The data complexity of $f_{\mathit{conf}}^{u}$ is \complexity{}. 
\end{theorem}

\begin{proof}
Let $M = \langle D, \Gamma\rangle$ be some fixed system configuration, $a \in {\cal A}_{D,U}$ be some fixed  action, $u \in {\cal U}$ be some fixed  user, $U \subseteq {\cal U}$ be some fixed  set of users, $\mathit{sec} \in \Omega^{\mathit{sec}}_{U,D}$ be some fixed  policy, $T$ be some fixed  set of triggers over $D$ whose owners are in $U$, $V$ be some fixed  set of views over $D$ whose owners are in $U$, and $c$ be some fixed  context.
The data complexity of $f_{\mathit{conf}}^{u}$ is the maximum of the  data complexities of $f^{u}_{\mathit{conf},\texttt{I},\texttt{D}}$, $f^{u}_{\mathit{conf},\texttt{G}}$, and $f^{u}_{\mathit{conf},\texttt{S}}$.
We claim that:
\begin{compactenum}
\item the data complexity of $f^{u}_{\mathit{conf},\texttt{I},\texttt{D}}$ is \complexity{},
\item the data complexity of $f^{u}_{\mathit{conf},\texttt{S}}$ is \complexity{}, and
\item the data complexity of $f^{u}_{\mathit{conf},\texttt{G}}$ is $O(1)$.
\end{compactenum}
From this, it follows that the data complexity of $f_{\mathit{conf}}^{u}$ is $\mathit{max}(AC^{0}, O(1))$.
From this, it follows that the data complexity of $f_{\mathit{conf}}^{u}$ is \complexity{}.

Our claims on the data complexity of $f^{u}_{\mathit{conf},\texttt{I},\texttt{D}}$, $f^{u}_{\mathit{conf},\texttt{S}}$, and $f^{u}_{\mathit{conf},\texttt{G}}$ are proved respectively in \thref{theorem:complexity:f:conf:insert:delete}, \thref{theorem:complexity:f:conf:select}, and \thref{theorem:complexity:f:conf:grant}.
\end{proof}

\begin{lemma}\thlabel{theorem:complexity:f:conf:insert:delete}
The data complexity of $f_{\mathit{conf}, \texttt{I},\texttt{D}}^{u}$ is \complexity{}. 
\end{lemma}

\begin{proof}
Let $M = \langle D, \Gamma\rangle$ be some fixed system configuration, $a \in {\cal A}_{D,U}$ be some fixed \texttt{INSERT} or \texttt{DELETE} action, $u \in {\cal U}$ be some fixed  user, $U \subseteq {\cal U}$ be some fixed  set of users, $\mathit{sec} \in \Omega^{\mathit{sec}}_{U,D}$ be some fixed  policy, $T$ be some fixed  set of triggers over $D$ whose owners are in $U$, $V$ be some fixed  set of views over $D$ whose owners are in $U$, and $c$ be some fixed  context.
Furthermore, let $\mathit{db} \in \Omega_{D}^\Gamma$ be a database state such that $\langle \mathit{db},U,\mathit{sec},T,V,c\rangle \in \Omega_{M}$.
We can check whether $f_{\mathit{conf}, \texttt{I},\texttt{D}}^{u}(\langle \mathit{db}, U,\mathit{sec},T,V,c\rangle,a) = \top$ as follows:
\begin{compactenum}
\item If $\mathit{trigger}(s) = \epsilon$ and $a \not\in {\cal A}_{D,u}$, return $\top$.
\item If $\mathit{trigger}(s) \neq \epsilon$ and $\mathit{invoker}(s) \neq u$, return $\top$.
\item Compute the result of $\mathit{noLeak}(s,a,u)$. If $\mathit{noLeak}(s,a,u) \\ = \bot$, then returns $\bot$.
\item Compute the set $\mathit{Dep}(\Gamma,a)$.
\item Compute $\mathit{secure}(u,\mathit{getInfo}(a),s)$. If its result is $\bot$, return $\bot$.
\item For each $\gamma \in \mathit{Dep}(\Gamma,a)$, compute $\mathit{secure}(u,\mathit{getInfoV}(a, \\ \gamma),s)$. If its result is $\bot$, return $\bot$.
\item For each $\gamma \in \mathit{Dep}(\Gamma,a)$, compute $\mathit{secure}(u,\mathit{getInfoS}(a, \\ \gamma),s)$. If its result is $\bot$, return $\bot$.
\item Return $\top$.
\end{compactenum}
The data complexity of the steps 1 and 2 is $O(1)$. 
We claim that also the data complexity of the third step is $O(1)$.
The  complexity of the fourth step is $O(|\Gamma|)$.
From the definition of $\mathit{getInfo}$, the resulting formula is constant in the size of the database.
Furthermore, also constructing the formula can be done in constant time in the size of the database.
From this and \thref{theorem:secure:complexity}, it follows that the data complexity of the fifth step is \complexity{}.
For a similar reason, the data complexity of the sixth and seventh steps is also \complexity{}.
Therefore, the overall data complexity of the $f_{\mathit{conf},\texttt{I},\texttt{D}}^{u}$ procedure is \complexity{}.

We now prove our claim that the data complexity of the $\mathit{noLeak}$ procedure is $O(1)$.
An algorithm implementing the $\mathit{noLeak}$ procedure is as follows:
\begin{compactenum}
\item for each view $v \in V$, for each grant $g \in \mathit{sec}$, if $g = \langle \mathit{op}, u, \langle \mathtt{SELECT}, v \rangle,u' \rangle$, then
\begin{compactenum}
\item compute the set $\mathit{tDet}(v,s,M)$.
\item if $R \in \mathit{tDet}(v,s,M)$, for each $o \in \mathit{tDet}(v,s,M)$, check whether  $\langle \mathit{op}, u, \langle \mathtt{SELECT}, o \rangle,u'' \rangle \in \mathit{sec}$.
\end{compactenum}
\end{compactenum}
The size of the set $\mathit{tDet}(v,s,M)$ is at most $|D|$.
From this, it follows that the complexity of the step 1.(b) is $O(|D| \concat |\mathit{sec}|)$.
From \thref{theorem:apprDet:complexity} and the definition of $\mathit{tDet}$, the complexity of computing $\mathit{tDet}(v,s,M)$ is $O(|\phi|^3)$, where $\phi$ is $v$'s definition.
The overall complexity is, therefore, $O(|V|\concat|\mathit{sec}|\concat(|D| \concat |\mathit{sec}| + 2^{|D|} \concat |\phi|))$, where $\phi$ is the definition of the longest view in $V$.
From this, it is easy to see that the data complexity of the $\mathit{noLeak}$ procedure is $O(1)$.
\end{proof}

\begin{lemma}\thlabel{theorem:complexity:f:conf:grant}
The data complexity of $f_{\mathit{conf}, \texttt{G}}^{u}$ is $O(1)$.
\end{lemma}

\begin{proof}
Let $M = \langle D, \Gamma\rangle$ be some fixed system configuration, $a \in {\cal A}_{D,U}$ be some fixed \texttt{GRANT} action, $u \in {\cal U}$ be some fixed  user, $U \subseteq {\cal U}$ be some fixed  set of users, $\mathit{sec} \in \Omega^{\mathit{sec}}_{U,D}$ be some fixed  policy, $T$ be some fixed  set of triggers over $D$ whose owners are in $U$, $V$ be some fixed  set of views over $D$ whose owners are in $U$, and $c$ be some fixed  context.
Furthermore, let $\mathit{db} \in \Omega_{D}^\Gamma$ be a database state such that $\langle \mathit{db},U,\mathit{sec},T,V,c\rangle \in \Omega_{M}$.
We can check whether $f_{\mathit{conf}, \texttt{G}}^{u}(\langle \mathit{db}, \\ U,\mathit{sec},T,V,c\rangle, \langle\mathit{op},u'',p,u' \rangle)) = \top$ as follows.
\begin{compactenum}
\item If $\mathit{trigger}(s) = \epsilon$ and $a \not\in {\cal A}_{D,u}$, return $\top$.
\item If $\mathit{trigger}(s) \neq \epsilon$ and $\mathit{invoker}(s) \neq u$, return $\top$.
\item If $p$ is not a \texttt{SELECT} privilege, return $\top$.
\item If $u'' \neq u$, return $\top$.
\item For each $g \in \mathit{sec}$, if $g = \langle \mathit{op}, u, p, u' \rangle$, return $\top$.
\item Return $\bot$. 
\end{compactenum}
The complexity of the fifth step is $O(|\mathit{sec}|)$, whereas the complexity of the other steps is $O(1)$. 
Therefore, the overall complexity of the $f_{\mathit{conf},\texttt{G}}^{u}$ procedure is $O(|\mathit{sec}|)$.
From this, it follows that the data complexity of $f_{\mathit{conf},\texttt{G}}^{u}$ procedure is $O(1)$.
\end{proof}

\begin{lemma}\thlabel{theorem:complexity:f:conf:select}
The data complexity of $f_{\mathit{conf}, \texttt{S}}^{u}$ is \complexity{}.
\end{lemma}

\begin{proof}
Let $M = \langle D, \Gamma\rangle$ be some fixed system configuration, $a \in {\cal A}_{D,U}$ be some fixed \texttt{SELECT} action, $u \in {\cal U}$ be some fixed  user, $U \subseteq {\cal U}$ be some fixed  set of users, $\mathit{sec} \in \Omega^{\mathit{sec}}_{U,D}$ be some fixed  policy, $T$ be some fixed  set of triggers over $D$ whose owners are in $U$, $V$ be some fixed  set of views over $D$ whose owners are in $U$, and $c$ be some fixed  context.
Furthermore, let $\mathit{db} \in \Omega_{D}^\Gamma$ be a database state such that $\langle \mathit{db},U,\mathit{sec},T,V,c\rangle \in \Omega_{M}$.
We can check whether $f_{\mathit{conf}, \texttt{S}}^{u}(\langle \mathit{db}, \\ U,\mathit{sec},T,V,c\rangle, a)) = \top$ as follows.
\begin{compactenum}
\item If $\mathit{trigger}(s) = \epsilon$ and $a \not\in {\cal A}_{D,u}$, return $\top$.
\item If $\mathit{trigger}(s) \neq \epsilon$ and $\mathit{invoker}(s) \neq u$, return $\top$.
\item Compute $\mathit{secure}(u,\phi,s)$ and return its result.
\end{compactenum}
The complexity of the first and second steps is $O(1)$.
From \thref{theorem:secure:complexity}, it follows that the data complexity of the third step is \complexity{}. 
From this, it follows that the data complexity of $f_{\mathit{conf},\texttt{S}}^{u}$ procedure is \complexity{}.
\end{proof}

\begin{lemma}\thlabel{theorem:secure:complexity}
The data complexity of $\mathit{secure}$ is \complexity{}.
\end{lemma}

\begin{proof}
Let $M = \langle D, \Gamma\rangle$ be some fixed system configuration, $\phi$ be some fixed sentence, $u \in {\cal U}$ be some fixed  user, $U \subseteq {\cal U}$ be some fixed  set of users, $\mathit{sec} \in \Omega^{\mathit{sec}}_{U,D}$ be some fixed  policy, $T$ be some fixed  set of triggers over $D$ whose owners are in $U$, $V$ be some fixed  set of views over $D$ whose owners are in $U$, and $c$ be some fixed  context.
Furthermore, let $\mathit{db} \in \Omega_{D}^\Gamma$ be a database state such that $\langle \mathit{db},U,\mathit{sec},T,V,c\rangle \in \Omega_{M}$.
We denote by $s$ the state $\langle \mathit{db},U,\mathit{sec},T,V,c\rangle$.
We can check whether $\mathit{secure}(u,\phi,\langle \mathit{db},U,\mathit{sec},T,V,c \rangle ) = \top$ as follows:
\begin{compactenum}
\item Compute the formula $\phi^{\mathit{rw}}_{s,u}$.
\item Compute $[\phi^{\mathit{rw}}_{s,u}]^{\mathit{db}}$. 
\item $\mathit{secure}(u,\phi,\langle \mathit{db},U,\mathit{sec},T,V,c \rangle ) = \top$ iff $[\phi^{\mathit{rw}}_{s,u}]^{\mathit{db}} = \bot$. 
\end{compactenum}
We claim that the first step can be done in constant time in terms of data complexity.
It is well-known that the data complexity of query execution is \complexity{} \cite{abiteboul1995foundations}.
From this, it follows that the data complexity of $\mathit{secure}$ is also \complexity{}.

We now prove our claim that computing the formula $\phi^{\mathit{rw}}_{s,u}$ can be done in constant time in terms of data complexity.
The extended vocabulary $\mathit{extVocabulary}(D,V)$ does not depend on the database state.
From this and the definition of $R^{v}_{s}$, where $R$ is a predicate symbol and $v \in \{\top, \bot\}$, the set $R^{v}_{s}$ (and the time needed to compute it) depends just on the database schema $D$ and the set of views $V$.
The set $\mathit{AUTH}_{s,u}$ and the time needed to compute it depend just on the size of the policy $\mathit{sec}$.
Furthermore, the time needed to compute $\mathit{AUTH}_{s,u}^{*}$ depends just on the size of the policy $\mathit{sec}$ and of the extended vocabulary.
Therefore, for any predicate $R$, the set $R^{v}_{s}$ can be computed in constant time in terms of database size.
The computation of the formula $\phi'$, obtained by replacing sub-formulae of the form $\exists \overline{x}. R(\overline{x},\overline{y})$ with the corresponding predicates in the extended vocabulary, can be done in linear time in terms of $|\phi|$ and in constant time in terms of $|\mathit{db}|$.
Note that the size of the resulting formula is linear in  $|\phi|$.
It is easy to see that also computing $\phi^{\top}_{s,u}$ and $\phi^{\bot}_{s,u}$ can be done in linear time in terms of $|\phi|$ and in constant time in terms of $|\mathit{db}|$.
As shown in \thref{theorem:rewritten:formula:linear:size},  the size of the resulting formula is linear in  $|\phi|$. 
Finally, we can replace the predicates in the extended vocabulary with the corresponding sub-formulae again in linear time in terms of $|\phi|$.
Note that, again, the size of the resulting formula is linear in  $|\phi|$.
Therefore, the overall rewriting process can be done in linear time in the size of $\phi$ and in constant time in the size of $\mathit{db}$.
\end{proof}

%% file: composition.tex
\clearpage
\section{Composition}\label{app:composition}

Here, we model the \acf{} $f$, presented in Section \ref{sect:enf:alg}, which is obtained by composing the \acf{}s  $f_{\mathit{int}}$ and $f_{\mathit{conf}}^{u}$ presented above.
The \acf{} $f$ is obtained by composing $f_{\mathit{int}}$ and $f_{\mathit{conf}}^{u}$ as follows:
\[
f(s,\mathit{act}) = f_{\mathit{int}}(s, \mathit{act}) \wedge f_{\mathit{conf}}^{\mathit{user}(\mathit{act},s)}(s,\mathit{act})
\]

The function $\mathit{user}$ takes as input an action and a state and returns the actual user executing the action.
It is defined as follows, where $i$ denotes the $\mathit{invoker}$ function and $\mathit{tr}$ denotes the $\mathit{trigger}$ function.
\[
\mathit{user}(\mathit{act},s) = \left\{ 
  \begin{array}{l l}
  \mathit{i}(s) & \text{if } \mathit{tr}(s) \neq \epsilon \\
   u & \text{if } \mathit{tr}(s) = \epsilon \;\text{and}\; \mathit{act} \in {\cal A}_{D,u}\\
  \end{array}\right.
\]

We now show our main results, namely that (1) $f$ provides both \correctness{} and \confidentiality{}, and (2) $f$'s data complexity is \complexity{}.

\begin{theorem}
Let $M$ be a system configuration, $f$ be as above, and $P=  \langle M, f \rangle$ be an \accessControlConfiguration{}.
\begin{compactenum}
\item For any user $u \in {\cal U}$, the \acf{} $f$ provides \confidentiality{} with respect to $\attMod$, $P$, and $u$.
\item The \acf{} $f$ provides \correctness{} with respect to $P$.
\end{compactenum}
\end{theorem}

\begin{proof}
It follows from \thref{theorem:composition:correctness} and \thref{theorem:composition:confidentiality}.
\end{proof}

\begin{theorem}\thlabel{theorem:composition:complexity}
The data complexity of $f$ is  \complexity{}. 
\end{theorem}

\begin{proof}
From $f$'s definition, it follows that $f$'s data complexity is the maximum complexity between $f_{\mathit{conf}}^{u}$'s complexity and $f_{\mathit{int}}$'s complexity.
From this,  \thref{theorem:integrity:complexity}, and \thref{theorem:confidentiality:complexity}, it follows that the data complexity of $f$ is \complexity{}.
\end{proof}

\subsection{\Correctness{}}

Here, we show that $f$ provides \correctness{}.

\begin{lemma}\thlabel{theorem:composition:correctness}
Let $M = \langle D, \Gamma \rangle$ be a system configuration, $f$ be as above, and $P=  \langle M, f \rangle$ be an \accessControlConfiguration{}.
The \acf{} $f$ provides \correctness{} with respect to $P$.
\end{lemma}

\begin{proof}
We prove the lemma by contradiction.
Assume, for contradiction's sake, that $f$ does not satisfy the \correctness{} property.
There are three cases:
\begin{compactitem}

\item there is a reachable state $s$ and an action $\mathit{act} \in {\cal A}_{D, {\cal U}}$ such that $\mathit{trigger}(s) = \epsilon$, $f(s,\mathit{act}) = \top$, and $s \not\auth \mathit{act}$. From $f(s,\mathit{act}) = \top$, it follows that 
$f_{\mathit{int}}(s,\mathit{act}) = \top$. From this fact, $\mathit{trigger}(s) = \epsilon$, and \thref{theorem:eopsec:enforcement:second:variant}, it follows $s \auth \mathit{act}$, which leads to a contradiction.

\item there is a reachable state $s$ and a trigger $t \in {\cal TRIGGER}_{D}$ such that $\mathit{trigger}(s) = t$, $f(s,c) = \top$, $[\psi]^{s.\mathit{db}} = \bot$, and $s \not\auth t$, where $c = \langle u, \mathtt{SELECT}, \psi \rangle$ is $t$'s condition. From $f(s,c) = \top$, it follows that $f_{\mathit{int}}(s,c) = \top$. 
From $f_{\mathit{int}}(s,c) = \top$, $[\psi]^{s.\mathit{db}} = \bot$, $\mathit{trigger}(s) = t$, and \thref{theorem:eopsec:enforcement:third:variant}, it follows $s \auth t$, which leads to a contradiction.

\item there is a reachable state $s$ and a trigger $t \in {\cal TRIGGER}_{D}$ such that $\mathit{trigger}(s) = t$, $f(s,c) = \top$, $[\psi]^{s.\mathit{db}} = \top$, $f(s',a ) = \top$, and $s \not\auth t$, where $c = \langle u, \mathtt{SELECT}, \psi \rangle$ is $t$'s condition, $a$ is $t$'s action, and $s'$ is the state obtained from $s$ by updating the context's history. From $f(s',a) = \top$, it follows that  $f_{\mathit{int}}(s',a) = \top$.
Since $s$ and $s'$ are equivalent modulo the context's history and $f_{\mathit{int}}$ does not depend on the context's history, it follows that $f_{\mathit{int}}(s,a) = \top$.
From $f_{\mathit{int}}(s,c) = \top$, $[\psi]^{s.\mathit{db}} = \top$, $f_{\mathit{int}}(s,a) = \top$, $\mathit{trigger}(s) = t$, and \thref{theorem:eopsec:enforcement:third:variant}, it follows $s \auth t$, which leads to a contradiction.

\end{compactitem}
This completes the proof.
\end{proof}

\begin{lemma}\thlabel{theorem:composition:policy:consistency}
Let $P = \langle M, f \rangle$ be an \accessControlConfiguration{}, where $M = \langle D, \Gamma \rangle$ is a system configuration and $f$ is as above, and $L$ be the $P$-LTS.
For each reachable state $s = \langle \mathit{db}, U, \mathit{sec}, T, V, c\rangle$, $s \auth g$ for all $g \in \mathit{sec}$.
\end{lemma}

\begin{proof}
The proof is very similar to that of \thref{theorem:integrity:policy:consistency}.
\end{proof}

\subsection{\Confidentiality{}}
Here, we show that $f$ provides the desired \confidentiality{} guarantees.
First, we show that the \acf{} $f'$, defined as $f'(s,\mathit{act}) := f_{\mathit{conf}}^{\mathit{user}(\mathit{act},s)}(s,\mathit{act})$, provides \confidentiality{}.
Afterwards, we analyse the security of $f$.

In \thref{theorem:f:conf:composition:soundness} and \thref{theorem:composition:f1:confidentiality}, we prove some preliminary results about $f'$.
These results will then be used to prove $f$'s security.

\begin{lemma}\thlabel{theorem:f:conf:composition:soundness}
Let $M = \langle D, \Gamma\rangle$ be a system configuration,  $a$ be an action in ${\cal A}_{D,{\cal U}}$, and $s,s' \in \Omega_{M}$ be two $M$-states such that $\mathit{pState}(s) \cong_{\mathit{user}(\mathit{a},s),M}^{\mathit{data}} \mathit{pState}(s')$, $\mathit{invoker}(s) = \mathit{invoker}(s')$, and $\mathit{trigger}(s) = \mathit{trigger}  (s')$.
Then, $f^{\mathit{user}(\mathit{a},s)}_{\mathit{conf}}(s, \\ a) = \top$ iff $f^{\mathit{user}(\mathit{a},s')}_{\mathit{conf}}(s', a) = \top$.
\end{lemma}

\begin{proof}
Let $s = \langle \mathit{db}, U, \mathit{sec}, T,V,c\rangle$ and $s' = \langle \mathit{db}', U', \mathit{sec}', \\  T',V',c'\rangle$ be two $M$-states such that $\mathit{pState}(s) \cong_{M,\mathit{user}(\mathit{a},s)}^{\mathit{data}} \mathit{pState}(s')$, $\mathit{invoker}(s) = \mathit{invoker}(s')$, and $\mathit{trigger}(s) = \mathit{trigger}  (s')$.

We first show that $\mathit{user}(\mathit{a},s) = \mathit{user}(\mathit{a},s')$.
Since $\mathit{trigger}(s)\\ = \mathit{trigger}(s')$, there are two cases:
\begin{compactitem}
\item $\mathit{trigger}(s) = \epsilon$.
In this case, the result of $\mathit{user}(\mathit{a},s)$ depends just on $\mathit{a}$. 
Therefore, $\mathit{user}(\mathit{a},s) = \mathit{user}(\mathit{a},s')$.

\item $\mathit{trigger}(s) \neq \epsilon$. 
In this case, $\mathit{user}(\mathit{a},s) = \mathit{invoker}(s)$ and $\mathit{user}(\mathit{a},s') = \mathit{invoker}(s')$. 
From $\mathit{invoker}(s) = \\ \mathit{invoker}(s')$, it follows that  $\mathit{user}(\mathit{a},s) = \mathit{user}(\mathit{a},s')$. 
\end{compactitem}

Let $u$ be the user $\mathit{user}(\mathit{a},s)$.
From \thref{theorem:f:conf:soundness}, it follows that $f^{u}_{\mathit{conf}}(s,  a) = f^{u}_{\mathit{conf}}(s', a)$.
This completes the proof.
\end{proof}

\begin{lemma}\thlabel{theorem:composition:f1:confidentiality}
Let $M$ be a system configuration, $f'$ be as above, and $P=  \langle M, f' \rangle$ be an \accessControlConfiguration{}.
For any user $u \in {\cal U}$, the \acf{} $f'$ satisfies the \confidentiality{} property with respect to $P$, $u$, $\attackerModel$, and $\cong_{P,u}$.
\end{lemma}

\begin{proof}
It is easy to see that Lemmas  \ref{theorem:f:conf:soundness}, \ref{theorem:secure:extend:on:runs:insert:delete}, \ref{theorem:secure:extend:on:runs:select:create}, \ref{theorem:secure:extend:on:runs:grant:revoke}, \ref{theorem:secure:extend:on:runs:triggers},  \ref{theorem:f:conf:pec:1}, and \ref{theorem:f:conf:pec:2} hold as well for $f'$.
Therefore, we can easily adapt the proof of  \thref{theorem:f:conf:confidentiality} to  $f'$.
\end{proof}
In \thref{theorem:f:composition:soundness}, we show that the \acf{} $f$ returns the same result in any two data-indistinguishable states.

\begin{lemma}\thlabel{theorem:f:composition:soundness}
Let $M = \langle D, \Gamma\rangle$ be a system configuration, $s,s' \in \Omega_{M}$ be two $M$-states such that $\mathit{pState}(s) \cong_{M,\mathit{user}(\mathit{a},s)}^{\mathit{data}} \mathit{pState}(s')$, $\mathit{tuple}(s) = \mathit{tuple}(s')$, $\mathit{invoker}(s) = \mathit{invoker}(s')$, and $\mathit{trigger}(s) = \mathit{trigger}(s')$, and $f$ be the \acf{} as above.
The following conditions hold:
\begin{compactenum}
\item If $\mathit{trigger}(s) = \epsilon$, for any action $a$ in ${\cal A}_{D,{\cal U}}$, $f(s,a) = \top$ iff $f(s', a) = \top$.

\item If $\mathit{trigger}(s) \in {\cal TRIGGER}_{D}$, $f(s,\mathit{trigCond}(s)) = \top$ iff $f(s', \mathit{trigCond}(s)) = \top$.

\item If $\mathit{trigger}(s) \in {\cal TRIGGER}_{D}$, $\mathit{trigCond}(s) = \langle u, \mathtt{SELECT}, \\ \psi\rangle$, $[\psi]^{s.\mathit{db}} = [\psi]^{s'.\mathit{db}} = \top$ , $f(s,\mathit{trigAct}(s)) = \top$ iff $f(s', \mathit{trigAct}(s')) = \top$.
\end{compactenum}

\end{lemma}

\begin{proof}
We prove our three claims by contradiction.

\begin{compactenum}
\item Assume, for contradiction's sake, that there are two states $s$ and $s'$ and an action $a$ such that $\mathit{trigger}(s) = \mathit{trigger}(s') = \epsilon$, $\mathit{pState}(s) \cong_{\mathit{user}(\mathit{a},s),M}^{\mathit{data}} \mathit{pState}(s')$, $f(s, a) \\ = \top$, and $f(s', a) = \bot$.
From $f$'s definition, $f(s, a) = \top$, $f(s', a) = \bot$, and \thref{theorem:f:conf:composition:soundness}, it follows that $f_{\mathit{int}}(s,a) \\ =  \top$, $f_{\mathit{int}}(s',a) = \bot$, and $f_{\mathit{conf}}^{\mathit{user}(a,s)}(s,a)=f_{\mathit{conf}}^{\mathit{user}(a,s')}(s', \\ a)=\top$.
From this, it follows that $s' \not \auth^{\mathit{approx}} a$.
From $f_{\mathit{int}}(s,a) = \top$, it follows $s \auth^{\mathit{approx}} a$.
From this, $a \in {\cal A}_{D, {\cal U}}$, and \thref{theorem:f:int:actions:independent:from:db}, it follows $s' \auth^{\mathit{approx}} a$, which contradicts $s' \not \auth^{\mathit{approx}} a$.
This completes the proof for the first claim.

\item Assume, for contradiction's sake, that there are two states $s$ and $s'$ such that $\mathit{trigger}(s) = \mathit{trigger}(s')$, $\mathit{trigger} \\ (s) \neq \epsilon$, $\mathit{pState}(s) \cong_{\mathit{user}(\mathit{a},s),M}^{\mathit{data}} \mathit{pState}(s')$, $f(s,a) = \top$, and $f(s', a) =  \bot$, where $\mathit{trigCond}(s)= \mathit{trigCond}(s') = a$.
From $f$'s definition, $f(s, a)  = \top$, $f(s', a) = \bot$, and \thref{theorem:f:conf:composition:soundness},  it follows that $f_{\mathit{int}}(s,a) = \top$, $f_{\mathit{int}}(s',a) = \bot$, and $f_{\mathit{conf}}^{\mathit{user}(a,s)}  (s,a)=f_{\mathit{conf}}^{\mathit{user}(a,s')}(s',a)=\top$.
From $f_{\mathit{int}}$'s definition, $\mathit{trigger}(s') \neq \epsilon$, and $a = \mathit{trigCond}(s')$, it follows that $f_{\mathit{int}}(s',a) = \top$, which contradicts $f_{\mathit{int}}(s', \\ a) = \bot$.
This completes the proof for the second claim.

\item Assume, for contradiction's sake, that there are two states $s$ and $s'$ such that $\mathit{trigger}(s) = \mathit{trigger}(s') = t$, $\mathit{trigger}  (s) \neq \epsilon$, $\mathit{pState}(s) \cong_{\mathit{user}(\mathit{a},s),M}^{\mathit{data}} \mathit{pState}(s')$, $[\psi]^{s.\mathit{db}} \\ = [\psi ]^{s'.\mathit{db}} = \top$, $f(s,a) = \top$, and $f(s', a) =  \bot$, where $a = \mathit{trigAct}(s) = \mathit{trigAct}(s')$.
From $f$'s definition, $f(s, a)  = \top$, $f(s', a) = \bot$, and \thref{theorem:f:conf:composition:soundness}, it follows that $f_{\mathit{conf}}^{\mathit{user}(a,s)}(s,a)=f_{\mathit{conf}}^{\mathit{user}(a,s')}(s',  a)=\top$, $f_{\mathit{int}}(s,a) = \top$, and $f_{\mathit{int}}(s',a) = \bot$.
From this, it follows that  $s' \not \auth^{\mathit{approx}} t$. 
From $f_{\mathit{int}}(s,a) = \top$, it follows $s \auth^{\mathit{approx}} t$. 
There are two cases depending on $t$'s security mode:
\begin{compactenum}
\item $\mathit{mode}(t) = A$.
From this and  $s \auth^{\mathit{approx}} t$, it follows that $s \auth^{\mathit{approx}} a$ and $s \auth^{\mathit{approx}} a'$, where $a' = \mathit{getAction}(\mathit{statement}(t), \mathit{owner}(t),  \mathit{tuple}(s))$ is the trigger's action associated with the trigger's owner.
Note that $s$ and $s'$ are data indistinguishable.
From this, $a,a' \in {\cal A}_{D, {\cal U}}$, and \thref{theorem:f:int:actions:independent:from:db}, it follows that $s' \auth^{\mathit{approx}} a$ and $s' \auth^{\mathit{approx}} a'$.
From $s' \auth^{\mathit{approx}} a$, $s' \auth^{\mathit{approx}} a'$, $[\psi ]^{s'.\mathit{db}} = \top$, and the rule \texttt{EXECUTE TRIGGER - 2}, it follows that $s' \auth^{\mathit{approx}} t$, which contradicts $s' \not \auth^{\mathit{approx}} t$.

\item $\mathit{mode}(t) = O$.
From this and  $s \auth^{\mathit{approx}} t$, it follows that $s \auth^{\mathit{approx}} a$.
Note that $s$ and $s'$ are data indistinguishable.
From this, $a,a' \in {\cal A}_{D, {\cal U}}$, and \thref{theorem:f:int:actions:independent:from:db}, it follows that $s' \auth^{\mathit{approx}} a$.
From this, $[\psi ]^{s'.\mathit{db}} = \top$, and the rule \texttt{EXECUTE TRIGGER - 1}, it follows that $s' \auth^{\mathit{approx}} t$, which contradicts $s' \not \auth^{\mathit{approx}} t$.

\end{compactenum}
This completes the proof for the third claim.
\end{compactenum}
This completes the proof.
\end{proof}

In \thref{theorem:composition:confidentiality}, we prove the main result of this section, namely that $f$ provides \confidentiality{}.
We first recall the concept of \emph{derivation}.
Given a judgment $r, i \attMod \phi$, a \emph{derivation of  $r, i \attMod \phi$ with respect to $\attackerModel$}, or \emph{a derivation of $r,i \attMod \phi$} for short, is a proof tree, obtained by applying the rules defining $\attackerModel$, that ends in $r, i \attMod \phi$.
With a slight abuse of notation, we use $r,i \attMod \phi$ to denote both the judgment and its derivation.
The length of a derivation, denoted $|r, i \attMod \phi|$, is the number of rule applications in it.

\begin{lemma}\thlabel{theorem:composition:confidentiality}
Let $M$ be a system configuration, $f$ be as above, and $P=  \langle M, f \rangle$ be an \accessControlConfiguration{}.
For any user $u \in {\cal U}$, the \acf{} $f$ provides  \confidentiality{}  with respect to $P$, $u$, $\attackerModel$, and $\cong_{P,u}$.
\end{lemma}

\begin{proof}
Let $u$ be a user in ${\cal U}$, $P = \langle M, f \rangle$ be an \accessControlConfiguration{}, where $M= \langle D,\Gamma\rangle$ is a system configuration and $f$ is as above, and $L$ be the $P$-LTS.
Furthermore, let $r$ be a run in $\mathit{traces}(L)$, $i$ be an integer such that $1 \leq i \leq |r|$, and $\phi$ be a sentence such that $r, i \attMod \phi$ holds.
We claim that also $\mathit{secure}_{P,u}(r,i \attMod \phi)$ holds.
The theorem follows trivially from the claim.

We now show that for all $r \in \mathit{traces}(L)$, all $i$ such that $1 \leq i \leq |r|$, and all sentences $\phi$ such that $r, i \attMod \phi$ holds, then also $\mathit{secure}_{P,u}(r, i \attMod \phi)$ holds.
We prove our claim by induction on the length of the derivation $r, i \attMod \phi$.
In the following, we denote by $e$ the function $\mathit{extend}$.

\smallskip
\noindent
{\bf Base Case: } 
Assume that $|r, i \attMod \phi| = 1$.
There are a number of cases depending on the rule used to obtain $r, i \attMod \phi$.
\begin{compactenum}
\item \emph{\texttt{SELECT} Success - 1}.
Let $i$ be such that $r^{i} = r^{i-1} \concat \langle u, \mathtt{SELECT}, \phi \rangle \concat s$, where $s = \langle \mathit{db}, U, \mathit{sec}, T,V,c\rangle \in \Omega_{M}$ and $\mathit{last}(r^{i-1}) = s'$, where $s' =  \langle \mathit{db}, U, \mathit{sec}, T,V,c'\rangle$. 
From the rules, it follows that $f(s', \langle u, \mathtt{SELECT}, \phi\rangle) = \top$.
From this and $f$'s definition, it follows that $f_{\mathit{int}}(s', \langle u, \\ \mathtt{SELECT}, \phi\rangle) = \top$ and $f_{\mathit{conf}}^{u}(s', \langle u, \mathtt{SELECT}, \phi\rangle) = \top$, because $\mathit{user}(s', \langle u, \mathtt{SELECT}, \phi\rangle) = u$.
From $\; f_{\mathit{conf}}^{u}(s',\;   \langle u, \\ \mathtt{SELECT}, \phi\rangle) = \top$, it follows $\mathit{secure}(u, \phi, s') = \top$.
From this, \thref{theorem:secure:equivalent:modulo:indistinguishable:state}, and $\mathit{pState}(s) = \mathit{pState}(s')$, it follows  $\mathit{secure}(u, \phi, s) = \top$.
From this,  \thref{theorem:secure:sound:under:approximation}, and $\mathit{last}(r^{i}) = s$, it follows that $\mathit{secure}_{P,u}(r,i \attMod \phi)$ holds.

\item \emph{\texttt{SELECT} Success - 2}.
The proof for this case is similar to that of \emph{\texttt{SELECT} Success - 1}.

\item \emph{\texttt{INSERT} Success}.
Let $i$ be such that $r^{i} = r^{i-1} \concat \langle u, \mathtt{INSERT}, \\ R, \overline{t} \rangle \concat s$ , where $s = \langle \mathit{db}, U, \mathit{sec}, T, V, c \rangle \in \Omega_{M}$ and $\mathit{last}(r^{i-1}) =  \langle \mathit{db}', U, \mathit{sec}, T, V, c' \rangle$, and $\phi$ be $R(\overline{t})$.
Then, $\mathit{secure}_{P,u}(r,i \attMod R(\overline{t}))$  holds.
Indeed, in all runs $r'$ $(P,u)$-indistinguishable from $r^{i}$ the last action is $\hfill \langle u, \hfill \\ \mathtt{INSERT},  R, \overline{t} \rangle$.
Furthermore, the action has been executed successfully.
Therefore, according to the LTS rules, $\overline{t} \in \mathit{last}(r').\mathit{db}(R)$ for all runs $r' \in \llbracket r^{i}\rrbracket_{P,u}$.
From this and the relational calculus semantics, it follows that $[R(\overline{t})]^{\mathit{last}(r').\mathit{db}} = \top$ for all runs $r' \in \llbracket r^{i}\rrbracket_{P,u}$.
Hence, $\mathit{secure}_{P,u}(r,i \attMod R(\overline{t}))$ holds.

\item \emph{\texttt{INSERT} Success - FD}.
Let $i$ be such that $r^{i} = r^{i-1} \concat \langle u, \mathtt{INSERT}, R, (\overline{v}, \overline{w}, \overline{q}) \rangle \concat s$, where $s = \langle \mathit{db}, U, \mathit{sec}, T, V, c \rangle \\ \in \Omega_{M}$ and $\mathit{last}(r^{i-1}) =  \langle \mathit{db}', U, \mathit{sec}, T, V, c' \rangle$, and $\phi$ be $\neg \exists \overline{y},\overline{z}.\, R(\overline{v}, \overline{y}, \overline{z}) \wedge \overline{y} \neq \overline{w}$.
From the rule's definition, it follows that $\mathit{secEx}(s) = \bot$.
From this and the LTS rules, it follows that $f(s', \langle u, \mathtt{INSERT}, R, (\overline{v}, \overline{w}, \overline{q}) \rangle) = \top$.
From this and $f$'s definition, it follows that $f_{\mathit{conf}}^{u}(s', \langle u,\\ \mathtt{INSERT}, R, (\overline{v}, \overline{w}, \overline{q})\rangle) = \top$, because $\hfill \mathit{user}(s', \hfill \langle u, \hfill \mathtt{INSERT}, \\ R, (\overline{v}, \overline{w}, \overline{q}) \rangle) = u$.
From this and $f_{\mathit{conf}}^{u}$'s definition, it follows that  $\mathit{secure}(u, \phi, \mathit{last}(r^{i-1})) = \top$ holds because $\phi$ is equivalent to $\mathit{getInfoS}(\gamma,a)$ for some $\gamma \in \mathit{Dep}(\Gamma, a)$, where $a = \langle u, \mathtt{INSERT}, R, (\overline{v}, \overline{w}, \overline{q}) \rangle$.
From this and \thref{theorem:secure:sound:under:approximation}, it follows that $\mathit{secure}_{P,u}(r,i-1 \attMod \phi)$ holds.
We claim that $\mathit{secure}^{\mathit{data}}_{P,u}(r,i \attMod \phi)$ holds.
From this and \thref{theorem:ibsec:correctness:secure:3}, it follows that also  $\mathit{secure}_{P,u}(r,i  \attMod \phi)$ holds.

We now prove our claim that  $\mathit{secure}^{\mathit{data}}_{P,u}(r,i \attMod \phi)$ holds.
Let $s'$ be the state $\mathit{last}(r^{i-1})$.
Furthermore, for brevity's sake, in the following we omit the $\mathit{pState}$ function where needed.
For instance, with a slight abuse of notation, we write $\llbracket s' \rrbracket^{\mathit{data}}_{u,M}$ instead of $\llbracket \mathit{pState}(s') \rrbracket^{\mathit{data}}_{u,M}$.
There are two cases:
\begin{compactenum}
\item the \texttt{INSERT} command has caused an integrity constraint violation, i.e., $\mathit{Ex}(s) \neq \emptyset$.
From $\mathit{secure}(u, \phi, \\ s') = \top$ and \thref{theorem:secure:sound:under:approximation}, it follows that $\mathit{secure}^{\mathit{data}}_{P,u}(r,\\i-1 \attMod \phi)$ holds.
From this, it follows that $[\phi]^{v} = [\phi]^{s'}$ for any $v \in \llbracket s' \rrbracket^{\mathit{data}}_{u,M}$.
From this and the fact that the \texttt{INSERT} command caused an exception (i.e., $s' = s$), it follows that $[\phi]^{v} = [\phi]^{s}$ for any $v \in \llbracket s \rrbracket^{\mathit{data}}_{u,M}$.
From this, it follows that $\mathit{secure}^{\mathit{data}}_{P,u}(r,i \attMod \phi)$ holds.

\item the \texttt{INSERT} command has not caused exceptions, i.e., $\mathit{Ex}(s) = \emptyset$.
From $\mathit{secure}(u, \phi,s') = \top$ and \thref{theorem:secure:sound:under:approximation}, it follows that $\mathit{secure}^{\mathit{data}}_{P,u}(r,i-1 \attMod \phi)$ holds.
From this, it follows that $[\phi]^{v} = [\phi]^{s'}$ for any $v \in \llbracket s' \rrbracket^{\mathit{data}}_{u,M}$.
Furthermore, from \ref{theorem:getInfo:sound:and:complete} and $\mathit{Ex}(s) = \emptyset$, it follows that $\phi$ holds in $s'$.
Let $A_{s',R,\overline{t}}$ be the set $\{\langle \mathit{db}[R \oplus \overline{t}], U, \mathit{sec}, T,V\rangle \in \Pi_{M} \,|\, \exists \mathit{db}' \in \Omega_{D}.\, \langle \mathit{db}', U,  \mathit{sec},  T, V\rangle \in \llbracket s'\rrbracket^{\mathit{data}}_{M,u}\}$.
It is easy to see that $\llbracket s\rrbracket^{\mathit{data}}_{M,u} \subseteq A_{s',R,\overline{t}}$.
We now show that $\phi$ holds for any $z \in A_{s',R,\overline{t}}$.
Let $z_{1} \in \llbracket s' \rrbracket^{\mathit{data}}_{M,u}$.
From $[\phi]^{v} = [\phi]^{s'}$ for any $v \in \llbracket s' \rrbracket^{\mathit{data}}_{u,M}$ and the fact that $\phi$ holds in $s'$, it follows that $[\phi]^{z_{1}} = \top$.
Therefore, for any $(\overline{k}_{1},\overline{k}_{2}, \overline{k}_{3}) \in R(z_{1})$ such that $|\overline{k}_1| = |\overline{v}|$, $|\overline{k}_2| = |\overline{w}|$, and $|\overline{k}_3| = | \overline{z}|$, if $k_{1} = \overline{v}$, then $k_{2} = \overline{w}$. 
Then, for any $(\overline{k}_{1},\overline{k}_{2}, \overline{k}_{3}) \in R(z_{1}) \cup \{(\overline{v}, \overline{w}, \overline{q})\}$ such that $|\overline{k}_1| = |\overline{v}|$, $|\overline{k}_2| = |\overline{w}|$, and $|\overline{k}_3| = | \overline{z}|$, if $k_{1} = \overline{v}$, then $k_{2} = \overline{w}$.
Therefore, $\phi$ holds also in $z_{1}[R \oplus \overline{t}] \in A_{\mathit{pState}(s'), R, \overline{t}}$.	
Hence, $[\phi]^{z} = \top$ for any $z \in A_{s',R,\overline{t}}$.
From this and $\llbracket s\rrbracket^{\mathit{data}}_{M,u} \subseteq A_{s',R,\overline{t}}$, it follows that  $[\phi]^{z} = \top$ for any $z \in \llbracket s\rrbracket^{\mathit{data}}_{M,u}$.
From this, it follows that $\mathit{secure}^{\mathit{data}}_{P,u}(r,i ,u, \phi)$ holds.
\end{compactenum}

\item \emph{\texttt{INSERT} Success - ID}.
The proof of this case is similar to that for the \emph{\texttt{INSERT} Success - FD}.

\item \emph{\texttt{DELETE} Success}.
The proof for this case is similar to that of \emph{\texttt{INSERT} Success}.

\item \emph{\texttt{DELETE} Success - ID}.
The proof of this case is similar to that for the \emph{\texttt{INSERT} Success - FD}.

\item \emph{\texttt{INSERT} Exception}.
Let $i$ be such that $r^{i} = r^{i-1} \concat \langle u, \mathtt{INSER}, R, \overline{t} \rangle \concat s$, where $s = \langle \mathit{db}, U, \mathit{sec}, T, V, c \rangle  \in  \Omega_{M}$ and $\mathit{last}(r^{i-1}) =  \langle \mathit{db}', U, \mathit{sec}, T, V, c' \rangle$, and $\phi$ be $\neg R(\overline{t})$.
From the rule's definition, it follows that $\mathit{secEx}(s) = \bot$.
From this and the LTS rules, it follows that $\hfill f(s', \hfill \langle u, \\ \mathtt{INSERT}, R, \overline{t} \rangle) = \top$.
From this and $f$'s definition, it follows that $f_{\mathit{conf}}^{u}(s', \langle u, \mathtt{INSERT}, R, \overline{t} \rangle) = \top$, because $ \mathit{user}(s', \langle u,  \mathtt{INSERT}, R, \overline{t} \rangle) = u$.
From this and $f_{\mathit{conf}}^{u}$'s definition, it follows that  $\mathit{secure}(u, \phi, \mathit{last}(r^{i-1})) = \top$ holds because $\phi = \mathit{getInfo}(\langle u, \mathtt{INSERT}, R, \overline{t} \rangle)$.
From this and \thref{theorem:secure:sound:under:approximation}, it follows that $\mathit{secure}_{P,u}(r,i-1 \attMod \phi)$ holds.
From the LTS semantics, it follows that $\mathit{pState}(s) \cong^{data}_{u,M} \mathit{pState}(\mathit{last}(r^{i-1}))$.
From this,  $\mathit{secure}(u, \\  \phi, \mathit{last}(r^{i-1})) = \top$, and \thref{theorem:secure:equivalent:modulo:indistinguishable:state}, it follows that  $\mathit{secure}(u, \phi,  \mathit{last}(r^{i})) = \top$.
From this and  \thref{theorem:secure:sound:under:approximation}, it follows that $\mathit{secure}_{P,u}(r,i \attMod \phi)$ holds.

\item \emph{\texttt{DELETE} Exception}. 
The proof for this case is similar to that of \emph{\texttt{INSERT} Exception}.

\item \emph{\texttt{INSERT} FD Exception}.
Let $i$ be such that $r^{i} = r^{i-1} \concat \langle u, \mathtt{INSERT}, R, (\overline{v}, \overline{w}, \overline{q}) \rangle \concat s$, where $s = \langle \mathit{db}, U, \mathit{sec}, T, V, c \rangle \\ \in \Omega_{M}$ and $\mathit{last}(r^{i-1}) =  \langle \mathit{db}', U, \mathit{sec}, T, V, c' \rangle$, and $\phi$ be $\exists \overline{y},\overline{z}.\, R(\overline{v}, \overline{y}, \overline{z}) \wedge \overline{y} \neq \overline{w}$.
From the rule's definition, it follows that $\mathit{secEx}(s) = \bot$.
From this and the LTS rules, it follows that $f(s', \langle u, \mathtt{INSERT}, R, (\overline{v}, \overline{w}, \overline{q}) \rangle) = \top$.
From this and $f$'s definition, it follows that $f_{\mathit{conf}}^{u}(s', \langle u, \\ \mathtt{INSERT}, R, (\overline{v}, \overline{w}, \overline{q}) \rangle) = \top$, because $ \mathit{user}(s',  \langle u,   \mathtt{INSERT}, \\ R, \overline{t} \rangle) = u$.
From this and $f_{\mathit{conf}}^{u}$'s definition, it follows that  $\mathit{secure}(u, \phi, \mathit{last}(r^{i-1})) = \top$ because $\phi = \mathit{getInfoV}(\gamma,  \langle u, \mathtt{INSERT}, R, (\overline{v}, \overline{w}, \overline{q}) \rangle)$ for some constraint $\gamma \in \mathit{Dep}(\Gamma, \langle u, \mathtt{INSERT}, R, (\overline{v}, \overline{w}, \overline{q}) \rangle)$.
From this and  \thref{theorem:secure:sound:under:approximation}, it follows that $\mathit{secure}_{P,u}(r,i-1 \attMod \phi)$ holds.
From the LTS semantics, it follows that $\mathit{pState}(s) \\ \cong^{data}_{u,M} \mathit{pState}(\mathit{last}(r^{i-1}))$.
From this,  \thref{theorem:secure:equivalent:modulo:indistinguishable:state}, and $\mathit{secure}(u, \phi, \mathit{last}(r^{i-1})) = \top$,  it follows that  $\mathit{secure}(u, \phi, \\ \mathit{last}(r^{i})) = \top$.
From this and  \thref{theorem:secure:sound:under:approximation}, it follows that also $\mathit{secure}_{P,u}(r,i \attMod \phi)$ holds.

\item \emph{\texttt{INSERT} ID Exception}.
The proof for this case is similar to that of \emph{\texttt{INSERT} FD Exception}.

\item \emph{\texttt{DELETE} FD Exception}.
The proof for this case is similar to that of \emph{\texttt{INSERT} FD Exception}.

\item \emph{Integrity Constraint}.
The proof of this case follows trivially from the fact that for any state $s = \langle \mathit{db}, U, \mathit{sec},  T, \\ V,  c \rangle  \in \Omega_{M}$ and any $\gamma \in \Gamma$, $[\gamma]^{\mathit{db}} = \top$  by definition.

\item \emph{Learn \texttt{GRANT}/\texttt{REVOKE} Backward}. 
Let $i$ be such that $r^{i} = r^{i-1} \concat t \concat s$, where $s = \langle \mathit{db}, U, \mathit{sec}, T, V, c \rangle  \in \Omega_{M}$, $\mathit{last}(r^{i-1}) =  \langle \mathit{db}, U, \mathit{sec}', T, V, c' \rangle$, and $t$ be a trigger whose \texttt{WHEN} condition is $\phi$ and whose action is either a \texttt{GRANT} or a \texttt{REVOKE}.
From the rule's definition, it follows that $\mathit{secEx}(s) = \bot$. 
From this and the LTS rules, it follows that $f(\mathit{last}(r^{i-1}), \langle u', \mathtt{SELECT}, \phi \rangle) = \top$, where $u'$ is either the trigger's owner or the trigger's invoker depending on the security mode.
From this and $f$'s definition, it follows $f_{\mathit{conf}}^{u}(\mathit{last}(r^{i-1}), \langle u', \mathtt{SELECT}, \phi \rangle) = \top$, because $\mathit{user}(\mathit{last}(r^{i-1}), \langle u',  \mathtt{SELECT}, \phi \rangle) = u$ because $t$'s invoker is $u$ according to the rules.
From this and $f_{\mathit{conf}}^{u}$'s definition, it follows $\mathit{secure}(u,\phi,\mathit{last}(r^{i-1})) = \top$.
From this and \ref{theorem:secure:sound:under:approximation}, it follows that $\mathit{secure}_{P,u}(r,i-1 \attMod \phi)$ holds.

\item \emph{Trigger \texttt{GRANT} Disabled Backward}. 
Let $i$ be such that $r^{i} = r^{i-1} \concat t \concat s$, where $s = \langle \mathit{db}, U, \mathit{sec}, T, V, c \rangle  \in \Omega_{M}$, $\mathit{last}(r^{i-1}) =  \langle \mathit{db}, U, \mathit{sec}', T, V, c' \rangle$, and $t$ be a trigger whose \texttt{WHEN} condition is $\psi$, and $\phi$ be $\neg \psi$.
From the rule's definition, it follows that $\mathit{secEx}(s) = \bot$. 
From this and the LTS rules, it follows that $f(\mathit{last}(r^{i-1}), \langle u', \\ \mathtt{SELECT}, \phi \rangle) = \top$, where $u'$ is either the trigger's owner or the trigger's invoker depending on the security mode.
From this and $f$'s definition, it follows $f_{\mathit{conf}}^{u}(\mathit{last}(r^{i-1}), \\ \langle u', \mathtt{SELECT}, \phi \rangle) = \top$, as $\mathit{user}(\mathit{last}(r^{i-1}), \langle u',   \mathtt{SELECT}, \phi \rangle) \\ = u$ because $t$'s invoker is $u$ according to the rules.
From this and $f_{\mathit{conf}}^{u}$'s definition, it follows that also $\mathit{secure}(u, \phi,\mathit{last}(r^{i-1})) = \top$.
From this and \ref{theorem:secure:sound:under:approximation}, it follows that $\mathit{secure}_{P,u}(r,i-1 \attMod \phi)$ holds.

\item \emph{Trigger \texttt{REVOKE} Disabled Backward}. 
The proof for this case is similar to that of \emph{Trigger \texttt{GRANT} Disabled Backward}.

\item \emph{Trigger \texttt{INSERT} FD Exception}.
Let $i$ be such that $r^{i} = r^{i-1} \concat t \concat s$, where $s = \langle \mathit{db}, U, \mathit{sec}, T, V, c \rangle  \in \Omega_{M}$, $\mathit{last}(r^{i-1}) =  \langle \mathit{db}, U, \mathit{sec}', T, V, c' \rangle$, and $t$ be a trigger whose \texttt{WHEN} condition is $\phi$ and whose action $\mathit{act}$ is a \texttt{INSERT} statement $\langle u', \mathtt{INSERT}, R, (\overline{v}, \overline{w}, \overline{q}) \rangle$.
Furthermore, let $\phi$ be $\exists \overline{y},\overline{z}.\, R(\overline{v}, \overline{y}, \overline{z}) \wedge \overline{y} \neq \overline{w}$.
From the rule's definition, it follows that $\mathit{secEx}(s) = \bot$.
From this and the LTS rules, it follows that $f(\mathit{last}(r^{i-1}), \mathit{act}) = \top$.
From this and $f$'s definition, it follows that $f_{\mathit{conf}}^{u}(\mathit{last}(r^{i-1}),  \mathit{act}) = \top$, because $\mathit{user}(\mathit{last}(r^{i-1}),\mathit{act}) = u$ because $t$'s invoker is $u$ according to the rules.
From this and $f_{\mathit{conf}}^{u}$'s definition, it follows that  $\mathit{secure}(u, \phi, \mathit{last}(r^{i-1})) = \top$ because $\phi = \mathit{getInfoV}(\gamma,  \mathit{act})$ for some constraint $\gamma \in \mathit{Dep}(\Gamma,  \mathit{act})$.
From this and  \thref{theorem:secure:sound:under:approximation}, it follows that $\mathit{secure}_{P,u}(r,i-1 \attMod \phi)$ holds.

\item \emph{Trigger \texttt{INSERT} ID Exception}.
The proof for this case is similar to that of \emph{Trigger \texttt{INSERT} ID Exception}.

\item \emph{Trigger \texttt{DELETE} ID Exception}.
The proof for this case is similar to that of \emph{Trigger \texttt{DELETE} ID Exception}.

\item \emph{Trigger Exception}.
Let $i$ be such that $r^{i} = r^{i-1} \concat t \concat s$, where $s = \langle \mathit{db}, U, \mathit{sec}, T, V, c \rangle  \in \Omega_{M}$, $\mathit{last}(r^{i-1}) =  \langle \mathit{db}, U, \mathit{sec}', T, V, c' \rangle$, and $t$ be a trigger whose \texttt{WHEN} condition is $\phi$ and whose action is $\mathit{act}$.
From the rule's definition, it follows that $f(\mathit{last}(r^{i-1}), \langle u', \mathtt{SELECT}, \phi \rangle) = \top$, where $u'$ is either the trigger's owner or the trigger's invoker depending on the security mode.
From this and $f$'s definition, it follows $f_{\mathit{conf}}^{u}(\mathit{last} (r^{i-1}), \langle u', \mathtt{SELECT}, \phi \rangle) \\ = \top$, because $\mathit{user}(\mathit{last}(r^{i-1}), \langle u', \mathtt{SELECT}, \phi \rangle) = u$ since $t$'s invoker is $u$ according to the rules.
From this and $f_{\mathit{conf}}^{u}$'s definition, it follows that $\mathit{secure}(u,\phi,\mathit{last}(r^{i-1})) \\ = \top$.
From this and \ref{theorem:secure:sound:under:approximation}, it follows that $\mathit{secure}_{P,u}(r,i-1 \attMod \phi)$ holds.

\item \emph{Trigger \texttt{INSERT} Exception}.
The proof for this case is similar to that of \emph{\texttt{INSERT} Exception}.

\item \emph{Trigger \texttt{DELETE} Exception}.
The proof for this case is similar to that of \emph{\texttt{DELETE} Exception}.

\item \emph{Trigger Rollback \texttt{INSERT}}.
Let $i$ be such that $r^{i} = r^{i-n-1}  \concat \langle u, \mathtt{INSERT}, R, \overline{t}\rangle \concat s_{1} \concat t_{1} \concat s_{2} \concat \ldots \concat t_{n} \concat s_{n}$, where $s_{1}, s_{2}, \ldots, s_{n} \\ \in \Omega_{M}$ and $t_{1}, \ldots, t_{n} \in {\cal TRIGGER}_{D}$, and $\phi$ be $\neg R(\overline{t})$.
Furthermore, let  $\mathit{last}(r^{i-n-1}) = \langle \mathit{db}', U', \mathit{sec}', T', V', c' \rangle$ and $s_n$ be $\langle \mathit{db}, U, \mathit{sec}, T, V, c \rangle$. 
From the rule's definition, it follows that $\mathit{secEx}(s_{1}) = \bot$.
From this,  it follows that $f(\mathit{last}(r^{i-n-1}), \langle u, \mathtt{INSERT}, R, \overline{t}\rangle) = \top$.
From this and $f$'s definition, it follows $f_{\mathit{conf}}^{u}(\mathit{last}(r^{i-n-1}), \langle u, \\ \mathtt{INSERT}, R, \overline{t}\rangle) = \top$ since $\mathit{user}(\mathit{last}(r^{i-n-1}),\langle u, \mathtt{INSERT}, \\ R, \overline{t}\rangle) = u$.
From this and $f_{\mathit{conf}}^{u}$'s definition, it follows $\mathit{secure}(u,\phi,\mathit{last}(r^{i-n-1})) = \top$ because $\phi = \mathit{getInfo}(\langle u, \\  \mathtt{INSERT}, R, \overline{t}\rangle)$.
From the LTS semantics, it follows that $\mathit{last}(r^{i-n-1}) \cong_{M,u}^{\mathit{data}} s_{n}$ because $\mathit{pState}(\mathit{last}(r^{i-n-1})) = \mathit{pState}(s_n)$.
From this,  \thref{theorem:secure:equivalent:modulo:indistinguishable:state}, and $\mathit{secure}(u,\phi,\mathit{last} \\ (r^{i-n-1})) = \top$, it follows $\mathit{secure}(u, \phi,s_{n}) = \top$.
From this and \thref{theorem:secure:sound:under:approximation}, it follows that $\mathit{secure}_{P,u}(r,i \attMod \phi)$ holds.

\item \emph{Trigger Rollback \texttt{DELETE}}.
The proof for this case is similar to that of \emph{Trigger Rollback \texttt{INSERT}}.
	
\end{compactenum}
	This completes the proof of the base step.

	\smallskip
	\noindent
	{\bf Induction Step: } 
Assume that the claim hold for any derivation of $r, j \attMod \psi$ such that $|r, j \attMod \psi| < |r,i \attMod \phi|$.
We now prove that the claim also holds for $r,i \attMod \phi$.
There are a number of cases depending on the rule used to obtain $r,i \attMod \phi$.
\begin{compactenum}
\item \emph{View}.
The proof of this case follows trivially from the semantics of the relational calculus extended over views.

\item \emph{Propagate Forward \texttt{SELECT}}.
Let $i$ be such that $r^{i+1} = r^{i} \concat \langle u, \mathtt{SELECT}, \psi \rangle \concat s$, where $s = \langle \mathit{db}, U, \mathit{sec}, T, V, c \rangle  \in \Omega_{M}$ and $\mathit{last}(r^{i}) =  \langle \mathit{db}', U', \mathit{sec}', T', V', c' \rangle$.
From the rule, it follows that $r,i \attMod \phi$ holds.
From this and the induction hypothesis, it follows that $\mathit{secure}_{P,u}(r,i  \attMod \phi)$ holds.
From \thref{theorem:f:composition:pec:1}, the action $\langle u, \mathtt{SELECT},  \psi \rangle$ preserves the equivalence class with respect to $r^{i}$, $P$, and $u$.
From this,  \thref{theorem:secure:extend:on:runs:select:create}, and $\mathit{secure}_{P,u}(r,i  \attMod \phi)$, it follows that also $\mathit{secure}_{P,u}(r,i+1 \attMod \phi)$ holds.

\item \emph{Propagate Forward \texttt{GRANT/REVOKE}}.
Let $i$ be such that $r^{i+1} = r^{i} \concat \langle \mathit{op}, u', p, u \rangle \concat s$, where $s = \langle \mathit{db}, U, \mathit{sec}, T, V, c \rangle  \in \Omega_{M}$ and $\mathit{last}(r^{i}) =  \langle \mathit{db}', U', \mathit{sec}', T', V', c' \rangle$.
From the rule, it follows that $r,i \attMod \phi$ holds.
From this and the induction hypothesis, it follows that $\mathit{secure}_{P,u}(r,i  \attMod \phi)$ holds.
From \thref{theorem:f:composition:pec:1}, the action $\langle \mathit{op}, u', p, u \rangle$ preserves the equivalence class with respect to $r^{i}$, $P$, and $u$.
From this,  \thref{theorem:secure:extend:on:runs:grant:revoke}, and $\mathit{secure}_{P,u}(r,i  \attMod \phi)$, it follows that also $\mathit{secure}_{P,u}(r,i+1 \attMod \phi)$ holds.

\item \emph{Propagate Forward \texttt{CREATE}}.
The proof for this case is similar to that of \emph{Propagate Forward \texttt{SELECT}}.

\item \emph{Propagate Backward \texttt{SELECT}}.
Let $i$ be such that $r^{i+1} = r^{i} \concat \langle u, \mathtt{SELECT}, \psi \rangle \concat s$, where $s = \langle \mathit{db}', U', \mathit{sec}', T', V', c' \rangle  \\ \in \Omega_{M}$ and $\mathit{last}(r^{i}) =  \langle \mathit{db}, U, \mathit{sec}, T, V, c \rangle$.
From the rule, it follows that $r,i+1 \attMod \phi$ holds.
From this and the induction hypothesis, it follows that $\mathit{secure}_{P,u}  (r,i+1  \attMod \phi)$ holds.
From \thref{theorem:f:composition:pec:1}, the action $\langle u, \mathtt{SELECT}, \psi \rangle$ preserves the equivalence class with respect to $r^{i}$, $P$, and $u$.
From this,  \thref{theorem:secure:extend:on:runs:select:create}, and $\mathit{secure}_{P,u}(r,i+1 \\ \attMod \phi)$, it follows that also $\mathit{secure}_{P,u}(r,i \attMod \phi)$ holds.

\item \emph{Propagate Backward \texttt{GRANT/REVOKE}}.
Let $i$ be such that $r^{i+1} = r^{i} \concat \langle \mathit{op},u',p,u \rangle \concat s$, where $s = \langle \mathit{db}', U', \mathit{sec}', T', V',\\ c' \rangle   \in \Omega_{M}$ and $\mathit{last}(r^{i}) =  \langle \mathit{db}, U, \mathit{sec}, T, V, c \rangle$.
From the rule, it follows that $r,i+1 \attMod \phi$ holds.
From this and the induction hypothesis, it follows that $\mathit{secure}_{P,u}  (r,i+1  \attMod \phi)$ holds.
From \thref{theorem:f:composition:pec:1}, the action $\langle \mathit{op},u',p, u \rangle$ preserves the equivalence class with respect to $r^{i}$, $P$, and $u$.
From this,  \thref{theorem:secure:extend:on:runs:grant:revoke}, and $\mathit{secure}_{P,u}(r,i+1  \attMod \phi)$, it follows that also $\mathit{secure}_{P,u}(r,i \attMod \phi)$ holds.

\item \emph{Propagate Backward \texttt{CREATE TRIGGER}}.
The proof for this case is similar to that of \emph{Propagate Backward \texttt{SELECT}}.

\item \emph{Propagate Backward \texttt{CREATE VIEW}}.
Note that the formulae $\psi$ and $\mathit{replace}(\psi,o)$ are semantically equivalent.
This is the only difference between the proof for this case and the one for the  \emph{Propagate Backward \texttt{SELECT}} case.

\item \emph{Rollback Backward - 1}.
Let $i$ be such that $r^{i} = r^{i-n-1}  \concat \langle u, \mathit{op}, R, \overline{t}\rangle \concat s_{1} \concat t_{1} \concat s_{2} \concat \ldots \concat t_{n} \concat s_{n}$, where  $s_{1}, s_{2}, \ldots, s_{n} \in \Omega_{M}$, $t_{1}, \ldots, t_{n} \in {\cal TRIGGER}_{D}$, and  $\mathit{op}$ is one of $\{\mathtt{INSERT}, \\ \texttt{DELETE}\}$.
Furthermore, let $s_{n}$ be $\langle \mathit{db}', U', \mathit{sec}', T',  V', c' \rangle$ and $\mathit{last}(r^{i-n-1})$ be $\langle \mathit{db}, U, \mathit{sec}, T, V, c \rangle$.
From the rule's definition, $r, i \attMod \phi$ holds.
From this and the induction hypothesis, it follows that $\mathit{secure}_{P,u}(r,i \attMod \phi)$ holds.
From \thref{theorem:f:composition:pec:2}, the triggers $t_{j}$ preserve the equivalence class with respect to $r^{i-n-1+j}$, $P$, and $u$ for any $1 \leq j \leq n$.
Therefore, for any $v \in \llbracket r^{i-1} \rrbracket_{{P,u}}$, the run $e(v,t_{n})$ contains the roll-back.
Therefore, for any $v \in \llbracket r^{i-1} \rrbracket_{P,u}$, the state $\mathit{last}(e(v,t_{n}))$ is the state just before the action $\langle u,\mathit{op}, R, \overline{t}\rangle$.
Let $A$ be the set of partial states associated with the roll-back states.
It is easy to see that $A$ is the same as $\{\mathit{pState}(\mathit{last}(t')) | t' \in \llbracket r^{i-n-1} \rrbracket_{P,u}\}$.
From $\mathit{secure}_{P,u}(r,i \attMod \phi)$, it follows that $\phi$ has the same result over all states in $A$.
From this and $A = \{\mathit{pState}(\mathit{last}(t')) | t' \in \llbracket r^{i-n-1} \rrbracket_{P,u}\}$, it follows that $\phi$ has the same result over all states in $ \{\mathit{pState}(\mathit{last}(t')) | \\ t' \in \llbracket r^{i-n-1} \rrbracket_{P,u}\}$.
From this, it follows that $\mathit{secure}_{P,u} \\ (r,i-n-1 \attMod \phi)$ holds.

\item \emph{Rollback Backward - 2}.
Let $i$ be such that $r^{i} = r^{i-1} \concat \langle u, op, R, \overline{t} \rangle \concat s$, where $s = \langle \mathit{db}', U', \mathit{sec}', T', V', c' \rangle  \in \Omega_{M}$, $\mathit{last}(r^{i-1}) =  \langle \mathit{db}, U,  \mathit{sec}, T, V, c \rangle$, and $\mathit{op}$ is one of $\{\mathtt{INSERT}, \mathtt{DELETE}\}$.
From the rule's definition, $r, i \attMod \phi$ holds.
From this and  the induction hypothesis, it follows that $\mathit{secure}_{P,u}(r,i \attMod \phi)$ holds.
From \thref{theorem:f:composition:pec:1}, the action $\langle u,\mathit{op}, R, \overline{t}\rangle$ preserves the equivalence class with respect to $r^{i-1}$, $P$, and $u$. 	
From this, \thref{theorem:secure:extend:on:runs:insert:delete}, the fact that the action does not modify the database state, and $\mathit{secure}_{P,u}(r,i \attMod \phi)$, it follows $\mathit{secure}_{P,u}(r, i-1 \attMod \phi)$.

\item \emph{Rollback Forward - 1}.
Let $i$ be such that $r^{i} = r^{i-n-1}  \concat \langle u, \mathit{op}, R, \overline{t}\rangle \concat s_{1} \concat t_{1} \concat s_{2} \concat \ldots \concat t_{n} \concat s_{n}$, where $s_{1}, s_{2}, \ldots, s_{n} \in \Omega_{M}$, $t_{1}, \ldots, t_{n} \in {\cal TRIGGER}_{D}$, and $\mathit{op}$ is one of  $\{\mathtt{INSERT}, \\ \mathtt{DELETE}\}$.
Furthermore, let $s_{n}$ be $\langle \mathit{db}, U, \mathit{sec}, T, V, c \rangle $ and $\mathit{last}(r^{i-n-1})$ be $\langle \mathit{db}', U', \mathit{sec}', T', V', c' \rangle$.
From the rule's definition, $r, i-n-1 \attMod \phi$ holds.
From this and the induction hypothesis, it follows that $\mathit{secure}_{P,u} (r,i-n-1 \attMod \phi)$ holds.
From \thref{theorem:f:composition:pec:2}, the triggers $t_{j}$ preserve the equivalence class with respect to $r^{i-n-1+j}$, $P$, and $u$ for any $1 \leq j \leq n$.
Independently on the cause of the roll-back (either a security exception or an integrity constraint violation), we claim that the set $A$ of roll-back partial states is $\{\mathit{pState}(\mathit{last}(t')) | t' \in \llbracket r^{i-n-1} \rrbracket_{P,u}\}$.
From $\mathit{secure}_{P,u}  (r, i-n-1 \attMod \phi)$, the result of $\phi$ is the same for all states in $A$.
From this and $A = \{\mathit{pState}(\mathit{last}(t')) | t' \in \llbracket r^{i-n-1} \rrbracket_{P,u}\}$, it follows that also  $\mathit{secure}_{P,u}(r,i \attMod \phi)$ holds.

We now prove our claim.
It is trivial to see (from the LTS's semantics) that the set of rollback's states is a subset of $\{\mathit{last}(v) | v \in \llbracket r^{i-n-1} \rrbracket_{P,u}\}$.
Assume, for contradiction's sake, that there is a state in $\{\mathit{last}(v) | v \in \llbracket r^{i-n-1} \rrbracket_{P,u}\}$ that is not a rollback state for the runs in $\llbracket r^{i} \rrbracket_{P,u}$.
This is impossible since all triggers $t_{1}, \ldots, t_{n}$ preserve the equivalence class. 

\item \emph{Rollback Forward - 2}.
Let $i$ be such that $r^{i} = r^{i-1} \concat \langle u, op, R, \overline{t} \rangle \concat s$, where $\mathit{op} \in \{\mathtt{INSERT}, \mathtt{DELETE}\}$, $s = \langle \mathit{db}, U, \\ \mathit{sec}, T, V, c \rangle  \in \Omega_{M}$ and $\mathit{last}(r^{i-1}) =  \langle \mathit{db}', U', \mathit{sec}',   T', V', \\ c' \rangle$.
From the rule's definition, $r, i-1 \attMod \phi$ holds.
From this and  the induction hypothesis, it follows that $\mathit{secure}_{P,u}(r,i-1 \attMod \phi)$ holds.
From \thref{theorem:f:composition:pec:1}, the action $\langle u,\mathit{op}, R, \overline{t}\rangle$ preserves the equivalence class with respect to $r^{i-1}$, $P$, and $u$. 	
From this, \thref{theorem:secure:extend:on:runs:insert:delete}, the fact that the action does not modify the database state, and $\mathit{secure}_{P,u}(r,i-1 \attMod \phi)$, it follows that also $\mathit{secure}_{P,u}(r,i \attMod \phi)$ holds.

\item \emph{Propagate Forward \texttt{INSERT/DELETE} Success}.
Let $i$ be such that $r^{i} = r^{i-1} \concat \langle u, op, R, \overline{t} \rangle \concat s$, where $\mathit{op} \in \{\mathtt{INSERT}, \\ \mathtt{DELETE}\}$, $s = \langle \mathit{db}, U, \mathit{sec}, T, V, c \rangle  \in \Omega_{M}$ and $\mathit{last}(r^{i-1}) =  \langle \mathit{db}', U', \mathit{sec}',  T', V', c' \rangle$.
From the rule's definition, $r, i-1 \attMod \phi$ holds.
From this and the induction hypothesis, it follows that $\mathit{secure}_{P,u}(r,i-1 \attMod \phi)$ holds.
From \thref{theorem:f:composition:pec:1}, the action $\langle u,\mathit{op}, R, \overline{t}\rangle$ preserves the equivalence class with respect to $r^{i-1}$, $P$, and $u$. 	
From $\mathit{reviseBelif}(r^{i-1}, \phi, r^{i})$, it follows that the execution of $\langle u,\mathit{op}, R, \overline{t}\rangle$ does not alter the content of the tables in $\mathit{tables}(\phi)$  for any $v \in \llbracket r^{i-1}\rrbracket_{P,u}$.
From this, \thref{theorem:secure:extend:on:runs:insert:delete}, and $\mathit{secure}_{P,u}(r,i-1 \attMod \phi)$, it follows that $\mathit{secure}_{P,u} \\ (r,i \attMod \phi)$ holds.

\item \emph{Propagate Forward \texttt{INSERT} Success - 1}.
Let $i$ be such that $r^{i} = r^{i-1} \concat \langle u, op, R, \overline{t} \rangle \concat s$, where $\mathit{op}$ is one of  $\{\mathtt{INSERT},  \mathtt{DELETE}\}$, $s = \langle \mathit{db}, U, \mathit{sec}, T, V, c \rangle  \in \Omega_{M}$, and $\mathit{last}(r^{i-1}) =  \langle \mathit{db}', U', \mathit{sec}',  T', V', c' \rangle$.
From the rule's definition, $r, i-1 \attMod \phi$ holds.
From this and the induction hypothesis, it follows that $\mathit{secure}_{P,u}(r,i-1 \attMod \phi)$ holds.
From \thref{theorem:f:composition:pec:1}, the action $\langle u,\mathit{op}, R, \overline{t}\rangle$ preserves the equivalence class with respect to $r^{i-1}$, $P$, and $u$. 	
We claim that the execution of $\langle u,\mathit{\mathtt{INSERT}}, R, \overline{t}\rangle$ does not alter the content of the tables in $\mathit{tables}(\phi)$.
From this, \thref{theorem:secure:extend:on:runs:insert:delete}, and $\mathit{secure}_{P,u}(r,i-1 \attMod \phi)$, it follows that $\mathit{secure}_{P,u}  (r,i \attMod \phi)$ holds.

We now prove our claim that the execution of $\langle u,\mathit{\mathtt{INSERT}}, \\ R, \overline{t}\rangle$ does not alter the content of the tables in $\mathit{tables}(\phi)$.
From the rule's definition, it follows that $r, i-1 \attMod R(\overline{t})$ holds.
From this and \thref{theorem:attacker:model:sound}, it follows that $[R(\overline{t})]^{\mathit{last}(r^{i-1}).\mathit{db}} = \top$.
From $r, i-1 \attMod R(\overline{t})$ and the induction hypothesis, it follows that $\mathit{secure}_{P,u}(r,i-1 \attMod R(\overline{t}))$ holds.
From this and $[R(\overline{t})]^{\mathit{last}(r^{i-1}).\mathit{db}} = \top$, it follows that $[R(\overline{t})]^{\mathit{last}(v).\mathit{db}} = \top$ for any $v \in \llbracket r^{i-1}\rrbracket_{P,u}$.
From this and the relational calculus semantics, it follows that the execution of $\langle u,\mathit{op}, R, \overline{t}\rangle$ does not alter the content of the tables in $\mathit{tables}(\phi)$ for any $v \in \llbracket r^{i-1}\rrbracket_{P,u}$.

\item \emph{Propagate Forward \texttt{DELETE} Success - 1}.
The proof for this case is similar to that of \emph{Propagate Forward \texttt{INSERT} Success - 1}.

\item \emph{Propagate Backward \texttt{INSERT/DELETE} Success}.
Let $i$ be such that $r^{i} = r^{i-1} \concat \langle u, op, R, \overline{t} \rangle \concat s$, where $\mathit{op} \in \{\mathtt{INSERT}, \\ \mathtt{DELETE}\}$, $s = \langle \mathit{db}, U, \mathit{sec}, T, V, c \rangle  \in \Omega_{M}$ and $\mathit{last}(r^{i-1}) =  \langle \mathit{db}', U', \mathit{sec}',  T', V', c' \rangle$.
From the rule's definition, $r, i  \\ \attMod \phi$ holds.
From this and the induction hypothesis, it follows that $\mathit{secure}_{P,u}(r,i \attMod \phi)$ holds.
From \thref{theorem:f:composition:pec:1}, the action $\langle u,\mathit{op}, R, \overline{t}\rangle$ preserves the equivalence class with respect to $r^{i-1}$, $P$, and $u$. 	
From $\mathit{reviseBelif}(r^{i-1}, \phi,  r^{i})$, it follows that the execution of $\langle u,\mathit{op}, R, \overline{t}\rangle$ does not alter the content of the tables in $\mathit{tables}(\phi)$  for any $v \in \llbracket r^{i-1}\rrbracket_{P,u}$.
From this, \thref{theorem:secure:extend:on:runs:insert:delete}, and $\mathit{secure}_{P,u}(r,i \attMod \phi)$, it follows that $\mathit{secure}_{P,u}  (r,i-1 \attMod \phi)$ holds.

\item \emph{Propagate Backward \texttt{INSERT} Success - 1}.
Let $i$ be such that $r^{i} = r^{i-1} \concat \langle u, op, R, \overline{t} \rangle \concat s$, where $\mathit{op}$ is one of $\{\mathtt{INSERT}, \mathtt{DELETE}\}$, $s = \langle \mathit{db}, U, \mathit{sec}, T, V, c \rangle  \in \Omega_{M}$ and $\mathit{last}(r^{i-1}) =  \langle \mathit{db}', U', \mathit{sec}',  T', V', c' \rangle$.
From the rule's definition, $r, i \attMod \phi$ holds.
From this and the induction hypothesis, it follows that $\mathit{secure}_{P,u}(r,i \attMod \phi)$ holds.
From \thref{theorem:f:composition:pec:1}, the action $\langle u,\mathit{op}, R, \overline{t}\rangle$ preserves the equivalence class with respect to $r^{i-1}$, $P$, and $u$. 	
We claim that the execution of $\langle u,\mathit{\mathtt{INSERT}}, R, \overline{t}\rangle$ does not alter the content of the tables in $\mathit{tables}(\phi)$ for any $v \in \llbracket r^{i-1}\rrbracket_{P,u}$ (the proof of this claim is in the proof of the \emph{Propagate Forward \texttt{INSERT} Success - 1} case).
From this, \thref{theorem:secure:extend:on:runs:insert:delete}, and $\mathit{secure}_{P,u}(r,i \attMod \phi)$, it follows that $\mathit{secure}_{P,u}  (r,i-1 \attMod \phi)$ holds.

\item \emph{Propagate Backward \texttt{DELETE} Success - 1}.
The proof for this case is similar to that of \emph{Propagate Forward \texttt{DELETE} Success - 1}.

\item \emph{Reasoning}.
Let $\Delta$ be a subset of $\{ \delta \,| \, r,i \attMod \delta \}$ and $\mathit{last}(r^{i}) =  \langle \mathit{db}, U, \mathit{sec},  T, V, c \rangle$.
From the induction hypothesis, it follows that $\mathit{secure}_{P,u}(r,i \attMod \delta)$ holds for any $\delta \in \Delta$.
Note that, given any $\delta \in \Delta$, from $r,i \attMod \delta$ and \thref{theorem:attacker:model:sound}, it follows that $\delta$ holds in $\mathit{last}(r^{i})$.
From this, $\mathit{secure}_{P,u}(r,i \attMod \delta)$ holds for any $\delta \in \Delta$, $\Delta \models_{\mathit{fin}} \phi$, and \thref{theorem:reasoning:preserves:security}, it follows that  $\mathit{secure}_{P,u}(r, \\ i \attMod \phi)$ holds.

\item \emph{Learn \texttt{INSERT} Backward - 3}.
Let $i$ be such that $r^{i} = r^{i-1} \concat \langle u, \mathtt{INSERT}, R, \overline{t} \rangle \concat s$, where $s = \langle \mathit{db}', U', \mathit{sec}', T', V', \\ c' \rangle  \in \Omega_{M}$ and $\mathit{last}(r^{i-1}) =  \langle \mathit{db}, U, \mathit{sec},  T, V, c \rangle$, and $\phi$ be $\neg R(\overline{t})$.
From the rule's definition, $\mathit{secEx}(s) = \bot$.
From this and the LTS rules, it follows that $f(\mathit{last}(r^{i-1}), \langle u, \\ \mathtt{INSERT}, R, \overline{t} \rangle) = \top$.
From this and $f$'s definition, it follows that $f_{\mathit{conf}}^{u}(\mathit{last}(r^{i-1}),  \langle u, \mathtt{INSERT}, R, \overline{t} \rangle) = \top$ because $\mathit{user}(\mathit{last}(r^{i-1}),  \langle u, \mathtt{INSERT}, R, \overline{t} \rangle) = u$.
From this and $f_{\mathit{conf}}^{u}$'s definition, it follows $\mathit{secure}(u,\phi,\mathit{last}(r^{i-1})) \\ = \top$ because $\phi = \mathit{getInfo}(\langle u,  \mathtt{INSERT}, R, \overline{t}\rangle)$.
From this and \thref{theorem:secure:sound:under:approximation}, it follows that $\mathit{secure}_{P,u}(r,i-1 \attMod \phi)$ holds.

\item \emph{Learn \texttt{DELETE} Backward - 3}.
The proof for this case is similar to that of \emph{Learn \texttt{INSERT} Backward - 3}.

\item \emph{Propagate Forward Disabled Trigger}.
Let $i$ be such that $r^{i} = r^{i-1} \concat t \concat s$, where $s = \langle \mathit{db}, U, \mathit{sec}, T, V, c \rangle  \in \Omega_{M}$, $\mathit{last}(r^{i-1}) =  \langle \mathit{db}, U, \mathit{sec},  T, V, c \rangle$, and $t$ be a trigger.
Furthermore, let $\psi$ be $t$'s condition where all free variables are replaced with $\mathit{tpl}(\mathit{last}(r^{i-1}))$.
From the rule, it follows that $r,i-1 \attMod \phi$.
From this and the induction hypothesis, it follows that $\mathit{secure}_{P,u} (r, i-1 \attMod \phi)$ holds.
Furthermore, from \thref{theorem:f:composition:pec:2}, it follows that $t$ preserves the equivalence class with respect to $r^{i-1}$, $P$, and $u$.	
If the trigger's action is an \texttt{INSERT} or a \texttt{DELETE} operation, we claim that the operation does not change the content of any table in $\mathit{tables}(\phi)$ for any run $v \in \llbracket r^{i-1} \rrbracket_{P,u}$.
From this, the fact that $t$ preserves the equivalence class with respect to $r^{i-1}$, $P$, and $u$, \thref{theorem:secure:extend:on:runs:triggers}, and $\mathit{secure}_{P,u} (r, i-1 \attMod \phi)$, it follows that also $\mathit{secure}_{P,u} (r, i \attMod \phi)$ holds.

We now prove our claim.
Assume that $t$'s action in either an  \texttt{INSERT} or a \texttt{DELETE} operation.
From the rule, it follows that $r,i-1 \attMod \neg \psi$.
From this and \thref{theorem:attacker:model:sound}, $[\psi]^{\mathit{last}(r^{i-1})} = \bot$.
From $r,i-1 \attMod \neg \psi$ and the induction hypothesis, it follows that $\mathit{secure}_{P,u}  (r, i-1 \attMod \psi)$ holds.
From this and $[\psi]^{\mathit{last}(r^{i-1}).\mathit{db}} = \bot$, it follows that $[\psi]^{v.\mathit{db}} = \bot$ for any run $v \in \llbracket r^{i-1} \rrbracket_{P,u}$.
Therefore, the trigger $t$ is disabled in any run $v \in \llbracket r^{i-1} \rrbracket_{P,u}$.
From this and the LTS semantics, it follows that $t$'s execution does not change the content of any  table in $\mathit{tables}(\phi)$ for any run $v \in \llbracket r^{i-1} \rrbracket_{P,u}$.

\item \emph{Propagate Backward Disabled Trigger}.
The proof for this case is similar to that of  \emph{Propagate Forward Disabled Trigger}.

\item \emph{Learn \texttt{INSERT} Forward}.
Let $i$ be such that $r^{i} = r^{i-1} \concat t \concat s$, where $s = \langle \mathit{db}, U, \mathit{sec}, T, V, c \rangle  \in \Omega_{M}$, $\mathit{last}(r^{i-1}) =  \langle \mathit{db}, U, \mathit{sec},  T, V, c \rangle$, and $t$ be a trigger, and $\phi$ be $R(\overline{t})$.
Furthermore, let $\psi$ be $t$'s condition where all free variables are replaced with $\mathit{tpl}(\mathit{last}(r^{i-1}))$.
From the rule's definition, it follows  that $t$'s action is  $\langle u',\mathtt{INSERT}, R, \overline{t}\rangle$ and that $r, i-1 \attMod \psi$ holds.
From \thref{theorem:attacker:model:sound} and $r, i-1 \attMod \psi$, it follows that $[\psi]^{\mathit{last}(r^{i-1}).\mathit{db}} = \top$.
From this, $\mathit{secEx}(s) = \bot$, and $\mathit{Ex}(s) = \emptyset$, it follows that $t$'s action has been executed successfully.
From this, it follows that $\overline{t} \in s.\mathit{db}(R)$.
From  $r, i-1 \attMod \psi$ and the induction hypothesis, it follows that $\mathit{secure}_{P,u} (r, i-1 \attMod \psi)$.
From this and $[\psi]^{\mathit{last}(r^{i-1}).\mathit{db}} = \top$, it follows that $[\psi]^{\mathit{last}(v).\mathit{db}} = \top$ for any $v \in \llbracket r^{i-1}\rrbracket_{P,u}$.
From this, it follows that the trigger $t$ is enabled in any run $v \in \llbracket r^{i-1}\rrbracket_{P,u}$.
From \thref{theorem:f:composition:pec:2}, it follows that $t$ preserves the equivalence class with respect to $r^{i-1}$, $P$, and $u$.
From this, $\mathit{secEx}(s) = \bot$, $\mathit{Ex}(s) = \emptyset$, and the fact that the trigger $t$ is enabled in any run $v \in \llbracket r^{i-1}\rrbracket_{P,u}$, it follows that $t$'s action is executed successfully in any run $e(v,t)$, where $v \in \llbracket r^{i-1}\rrbracket_{P,u}$.
From this,  it follows that $\mathit{db}''(R)$, where $\mathit{db}'' = \overline{t} \in \mathit{last}(e(v,t)).\mathit{db}$, for any $v \in \llbracket r^{i-1}\rrbracket_{P,u}$.
Therefore, $\mathit{secure}_{P,u} (r, i \attMod \phi)$ holds.

\item \emph{Learn \texttt{INSERT} - FD}.
Let $i$ be such that $r^{i} = r^{i-1} \concat t \concat s$, where $s = \langle \mathit{db}, U, \mathit{sec}, T, V, c \rangle \in \Omega_{M}$, $\mathit{last}(r^{i-1}) =  \langle \mathit{db}', U', \mathit{sec}', T', V', c' \rangle$, and $t \in {\cal TRIGGER}_{D}$, and $\phi$ be $\neg \exists \overline{y},\overline{z}.\, R(\overline{v}, \overline{y}, \overline{z}) \wedge \overline{y} \neq \overline{w}$.
Furthermore,  let $\psi$ be $t$'s condition where all free variables are replaced with the values in $\mathit{tpl}(\mathit{last}(r^{i-1}))$ and $\langle u', \mathtt{INSERT}, R, (\overline{v}, \overline{w}, \overline{q}) \rangle$ be $t$'s actual action.
From the rule, it follows that $r, i-1 \attMod \psi$.
From this and \thref{theorem:attacker:model:sound}, it follows that $[\psi]^{\mathit{last}(r^{i-1}).\mathit{db}} = \top$.
From this, $\mathit{Ex}(s) = \emptyset$, and  $\mathit{secEx}(s) \\ = \bot$, it follows that $f(s',\langle u',\texttt{INSERT},R, \overline{t}\rangle)  = \top$, where $s'$ is the state just after the execution of the \texttt{SELECT} statement associated with $t$'s \texttt{WHEN} clause.
From this and $f$'s definition, it follows that $f_{\mathit{conf}}^{u} (s',\langle u', \texttt{INSERT},R,\\  \overline{t}\rangle) = \top$ because $\mathit{user}(s', \langle u',\texttt{INSERT},R, \overline{t}\rangle) = u$ since $u$ is $t$'s invoker.
From this and $f_{\mathit{conf}}^{u}$'s definition, it follows that $\mathit{secure}(u, \phi, s') = \top$.
From this, $\mathit{pState}(s')  = \mathit{pState}(\mathit{last}(r^{i-1}))$, and \thref{theorem:secure:equivalent:modulo:indistinguishable:state}, it follows $\mathit{secure}(u, \\ \phi, \mathit{last}(r^{i-1})) = \top$.
From this and \thref{theorem:secure:sound:under:approximation}, it follows $\mathit{secure}_{P,u} (r, i-1 \attMod \phi)$.
We claim that $\mathit{secure}^{\mathit{data}}_{P,u} (r, \\ i \attMod \phi)$ holds.
From this and \thref{theorem:ibsec:correctness:secure:3}, it follows that also  $\mathit{secure}_{P,u}(r,i  \attMod \phi)$ holds.

We now prove our claim that  $\mathit{secure}^{\mathit{data}}_{P,u}(r,i \attMod \phi)$ holds.
Let $s'$ be the state just after the execution of the \texttt{SELECT} statement associated with $t$'s \texttt{WHEN} clause and $s''$ be the state $\mathit{last}(r^{i-1})$.
Furthermore, for brevity's sake, in the following we omit the $\mathit{pState}$ function where needed.
For instance, with a slight abuse of notation, we write $\llbracket s' \rrbracket^{\mathit{data}}_{u,M}$ instead of $\llbracket \mathit{pState}(s') \rrbracket^{\mathit{data}}_{u,M}$.
From $\mathit{secure}(u, \phi,s') = \top$, $s' \cong_{M,u}^{\mathit{data}} s''$, \thref{theorem:secure:equivalent:modulo:indistinguishable:state}, and \thref{theorem:secure:sound:under:approximation}, it follows that $\mathit{secure}^{\mathit{data}}_{P,u}(r,i-1 \attMod \phi)$ holds.
From this, it follows that $[\phi]^{v} = [\phi]^{s''}$ for any $v \in \llbracket s'' \rrbracket^{\mathit{data}}_{u,M}$.
Furthermore, from \thref{theorem:getInfo:sound:and:complete} and $\mathit{Ex}(s) = \emptyset$, it follows that $\phi$ holds in $s''$.
Let $A_{s'',R,\overline{t}}$ be the set $\{\langle \mathit{db}[R \oplus \overline{t}], U, \mathit{sec}, T,V\rangle \in \Pi_{M} \,|\, \exists \mathit{db}' \in \Omega_{D}.\, \langle \mathit{db}', \\  U,  \mathit{sec},  T, V\rangle \in \llbracket s''\rrbracket^{\mathit{data}}_{M,u}\}$.
It is easy to see that $\llbracket s\rrbracket^{\mathit{data}}_{M,u} \subseteq A_{s'',R,\overline{t}}$.
We now show that $\phi$ holds for any $z \in A_{s'',R,\overline{t}}$.
Let $z_{1} \in \llbracket s'' \rrbracket^{\mathit{data}}_{M,u}$.
From $[\phi]^{v} = [\phi]^{s''}$ for any $v \in \llbracket s'' \rrbracket^{\mathit{data}}_{u,M}$ and the fact that $\phi$ holds in $s''$, it follows that $[\phi]^{z_{1}} = \top$.
Therefore, for any $(\overline{k}_{1},\overline{k}_{2}, \overline{k}_{3}) \in R(z_{1})$ such that $|\overline{k}_1| = |\overline{v}|$, $|\overline{k}_2| = |\overline{w}|$, and $|\overline{k}_3| = |\overline{q}|$, if $k_{1} = \overline{v}$, then $k_{2} = \overline{w}$. 
Then, for any $(\overline{k}_{1},\overline{k}_{2}, \overline{k}_{3}) \in R(z_{1}) \cup \{(\overline{v}, \overline{w}, \overline{q})\}$ such that $|\overline{k}_1| = |\overline{v}|$, $|\overline{k}_2| = |\overline{w}|$, and $|\overline{k}_3| = |\overline{q}|$, if $k_{1} = \overline{v}$, then $k_{2} = \overline{w}$.
Therefore, $\phi$ holds also in $z_{1}[R \oplus \overline{t}] \in A_{\mathit{pState}(s''), R, \overline{t}}$.	
Hence, $[\phi]^{z} = \top$ for any $z \in A_{s'',R,\overline{t}}$.
From this and $\llbracket s\rrbracket^{\mathit{data}}_{M,u} \subseteq A_{s'',R,\overline{t}}$, it follows that  $[\phi]^{z} = \top$ for any $z \in \llbracket s\rrbracket^{\mathit{data}}_{M,u}$.
From this, it follows that $\mathit{secure}^{\mathit{data}}_{P,u}(r,i \attMod \phi)$ holds.

\item \emph{Learn \texttt{INSERT} - FD - 1}.
The proof of this case is similar to that of \emph{Learn \texttt{INSERT} - FD}.

\item \emph{Learn \texttt{INSERT} - ID}.
The proof of this case is similar to that of \emph{Learn \texttt{INSERT} - FD}.
See also the proof of \emph{\texttt{INSERT} Success - ID}.	

\item \emph{Learn \texttt{INSERT} - ID - 1}.
The proof of this case is similar to that of \emph{Learn \texttt{INSERT} - ID}.

\item \emph{Learn \texttt{INSERT} Backward - 1}.
Let $i$ be such that $r^{i} = r^{i-1} \concat t \concat s$, where $s = \langle \mathit{db}', U', \mathit{sec}', T', V', c' \rangle \in \Omega_{M}$, $\mathit{last}(r^{i-1}) =  \langle \mathit{db}, U, \mathit{sec}, T, V, c \rangle$, and $t \in {\cal TRIGGER}_{D}$, and $\phi$ be $t$'s actual \texttt{WHEN} condition, where all free variables are replaced with the values in $\mathit{tpl}(\mathit{last}(r^{i-1}))$.
From the rule's definition, it follows that $\mathit{secEx}(s) = \top$.
From this, the LTS semantics, and $\mathit{secEx}(s) = \top$, it follows that $f(\mathit{last}(r^{i-1}), \langle u', \mathtt{SELECT}, \phi \rangle) = \top$.
From this and $f$'s definition, it follows $f^{u}_{\mathit{conf}}(\mathit{last}(r^{i-1}), \langle u', \mathtt{SELECT}, \\ \phi \rangle) = \top$ because $\mathit{user}(\mathit{last}(r^{i-1}), \langle u',  \mathtt{SELECT}, \phi \rangle) = u$ since $u$ is $t$'s invoker.
From this and $f^{u}_{\mathit{conf}}$'s definition, it follows that $\mathit{secure}(u,\phi,\mathit{last}(r^{i-1})) = \top$.
From this and \thref{theorem:secure:sound:under:approximation}, it follows that also $\mathit{secure}_{P,u}(r,i - 1 \attMod \phi)$ holds.

\item \emph{Learn \texttt{INSERT} Backward - 2}.
Let $i$ be such that $r^{i} = r^{i-1} \concat t \concat s$, where $s = \langle \mathit{db}', U', \mathit{sec}', T', V', c' \rangle \in \Omega_{M}$, $\mathit{last}(r^{i-1}) =  \langle \mathit{db}, U, \mathit{sec}, T, V, c \rangle$, and $t \in {\cal TRIGGER}_{D}$, and $\phi$ be $\neg R(\overline{t})$.
Furthermore, let $\mathit{act}=\langle u', \mathtt{INSERT}, R, \\ \overline{t} \rangle$ be $t$'s actual action and $\gamma$ be $t$'s actual \texttt{WHEN} condition obtained by replacing all free variables with the values in $\mathit{tpl}(\mathit{last}(r^{i-1}))$.
From the rule's definition, it follows $\mathit{secEx}(s) = \top$ and there is a $\psi$ such that $r, i-1 \attMod \psi$ and $r, i \attMod \neg \psi$.
We claim that $[\gamma]^{\mathit{db}} = \top$.
From this and  $\mathit{secEx}(s) = \top$, it follows that $f(s', \langle u', \mathtt{INSERT}, R, \overline{t}\rangle) = \top$, where $s'$ is the state obtained after the evaluation of $t$'s \texttt{WHEN} condition.
From this and $f$'s definition, it follows $f_{\mathit{conf}}^{u}(s', \langle u', \mathtt{INSERT}, R, \overline{t}\rangle) = \top$ as $\mathit{user}(s', \langle u', \\ \mathtt{INSERT}, R, \overline{t}\rangle) = u$ because $u$ is $t$'s invoker.
From this and $f^{u}_{\mathit{conf}}$'s definition, it follows $\mathit{secure}(u,\phi,s') = \top$ since $\phi$ is equivalent to $\mathit{getInfo}(\langle u', \mathtt{INSERT}, R, \overline{t}\rangle)$.
From this, \thref{theorem:secure:equivalent:modulo:indistinguishable:state}, and $\mathit{pState}(s') = \mathit{pState}(\mathit{last}(r^{i-1}))$,  it follows  $\mathit{secure}(u, \phi,\mathit{last}(r^{i-1})) = \top$.
From this and \thref{theorem:secure:sound:under:approximation}, it follows $\mathit{secure}_{P,u}(r,i - 1 \attMod \phi)$.

We now prove our claim that $[\gamma]^{\mathit{db}} = \top$.
Assume, for contradiction's sake, that this is not the case.
From this and the LTS rules, it follows that $\mathit{db} = \mathit{db}'$.
From the rule's definition, it follows that there is a $\psi$ such that $r, i-1 \attMod \psi$ and $r, i \attMod \neg \psi$.
From this, \thref{theorem:attacker:model:sound}, $s = \langle \mathit{db}', U', \mathit{sec}', T', V', c' \rangle$, and $\mathit{last}(r^{i-1}) =  \langle \mathit{db}, U, \mathit{sec}, T, \\ V, c \rangle$, it follows that $[\psi]^{\mathit{db}} = \top$ and $[\neg \psi]^{\mathit{db}'} = \top$.
Therefore, $[\psi]^{\mathit{db}} = \top$ and $[\psi]^{\mathit{db}'} = \bot$.
Hence, $\mathit{db} \neq \mathit{db}'$, which contradicts $\mathit{db} = \mathit{db}'$.

	\item \emph{Learn \texttt{DELETE} Forward}.
The proof of this case is similar to that of \emph{Learn \texttt{INSERT} Forward}.

\item \emph{Learn \texttt{DELETE} - ID}.
The proof of this case is similar to that of \emph{Learn \texttt{INSERT} - FD}.
See also the proof of \emph{\texttt{DELETE} Success - ID}.

\item \emph{Learn \texttt{DELETE} - ID - 1}.
The proof of this case is similar to that of \emph{Learn \texttt{DELETE} - ID}.

\item \emph{Learn \texttt{DELETE} Backward - 1}.
The proof of this case is similar to that of \emph{Learn \texttt{INSERT} Backward - 1}.

\item \emph{Learn \texttt{DELETE} Backward - 2}.
The proof of this case is similar to that of \emph{Learn \texttt{INSERT} Backward - 2}.

\item \emph{Propagate Forward Trigger Action}.
Let $i$ be such that $r^{i} = r^{i-1} \concat t \concat s$, where $t$ is a trigger, $s = \langle \mathit{db}, U, \mathit{sec}, T, V, c \rangle \\ \in \Omega_{M}$ and $\mathit{last}(r^{i-1}) =  \langle \mathit{db}', U', \mathit{sec}',  T', V', c' \rangle$.
From the rule's definition, $r, i-1 \attMod \phi$ holds.
From this and the induction hypothesis, it follows that $\mathit{secure}_{P,u}  (r,i-1 \attMod \phi)$ holds.
From \thref{theorem:f:composition:pec:2}, the trigger $t$ preserves the equivalence class with respect to $r^{i-1}$, $P$, and $u$. 	
We claim that the execution of $t$ does not alter the content of the tables in $\mathit{tables}(\phi)$.
From this, \thref{theorem:secure:extend:on:runs:insert:delete}, and $\mathit{secure}_{P,u}(r,i-1 \attMod \phi)$, it follows  $\mathit{secure}_{P,u} (r,i \attMod \phi)$.

We now prove our claim that the execution of $t$ does not alter the content of the tables in $\mathit{tables}(\phi)$.
If the trigger is not enabled, the claim is trivial.
In the following, we assume the trigger is enabled.
There are four cases:
\begin{compactitem}
\item $t$'s action is an \texttt{INSERT} statement.
This case amount to claiming that the \texttt{INSERT} statement $\langle u',\mathtt{INSERT}, \\ R,  \overline{t} \rangle$ does not alter the content of the tables in $\mathit{tables}(\phi)$ in case $\mathit{reviseBelif}(r^{i-1}, \phi, r^{i}) = \top$.
We proved the claim above in the \emph{Propagate Forward \texttt{INSERT/DELETE} Success} case.

\item $t$'s action is an \texttt{DELETE} statement.
The proof is similar to that of the \texttt{INSERT} case.

\item $t$'s action is an \texttt{GRANT} statement.
In this case, the action does not alter the database state and the claim follows trivially.

\item $t$'s action is an \texttt{REVOKE} statement.
The proof is similar to that of the \texttt{GRANT} case.

\end{compactitem}

\item \emph{Propagate Backward Trigger Action}.
The proof of this case is similar to \emph{Propagate Backward Trigger Action}.

\item \emph{Propagate Forward \texttt{INSERT} Trigger Action}.
Let $i$ be such that $r^{i} = r^{i-1} \concat t \concat s$, where $t$ is a trigger, $s = \langle \mathit{db}, U, \mathit{sec}, T, V, c \rangle  \in \Omega_{M}$ and $\mathit{last}(r^{i-1}) =  \langle \mathit{db}', U', \mathit{sec}', \\  T', V', c' \rangle$.
From the rule's definition, $r, i-1 \attMod \phi$ holds.
From this and the induction hypothesis, it follows that $\mathit{secure}_{P,u}  (r,i-1 \attMod \phi)$ holds.
From \thref{theorem:f:composition:pec:2}, the trigger $t$ preserves the equivalence class with respect to $r^{i-1}$, $P$, and $u$. 	
We claim that the execution of $t$ does not alter the content of the tables in $\mathit{tables}(\phi)$.
From this, \thref{theorem:secure:extend:on:runs:insert:delete}, and $\mathit{secure}_{P,u}  (r,i-1 \attMod \phi)$, it follows $\mathit{secure}_{P,u}  (r,i \attMod \phi)$.

We now prove our claim that the execution of $t$ does not alter the content of the tables in $\mathit{tables}(\phi)$.
If the trigger is not enabled, the claim is trivial.
In the following, we assume the trigger is enabled.
Then, $t$'s action is an \texttt{INSERT} statement.
This case amount to claiming that the \texttt{INSERT} statement $\langle u',\mathtt{INSERT}, R, \overline{t} \rangle$ does not alter the content of the tables in $\mathit{tables}(\phi)$ in case $r,i-1 \attMod R(\overline{t})$ holds.
We proved the claim above in the \emph{Propagate Forward \texttt{INSERT} Success - 1} case.

\item \emph{Propagate Forward \texttt{DELETE} Trigger Action}.
The proof of this case is similar to that of  \emph{Propagate Forward \texttt{INSERT} Trigger Action}.

\item \emph{Propagate Backward \texttt{INSERT} Trigger Action}.
The proof of this case is similar to that of  \emph{Propagate Forward \texttt{INSERT} Trigger Action}.

\item \emph{Propagate Backward \texttt{DELETE} Trigger Action}.
The proof of this case is similar to that of  \emph{Propagate Forward \texttt{INSERT} Trigger Action}.

\item \emph{Trigger FD \texttt{INSERT} Disabled Backward}.
Let $i$ be such that $r^{i} = r^{i-1} \concat t \concat s$, where $s = \langle \mathit{db}', U', \mathit{sec}', T', V', c' \rangle \in \Omega_{M}$, $t \in {\cal TRIGGER}_{D}$, $\mathit{last}(r^{i-1}) =  \langle \mathit{db}, U, \mathit{sec}, T, V, c \rangle$, and $\phi$ be $t$'s actual \texttt{WHEN} condition obtained by replacing all free variables with the values in $\mathit{tpl}(\mathit{last}(r^{i-1}))$.
Furthermore, let $\mathit{act}=\langle u', \mathtt{INSERT}, R, (\overline{v}, \overline{w}, \overline{q}) \rangle$ be $t$'s actual action and $\alpha$ be $\exists \overline{y},\overline{z}. R(\overline{v},  \overline{y}, \overline{z}) \wedge \overline{y} \neq \overline{w}$.
From the rule's definition, it follows that  $\mathit{secEx}(s) = \bot$.
From this, it follows that $f(\mathit{last}(r^{i-1}), \langle u', \mathtt{SELECT}, \phi \rangle)  = \top$.
From this and $f$'s definition, it follows $f^{u}_{\mathit{conf}}(\mathit{last}(r^{i-1}), \\ \langle u', \mathtt{SELECT}, \phi \rangle)  = \top$ since $\mathit{user}(\mathit{last}(r^{i-1}), \langle u', \mathtt{SELECT}, \\ \phi \rangle)  = u$ since $u$ is $t$'s invoker.
From this and $f^{u}_{\mathit{conf}}$'s definition, it follows that $\mathit{secure}(u, \neg \phi, \mathit{last}(r^{i-1}))  = \top$.
From this, it follows that  $\mathit{secure}(u,  \phi, \mathit{last}(r^{i-1})) = \top$.
From this and \thref{theorem:secure:sound:under:approximation}, it follows that also $\mathit{secure}_{P,u}  (r,i-1 \attMod \phi)$.

\item \emph{Trigger ID \texttt{INSERT} Disabled Backward}.
The proof of this case is similar to that of \emph{Trigger FD \texttt{INSERT} Disabled Backward}.

\item \emph{Trigger ID \texttt{DELETE} Disabled Backward}.
The proof of this case is similar to that of \emph{Trigger FD \texttt{INSERT} Disabled Backward}.

	\end{compactenum}
	This completes the proof of the induction step.
	
This completes the proof.
\end{proof}

In \thref{theorem:f:composition:pec:1} and \thref{theorem:f:composition:pec:2}, we show that actions and triggers preserve the equivalence class for any LTS that uses $f$ as \acf{}.

\begin{lemma}\thlabel{theorem:f:composition:pec:1}
Let $u$ be a user in ${\cal U}$, $P = \langle M, f \rangle$ be an \accessControlConfiguration{}, where $M = \langle D,\Gamma\rangle$ is a system configuration and $f$ is as above, $L$ be the $P$-LTS.
For any run $r \in \mathit{traces}(L)$ and any action $a \in {\cal A}_{D,u}$,
if $\mathit{extend}(r,a)$ is defined, then $a$ preserves the equivalence class for $r$, $P$, and $u$.
\end{lemma}

\begin{proof}

Let $u$ be a user in ${\cal U}$, $P = \langle M, f \rangle$ be an \accessControlConfiguration{}, where $M = \langle D,\Gamma\rangle$ is a system configuration and $f$ is as above, and $L$ be the $P$-LTS.
In the following, we use $e$ to refer to the $\mathit{extend}$ function.
We prove our claim by contradiction.
Assume, for contradiction's sake, that there is a run $r  \in \mathit{traces}(L)$ and an action $a \in {\cal A}_{D,u}$ such that $\mathit{e}(r,a)$ is defined and $a$ does not preserve the equivalence class for $r$, $P$, and $u$.
According to the LTS semantics, the fact that  $\mathit{e}(r,a)$ is defined implies that $\mathit{triggers}(\mathit{last}(r)) = \epsilon$.
Therefore,  $\mathit{triggers}(\mathit{last}(r')) = \epsilon$ holds as well for any for any $r' \in \llbracket r \rrbracket_{P,u}$ (because $r$ and $r'$ are indistinguishable and, therefore, their projections are consistent), and, thus, $\mathit{e}(r',a)$ is defined as well for any $r' \in \llbracket r \rrbracket_{P,u}$.
There are a number of cases depending on $a$:
\begin{compactenum}
\item $a = \langle u, \mathtt{SELECT}, q \rangle$.
There are two cases:
\begin{compactenum}
\item $\mathit{secEx}(\mathit{last}(e(r,a))) = \bot$.
From the LTS rules and  $\mathit{secEx}( \mathit{last}(e(r,a))) = \bot$, it follows that $f(\mathit{last}(r),a) \\ = \top$.
From this and \thref{theorem:f:composition:soundness}, it follows that $f(\mathit{last}(r'), a) = \top$ for any $r' \in \llbracket r \rrbracket_{P,u}$.
From this and the LTS rules, it follows $\mathit{secEx}(\mathit{last}(e(r',a))) = \bot$ for any $r' \in \llbracket r \rrbracket_{P,u}$.
From $f(\mathit{last}(r'),a) = \top$ for any $r' \in \llbracket r \rrbracket_{P,u}$, it follows $f_{\mathit{conf}}^{\mathit{user}(\mathit{last} (r'),a)} (\mathit{last}(r'),a) \\ = \top$ for any $r' \in \llbracket r \rrbracket_{P,u}$.
Note that $\mathit{user}(\mathit{last}(r'),a) \\ = u$ for any $r' \in \llbracket r \rrbracket_{P,u}$ because $\mathit{trigger}(\mathit{last}(r')) = \epsilon$ and $u \in {\cal A}_{D,u}$.
From this, $f_{\mathit{conf}}^{\mathit{user}(\mathit{last}(r'),a)}(\mathit{last}(r'),\\ a)  = \top$ for any $r' \in \llbracket r \rrbracket_{P,u}$, and $f_{\mathit{conf}}^{u}$'s definition, it follows that $\mathit{secure}(u,q,  \mathit{last}(r'))  = \top$ for any $r' \in \llbracket r \rrbracket_{P,u}$.
From this and \thref{theorem:secure:sound:under:approximation}, it follows that $[q]^{\mathit{last}(r').\mathit{db}} = [q]^{\mathit{last}(r).\mathit{db}}$ for all $r' \in \llbracket r \rrbracket_{P,u}$.
Furthermore, it follows trivially from the LTS rule \emph{\texttt{SELECT} Success}, that the state after $a$'s execution is data indistinguishable from $\mathit{last}(r)$.
It is also easy to see that $e(r',a)$ is well-defined for any $r' \in \llbracket r \rrbracket_{P,u}$.
From the considerations above and $r' \in \llbracket r \rrbracket_{P,u}$, it follows trivially that $e(r',a) \in  \llbracket e(r,a) \rrbracket_{P,u}$. The bijection $b$ is trivially $b(r') = e(r',a)$. 
This leads to a contradiction.

\item $\mathit{secEx}(\mathit{last}(e(r,a))) = \top$.
From the LTS rules and  $\mathit{secEx}(\mathit{last}(e(r,a))) = \top$, it follows that $f(\mathit{last}(r),a) \\ = \bot$.
From this and \thref{theorem:f:composition:soundness}, it follows that $f(\mathit{last}(r'), a) = \bot$ for any $r' \in \llbracket r \rrbracket_{P,u}$.
From this and the LTS rules, it follows $\mathit{secEx}(\mathit{last}(e(r',a))) = \top$ for any $r' \in \llbracket r \rrbracket_{P,u}$.
The data indistinguishability between $\mathit{last}(e(r',a))$ and $\mathit{last}(e(r,a))$ follows trivially from the data indistinguishability between $\mathit{last}(r')$ and $\mathit{last}(r)$.
Therefore, for any run $r' \in \llbracket r \rrbracket_{P,C}$, there is exactly one run $e(r',a)$. 
From the considerations above, it follows trivially that $e(r',a) \\ \in  \llbracket e(r,a) \rrbracket_{P,u}$.
The bijection $b$ is trivially $b(r') = e(r',a)$. This leads to a contradiction.

\end{compactenum}
Both cases leads to a contradiction. 
This completes the proof for $a = \langle u, \mathtt{SELECT}, q \rangle$.

\item $a = \langle u, \mathtt{INSERT}, R, \overline{t} \rangle$. 
In the following, we denote by $\mathit{gI}$ the function $\mathit{getInfo}$, by $\mathit{gS}$ the function $\mathit{getInfoS}$, and by $\mathit{gV}$ the function $\mathit{getInfoV}$.
There are three cases:

\begin{compactenum}
\item $\mathit{secEx}(\mathit{last}(e(r,a))) = \bot$ and $\mathit{Ex}(\mathit{last}(e(r,a))) = \emptyset$. 
From the LTS rules and  $\mathit{secEx}(\mathit{last}(e(r,a))) = \bot$, it follows that $f(\mathit{last}(r),a) = \top$.
From this and \thref{theorem:f:composition:soundness}, it follows that $f(\mathit{last}(r'),a) = \top$ for any $r' \in \llbracket r \rrbracket_{P,u}$.
From this and the LTS rules, it follows that $\mathit{secEx}(\mathit{last}(e(r',a))) = \bot$ for any $r' \in \llbracket r \rrbracket_{P,u}$.
From $f(\mathit{last}(r),a) = \top$, it follows that $f_{\mathit{conf}}^{u}(\mathit{last}(r),a) = \top$ because $\mathit{user}(\mathit{last}(r),a) = u$ since $\mathit{trigger}(\mathit{last}(r), a) = \epsilon$ and $a \in {\cal A}_{D,u}$.
From this and $f_{\mathit{conf}}^{u}$'s definition, it follows that $\mathit{secure}(u, \\ \mathit{gS}(\gamma, \mathit{act}), \mathit{last}(r))$ holds for any integrity constraint $\gamma$ in $\mathit{Dep}(\Gamma,a)$.
From $\mathit{Ex}(\mathit{last}(e(r,a))) = \emptyset$ and \thref{theorem:getInfo:sound:and:complete}, it follows $[\mathit{gS}(\gamma, \mathit{act})]^{\mathit{last}(r).\mathit{db}} = \top$.
From this,  $\mathit{secure}(u,  \mathit{gS}(\gamma, \mathit{act}), \mathit{last}(r))$, and \thref{theorem:secure:sound:under:approximation}, it follows that $[\mathit{gS}(\gamma, \mathit{act})]^{\mathit{last}(r').\mathit{db}} = \top$ for any $r'  \in \llbracket r\rrbracket_{P,u}$. 
From this and \thref{theorem:getInfo:sound:and:complete}, it follows that $\mathit{Ex}(\mathit{last}(e(r',a))) = \emptyset$ for any $r'  \in \llbracket r\rrbracket_{P,u}$. 
We  claim that, for any $r'  \in \llbracket r\rrbracket_{P,u}$, $\mathit{last}(e(r,a))$ and $\mathit{last}(e(r', a))$ are data indistinguishable.
From this and the above considerations, it follows trivially that $e(r',a) \in  \llbracket e(r,a) \rrbracket_{P,u}$.
The bijection $b$ is trivially $b(r') = e(r',a)$.
This leads to a contradiction.

We now prove our claim that for any $r'  \in \llbracket r\rrbracket_{P,u}$, $\mathit{last}(e(r,a))$ and $\mathit{last}(e(r',a))$ are data indistinguishable.
We prove the claim by contradiction.
Let $s_{2} = \langle \mathit{db}_{2}, U_{2},\mathit{sec}_{2}, T_{2},V_{2} \rangle$ be $\mathit{pState}(\mathit{last}(e(r, a)))$, $s_{2}' = \langle \mathit{db}_{2}', U_{2}',\mathit{sec}_{2}', T_{2}',V_{2}' \rangle$ be $\mathit{pState}(\mathit{last}(e  (r',a)))$, $s_{1} = \langle \mathit{db}_{1}, U_{1},\mathit{sec}_{1}, T_{1},V_{1} \rangle$ be $\mathit{pState}(\mathit{last}  (r))$, and $s_{1}' = \langle \mathit{db}_{1}', U_{1}',\mathit{sec}_{1}', T_{1}',V_{1}'\rangle$ be $\mathit{pState}(\mathit{last}  (r'))$.
In the following, we denote  the $\mathit{permissions}$ function by $p$.
Furthermore, note that $s_{1}$ and $s_{1}'$ are data-indistinguishable because $r'  \in \llbracket r\rrbracket_{P,u}$.
There are a number of cases:
\begin{compactenum}

\item $U_{2} \neq U_{2}'$. 
Since $a$ is an \texttt{INSERT} operation, it follows that $U_{1} = U_{2}$ and $U_{1}' = U_{2}'$.
Furthermore, from $s_{1} \cong_{M,u}^{\mathit{data}} s_{1}'$, it follows that $U_{1} = U_{1}'$.
Therefore, $U_{2} = U_{2}'$ leading to a contradiction.

\item $\mathit{sec}_{2} \neq \mathit{sec}_{2}'$. 
The proof is similar to the case $U_{2} \neq U_{2}'$. 

\item $\mathit{T}_{2} \neq \mathit{T}_{2}'$. 
The proof is similar to the case $U_{2} \neq U_{2}'$. 

\item $\mathit{V}_{2} \neq \mathit{V}_{2}'$. 
The proof is similar to the case $U_{2} \neq U_{2}'$.

\item there is a table $R'$ for which $\langle \oplus, \mathtt{SELECT},  R\rangle \in p(s_{2},u)$ and $\mathit{db}_{2}(R') \neq \mathit{db}_{2}'(R')$.
Note that $p(s_{2},u) = p(s_{1},u)$.
There are two cases:
\begin{compactitem}
\item $R = R'$.
From $s_{1} \cong_{M,u}^{\mathit{data}} s_{1}'$ and $\langle \oplus, \mathtt{SELECT}, R\rangle \\ \in p(s_{2},u) $, it follows that $\mathit{db}_{1}(R') = \mathit{db}_{1}'(R')$.
From this and the fact that $a$ has been executed successfully both in $e(r,a)$ and $e(r',a)$, it follows that $\mathit{db}_{2}(R') = \mathit{db}_{1}  (R') \cup \{\overline{t}\}$ and  $\mathit{db}_{2}'(R') = \mathit{db}_{1}'(R') \cup \{\overline{t}\}$.
From this and $\mathit{db}_{1}(R') = \mathit{db}_{1}'(R')$, it follows that $\mathit{db}_{2}(R') = \mathit{db}_{2}'(R')$ leading to a contradiction.

\item $R \neq R'$.
From $s_{1} \cong_{M,u}^{\mathit{data}} s_{1}'$ and $\langle \oplus, \mathtt{SELECT}, R\rangle \\ \in p(s_{2},u)$, it follows that $\mathit{db}_{1}(R') = \mathit{db}_{1}'(R')$.
From this and the fact that $a$ does not modify $R'$, it follows that $\mathit{db}_{1}(R') = \mathit{db}_{2}(R')$ and  $\mathit{db}_{1}'(R') = \mathit{db}_{2}'(R')$.
From this and $\mathit{db}_{1}(R') = \mathit{db}_{1}'(R')$, it follows that $\mathit{db}_{2}(R') = \mathit{db}_{2}'(R')$ leading to a contradiction.
\end{compactitem}

\item there is a view $v$ for which $\langle \oplus, \mathtt{SELECT},  v\rangle \in p(s_{2},  u)$ and $\mathit{db}_{2}(v) \neq \mathit{db}_{2}'(v)$.
Note that $p(s_{2},  u) \\ = p(s_{1},u)$.
Since $a$ has been successfully executed in both states, we know that $\mathit{leak}(s_{1},  a, u)$ hold.
There are two cases:
\begin{compactitem}
\item $R \not\in \mathit{tDet}(v,s,M)$.
Then, $v(s_{1}) = v(s_{2})$ and  $v(s_{1}') = v(s_{2}')$ (because $R$'s content does not determine $v$'s materialization).
From $s_{1} \cong_{M,u}^{\mathit{data}} s_{1}'$ and the fact that $a$ modifies only $R$, it follows that  $v(\mathit{db}_{2}) = v(\mathit{db}_{2}')$ leading to  a contradiction.

\item $R \in \mathit{tDet}(v,s,M)$ and for all $o \in \mathit{tDet}(v,s, M)$, $\langle \oplus, \mathtt{SELECT}, o \rangle \in  p(s_{1},u)$. 
From this and $s_{1} \\ \cong_{M,u}^{\mathit{data}} s_{1}'$, it follows that, for all $o \in \mathit{tDet}(v, s, \\ M)$, $o(s_{1}) = o(s_{1}')$.
If $o \neq R$, $o(s_{1}) = o(s_{1}') = o(s_{2}) = o(s_{2}')$.
From  $\langle \oplus, \mathtt{SELECT},  R\rangle \in p(s_{1},u) $ and $s_{1} \cong_{M,u}^{\mathit{data}} s_{1}'$, it follows that $\mathit{db}_{1}(R) = \mathit{db}_{1}'(R)$.
From this and the fact that $a$ has been executed successfully both in $e(r,a)$ and $e(r',a)$, it follows that $\mathit{db}_{2}(R) = \mathit{db}_{1}  (R) \cup \{\overline{t}\}$ and  $\mathit{db}_{2}'(R) = \mathit{db}_{1}'(R) \cup \{\overline{t}\}$.
From this and $\mathit{db}_{1}(R) = \mathit{db}_{1}'(R)$, it follows that $\mathit{db}_{2}(R) = \mathit{db}_{2}'(R)$.
From this and for all $o \in \mathit{tDet}(v,s, M)$ such that $o \neq R$, $o(s_{2}) = o(s_{2}')$, it follows that for all $o \in \mathit{tDet}(v,s,M)$, $o(s_{2}) = o(s_{2}')$.
Since the content of all tables determining $v$ is the same in $s_{2}$ and $s_{2}'$, it follows that  $\mathit{db}_{2}(v) = \mathit{db}_{2}'(v)$ leading to  a contradiction.
\end{compactitem}

\end{compactenum}
All the cases lead to a contradiction.

\item $\mathit{secEx}(\mathit{last}(e(r,a))) = \bot$ and $\mathit{Ex}(\mathit{last}(e(r,a))) \neq \emptyset$.
From the LTS rules and  $\mathit{secEx}(\mathit{last}(e(r,a))) = \bot$, it follows that $f(\mathit{last}(r),a) = \top$.
From this and \thref{theorem:f:composition:soundness}, it follows that $f(\mathit{last}(r'),a) = \top$ for any $r' \in \llbracket r \rrbracket_{P,u}$.
From this and the LTS rules, it follows that $\mathit{secEx}( \mathit{last}(e(r',a))) = \bot$ for any $r' \in \llbracket r \rrbracket_{P,u}$.
Assume that the exception has been caused by the constraint $\gamma$, i.e., $\gamma \in \mathit{Ex}(\mathit{last}(e(r,a)))$.
From this and \thref{theorem:getInfo:sound:and:complete}, it follows that $\mathit{gV}(\gamma, \mathit{a})$ holds in $\mathit{last}(r).\mathit{db}$.
From  $f(\mathit{last}(r),a) = \top$ and $f$'s definition, it follows that $f_{\mathit{conf}}^{u}(\mathit{last}(r),a) = \top$ because $\mathit{user}(\mathit{last}(r),a) = u$ since $\mathit{trigger}(\mathit{last}(r))  = \epsilon$ and $a \in {\cal A}_{D,u}$.
From this and $f_{\mathit{conf}}^{u}$'s definition, it follows that $\mathit{secure}(u, \mathit{gV}(\gamma, \mathit{a}), \mathit{last}(r))$ holds.
From this, \thref{theorem:secure:sound:under:approximation}, and $[\mathit{gV}(\gamma, \mathit{a})]^{\mathit{last}(r).\mathit{db}} = \top$, it follows that also  $[\mathit{gV}(\gamma, \mathit{act})]^{\mathit{last}(r').\mathit{db}} = \top$ for any $r'  \in \llbracket r\rrbracket_{P,u}$. 
From this and \thref{theorem:getInfo:sound:and:complete}, it follows that $\gamma \in \mathit{Ex}(\mathit{last}(e(r', a)))$ for any $r'  \in \llbracket r\rrbracket_{P,u}$.
The data indistinguishability between $\mathit{last}(e(r,a))$ and $\mathit{last}(e(r',a))$ follows trivially from the data indistinguishability between $\mathit{last}(r)$ and $\mathit{last}(r')$ for any $r' \in \llbracket r \rrbracket_{P,u}$. 
Therefore, for any run $r' \in \llbracket r \rrbracket_{P,u}$, there is exactly one run $e(r',a)$.
From the considerations above, it follows trivially that $e(r',a) \in  \llbracket e(r,a) \rrbracket_{P,u}$.
The bijection $b$ is trivially $b(r') = e(r',a)$.
This leads to a contradiction.

\item $\mathit{secEx}(\mathit{last}(e(r,a))) = \top$.
From the LTS rules and  $\mathit{secEx}(\mathit{last}(e(r,a))) = \top$, it follows that $f(\mathit{last}(r),a) \\ = \bot$.
From this and \thref{theorem:f:composition:soundness}, it follows that $f(\mathit{last}(r'), a) = \bot$ for any $r' \in \llbracket r \rrbracket_{P,u}$.
From this and the LTS rules, it follows $\mathit{secEx}(\mathit{last}(e(r',a))) = \top$ for any $r' \in \llbracket r \rrbracket_{P,u}$.
The data indistinguishability between $\mathit{last}(e(r,a))$ and $\mathit{last}(e(r',a))$ follows trivially from that between $\mathit{last}(r)$ and $\mathit{last}(r')$ for any $r' \in \llbracket r \rrbracket_{P,u}$.
Therefore, for any run $r' \in \llbracket r \rrbracket_{P,u}$, there is exactly one run $e(r',a)$.
From the considerations above, it follows trivially that $e(r',a) \in  \llbracket e(r,a) \rrbracket_{P,u}$.
The bijection $b$ is trivially $b(r') = e(r',a)$. This leads to a contradiction.
\end{compactenum}
All cases lead to a contradiction.
This completes the proof for $a = \langle u, \mathtt{INSERT}, R, \overline{t} \rangle$.

\item $a = \langle u, \mathtt{DELETE}, R, \overline{t} \rangle$. 
The proof is similar to that for $a = \langle u, \mathtt{INSERT}, R, \overline{t} \rangle$.

\item $a = \langle \oplus, u', p, u\rangle$.
There are two cases:
\begin{compactenum}
\item $\mathit{secEx}(\mathit{last}(e(r,a))) = \bot$.
We assume that $p = \langle \mathtt{SELECT}, \\ O\rangle$ for some $O \in D \cup V$.
If this is not the case, the proof is trivial.
Furthermore, we also assume that $u' = u$, otherwise the proof is, again, trivial since the new permission does not influence $u$'s permissions.
From the LTS rules and  $\mathit{secEx}(\mathit{last}(e(r,a))) = \bot$, it follows that $f(\mathit{last}(r),a) = \top$.
From this and \thref{theorem:f:composition:soundness}, it follows that $f(\mathit{last}(r'),a) = \top$ for any $r' \in \llbracket r \rrbracket_{P,u}$.
From this and the LTS rules, it follows that $\mathit{secEx}(\mathit{last}(e(r',a))) = \bot$ for any $r' \in \llbracket r \rrbracket_{P,u}$. 
From $\mathit{secEx}(\mathit{last}(e(r',a))) = \bot$ and $f_{\mathit{conf}}^{u}$'s definition, it follows that $\mathit{last}(r').\mathit{sec} = \mathit{last}(e(r',  a)). \mathit{sec}$.
Therefore, since $\mathit{last}(r)$ and $\mathit{last}(r')$ are data indistinguishable, for any $r' \in \llbracket r \rrbracket_{P,u}$, then also   $\mathit{last}(e(r,a))$ and $\mathit{last}(e(r',a))$ are data indistinguishable.
Therefore, for any run $r' \in \llbracket r \rrbracket_{P,u}$, there is exactly one run $e(r',a)$.
From the considerations above, it follows trivially that $e(r',a) \in  \llbracket e(r,a) \rrbracket_{P,u}$.
The bijection $b$ is trivially $b(r') = e(r',a)$.
This leads to a contradiction.

\item $\mathit{secEx}(\mathit{last}(e(r,a))) = \top$.  
From the LTS rules and  $\mathit{secEx}(\mathit{last}(e(r,a))) = \top$, it follows $f(\mathit{last}(r),a) = \bot$.
From this and \thref{theorem:f:composition:soundness}, it follows that $f(\mathit{last}(r'),a) = \bot$ for any $r' \in \llbracket r \rrbracket_{P,u}$.
From this and the LTS rules, it follows $\mathit{secEx}(\mathit{last}(e(r',a))) = \top$ for any $r' \in \llbracket r \rrbracket_{P,u}$.
The data indistinguishability between $\mathit{last}(e(r',a))$ and $\mathit{last}(e(r,a))$ follows trivially from the data indistinguishability between $\mathit{last}(r')$ and $\mathit{last}(r)$.
Therefore, for any run $r' \in \llbracket r \rrbracket_{P,u}$, there is exactly one run $e(r',a)$.
From the considerations above, it follows trivially $e(r',a)  \in  \llbracket e(r,a) \rrbracket_{P,u}$.
The bijection $b$ is trivially $b(r') = e(r',a)$.
This leads to a contradiction.

\end{compactenum}
Both cases lead to a contradiction.
This completes the proof for $a = \langle \oplus, u', p, u\rangle$.

\item $a = \langle \oplus^{*}, u', p, u\rangle$.
The proof is similar to that for $a = \langle \oplus, u', p, u\rangle$.

\item $a = \langle \ominus, u', p, u\rangle$.
The proof is similar to that for $a = \langle u, \mathtt{SELECT}, q \rangle$.
The only difference is in proving that for any $r'  \in \llbracket r\rrbracket_{P,u}$, $\mathit{last}(e(r,a))$ and $\mathit{last}(e(r',a))$ are data indistinguishable.
Assume, for contradiction's sake, that this is not the case.
Let $s_{2} = \langle \mathit{db}_{2}, U_{2},\mathit{sec}_{2}, T_{2}, V_{2} \rangle$ be $\mathit{pState}(\mathit{last}(e(r, a)))$ and $s_{2}' = \langle \mathit{db}_{2}', U_{2}',\mathit{sec}_{2}', T_{2}',V_{2}' \rangle$ be $\mathit{pState}(\mathit{last}(e (r',a)))$.
Furthermore, let $s_{1} = \langle \mathit{db}_{1}, U_{1}, \\ \mathit{sec}_{1}, T_{1},V_{1} \rangle$ be $\mathit{pState}(\mathit{last}  (r))$ and  $s_{1}' = \langle \mathit{db}_{1}', U_{1}',\mathit{sec}_{1}', \\ T_{1}',V_{1}'\rangle$ be $\mathit{pState}(\mathit{last}  (r'))$.
In the following, we denote  the $\mathit{permissions}$ function by $p$.
Furthermore, note that $s_{1}$ and $s_{1}'$ are data-indistinguishable because $r'  \in \llbracket r\rrbracket_{P,u}$.
There are a number of cases:
\begin{compactenum}

\item $U_{2} \neq U_{2}'$. 
Since $a$ is an \texttt{REVOKE} operation, it follows that $U_{1} = U_{2}$ and $U_{1}' = U_{2}'$.
Furthermore, from $s_{1} \cong_{u,M}^{\mathit{data}} s_{1}'$, it follows that $U_{1} = U_{1}'$.
Therefore, $U_{2} = U_{2}'$ leading to a contradiction.

\item $\mathit{sec}_{2} \neq \mathit{sec}_{2}'$. 
From $s_{1} \cong_{u,M}^{\mathit{data}} s_{1}'$, it follows that $\mathit{sec}_{1} = \mathit{sec}_{1}'$.
From  $a$'s definition and the LTS rules, it follows that $\mathit{sec}_{2} = \mathit{revoke}(\mathit{sec}_{1}, u',  p,u)$ and $\mathit{sec}_{2}' = \mathit{revoke}(\mathit{sec}_{1}', u',p,u)$.
From this and $\mathit{sec}_{1} = \mathit{sec}_{1}'$, it follows that $\mathit{sec}_{2} = \mathit{sec}_{2}'$ leading to a contradiction.

\item $\mathit{T}_{2} \neq \mathit{T}_{2}'$. 
The proof is similar to the case $U_{2} \neq U_{2}'$. 

\item $\mathit{V}_{2} \neq \mathit{V}_{2}'$. 
The proof is similar to the case $U_{2} \neq U_{2}'$.

\item there is a table $R$ for which $\langle \oplus, \mathtt{SELECT},  R\rangle \in p(s_{2},u)$ and $\mathit{db}_{2}(R) \neq \mathit{db}_{2}'(R)$.
Since $a$ is an \texttt{REVOKE} operation, it follows that $\mathit{db}_{1} = \mathit{db}_{2}$ and $\mathit{db}_{1}' = \mathit{db}_{2}'$.
Furthermore, from $s_{1} \cong_{u,M}^{\mathit{data}} s_{1}'$, it follows that $\mathit{db}_{1}(R) = \mathit{db}_{1}'(R)$.
From this, $\mathit{db}_{1} = \mathit{db}_{2}$, and $\mathit{db}_{1}' = \mathit{db}_{2}'$, it follows that $\mathit{db}_{2}(R) = \mathit{db}_{2}'(R)$ leading to a contradiction.

\item there a view $v$ for which $\langle \oplus, \mathtt{SELECT},  v\rangle \in p(s_{2}, \\ u)$ and $\mathit{db}_{2}(v) \neq \mathit{db}_{2}'(v)$.
Since $a$ is an \texttt{REVOKE} operation, it follows that $\mathit{db}_{1} = \mathit{db}_{2}$ and $\mathit{db}_{1}' = \mathit{db}_{2}'$.
Furthermore, from $s_{1} \cong_{u,M}^{\mathit{data}} s_{1}'$, it follows that $\mathit{db}_{1}(v) = \mathit{db}_{1}'(v)$.
From this, $\mathit{db}_{1} = \mathit{db}_{2}$, and $\mathit{db}_{1}' = \mathit{db}_{2}'$, it follows that $\mathit{db}_{2}(v) = \mathit{db}_{2}'(v)$ leading to a contradiction.

\end{compactenum}
All the cases lead to a contradiction.

\item $a = \langle u, \mathtt{CREATE}, o \rangle$. 
The proof is similar to that for $a = \langle \ominus, u', p, u\rangle$.

\item $a = \langle u, \mathtt{ADD\_USER}, u' \rangle$. 
The proof is similar to that for $a = \langle \ominus, u', p, u\rangle$.

\end{compactenum}

This completes the proof.
\end{proof}

\begin{lemma}\thlabel{theorem:f:composition:pec:2}
Let $u$ be a user in  ${\cal U}$, $P = \langle M, f \rangle$ be an \accessControlConfiguration{}, where $M = \langle D,\Gamma\rangle$ is a system configuration  and $f$ is as above, and $L$ be the $P$-LTS.
For any run $r \in \mathit{traces}(L)$ such that $\mathit{invoker}(\mathit{last}(r)) = u$ and any trigger $t \in {\cal TRIGGER}_{D}$,
if $\mathit{extend}(r,t)$ is defined, then $t$ preserves the equivalence class for $r$, $M$, and $u$.
\end{lemma}

\begin{proof}

Let $u$ be a user in ${\cal U}$, $P = \langle M, f_{\mathit{conf}}^{u} \rangle$ be an \accessControlConfiguration{}, where $M = \langle D,\Gamma\rangle$ is a system configuration and $f_{\mathit{conf}}^{u}$ is as above, and $L$ be the $P$-LTS.
In the following, we use $e$ to refer to the $\mathit{extend}$ function.
The proof in cases where the trigger $t$ is not enabled or $t$'s \texttt{WHEN} condition is not secure are similar to the proof of the \texttt{SELECT} case of \thref{theorem:f:composition:pec:1}.
In the following, we therefore assume that the trigger $t$ is enabled and that its \texttt{WHEN} condition is secure.
We prove our claim by contradiction.
Assume, for contradiction's sake, that there is a run $r  \in \mathit{traces}(L)$ such that $\mathit{invoker}(\mathit{last}(r)) = u$ and a trigger $t$ such that $\mathit{e}(r,t)$ is defined and $t$ does not preserve the equivalence class for $r$, $P$, and $u$.
Since $\mathit{invoker}(\mathit{last}(r))=u$ and $\mathit{e}(r,t)$ is defined, then $\mathit{e}(r',t)$ is defined as well for any $r' \in \llbracket r \rrbracket_{P,u}$ (indeed, from $\mathit{invoker}(\mathit{last}(r))=u$, it follows that the last action in $r$ is either an action issued by $u$ or a trigger invoker by $u$.
From this, the fact that $e(r,t)$ is defined, and the fact that $r$ and $r'$ are indistinguishable, it follows that $\mathit{trigger}(\mathit{last}(r)) = \mathit{trigger}(\mathit{last}(r')) = t$).
Let $a$ be $t$'s action and $w = \langle u', \texttt{SELECT},q\rangle$ be the \texttt{SELECT} command associated with $t$'s \texttt{WHEN} condition.
Let $s$ be the state $\mathit{last}(r)$, $s'$ be the state obtained just after the execution of the \texttt{WHEN} condition, and $s''$ be the state $\mathit{last}(e(r,t))$.
There are a number of cases depending on $t$'s action $a$:
\begin{compactenum}

\item $a = \langle u', \mathtt{INSERT}, R, \overline{t} \rangle$. 
There are three cases:

\begin{compactenum}
\item $\mathit{secEx}(\mathit{last}(e(r,a))) = \bot$ and $\mathit{Ex}(\mathit{last}(e(r,a))) = \emptyset$. 
The proof of this case is similar to that of the corresponding case in \thref{theorem:f:composition:pec:1}.

\item $\mathit{secEx}(\mathit{last}(e(r,a))) = \bot$ and $\mathit{Ex}(\mathit{last}(e(r,a))) \neq \emptyset$.
The only difference between the proof of this case in this Lemma and in that of \thref{theorem:f:composition:pec:1} is that we have to establish again the data indistinguishability between $\mathit{last}(e(r,t))$ and $\mathit{last}(e(r',t))$.
Indeed, for triggers the roll-back state is, in general, different from the one immediately before the trigger's execution, i.e., it may be that $\mathit{pState}(\mathit{last}(e(r,t))) \\ \neq \mathit{pState}(\mathit{last}(r))$. 
We now prove that $\mathit{last}(e(r,t))$ and $\mathit{last}(e(r',t))$ are data indistinguishable. 
From the LTS semantics, it follows that $r = p \concat s_{0} \concat \langle \mathit{invoker}(\mathit{last}(r)), \mathit{op}, R',\overline{v} \rangle \concat  s_{1} \concat t_{1} \concat \ldots \concat s_{n-1} \concat t_{n} \concat s_{n}$, where $p \in \mathit{traces}(L)$ and $t_{1}, \ldots, t_{n} \in {\cal TRIGGER}_{D}$.
Similarly, $r' = p' \concat s_{0}' \concat \langle \mathit{invoker}(\mathit{last}(r)), \mathit{op}, R',\overline{v} \rangle \concat  s_{1}' \concat t_{1} \concat \ldots \concat s_{n-1}' \concat t_{n} \concat s_{n}'$, where $p' \in \mathit{traces}(L)$, $p \cong_{P,u} p'$, and all states $s_{i}$ and $s_{i}'$ are data indistinguishable.
Then, the roll-back states are, respectively, $s_{0}$ and $s_{0}'$, which are data indistinguishable.
From the LTS rules, $\mathit{last}(e(r,a)) = s_{0}$ and $\mathit{last}(e(r',a)) = s_{0}'$.
Therefore, the data indistinguishability between $\mathit{last}(e(r,a))$ and $\mathit{last}(e(r', a))$ follows trivially for any $r' \in \llbracket r \rrbracket_{P,u}$.

\item $\mathit{secEx}(e(r,a)) = \top$. 
The proof is similar to the previous case.
\end{compactenum}
All cases lead to a contradiction.
This completes the proof for $a = \langle u', \mathtt{INSERT}, R, \overline{t} \rangle$.

\item $a = \langle u', \mathtt{DELETE}, R, \overline{t} \rangle$.
The proof is similar to that for $a = \langle u', \mathtt{INSERT}, R, \overline{t} \rangle$.

\item $a = \langle \oplus, u'', p, u'\rangle$.
There are two cases:
\begin{compactenum}
\item $\mathit{secEx}(\mathit{last}(e(r,a))) = \bot$. 
In this case, the proof is similar to the  corresponding case in \thref{theorem:f:composition:pec:1}.

\item $\mathit{secEx}(\mathit{last}(e(r,a))) = \top$. 
The proof is similar to the $\mathit{secEx}(\mathit{last}(e(r,a))) = \top$ case of $a = \langle u', \mathtt{INSERT}, \\ R, \overline{t} \rangle$.
\end{compactenum}
Both cases lead to a contradiction.
This completes the proof for $a = \langle \oplus, u'', p, u'\rangle$.

\item $a = \langle \oplus^{*}, u'', p, u'\rangle$. 
The proof is similar to that for $a = \langle \oplus, u'', p, u'\rangle$.

\item $a = \langle \ominus, u'', p, u'\rangle$. 
The proof is similar to that for $a = \langle u', \mathtt{INSERT}, R, \overline{t} \rangle$.

\end{compactenum}
This completes the proof.
\end{proof}

%% file: nonInterference.tex
\clearpage
\section{Database Access Control and Information Flow Control}\label{app:data:conf:non:interference}
Here, we first show that the notion of secure judgment can be seen as an instance of non-interference.
Afterwards, we present NI-\confidentiality{}, a security notion for database access control that is an instance of non-interference.
Finally, we show that \confidentiality{} and NI-\confidentiality{} are equivalent.
For non-interference, we use terminology and notation taken from \cite{hedin2011perspective}.

It is easy to see that the notion of secure judgment is an instance of non-interference over relational calculus sentences.
Indeed, the set of all programs is just the set of all sentences, the set of inputs is the set of all runs, the equivalence relation between the inputs is $\cong_{P,u}$, the set of outputs is $\{\top,\bot\}$, the equivalence relation between the outputs is the equality, and the semantics of the programs is obtained by evaluating the sentences, according to the relational calculus semantics, over the  database state in the last state of a run.
Using a similar argument, one can easily show that both determinacy~\cite{nash2010views} and instance-based determinacy~\cite{Koutris:2012:QDP:2213556.2213582} are just instances of non-interference over relational calculus sentences.

Before defining NI-\confidentiality{}, we need some machinery.
Let $P = \langle M,f\rangle$ be an \accessControlConfiguration{}, $L$ be the $P$-LTS, $u \in {\cal U}$ be a user, $\attMod$ be  a $(P,u)$-attacker model, and $\cong$ be a $P$-indistinguishability relation.
Given a run $r$, we denote by $K(r)$ the set of all formulae that the user $u$ can derive from any extension of $r$ using $A$, i.e., $\{ \phi \in \mathit{RC}_{\mathit{bool}} \,|\, \exists s \in \mathit{traces}(L), i \in \mathbb{N}.\, s,i \attMod \phi \in A \wedge s^{|r|} = r\}$.
Moreover, given a set of formulae $K$, we say that two runs $r$ and $r'$ \emph{agree} on $K$, denoted by $r \equiv_K r'$, iff for all $\phi \in K$, $\phi$ holds in the last states of $r$ and $r'$.
Given a system state $s = \langle \mathit{db}, U,\mathit{sec},T,V,c\rangle$, we denote by $s.\mathit{db}$ the database state $\mathit{db}$.

We are now ready to define NI-\confidentiality{} notion.

\begin{definition}\label{definition:ni:data:confidentiality}
Let $P = \langle M,f\rangle$ be an \accessControlConfiguration{}, $L$ be the $P$-LTS, $u \in {\cal U}$ be a user, $A$ be a $(P,u)$-attacker model, and $\cong$ be a $P$-indistinguishability relation.
We say that \emph{$f$ provides NI-\confidentiality{} with respect to $P$, $u$, $A$, and $\cong$} iff for all runs $r,r' \in \mathit{traces}(L)$, if $r \cong r'$ holds, then $r \equiv_{K(r) \cup K(r')} r'$ holds. 
\end{definition}

Finally, we prove that NI-\confidentiality{} and \confidentiality{} are equivalent.

\begin{proposition}
 Let $P = \langle M,f\rangle$ be an \accessControlConfiguration{}, $L$ be the $P$-LTS,
    $u \in {\cal U}$ be a user, $\attMod$ be  a $(P,u)$-attacker model, and $\cong_{P,u}$ be a $(P,u)$-indistinguishability relation.
    The \acf{} $f$ provides \confidentiality{} iff it provides NI-\confidentiality{}.
\end{proposition}

\begin{proof}
We prove the two directions separately.

\smallskip
\noindent
$(\Rightarrow)$ 
We prove this direction by contradiction.
Assume that $f$ provides \confidentiality{} but it does not provide NI-\confidentiality{}.
From the fact that NI-\confidentiality{} does not hold, it follows that there are two runs $r,r' \in \mathit{traces}(L)$ such that $r \cong r'$ but $r \not\equiv_{K(r) \cup K(r')} r'$.
From $r \not\equiv_{K(r) \cup K(r')} r'$, it follows that there are two cases:
\begin{compactenum}
\item there is a run $s \in \mathit{traces}(L)$ such that $s^{|r|} = r$, $s, |r| \attMod \phi \in A$, and $[\phi]^{\mathit{last}(r).\mathit{db}} \neq [\phi]^{\mathit{last}(r').\mathit{db}}$.
From this, it follows that $\mathit{secure}_{P,\cong}(s, |r| \attMod \phi)$ does not hold, since $s^{|r|} = r$, $[\phi]^{\mathit{last}(r).\mathit{db}} \neq [\phi]^{\mathit{last}(r').\mathit{db}}$, and $r \cong r'$.
This contradicts the fact that $f$ provides \confidentiality{}.

\item there is a run $s \in \mathit{traces}(L)$ such that $s^{|r'|} = r'$, $s, |r'| \attMod \phi \in A$, and $[\phi]^{\mathit{last}(r).\mathit{db}} \neq [\phi]^{\mathit{last}(r').\mathit{db}}$.
From this, it follows that $\mathit{secure}_{P,\cong}(s, |r'| \attMod \phi)$ does not hold, that is not secure, since $s^{|r'|} = r'$, $[\phi]^{\mathit{last}(r).\mathit{db}} \neq [\phi]^{\mathit{last}(r').\mathit{db}}$, and $r \cong r'$.
This contradicts the fact that $f$ provides \confidentiality{}.
\end{compactenum}
Since both cases lead to a contradiction, this concludes the proof of this direction.

\smallskip
\noindent
$(\Leftarrow)$ 
We prove this direction by contradiction.
Assume that $f$ provides NI-\confidentiality{} but it does not provide \confidentiality{}.
From the fact that \confidentiality{} does not hold, it follows that there is a runs $r \in \mathit{traces}(L)$, an index $i$, and a sentence $\phi$ such that $r,i \attMod \phi \in A$ and $\mathit{secure}_{P,\cong}(r,i \attMod \phi)$ does not hold.
From this and $\mathit{secure}_{P,\cong}\\(r,i \attMod \phi)$'s definition, it follows that there are two runs $r,r' \in \mathit{traces}(L)$, an index $i$, and a sentence $\phi$ such that $r,i \attMod \phi \in A$, $r^i \cong r'$, and $[\phi]^{\mathit{last}(r^i).\mathit{db}} \neq [\phi]^{\mathit{last}(r').\mathit{db}}$.
From this and $|r^i|=i$, it follows that there are two runs $r,r' \in \mathit{traces}(L)$ and a sentence $\phi$ such that $r,|r^i| \attMod \phi \in A$, $r^i \cong r'$, and $[\phi]^{\mathit{last}(r^i).\mathit{db}} \neq [\phi]^{\mathit{last}(r').\mathit{db}}$.
By renaming $r^i$ as $k$ and by considering the fact that $r$ is, by definition, an extension of $k$, it follows that there are two runs $r,r' \in \mathit{traces}(L)$ and a sentence $\phi$ such that $r,|k| \attMod \phi \in A$, $r^{|k|} = k$, $k \cong r'$, and $[\phi]^{\mathit{last}(k).\mathit{db}} \neq [\phi]^{\mathit{last}(r').\mathit{db}}$.
From this and $K(k)$'s definition, it follows that there are two runs $k,r' \in \mathit{traces}(L)$ and a sentence $\phi$ such that $\phi \in K(k)$, $k \cong r'$, and $[\phi]^{\mathit{last}(k).\mathit{db}} \neq [\phi]^{\mathit{last}(r').\mathit{db}}$.
From this and $\phi \in K(k)$, it follows that there are two runs $k,r' \in \mathit{traces}(L)$ and a sentence $\phi$ such that $k \cong r'$ and $k \not\equiv_{K(k)} r'$.
From this, it follows that there are two runs $k,r' \in \mathit{traces}(L)$ and a sentence $\phi$ such that $k \cong_{P,u} r'$, and $k \not\equiv_{K(k) \cup K(r')} r'$.
This contradicts the fact that $f$ provides NI-\confidentiality{}.
\end{proof}

We now show that NI-\confidentiality{} can be seen as an instance of non-interference.
Let $M$ be a system configuration and $u$ be a user.
The set of programs ${\cal P}$ is the set of all pairs of the form $(f,\attMod)$, where $f$ is a system configuration and $\attMod$ is a $(\langle M, f \rangle, u )$-attacker model.
The set of inputs ${\cal I}$ is the set $\{ (s, \mathit{evs}) \,|\, s \in {\cal I}_{M} \wedge \mathit{evs} \in ({\cal A}_{D,{\cal U}} \cup {\cal TRIGGER}_D )^* \}$.
The set of outputs ${\cal O}$ is the set of all possible sequences of $M$-states and labels in ${\cal A}_{D,{\cal U}} \cup {\cal TRIGGER}_D$.
The semantics of the programs $\sigma : {\cal P} \times {\cal I} \rightarrow ({\cal O} \cup \{\bot\})$ is a total function defined as follows:
$\sigma( (f,\attMod) , (s , \mathit{evs})) = r$ iff (1) $r$ is a run in $\mathit{traces}(L)$, where $L$ is the $\langle M, f \rangle$-LTS, (2) $r$ starts from the state $s$, and (3) the labels of $r$ are equivalent to $\mathit{evs}$; $\sigma( (f,\attMod) , (s , \mathit{evs})) = \bot$ otherwise.
Finally, the relation $\sim$  over the set ${\cal I}$ is $\sim = {\cal I} \times {\cal I}$, i.e., any two inputs are indistinguishable, whereas the relation $\equiv$ over the set ${\cal O}$ is as follows:
for any two $r , r' \in {\cal O}$, $r \equiv r'$ iff  
(1) $r = \bot$,
(2) $r' = \bot$, or
(3) $r \neq \bot$, $r' \neq \bot$, and if $r \cong_{P,u} r'$, then $r \equiv_{K(r) \cup K(r')} r'$.
Note that $\equiv$ is not an equivalence relation, i.e., it is reflexive and symmetric but it is not transitive.
Therefore, a \acf{} $f$ provides NI-\confidentiality{} (and, therefore, \confidentiality{}) with respect to an attacker model $\attMod$ iff $(f,\attMod)$ satisfies non-interferences, where ${\cal P}$, ${\cal I}$, ${\cal O}$, $\sigma$, $\sim$, and $\equiv$ are as above.

%% file: sql_fragment.tex
\begin{figure*}
\begin{lstlisting}[mathescape, basicstyle=\ttfamily]
$\mathit{SqlStmt}$ := $\mathit{SelectStmt}$ | $\mathit{SqlBasicStmt}$ | $\mathit{CreateTrigger}$ | $\mathit{CreateView}$ 
$\mathit{SqlBasicStmt}$ :=  $\mathit{InsertStmt}$ | $\mathit{DeleteStmt}$ | $\mathit{GrantStmt}$ | $\mathit{RevokeStmt}$
$\mathit{SelectStmt}$ := "SELECT DISTINCT" $\mathit{columnList}$ "FROM" $\mathit{tableList}$ "WHERE" $\mathit{expr}$
$\mathit{columnList}$ := $\mathbf{columnId}$ | $\mathit{columnList}$ "," $\mathbf{columnId}$
$\mathit{tableList}$ := $\mathbf{tableId}$ | $\mathit{tableList}$ "," $\mathbf{tableId}$
$\mathit{expr}$ := $\mathbf{varId}$ "=" $\mathbf{const}$ | $\mathbf{varId}$ "=" $\mathbf{varId}$ | "NOT" "("$\mathit{expr}$")" | $\mathit{expr}$ ("AND"|"OR") $\mathit{expr}$ |
	"EXISTS" "("$\mathit{SelectStmt}$")"
$\mathit{InsertStmt}$ := "INSERT INTO" $\mathbf{tableId}$ "VALUES ("$\mathit{valueList}$")"
$\mathit{valueList}$ := $\mathbf{const}$ | $\mathit{valueList}$ "," $\mathbf{const}$
$\mathit{DeleteStmt}$ := "DELETE FROM" $\mathbf{tableId}$ "WHERE" $\mathit{restrictedExpr}$
$\mathit{restrictedExpr}$ := $\mathbf{varId}$ "=" $\mathbf{const}$ | $\mathit{restrictedExpr}$ "AND" $\mathbf{varId}$ "=" $\mathbf{const}$
$\mathit{GrantStmt}$ := "GRANT" $\mathit{privilege}$ "TO" $\mathbf{userId}$ ("WITH GRANT OPTION")
$\mathit{RevokeStmt}$ := "REVOKE" $\mathit{privilege}$ "FROM" $\mathbf{userId}$ "WITH CASCADE"
$\mathit{privilege}$ := "SELECT ON" ($\mathbf{tableId}$ | $\mathbf{viewId}$) | "CREATE VIEW" |
	( "INSERT" | "DELETE" | "CREATE TRIGGER" ) "ON" $\mathbf{tableId}$ 
$\mathit{CreateTrigger}$ := "CREATE TRIGGER" $\mathbf{triggerId}$ "AFTER" ("INSERT" | "DELETE") "ON" $\mathbf{tableId}$
		("SECURITY DEFINER" | "SECURITY INVOKER") $\mathit{SqlBasicStmt}$
$\mathit{CreateView}$ := "CREATE VIEW" $\mathbf{viewId}$ ("SECURITY DEFINER" | "SECURITY INVOKER")
		AS $\mathit{SelectStmt}$
\end{lstlisting}
\caption{This is the syntax of the SQL fragment that corresponds to the features  we support in this paper.}\label{table:sql:syntax}
\end{figure*}

%% file: draft.bbl
\begin{thebibliography}{10}
\providecommand{\url}[1]{#1}
\csname url@samestyle\endcsname
\providecommand{\newblock}{\relax}
\providecommand{\bibinfo}[2]{#2}
\providecommand{\BIBentrySTDinterwordspacing}{\spaceskip=0pt\relax}
\providecommand{\BIBentryALTinterwordstretchfactor}{4}
\providecommand{\BIBentryALTinterwordspacing}{\spaceskip=\fontdimen2\font plus
\BIBentryALTinterwordstretchfactor\fontdimen3\font minus
  \fontdimen4\font\relax}
\providecommand{\BIBforeignlanguage}[2]{{%
\expandafter\ifx\csname l@#1\endcsname\relax
\typeout{** WARNING: IEEEtranS.bst: No hyphenation pattern has been}%
\typeout{** loaded for the language `#1'. Using the pattern for}%
\typeout{** the default language instead.}%
\else
\language=\csname l@#1\endcsname
\fi
#2}}
\providecommand{\BIBdecl}{\relax}
\BIBdecl

\bibitem{sybase2003security}
``{New Security Features in Sybase Adaptive Server Enterprise},'' \emph{{Sybase
  Technical White Paper, Sybase, an SAP company}}, 2003.

\bibitem{microsoft2014}
\BIBentryALTinterwordspacing
(2014, Sep.) Manage trigger security, \emph{Microsoft MSDN Library}. [Online].
  Available: \url{http://msdn.microsoft.com/en-us/library/ms191134.aspx/}
\BIBentrySTDinterwordspacing

\bibitem{abiteboul1995foundations}
S.~Abiteboul, R.~Hull, and V.~Vianu, \emph{Foundations of databases}.\hskip 1em
  plus 0.5em minus 0.4em\relax Addison-Wesley Reading, 1995, vol.~8.

\bibitem{agrawal2005extending}
R.~Agrawal, P.~Bird, T.~Grandison, J.~Kiernan, S.~Logan, and W.~Rjaibi,
  ``Extending relational database systems to automatically enforce privacy
  policies,'' in \emph{Proc. 2005 IEEE Int. Conf. Data Engineering}.

\bibitem{askarov2012learning}
A.~Askarov and S.~Chong, ``Learning is change in knowledge: Knowledge-based
  security for dynamic policies,'' in \emph{Proc. 2012 IEEE Symp. Computer
  Security Foundations}.

\bibitem{askarov2007gradual}
A.~Askarov and A.~Sabelfeld, ``Gradual release: Unifying declassification,
  encryption and key release policies,'' in \emph{Proc. 2007 IEEE Symp.
  Security and Privacy}.

\bibitem{askarov2009tight}
------, ``Tight enforcement of in\-for\-ma\-tion-release policies for dynamic
  languages,'' in \emph{Proc. 2009 IEEE Symp. Computer Security Foundations}.

\bibitem{Bender:2014:ESR:2588555.2593663}
G.~Bender, L.~Kot, and J.~Gehrke, ``Explainable security for relational
  databases,'' in \emph{Proc. 2014 ACM Intl. Conf. Management of data}.

\bibitem{Bender:2013:FDC:2463676.2467798}
G.~M. Bender, L.~Kot, J.~Gehrke, and C.~Koch, ``Fine-grained disclosure control
  for app ecosystems,'' in \emph{Proc. 2013 ACM Intl. Conf. Management of
  data}.

\bibitem{bohannon2009reactive}
A.~Bohannon, B.~C. Pierce, V.~Sj{\"o}berg, S.~Weirich, and S.~Zdancewic,
  ``Reactive noninterference,'' in \emph{Proc. 2009 ACM Conf. Computer and
  Communications Security}.

\bibitem{bonatti1995foundations}
P.~A. Bonatti, S.~Kraus, and V.~Subrahmanian, ``Foundations of secure deductive
  databases,'' \emph{IEEE Trans. Knowl. Data Eng.}, vol.~7, no.~3, 1995.

\bibitem{brodsky2000secure}
A.~Brodsky, C.~Farkas, and S.~Jajodia, ``Secure databases: Constraints,
  inference channels, and monitoring disclosures,'' \emph{IEEE Trans. Knowl.
  Data Eng.}, vol.~12, no.~6, 2000.

\bibitem{browder2002virtual}
K.~Browder and M.~Davidson, ``The virtual private database in oracle9ir2,''
  \emph{Oracle Technical White Paper, Oracle Corporation}, vol. 500, 2002.

\bibitem{clavel2003maude}
M.~Clavel, F.~Dur{\'a}n, S.~Eker, P.~Lincoln, N.~Mart{\'\i}-Oliet, J.~Meseguer,
  and C.~Talcott, ``The maude 2.0 system,'' in \emph{Rewriting Techniques and
  Applications}.\hskip 1em plus 0.5em minus 0.4em\relax Springer, 2003.

\bibitem{corcoran2009cross}
B.~J. Corcoran, N.~Swamy, and M.~Hicks, ``Cross-tier, label-based security
  enforcement for web applications,'' in \emph{Proc. 2009 ACM Intl. Conf.
  Management of data}.

\bibitem{davis2010dbtaint}
B.~Davis and H.~Chen, ``{DBTaint}: cross-application information flow tracking
  via databases,'' \emph{Proc. 2010 USENIX Conf. Web Application Development}.

\bibitem{denning1987multilevel}
D.~E. Denning and T.~F. Lunt, ``A multilevel relational data model,'' in
  \emph{Proc. 1987 IEEE Symp. Security and Privacy}.

\bibitem{devriese2010noninterference}
D.~Devriese and F.~Piessens, ``Noninterference through secure
  multi-execution,'' in \emph{Proc. 2010 IEEE Symp. Security and Privacy}.

\bibitem{dolev1983security}
D.~Dolev and A.~C. Yao, ``On the security of public key protocols,'' \emph{IEEE
  Trans. Inf. Theory}, vol.~29, no.~2, 1983.

\bibitem{farkas2002inference}
C.~Farkas and S.~Jajodia, ``The inference problem: a survey,'' \emph{ACM SIGKDD
  Explorations}, vol.~4, no.~2, 2002.

\bibitem{Giacobazzi:2004:ANP:964001.964017}
R.~Giacobazzi and I.~Mastroeni, ``Abstract non-interference: Parameterizing
  non-interference by abstract interpretation,'' in \emph{Proc. 2004 ACM Symp.
  Principles of Programming Languages}.

\bibitem{GoguenM82}
J.~A. Goguen and J.~Meseguer, ``Security policies and security models,'' in
  \emph{IEEE Symp. Security and Privacy}, 1982.

\bibitem{griffiths1976authorization}
P.~P. Griffiths and B.~W. Wade, ``An authorization mechanism for a relational
  database system,'' \emph{ACM Trans. on Database Syst.}, vol.~1, no.~3, 1976.

\bibitem{guarnieri2014optimal}
M.~Guarnieri and D.~Basin, ``Optimal security-aware query processing,'' in
  \emph{Proc. 2014 Int. Conf. Very Large Data Bases}.

\bibitem{guarnieri2016strong}
M.~Guarnieri, S.~Marinovic, and D.~Basin, ``Strong and provably secure database
  access control,'' in \emph{Proc.~2016 IEEE European Symp. Security and
  Privacy}.

\bibitem{prototype}
\BIBentryALTinterwordspacing
------. {Strong and Provably Secure Database Access Control --- Prototype and
  Maude models}. [Online]. Available:
  \url{http://www.infsec.ethz.ch/research/projects/FDAC.html}
\BIBentrySTDinterwordspacing

\bibitem{halder2013fine}
R.~Halder and A.~Cortesi, ``Fine grained access control for relational
  databases by abstract interpretation,'' in \emph{Software and Data
  Technologies}, 2013, vol. 170.

\bibitem{hedin2011perspective}
D.~Hedin and A.~Sabelfeld, ``A perspective on in\-for\-ma\-tion-flow control,''
  \emph{Proc. 2011 Marktoberdorf Summer School. IOS Press}, 2011.

\bibitem{jajodia1990polyinstantiation}
S.~Jajodia and R.~Sandhu, ``Polyinstantiation integrity in multilevel
  relations,'' in \emph{Proc. 1990 IEEE Symp. Security and Privacy}.

\bibitem{Koutris:2012:QDP:2213556.2213582}
P.~Koutris, P.~Upadhyaya, M.~Balazinska, B.~Howe, and D.~Suciu, ``Query-based
  data pricing,'' in \emph{Proc. 2012 ACM Symp. Principles of Database
  Systems}.

\bibitem{lefevre2004limiting}
K.~LeFevre, R.~Agrawal, V.~Ercegovac, R.~Ramakrishnan, Y.~Xu, and D.~DeWitt,
  ``Limiting disclosure in hippocratic databases,'' in \emph{Proc. 2004 Int.
  Conf. Very Large Data Bases}.

\bibitem{li2005practical}
P.~Li and S.~Zdancewic, ``Practical information flow control in web-based
  information systems,'' in \emph{Proc. 2005 IEEE Workshop on Computer Security
  Foundations}.

\bibitem{Myers:1997:DMI:268998.266669}
A.~C. Myers and B.~Liskov, ``A decentralized model for information flow
  control,'' in \emph{Proc. 1997 ACM Symp. Operating Systems Principles}.

\bibitem{nash2010views}
A.~Nash, L.~Segoufin, and V.~Vianu, ``Views and queries: Determinacy and
  rewriting,'' \emph{ACM Trans. Database Syst.}, vol.~35, no.~3, 2010.

\bibitem{rizvi2004extending}
S.~Rizvi, A.~Mendelzon, S.~Sudarshan, and P.~Roy, ``Extending query rewriting
  techniques for fine-grained access control,'' in \emph{Proc. 2004 ACM Int.
  Conf. Management of data}.

\bibitem{sabelfeld2003language}
A.~Sabelfeld and A.~C. Myers, ``Language-based in\-for\-ma\-tion-flow
  security,'' \emph{IEEE J. Sel. Areas Commun.}, vol.~21, no.~1, 2003.

\bibitem{samarati2001access}
P.~Samarati and S.~Capitani~de Vimercati, ``{Access Control: Policies, Models,
  and Mechanisms},'' \emph{Springer Lecture Notes in Computer Science}, vol.
  2171, 2001.

\bibitem{sandhu1998multilevel}
R.~Sandhu and F.~Chen, ``{The multilevel relational (MLR) data model},''
  \emph{ACM Trans. Inf. Syst. Sec.}, vol.~1, no.~1, 1998.

\bibitem{schoepe2014selinq}
D.~Schoepe, D.~Hedin, and A.~Sabelfeld, ``{SeLINQ}: tracking information across
  application-database boundaries,'' in \emph{Proc. 2014 ACM Intl. Conf.
  Functional Programming}.

\bibitem{schultz2013ifdb}
D.~Schultz and B.~Liskov, ``{IFDB}: decentralized information flow control for
  databases,'' in \emph{Proc. 2013 ACM European Conf. Computer Systems}.

\bibitem{shi2009soundness}
J.~Shi, H.~Zhu, G.~Fu, and T.~Jiang, ``On the soundness property for sql
  queries of fine-grained access control in dbmss,'' in \emph{Proc. 2009
  IEEE/ACIS Int. Conf. Computer and Information Science}.

\bibitem{smith1993multilevel}
K.~Smith and M.~Winslett, ``Multilevel secure rules: Integrating the multilevel
  secure and active data models,'' in \emph{Database Security VI: Status and
  Prospects}.\hskip 1em plus 0.5em minus 0.4em\relax North-Holland, 1993.

\bibitem{stonebraker1974access}
M.~Stonebraker and E.~Wong, ``Access control in a relational data base
  management system by query modification,'' in \emph{Proc. 1974 ACM Annual
  Conference}.

\bibitem{toland2010inference}
T.~S. Toland, C.~Farkas, and C.~M. Eastman, ``The inference problem:
  Maintaining maximal availability in the presence of database updates,''
  \emph{Computers \& Security}, vol.~29, no.~1, 2010.

\bibitem{VanGelder:1991:STR:114325.103712}
A.~Van~Gelder and R.~W. Topor, ``Safety and translation of relational
  calculus,'' \emph{ACM Trans. Database Syst.}, vol.~16, no.~2, pp. 235--278,
  May 1991.

\bibitem{wang2007correctness}
Q.~Wang, T.~Yu, N.~Li, J.~Lobo, E.~Bertino, K.~Irwin, and J.-W. Byun, ``On the
  correctness criteria of fine-grained access control in relational
  databases,'' in \emph{Proc. 2007 Int. Conf. Very Large Data Bases}.

\end{thebibliography}
